%% file: main.tex
\PassOptionsToPackage{dvipsnames}{xcolor}
\documentclass[sigplan,10pt,review=false,screen]{acmart}
\setcopyright{none}

\usepackage{import}
\usepackage{preamble}

\begin{document}
\renewcommand\footnotetextcopyrightpermission[1]{} 

\title{Recursive Session Logical Relations}

\author{Farzaneh Derakhshan}
\affiliation{%
  \institution{Carnegie Mellon University}
  \country{USA}
}
\author{Stephanie Balzer}
\affiliation{%
  \institution{Carnegie Mellon University}
  \country{USA}
}

\renewcommand{\shortauthors}{Derakhshan and Balzer}

\begin{abstract}
\input{sections/abstract}
\end{abstract}



%

\maketitle
\input{sections/1.intro} 
\input{sections/2.background} 
\input{sections/3.recursion} 
\input{sections/4.logical_relation} 
\input{sections/5.metatheory} 
\input{sections/6.related_work} 
 \section*{\acksname}
This work was supported by the United States Air Force Office of Scientific Research (AFOSR) award FA9550-21-1-0385 (Tristan Nguyen, program manager).  Any opinions, findings and conclusions or recommendations expressed in this material are those of the authors and do not necessarily reflect the views of the AFOSR.

\bibliographystyle{ACM-Reference-Format}
\bibliography{ref.bib}

\clearpage



\section*{Supplementary material}

# A  Further examples

This section illustrates the challenges recursion poses for information flow control based on various examples and motivates the need for processes that are polymorphic in their maximal secrecy and running secrecy. The examples showcase scenarios of the banking application used in the main part of the paper, given the security lattice

$$\textbf{guest} \sqsubseteq \textbf{alice} \sqsubseteq \textbf{bank} \quad\quad \textbf{guest} \sqsubseteq \textbf{bob} \sqsubseteq \textbf{bank}$$

## A.1  Examples of leakage — rejected by SILL$_{\text{sec+}}$

This subsection gives several example programs that leak information by exploiting non-termination and concurrency. None of the examples use secrecy polymorphism to orchestrate the balance between a process' maximal secrecy and running secrecy and to ensure that a process' running secrecy stands for an arbitrary iteration of the process. The examples thus do not use secrecy variables nor running secrecy annotations, and hence do not type check in SILL$_{\text{sec+}}$.

### A.1.1  Different recursive calls after branching on a high secrecy channel

In our first example, the sneaky verification process (Sneaky_Verifier$_1$) leaks, with the help of an accomplice (Sneaky_Partner$_1$), the information of whether Alice's PIN provides a correct token. If verification of Alice's PIN is successful, the sneaky verifier calls itself recursively, and if not, it calls its partner. Sneaky_Verifier$_1$ and its partner, Sneaky_Partner$_1$, differ on how they interact with the attacker right after being spawned: the sneaky verifier sends the label $s$ to the attacker, signaling preceding successful authentication, and its partner sends the label $f$ to the attacker, signaling preceding unsuccessful authentication. We point out that the code would be perfectly secure without the recursive calls: the sends of label $s$ and $f$ in the Sneaky_Verifier$_1$ and Sneaky_Partner$_1$, resp., happen before the security token is received. Recursion, however, provides a way of "rolling forward" information learned in one iteration to the next.

ver = pin ⊸ ⊕{$succ$:pin ⊗ ver, $fail$:pin ⊗ ver}
pin = ⊕{$tok_1$:pin, ..., $tok_n$:pin}
attacker = &{$s$:attacker, $f$:attacker}

$y$:attacker[**guest**] ⊢ Sneaky_Verifier$_1$ :: $x$:ver[**alice**]
$x$ ← Sneaky_Verifier$_1$ ← $y$ = ($y.s$; $z$ ← **recv** $x$
  **case** $z$ ($tok_j$ ⇒ $x.succ$; **send** $z$ $x$; ($x$:ver ← Sneaky_Verifier$_1$ ← $y$:attacker)
       | $tok_{i \neq j}$ ⇒ $x.fail$; **send** $z$ $x$; ($x$:ver ← Sneaky_Partner$_1$ ← $y$:attacker)))

$y$:attacker[**guest**] ⊢ Sneaky_Partner$_1$ :: $x$:ver[**alice**]
$x$ ← Sneaky_Partner$_1$ ← $y$ = ($y.f$; $z$ ← **recv** $x$
  **case** $z$ ($tok_j$ ⇒ $x.succ$; **send** $z$ $x$; ($x$:ver ← Sneaky_Verifier$_1$ ← $y$:attacker)
       | $tok_{i \neq j}$ ⇒ $x.fail$; **send** $z$ $x$; ($x$:ver ← Sneaky_Partner$_1$ ← $y$:attacker)))

### A.1.2  Exploiting non-termination

Recursive protocols significantly improve the expressiveness of the language but at the same time degrade liveness. In particular, we lose the guarantee that a process will eventually communicate with its provider or client. The process Diverge is an example of a process that is supposed to provide an infinite stream of terminating channel references $x_1$ : term, but instead calls itself recursively without ever interacting with its client.

term = 1 ⊗ term

· ⊢ Diverge :: $x_1$:term[**guest**]
$x_1$ ← Diverge ← · = ($x_1$:term ← Diverge ← ·)

Diverge is only a trivial implementation of a non-reactive process, but there are more ingenious ways to implement such processes, and identifying them was shown to be undecidable [1].

---
[1] Farzaneh Derakhshan and Frank Pfenning. Circular proofs as session-typed processes: A local validity condition. CoRR abs/1908.01909 (2019), http://arxiv.org/abs/1908. 01909.

A sneaky verifier can exploit the non-reactivity of Diverge to leak information to the attacker. In our second example, the sneaky verifier, Sneaky_Verifier$_2$, spawns the diverging process along channel $x_1$ and continues as its partner Sneaky_Partner$_2$. The partner only interacts with the diverging channel $x_1$ when Alice's PIN provides the incorrect token. As a result, if Alice provides an incorrect token, the process will wait forever for $x_1$ to send it a terminating channel $x_2$, and the attacker will receive no further messages. Otherwise, if Alice's token is correct, Sneaky_Partner$_2$ continues by recursively calling itself, and right at the beginning of the recursive call, sends the label $s$ to the attacker, signaling the previous successful login.

$y$:attacker[**guest**] ⊢ Sneaky_Verifier$_2$ :: $x$:ver[**alice**]
$x$ ← Sneaky_Verifier$_2$ ← $y$ = (($x_1$:term[**guest**] ← Diverge ← ·);
$\qquad\qquad\qquad\qquad$ ($x$:ver ← Sneaky_Partner$_2$ ← $y$:attacker, $x_1$:term))

$y$:attacker[**guest**], $x_1$:term[**guest**] ⊢ Sneaky_Partner$_2$ :: $x$:ver[**alice**]
$x$ ← Sneaky_Partner$_2$ ← $y$, $x_1$ =
($y.s; z$ ← **recv** $x$ // $y$:attacker, $z$:pin, $x_1$:term ⊢ $x$: ⊕ $\{succ : \text{pin} \otimes \text{ver}, fail:\text{pin} \otimes \text{ver}\}$
$\quad$ **case** $z$ ($tok_j \Rightarrow x.succ$; **send** $z\ x$; ($x$:ver ← Sneaky_Partner$_2$ ← $y$:attacker, $x_1$:term)
$\qquad\quad$ | $tok_{i \neq j} \Rightarrow x.fail$; **send** $z\ x$; $x_2$ ← **recv** $x_1$; **wait** $x_2$;
$\qquad\qquad\quad$ ($x$:ver ← Sneaky_Partner$_2$ ← $y$:attacker, $x_1$:term)))

### A.1.3 Exploiting concurrency

The following example shows how two attackers can exploit concurrency to infer Alice's success or failure in verifying her PIN by observing the verifier's pattern in producing their low-secrecy messages. Process Sneaky_Verifier$_3$ leaks secret information with the help of process Alternate_bits that produces an alternating sequence of bits $b0; b1; b0; b1; \ldots$ and its two mutually recursive partners, Sneaky_Partner$_{b0}$ and Sneaky_Partner$_{b1}$.

bits = ⊕$\{b0 : \text{bits}, b1 : \text{bits}\}$
cobits = &$\{b0 : \text{bits}, b1 : \text{bits}\}$
leak = &$\{succ : \text{leak}, fail : \text{leak}\}$

$x_1$:bits[**guest**], $x_2$:cobits[**guest**] ⊢ Sneaky_Verifier$_3$ :: $x$:ver[**alice**]
$x$ ← Sneaky_Verifier$_3$ ← $x_1, x_2$ =
**case** $x_1$ ($b0 \Rightarrow x_2.b0; z$ ← **recv** $x$; **case** $z$ ($tok_j \Rightarrow x.succ$; **send** $z\ x$;
$\qquad\qquad\qquad\qquad\qquad\qquad\qquad$ **case** $x_1$($b0 \Rightarrow x$ ← Sneaky_Verifier$_3$ ← $x_1, x_2$,
$\qquad\qquad\qquad\qquad\qquad\qquad\qquad\qquad$ | $b1 \Rightarrow x$ ← Sneaky_Verifier$_3$ ← $x_1, x_2$)
$\qquad\qquad\qquad\qquad\qquad\quad$ | $tok_{i \neq j} \Rightarrow x$ ← Sneaky_Verifier$_3$ ← $x_1, x_2$)
$\quad$ | $b1 \Rightarrow x_2.b1; z$ ← **recv** $x$; **case** $z$ ($tok_j \Rightarrow x.succ$; **send** $z\ x$;
$\qquad\qquad\qquad\qquad\qquad\qquad\qquad$ **case** $x_1$($b0 \Rightarrow x$ ← Sneaky_Verifier$_3$ ← $x_1, x_2$,
$\qquad\qquad\qquad\qquad\qquad\qquad\qquad\qquad$ | $b1 \Rightarrow x$ ← Sneaky_Verifier$_3$ ← $x_1, x_2$)
$\qquad\qquad\qquad\qquad\qquad\quad$ | $tok_{i \neq j} \Rightarrow x$ ← Sneaky_Verifier$_3$ ← $x_1, x_2$)))

· ⊢ Alternate_bits :: $x_1$:bits[**guest**]
$x_1$ ← Alternate_bits ← · = ($x_1.b0; x_2.b1; (x_1$:bits ← Alternate_bits ← ·))

$x_3$:leak[**guest**] ⊢ Sneaky_Partner$_{b0}$ :: $x_2$:cobits[**guest**]
$x_2$ ← Sneaky_Partner$_{b0}$ ← $x_3$ = **case** $x_2$($b0 \Rightarrow x_3.succ; x_2$ ← Sneaky_Partner$_{b0}$ ← $x_3$;
$\qquad\qquad\qquad\qquad\qquad\quad$ | $b1 \Rightarrow x_3.fail; x_2$ ← Sneaky_Partner$_{b1}$ ← $x_3$))

$x_3$:leak[**guest**] ⊢ Sneaky_Partner$_{b1}$ :: $x_2$:cobits[**guest**]
$x_2$ ← Sneaky_Partner$_{b1}$ ← $x_3$ = **case** $x_2$($b0 \Rightarrow x_3.fail; x_2$ ← Sneaky_Partner$_{b0}$ ← $x_3$;
$\qquad\qquad\qquad\qquad\qquad\quad$ | $b1 \Rightarrow x_3.succ; x_2$ ← Sneaky_Partner$_{b1}$ ← $x_3$))

At the beginning, the Sneaky_Verifier$_3$ is connected to Alternate_bits along channel $x_1$, and to Sneaky_Partner$_{b0}$ along channel $x_2$. The sneaky verifier starts by relaying the first bit ($b0$) that it receives from $x_1$ to its accomplice along $x_2$. Next, if the verification of Alice's PIN is successful it receives one extra bit ($b1$) from the sequence provided by Alternate_bits without relaying it to its partner and then calls itself recursively. In this case, again, the next bit that it sends to Sneaky_Partner$_{b0}$ is $b0$. In fact, as long as the verification is successful, the verifier only sends $b0$ messages to its partner. As soon as a failure occurs, the verifier skips receiving the extra bit from Alternate_bits and instead relays a $b1$ label to Sneaky_Partner$_{b0}$ after its recursive call. With the same argument, we can observe that Sneaky_Verifier$_3$ continues to send $b1$ labels until the next verification failure, after which it again flips to sending $b0$. The mutually recursive partners observe these flips in the sequence they receive along $x_2$ and restore the sequence of Alice's successful and failed logins by building a sequence of type leak.

In all the above examples, the leak is caused by sending to an attacker after receiving along a high-secrecy channel, which is prevented by a flow-sensitive IFC type system. However, in contrast to the non-recursive case, causality between actions

can no longer determined locally, by just considering one iteration of either a sneaky verifier or its partners. The leaks come about because of mutually recursive calls, non-termination of spawned processes, and concurrent communications.

Our solution is to type processes not only by considering one iteration in isolation but considering both the maximal and running secrecy. The running secrecy must thereby be a sound approximation of an arbitrary iteration. To ensure compositionality, we moreover allow process definitions to be *polymorphic* in their maximal and running secrecy, given restrictions on the relationships between those variables.

To illustrate the use of polymorphic process definitions, we revisit the unsafe example from §A.1.1 and explain why our type system rejects it. We focus on this example, but the same argument holds for the remaining ones.

Consider the below definition of Sneaky_Verifier$_1$ and assume that the extended lattice $\Psi$ contains the constraints $\psi_1 = $ **guest** and $\psi = $ **alice**.

$\Psi; y$:attacker$[\psi_1] \vdash$ Sneaky_Verifier$_1@\psi' :: x$:ver$[\psi]$
$x \leftarrow$ Sneaky_Verifier$_1 \leftarrow y = (y.s; z \leftarrow \mathbf{recv}\, x$
  **case** $z\, (tok_j \Rightarrow x.succ; \mathbf{send}\, z\, x; x$:ver $\leftarrow$ Sneaky_Verifier$_1@d_2 \leftarrow y$:attacker
    $|\ tok_{i \neq j} \Rightarrow x.fail; \mathbf{send}\, z\, x; x$:ver $\leftarrow$ Sneaky_Partner$_1@d_3 \leftarrow y$:attacker))$@\psi'$

Following the signature typing rule, the extended lattice $\Psi$, must satisfy $\Psi \Vdash \mathbf{guest} = \psi_1 \sqsubset \mathbf{alice} = \psi$, $\psi' \sqsubseteq \psi = \mathbf{alice}$. Moreover, by &$L$, we must have $\Psi \Vdash \psi' \sqsubseteq \psi_1$, to execute $y.s$. By $\multimap R$, we know that the running secrecy of the process after executing $z \leftarrow \mathbf{recv}\, x^\psi$ increases to $\psi$. And by the spawn rule, we need to know that the running secrecy $d_2$ specified for the callee is at least as high as the caller's, i.e. $\Psi \Vdash \psi \sqsubseteq d_2$. It is straightforward to observe that there is no possible substitution $\gamma$ to unify $d_2$ with $\psi'$ as required by the spawn rule, i.e. $\hat{\gamma}(\Psi) \nVdash d_2 \sqsubseteq \psi_1 = \mathbf{guest} \sqsubset \mathbf{alice} = \psi \sqsubseteq d_2$.

## A.2 Secrecy-polymorphic processes to the rescue

As a bonus to our polymorphic treatment of processes, we get to define generic process definitions. This section completes our banking example by providing generic implementations for a customer and its authorization process that can be spawned for any specific customer of the bank. We add the following session types to our example:

customer = pin $\multimap$ auth$_{\mathsf{out}}$ $\multimap$ 1
pin $= \oplus\{tok_1$:pin$, \ldots, tok_n$:pin$\}$
auth$_{\mathsf{out}} = $ pin $\multimap \oplus\{succ$:account $\otimes$ auth$_{\mathsf{in}}$, $fail$:pin $\otimes$ auth$_{\mathsf{out}}\}$
auth$_{\mathsf{in}} = $ account $\multimap$ pin $\otimes$ auth$_{\mathsf{out}}$
ver = pin $\multimap \oplus\{succ$:pin $\otimes$ ver, $fail$:pin $\otimes$ ver$\}$
account $= \oplus\{high$:account, $med$:account, $low$:account$\}$

The bank sets up a new customer by sending the customer an authentication PIN and an authorization process that guards their account.

$\Psi; y_1$: customer$[\psi_1], u_1$: pin$[\psi_1], z_1$:auth$_{\mathsf{out}}[\psi_1]$,  $\Psi := \psi_1 = \mathbf{alice}, \psi_2 = \mathbf{bob}, \psi = \mathbf{bank}, \psi' = \mathbf{guest}, \Psi_0$
    $y_2$: customer$[\psi_2], u_2$:pin$[\psi_2], z_2$:auth$_{\mathsf{out}}[\psi_2] \vdash$ Bank $:: w$:1$[\psi]$
$w \leftarrow$ Bank$@\psi' \leftarrow y = (\ \mathbf{send}\, u_1^{\psi_1}\, y_1; \mathbf{send}\, z_1^{\psi_1}\, y_1; \mathbf{send}\, u_2^{\psi_2}\, y_2; \mathbf{send}\, z_2^{\psi_2}\, y_2; \mathbf{wait}\, y_1; \mathbf{wait}\, y_2; \mathbf{close}\, w)$
$\Psi_1; \cdot \vdash$ New_Customer$@\psi' :: y$: customer$[\psi]$    $\Psi_1 := \psi' \sqsubseteq \psi, \Psi_0$
$y \leftarrow$ New_Customer $\leftarrow \cdot = (\ u^\psi \leftarrow \mathbf{recv}\, y; z^\psi \leftarrow \mathbf{recv}\, y;\ //\ u$:pin$[\psi], z$:auth$_{\mathsf{out}}[\psi] \vdash y$:1$[\psi]$
                  $y^\psi$:1 $\leftarrow$ Customer$@\psi' \leftarrow u$:pin, $z$:auth$_{\mathsf{out}})@\psi'$

The types auth$_{\mathsf{out}}$ and auth$_{\mathsf{in}}$ describe the communication of the authorization process with the customer. The authorization process receives a PIN from the customer and sends it to the verifier to validate the login. If the verifier process signals a successful verification, the authorization goes ahead and logs in the customer by sending a *succ* label, followed by the account to the customer. If the PIN verification is unsuccessful, the authorization process returns the incorrect PIN to the customer and remains logged out. In the case of a successful verification, the customer remains logged in until it requests a log out by sending back its account to the authorization process; upon such a request, the account sends back the PIN to the customer and logs them out.

$\Psi_1; x{:}\mathsf{ver}[\psi], v{:}\mathsf{account}[\psi] \vdash \mathsf{Auth_{out}}@\psi' :: z{:}\mathsf{auth_{out}}[\psi] \qquad \Psi_1 = \psi' \sqsubseteq \psi, \Psi_0$
$z{:}\mathsf{auth_{out}} \leftarrow \mathsf{Auth_{out}} \leftarrow x{:}\mathsf{ver}, v{:}\mathsf{account} = (\ //\ x{:}\mathsf{ver}, v{:}\mathsf{account} \vdash z{:}\mathsf{auth_{out}}$
$\quad w_3^\psi \leftarrow \mathbf{recv}\ z; \mathbf{send}\,w_3^\psi\,x;$
$\quad \mathbf{case}\ x\ (succ \Rightarrow z.succ; \mathbf{send}\,v^\psi\,z;\ //\ x{:}\mathsf{pin} \otimes \mathsf{ver} \vdash z{:}\mathsf{auth_{in}}$
$\qquad\qquad\qquad (z^\psi{:}\mathsf{auth_{in}} \leftarrow \mathsf{Auth_{in}}@\psi \leftarrow x{:}\mathsf{pin} \otimes \mathsf{ver})$
$\qquad\quad |\ fail \Rightarrow w_4^\psi \leftarrow \mathbf{recv}\,x; \mathbf{send}\,w_4^\psi\,z;\ //\ x{:}\mathsf{ver}, v{:}\mathsf{account} \vdash z{:}\mathsf{auth_{out}}$
$\qquad\qquad\qquad (z^\psi{:}\mathsf{auth_{out}} \leftarrow \mathsf{Auth_{out}}@\psi \leftarrow x{:}\mathsf{ver}, v{:}\mathsf{account})))@\psi'$

$\Psi_1; x{:}\mathsf{pin} \otimes \mathsf{ver}[\psi] \vdash \mathsf{Auth_{in}}@\psi' :: z{:}\mathsf{auth_{in}}[\psi] \qquad \Psi_1 = \psi' \sqsubseteq \psi, \Psi_0$
$z{:}\mathsf{auth_{in}} \leftarrow \mathsf{Auth_{in}} \leftarrow x{:}\mathsf{pin} \otimes \mathsf{ver} = ($
$\quad v^\psi \leftarrow \mathbf{recv}\ z; w_5^\psi \leftarrow \mathbf{recv}\,x; \mathbf{send}\,w_5^\psi\,z;\ //\ x{:}\mathsf{ver}, v{:}\mathsf{account} \vdash z{:}\mathsf{auth_{out}}\}$
$\quad (z^\psi{:}\mathsf{auth_{out}} \leftarrow \mathsf{Auth_{out}}@\psi \leftarrow x{:}\mathsf{ver}, v{:}\mathsf{account})))@\psi'$

$\Psi_1; u{:}\mathsf{pin}[\psi], z{:}\mathsf{auth_{out}}[\psi] \vdash \mathsf{Customer}@\psi' :: y{:}\mathbf{1}[\psi] \qquad \Psi_1 = \psi' \sqsubseteq \psi, \Psi_0$
$y \leftarrow \mathsf{Customer} \leftarrow \cdot = (\ //\ u{:}\mathsf{pin}, z{:}\mathsf{auth_{out}} \vdash y{:}\mathsf{customer}$
$\mathbf{send}\,u^\psi\,z; \mathbf{case}\ z\ (succ \Rightarrow v^\psi \leftarrow \mathbf{recv}\ z;\ //\ v{:}\mathsf{account}, z{:}\mathsf{auth_{in}} \vdash y{:}\mathbf{1}$
$\quad\quad \mathbf{case}\ v\ (high \Rightarrow \mathbf{send}\,v^\psi\,z, w_6^\psi \leftarrow \mathbf{recv}\,z; (y^\psi{:}\mathbf{1} \leftarrow \mathsf{Customer}@\psi \leftarrow w_6{:}\mathsf{pin}, z{:}\mathsf{auth_{out}})$
$\quad\quad\quad |\ med \Rightarrow \mathbf{send}\,v^\psi\,z, w_6^\psi \leftarrow \mathbf{recv}\,z; (y^\psi{:}\mathbf{1} \leftarrow \mathsf{Customer}@\psi \leftarrow w_6{:}\mathsf{pin}, z{:}\mathsf{auth_{out}})$
$\quad\quad\quad |\ low \Rightarrow \mathbf{send}\,v^\psi\,z, w_6^\psi \leftarrow \mathbf{recv}\,z; (y^\psi{:}\mathbf{1} \leftarrow \mathsf{Customer}@\psi \leftarrow w_6{:}\mathsf{pin}, z{:}\mathsf{auth_{out}})$
$\quad\quad\quad |\ fail \Rightarrow w_6^\psi \leftarrow \mathbf{recv}\ z; (y^\psi{:}\mathbf{1} \leftarrow \mathsf{Customer}@\psi \leftarrow w_6{:}\mathsf{pin}, z{:}\mathsf{auth_{out}})))@\psi'$

The account provides a customer's balance by signaling the corresponding label along $v{:}\mathsf{account}$. For example, Alice's account sends label *high* to her after being authorized.

$\Psi_2; \cdot \vdash \mathsf{aAccount}@\psi' :: v{:}\mathsf{account}[\psi] \qquad \Psi_2 := \psi = \mathbf{alice}, \psi' \sqsubseteq \psi, \Psi_0$
$\quad v \leftarrow \mathsf{aAccount} \leftarrow \cdot = (v.high;$
$\qquad v^\psi{:}\mathsf{account} \leftarrow \mathsf{aAccount}@\psi' \leftarrow \cdot)@\psi'$

We repeat the definition of Alice's verifier and her pin from the main part of the paper here for convenience.

$\Psi_2; \cdot \vdash \mathsf{aVerifier}@\psi' :: x{:}\mathsf{ver}[\psi] \qquad \Psi_2 := \psi = \mathbf{alice}, \psi' \sqsubseteq \psi, \Psi_0$
$x \leftarrow \mathsf{aVerifier} \leftarrow \cdot = (z^\psi \leftarrow \mathbf{recv}\ x\ //\ z{:}\mathsf{pin}[\psi] \vdash x{:} \oplus \{succ : \mathsf{pin}[\psi] \otimes \mathsf{ver}, fail{:}\mathsf{pin}[\psi] \otimes \mathsf{ver}\}[\psi]$
$\quad \mathbf{case}\ z\ (tok_j \Rightarrow x.succ; \mathbf{send}\ z^\psi\ x; (x^\psi{:}\mathsf{ver} \leftarrow \mathsf{aVerifier}@\psi \leftarrow \cdot)$
$\qquad |\ tok_{i \neq j} \Rightarrow x.fail; \mathbf{send}\ z^\psi\ x; (x^\psi{:}\mathsf{ver} \leftarrow \mathsf{aVerifier}@\psi \leftarrow \cdot)))@\psi'$

$\Psi_2; \cdot \vdash \mathsf{apin}@\psi' :: u{:}\mathsf{pin}[\psi] \qquad \Psi_2 := \psi = \mathbf{alice}, \psi' \sqsubseteq \psi, \Psi_0$
$u \leftarrow \mathsf{apin} \leftarrow \cdot = (u.tok_j; (u^\psi{:}\mathsf{pin} \leftarrow \mathsf{apin}@\psi' \leftarrow \cdot)))@\psi'$

We provide the extended lattice of each process variable consisting of its essential constraints next to its definition. We leave it to the reader to check that in each spawn, the caller can assert the extended lattice of their callee.

In contrast to Alice's verifier, her PIN, and her account, the customer and authorization processes do not have any specific information about Alice, e.g., her correct PIN, hard-coded in their code. Thus, we can define them as generic processes that can be spawned for any customer. Their caller only needs to instantiate their secrecy variables with the customer's correct secrecy level when spawned for a specific customer.

# B   A secure typing system

**External and Internal Choice**

$$\frac{\Psi; \Delta \vdash P@d_1 :: y{:}A_k[c] \quad k \in L}{\Psi; \Delta \vdash (y^c.k; P)@d_1 :: y{:} \oplus \{\ell{:}A_\ell\}_{\ell \in L}[c]} \oplus R$$

$$\frac{\Psi \Vdash d_2 = c \sqcup d_1 \quad \Psi; \Delta, x{:}A_k[c] \vdash Q_k@d_2 :: y{:}C[c'] \quad \forall k \in L}{\Psi; \Delta, x{:} \oplus \{\ell : A_\ell\}_{\ell \in L}[c] \vdash (\mathbf{case}\, x^c(\ell \Rightarrow Q_\ell)_{\ell \in L})@d_1 :: y{:}C[c']} \oplus L$$

$$\frac{\Psi; \Delta \vdash Q_k@c :: y{:}A_k[c] \quad \forall k \in L}{\Psi; \Delta \vdash (\mathbf{case}\, y^c(\ell \Rightarrow Q_\ell)_{\ell \in I})@d_1 :: y{:}\&\{\ell : A_\ell\}_{\ell \in L}[c]} \&R$$

$$\frac{\Psi \Vdash d_1 \sqsubseteq c \quad \Psi; \Delta, x{:}A_k[c] \vdash P@d_1 :: y{:}C[c'] \quad k \in L}{\Psi; \Delta, x{:}\&\{\ell : A_\ell\}_{\ell \in I}[c] \vdash (x^c.k; P)@d_1 :: y{:}C[c']} \&L$$

**Channel input/output**

$$\frac{\Psi; \Delta \vdash P@d_1 :: y{:}B[c]}{\Psi; \Delta, z{:}A[c] \vdash (\mathbf{send}\, z\, y^c; P)@d_1 :: y{:}A \otimes B[c]} \otimes R$$

$$\frac{\Psi \Vdash d_2 = c \sqcup d_1 \quad \Psi; \Delta, z{:}A[c], x{:}B[c] \vdash P@d_2 :: y{:}C[c']}{\Psi; \Delta, x{:}A \otimes B[c] \vdash (z \leftarrow \mathbf{recv}\, x^c; P)@d_1 :: y{:}C[c']} \otimes L$$

$$\frac{\Psi; \Delta, z{:}A[c] \vdash P@c :: y{:}B[c]}{\Psi; \Delta \vdash (z \leftarrow \mathbf{recv}\, y^c; P)@d_1 :: y{:}A \multimap B[c]} \multimap R$$

$$\frac{\Psi \Vdash d_1 \sqsubseteq c \quad \Psi; \Delta, x{:}B[c] \vdash P@d_1 :: y{:}C[c']}{\Psi; \Delta, z{:}A[c], x{:}A \multimap B[c] \vdash (\mathbf{send}\, z\, x^c; P)@d_1 :: y{:}C[c']} \multimap L$$

**Identity and spawn**

$$\frac{}{\Psi; x : T[d] \vdash_\Sigma (y^d \leftarrow x^d)@c :: y : T[d]} \text{FWD}$$

$$\frac{\begin{array}{c}\hat{\gamma}(\Delta'_1; \psi; \psi') = \Delta_1; d'; d_2 \quad \Psi \Vdash \hat{\gamma}(\Psi') \quad \Psi \Vdash d_1 \sqsubseteq d_2, \Psi \Vdash d' \sqsubseteq d \\ \Psi'; \Delta'_1 \vdash X = P@\psi' :: x : B[\psi] \in \Sigma \quad \Psi; x : B[d'], \Delta_2 \vdash_\Sigma Q@d_1 :: y : T[d]\end{array}}{\Psi; \Delta_1, \Delta_2 \vdash_\Sigma ((x^{d'} \leftarrow X \leftarrow \Delta_1@d_2); Q)@d_1 :: y : T[d]} \text{Spawn}$$

**Termination Rules**

$$\frac{}{\Psi; \cdot \vdash_\Sigma (\mathbf{close}\, y^c)@d_1 :: y : 1[c]} 1R \qquad \frac{\Psi; \Delta \vdash_\Sigma Q@d_1 :: y : T[d]}{\Psi; \Delta, x : 1[c] \vdash_\Sigma (\mathbf{wait}\, x^c; Q)@d_1 :: y : T[d]} 1L$$

**Signature checking**

$$\frac{\begin{array}{c}\mathtt{concrete}(\Psi) = \Psi_0 \quad \Psi \Vdash \psi' \sqsubseteq \psi, \forall i.\Psi \Vdash \psi_i \sqsubseteq \psi \\ \Psi; x_1{:}A_1[\psi_1], \cdots, x_n{:}A_n[\psi_n] \vdash_\Sigma P@\psi' :: x : B[\psi] \quad \Vdash_\Sigma \Sigma'\ \mathbf{sig}\end{array}}{\Vdash_{\Sigma, \Psi_0} \Psi; x_1{:}A_1[\psi_1], \cdots, x_n{:}A_n[\psi_n] \vdash_\Sigma X = P@\psi' :: x : B[\psi], \Sigma'\ \mathbf{sig}}$$

$$\frac{}{\Vdash_\Sigma (\cdot)\ \mathbf{sig}}$$

5

**Configuration Typing**

$$\overline{\Psi_0; x{:}A[d] \Vdash \cdot :: (x{:}A[d])} \ \mathbf{emp}_1 \qquad \overline{\Psi_0; \cdot \Vdash \cdot :: (\cdot)} \ \mathbf{emp}_2$$

$$\frac{\Psi_0 \Vdash d_1 \sqsubseteq d \qquad \forall y{:}B[d'] \in \Delta'_0, \Delta \ (\Psi_0 \Vdash d' \sqsubseteq d)}{\Psi_0; \Delta_0, \Delta'_0 \Vdash \mathcal{C}, \mathbf{proc}(x[d], P@d_1) :: (x{:}A[d])} \ \mathbf{proc}$$

$$\frac{\forall y{:}B[d'] \in \Delta'_0, \Delta \ (\Psi_0 \Vdash d' \sqsubseteq d) \qquad \Psi_0; \Delta_0 \Vdash \mathcal{C} :: \Delta \qquad \Psi_0; \Delta'_0, \Delta \vdash P@d :: (x{:}A[d])}{\Psi_0; \Delta_0, \Delta'_0 \Vdash \mathcal{C}, \mathbf{msg}(P) :: (x{:}A[d])} \ \mathbf{msg}$$

$$\frac{\Psi_0; \Delta_0 \Vdash \mathcal{C}_1 :: y{:}B[d_1] \qquad \Psi_0; \Delta'_0, z{:}B[d_1] \Vdash \mathcal{C}_2 :: x{:}A[d]}{\Psi_0; \Delta_0, \Delta'_0 \Vdash \mathcal{C}_1, [y/z], \mathcal{C}_2 :: x{:}A[d]} \ \mathbf{sub}_1$$

$$\frac{\Psi_0; \Delta_0 \Vdash \mathcal{C} :: \Delta \qquad \Psi_0; \Delta'_0 \Vdash \mathcal{C}_1 :: x{:}A[d]}{\Psi_0; \Delta_0, \Delta'_0 \Vdash \mathcal{C}, \mathcal{C}_1 :: \Delta, x{:}A[d]} \ \mathbf{comp}$$

Since $\Psi_0$ and $\Sigma$ are fixed, we may drop $\Psi_0$ and $\Sigma$ in a configuration typing judgment and a process typing judgment, respectively, for brevity.

**Definition B.1** (Poised Configuration). *A configuration* $\Delta_1, \Delta_2 \Vdash \mathcal{C}_1, \mathcal{C}_2 :: \Lambda, w{:}A'[c]$ *is poised iff* $\Delta_1 \Vdash \mathcal{C}_1 :: \Lambda$ *is poised, and* $\Delta_2 \Vdash \mathcal{C}_2 :: w{:}A'[c]$ *is poised.*

*The configuration* $\Delta_2 \Vdash \mathcal{C}_2 :: w{:}A'[c]$ *is poised iff either* $\mathcal{C}_2$ *is empty, or at least one of the following conditions hold:*

1. $\mathcal{C}_2 = \mathcal{C}'_2 \mathbf{msg}(P) \mathcal{C}''_2$ *such that* $\mathbf{msg}(P)$ *is a negative message along* $y{:}C[d'] \in \Delta_2$, *i.e.*

    $y{:}\&\{\ell{:}A_\ell\}_{\ell \in L}[c_1] \Vdash \mathbf{msg}(P) :: x{:}A_k[c_1]$, *and* $y{:}C[d'] = y{:}\&\{\ell{:}A_\ell\}_{\ell \in L}[c_1]$ *or*
    $y{:}(A \multimap B)[d_1], z{:}A[d_1] \Vdash \mathbf{msg}(P) :: x{:}B[c_1]$ *and* $y{:}C[d'] = y{:}(A \multimap B)[d_1]$.

2. $\mathcal{C}_2 = \mathbf{proc}(x[c'], P@d_1)\mathcal{C}'_2$ *such that* $\mathbf{proc}(x[c'], P@d_1)$ *attempts to receive along* $y{:}C[d'] \in \Delta_2$, *i.e.*

    $\Delta'_2, y{:} \oplus \{\ell{:}A_\ell\}_{\ell \in L}[c_1] \Vdash \mathbf{proc}(x[c'], \mathbf{case}\, y^{c_1}(\ell \Rightarrow P'_\ell)@d_1) :: x{:}T[c']$, *or*
    $\Delta'_2, y{:}(A \otimes B)[c_1] \Vdash \mathbf{proc}(x[c'], v^{c_1} \leftarrow \mathbf{recv}\, y; P'@d_1) :: x{:}T[c']$, *or*
    $\Delta'_2, y{:}\mathbf{1}[c] \Vdash \mathbf{proc}(x[c'], \mathbf{wait}\, y^c; P'@d_1) :: x{:}T[c']$.

3. $\mathcal{C}_2 = \mathcal{C}'_2 \mathbf{msg}(P)$ *such that* $\mathbf{msg}(P)$ *is a positive message along* $y{:}C[d'] = w{:}A'[c]$, *i.e.*

    $x{:}A_k[c] \Vdash \mathbf{msg}(P) :: w{:} \oplus \{\ell{:}A_\ell\}_{\ell \in L}[c]$, *or*
    $x{:}B[c], z{:}A[c] \Vdash \mathbf{msg}(P) :: w{:}(A \otimes B)[c]$, *or*
    $\cdot \Vdash \mathbf{msg}(P) :: w{:}\mathbf{1}[c]$.

4. $\mathcal{C}_2 = \mathbf{proc}(w[c], P@d_1)\mathcal{C}'_2$ *such that* $\mathbf{proc}(w[c], P@d_1)$ *attempts to receive along* $y{:}C[d'] = w{:}A'[c]$, *i.e.*

    $\Lambda \Vdash \mathbf{proc}(w[c], \mathbf{case}\, w^c(\ell \Rightarrow P'_\ell)@d_1) :: w{:}\&\{\ell{:}A_\ell\}_{\ell \in L}[c]$, *or*
    $\Lambda \Vdash \mathbf{proc}(w[c], v^c \leftarrow \mathbf{recv}\, w^c; P'@d_1) :: w{:}(A \multimap B)[c]$.

**Definition B.2.** *The configuration* $\Delta \Vdash \mathcal{C} :: \Delta'$ *is ready to send along* $\Lambda \subseteq \Delta, \Delta'$ *if* $y{:}C[d'] \in \Lambda$ *iff there is either a message* $\mathbf{msg}(P)$ *in* $\mathcal{C}$ *sending along* $y$.

**Definition B.3.** *Here is the definition of some notations that we use in the logical relation.*

- *We define* $(\mathcal{C}_1, \mathcal{D}_1, \mathcal{F}_1; \mathcal{C}_2, \mathcal{D}_2, \mathcal{F}_2) \in \mathsf{Tree}_\Psi(\Delta \Vdash K)$ *as*

    $\Psi_0; \cdot \Vdash \mathcal{C}_1 :: \Delta$, and $\Psi_0; \cdot \Vdash \mathcal{C}_2 :: \Delta$, and $\Psi_0; \Delta \Vdash \mathcal{D}_1 :: K$ and $\Psi_0; \Delta \Vdash \mathcal{D}_2 :: K$ and $\Psi_0; K \Vdash \mathcal{F}_1 :: \cdot$ and $\Psi_0; K \Vdash \mathcal{F}_2 :: \cdot$

- *We define* $(\mathcal{C}, \mathcal{D}, \mathcal{F}) \in \mathsf{Tree}_{\Psi_0}(\Delta \Vdash K)$ *as*

    $\Psi_0; \cdot \Vdash \mathcal{C}_1 :: \Delta$, and $\Psi_0; \Delta \Vdash \mathcal{D}_1 :: K$ and $\Psi_0; K \Vdash \mathcal{F}_1 :: \cdot$

- *The notation* $\mapsto^*_{\Delta \Vdash K}$ *refers to taking none or many steps with* $\mapsto_{\Delta \Vdash K}$. *The notation* $\mapsto^j_{\Delta \Vdash K}$ *refers to taking $j$ steps with* $\mapsto_{\Delta \Vdash K}$.

- The configuration $\mathcal{C}', \mathcal{D}', \mathcal{F}'$ is ready to send along $\Theta_1, \Theta_2; \Upsilon$ iff $\mathcal{C}'$ is ready to send along $\Theta_1$, $\mathcal{D}'$ is ready to send along $\Upsilon$, and $\mathcal{F}'$ is ready to send along $\Theta_3$.
- We write $\mathcal{C}, \mathcal{D}, \mathcal{F} \mapsto_{\Delta \Vdash K}^{*\Theta;\Upsilon} \mathcal{C}', \mathcal{D}', \mathcal{F}'$ stating that $\mathcal{C}, \mathcal{D}, \mathcal{F} \mapsto_{\Delta \Vdash K}^{*} \mathcal{C}', \mathcal{D}', \mathcal{F}'$ and $\mathcal{C}'_i, \mathcal{D}'_i, \mathcal{F}'_i$ is ready to send along $\Theta; \Upsilon$.

**Definition B.4.** *Let $\Psi_0 = \langle \mathcal{L}, \sqsubseteq, \sqcup, \sqcap \rangle$ be a lattice formed over (concrete) security levels $\mathcal{L}$. We define an extended lattice $\Psi = \langle \Psi_0, \mathcal{V}, \mathcal{E} \rangle$ over $\Psi_0$, where*

- *$\mathcal{V}$ is a set of security variables disjoint from the (concrete) levels $\mathcal{L}$ in $\Psi_0$.*
- *the elements $\ell$ of the extended lattice are defined inductively as*

$$\ell := \ell \sqcup \ell \mid \ell \sqcap \ell \mid d \in \mathcal{L} \cup \mathcal{V}$$

- *$\mathcal{E}$ is a set of equations formed over the elements $\ell$ of the extended lattice with the following grammar:*

$$\mathcal{E} := \mathcal{E}', \ell \sqsubseteq \ell \mid \cdot$$

*We write $\Psi \Vdash E$ if $E$ can be deduced from (reflexive and transitive closure of) the union of the underlying lattice structure ($\Psi_0$) and the equations $\mathcal{E}$ in $\Psi$:*

$$\Psi_0 \cup \mathcal{E} \Vdash E.$$

**Definition B.5.** *A substitution $\gamma$ is a mapping that maps a secrecy variable $v_i$ to a secrecy term $d_i$ ($d_i$ is either a concrete security level in $\Psi_0$ or a security variable). It is straightforward to extend the substitution $\gamma$ to an extended lattice $\Psi$, and process terms, such that we replace simultaneously all occurrences of every $v_i$, $1 \leq i \leq n$ in the structure with the corresponding term $d_i$. We call the extended substitution $\hat{\gamma}$, and leave its definition for each structure to the reader.*

*We only consider order preserving substitutions, meaning that for every extended lattice $\Psi$ and every equation $E$, if $\Psi \Vdash E$, then $\hat{\gamma}(\Psi) \Vdash \hat{\gamma}(E)$.*

*We put $\gamma \circ \delta$ for the composition of two substitutions $\gamma$ and $\delta$ which is defined as usual. It is straightforward to observe that if both $\gamma$ and $\delta$ are order preserving so is their composition.*

*We write $\texttt{concrete}(\Psi) = \Psi_0$ to express that for every concrete secrecy level $c, d \in \Psi_0$, we have $\Psi \Vdash c \sqsubseteq d$ iff $\Psi_0 \Vdash c \sqsubseteq d$. It is straightforward to observe that $\texttt{concrete}(\Psi) = \Psi_0$ implies $\hat{\gamma}(\Psi) \Vdash \Psi_0$ for every substitution $\gamma$.*

# C Detailed proofs

## C.1 Standard properties of Session types

**Lemma C.1** (Substitution). *If $\Psi; \Delta \Vdash_\Sigma P@c :: x : B[d]$ and it satisfies the tree invariants, then for every substitution $\gamma$ that unifies $\Psi$ with respect to $\Psi_0$ (i.e., $\Psi_0 \Vdash \hat{\gamma}(\Psi)$), we have*

$$\hat{\gamma}(\Psi); \hat{\gamma}(\Delta) \Vdash_\Sigma \hat{\gamma}(P)@\hat{\gamma}(c) :: x{:}\hat{\gamma}(B)[\hat{\gamma}(d)]),$$

*that also satisfies the tree invariants.*

*Proof.* The proof is by induction on the structure of typing inference for $\Psi; \Delta \Vdash_\Sigma P@c :: x : B[d]$. We consider cases for the last step in the inference.

There are two subtle observations that are common to all cases: (1) By the nature of presupposition, if in the last step of the inference the conclusion satisfies the tree invariant, so do the premises. (2) If $\Psi$ in the original sequent is a lattice that agrees with $\Psi_0$ on the concrete elements and the functions/relations between them, $\hat{\gamma}(\Psi)$ has the same property. The second observation is the straightforward consequence of the assumption $\Psi_0 \Vdash \hat{\gamma}(\Psi)$.

**Case 1.**

$$\frac{\Psi; \Delta \vdash_\Sigma Q_k@c :: y : A_k[c] \quad \forall k \in I}{\Psi; \Delta \vdash_\Sigma (\mathbf{case}\, y^c(\ell \Rightarrow Q_\ell)_{\ell \in I})@d_1 :: y : \&\{\ell : A_\ell\}_{\ell \in I}[c]} \ \&R$$

By the induction hypothesis, for all $k \in I$, we have

$$\hat{\gamma}(\Psi); \hat{\gamma}(\Delta) \vdash_\Sigma \hat{\gamma}(Q_k)@\hat{\gamma}(c) :: y{:}\hat{\gamma}(A_k)[\hat{\gamma}(c)],$$

which preserves the invariants.

By the $\&R$ rule, we get

$$\hat{\gamma}(\Psi); \hat{\gamma}(\Delta) \vdash_\Sigma \mathbf{case}\, y^{\hat{\gamma}(c)}(\ell \Rightarrow \hat{\gamma}(Q_\ell))_{\ell \in I} @ \hat{\gamma}(d_1) :: y : \&\{\ell : \hat{\gamma}(A_\ell)\}_{\ell \in I}[\hat{\gamma}(c)].$$

By the the first part of the tree invariant for the original sequent, we get $\Psi \Vdash d_1 \sqsubseteq c$, and thus $\hat{\gamma}(\Psi) \Vdash \hat{\gamma}(d_1) \sqsubseteq \hat{\gamma}(c)$. The second part of the tree invariant is guaranteed by the induction hypothesis on the premise.
By a simple rewrite we get

$$\hat{\gamma}(\Psi); \hat{\gamma}(\Delta) \vdash_\Sigma \hat{\gamma}(\mathbf{case}\, y^c(\ell \Rightarrow Q_\ell)_{\ell \in I}) @ \hat{\gamma}(d_1) :: y{:}\hat{\gamma}(\&\{\ell : A_\ell\}_{\ell \in I})[\hat{\gamma}(c)].$$

**Case 2.**

$$\frac{\Psi \Vdash d_1 \sqsubseteq c \quad \Psi; \Delta, x : A_k[c] \vdash_\Sigma P @ d_1 :: y : T[d] \quad k \in I}{\Psi; \Delta, x : \&\{\ell : A_\ell\}_{\ell \in I}[c] \vdash_\Sigma (x^c.k; P) @ d_1 :: y : T[d]} \&L$$

By the induction hypothesis on the first premise we have

$$\hat{\gamma}(\Psi); \hat{\gamma}(\Delta), x{:}\hat{\gamma}(A_k[c]) \vdash_\Sigma \hat{\gamma}(P) @ \hat{\gamma}(d_1) :: y{:}\hat{\gamma}(T)[\hat{\gamma}(d)],$$

which preserves the tree invariants. By applying substitution on the second premise ($\Psi \Vdash d_1 \sqsubseteq c$) we get $\hat{\gamma}(\Psi) \Vdash \hat{\gamma}(d_1) \sqsubseteq \hat{\gamma}(c)$.
By the $\&L$ rule, we get

$$\hat{\gamma}(\Psi); \hat{\gamma}(\Delta), x{:}\&\{\hat{\gamma}(A_\ell)\}_{\ell \in I}[\hat{\gamma}(c)] \vdash_\Sigma x^{\hat{\gamma}(c)}.k; \hat{\gamma}(P) @ \hat{\gamma}(d_1) :: y : \hat{\gamma}(T)[\hat{\gamma}(d)],$$

which satisfies the tree invariants since the premise satisfies them.
Again by a simple rewrite we get

$$\hat{\gamma}(\Psi); \hat{\gamma}(\Delta), x{:}\hat{\gamma}(\&\{(A_\ell)_{\ell \in I}[(c)]) \vdash_\Sigma \hat{\gamma}(x^c.k; P) @ \hat{\gamma}(d_1) :: y : \hat{\gamma}(T)[\hat{\gamma}(d)].$$

**Case 3.**

$$\frac{\Psi; \Delta \vdash_\Sigma P @ c :: y : B[d]}{\Psi; \Delta, z : A[d] \vdash_\Sigma (\mathbf{send}\, z^d\, y; P) @ c :: y : (A[d] \otimes B)\,[d]} \otimes R$$

By the induction hypothesis on the premise we have

$$\hat{\gamma}(\Psi); \hat{\gamma}(\Delta) \vdash_\Sigma \hat{\gamma}(P) @ \hat{\gamma}(c) :: y{:}\hat{\gamma}(B)[\hat{\gamma}(d)],$$

which preserves the tree invariants.
By the $\otimes R$ rule, we get

$$\hat{\gamma}(\Psi); \hat{\gamma}(\Delta), z{:}\hat{\gamma}(A)[\hat{\gamma}(d)] \vdash_\Sigma \mathbf{send}\, z^{\hat{\gamma}(d)}\, y; \hat{\gamma}(P) @ \hat{\gamma}(c) :: y : (\hat{\gamma}(A) \otimes \hat{\gamma}(B))[\hat{\gamma}(d)],$$

the resulting sequent satisfies the tree invariant.

**Case 4.**

$$\frac{\Psi \Vdash c' = c \sqcup d \quad \Psi; \Delta, z : A[d], x : B[d] \vdash_\Sigma P @ c' :: y : T[d_1]}{\Psi; \Delta, x : (A \otimes B)\,[d] \vdash_\Sigma (z^d \leftarrow \mathbf{recv}\, x; P) @ c :: y : T[d_1]} \otimes L$$

By the induction hypothesis on the first premise we have By the induction hypothesis on the premise we have

$$\hat{\gamma}(\Psi); \hat{\gamma}(\Delta), z{:}\hat{\gamma}(A)[\hat{\gamma}(d)], x{:}\hat{\gamma}(B)[\hat{\gamma}(d)] \vdash_\Sigma \hat{\gamma}(P) @ \hat{\gamma}(c') :: y{:}\hat{\gamma}(T)[\hat{\gamma}(d_1)],$$

which preserves the tree invariants. Moreover, by applying the substitution on the second and sthird premises, we get $\hat{\gamma}(\Psi) \Vdash \hat{\gamma}(c') = \hat{\gamma}(c) \sqcup \hat{\gamma}(d)$.
By the $\otimes L$ rule, we get

$$\hat{\gamma}(\Psi); \hat{\gamma}(\Delta), x{:}(\hat{\gamma}(A) \otimes \hat{\gamma}(B))[\hat{\gamma}(d)] \vdash_\Sigma z^{\hat{\gamma}(d)} \leftarrow \mathbf{recv}\, x; \hat{\gamma}(P) @ \hat{\gamma}(c) :: y{:}\hat{\gamma}(T)[\hat{\gamma}(d_1)],$$

It is straightforward to observe that the resulting sequent satisfies the tree invariant.

| | |
|---|---|
| FWD | $\mathbf{proc}(y_\alpha[c], (y_\alpha \leftarrow x_\beta)@d_1) \mapsto_{\Delta \Vdash \Delta'} [x_\beta/y_\alpha]$ |
| SPAWN | $\mathbf{proc}(y_\alpha[c], (x^d \leftarrow X \leftarrow \Delta)@d_2; Q@d_1) \mapsto_{\Delta \Vdash \Delta'}$     $(\Psi'; \Delta' \vdash X = P@\psi' :: x : B'[\psi] \in \Sigma)$ |
| | $\mathbf{proc}(x_0[d], ([x_0, \Delta/x, \mathrm{dom}(\Delta')] \hat{\gamma}(P))@d_2) \, \mathbf{proc}(y_\alpha[c], ([x_0/x] \, Q)@d_1)$    $(\hat{\gamma}(\Psi') = \Psi_0 \text{ and } x_0 \text{ fresh})$ |
| $1_{\mathsf{snd}}$ | $\mathbf{proc}(y_\alpha[c], (\mathbf{close}\, y_\alpha)@d_1) \mapsto_{\Delta \Vdash \Delta'} \mathbf{msg}(\mathbf{close}\, y_\alpha)$ |
| $1_{\mathsf{rcv}}$ | $\mathbf{msg}(\mathbf{close}\, y_\alpha) \, \mathbf{proc}(x_\beta[c'], (\mathbf{wait}\, y_\alpha; Q)@d_1) \mapsto_{\Delta \Vdash \Delta'} \mathbf{proc}(x_\beta[c'], Q@(d_1 \sqcup c))$ |
| $\oplus_{\mathsf{snd}}$ | $\mathbf{proc}(y_\alpha[c], y_\alpha.k; P@d_1) \mapsto_{\Delta \Vdash \Delta'} \mathbf{proc}(y_{\alpha+1}[c], ([y_{\alpha+1}/y_\alpha]P)@d_1) \, \mathbf{msg}(y_\alpha.k; y_\alpha \leftarrow y_{\alpha+1})$ |
| $\oplus_{\mathsf{rcv}}$ | $\mathbf{msg}(y_\alpha[c].k; y_\alpha \leftarrow v_\delta)) \, \mathbf{proc}(u_\gamma[c'], \mathbf{case}\, y_\alpha((\ell \Rightarrow P_\ell)_{\ell \in L})@d_1) \mapsto_{\Delta \Vdash \Delta'} \mathbf{proc}(u_\gamma[c'], ([v_\delta/y_\alpha]P_k)@(d_1 \sqcup c))$ |
| $\&_{\mathsf{snd}}$ | $\mathbf{proc}(y_\alpha[c], (x_\beta.k; P)@d_1) \mapsto_{\Delta \Vdash \Delta'} \mathbf{msg}(x_\beta.k; x_{\beta+1} \leftarrow x_\beta) \, \mathbf{proc}(y_\alpha[c], ([x_{\beta+1}/x_\beta]P)@d_1)$ |
| $\&_{\mathsf{rcv}}$ | $\mathbf{proc}(y_\alpha[c], (\mathbf{case}\, y_\alpha(\ell \Rightarrow P_\ell)_{\ell \in L})@d_1) \, \mathbf{msg}(y_\alpha.k; v_\delta \leftarrow y_\alpha) \mapsto_{\Delta \Vdash \Delta'} \mathbf{proc}(v_\delta[c], ([v_\delta/y_\alpha]P_k)@c)$ |
| $\otimes_{\mathsf{snd}}$ | $\mathbf{proc}(y_\alpha[c], (\mathbf{send}\, x_\beta\, y_\alpha; P)@d_1) \mapsto_{\Delta \Vdash \Delta'} \mathbf{proc}(y_{\alpha+1}[c], ([y_{\alpha+1}/y_\alpha]P)@d_1) \, \mathbf{msg}(\mathbf{send}\, x_\beta\, y_\alpha; y_\alpha \leftarrow y_{\alpha+1})$ |
| $\otimes_{\mathsf{rcv}}$ | $\mathbf{msg}(\mathbf{send}\, x_\beta\, y_\alpha; y_\alpha \leftarrow v_\delta) \, \mathbf{proc}(u_\gamma[c'], (w_\eta \leftarrow \mathbf{recv}\, y_\alpha; P)@d_1) \mapsto_{\Delta \Vdash \Delta'} \mathbf{proc}(u_\gamma[c'], ([x_\beta/w_\eta][v_\delta/y_\alpha]P)@(d_1 \sqcup c))$ |
| $\multimap_{\mathsf{snd}}$ | $\mathbf{proc}(y_\alpha[c], (\mathbf{send}\, x_\beta\, u_\gamma; P)@d_1) \mapsto_{\Delta \Vdash \Delta'} \mathbf{msg}(\mathbf{send}\, x_\beta\, u_\gamma; u_{\gamma+1} \leftarrow u_\gamma) \, \mathbf{proc}(y_\alpha[c], ([u_{\gamma+1}/u_\gamma]P)@d_1)$ |
| $\multimap_{\mathsf{rcv}}$ | $\mathbf{proc}(y_\alpha[c], (w_\eta \leftarrow \mathbf{recv}\, y_\alpha; P)@d_1) \, \mathbf{msg}(\mathbf{send}\, x_\beta\, y_\alpha; v_\delta \leftarrow y_\alpha) \mapsto_{\Delta \Vdash \Delta'} \mathbf{proc}(v_\delta[c], ([x_\beta/w_\eta][v_\delta/y_\alpha]P)@c)$ |

Figure 1: Asynchronous dynamics of $\mathsf{SILL}_{\mathsf{sec}+}$.

**Case 5.**
$$\frac{}{\Psi; x : T[d] \vdash_\Sigma (y^d \leftarrow x^d)@c :: y : T[d]} \text{ FWD}$$

The following sequent is the conclusion of a FWD rule:

$$\hat{\gamma}(\Psi); x{:}\hat{\gamma}(T)[\hat{\gamma}(d)] \vdash_\Sigma (y^{\hat{\gamma}(d)} \leftarrow x^{\hat{\gamma}(d)})@\hat{\gamma}(c) :: y{:}\hat{\gamma}(T)[\hat{\gamma}(d)]$$

It is straightforward to observe that this sequent satisfies the tree invariant. In particular, by applying substitution on the assumption $\Psi \Vdash c \sqsubseteq d$, we get $\hat{\gamma}(\Psi) \Vdash \hat{\gamma}(c) \sqsubseteq \hat{\gamma}(d)$.

**Case 6.** Here is the interesting case!

$$\frac{\hat{\delta}(\Delta'_1; B[\psi]; \psi') = \Delta_1; B[d']; d_2 \quad \Psi \Vdash \hat{\delta}(\Psi') \quad \Psi \Vdash d_1 \sqsubseteq d_2 \quad d' \sqsubseteq d \quad \Psi'; \Delta'_1 \vdash X = P@\psi' :: x : B[\psi] \in \Sigma \quad \Psi; x : B[d'], \Delta_2 \vdash_\Sigma Q@d_1 :: y : T[d]}{\Psi; \Delta_1, \Delta_2 \vdash_\Sigma ((x^{d'} \leftarrow X \leftarrow \Delta_1@d_2); Q)@d_1 :: y : T[d]} \text{ SPAWN}$$

By the induction hypothesis we have

$$\star \hat{\gamma}(\Psi); x : \hat{\gamma}(B[d']), \hat{\gamma}(\Delta_2) \vdash_\Sigma \hat{\gamma}(Q)@\hat{\gamma}(d_1) :: y{:}\hat{\gamma}(T)[\hat{\gamma}(d)],$$

which preserves the tree invariants. By applying $\gamma$ on the 5th and 6th premises, we get

$$\hat{\gamma}(\Psi) \Vdash \hat{\gamma}(d_1) \sqsubseteq \hat{\gamma}(d_2), \hat{\gamma}(d') \sqsubseteq \hat{\gamma}(d).$$

Again by applying the substitution $\gamma$ on the third premise $\Psi \Vdash \hat{\delta}(\Psi')$, we get $\hat{\gamma}(\Psi) \Vdash \hat{\gamma}(\hat{\delta}(\Psi'))$. And similarly, we have $\hat{\gamma}(\hat{\delta}(\Delta'_1; B[\psi]; \psi')) = \hat{\gamma}(\Delta_1; B[d']; d_2)$ from the second premise.
By the SPAWN rule for the substitution $\gamma \circ \delta$, we get

$$\hat{\gamma}(\Psi); \hat{\gamma}(\Delta_1), \hat{\gamma}(\Delta_2) \vdash_\Sigma ((x^{\hat{\gamma}(d')} \leftarrow X \leftarrow \hat{\gamma}(\Delta_1)@\hat{\gamma}(d_2)); \hat{\gamma}(Q))@\hat{\gamma}(d_1) :: y : \hat{\gamma}(T)[\hat{\gamma}(d)].$$

To show that this sequent satisfies the tree invariant, we use (a) the fact that the sequent $\star$ satisfies the tree invariant and (b) the condition on process definitions in the signature. The condition we imposed on process definitions ensures that $\Psi' \Vdash \psi' \sqsubseteq \psi$ and $\forall y{:}A[\psi_i] \in \Delta'_1. \Psi' \Vdash \psi_i \sqsubseteq \psi$. By applying substitution on these judgments, we get

$$\forall y{:}\gamma \,\hat{\circ}\, \delta(A')[\gamma \,\hat{\circ}\, \delta(\psi_i)] \in \gamma \,\hat{\circ}\, \delta(\Delta'_1). \; \gamma \,\hat{\circ}\, \delta(\Psi') \Vdash \gamma \,\hat{\circ}\, \delta(\psi_i) \sqsubseteq \gamma \,\hat{\circ}\, \delta(\psi).$$

By $\hat{\gamma}(\Psi) \Vdash \hat{\gamma}(\hat{\delta}(\Psi'))$, and $\hat{\gamma}(\hat{\delta}(\Delta'_1; B'[\psi]; \psi')) = \hat{\gamma}(\Delta_1; B[d']; d_2)$, and $\hat{\gamma}(\Psi) \Vdash \hat{\gamma}(d') \sqsubseteq \hat{\gamma}(d)$, we can rewrite the above judgment as

$$\forall y{:}\hat{\gamma}(A)[\hat{\gamma}(\psi_i)] \in \hat{\gamma}(\Delta_1). \; \hat{\gamma}(\Psi) \Vdash \hat{\gamma}(\psi_i) \sqsubseteq \hat{\gamma}(d') \sqsubseteq \hat{\gamma}(d).$$

$\square$

**Lemma C.2.** *If $\Psi_0; \Delta \Vdash \mathcal{C}\,\mathcal{C}' :: \Delta'$, then for some $\Delta_1$ we have $\Psi_0; \Lambda_1 \Vdash \mathcal{C} :: \Lambda'_1, \Delta_1$ and $\Psi_0; \Lambda_2, \Delta_1 \Vdash \mathcal{C}' :: \Lambda'_2$, where $\Delta = \Lambda_1, \Lambda_2$ and $\Delta' = \Lambda'_1, \Lambda'_2$.*

*Proof.* The proof is by a straightforward induction on the configuration typing rules. □

**Theorem C.3** (Progress). *If $\Psi_0; \Delta \Vdash \mathcal{C} :: \Delta'$, then either $\mathcal{C} \mapsto_{\Delta \Vdash \Delta'} \mathcal{C}'$, or $\mathcal{C}$ is poised.*

*Proof.* The proof is by induction on the configuration typing of $\mathcal{C}$. If $\mathcal{C} \mapsto_{\Delta \Vdash \Delta'} \mathcal{C}'$ then the proof is complete. Otherwise we consider the last rule in the typing of $\mathcal{C}$.

**Case 1.**
$$\overline{\Psi_0; x{:}A[d] \Vdash \cdot :: (x{:}A[d])} \; \mathbf{emp}_1 \qquad \overline{\Psi_0; \cdot \Vdash \cdot :: (\cdot)} \; \mathbf{emp}_2$$

In this case $\mathcal{C}$ is empty and the proof is complete.

**Case 2.**
$$\frac{\Psi_0; \Delta_0 \Vdash \mathcal{C}' :: \Delta \quad \begin{array}{c} \Psi_0 \Vdash d_1 \sqsubseteq d \\ \forall (y{:}B[d']) \in \Delta_0', \Delta \quad (\Psi_0 \Vdash d' \sqsubseteq d) \\ \Psi_0; \Delta_0', \Delta \vdash P@d_1 :: (x{:}A[d]) \end{array}}{\Psi_0; \Delta_0, \Delta_0' \Vdash \mathcal{C}', \mathbf{proc}(x[d], P@d_1) :: (x{:}A[d])} \; \mathbf{proc}$$

If $\mathbf{proc}(x[d], P@d_1) :: (x{:}A[d])$ attempts to receive on $x$ or is forward then the proof is complete. Otherwise, if it wants to send or spawn along one of its channels $\mathcal{C}$ can take a step. It is straightforward to see that the substitution $\hat{\gamma}$ required by the dynamics is the same as the substitution we get by inversion on the Spawn rule. In particular, by inversion on the Spawn rule, we know that there exists a substitution $\hat{\gamma}$ such that $(\star)$ $\Psi_0 \Vdash \hat{\gamma}(\Psi')$. From $(\star)$ and $\mathsf{concrete}(\Psi') = \Psi_0$ which we get from the signature typing, we can show that $\hat{\gamma}(\Psi') = \Psi_0$.

It remains to consider the case in which $\mathcal{C}$ wants to receive along one of its positive resources $y{:}B[c] \in \Delta_0', \Delta$. If $y{:}B[c] \in \Delta_0'$, then the proof is complete by the definition of poised configurations. If $y{:}B[c] \in \Delta$, we apply the induction hypothesis on $\mathcal{C}'$. If $\mathcal{C}'$ can take a step, so does $\mathcal{C}$ and the proof is complete. If $\mathcal{C}'$ is poised, the subconfiguration $\Delta_0'' \Vdash \mathcal{C}'' :: y{:}B[c]$ for some $\Delta_0'' \subseteq \Delta_0$ is also poised. If $\mathcal{C}''$ is empty, then $y{:}B[c] \in \Delta_0'' \subseteq \Delta_0$ and by definition $\mathcal{C}$ is poised. Otherwise, by assumption $\mathcal{C}''$ cannot attempt to receive along $y{:}B[c]$ since $B$ is positive. If it offers a positive message along $y{:}B[c]$, then the proof is complete since $\mathcal{C}$ can take a step. If $\mathcal{C}'$ has a forwarding on the root, then $\mathcal{C}'$ can take a step again. In the other cases, poisedness of $\mathcal{C}$ follows by definition.

**Case 3.**
$$\frac{\Psi_0; \Delta_0 \Vdash \mathcal{C}' :: \Delta \quad \begin{array}{c} \forall (y{:}B[d']) \in \Delta_0', \Delta \quad (\Psi_0 \Vdash d' \sqsubseteq d) \\ \Psi_0; \Delta_0', \Delta \vdash P@d :: (x{:}A[d]) \end{array}}{\Psi_0; \Delta_0, \Delta_0' \Vdash \mathcal{C}', \mathbf{msg}(P) :: (x{:}A[d])} \; \mathbf{msg}$$

If $\mathbf{msg}(P)$ is a positive message or a negative message along $y_\alpha{:}A[c] \in \Delta_0'$ then the proof is complete by definition.

If $\mathbf{msg}(P)$ is a negative message along $y_\alpha{:}A[c] \in \Delta$ then we proceed the proof by induction on $\mathcal{C}'$: if it can take a step, so does $\mathcal{C}$. If $\mathcal{C}'$ is poised, the subconfiguration $\Delta_0'' \Vdash \mathcal{C}'' :: y_\alpha{:}A[c]$ for some $\Delta_0'' \subseteq \Delta_0$ is also poised. If $\mathcal{C}''$ is empty then by typing rules $y_\alpha{:}A[c] \in \Delta_0'' \subseteq \Delta_0$ and by definition $\mathcal{C}$ is poised. Otherwise, by assumption $\mathcal{C}''$ cannot offer a positive message along $y_\alpha{:}A[c]$. If it attempts to receive along $y_\alpha{:}A[c]$ then the proof is complete since $\mathcal{C}$ can take a step. If it has a forwarding on its root, then $\mathcal{C}$ can take a step again. In the other cases, poisedness of $\mathcal{C}$ follows by definition.

**Case 4.**
$$\frac{\Psi_0; \Delta_0 \Vdash \mathcal{C}' :: \Delta \quad \Psi_0; \Delta_0' \Vdash \mathcal{C}'' :: x{:}A[d]}{\Psi_0; \Delta_0, \Delta_0' \Vdash \mathcal{C}', \mathcal{C}'' :: \Delta, x{:}A[d]} \; \mathbf{comp}$$

where $\mathcal{C} = \mathcal{C}', \mathcal{C}''$. By induction hypothesis, either (i) $\mathcal{C}'$ can take a step or (ii) $\mathcal{C}'$ is empty or (iii) $\mathcal{C}'$ is poised. In (i) the proof is complete, since $\mathcal{C}'\mathcal{C}''$ also can take a step. In (ii) the proof is complete since $\mathcal{C} = \mathcal{C}''$ and we can apply the induction hypothesis on $\mathcal{C}''$. In (iii) we apply the induction hypothesis on $\mathcal{C}''$ and consider the cases: (i') $\mathcal{C}''$ can take a step which completes the proof, or (ii') $\mathcal{C}''$ is empty which again completes the proof, or (iii') $\mathcal{C}''$ is poised, which is enough to prove that $\mathcal{C}$ is poised and completes the proof.

□

**Theorem C.4** (Preservation). *If $\Psi_0; \Delta \Vdash \mathcal{C} :: \Delta'$ and $\mathcal{C} \mapsto_{\Delta \Vdash \Delta'} \mathcal{C}'$, then $\Psi_0; \Delta \Vdash \mathcal{C}' :: \Delta'$.*

*Proof.* The proof is by considering different cases of $\mathcal{C} \mapsto_{\Delta \Vdash \Delta'} \mathcal{C}'$. And then by inversion on the typing derivations.

**Case 1.** (fwd)
$$\mathcal{C}_1 \mathbf{proc}(y_\alpha[c], (y_\alpha \leftarrow x_\beta)@d_1) \mathcal{C}_2 \mapsto_{\Delta \Vdash \Delta'} \mathcal{C}_1[x_\beta/y_\alpha]\mathcal{C}_2$$

**By assumption of the theorem:** $\Psi_0; \Delta \Vdash \mathcal{C}_1 \mathbf{proc}(y_\alpha[c], (y_\alpha \leftarrow x_\beta)@d_1) \mathcal{C}_2 :: \Delta'$.
**By Lemma C.2 and inversion on proc and Fwd rules:** $\Psi_0; x_\beta{:}A[c] \Vdash \mathbf{proc}(y_\alpha[c], (y_\alpha \leftarrow x_\beta)@d_1) :: y_\alpha{:}A[c]$.
If $x_\beta{:}A[c] \in \Delta$, then $\Psi_0; \Lambda_1 \Vdash \mathcal{C}_1 :: \Delta_2, \Lambda_1'$, and otherwise $\Psi_0; \Lambda_1 \Vdash \mathcal{C}_1 :: x_\beta{:}A[c], \Delta_2, \Lambda_1'$.

10

If $y_\alpha{:}A[c] \in \Delta'$ then $\Psi; \Lambda_3, \Delta_2 \Vdash \mathcal{C}_2 :: \Lambda'_2$, and otherwise $\Psi_0; \Lambda_3, \Delta_2, y_\alpha{:}A[c] \Vdash \mathcal{C}_2 :: \Lambda'_2$.
Where if $x_\beta{:}A[c] \in \Delta$ then $\Delta = \Lambda_1, \Lambda_3, x_\beta{:}A[c]$ and otherwise $\Delta = \Lambda_1, \Lambda_3$. And if $y_\alpha{:}A[c] \in \Delta'$ then $\Delta' = \Lambda'_1, \Lambda'_2, y_\alpha{:}A[c]$ and otherwise $\Delta' = \Lambda'_1, \Lambda'_2$.
**By proc rule and $(\star)$ and $(\star')$ and $(\star'')$:** $\Psi; \Delta'_1, \Lambda'_2 \Vdash \mathbf{proc}(x_0[d], [\Lambda/\Lambda']\hat{\gamma} P_{x_0^d}@d_2) :: x_0{:}B[d]$ and $\Psi; \Delta''_1, \Lambda''_2, x_0{:}B[d] \Vdash \mathbf{proc}(y_\alpha[c], Q_{x_0^d}@d_1) :: y_\alpha{:}A[c]$ where $\Delta_1 = \Delta'_1, \Delta''_1$ and $\Lambda_2 = \Lambda'_2, \Lambda''_2$.

**Case 2.** (spawn)

$$\mathcal{C}_1\mathbf{proc}(y_\alpha[c], (x^d \leftarrow X \leftarrow \Lambda)@d_2; Q@d_1)\mathcal{C}_2 \qquad !\mathbf{def}(\Psi'; \Lambda' \vdash X = P@\psi' :: x : B'[\psi])$$

$$\mapsto_{\Delta \Vdash \Delta'} \mathcal{C}_1 \mathbf{proc}(x_0[d], ([x_0, \Lambda/x, \Lambda']\hat{\gamma} P)@d_2)\mathbf{proc}(y_\alpha[c], ([x_0/x] Q)@d_1)\mathcal{C}_2 \qquad \hat{\gamma}(\Psi') = \Psi_0 \ (x_0 \text{ fresh})$$

**By assumption of the theorem:** $\Psi_0; \Delta \Vdash \mathcal{C}_1\mathbf{proc}(y_\alpha[c], (x^d \leftarrow X \leftarrow \Lambda)@d_2; Q@d_1)\mathcal{C}_2 :: \Delta'$.
**By Lemma C.2:** $\Psi_0; \Lambda_1 \Vdash \mathcal{C}_1 :: \Delta_1, \Delta_2, \Lambda'_1$ and $\Psi_0; \Delta_1, \Lambda_2 \Vdash \mathbf{proc}(y_\alpha[c], (x^d \leftarrow X \leftarrow \Lambda)@d_2; Q@d_1) :: y_\alpha{:}A[c]$. If $y_\alpha{:}A[c] \notin \Delta'$ then $\Psi_0; \Lambda_3, \Delta_2, y_\alpha{:}A[c] \Vdash \mathcal{C}_2 :: \Lambda'_2$ and otherwise $\Psi; \Lambda_3, \Delta_2 \Vdash \mathcal{C}_2 :: \Lambda'_2$. Where $\Delta = \Lambda_1, \Lambda_2, \Lambda_3$, if $y_\alpha{:}A[c] \in \Delta'$ then $\Delta' = \Lambda'_1, \Lambda'_2, y_\alpha{:}A[c]$ and otherwise $\Delta' = \Lambda'_1, \Lambda'_2$.
**By inversion on proc rule** $\Psi_0; \Delta_1, \Lambda_2 \vdash (x^d \leftarrow X \leftarrow \Lambda)@d_2; Q@d_1) :: y_\alpha{:}A[c]$. Moreover, $(\star)$ $\Psi_0 \Vdash d_1 \sqsubseteq c$ and $\forall u_\gamma{:}T[d'] \in \Delta_1, \Lambda_2. \Psi \Vdash d' \sqsubseteq c$.
**By inversion on Spawn rule:** $\dagger$ $\Psi'; \Lambda' \vdash X = P_{x_0^\psi}@\psi' :: x_0{:}B'[\psi]$ and $\Psi_2; \Delta''_1, \Lambda''_2, x_0{:}B[d] \Vdash Q_{x_0^d}@d_1 :: y_\alpha{:}A[c]$ where $\Delta_1 = \Delta'_1, \Delta''_1$ and $\Lambda_2 = \Lambda'_2, \Lambda''_2$, and $\Lambda = \Delta'_1, \Lambda'_2$. Moreover $\dagger_1$ $\hat{\gamma}(\Lambda'; B'[\psi]; \psi') = (\Lambda; B[d]; d_2)$ and $\dagger_2$ $\Psi_0 \Vdash \hat{\gamma}(\Psi')$ and $(\star')$ $\Psi_0 \Vdash d_1 \sqsubseteq d_2$ and $d \sqsubseteq c$.
**By inversion on signature typing,** $\dagger$ and $\dagger_1$ and $\dagger_2$: $\dagger'$ $\Psi'; \Lambda' \vdash_\Sigma P@\psi' :: (x_0{:}B'[\psi])$ and $\mathtt{concrete}(\Psi') = \Psi_0$ and $\dagger_3$ $\Psi' \Vdash \psi' \sqsubseteq \psi$ and $\forall u_\gamma{:}T[\psi''] \in \Lambda'. \Psi \Vdash \psi'' \sqsubseteq \psi$.
**By Lemma C.1 and $\dagger'$:** $\Psi_0; \Lambda \vdash_\Sigma [\Lambda/\Lambda']\hat{\gamma}(P_{x_0^d})@d_1 :: (x_0{:}B[d])$. Moreover by applying the order-preserving substitution $\gamma$ on the inequalities of $\dagger_3$ we get $(\star'')$ $\Psi_0 \Vdash d_1 \sqsubseteq d$ and $\forall u_\gamma{:}T[d'] \in \Lambda. \Psi_0 \Vdash d' \sqsubseteq d$.
**By proc rule and $(\star)$ and $(\star')$ and $(\star'')$:** $\Psi; \Delta'_1, \Lambda'_2 \Vdash \mathbf{proc}(x_0[d], [\Lambda/\Lambda']\hat{\gamma} P_{x_0^d}@d_2) :: x_0{:}B[d]$ and $\Psi; \Delta''_1, \Lambda''_2, x_0{:}B[d] \Vdash \mathbf{proc}(y_\alpha[c], Q_{x_0^d}@d_1) :: y_\alpha{:}A[c]$ where $\Delta_1 = \Delta'_1, \Delta''_1$ and $\Lambda_2 = \Lambda'_2, \Lambda''_2$.
**By configuration typing rules:** $\Psi; \Delta \Vdash \mathcal{C}_1\mathbf{proc}(x_0[d], [\Lambda/\Lambda]\hat{\gamma} P_{x_0^d}@d_2)\mathbf{proc}(y_\alpha[c], Q_{x_0^d}@d_1)\mathcal{C}_2 :: \Delta'$.

**Case 3.** $(\otimes)$
  **Subcase 1.** (send)

$$\mathcal{C}_1\mathbf{proc}(y_\alpha[c], (\mathbf{send}\, x_\beta^c\, y_\alpha; P)@d_1)\mathcal{C}_2 \mapsto_{\Delta \Vdash \Delta'} \mathcal{C}_1\mathbf{proc}(y_{\alpha+1}[c], ([y_{\alpha+1}/y_\alpha]P)@d_1)\mathbf{msg}(\mathbf{send}\, x_\beta^c\, y_\alpha; y_\alpha \leftarrow y_{\alpha+1})\mathcal{C}_2$$

**By assumption of the theorem:** $\Psi_0; \Delta \Vdash \mathcal{C}_1\mathbf{proc}(y_\alpha[c], (\mathbf{send}\, x_\beta^c\, y_\alpha; P)@d_1)\mathcal{C}_2 :: \Delta'$.
**By Lemma C.2:** $\Psi_0; \Lambda_1 \Vdash \mathcal{C}_1 :: \Delta_1, \Delta_2, \Lambda'_1$ and $\Psi_0; \Lambda''_2, \Delta'_1, x_\beta{:}A[c] \Vdash \mathbf{proc}(y_\alpha[c], (\mathbf{send}\, x_\beta^c\, y_\alpha; P)@d_1) :: y_\alpha{:}(A \otimes B)[c]$ and $\Psi_0; \Lambda_3, \Delta_2 \Vdash \mathcal{C}_2 :: \Lambda'_2$.
Where $\Delta = \Lambda_1, \Lambda_2, \Lambda_3$ and $\Delta' = \Lambda'_1, \Lambda'_2$. If $x_\beta{:}A[c] \in \Delta$ we have $\Delta'_1 = \Delta_1$ and $\Lambda''_2, x_\beta{:}A[c] = \Lambda_2$, and otherwise $\Delta'_1, x_\beta{:}A[c] = \Delta_1$ and $\Lambda''_2 = \Lambda_2$.
**By inversion on proc rule** $\Psi_0; \Lambda''_2, \Delta'_1, x_\beta{:}A[c] \vdash (\mathbf{send}\, x_\beta^c\, y_\alpha^c; P)@d_1 :: y_\alpha{:}(A \otimes B)[c]$. Moreover, $(\star)$ $\forall u_\gamma{:}T[d_2] \in \Lambda''_2, \Delta'_1, x_\beta{:}A[c]. \Psi_0 \Vdash d_2 \sqsubseteq c$.
**By inversion on $\otimes R$ rule** $\Psi_0; \Lambda''_2, \Delta'_1 \vdash P@d_1 :: y_\alpha{:}B[c]$.
**By substitution of $y_{\alpha+1}$ for $y_\alpha$:**
$$\Psi_0; \Lambda''_2, \Delta'_1 \vdash [y^{\alpha+1}/y^\alpha]P@d_1 :: y_{\alpha+1}{:}B[c].$$
**By proc rule and $(\star)$:** $\dagger$ $\Psi_0; \Lambda''_2, \Delta'_1 \Vdash \mathbf{proc}(y_{\alpha+1}[c], [y^{\alpha+1}/y^\alpha]P@d_1) :: y_{\alpha+1}{:}B[c]$.
**By $\otimes R$ and fwd rule:** $\Psi_0; x_\beta{:}A[c], y_{\alpha+1}{:}B[c] \vdash (\mathbf{send}\, x_\beta^c\, y_\alpha^c; y_\alpha^c \leftarrow y_{\alpha+1}^c)@c :: y_\alpha{:}(A \otimes B)[c]$.
**By msg, $\star$, and $\dagger$:**

$$\Psi_0; \Lambda''_2, \Delta'_1, x_\beta{:}A[c] \vdash \mathbf{proc}(y_{\alpha+1}[c], [y^{\alpha+1}/y^\alpha]P@d_1)\mathbf{msg}(\mathbf{send}\, x_\beta^c\, y_\alpha^c; y_\alpha^c \leftarrow y_{\alpha+1}^c) :: y_\alpha{:}(A \otimes B)[c].$$

**By configuration typing rules:** $\Psi_0; \Delta \Vdash \mathcal{C}_1\mathbf{proc}(y_{\alpha+1}[c], [y^{\alpha+1}/y^\alpha]P@d_1)\mathbf{msg}(\mathbf{send}\, x_\beta^c\, y_\alpha^c; y_\alpha^c \leftarrow y_{\alpha+1}^c)\mathcal{C}_2 :: \Delta'$
  **Subcase 2.** (receive)
$$\mathcal{C}_1\mathbf{msg}(\mathbf{send}\, x_\beta^c\, y_\alpha; y_\alpha \leftarrow v_\delta)\mathcal{C}'\mathbf{proc}(u_\gamma[c'], (w_\eta^c \leftarrow \mathbf{recv}\, y_\alpha; P)@d_1)\mathcal{C}_2$$

$$\mapsto_{\Delta \Vdash \Delta'} \mathcal{C}_1\mathcal{C}'\mathbf{proc}(u_\gamma[c'], ([x_\beta/w_\eta][v_\delta/y_\alpha]P)@(d_1 \sqcup c))\mathcal{C}_2$$

**By assumption of the theorem:** $\Psi_0; \Delta \Vdash \mathcal{C}_1\mathbf{msg}(\mathbf{send}\, x_\beta^c\, y_\alpha^c; y_\alpha^c \leftarrow v_\delta^c)\mathcal{C}'\mathbf{proc}(u_\gamma^{c'}, w_\eta^c \leftarrow \mathbf{recv}\, y_\alpha^c; P@d_1)\mathcal{C}_2 :: \Delta'$.
**By Lemma C.2:** $\Psi_0; \Lambda_1 \Vdash \mathcal{C}_1 :: \Delta_1, \Delta_2, \Delta_3, \Delta_4, \Lambda'_1$ and $\Psi_0; \Lambda''_2, \Delta'_1, x_\beta{:}A[c] \Vdash \mathbf{msg}(\mathbf{send}\, x_\beta^c\, y_\alpha^c; y_\alpha^c \leftarrow v_\delta^c) :: y_\alpha{:}(A \otimes B)[c]$ and $\Psi_0; \Lambda_3, \Delta_2 \Vdash \mathcal{C}' :: \Lambda'_2, \Delta_5, \Delta_6$ and $\Psi_0; \Lambda_4, \Delta_3, y_\alpha{:}(A \otimes B)[c], \Delta_5 \Vdash \mathbf{proc}(u_\gamma^{c'}, w_\eta^c \leftarrow \mathbf{recv}\, y_\alpha^c; P\, @d_1) :: (u_\gamma{:}C[c'])$ and $\Psi_0; \Lambda_5, \Delta_4, \Delta'_6 \Vdash \mathcal{C}_2 :: \Lambda'_3$.

11

Where $\Delta = \Lambda_1, \Lambda_2, \Lambda_3, \Lambda_4, \Lambda_5$ and $\Delta' = \Lambda'_1, \Lambda'_2, \Lambda'_3$. If $x_\beta{:}A[c] \in \Delta$ we have $\Delta'_1 = \Delta_1$ and $\Lambda''_2, x_\beta{:}A[c] = \Lambda_2$, and otherwise $\Delta'_1, x_\beta{:}A[c] = \Delta_1$ and $\Lambda''_2 = \Lambda_2$. Also, $\Delta'_6 = \Delta_6$ if $u_\gamma{:}C[c'] \in \Delta'$ and otherwise $\Delta'_6 = \Delta_6, u_\gamma{:}C[c']$.

**By inversion on msg rule** $\Psi_0; \Lambda''_2, \Delta'_1, x_\beta{:}A[c] \vdash (\mathbf{send}\ x_\beta^c\ y_\alpha^c; y_\alpha^c \leftarrow v_\delta^c)@c :: y_\alpha{:}(A \otimes B)[c]$. Moreover, $(\star)\ \forall u_\gamma{:}T[d_2] \in \Lambda''_2, \Delta'_1, x_\beta{:}A[c].\ \Psi_0 \Vdash d_2 \sqsubseteq c$. By inversion on $\otimes R$ rule, we know that $\Lambda''_2, \Delta'_1 = v_\delta^c{:}B[c]$

**By inversion on proc rule** $\Psi_0; \Lambda_4, \Delta_3, y_\alpha{:}(A[c] \otimes B)[c], \Delta_5 \vdash (w_\eta^c \leftarrow \mathbf{recv}\ y_\alpha^c; P)@d_1 :: u_\gamma{:}C[c']$. Moreover, $(\star')\ \forall u_\gamma{:}T[d_2] \in \Lambda_4, \Delta_3, y_\alpha{:}(A \otimes B)[c], \Delta_5.\ \Psi_0 \Vdash d_2 \sqsubseteq c'$, and $\Psi_0 \Vdash d_1 \sqsubseteq c'$.

**By inversion on $\otimes L$ rule** $\Psi_0; \Lambda_4, \Delta_3, y_\alpha{:}B[c], w_\eta{:}A[c], \Delta_5 \vdash P @d_1 \sqcup c :: u_\gamma{:}C[c']$.

**By substitution of $x_\beta$ for $w_\eta$ and $v_\delta$ for $y_\alpha$:**

$$\Psi_0; \Lambda_4, \Delta_3, v_\delta{:}B[c], x_\beta{:}A[c], \Delta_5 \vdash [v_\delta/y_\alpha][x_\beta/w_\eta]P @d_1 \sqcup c :: u_\gamma{:}C[c'].$$

**By proc rule, $(\star)$ and $(\star')$:** $\Psi_0; \Lambda_4, \Delta_3, v_\delta{:}B[c], x_\beta{:}A[c], \Delta_5 \Vdash \mathbf{proc}(u_\gamma[c'], [x_\beta/w_\eta][v_\delta/y_\alpha]P @d_1 \sqcup c) :: u_\gamma{:}C[c']$.

**By configuration typing rules** $\Psi_0; \Delta \Vdash \mathcal{C}_1\mathcal{C}'\mathbf{proc}(u_\gamma[c'], [x_\beta/w_\eta][v_\delta/y_\alpha]P @d_1 \sqcup c)\mathcal{C}_2 :: \Delta'$

$\square$

**Remark**: To keep the proofs concise we may use $\Psi$ instead of the term secrecy lattice $\Psi_0$, whenever we are clearly working with configurations defined in the run-time.

**Definition C.1** (Projections). *Projection for linear context is defined as follows:*

$$\begin{aligned}
\Delta, x_\alpha{:}T[c] \Downarrow \xi &\overset{\text{def}}{=} \Delta \Downarrow \xi, x_\alpha{:}T[c] && \text{if } c \sqsubseteq \xi \\
\Delta, x_\alpha{:}T[c] \Downarrow \xi &\overset{\text{def}}{=} \Delta \Downarrow \xi && \text{if } c \not\sqsubseteq \xi \\
\cdot \Downarrow \xi &\overset{\text{def}}{=} \cdot \\
x_\alpha{:}T[c] \Downarrow \xi &\overset{\text{def}}{=} x_\alpha{:}T[c] && \text{if } c \sqsubseteq \xi \\
x_\alpha{:}T[c] \Downarrow \xi &\overset{\text{def}}{=} \cdot && \text{if } c \not\sqsubseteq \xi
\end{aligned}$$

**Definition C.2.** *For each dynamic step in Figure 1, we say that a process is* produced *in the step if it occurs in the post-step (right side of $\mapsto_{\Delta \Vdash \Delta'}$ notation). For example, the transition step for* Spawn, *produces two processes*

$$\mathbf{proc}(x_0[d], ([x_0, \Lambda/x, \Lambda']\,\hat\gamma\,P)@d_2)\ \text{and}\ \mathbf{proc}(y_\alpha[c], ([x_0/x]\,Q)@d_1),$$

*and the transition step for $\multimap_{\mathsf{snd}}$ produces a message and a process*

$$\mathbf{msg}(\mathbf{send}\ x_\beta\ u_\gamma; u_{\gamma+1} \leftarrow u_\gamma)\ \text{and}\ \mathbf{proc}(y_\alpha[c], ([u_{\gamma+1}/u_\gamma]P)@d_1).$$

*The transition step for* fwd *does not produce any processes.*

There are two ways to interpret the semantics of a forwarding process in Figure 1,

$$\mathbf{proc}(y_\alpha[c], (y_\alpha \leftarrow x_\beta)@d_1)\ \mapsto_{\Delta \Vdash \Delta'}\ [x_\beta/y_\alpha].$$

The common interpretation is to globally substitute the channel $x_\beta$ for $y_\alpha$ in the configuration and remove $[x_\beta/y_\alpha]$. Alternatively, we can keep both channels $x_\beta$ and $y_\alpha$ alive, and keep the relation $[x_\beta/y_\alpha]$ in the configuration to remember that $y_\alpha$ can always be instantiated with $x_\beta$. In each reduction step involving a communication along two channels ($\{1_{\mathsf{rcv}}, \oplus_{\mathsf{rcv}}, \&_{\mathsf{rcv}}, \otimes_{\mathsf{rcv}}, \multimap_{\mathsf{rcv}}\}$), we consult the set of available substitutions to unify the channels. After the communication is done, the channels cease to exist and we can safely ignore the substitutions that unified them, i.e. we consult each substitution at most once. In the remaining lemmas of this subsection, we consider the latter interpretation to simplify the proofs, particularly, the proof of Lemma C.8 in which we build the minimal sending configuration. Since these two interpretations are equivalent, the result of Lemma C.8 holds independent from which interpretation we use. It is important to observe that in both of these interpretation the offering channel of a process/message will always maintain its name.

**Lemma C.5** (Uniqueness of process productions). *Consider $\Psi; \Delta \Vdash \mathcal{D} :: K$, and $\mathcal{D} \mapsto^* \mathcal{D}' \mapsto \mathcal{D}_1 \mathbf{proc}(x, P)\mathcal{D}_2$, such that $\mathbf{proc}(x, P)$ is produced in the last step. Then process $\mathbf{proc}(x, P)$ does not occur in any of the previous steps $\mathcal{D} \mapsto^* \mathcal{D}'$.*

*The same result hold for the production of a message.*

*Proof.* Observe the following invariant in the dynamics

1. In each configuration, there exists at most one process or message offering along a particular generation of $x$.
2. The offering channel of processes that are not active or produced in the step does not change. In particular, fwd step does not change the offering channel of any process or message.

3. The offering channel of a negative message is a fresh generation of a channel (and no process offers along it before the message is received).
4. In the production of $\mathbf{proc}(x_\alpha, P)$, either (a) $x^\alpha$ is freshly generated (in the case of Spawn for the callee), or (b) $\mathbf{proc}(x_\alpha, P)$ replaces another process $\mathbf{proc}(x_\alpha, P')$ that offers along $x^\alpha$ and has a larger process term, i.e. $|P| < |P'|$ (in $1_{\text{rcv}}, \oplus_{\text{rcv}}, \otimes_{\text{rcv}}, \&_{\text{snd}}, \multimap_{\text{snd}}$, and spawn for the continuation of the caller), or (c) $x_\alpha$ is a fresh generation of $x$ (in $\oplus_{\text{snd}}, \otimes_{\text{snd}}, \&_{\text{rcv}}, \multimap_{\text{rcv}}$).

Consider production of a process $\mathbf{proc}(x_\alpha, P)$ and the cases described in 4. If 4.(a) or 4.(c) hold, then by freshness of $x_\alpha$, such process has not been produced before. It is enough to consider case 4.(b). Assume that there is another occurrence of process $\mathbf{proc}(x_\alpha, P)$ before this production; By the observations 1, 2 and 4 we get to a contradiction: there must be a chain of productions satisfying 4.(b) with decreasing sizes from the earlier $\mathbf{proc}(x_\alpha, P)$ to the later one $\mathbf{proc}(x_\alpha, P)$, which is contradictory.

With a similar reasoning we can prove that if $\Psi; \Delta \Vdash \mathcal{D} :: K$, and $\mathcal{D} \mapsto^* \mathcal{D}' \mapsto \mathcal{D}_1 \mathbf{msg}(P) \mathcal{D}_2$, such that $\mathbf{msg}(P)$ is produced in the last step and $\overline{u}$ is the name of the message resources, then $\mathbf{msg}(P])$ does not occur in any of $\mathcal{D} \mapsto^* \mathcal{D}'$ steps. □

**Definition C.3.** *For each dynamic step in Figure 1, we define the active configuration $\mathcal{A}$ as the list of processes occurring in the pre-step (the left side of $\mapsto_{\Delta \Vdash \Delta'}$ notation). For example, in the transition step for Spawn, the active configuration is the single process*
$$\mathbf{proc}(y_\alpha[c], (x^d \leftarrow X \leftarrow \Lambda)@d_2; Q@d_1),$$
*and in the transition step for $\multimap_{\text{rcv}}$, the active configuration is*
$$\mathbf{proc}(y_\alpha[c], (w_\eta \leftarrow \mathbf{recv}\ y_\alpha; P)@d_1)\ \mathbf{msg}(\mathbf{send}\ x_\beta\ y_\alpha; v_\delta \leftarrow y_\alpha).$$

*We define the set of active configurations, i.e. active, for $\mathcal{D}_1 \mapsto^{*\Upsilon}_{\Delta \Vdash K} \mathcal{D}'_1$ as the union of every step's active configuration. A process is called active in $\mathcal{D}_1 \mapsto^{*\Upsilon}_{\Delta \Vdash K} \mathcal{D}'_1$, if it is in the set active.*

**Lemma C.6** (Diamond Property). *If $\Psi; \Delta \Vdash \mathcal{D}_1 :: K$ and $\dagger\ \mathcal{D}_1 \mapsto^{\Upsilon}_{\Delta \Vdash K} \mathcal{D}'_1$ and $\dagger'\ \mathcal{D}_1 \mapsto^{\Upsilon'}_{\Delta \Vdash K} \mathcal{D}''_1$, and $\mathcal{D}'_1 \neq \mathcal{D}''_1$ then there is a configuration $\mathcal{D}$ such that $\star\ \mathcal{D}'_1 \mapsto^{\Upsilon''}_{\Delta \Vdash K} \mathcal{D}$, and $\star'\ \mathcal{D}''_1 \mapsto^{\Upsilon''}_{\Delta \Vdash K} \mathcal{D}$, where $\Upsilon \cup \Upsilon' = \Upsilon''$. The messages produced along the channels $\Upsilon \cap \Upsilon'$ are identical in $\mathcal{D}'_1$ and $\mathcal{D}''_1$ and $\mathcal{D}$.*

*Moreover, every process in $\mathcal{D}'_1$ that is not an active process of $\mathcal{D}_1 \mapsto^{\Theta'; \Upsilon'}_{\Delta \Vdash K} \mathcal{D}''_1$ is in $\mathcal{D}$. And every process in $\mathcal{D}''_1$ that is not an active process of $\mathcal{D}_1 \mapsto^{\Theta; \Upsilon}_{\Delta \Vdash K} \mathcal{D}'_1$ is in $\mathcal{D}$.*

*Proof.* The proof is straightforward by cases. The key is to build $\star$ (locally) identical to $\dagger'$, and $\star'$ (locally) identical to $\dagger$. □

**Lemma C.7** (Confluence). *If $\Psi; \Delta \Vdash \mathcal{D}_1 :: K$ and $\dagger\ \mathcal{D}_1 \mapsto^{m'_{\Upsilon}}_{\Delta \Vdash K} \mathcal{D}'_1$ and $\dagger'\ \mathcal{D}_1 \mapsto^{n'_{\Upsilon'}}_{\Delta \Vdash K} \mathcal{D}''_1$, then there is a configuration $\mathcal{D}$ such that $\star\ \mathcal{D}'_1 \mapsto^{j_{\Upsilon''}}_{\Delta \Vdash K} \mathcal{D}$ for some $j \leq n'$, and $\star'\ \mathcal{D}''_1 \mapsto^{k_{\Upsilon''}}_{\Delta \Vdash K} \mathcal{D}$ for some $k \leq m'$, where $\Upsilon \cup \Upsilon' = \Upsilon''$. The messages produced along the channels $\Upsilon \cap \Upsilon'$ are identical in $\mathcal{D}'_1$ and $\mathcal{D}''_1$ and $\mathcal{D}_1$. The steps in $\star$ are (locally) identical to the steps of $\dagger'$ that do not occur in $\dagger$. And the steps in $\star'$ are (locally) identical to the steps of $\dagger$ that do not occur in $\dagger'$.*

*Moreover, every process in $\mathcal{D}'_1$ that is not an active process of $\mathcal{D}_1 \mapsto^{*\Upsilon'}_{\Delta \Vdash K} \mathcal{D}''_1$ is in $\mathcal{D}$. And every process in $\mathcal{D}''_1$ that is not an active process of $\mathcal{D}_1 \mapsto^{*\Upsilon}_{\Delta \Vdash K} \mathcal{D}'_1$ is in $\mathcal{D}$.*

*Proof.* It follows by standard inductions from the diamond property. The induction is on the pair $(n', m')$. If $n' = 0$ and $m' = 0$, the proof is straightforward. Similarly, if $n' = 1$ and $m' = 0$ or $n' = 0$ and $m' = 1$ the proof is straightforward. For $n' = 1$ and $m' = 1$, we apply the diamond property. Assume that $n' = n + 1$ and $m' = m + 1$. We form the following diagram to sketch the structure of the proof.

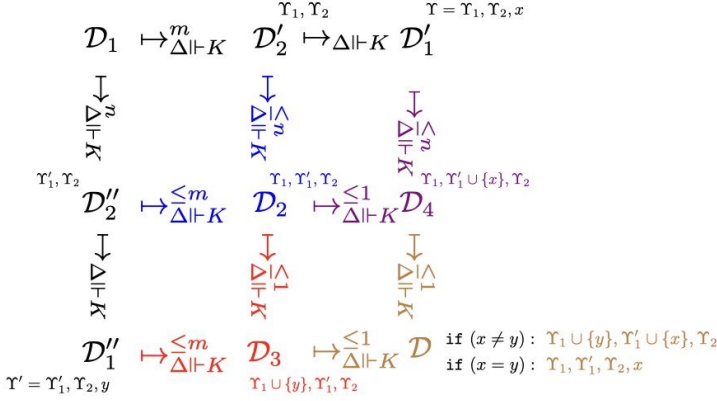

By induction hypothesis, from $\mathcal{D}_2''$ and $\mathcal{D}_2'$, we build $\mathcal{D}_2$ (with the blue steps) that satisfies the required properties. Then, again by induction we build $\mathcal{D}_3$ (with the red steps) and $\mathcal{D}_4$ (with the violet steps). And finally, with a last induction, we build $\mathcal{D}$ (with the brown steps). The diagram depicts how the required properties move along the steps. For each configuration, we put the set of channels that it sends along them near the configuration. In particular, for $\mathcal{D}_1'$, we put $\Upsilon = \Upsilon_1, \Upsilon_2, x$ and for $\mathcal{D}_1''$, we put $\Upsilon' = \Upsilon_1', \Upsilon_2, y$. The set $\Upsilon_2$ is in both $\Upsilon$ and $\Upsilon'$, and we have $\Upsilon_1 \cap \Upsilon_1' = \emptyset$, i.e. we put all the common channels (except possibly $x$ and $y$) in $\Upsilon_2$. We assume that the step $\mathcal{D}_2' \mapsto_{\Delta \vdash K} \mathcal{D}_1'$ produces a message along channel $x$. In the case that this step does not produce any message we can simply ignore $x$. Similarly, we assume that the step $\mathcal{D}_2'' \mapsto_{\Delta \vdash K} \mathcal{D}_1''$ produces a message along channel $y$. In the case that this step does not produce any message we can simply ignore $y$. At the end, we can build $\mathcal{D}$ with at most $m+1$ steps from $\mathcal{D}_1''$ and at most $n+1$ steps from $\mathcal{D}_1'$. The configuration $\mathcal{D}$ sends messages along the union of $\Upsilon$ and $\Upsilon'$, and by induction the messages along the intersection of $\Upsilon \cap \Upsilon'$ are identical in $\mathcal{D}_1'$ and $\mathcal{D}_1''$ and $\mathcal{D}$. The proof of the second part of the lemma is straightforward by stating the required property for each inductive step and passing them down to $\mathcal{D}$ in the diagram.

□

As a straightforward corollary to the first part of the confluence lemma, we get that the messages produced along the channels $\Upsilon \cap \Upsilon'$ are identical in $\mathcal{D}_1'$ and $\mathcal{D}_1''$, i.e. identical messages will be produced along the same channels, independent from the non-deterministic path that we take to produce them.

**Corollary C.7.1.** *If $\mathcal{D}_1 \mapsto^* \mathcal{D}_1'\mathcal{A}\mathcal{D}_1'' \mapsto \mathcal{D}_1^1\mathbf{proc}(x,P)\mathcal{D}_1^2$ and $\mathcal{D}_1 \mapsto^* \mathcal{D}_2'\mathcal{A}'\mathcal{D}_2'' \mapsto \mathcal{D}_2^1\mathbf{proc}(x,P)\mathcal{D}_2^2$, where $\mathcal{A}$ and $\mathcal{A}'$ are the active parts that produce $\mathbf{proc}(x,P)$, then $\mathcal{A} = \mathcal{A}'$. Moreover, $\mathcal{A}$ and $\mathcal{A}'$ consult the same set of substitutions to produce $\mathbf{proc}(x,P)$.*

*Proof.* This is another corollary to the second part of the confluence lemma (Lemma C.7). First observe that $\mathcal{A}$ is not active in $\mathcal{D}_1 \mapsto^* \mathcal{D}_2'\mathcal{A}'\mathcal{D}_2''$: if it is active, we produce $\mathbf{proc}(x,P)$ twice in $\mathcal{D}_1 \mapsto^* \mathcal{D}_2'\mathcal{A}'\mathcal{D}_2'' \mapsto \mathcal{D}_2^1\mathbf{proc}(x,P)\mathcal{D}_2^2$, which contradicts with uniqueness of process productions (Lemma C.5). Similarly, $\mathcal{A}'$ is not active in $\mathcal{D}_1 \mapsto^* \mathcal{D}_1'\mathcal{A}\mathcal{D}_1''$. With the second part of the confluence lemma, for some $\mathcal{D}$, we have $\mathcal{D}_1'\mathcal{A}\mathcal{D}_1'' \mapsto^* \mathcal{D}$ and $\mathcal{D}_2'\mathcal{A}'\mathcal{D}_2'' \mapsto^* \mathcal{D}$ such that both $\mathcal{A}$ and $\mathcal{A}'$ occur in $\mathcal{D}$. If $\mathcal{A} \neq \mathcal{A}'$, we can produce $\mathbf{proc}(x,P)$ twice from $\mathcal{D}$ which by Lemma C.5 is contradictory. Thus, know that $\mathcal{A} = \mathcal{A}'$.

Note: here we rely on the fact that the steps in Figure 1 produce the post-steps uniquely from the pre-steps. In particular, in the spawn rule we assume that the fresh channel name is uniquely determined based on the offering channel and the process term of the caller.

For the second part, we consider the cases in which a process receives a message, i.e. the local step is in the set $\{1_{\mathsf{rcv}}, \oplus_{\mathsf{rcv}}, \&_{\mathsf{rcv}}, \otimes_{\mathsf{rcv}}, \multimap_{\mathsf{rcv}}\}$). These cases are the only ones in which a local step needs to consult a substitution. We can rewrite $\mathcal{A} = \mathcal{A}'$ as $\mathcal{A} = \mathcal{A}_1\mathcal{A}_2$ where the channels connecting $\mathcal{A}_1$ and $\mathcal{A}_2$ has to be unified by substitutions. It is straightforward to observe that $\mathcal{A} = \mathcal{A}'$ needs a unique substitution to match up its communicating channels and take the step. However, this substitution can be a transitive result of multiple substitutions, e.g. for $\mathbf{msg}(\mathbf{close}\, z_\delta)\, \mathbf{proc}(x_\beta[c'], (\mathbf{wait}\, y_\alpha; Q)@d_1) \mapsto_{\Delta \Vdash \Delta'} \mathbf{proc}(x_\beta[c'], Q@(d_1 \sqcup c))$, we require the substitution $[z_\delta/y_\alpha]$, which can be inferred by transitivity from $\{[z_\delta/w_\eta], [w_\eta/y_\alpha]\}$. We need to show that $\mathcal{A} = \mathcal{A}'$ always consult the same set of substitutions. Let $S$ and $S'$ be the set of substitutions consulted by $\mathcal{A}$ and $\mathcal{A}'$, respectively. Since the substitutions in $S$ and $S'$ are still alive we know that they have never been consulted before, and thus we can safely remove the Fwd steps $\mathcal{S}$ and $\mathcal{S}'$ that create them from $\mathcal{D}_1 \mapsto^* \mathcal{D}_1'\mathcal{A}\mathcal{D}_1''$ and $\mathcal{D}_2 \mapsto^* \mathcal{D}_2'\mathcal{A}'\mathcal{D}_2''$ to get $\mathcal{D}_1 \mapsto^* \mathcal{D}_1'\mathcal{A}_1\mathcal{S}\mathcal{A}_2\mathcal{D}_1''$ and $\mathcal{D}_2 \mapsto^* \mathcal{D}_2'\mathcal{A}_1\mathcal{S}'\mathcal{A}_2\mathcal{D}_2''$, respectively. Similar to the above reasoning, we get a contradiction from $S \neq S$, the fact that the pre-step for each substitution is unique and the confluence lemma. □

14

Similarly, we can prove that if $\mathcal{D}_1 \mapsto^* \mathcal{D}'_1 \mathcal{A} \mathcal{D}''_1 \mapsto \mathcal{D}^1_1 \mathbf{msg}(P) \mathcal{D}^2_1$ and $\mathcal{D}_1 \mapsto^* \mathcal{D}'_2 \mathcal{A}' \mathcal{D}''_2 \mapsto \mathcal{D}^1_2 \mathbf{msg}(P) \mathcal{D}^2_2$, where $\mathcal{A}$ and $\mathcal{A}'$ are the active parts that produce $\mathbf{msg}(P)$, then $\mathcal{A} = \mathcal{A}'$.

In the proof of Corollary C.7.1, we used the fact that the pre-step for each substitution is unique. This can be proved independently by a straightforward observation that if $\mathcal{D}_1 \mapsto^* \mathcal{D}'_1 \mathcal{A} \mathcal{D}''_1 \mapsto \mathcal{D}^1_1 [x_\beta/y_\alpha] \mathcal{D}^2_1$ and $\mathcal{D}_1 \mapsto^* \mathcal{D}'_2 \mathcal{A}' \mathcal{D}''_2 \mapsto \mathcal{D}^1_2 [x_\beta/y_\alpha] \mathcal{D}^2_2$, where $\mathcal{A}$ and $\mathcal{A}'$ are forwarding processes that step by renaming the resource $y_\alpha$ to $x_\beta$ ($[x_\beta/y_\alpha]$), then $\mathcal{A} = \mathcal{A}'$.

**Lemma C.8** (Build the minimal sending configuration). *Consider $\Delta \Vdash \mathcal{D}_2 :: K$ and $\mathcal{D}_2 \mapsto^{*\Upsilon} \mathcal{D}'_2$. There exists a configuration $\mathcal{D}''_2$ such that $\mathcal{D}_2 \mapsto^{*\Upsilon} \mathcal{D}''_2$ and $\forall \mathcal{D}^1_2, \Upsilon_1 \supseteq \Upsilon$. if $\mathcal{D}_2 \mapsto^{*\Upsilon_1} \mathcal{D}^1_2$ then $\mathcal{D}''_2 \mapsto^{*\Upsilon_1} \mathcal{D}^1_2$. We call $\mathcal{D}''_2$ the minimal sending configuration with respect to $\Upsilon$ and $\mathcal{D}_2$.*

*Proof.* We first provide an algorithm to build $\mathcal{D}''_2$ based on the transition steps in $\dagger \mathcal{D}_2 \mapsto^{*\Upsilon} \mathcal{D}'_2$ and we show that $\mathcal{D}_2 \mapsto^{*\Upsilon} \mathcal{D}''_2$ and $\mathcal{D}'_2 \mapsto^{*\Upsilon} \mathcal{D}''_2$. Then, we prove that if we apply the algorithm on every $\mathcal{D}^1_2$ with $\mathcal{D}_2 \mapsto^{*\Upsilon_1} \mathcal{D}^1_2$ and $\Upsilon \subseteq \Upsilon_1$, we build the same $\mathcal{D}'_2$ that satisfies $\mathcal{D}'_2 \mapsto^{*\Upsilon} \mathcal{D}''_2$.

---
**Algorithm 1** Building the minimal sending configuration
---
**Require:** A set $\Upsilon$ of channels, configuration $\mathcal{D}_2$, and a configuration $\mathcal{D}$ with $\mathcal{D}_2 \mapsto^{*\Upsilon_1} \mathcal{D}$ and $\Upsilon_1 \supseteq \Upsilon$.
**Ensure:** A configuration $\mathcal{D}''_2$ with $\mathcal{D}_2 \mapsto^{*\Upsilon} \mathcal{D}''_2$ and $\mathcal{D}''_2 \mapsto^{*\Upsilon_1} \mathcal{D}$.
  $S :=$ the local transition steps in $\mathcal{D}_2 \mapsto^{*\Upsilon_1} \mathcal{D}$
  $i := 0$
  $X_0 := \emptyset$
  $M :=$ the messages in $\mathcal{D}$ along $\Upsilon$
  $A := \emptyset$
  $\mathsf{cSub} := \emptyset$
  **while** $(M \neq \emptyset)$ **do**
    **for** $(s \text{ in } S)$ **do**
      **if** $(post(s) \subseteq M$ or $post(s) \in \mathsf{cSub})$ **then**      ▷ $post(s)$ is the list on the right-hand side of s.
        $A := A \cup \{pre(s)\}$          ▷ $pre(s)$ is the set of processes on the left-hand side of s.
        $X_i := X_i \cup \{s\}$
        **if** $(s \in \{1_{\mathsf{rcv}}, \oplus_{\mathsf{rcv}}, \&_{\mathsf{rcv}}, \otimes_{\mathsf{rcv}}, \multimap_{\mathsf{rcv}}\})$ **then**
          $\mathsf{cSub} := \mathsf{cSub} \cup \{\text{the set of substitutions required for taking the step.}\}$
        **end if**
      **end if**
    **end for**
    $i := i + 1$
    $X_i = \emptyset$
    $M := A$
    $A := \emptyset$
  **end while**

  $\mathsf{Cfg} := \mathcal{D}_2$
  $\mathsf{Tns} := \cdot$
  $i = i - 1$
  **while** $(i >= 0)$ **do**
    **for** $(s \text{ in } X_i)$ **do**
      $\mathsf{Cfgpost} :=$ the global post-state when the local step $s$ applies to the configuration $\mathsf{Cfg}$.
      $\mathsf{Tns}.\mathsf{append}(\mathsf{Cfg} \mapsto_{\Delta \Vdash K} \mathsf{Cfgpost})$
      $\mathsf{Cfg} := \mathsf{Cfgpost}$
    **end for**
    $i = i - 1$
  **end while**
  $\mathcal{D}''_2 := \mathsf{Cfg}$
  **return** $\mathcal{D}''_2, \mathsf{Tns}$
---

For every given $\Upsilon$, $\mathcal{D}_2$, and $\mathcal{D}$ with $\mathcal{D}_2 \mapsto^{*\Upsilon_1} \mathcal{D}$ and $\Upsilon \subseteq \Upsilon_1$, Algorithm 1 returns the configuration $\mathcal{D}''$ and builds the dynamic transition $\mathcal{D}_2 \mapsto^{*\Upsilon} \mathcal{D}''_2$ in $\mathsf{Tns}$. Observe that the steps of $\mathcal{D}_2 \mapsto^{*\Upsilon} \mathcal{D}''_2$ are all local transitions of $\mathcal{D}_2 \mapsto^{*\Upsilon_1} \mathcal{D}$. Thus, by the confluence lemma (Lemma C.7), we know that $\mathcal{D}''_2 \mapsto^{*\Upsilon_1} \mathcal{D}$.

It remains to show that the $\mathcal{D}_2''$ that Algorithm 1 builds is independent from the choice of $\mathcal{D}$ and $\Upsilon_1$, and is uniquely identified for each $\mathcal{D}_2$ and $\Upsilon$. By the confluence lemma, we know that for each $\mathcal{D}$ with $\mathcal{D}_2 \mapsto^{*\Upsilon_1} \mathcal{D}$ and $\Upsilon \subseteq \Upsilon_1$, the algorithm initialized set $M$ with the same messages. In the first for loop, the algorithm collects the generators of the set $M$ in $A$ and the local steps that produces the set $M$ in $X_i$. By Corollary C.7.1, the set of generators of a set $M$ is the same for each $\mathcal{D}$ with $\mathcal{D}_2 \mapsto^{*\Upsilon_1} \mathcal{D}$. Similarly, the set of local steps that produces the set $M$ is unique for each $\mathcal{D}$ with $\mathcal{D}_2 \mapsto^{*\Upsilon_1} \mathcal{D}$. As a result, we collect the same local transition steps in all $X_i$s for every $\mathcal{D}$ with $\mathcal{D}_2 \mapsto^{*\Upsilon_1} \mathcal{D}$. The local transition steps in all $X_i$s are those with which we construct $\mathcal{D}_2''$ from $\mathcal{D}_2$, and the proof is complete. □

## C.2 Useful properties of the logical relation

**Lemma C.9** (Backward closure on the second substitutions). *Substitutions of the second run enjoy backward closure.*

1. *If $(\mathcal{B}_1; \mathcal{B}_2) \in \mathcal{E}^{\xi}_{\Psi_0}[\![\Delta \Vdash K]\!]^k$ and $\mathcal{B}_i = \mathcal{C}_i \mathcal{D}_i \mathcal{F}_i$ with $(\mathcal{C}_1, \mathcal{D}_1, \mathcal{F}_1; \mathcal{C}_2, \mathcal{D}_2, \mathcal{F}_2) \in \mathsf{Tree}(\Delta \Vdash K)$ and $\mathcal{C}_2^2 \mapsto^* \mathcal{C}_2$ then $(\mathcal{C}_1 \mathcal{D}_1 \mathcal{F}_1; \mathcal{C}_2^2 \mathcal{D}_2 \mathcal{F}_2) \in \mathcal{E}^{\xi}_{\Psi_0}[\![\Delta \Vdash K]\!]^k$.*
2. *If $(\mathcal{B}_1; \mathcal{B}_2) \in \mathcal{E}^{\xi}_{\Psi_0}[\![\Delta \Vdash K]\!]^k$ and $\mathcal{B}_i = \mathcal{C}_i \mathcal{D}_i \mathcal{F}_i$ with $(\mathcal{C}_1, \mathcal{D}_1, \mathcal{F}_1; \mathcal{C}_2, \mathcal{D}_2, \mathcal{F}_2) \in \mathsf{Tree}(\Delta \Vdash K)$ and $\mathcal{F}_2^2 \mapsto^* \mathcal{F}_2$ then $(\mathcal{C}_1 \mathcal{D}_1 \mathcal{F}_1; \mathcal{C}_2 \mathcal{D}_2 \mathcal{F}_2^2) \in \mathcal{E}^{\xi}_{\Psi_0}[\![\Delta \Vdash K]\!]^k$.*
3. *Consider $k \geq 1$ and $(\mathcal{C}_1, \mathcal{D}_1, \mathcal{F}_1; \mathcal{C}_2, \mathcal{D}_2, \mathcal{F}_2) \in \mathcal{V}[\![\Delta \Vdash K]\!]^k_{\cdot; y_\alpha}$ and $\mathcal{C}_2^2 \mapsto^* \mathcal{C}_2$. We have $(\mathcal{C}_1, \mathcal{D}_1, \mathcal{F}_1; \mathcal{C}_2^2, \mathcal{D}_2, \mathcal{F}_2) \in \mathcal{V}[\![\Delta \Vdash K]\!]^k_{\cdot; y_\alpha}$.*
4. *Consider $k \geq 1$ and $(\mathcal{C}_1, \mathcal{D}_1, \mathcal{F}_1; \mathcal{C}_2, \mathcal{D}_2, \mathcal{F}_2) \in \mathcal{V}[\![\Delta \Vdash K]\!]^k_{\cdot; y_\alpha}$ and $\mathcal{F}_2^2 \mapsto^* \mathcal{F}_2$. We have $(\mathcal{C}_1, \mathcal{D}_1, \mathcal{F}_1; \mathcal{C}_2, \mathcal{D}_2, \mathcal{F}_2^2) \in \mathcal{V}[\![\Delta \Vdash K]\!]^k_{\cdot; y_\alpha}$.*
5. *Consider $k \geq 1$ and $(\mathcal{C}_1, \mathcal{D}_1, \mathcal{F}_1; \mathcal{C}_2, \mathcal{D}_2, \mathcal{F}_2) \in \mathcal{V}[\![\Delta \Vdash K]\!]^k_{y_\alpha; \cdot}$ and $\mathcal{C}_2^2 \mapsto^* \mathcal{C}_2$ such that such that $\mathcal{C}_2^2$ and $\mathcal{C}_2$ both send along the channel $y_\alpha$ or $\mathcal{F}_2$ sends along the channel $y_\alpha$. We have $(\mathcal{C}_1, \mathcal{D}_1, \mathcal{F}_1; \mathcal{C}_2^2, \mathcal{D}_2, \mathcal{F}_2) \in \mathcal{V}[\![\Delta \Vdash K]\!]^k_{y_\alpha; \cdot}$.*
6. *Consider $k \geq 1$ and $(\mathcal{C}_1, \mathcal{D}_1, \mathcal{F}_1; \mathcal{C}_2, \mathcal{D}_2, \mathcal{F}_2) \in \mathcal{V}[\![\Delta \Vdash K]\!]^k_{y_\alpha; \cdot}$ and $\mathcal{F}_2^2 \mapsto^* \mathcal{F}_2$ such that such that $\mathcal{F}_2^2$ and $\mathcal{F}_2$ both send along the channel $y_\alpha$ or $\mathcal{C}_2$ sends along the channel $y_\alpha$. We have $(\mathcal{C}_1, \mathcal{D}_1, \mathcal{F}_1; \mathcal{C}_2, \mathcal{D}_2, \mathcal{F}_2^2) \in \mathcal{V}[\![\Delta \Vdash K]\!]^k_{y_\alpha; \cdot}$.*

*Proof.* We prove statements 1-6 by a simultaneous lexicographic induction on $k$ and $\mathcal{V} < \mathcal{E}$. If $k = 0$, the proof is trivial. Consider $k = m + 1$.

1. By assumption, we have $(\mathcal{C}_1 \mathcal{D}_1 \mathcal{F}_1; \mathcal{C}_2 \mathcal{D}_2 \mathcal{F}_2) \in \mathcal{E}^{\xi}_{\Psi_0}[\![\Delta \Vdash K]\!]^{m+1}$. By line (13) of the logical relation we get

⋆ if $\mathcal{C}_1, \mathcal{D}_1, \mathcal{F}_1 \mapsto^{*\Theta, \Upsilon_1}_{\Delta \Vdash K} \mathcal{C}_1, \mathcal{D}_1', \mathcal{F}_1$, then
$\exists \mathcal{D}_2'. \mathcal{C}_2, \mathcal{D}_2, \mathcal{F}_2 \mapsto^{*\Theta''; \Upsilon_1}_{\Delta \Vdash K} \mathcal{C}_2, \mathcal{D}_2', \mathcal{F}_2$, and $\forall \mathcal{C}_1', \mathcal{F}_1', \mathcal{C}_2', \mathcal{F}_2'$.if $\mathcal{C}_1, \mathcal{D}_1', \mathcal{F}_1 \mapsto^{*\Theta_1, \Upsilon_1}_{\Delta \Vdash K} \mathcal{C}_1', \mathcal{D}_1', \mathcal{F}_1'$ and $\mathcal{C}_2, \mathcal{D}_2', \mathcal{F}_2 \mapsto^{*\Theta_2; \Upsilon_1}_{\Delta \Vdash K} \mathcal{C}_2', \mathcal{D}_2', \mathcal{F}_2'$, then
$\forall x \in \Upsilon_1. (\mathcal{C}_1', \mathcal{D}_1', \mathcal{F}_1'; \mathcal{C}_2', \mathcal{D}_2', \mathcal{F}_2') \in \mathcal{V}^{\xi}_{\Psi}[\![\Delta \Vdash K]\!]^{m+1}_{\cdot; x}$ and
$\forall x \in (\Theta_1 \cap \Theta_2). (\mathcal{C}_1', \mathcal{D}_1', \mathcal{F}_1'; \mathcal{C}_2', \mathcal{D}_2', \mathcal{F}_2') \in \mathcal{V}^{\xi}_{\Psi}[\![\Delta \Vdash K]\!]^{m+1}_{x; \cdot}$.

Consider $\mathcal{C}_2^2$, for which by the assumption we have $\mathcal{C}_2^2 \mapsto^* \mathcal{C}_2$. Our goal is to prove $(\mathcal{C}_1 \mathcal{D}_1 \mathcal{F}_1; \mathcal{C}_2^2 \mathcal{D}_2 \mathcal{F}_2) \in \mathcal{E}^{\xi}_{\Psi_0}[\![\Delta \Vdash K]\!]^{m+1}$. We need to show that

† if $\mathcal{C}_1, \mathcal{D}_1, \mathcal{F}_1 \mapsto^{*\Theta, \Upsilon_1}_{\Delta \Vdash K} \mathcal{C}_1, \mathcal{D}_1', \mathcal{F}_1$, then
$\exists \mathcal{D}_2'. \mathcal{C}_2^2, \mathcal{D}_2, \mathcal{F}_2 \mapsto^{*\Theta''; \Upsilon_1}_{\Delta \Vdash K} \mathcal{C}_2^2, \mathcal{D}_2', \mathcal{F}_2$, and $\forall \mathcal{C}_1', \mathcal{F}_1', \mathcal{C}_2', \mathcal{F}_2'$.if $\mathcal{C}_1, \mathcal{D}_1', \mathcal{F}_1 \mapsto^{*\Theta_1, \Upsilon_1}_{\Delta \Vdash K} \mathcal{C}_1', \mathcal{D}_1', \mathcal{F}_1'$ and $\mathcal{C}_2^2, \mathcal{D}_2', \mathcal{F}_2 \mapsto^{*\Theta_2; \Upsilon_1}_{\Delta \Vdash K} \mathcal{C}_2', \mathcal{D}_2', \mathcal{F}_2'$, then
$\forall x \in \Upsilon_1. (\mathcal{C}_1', \mathcal{D}_1', \mathcal{F}_1'; \mathcal{C}_2', \mathcal{D}_2', \mathcal{F}_2') \in \mathcal{V}^{\xi}_{\Psi}[\![\Delta \Vdash K]\!]^{m+1}_{\cdot; x}$ and
$\forall x \in (\Theta_1 \cap \Theta_2). (\mathcal{C}_1', \mathcal{D}_1', \mathcal{F}_1'; \mathcal{C}_2', \mathcal{D}_2', \mathcal{F}_2') \in \mathcal{V}^{\xi}_{\Psi}[\![\Delta \Vdash K]\!]^{m+1}_{x; \cdot}$.

We apply an if-Introduction on the goal followed by an if- Elimination on the assumption. Next, we apply an ∃-Elimination on the assumption to get $\mathcal{D}_2'$ that satisfies the conditions and use the same $\mathcal{D}_2'$ to instantiate the existential quantifier in the goal. We get the assumption

⋆′ $\forall \mathcal{C}_1', \mathcal{F}_1', \mathcal{C}_2', \mathcal{F}_2'$.if $\mathcal{C}_1, \mathcal{D}_1', \mathcal{F}_1 \mapsto^{*\Theta_1, \Upsilon_1}_{\Delta \Vdash K} \mathcal{C}_1', \mathcal{D}_1', \mathcal{F}_1'$ and $\mathcal{C}_2, \mathcal{D}_2', \mathcal{F}_2 \mapsto^{*\Theta_2; \Upsilon_1}_{\Delta \Vdash K} \mathcal{C}_2', \mathcal{D}_2', \mathcal{F}_2'$, then
$\forall x \in \Upsilon_1. (\mathcal{C}_1', \mathcal{D}_1', \mathcal{F}_1'; \mathcal{C}_2', \mathcal{D}_2', \mathcal{F}_2') \in \mathcal{V}^{\xi}_{\Psi}[\![\Delta \Vdash K]\!]^{m+1}_{\cdot; x}$ and
$\forall x \in (\Theta_1 \cap \Theta_2). (\mathcal{C}_1', \mathcal{D}_1', \mathcal{F}_1'; \mathcal{C}_2', \mathcal{D}_2', \mathcal{F}_2') \in \mathcal{V}^{\xi}_{\Psi}[\![\Delta \Vdash K]\!]^{m+1}_{x; \cdot}$.

and the goal

†′ $\forall \mathcal{C}_1', \mathcal{F}_1', \mathcal{C}_2', \mathcal{F}_2'$.if $\mathcal{C}_1, \mathcal{D}_1', \mathcal{F}_1 \mapsto^{*\Theta_1, \Upsilon_1}_{\Delta \Vdash K} \mathcal{C}_1', \mathcal{D}_1', \mathcal{F}_1'$ and $\mathcal{C}_2^2, \mathcal{D}_2', \mathcal{F}_2 \mapsto^{*\Theta_2; \Upsilon_1}_{\Delta \Vdash K} \mathcal{C}_2', \mathcal{D}_2', \mathcal{F}_2'$, then
$\forall x \in \Upsilon_1. (\mathcal{C}_1', \mathcal{D}_1', \mathcal{F}_1'; \mathcal{C}_2', \mathcal{D}_2', \mathcal{F}_2') \in \mathcal{V}^{\xi}_{\Psi}[\![\Delta \Vdash K]\!]^{m+1}_{\cdot; x}$ and
$\forall x \in (\Theta_1 \cap \Theta_2). (\mathcal{C}_1', \mathcal{D}_1', \mathcal{F}_1'; \mathcal{C}_2', \mathcal{D}_2', \mathcal{F}_2') \in \mathcal{V}^{\xi}_{\Psi}[\![\Delta \Vdash K]\!]^{m+1}_{x; \cdot}$.

We apply a ∀-Introduction and if-Introduction on the goal to get $\mathcal{C}_1', \mathcal{F}_1', \mathcal{C}_2', \mathcal{F}_2'$ that satisfies the conditions $\mathcal{C}_1, \mathcal{D}_1', \mathcal{F}_1 \mapsto^{*\Theta_1, \Upsilon_1}_{\Delta \Vdash K} \mathcal{C}_1', \mathcal{D}_1', \mathcal{F}_1'$ and $\mathcal{C}_2^2, \mathcal{D}_2', \mathcal{F}_2 \mapsto^{*\Theta_2; \Upsilon_1}_{\Delta \Vdash K} \mathcal{C}_2', \mathcal{D}_2', \mathcal{F}_2'$. By Lemma C.7 there exists $\mathcal{C}_2^1$ such that $\mathcal{C}_2 \mapsto^*_{\Theta_2'} \mathcal{C}_2^1$ and $\mathcal{C}_2' \mapsto^* \mathcal{C}_2^1$. We use $\mathcal{C}_1'$ and $\mathcal{F}_1'$ and $\mathcal{C}_2^1$ and $\mathcal{F}_2'$ to instantiate the ∀ quantifier in the assumption. We get the assumption

⋆″ $\forall x \in \Upsilon_1. (\mathcal{C}_1', \mathcal{D}_1', \mathcal{F}_1'; \mathcal{C}_2^1, \mathcal{D}_2', \mathcal{F}_2') \in \mathcal{V}^{\xi}_{\Psi}[\![\Delta \Vdash K]\!]^{m+1}_{\cdot; x}$ and
$\forall x \in (\Theta_1 \cap \Theta_2'). (\mathcal{C}_1', \mathcal{D}_1', \mathcal{F}_1'; \mathcal{C}_2^1, \mathcal{D}_2', \mathcal{F}_2') \in \mathcal{V}^{\xi}_{\Psi}[\![\Delta \Vdash K]\!]^{m+1}_{x; \cdot}$.

16

(1) $(\cdot, \mathcal{D}_1, \mathcal{F}_1; \cdot, \mathcal{D}_1, \mathcal{F}_1) \in$     $(\cdot, \mathcal{D}_1, \mathcal{F}_1; \cdot, \mathcal{D}_2, \mathcal{F}_2) \in \mathsf{Tree}_\Psi(\cdot \Vdash y_\alpha)$ and
$\mathcal{V}_\Psi^\xi [\![ \cdot \Vdash y_\alpha{:}1[c] ]\!]_{\cdot; y^\alpha}^{m+1}$     $\mathcal{D}_1 = \mathbf{msg}(\mathbf{close}\ y_\alpha^c)$ and $\mathcal{D}_2 = \mathbf{msg}(\mathbf{close}\ y_\alpha^c)$

(2) $(\mathcal{C}_1, \mathcal{D}_1, \mathcal{F}_1; \mathcal{C}_2, \mathcal{D}_2, \mathcal{F}_2) \in$     $(\mathcal{C}_1, \mathcal{D}_1, \mathcal{F}_1; \mathcal{C}_2, \mathcal{D}_2, \mathcal{F}_2) \in \mathsf{Tree}_\Psi(\Delta \Vdash y_\alpha : \oplus\{\ell{:}A_\ell\}_{\ell \in I}[c])$ and
$\mathcal{V}_\Psi^\xi [\![ (\Delta \Vdash y_\alpha{:} \oplus \{\ell{:}A_\ell\}_{\ell \in I}[c]) ]\!]_{\cdot; y^\alpha}^{m+1}$     $\mathcal{D}_1 = \mathcal{D}_1'\mathbf{msg}(y_\alpha^c.k; y_\alpha^c \leftarrow u_\delta^c)$ and $\mathcal{D}_2 = \mathcal{D}_2'\mathbf{msg}(y_\alpha^c.k; y_\alpha^c \leftarrow u_\delta^c)$ and
    $(\mathcal{C}_1\mathcal{D}_1'\mathbf{msg}(y_\alpha^c.k; y_\alpha^c \leftarrow u_\delta^c)\mathcal{F}_1, \mathcal{C}_2\mathcal{D}_2'\mathbf{msg}(y_\alpha^c.k; y_\alpha^c \leftarrow u_\delta^c)\mathcal{F}_2) \in \mathcal{E}_\Psi^\xi [\![ \Delta \Vdash u_\delta{:}A_k[c] ]\!]^m$

(3) $(\mathcal{C}_1, \mathcal{D}_1, \mathcal{F}_1; \mathcal{C}_2, \mathcal{D}_2, \mathcal{F}_2) \in$     $(\mathcal{C}_1, \mathcal{D}_1, \mathcal{F}_1; \mathcal{C}_2, \mathcal{D}_2, \mathcal{F}_2) \in \mathsf{Tree}_\Psi(\Delta \Vdash y_\alpha{:}\&\{\ell{:}A_\ell\}_{\ell\in I}[c])$, and
$\mathcal{V}_\Psi^\xi [\![ \Delta \Vdash y_\alpha{:}\&\{\ell{:}A_\ell\}_{\ell\in I}[c] ]\!]_{y_\mathcal{F}^\alpha;\cdot}^{m+1}$     $\mathcal{F}_1 = \mathbf{msg}(y_\alpha^c.k; u_\delta^c \leftarrow y_\alpha^c)\mathcal{F}_1'$ and if $\mathcal{F}_2 = \mathbf{msg}(y_\alpha^c.k; u_\delta^c \leftarrow y_\alpha^c)\mathcal{F}_2'$ then
    $(\mathcal{C}_1\mathcal{D}_1\mathbf{msg}(y_\alpha^c.k; u_\delta^c \leftarrow y_\alpha^c)\mathcal{F}_1', \mathcal{C}_2\mathcal{D}_2\mathbf{msg}(y_\alpha^c.k; u_\delta^c \leftarrow y_\alpha^c)\mathcal{F}_2') \in \mathcal{E}_\Psi^\xi [\![ \Delta \Vdash u_\delta{:}A_k[c] ]\!]^m$

(4) $(\mathcal{C}_1'\mathcal{C}_1'', \mathcal{D}_1, \mathcal{F}_1; \mathcal{C}_2'\mathcal{C}_2'', \mathcal{D}_2, \mathcal{F}_2) \in$     $\mathcal{C}_i' \in \mathsf{Tree}_\Psi(\cdot \Vdash \Delta')$ and $(\mathcal{C}_1'\mathcal{C}_1'', \mathcal{D}_1, \mathcal{F}_1; \mathcal{C}_2'\mathcal{C}_2'', \mathcal{D}_2, \mathcal{F}_2) \in \mathsf{Tree}_\Psi(\Delta', \Delta'' \Vdash y_\alpha{:}A \otimes B[c])$ and
$\mathcal{V}_\Psi^\xi [\![ \Delta', \Delta'' \Vdash y_\alpha{:}A \otimes B[c] ]\!]_{\cdot; y^\alpha}^{m+1}$     $\mathcal{D}_1 = \mathcal{D}_1'\mathcal{T}_1\mathbf{msg}(\mathbf{send}\,x_\beta^c\,y_\alpha^c; y_\alpha^c \leftarrow u_\delta^c)$ for $\mathcal{T}_1 \in \mathsf{Tree}_\Psi(\Delta'' \Vdash x_\beta{:}A[c])$ and
    $\mathcal{D}_2 = \mathcal{D}_2'\mathcal{T}_2\mathbf{msg}(\mathbf{send},\,x_\beta^c\,y_\alpha^c; y_\alpha^c \leftarrow u_\delta^c)$ for $\mathcal{T}_2 \in \mathsf{Tree}_\Psi(\Delta'' \Vdash x_\beta{:}A[c])$ and
    $(\mathcal{C}_1''\mathcal{T}_1\mathcal{C}_1'\mathcal{D}_1'\mathbf{msg}(\mathbf{send}\,x_\beta^c\,y_\alpha^c; y_\alpha^c \leftarrow u_\delta^c)\mathcal{F}_1,$
       $\mathcal{C}_2''\mathcal{T}_2\mathcal{C}_2'\mathcal{D}_2'\mathbf{msg}(\mathbf{send}\,x_\beta^c\,y_\alpha^c; y_\alpha^c \leftarrow u_\delta^c)\mathcal{F}_2) \in \mathcal{E}_\Psi^\xi [\![ \Delta'' \Vdash x_\beta{:}A[c] ]\!]^m$ and
    $(\mathcal{C}_1'\mathcal{D}_1'\mathcal{C}_1''\mathcal{T}_1\mathbf{msg}(\mathbf{send}\,x_\beta^c\,y_\alpha^c; y_\alpha^c \leftarrow u_\delta^c)\mathcal{F}_1,$
       $\mathcal{C}_2'\mathcal{D}_2'\mathcal{C}_2''\mathcal{T}_2\mathbf{msg}(\mathbf{send}\,x_\beta^c\,y_\alpha^c; y_\alpha^c \leftarrow u_\delta^c)\mathcal{F}_2) \in \mathcal{E}_\Psi^\xi [\![ \Delta' \Vdash u_\delta{:}B[c] ]\!]^m$

(5) $(\mathcal{C}_1, \mathcal{D}_1, \mathcal{F}_1; \mathcal{C}_2, \mathcal{D}_2, \mathcal{F}_2) \in$     $(\mathcal{C}_1, \mathcal{D}_1, \mathcal{F}_1; \mathcal{C}_2, \mathcal{D}_2, \mathcal{F}_2) \in \mathsf{Tree}_\Psi(\Delta \Vdash y_\alpha{:}A \multimap B[c])$ and
$\mathcal{V}_\Psi^\xi [\![ \Delta \Vdash y_\alpha{:}A \multimap B[c] ]\!]_{y_\mathcal{F}^\alpha;\cdot}^{m+1}$     $\mathcal{F}_1 = \mathcal{T}_1\mathbf{msg}(\mathbf{send}\,x_\beta^c\,y_\alpha^c; u_\delta^c \leftarrow y_\alpha^c)\mathcal{F}_1'$ for $\mathcal{T}_1 \in \mathsf{Tree}_\Psi(\cdot \Vdash x_\beta{:}A[c])$ and if
    $\mathcal{F}_2 = \mathcal{T}_2\mathbf{msg}(\mathbf{send}\,x_\beta^c\,y_\alpha^c; u_\delta^c \leftarrow y_\alpha^c)\mathcal{F}_2'$ for $\mathcal{T}_2 \in \mathsf{Tree}_\Psi(\cdot \Vdash x_\beta{:}A[c])$ then
    $(\mathcal{C}_1\mathcal{T}_1\mathcal{D}_1\mathbf{msg}(\mathbf{send}\,x_\beta^c\,y_\alpha^c; u_\delta^c \leftarrow y_\alpha^c)\mathcal{F}_1',$
       $\mathcal{C}_2\mathcal{T}_2\mathcal{D}_2\mathbf{msg}(\mathbf{send}\,x_\beta^c\,y_\alpha^c; u_\delta^c \leftarrow y_\alpha^c)\mathcal{F}_2') \in \mathcal{E}_\Psi^\xi [\![ \Delta, x_\beta{:}A[c] \Vdash u_\delta{:}B[c] ]\!]^m$

(6) $(\mathcal{C}_1, \mathcal{F}_1, \mathcal{D}_1; \mathcal{C}_2, \mathcal{F}_2, \mathcal{D}_2) \in$     $(\mathcal{C}_1, \mathcal{F}_1, \mathcal{D}_1; \mathcal{C}_2, \mathcal{F}_2, \mathcal{D}_2) \in \mathsf{Tree}_\Psi(\Delta, y_\alpha{:}1[c] \Vdash K)$ and
$\mathcal{V}_\Psi^\xi [\![ \Delta, y_\alpha{:}1[c] \Vdash K ]\!]_{y_\mathcal{C}^\alpha;\cdot}^{m+1}$     $\mathcal{C}_1 = \mathcal{C}_1'\mathbf{msg}(\mathbf{close}\ y_\alpha^c)$ and if $\mathcal{C}_2 = \mathcal{C}_2'\mathbf{msg}(\mathbf{close}\ y_\alpha^c)$ then
    $(\mathcal{C}_1'\mathbf{msg}(\mathbf{close}\ y_\alpha^c)\mathcal{D}_1\mathcal{F}_1, \mathcal{C}_2'\mathbf{msg}(\mathbf{close}\ y_\alpha^c)\mathcal{D}_2\mathcal{F}_2)\mathcal{E}_\Psi^\xi [\![ \Delta \Vdash K ]\!]^m$

(7) $(\mathcal{C}_1, \mathcal{D}_1, \mathcal{F}_1; \mathcal{C}_2, \mathcal{D}_2, \mathcal{F}_2) \in$     $(\mathcal{C}_1, \mathcal{D}_1, \mathcal{F}_1; \mathcal{C}_2, \mathcal{D}_2, \mathcal{F}_2) \in \mathsf{Tree}_\Psi(\Delta, y_\alpha{:}\oplus\{\ell{:}A_\ell\}_{\ell\in I}[c] \Vdash K)$ and
$\mathcal{V}_\Psi^\xi [\![ \Delta, y_\alpha : \oplus\{\ell{:}A_\ell\}_{\ell\in I}[c] \Vdash K ]\!]_{y_\mathcal{C}^\alpha;\cdot}^{m+1}$     $\mathcal{C}_1 = \mathcal{C}_1'\mathbf{msg}(y_\alpha^c.k; y_\alpha^c \leftarrow u_\delta^c)$ and if $\mathcal{C}_2 = \mathcal{C}_2'\mathbf{msg}(y_\alpha^c.k; y_\alpha^c \leftarrow u_\delta^c)$ then
    $(\mathcal{C}_1'\mathbf{msg}(y_\alpha^c.k; y_\alpha^c \leftarrow u_\delta^c)\mathcal{D}_1\mathcal{F}_1, \mathcal{C}_2'\mathbf{msg}(y_\alpha^c.k; y_\alpha^c \leftarrow u_\delta^c)\mathcal{D}_2\mathcal{F}_2)$
       $\in \mathcal{E}_\Psi^\xi [\![ \Delta, u_\delta{:}A_k[c] \Vdash K ]\!]^m$

(8) $(\mathcal{C}_1, \mathcal{D}_1, \mathcal{F}_1; \mathcal{C}_2, \mathcal{D}_2, \mathcal{F}_2) \in$     $(\mathcal{C}_1, \mathcal{D}_1, \mathcal{F}_1; \mathcal{C}_2, \mathcal{D}_2, \mathcal{F}_2) \in \mathsf{Tree}_\Psi(\Delta, y_\alpha{:}\&\{\ell{:}A_\ell\}_{\ell\in I}[c] \Vdash K)$ and
$\mathcal{V}_\Psi^\xi [\![ \Delta, y_\alpha{:}\&\{\ell{:}A_\ell\}_{\ell\in I}[c] \Vdash K ]\!]_{\cdot; y^\alpha}^{m+1}$     $\mathcal{D}_1 = \mathbf{msg}(y_\alpha^c.k; u_\delta^c \leftarrow y_\alpha^c)\mathcal{D}_1'$ and $\mathcal{D}_2 = \mathbf{msg}(y_\alpha^c.k; u_\delta^c \leftarrow y_\alpha^c)\mathcal{D}_2'$ and
    $(\mathcal{C}_1\mathbf{msg}(y_\alpha^c.k; u_\delta^c \leftarrow y_\alpha^c)\mathcal{D}_1'\mathcal{F}_1,$
       $\mathcal{C}_2\mathbf{msg}(y_\alpha^c.k; u_\delta^c \leftarrow y_\alpha^c)\mathcal{D}_2'\mathcal{F}_2) \in \mathcal{E}_\Psi^\xi [\![ \Delta, u_\delta{:}A_k[c] \Vdash K ]\!]^m$

(9) $(\mathcal{C}_1, \mathcal{D}_1, \mathcal{F}_1; \mathcal{C}_2, \mathcal{D}_2, \mathcal{F}_2) \in$     $(\mathcal{C}_1, \mathcal{D}_1, \mathcal{F}_1; \mathcal{C}_2, \mathcal{D}_2, \mathcal{F}_2) \in \mathsf{Tree}_\Psi(\Delta, y_\alpha{:}A \otimes B[c] \Vdash K)$ and
$\mathcal{V}_\Psi^\xi [\![ \Delta, y_\alpha{:}A \otimes B[c] \Vdash K ]\!]_{y_\mathcal{C}^\alpha;\cdot}^{m+1}$     $\mathcal{C}_1 = \mathcal{C}_1'\mathbf{msg}(\mathbf{send}\,x_\beta^c\,y_\alpha^c; y_\alpha^c \leftarrow u_\delta^c)$ and if
    $\mathcal{C}_2 = \mathcal{C}_2'\mathbf{msg}(\mathbf{send}\,x_\beta^c\,y_\alpha^c; y_\alpha^c \leftarrow u_\delta^c$ then
       $(\mathcal{C}_1'\mathbf{msg}(\mathbf{send}\,x_\beta^c\,y_\alpha^c; y_\alpha^c \leftarrow u_\delta^c)\mathcal{D}_1\mathcal{F}_1,$
       $\mathcal{C}_2'\mathbf{msg}(\mathbf{send}\,x_\beta^c\,y_\alpha^c; y_\alpha^c \leftarrow u_\delta^c)\mathcal{D}_2\mathcal{F}_2) \in \mathcal{E}_\Psi^\xi [\![ \Delta, x_\beta{:}A[c], u_\delta{:}B[c] \Vdash K ]\!]^m$

(10) $(\mathcal{C}_1'\mathcal{C}_1'', \mathcal{D}_1, \mathcal{F}_1; \mathcal{C}_2'\mathcal{C}_2'', \mathcal{D}_2, \mathcal{F}_2) \in$     $\mathcal{C}_i'' \in \mathsf{Tree}_\Psi(\cdot \Vdash \Delta')$ and $(\mathcal{C}_1'\mathcal{C}_1'', \mathcal{D}_1, \mathcal{F}_1; \mathcal{C}_2'\mathcal{C}_2'', \mathcal{D}_2, \mathcal{F}_2) \in \mathsf{Tree}_\Psi(\Delta', \Delta'', y_\alpha{:}A \multimap B[c] \Vdash K)$ and
$\mathcal{V}_\Psi^\xi [\![ \Delta', \Delta'' , y_\alpha{:}A \multimap B[c] \Vdash K ]\!]_{\cdot; y^\alpha}^{m+1}$     $\mathcal{D}_1 = \mathcal{T}_1\mathbf{msg}(\mathbf{send}\,x_\beta^c\,y_\alpha^c; u_\delta^c \leftarrow y_\alpha^c)\,\mathcal{D}_1''$ and for $\mathcal{T}_1 \in \mathsf{Tree}_\Psi(\Delta' \Vdash x_\beta{:}A[c])$ and
    $\mathcal{D}_2 = \mathcal{T}_2\mathbf{msg}(\mathbf{send}\,x_\beta^c\,y_\alpha^c; u_\delta^c \leftarrow y_\alpha^c)\,\mathcal{D}_2''$ and for $\mathcal{T}_2 \in \mathsf{Tree}_\Psi(\Delta' \Vdash x_\beta{:}A[c])$ and
    $(\mathcal{C}_1''\mathcal{T}_1\mathcal{C}_1'\mathbf{msg}(\mathbf{send}\,x_\beta^c\,y_\alpha^c; u_\delta^c \leftarrow y_\alpha^c)\mathcal{D}_1''\mathcal{F}_1,$
       $\mathcal{C}_2''\mathcal{T}_2\mathcal{C}_2'\mathbf{msg}(\mathbf{send}\,x_\beta^c\,y_\alpha^c; u_\delta^c \leftarrow y_\alpha^c)\mathcal{D}_2''\mathcal{F}_2) \in \mathcal{E}_\Psi^\xi [\![ \Delta' \Vdash x_\beta{:}A[c] ]\!]^m$ and
    $(\mathcal{C}_1'\mathcal{C}_1''\mathcal{T}_1\mathbf{msg}(\mathbf{send}\,x_\beta^c\,y_\alpha^c; u_\delta^c \leftarrow y_\alpha^c)\mathcal{D}_1''\mathcal{F}_1,$
       $\mathcal{C}_2'\mathcal{C}_2''\mathcal{T}_2\mathbf{msg}(\mathbf{send}\,x_\beta^c\,y_\alpha^c; u_\delta^c \leftarrow y_\alpha^c)\mathcal{D}_2''\mathcal{F}_2) \in \mathcal{E}_\Psi^\xi [\![ \Delta'', u_\delta{:}B[c] \Vdash K ]\!]^m$

(11) $(\mathcal{B}_1, \mathcal{B}_2) \in \mathcal{E}_\Psi^\xi [\![ \Delta \Vdash K ]\!]^{m+1}$     $\mathcal{B}_1 = \mathcal{C}_1\mathcal{D}_1\mathcal{F}_1$ and $\mathcal{B}_2 = \mathcal{C}_2\mathcal{D}_2\mathcal{F}_2$ and $(\mathcal{C}_1, \mathcal{D}_1, \mathcal{F}_1; \mathcal{C}_2, \mathcal{D}_2, \mathcal{F}_2) \in \mathsf{Tree}_\Psi(\Delta \Vdash K)$, and
    if $\mathcal{C}_1, \mathcal{D}_1, \mathcal{F}_1 \mapsto^{*\Theta;\Upsilon_1}_{\Delta \Vdash K} \mathcal{C}_1, \mathcal{D}_1', \mathcal{F}_1$ then $\exists \mathcal{D}_2'. \mathcal{C}_2, \mathcal{D}_2, \mathcal{F}_2 \mapsto^{*\Theta';\Upsilon_1}_{\Delta \Vdash K} \mathcal{C}_2, \mathcal{D}_2', \mathcal{F}_2$, and
    $\forall \mathcal{C}_1', \mathcal{F}_1', \mathcal{C}_2', \mathcal{F}_2'$ if $\mathcal{C}_1, \mathcal{D}_1', \mathcal{F}_1 \mapsto^{*\Theta_1;\Upsilon_1}_{\Delta \Vdash K} \mathcal{C}_1', \mathcal{D}_1', \mathcal{F}_1'$ and $\mathcal{C}_2, \mathcal{D}_2', \mathcal{F}_2 \mapsto^{*\Theta_2;\Upsilon_1}_{\Delta \Vdash K} \mathcal{C}_2', \mathcal{D}_2', \mathcal{F}_2'$, then
    $\forall x \in \Upsilon_1. (\mathcal{C}_1', \mathcal{D}_1', \mathcal{F}_1'; \mathcal{C}_2', \mathcal{D}_2', \mathcal{F}_2') \in \mathcal{V}_\Psi^\xi [\![ \Delta \Vdash K ]\!]_{\cdot;x}^{m+1}$ and
    $\forall x \in (\Theta_1 \cap \Theta_2). (\mathcal{C}_1', \mathcal{D}_1', \mathcal{F}_1'; \mathcal{C}_2', \mathcal{D}_2', \mathcal{F}_2') \in \mathcal{V}_\Psi^\xi [\![ \Delta \Vdash K ]\!]_{x;\cdot}^{m+1}.$

(12) $(\mathcal{B}_1, \mathcal{B}_2) \in \mathcal{E}_\Psi^\xi [\![ \Delta \Vdash K ]\!]^0$     $\mathcal{B}_1 = \mathcal{C}_1\mathcal{D}_1\mathcal{F}_1$ and $\mathcal{B}_2 = \mathcal{C}_2\mathcal{D}_2\mathcal{F}_2$ and $(\mathcal{C}_1, \mathcal{D}_1, \mathcal{F}_1; \mathcal{C}_2, \mathcal{D}_2, \mathcal{F}_2) \in \mathsf{Tree}_\Psi(\Delta \Vdash K)$.

Figure 2: Session Logical Relation for Noninterference

and the goal
$$\dagger'' \, \forall x \in \Upsilon_1. \, (\mathcal{C}'_1, \mathcal{D}'_1, \mathcal{F}'_1; \mathcal{C}'_2, \mathcal{D}'_2, \mathcal{F}'_2) \in \mathcal{V}^\xi_\Psi [\![\Delta \Vdash K]\!]^{m+1}_{\cdot;x} \text{ and}$$
$$\forall x \in (\Theta_1 \cap \Theta_2). \, (\mathcal{C}'_1, \mathcal{D}'_1, \mathcal{F}'_1; \mathcal{C}'_2, \mathcal{D}'_2, \mathcal{F}'_2) \in \mathcal{V}^\xi_\Psi [\![\Delta \Vdash K]\!]^{m+1}_{x;\cdot}.$$

Having $\mathcal{C}'_2 \mapsto^* \mathcal{C}^1_2$, we can observe that $\Theta_1 \cap \Theta_2 \subseteq \Theta_1 \cap \Theta'_2$. Thus we can apply the induction hypothesis to get the goal from the assumption.

2. The proof is similar to the previous statement.

3. The proof is by considering the row of the logical relation that ensures the assumption $(\mathcal{C}_1, \mathcal{D}_1, \mathcal{F}_1; \mathcal{C}_2, \mathcal{D}_2, \mathcal{F}_2) \in \mathcal{V}[\![\Delta \Vdash K]\!]^k_{\cdot;y_\alpha}$. Here we only consider two interesting cases, the proof of other cases is similar.

   **Row 4.** By the conditions of this row, we know that $\Delta = \Delta', \Delta''$, and $K = y_\alpha{:}A \otimes B[c]$. Moreover, $\mathcal{C}_1 = \mathcal{C}'_1 \mathcal{C}''_1$ and $\mathcal{C}_2 = \mathcal{C}'_2 \mathcal{C}''_2$ with $\mathcal{C}'_i \in \mathsf{Tree}_\Psi(\cdot \Vdash \Delta')$ and $\mathcal{C}''_i \in \mathsf{Tree}_\Psi(\cdot \Vdash \Delta'')$. We have
   $$\mathcal{D}_1 = \mathcal{D}'_1 \mathcal{T}_1 \mathbf{msg}(\mathbf{send}x^c_\beta \, y^c_\alpha; y^c_\alpha \leftarrow u^c_\delta) \text{ for } \mathcal{T}_1 \in \mathsf{Tree}_\Psi(\Delta'' \Vdash x_\beta{:}A[c])$$
   and
   $$\mathcal{D}_2 = \mathcal{D}'_2 \mathcal{T}_2 \mathbf{msg}(\mathbf{send}, x^c_\beta \, y^c_\alpha; y^c_\alpha \leftarrow u^c_\delta) \text{ for } \mathcal{T}_2 \in \mathsf{Tree}_\Psi(\Delta'' \Vdash x_\beta{:}A[c])$$
   and
   $$\dagger_1 \, (\mathcal{C}''_1 \mathcal{T}_1 \mathcal{C}'_1 \mathcal{D}'_1 \mathbf{msg}(\mathbf{send}x^c_\beta \, y^c_\alpha; y^c_\alpha \leftarrow u^c_\delta) \mathcal{F}_1, \mathcal{C}''_2 \mathcal{T}_2 \mathcal{C}'_2 \mathcal{D}'_2 \mathbf{msg}(\mathbf{send}x^c_\beta \, y^c_\alpha; y^c_\alpha \leftarrow u^c_\delta) \mathcal{F}_2) \in \mathcal{E}^\xi_\Psi [\![\Delta'' \Vdash x_\beta{:}A[c]]\!]^m$$
   and
   $$\dagger_2 \, (\mathcal{C}'_1 \mathcal{D}'_1 \mathcal{C}''_1 \mathcal{T}_1 \mathbf{msg}(\mathbf{send}x^c_\beta \, y^c_\alpha; y^c_\alpha \leftarrow u^c_\delta) \mathcal{F}_1, \mathcal{C}'_2 \mathcal{D}'_2 \mathcal{C}''_2 \mathcal{T}_2 \mathbf{msg}(\mathbf{send}x^c_\beta \, y^c_\alpha; y^c_\alpha \leftarrow u^c_\delta) \mathcal{F}_2) \in \mathcal{E}^\xi_\Psi [\![\Delta' \Vdash u_\delta{:}B[c]]\!]^m$$

   These are also the statements that we need to prove when replacing $\mathcal{C}_2$ with $\mathcal{C}^2_2$.
   By the typing rules, we know that $\mathcal{C}^2_2 = \mathcal{C}^{2'}_2 \mathcal{C}^{2''}_2$ with $\mathcal{C}^{2'}_2 \in \mathsf{Tree}_\Psi(\cdot \Vdash \Delta')$ and $\mathcal{C}^{2''}_2 \in \mathsf{Tree}_\Psi(\cdot \Vdash \Delta'')$. Moreover, since $\mathcal{C}^{2'}_2$ and $\mathcal{C}^{2''}_2$ are separate trees, we have $\mathcal{C}^{2'}_2 \mapsto^* \mathcal{C}'_2$ and $\mathcal{C}^{2''}_2 \mapsto^* \mathcal{C}''_2$.
   From this we get
   $$\mathcal{C}^{2'}_2 \mathcal{D}'_2 \mathbf{msg}(\mathbf{send}x^c_\beta \, y^c_\alpha; y^c_\alpha \leftarrow u^c_\delta) \mathcal{F}_2 \mapsto^* \mathcal{C}'_2 \mathcal{D}'_2 \mathbf{msg}(\mathbf{send}x^c_\beta \, y^c_\alpha; y^c_\alpha \leftarrow u^c_\delta) \mathcal{F}_2$$
   and
   $$\mathcal{C}^{2''}_2 \mathcal{T}_2 \mathbf{msg}(\mathbf{send}x^c_\beta \, y^c_\alpha; y^c_\alpha \leftarrow u^c_\delta) \mathcal{F}_2 \mapsto^* \mathcal{C}''_2 \mathcal{T}_2 \mathbf{msg}(\mathbf{send}x^c_\beta \, y^c_\alpha; y^c_\alpha \leftarrow u^c_\delta) \mathcal{F}_2.$$

   Now we can apply the induction hypothesis on $\dagger_1$ to get
   $$\dagger'_1 \, (\mathcal{C}''_1 \mathcal{T}_1 \mathcal{C}'_1 \mathcal{D}'_1 \mathbf{msg}(\mathbf{send}x^c_\beta \, y^c_\alpha; y^c_\alpha \leftarrow u^c_\delta) \mathcal{F}_1, \mathcal{C}^{2''}_2 \mathcal{T}_2 \mathcal{C}'_2 \mathcal{D}'_2 \mathbf{msg}(\mathbf{send}x^c_\beta \, y^c_\alpha; y^c_\alpha \leftarrow u^c_\delta) \mathcal{F}_2) \in \mathcal{E}^\xi_\Psi [\![\Delta'' \Vdash x_\beta{:}A[c]]\!]^m.$$

   We again apply the induction hypothesis on $\dagger'_1$ to get
   $$\dagger''_1 \, (\mathcal{C}''_1 \mathcal{T}_1 \mathcal{C}'_1 \mathcal{D}'_1 \mathbf{msg}(\mathbf{send}x^c_\beta \, y^c_\alpha; y^c_\alpha \leftarrow u^c_\delta) \mathcal{F}_1, \mathcal{C}^{2''}_2 \mathcal{T}_2 \mathcal{C}^{2'}_2 \mathcal{D}'_2 \mathbf{msg}(\mathbf{send}x^c_\beta \, y^c_\alpha; y^c_\alpha \leftarrow u^c_\delta) \mathcal{F}_2) \in \mathcal{E}^\xi_\Psi [\![\Delta'' \Vdash x_\beta{:}A[c]]\!]^m.$$

   Similarly we apply the induction hypothesis on $\dagger_2$ (in two steps) to get
   $$\dagger''_2 \, (\mathcal{C}'_1 \mathcal{D}'_1 \mathcal{C}''_1 \mathcal{T}_1 \mathbf{msg}(\mathbf{send}x^c_\beta \, y^c_\alpha; y^c_\alpha \leftarrow u^c_\delta) \mathcal{F}_1, \mathcal{C}^{2'}_2 \mathcal{D}'_2 \mathcal{C}^{2''}_2 \mathcal{T}_2 \mathbf{msg}(\mathbf{send}x^c_\beta \, y^c_\alpha; y^c_\alpha \leftarrow u^c_\delta) \mathcal{F}_2) \in \mathcal{E}^\xi_\Psi [\![\Delta' \Vdash u_\delta{:}B[c]]\!]^m$$

   and the proof of this subcase is complete.

   **Row 11.** By the conditions of this row, we know that
   $$\mathcal{D}_1 = \mathcal{D}'_1 \mathbf{proc}(y_\alpha[c], y^c_\alpha \leftarrow x^d_\beta @ d_1) \text{ and } \mathcal{D}_2 = \mathcal{D}'_2 \mathbf{proc}(y_\alpha[c], y^c_\alpha \leftarrow x^d_\beta @ d_2)$$
   and
   $$\dagger_1 \, ([x^d_\beta/y^c_\alpha]\mathcal{C}_1 \mathcal{D}'_1 \mathcal{F}_1, [x^d_\beta/y^c_\alpha]\mathcal{C}_2 \mathcal{D}'_2 \mathcal{F}_2) \in \mathcal{E}^\xi_\Psi [\![\Delta \Vdash x_\beta{:}A[c]]\!]^m.$$

   These are also the statements that we need to prove when replacing $\mathcal{C}_2$ with $\mathcal{C}^2_2$. It is important to observe that $y^c_\beta$ does not occur in $\mathcal{C}_2$ and $\mathcal{C}^2_2$.
   So we can apply the induction hypothesis on $\dagger_1$ to get
   $$\dagger_1 \, ([x^d_\beta/y^c_\alpha]\mathcal{C}_1 \mathcal{D}'_1 \mathcal{F}_1, [x^d_\beta/y^c_\alpha]\mathcal{C}^2_2 \mathcal{D}'_2 \mathcal{F}_2) \in \mathcal{E}^\xi_\Psi [\![\Delta \Vdash x_\beta{:}A[c]]\!]^m.$$

   and the proof of this subcase is complete.
   **Rows 1,2,8,10.** Similar to the previous cases.

4. The proof is similar to the previous statement.

5. The proof is by considering the row of the logical relation that ensures the assumption $(\mathcal{C}_1, \mathcal{D}_1, \mathcal{F}_1; \mathcal{C}_2, \mathcal{D}_2, \mathcal{F}_2) \in \mathcal{V}[\![\Delta \Vdash K]\!]^k_{y_\alpha;.}$.
   Here we only consider three interesting cases, the proof of other cases is similar.
   **Row 5.** By the conditions of this row, we know that $K = y_\alpha{:}A \multimap B[c]$. We have $(\mathcal{C}_1, \mathcal{D}_1, \mathcal{F}_1; \mathcal{C}_2, \mathcal{D}_2, \mathcal{F}_2) \in \mathsf{Tree}_\Psi(\Delta \Vdash y_\alpha{:}A \multimap B[c])$ and $\mathcal{F}_1 = \mathcal{T}_1 \mathbf{msg}(\mathbf{send} x^c_\beta \, y^c_\alpha; u^c_\delta \leftarrow y^c_\alpha)\mathcal{F}'_1$ for $\mathcal{T}_1 \in \mathsf{Tree}_\Psi(\cdot \Vdash x_\beta{:}A[c])$. Moreover, we have

   $$\text{if } \mathcal{F}_2 = \mathcal{T}_2 \mathbf{msg}(\mathbf{send} x^c_\beta \, y^c_\alpha; u^c_\delta \leftarrow y^c_\alpha)\mathcal{F}'_2 \text{ for } \mathcal{T}_2 \in \mathsf{Tree}_\Psi(\cdot \Vdash x_\beta{:}A[c]) \text{ then}$$

   $$\dagger \, (\mathcal{C}_1 \mathcal{T}_1 \mathcal{D}_1 \mathbf{msg}(\mathbf{send} x^c_\beta \, y^c_\alpha; u^c_\delta \leftarrow y^c_\alpha)\mathcal{F}'_1, \mathcal{C}_2 \mathcal{T}_2 \mathcal{D}_2 \mathbf{msg}(\mathbf{send} x^c_\beta \, y^c_\alpha; u^c_\delta \leftarrow y^c_\alpha)\mathcal{F}'_2) \in \mathcal{E}^\xi_\Psi[\![\Delta, x_\beta{:}A[c] \Vdash u_\delta{:}B[c]]\!]^m$$

   These are also the statements that we need to prove when replacing $\mathcal{C}_2$ with $\mathcal{C}^2_2$. The first tree statements are straightforward by type preservation. Now assume $\mathcal{F}_2 = \mathcal{T}_2 \mathbf{msg}(\mathbf{send} x^c_\beta \, y^c_\alpha; u^c_\delta \leftarrow y^c_\alpha)\mathcal{F}'_2$ for $\mathcal{T}_2 \in \mathsf{Tree}_\Psi(\cdot \Vdash x_\beta{:}A[c])$. By assumption, we know that $\mathcal{C}^2_2 \mathcal{T}_2 \mapsto^* \mathcal{C}_2 \mathcal{T}_2$. We can apply the induction hypothesis on $\dagger$ to get

   $$\dagger' \, (\mathcal{C}_1 \mathcal{T}_1 \mathcal{D}_1 \mathbf{msg}(\mathbf{send} x^c_\beta \, y^c_\alpha; u^c_\delta \leftarrow y^c_\alpha)\mathcal{F}'_1, \mathcal{C}^2_2 \mathcal{T}_2 \mathcal{D}_2 \mathbf{msg}(\mathbf{send} x^c_\beta \, y^c_\alpha; u^c_\delta \leftarrow y^c_\alpha)\mathcal{F}'_2) \in \mathcal{E}^\xi_\Psi[\![\Delta, x_\beta{:}A[c] \Vdash u_\delta{:}B[c]]\!]^m$$

   which completes the proof of this case.
   **Row 9.** By the conditions of this row, we know that $\Delta = \Delta', y_\alpha{:}A \otimes B[c]$. We have $(\mathcal{C}_1, \mathcal{D}_1, \mathcal{F}_1; \mathcal{C}_2, \mathcal{D}_2, \mathcal{F}_2) \in \mathsf{Tree}_\Psi(\Delta', y_\alpha{:}A \otimes B[c] \Vdash K)$ and $\mathcal{C}_1 = \mathcal{C}'_1 \mathbf{msg}(\mathbf{send} x^c_\beta \, y^c_\alpha; y^c_\alpha \leftarrow u^c_\delta)$ and

   $$\text{if } \mathcal{C}_2 = \mathcal{C}'_2 \mathbf{msg}(\mathbf{send} x^c_\beta \, y^c_\alpha; y^c_\alpha \leftarrow u^c_\delta) \text{then}$$

   $$\dagger \, (\mathcal{C}'_1 \mathbf{msg}(\mathbf{send} x^c_\beta \, y^c_\alpha; y^c_\alpha \leftarrow u^c_\delta)\mathcal{D}_1 \mathcal{F}_1, \mathcal{C}'_2 \mathbf{msg}(\mathbf{send} x^c_\beta \, y^c_\alpha; y^c_\alpha \leftarrow u^c_\delta)\mathcal{D}_2 \mathcal{F}_2) \in \mathcal{E}^\xi_\Psi[\![\Delta', x_\beta{:}A[c], u_\delta{:}B[c] \Vdash K]\!]^m$$

   These are also the statements that we need to prove when replacing $\mathcal{C}_2$ with $\mathcal{C}^2_2$. The first tree statements are straightforward by type preservation. Assume

   $$\mathcal{C}^2_2 = \mathcal{C}^{2'}_2 \mathbf{msg}(\mathbf{send} x^c_\beta \, y^c_\alpha; y^c_\alpha \leftarrow u^c_\delta),$$

   By the assumption that $\mathcal{C}_2$ also sends along $y$, uniqueness of channels, and $\mathcal{C}^2_2 \mapsto^* \mathcal{C}_2$ we get

   $$\mathcal{C}_2 = \mathcal{C}'_2 \mathbf{msg}(\mathbf{send} x^c_\beta \, y^c_\alpha; y^c_\alpha \leftarrow u^c_\delta).$$

   Moreover, $\mathcal{C}^{2'}_2 \mapsto^* \mathcal{C}'_2$. Now, we can apply the induction hypothesis on $\dagger$ to get

   $$\dagger' \, (\mathcal{C}'_1 \mathbf{msg}(\mathbf{send} x^c_\beta \, y^c_\alpha; y^c_\alpha \leftarrow u^c_\delta)\mathcal{D}_1 \mathcal{F}_1, \mathcal{C}^{2'}_2 \mathbf{msg}(\mathbf{send} x^c_\beta \, y^c_\alpha; y^c_\alpha \leftarrow u^c_\delta)\mathcal{D}_2 \mathcal{F}_2) \in \mathcal{E}^\xi_\Psi[\![\Delta', x_\beta{:}A[c], u_\delta{:}B[c] \Vdash K]\!]^m.$$

   **Row 12.** By the conditions of this row, we know that $\Delta = \Delta', y_\alpha{:}A[c]$. We have $(\mathcal{C}_1, \mathcal{D}_1, \mathcal{F}_1; \mathcal{C}_2, \mathcal{D}_2, \mathcal{F}_2) \in \mathsf{Tree}_\Psi(\Delta', y_\alpha{:}Ac] \Vdash K)$ and $\mathcal{C}_1 = \mathcal{C}'_1 \mathbf{proc}(y_\alpha[c], y^c_\alpha \leftarrow x^c_\beta @ d_1)$ and

   $$\text{if } \mathcal{C}_2 = \mathcal{C}'_2 \mathbf{proc}(y_\alpha[c], y^c_\alpha \leftarrow x^c_\beta @ d_2) \text{then}$$

   $$\dagger \, ([x^c_\beta/y^c_\alpha]\mathcal{C}'_1 \mathcal{D}_1 \mathcal{F}_1, [x^c_\beta/y^c_\alpha]\mathcal{C}'_2 \mathcal{D}_2 \mathcal{F}_2) \in \mathcal{E}^\xi_\Psi[\![\Delta', x_\beta{:}A[c] \Vdash K]\!]^m$$

   These are also the statements that we need to prove when replacing $\mathcal{C}_2$ with $\mathcal{C}^2_2$. The first two statements are straightforward by type preservation. Assume

   $$\mathcal{C}^2_2 = \mathcal{C}^{2'}_2 \mathbf{proc}(y_\alpha[c], y^c_\alpha \leftarrow x^c_\beta @ d_2),$$

   By the assumption that $\mathcal{C}_2$ also sends along $y$, uniqueness of channels, and $\mathcal{C}^2_2 \mapsto^* \mathcal{C}_2$ we get

   $$\mathcal{C}_2 = \mathcal{C}'_2 \mathbf{proc}(y_\alpha[c], y^c_\alpha \leftarrow x^c_\beta @ d_2).$$

   Moreover, $\mathcal{C}^{2'}_2 \mapsto^* \mathcal{C}'_2$. Now, we can apply the induction hypothesis on $\dagger$ to get

   $$\dagger' \, ([x^c_\beta/y^c_\alpha]\mathcal{C}'_1 \mathcal{D}_1 \mathcal{F}_1, [x^c_\beta/y^c_\alpha]\mathcal{C}^{2'}_2 \mathcal{D}_2 \mathcal{F}_2) \in \mathcal{E}^\xi_\Psi[\![\Delta', x_\beta{:}A[c] \Vdash K]\!]^m$$

   **Rows 3,6,7.** Similar to the previous cases.
6. The proof is similar to the previous statement.

$\square$

**Lemma C.10** (Backward closure on the second run). *The second run enjoys backward closure:*

1. If $(\mathcal{B}_1; \mathcal{B}_2) \in \mathcal{E}^\xi_{\Psi_0}\llbracket \Delta \Vdash K \rrbracket^k$ and $\mathcal{B}_i = \mathcal{C}_i\mathcal{D}_i\mathcal{F}_i$ with $(\mathcal{C}_1, \mathcal{D}_1, \mathcal{F}_1; \mathcal{C}_2, \mathcal{D}_2, \mathcal{F}_2) \in \mathsf{Tree}(\Delta \Vdash K)$ and $\mathcal{D}''_2 \mapsto^* \mathcal{D}_2$ then $(\mathcal{C}_1\mathcal{D}_1\mathcal{F}_1; \mathcal{C}_2\mathcal{D}''_2\mathcal{F}_2) \in \mathcal{E}^\xi_{\Psi_0}\llbracket \Delta \Vdash K \rrbracket^k$.

2. Consider $(\mathcal{C}_1, \mathcal{D}_1, \mathcal{F}_1; \mathcal{C}_2, \mathcal{D}_2, \mathcal{F}_2) \in \mathcal{V}\llbracket \Delta \Vdash K \rrbracket^{k+1}_{:;y_\alpha}$ and $\mathcal{D}''_2 \mapsto^* \mathcal{D}_2$ such that $\mathcal{D}''_2$ and $\mathcal{D}_2$ send along the same set of channels $\Upsilon \subseteq \Delta, K$. We have $(\mathcal{C}_1, \mathcal{D}_1, \mathcal{F}_1; \mathcal{C}_2, \mathcal{D}''_2, \mathcal{F}_2) \in \mathcal{V}\llbracket \Delta \Vdash K \rrbracket^{k+1}_{:;y_\alpha}$.

3. Consider $(\mathcal{C}_1, \mathcal{D}_1, \mathcal{F}_1; \mathcal{C}_2, \mathcal{D}_2, \mathcal{F}_2) \in \mathcal{V}\llbracket \Delta \Vdash K \rrbracket^{k+1}_{y_\alpha;:}$ and $\mathcal{D}''_2 \mapsto^* \mathcal{D}_2$ such that $\mathcal{D}''_2$ and $\mathcal{D}_2$ send along the same set of channels $\Upsilon \subseteq \Delta, K$. We have $(\mathcal{C}_1, \mathcal{D}_1, \mathcal{F}_1; \mathcal{C}_2, \mathcal{D}''_2, \mathcal{F}_2) \in \mathcal{V}\llbracket \Delta \Vdash K \rrbracket^{k+1}_{y_\alpha;:}$.

*Proof.* We prove the first statement separately and then use it to prove the second and third statements.

1. If $k = 0$, the proof is trivial. Consider $k = m + 1$. By assumption, we have $(\mathcal{C}_1\mathcal{D}_1\mathcal{F}_1; \mathcal{C}_2\mathcal{D}_2\mathcal{F}_2) \in \mathcal{E}^\xi_{\Psi_0}\llbracket \Delta \Vdash K \rrbracket^{m+1}$. By line (13) of the logical relation we get

   $\star$ if $\mathcal{C}_1, \mathcal{D}_1, \mathcal{F}_1 \mapsto^{*\Theta, \Upsilon_1}_{\Delta \Vdash K} \mathcal{C}_1, \mathcal{D}'_1, \mathcal{F}_1$, then
   $\exists \mathcal{D}'_2. \mathcal{C}_2, \mathcal{D}_2, \mathcal{F}_2 \mapsto^{*\Theta'; \Upsilon_1}_{\Delta \Vdash K} \mathcal{C}_2, \mathcal{D}'_2, \mathcal{F}_2$, and $\forall \mathcal{C}'_1, \mathcal{F}'_1, \mathcal{C}'_2, \mathcal{F}'_2$.if $\mathcal{C}_1, \mathcal{D}'_1, \mathcal{F}_1 \mapsto^{*\Theta_1; \Upsilon_1}_{\Delta \Vdash K} \mathcal{C}'_1, \mathcal{D}'_1, \mathcal{F}'_1$ and $\mathcal{C}_2, \mathcal{D}'_2, \mathcal{F}_2 \mapsto^{*\Theta_2; \Upsilon_1}_{\Delta \Vdash K} \mathcal{C}'_2, \mathcal{D}'_2, \mathcal{F}'_2$, then
   $\forall x \in \Upsilon_1. (\mathcal{C}'_1, \mathcal{D}'_1, \mathcal{F}'_1; \mathcal{C}'_2, \mathcal{D}'_2, \mathcal{F}'_2) \in \mathcal{V}^\xi_\Psi\llbracket \Delta \Vdash K \rrbracket^{m+1}_{:;x}$ and
   $\forall x \in (\Theta_1 \cap \Theta_2). (\mathcal{C}'_1, \mathcal{D}'_1, \mathcal{F}'_1; \mathcal{C}'_2, \mathcal{D}'_2, \mathcal{F}'_2) \in \mathcal{V}^\xi_\Psi\llbracket \Delta \Vdash K \rrbracket^{m+1}_{x;:}$.

   Consider $\mathcal{D}''_2$, for which by the assumption we have $\mathcal{D}''_2 \mapsto^* \mathcal{D}_2$. Our goal is to prove $(\mathcal{C}_1\mathcal{D}_1\mathcal{F}_1; \mathcal{C}_2\mathcal{D}''_2\mathcal{F}_2) \in \mathcal{E}^\xi_{\Psi_0}\llbracket \Delta \Vdash K \rrbracket^{m+1}$. We need to show that

   † if $\mathcal{C}_1, \mathcal{D}_1, \mathcal{F}_1 \mapsto^{*\Theta, \Upsilon_1}_{\Delta \Vdash K} \mathcal{C}'_1, \mathcal{D}'_1, \mathcal{F}'_1$, then
   $\exists \mathcal{D}'_2. \mathcal{C}_2, \mathcal{D}''_2, \mathcal{F}_2 \mapsto^{*\Theta'; \Upsilon_1}_{\Delta \Vdash K} \mathcal{C}_2, \mathcal{D}'_2, \mathcal{F}_2$, and $\forall \mathcal{C}'_1, \mathcal{F}'_1, \mathcal{C}'_2, \mathcal{F}'_2$.if $\mathcal{C}_1, \mathcal{D}'_1, \mathcal{F}_1 \mapsto^{*\Theta_1; \Upsilon_1}_{\Delta \Vdash K} \mathcal{C}'_1, \mathcal{D}'_1, \mathcal{F}'_1$ and $\mathcal{C}_2, \mathcal{D}'_2, \mathcal{F}_2 \mapsto^{*\Theta_2; \Upsilon_1}_{\Delta \Vdash K} \mathcal{C}'_2, \mathcal{D}'_2, \mathcal{F}'_2$, then
   $\forall x \in \Upsilon_1. (\mathcal{C}'_1, \mathcal{D}'_1, \mathcal{F}'_1; \mathcal{C}'_2, \mathcal{D}'_2, \mathcal{F}'_2) \in \mathcal{V}^\xi_\Psi\llbracket \Delta \Vdash K \rrbracket^{m+1}_{:;x}$ and
   $\forall x \in (\Theta_1 \cap \Theta_2). (\mathcal{C}'_1, \mathcal{D}'_1, \mathcal{F}'_1; \mathcal{C}'_2, \mathcal{D}'_2, \mathcal{F}'_2) \in \mathcal{V}^\xi_\Psi\llbracket \Delta \Vdash K \rrbracket^{m+1}_{x;:}$.

   With an if-Introduction on the goal followed by an if- Elimination on the assumption, we get the assumption

   $\star'$ $\exists \mathcal{D}'_2. \mathcal{C}_2, \mathcal{D}_2, \mathcal{F}_2 \mapsto^{*\Theta'; \Upsilon_1}_{\Delta \Vdash K} \mathcal{C}_2, \mathcal{D}'_2, \mathcal{F}_2$, and $\forall \mathcal{C}'_1, \mathcal{F}'_1, \mathcal{C}'_2, \mathcal{F}'_2$.if $\mathcal{C}_1, \mathcal{D}'_1, \mathcal{F}_1 \mapsto^{*\Theta_1; \Upsilon_1}_{\Delta \Vdash K} \mathcal{C}'_1, \mathcal{D}'_1, \mathcal{F}'_1$ and $\mathcal{C}_2, \mathcal{D}'_2, \mathcal{F}_2 \mapsto^{*\Theta_2; \Upsilon_1}_{\Delta \Vdash K} \mathcal{C}'_2, \mathcal{D}'_2, \mathcal{F}'_2$, then
   $\forall x \in \Upsilon_1. (\mathcal{C}'_1, \mathcal{D}'_1, \mathcal{F}'_1; \mathcal{C}'_2, \mathcal{D}'_2, \mathcal{F}'_2) \in \mathcal{V}^\xi_\Psi\llbracket \Delta \Vdash K \rrbracket^{m+1}_{:;x}$ and
   $\forall x \in (\Theta_1 \cap \Theta_2). (\mathcal{C}'_1, \mathcal{D}'_1, \mathcal{F}'_1; \mathcal{C}'_2, \mathcal{D}'_2, \mathcal{F}'_2) \in \mathcal{V}^\xi_\Psi\llbracket \Delta \Vdash K \rrbracket^{m+1}_{x;:}$.

   and the goal

   †' $\exists \mathcal{D}'_2. \mathcal{C}_2, \mathcal{D}''_2, \mathcal{F}_2 \mapsto^{*\Theta'; \Upsilon_1}_{\Delta \Vdash K} \mathcal{C}_2, \mathcal{D}'_2, \mathcal{F}_2$, and $\forall \mathcal{C}'_1, \mathcal{F}'_1, \mathcal{C}'_2, \mathcal{F}'_2$.if $\mathcal{C}_1, \mathcal{D}'_1, \mathcal{F}_1 \mapsto^{*\Theta_1; \Upsilon_1}_{\Delta \Vdash K} \mathcal{C}'_1, \mathcal{D}'_1, \mathcal{F}'_1$ and $\mathcal{C}_2, \mathcal{D}'_2, \mathcal{F}_2 \mapsto^{*\Theta_2; \Upsilon_1}_{\Delta \Vdash K} \mathcal{C}'_2, \mathcal{D}'_2, \mathcal{F}'_2$, then
   $\forall x \in \Upsilon_1. (\mathcal{C}'_1, \mathcal{D}'_1, \mathcal{F}'_1; \mathcal{C}'_2, \mathcal{D}'_2, \mathcal{F}'_2) \in \mathcal{V}^\xi_\Psi\llbracket \Delta \Vdash K \rrbracket^{m+1}_{:;x}$ and
   $\forall x \in (\Theta_1 \cap \Theta_2). (\mathcal{C}'_1, \mathcal{D}'_1, \mathcal{F}'_1; \mathcal{C}'_2, \mathcal{D}'_2, \mathcal{F}'_2) \in \mathcal{V}^\xi_\Psi\llbracket \Delta \Vdash K \rrbracket^{m+1}_{x;:}$.

   We apply an $\exists$-Elimination on the assumption to get $\mathcal{D}'_2$ that satisfies the conditions and use the same $\mathcal{D}'_2$ to instantiate the existential quantifier in the goal. Since $\mathcal{D}''_2 \mapsto^* \mathcal{D}_2$, we get $\mathcal{C}_2, \mathcal{D}''_2, \mathcal{F}_2 \mapsto^{*\Theta'; \Upsilon_1}_{\Delta \Vdash K} \mathcal{C}_2, \mathcal{D}'_2, \mathcal{F}_2$, and the proof is complete.

2. The proof is by cases on the row of the logical relation that makes the assumption $(\mathcal{C}_1, \mathcal{D}_1, \mathcal{F}_1; \mathcal{C}_2, \mathcal{D}_2, \mathcal{F}_2) \in \mathcal{V}\llbracket \Delta \Vdash K \rrbracket^{k+1}_{:;y_\alpha}$ correct. Here we only consider two interesting cases, the proof of other cases is similar.

   **Row 4.** By the conditions of this row, we know that $\Delta = \Delta', \Delta''$, and $K = y_\alpha{:}A \otimes B[c]$. Moreover, $\mathcal{C}_1 = \mathcal{C}'_1\mathcal{C}''_1$ and $\mathcal{C}_2 = \mathcal{C}'_2\mathcal{C}''_2$ with $\mathcal{C}'_i \in \mathsf{Tree}_\Psi(\cdot \Vdash \Delta')$ and $\mathcal{C}''_i \in \mathsf{Tree}_\Psi(\cdot \Vdash \Delta'')$. We have

   $$\mathcal{D}_1 = \mathcal{D}'_1 \mathcal{T}_1 \mathbf{msg}(\mathbf{send}\, x^c_\beta\, y^c_\alpha; y^c_\alpha \leftarrow u^c_\delta) \text{ for } \mathcal{T}_1 \in \mathsf{Tree}_\Psi(\Delta'' \Vdash x_\beta{:}A[c])$$

   and

   $$\mathcal{D}_2 = \mathcal{D}'_2 \mathcal{T}_2 \mathbf{msg}(\mathbf{send},\, x^c_\beta\, y^c_\alpha; y^c_\alpha \leftarrow u^c_\delta) \text{ for } \mathcal{T}_2 \in \mathsf{Tree}_\Psi(\Delta'' \Vdash x_\beta{:}A[c])$$

   and

   $\dagger_1$ $(\mathcal{C}''_1\mathcal{T}_1\mathcal{C}'_1\mathcal{D}'_1\mathbf{msg}(\mathbf{send}\, x^c_\beta\, y^c_\alpha; y^c_\alpha \leftarrow u^c_\delta)\mathcal{F}_1, \mathcal{C}''_2\mathcal{T}_2\mathcal{C}'_2\mathcal{D}'_2\mathbf{msg}(\mathbf{send}\, x^c_\beta\, y^c_\alpha; y^c_\alpha \leftarrow u^c_\delta)\mathcal{F}_2) \in \mathcal{E}^\xi_\Psi\llbracket \Delta'' \Vdash x_\beta{:}A[c]\rrbracket^k$

   and

   $\dagger_2$ $(\mathcal{C}'_1\mathcal{D}'_1\mathcal{C}''_1\mathcal{T}_1\mathbf{msg}(\mathbf{send}\, x^c_\beta\, y^c_\alpha; y^c_\alpha \leftarrow u^c_\delta)\mathcal{F}_1, \mathcal{C}'_2\mathcal{D}'_2\mathcal{C}''_2\mathcal{T}_2\mathbf{msg}(\mathbf{send}\, x^c_\beta\, y^c_\alpha; y^c_\alpha \leftarrow u^c_\delta)\mathcal{F}_2) \in \mathcal{E}^\xi_\Psi\llbracket \Delta' \Vdash u_\delta{:}B[c]\rrbracket^k$

   These are also the statements that we need to prove when replacing $\mathcal{D}_2$ with $\mathcal{D}''_2$. By the assumption that $\mathcal{D}''_2$ also sends along $y$, uniqueness of channels, and $\mathcal{D}''_2 \mapsto^* \mathcal{D}_2$ we get

   $$\mathcal{D}''_2 = \mathcal{D}^1_2 \mathcal{T}^1_2 \mathbf{msg}(\mathbf{send},\, x^c_\beta\, y^c_\alpha; y^c_\alpha \leftarrow u^c_\delta) \text{ for } \mathcal{T}^1_2 \in \mathsf{Tree}_\Psi(\Delta'' \Vdash x_\beta{:}A[c])$$

Moreover, since $\mathcal{T}_2^1$ and $\mathcal{D}_2^1$ are disjoint sub-trees and cannot communicate with each other internally, we have $\mathcal{T}_2^1 \mapsto^* \mathcal{T}_2$ and $\mathcal{D}_2^1 \mapsto^* \mathcal{D}_2'$.
From this we get

$$\mathcal{C}_2' \mathcal{D}_2^1 \mathbf{msg}(\mathbf{send} x_\beta^c\, y_\alpha^c; y_\alpha^c \leftarrow u_\delta^c) \mathcal{F}_2 \mapsto^* \mathcal{C}_2' \mathcal{D}_2' \mathbf{msg}(\mathbf{send} x_\beta^c\, y_\alpha^c; y_\alpha^c \leftarrow u_\delta^c) \mathcal{F}_2$$

and

$$\mathcal{C}_2'' \mathcal{T}_2^1 \mathbf{msg}(\mathbf{send} x_\beta^c\, y_\alpha^c; y_\alpha^c \leftarrow u_\delta^c) \mathcal{F}_2 \mapsto^* \mathcal{C}_2'' \mathcal{T}_2 \mathbf{msg}(\mathbf{send} x_\beta^c\, y_\alpha^c; y_\alpha^c \leftarrow u_\delta^c) \mathcal{F}_2.$$

Now we can apply the the first item of this lemma on $\dagger_1$ to get

$$\dagger_1'\ (\mathcal{C}_1'' \mathcal{T}_1 \mathcal{C}_1' \mathcal{D}_1' \mathbf{msg}(\mathbf{send} x_\beta^c\, y_\alpha^c; y_\alpha^c \leftarrow u_\delta^c) \mathcal{F}_1, \mathcal{C}_2'' \mathcal{T}_2^1 \mathcal{C}_2' \mathcal{D}_2' \mathbf{msg}(\mathbf{send} x_\beta^c\, y_\alpha^c; y_\alpha^c \leftarrow u_\delta^c) \mathcal{F}_2) \in \mathcal{E}_\Psi^\xi [\![\Delta'' \Vdash x_\beta{:}A[c]]\!]^k.$$

We then apply Lemma C.9 on $\dagger_1'$ to get

$$\dagger_1''\ (\mathcal{C}_1'' \mathcal{T}_1 \mathcal{C}_1' \mathcal{D}_1' \mathbf{msg}(\mathbf{send} x_\beta^c\, y_\alpha^c; y_\alpha^c \leftarrow u_\delta^c) \mathcal{F}_1, \mathcal{C}_2'' \mathcal{T}_2^1 \mathcal{C}_2' \mathcal{D}_2^1 \mathbf{msg}(\mathbf{send} x_\beta^c\, y_\alpha^c; y_\alpha^c \leftarrow u_\delta^c) \mathcal{F}_2) \in \mathcal{E}_\Psi^\xi [\![\Delta'' \Vdash x_\beta{:}A[c]]\!]^k.$$

Similarly we apply the first item of this lemma on $\dagger_2$ followed by an application of Lemma C.9 to get

$$\dagger_2''\ (\mathcal{C}_1' \mathcal{D}_1' \mathcal{C}_1'' \mathcal{T}_1 \mathbf{msg}(\mathbf{send} x_\beta^c\, y_\alpha^c; y_\alpha^c \leftarrow u_\delta^c) \mathcal{F}_1, \mathcal{C}_2' \mathcal{D}_2^1 \mathcal{C}_2'' \mathcal{T}_2^1 \mathbf{msg}(\mathbf{send} x_\beta^c\, y_\alpha^c; y_\alpha^c \leftarrow u_\delta^c) \mathcal{F}_2) \in \mathcal{E}_\Psi^\xi [\![\Delta' \Vdash u_\delta{:}B[c]]\!]^k$$

and the proof of this subcase is complete.

**Row 11.** By the conditions of this row, we know that $K = y_\alpha{:}A[c]$. Moreover, we have

$$\mathcal{D}_1 = \mathcal{D}_1' \mathbf{proc}(y_\alpha[c], y_\alpha^c \leftarrow x_\beta^c @ d_1)$$

and

$$\mathcal{D}_2 = \mathcal{D}_2' \mathbf{proc}(y_\alpha[c], y_\alpha^c \leftarrow x_\beta^c @ d_2)$$

and

$$\dagger_1\ ([x_\beta^c/y_\alpha^c]\mathcal{C}_1 \mathcal{D}_1' \mathcal{F}_1, [x_\beta^c/y_\alpha^c]\mathcal{C}_2 \mathcal{D}_2' \mathcal{F}_2) \in \mathcal{E}_\Psi^\xi [\![\Delta \Vdash x_\beta{:}A[c]]\!]^k$$

These are also the statements that we need to prove when replacing $\mathcal{D}_2$ with $\mathcal{D}_2''$. By the assumption that $\mathcal{D}_2''$ also sends along $y$, uniqueness of channels, and $\mathcal{D}_2'' \mapsto^* \mathcal{D}_2$ we get

$$\mathcal{D}_2'' = \mathcal{D}_2^1 \mathbf{proc}(y_\alpha[c], y_\alpha^c \leftarrow x_\beta^c @ d_2)$$

Moreover, we have $\mathcal{D}_2^1 \mapsto^* \mathcal{D}_2'$.
Now we can apply the the first item of this lemma on $\dagger_1$ to get

$$\dagger_1'\ ([x_\beta^c/y_\alpha^c]\mathcal{C}_1 \mathcal{D}_1' \mathcal{F}_1, [x_\beta^c/y_\alpha^c]\mathcal{C}_2 \mathcal{D}_2^1 \mathcal{F}_2) \in \mathcal{E}_\Psi^\xi [\![\Delta \Vdash x_\beta{:}A[c]]\!]^k$$

and the proof of this subcase is complete.
**Rows 1,2,8,10.** Similar to the previous cases.

3. The proof is by considering the row of the logical relation that ensures the assumption $(\mathcal{C}_1, \mathcal{D}_1, \mathcal{F}_1; \mathcal{C}_2, \mathcal{D}_2, \mathcal{F}_2) \in \mathcal{V}[\![\Delta \Vdash K]\!]^k_{y_\alpha;\cdot}$. Here we only consider three interesting cases, the proof of other cases is similar.
**Row 5.** By the conditions of this row, we know that $K = y_\alpha{:}A \multimap B[c]$. We have $(\mathcal{C}_1, \mathcal{D}_1, \mathcal{F}_1; \mathcal{C}_2, \mathcal{D}_2, \mathcal{F}_2) \in \mathsf{Tree}_\Psi(\Delta \Vdash y_\alpha{:}A \multimap B[c])$ and $\mathcal{F}_1 = \mathcal{T}_1 \mathbf{msg}(\mathbf{send} x_\beta^c\, y_\alpha^c; u_\delta^c \leftarrow y_\alpha^c) \mathcal{F}_1'$ for $\mathcal{T}_1 \in \mathsf{Tree}_\Psi(\cdot \Vdash x_\beta{:}A[c])$. Moreover, we have

if $\mathcal{F}_2 = \mathcal{T}_2 \mathbf{msg}(\mathbf{send} x_\beta^c\, y_\alpha^c; u_\delta^c \leftarrow y_\alpha^c) \mathcal{F}_2'$ for $\mathcal{T}_2 \in \mathsf{Tree}_\Psi(\cdot \Vdash x_\beta{:}A[c])$ then

$$\dagger\ (\mathcal{C}_1 \mathcal{T}_1 \mathcal{D}_1 \mathbf{msg}(\mathbf{send} x_\beta^c\, y_\alpha^c; u_\delta^c \leftarrow y_\alpha^c) \mathcal{F}_1', \mathcal{C}_2 \mathcal{T}_2 \mathcal{D}_2 \mathbf{msg}(\mathbf{send} x_\beta^c\, y_\alpha^c; u_\delta^c \leftarrow y_\alpha^c) \mathcal{F}_2') \in \mathcal{E}_\Psi^\xi [\![\Delta, x_\beta{:}A[c] \Vdash u_\delta{:}B[c]]\!]^k$$

These are also the statements that we need to prove when replacing $\mathcal{D}_2$ with $\mathcal{D}_2''$. The first tree statements are straightforward by type preservation. Now assume $\mathcal{F}_2 = \mathcal{T}_2 \mathbf{msg}(\mathbf{send} x_\beta^c\, y_\alpha^c; u_\delta^c \leftarrow y_\alpha^c) \mathcal{F}_2'$ for $\mathcal{T}_2 \in \mathsf{Tree}_\Psi(\cdot \Vdash x_\beta{:}A[c])$. By assumption, we know that $\mathcal{D}_2'' \mathbf{msg}(\mathbf{send} x_\beta^c\, y_\alpha^c; u_\delta^c \leftarrow y_\alpha^c) \mapsto^* \mathcal{D}_2 \mathbf{msg}(\mathbf{send} x_\beta^c\, y_\alpha^c; u_\delta^c \leftarrow y_\alpha^c)$. We can apply the first item of this lemma on $\dagger$ to get

$$\dagger'\ (\mathcal{C}_1 \mathcal{T}_1 \mathcal{D}_1 \mathbf{msg}(\mathbf{send} x_\beta^c\, y_\alpha^c; u_\delta^c \leftarrow y_\alpha^c) \mathcal{F}_1', \mathcal{C}_2 \mathcal{T}_2 \mathcal{D}_2'' \mathbf{msg}(\mathbf{send} x_\beta^c\, y_\alpha^c; u_\delta^c \leftarrow y_\alpha^c) \mathcal{F}_2') \in \mathcal{E}_\Psi^\xi [\![\Delta, x_\beta{:}A[c] \Vdash u_\delta{:}B[c]]\!]^k$$

which completes the proof of this case.

**Row 9.** By the conditions of this row, we know that $\Delta = \Delta', y_\alpha{:}A \otimes B[c]$. We have $(\mathcal{C}_1, \mathcal{D}_1, \mathcal{F}_1; \mathcal{C}_2, \mathcal{D}_2, \mathcal{F}_2) \in \mathsf{Tree}_\Psi(\Delta', y_\alpha{:}A \otimes B[c] \Vdash K)$ and $\mathcal{C}_1 = \mathcal{C}_1'\mathbf{msg}(\mathbf{send}x_\beta^c\, y_\alpha^c; y_\alpha^c \leftarrow u_\delta^c)$ and

$$\text{if } \mathcal{C}_2 = \mathcal{C}_2'\mathbf{msg}(\mathbf{send}x_\beta^c\, y_\alpha^c; y_\alpha^c \leftarrow u_\delta^c)\text{then}$$

$$\dagger\, (\mathcal{C}_1'\mathbf{msg}(\mathbf{send}x_\beta^c\, y_\alpha^c; y_\alpha^c \leftarrow u_\delta^c)\mathcal{D}_1\mathcal{F}_1, \mathcal{C}_2'\mathbf{msg}(\mathbf{send}x_\beta^c\, y_\alpha^c; y_\alpha^c \leftarrow u_\delta^c)\mathcal{D}_2\mathcal{F}_2) \in \mathcal{E}_\Psi^\xi[\![\Delta', x_\beta{:}A[c], u_\delta{:}B[c] \Vdash K]\!]^k$$

These are also the statements that we need to prove when replacing $\mathcal{D}_2$ with $\mathcal{D}_2''$. The first tree statements are straightforward by type preservation. Assume

$$\mathcal{C}_2 = \mathcal{C}_2'\mathbf{msg}(\mathbf{send}x_\beta^c\, y_\alpha^c; y_\alpha^c \leftarrow u_\delta^c),$$

By assumption, $\mathbf{msg}(\mathbf{send}x_\beta^c\, y_\alpha^c; y_\alpha^c \leftarrow u_\delta^c)\mathcal{D}_2'' \mapsto^* \mathbf{msg}(\mathbf{send}x_\beta^c\, y_\alpha^c; y_\alpha^c \leftarrow u_\delta^c)\mathcal{D}_2$. Now, we can apply the induction hypothesis on † to get

$$\dagger'\, (\mathcal{C}_1'\mathbf{msg}(\mathbf{send}x_\beta^c\, y_\alpha^c; y_\alpha^c \leftarrow u_\delta^c)\mathcal{D}_1\mathcal{F}_1, \mathcal{C}_2'\mathbf{msg}(\mathbf{send}x_\beta^c\, y_\alpha^c; y_\alpha^c \leftarrow u_\delta^c)\mathcal{D}_2''\mathcal{F}_2) \in \mathcal{E}_\Psi^\xi[\![\Delta', x_\beta{:}A[c], u_\delta{:}B[c] \Vdash K]\!]^k$$

**Row 12.** By the conditions of this row, we know that $\Delta = \Delta', y_\alpha{:}A[c]$. We have $(\mathcal{C}_1, \mathcal{D}_1, \mathcal{F}_1; \mathcal{C}_2, \mathcal{D}_2, \mathcal{F}_2) \in \mathsf{Tree}_\Psi(\Delta', y_\alpha{:}A[c] \Vdash K)$ and $\mathcal{C}_1 = \mathcal{C}_1'\mathbf{proc}(y_\alpha[c], y_\alpha^c \leftarrow x_\beta^c@d_1)$ and

$$\text{if } \mathcal{C}_2 = \mathcal{C}_2'\mathbf{proc}(y_\alpha[c], y_\alpha^c \leftarrow x_\beta^c@d_2)\text{then}$$

$$\dagger\, ([x_\beta^c/y_\alpha^c]\mathcal{C}_1'\mathcal{D}_1\mathcal{F}_1, [x_\beta^c/y_\alpha^c]\mathcal{C}_2'\mathcal{D}_2\mathcal{F}_2) \in \mathcal{E}_\Psi^\xi[\![\Delta', x_\beta{:}A[c] \Vdash K]\!]^k$$

These are also the statements that we need to prove when replacing $\mathcal{D}_2$ with $\mathcal{D}_2''$. The first tree statements are straightforward by type preservation. Assume

$$\mathcal{C}_2 = \mathcal{C}_2'\mathbf{proc}(y_\alpha[c], y_\alpha^c \leftarrow x_\beta^c@d_2),$$

By assumption, $[x_\beta^c/y_\alpha^c]\mathcal{D}_2'' \mapsto^* [x_\beta^c/y_\alpha^c]\mathcal{D}_2$. Now, we can apply the induction hypothesis on † to get

$$\dagger'\, ([x_\beta^c/y_\alpha^c]\mathcal{C}_1'\mathcal{D}_1\mathcal{F}_1, [x_\beta^c/y_\alpha^c]\mathcal{C}_2'\mathcal{D}_2''\mathcal{F}_2) \in \mathcal{E}_\Psi^\xi[\![\Delta', x_\beta{:}A[c] \Vdash K]\!]^k.$$

**Rows 3,6,7.** Similar to the previous cases.

□

**Lemma C.11** (Forward closure on the substitutions). *We prove forward closure for the substitutions of the first and second runs.*

- *Consider $(\mathcal{B}_1; \mathcal{B}_2) \in \mathcal{E}_{\Psi_0}^\xi[\![\Delta \Vdash K]\!]^k$ and $\mathcal{B}_i = \mathcal{C}_i\mathcal{D}_i\mathcal{F}_i$ with $(\mathcal{C}_1, \mathcal{D}_1, \mathcal{F}_1; \mathcal{C}_2, \mathcal{D}_2, \mathcal{F}_2) \in \mathsf{Tree}(\Delta \Vdash K)$ and $\mathcal{C}_2 \mapsto^* \mathcal{C}_2''$. We can prove $(\mathcal{C}_1\mathcal{D}_1\mathcal{F}_1; \mathcal{C}_2''\mathcal{D}_2\mathcal{F}_2) \in \mathcal{E}_{\Psi_0}^\xi[\![\Delta \Vdash K]\!]^k$.*
- *Consider $(\mathcal{B}_1; \mathcal{B}_2) \in \mathcal{E}_{\Psi_0}^\xi[\![\Delta \Vdash K]\!]^k$ and $\mathcal{B}_i = \mathcal{C}_i\mathcal{D}_i\mathcal{F}_i$ with $(\mathcal{C}_1, \mathcal{D}_1, \mathcal{F}_1; \mathcal{C}_2, \mathcal{D}_2, \mathcal{F}_2) \in \mathsf{Tree}(\Delta \Vdash K)$ and $\mathcal{F}_2 \mapsto^* \mathcal{F}_2''$. We can prove $(\mathcal{C}_1\mathcal{D}_1\mathcal{F}_1; \mathcal{C}_2\mathcal{D}_2\mathcal{F}_2'') \in \mathcal{E}_{\Psi_0}^\xi[\![\Delta \Vdash K]\!]^k$.*
- *Consider $(\mathcal{B}_1; \mathcal{B}_2) \in \mathcal{E}_{\Psi_0}^\xi[\![\Delta \Vdash K]\!]^k$ and $\mathcal{B}_i = \mathcal{C}_i\mathcal{D}_i\mathcal{F}_i$ with $(\mathcal{C}_1, \mathcal{D}_1, \mathcal{F}_1; \mathcal{C}_2, \mathcal{D}_2, \mathcal{F}_2) \in \mathsf{Tree}(\Delta \Vdash K)$ and $\mathcal{C}_1 \mapsto^* \mathcal{C}_1''$. We can prove $(\mathcal{C}_1''\mathcal{D}_1\mathcal{F}_1; \mathcal{C}_2\mathcal{D}_2\mathcal{F}_2) \in \mathcal{E}_{\Psi_0}^\xi[\![\Delta \Vdash K]\!]^k$.*
- *Consider $(\mathcal{B}_1; \mathcal{B}_2) \in \mathcal{E}_{\Psi_0}^\xi[\![\Delta \Vdash K]\!]^k$ and $\mathcal{B}_i = \mathcal{C}_i\mathcal{D}_i\mathcal{F}_i$ with $(\mathcal{C}_1, \mathcal{D}_1, \mathcal{F}_1; \mathcal{C}_2, \mathcal{D}_2, \mathcal{F}_2) \in \mathsf{Tree}(\Delta \Vdash K)$ and $\mathcal{F}_1 \mapsto^* \mathcal{F}_1''$. We can prove $(\mathcal{C}_1\mathcal{D}_1\mathcal{F}_1''; \mathcal{C}_2\mathcal{D}_2\mathcal{F}_2) \in \mathcal{E}_{\Psi_0}^\xi[\![\Delta \Vdash K]\!]^k$.*

*Proof.* Here we provide the proof for the first statement. The proof of the other statements are similar. If $k = 0$ the proof is trivial. Consider $k = m + 1$. By assumption, we have $(\mathcal{C}_1\mathcal{D}_1\mathcal{F}_1; \mathcal{C}_2\mathcal{D}_2\mathcal{F}_2) \in \mathcal{E}_{\Psi_0}^\xi[\![\Delta \Vdash K]\!]^{m+1}$. By line (13) of the logical relation we get

⋆ if $\mathcal{C}_1, \mathcal{D}_1, \mathcal{F}_1 \mapsto_{\Delta \Vdash K}^{*\Theta; \Upsilon_1} \mathcal{C}_1, \mathcal{D}_1', \mathcal{F}_1$, then
$\exists \mathcal{D}_2'.\, \mathcal{C}_2, \mathcal{D}_2, \mathcal{F}_2 \mapsto_{\Delta \Vdash K}^{*\Theta'; \Upsilon_1} \mathcal{C}_2, \mathcal{D}_2', \mathcal{F}_2$, and $\forall \mathcal{C}_1', \mathcal{F}_1', \mathcal{C}_2', \mathcal{F}_2'.\text{if } \mathcal{C}_1, \mathcal{D}_1', \mathcal{F}_1 \mapsto_{\Delta \Vdash K}^{*\Theta_1; \Upsilon_1} \mathcal{C}_1', \mathcal{D}_1', \mathcal{F}_1'$ and $\mathcal{C}_2, \mathcal{D}_2', \mathcal{F}_2 \mapsto_{\Delta \Vdash K}^{*\Theta_2; \Upsilon_1} \mathcal{C}_2', \mathcal{D}_2', \mathcal{F}_2'$, then
$\forall x \in \Upsilon_1.\, (\mathcal{C}_1', \mathcal{D}_1', \mathcal{F}_1'; \mathcal{C}_2', \mathcal{D}_2', \mathcal{F}_2') \in \mathcal{V}_\Psi^\xi[\![\Delta \Vdash K]\!]_{;;x}^{m+1}$ and
$\forall x \in (\Theta_1 \cap \Theta_2).\, (\mathcal{C}_1', \mathcal{D}_1', \mathcal{F}_1'; \mathcal{C}_2', \mathcal{D}_2', \mathcal{F}_2') \in \mathcal{V}_\Psi^\xi[\![\Delta \Vdash K]\!]_{x;:}^{m+1}$.

Consider $\mathcal{C}_2''$, for which by the assumption we have $\mathcal{C}_2 \mapsto^* \mathcal{C}_2''$. Our goal is to prove $(\mathcal{C}_1\mathcal{D}_1\mathcal{F}_1; \mathcal{C}_2''\mathcal{D}_2\mathcal{F}_2) \in \mathcal{E}_{\Psi_0}^\xi[\![\Delta \Vdash K]\!]^{m+1}$. We need to show that

† if $\mathcal{C}_1, \mathcal{D}_1, \mathcal{F}_1 \mapsto_{\Delta \Vdash K}^{*\Theta; \Upsilon_1} \mathcal{C}_1, \mathcal{D}_1', \mathcal{F}_1$, then
$\exists \mathcal{D}_2'.\, \mathcal{C}_2'', \mathcal{D}_2, \mathcal{F}_2 \mapsto_{\Delta \Vdash K}^{*\Theta''; \Upsilon_1} \mathcal{C}_2'', \mathcal{D}_2', \mathcal{F}_2$, and $\forall \mathcal{C}_1', \mathcal{F}_1', \mathcal{C}_2', \mathcal{F}_2'.\text{if } \mathcal{C}_1, \mathcal{D}_1', \mathcal{F}_1 \mapsto_{\Delta \Vdash K}^{*\Theta_1; \Upsilon_1} \mathcal{C}_1', \mathcal{D}_1', \mathcal{F}_1'$ and $\mathcal{C}_2'', \mathcal{D}_2', \mathcal{F}_2 \mapsto_{\Delta \Vdash K}^{*\Theta_2; \Upsilon_1} \mathcal{C}_2', \mathcal{D}_2', \mathcal{F}_2'$, then
$\forall x \in \Upsilon_1.\, (\mathcal{C}_1', \mathcal{D}_1', \mathcal{F}_1'; \mathcal{C}_2', \mathcal{D}_2', \mathcal{F}_2') \in \mathcal{V}_\Psi^\xi[\![\Delta \Vdash K]\!]_{;;x}^{m+1}$ and
$\forall x \in (\Theta_1 \cap \Theta_2).\, (\mathcal{C}_1', \mathcal{D}_1', \mathcal{F}_1'; \mathcal{C}_2', \mathcal{D}_2', \mathcal{F}_2') \in \mathcal{V}_\Psi^\xi[\![\Delta \Vdash K]\!]_{x;:}^{m+1}$.

With an if-Introduction on the goal followed by an if- Elimination on the assumption, we get the assumption

$\star'$ $\exists \mathcal{D}_2'. \mathcal{C}_2, \mathcal{D}_2, \mathcal{F}_2 \mapsto_{\Delta \Vdash K}^{*\Theta'; \Upsilon_1} \mathcal{C}_2, \mathcal{D}_2', \mathcal{F}_2$, and $\forall \mathcal{C}_1', \mathcal{F}_1', \mathcal{C}_2', \mathcal{F}_2'.$ if $\mathcal{C}_1, \mathcal{D}_1', \mathcal{F}_1 \mapsto_{\Delta \Vdash K}^{*\Theta_1; \Upsilon_1} \mathcal{C}_1', \mathcal{D}_1', \mathcal{F}_1'$ and $\mathcal{C}_2, \mathcal{D}_2', \mathcal{F}_2 \mapsto_{\Delta \Vdash K}^{*\Theta_2; \Upsilon_1} \mathcal{C}_2', \mathcal{D}_2', \mathcal{F}_2'$, then
$\forall x \in \Upsilon_1. (\mathcal{C}_1', \mathcal{D}_1', \mathcal{F}_1'; \mathcal{C}_2', \mathcal{D}_2', \mathcal{F}_2') \in \mathcal{V}_\Psi^\xi [\![\Delta \Vdash K]\!]_{:;x}^{m+1}$ and
$\forall x \in (\Theta_1 \cap \Theta_2). (\mathcal{C}_1', \mathcal{D}_1', \mathcal{F}_1'; \mathcal{C}_2', \mathcal{D}_2', \mathcal{F}_2') \in \mathcal{V}_\Psi^\xi [\![\Delta \Vdash K]\!]_{x;:}^{m+1}$.

and the goal

$\dagger'$ $\exists \mathcal{D}_2'. \mathcal{C}_2'', \mathcal{D}_2, \mathcal{F}_2 \mapsto_{\Delta \Vdash K}^{*\Theta''; \Upsilon_1} \mathcal{C}_2'', \mathcal{D}_2', \mathcal{F}_2$, and $\forall \mathcal{C}_1', \mathcal{F}_1', \mathcal{C}_2', \mathcal{F}_2'.$ if $\mathcal{C}_1, \mathcal{D}_1', \mathcal{F}_1 \mapsto_{\Delta \Vdash K}^{*\Theta_1; \Upsilon_1} \mathcal{C}_1', \mathcal{D}_1', \mathcal{F}_1'$ and $\mathcal{C}_2'', \mathcal{D}_2', \mathcal{F}_2 \mapsto_{\Delta \Vdash K}^{*\Theta_2; \Upsilon_1} \mathcal{C}_2', \mathcal{D}_2', \mathcal{F}_2'$, then
$\forall x \in \Upsilon_1. (\mathcal{C}_1', \mathcal{D}_1', \mathcal{F}_1'; \mathcal{C}_2', \mathcal{D}_2', \mathcal{F}_2') \in \mathcal{V}_\Psi^\xi [\![\Delta \Vdash K]\!]_{:;x}^{m+1}$ and
$\forall x \in (\Theta_1 \cap \Theta_2). (\mathcal{C}_1', \mathcal{D}_1', \mathcal{F}_1'; \mathcal{C}_2', \mathcal{D}_2', \mathcal{F}_2') \in \mathcal{V}_\Psi^\xi [\![\Delta \Vdash K]\!]_{x;:}^{m+1}$.

We apply an $\exists$-Elimination on the assumption to get $\mathcal{D}_2'$ that satisfies the conditions. We use the same $\mathcal{D}_2'$ to instantiate the existential quantifier in the goal.

Applying $\forall$-Introduction and if-Introduction on the goal, assume $\mathcal{C}_2'', \mathcal{D}_2', \mathcal{F}_2 \mapsto_{\Delta \Vdash K}^{*\Theta_2; \Upsilon_1} \mathcal{C}_2', \mathcal{D}_2', \mathcal{F}_2'$. From the assumption $\mathcal{C}_2 \mapsto^* \mathcal{C}_2''$, we get $\mathcal{C}_2, \mathcal{D}_2', \mathcal{F}_2 \mapsto_{\Delta \Vdash K}^{*\Theta_2; \Upsilon_1} \mathcal{C}_2', \mathcal{D}_2', \mathcal{F}_2'$ and we can use it to eliminate $\forall$ and if in the assumption. By that, we get

$\forall x \in \Upsilon_1. (\mathcal{C}_1', \mathcal{D}_1', \mathcal{F}_1'; \mathcal{C}_2', \mathcal{D}_2', \mathcal{F}_2') \in \mathcal{V}_\Psi^\xi [\![\Delta \Vdash K]\!]_{:;x}^{m+1}$ and $\forall x \in (\Theta_1 \cap \Theta_2). (\mathcal{C}_1', \mathcal{D}_1', \mathcal{F}_1'; \mathcal{C}_2', \mathcal{D}_2', \mathcal{F}_2') \in \mathcal{V}_\Psi^\xi [\![\Delta \Vdash K]\!]_{x;:}^{m+1}$.

which is exactly what we need to prove and the proof is complete.

$\square$

**Lemma C.12** (Forward closure on the first run). *Consider $(\mathcal{B}_1; \mathcal{B}_2) \in \mathcal{E}_{\Psi_0}^\xi [\![\Delta \Vdash K]\!]^k$ and $\mathcal{B}_i = \mathcal{C}_i \mathcal{D}_i \mathcal{F}_i$ with $(\mathcal{C}_1, \mathcal{D}_1, \mathcal{F}_1; \mathcal{C}_2, \mathcal{D}_2, \mathcal{F}_2) \in$ Tree$(\Delta \Vdash K)$ and $\mathcal{D}_1 \mapsto^* \mathcal{D}_1''$. We can prove $(\mathcal{C}_1 \mathcal{D}_1'' \mathcal{F}_1; \mathcal{C}_2 \mathcal{D}_2 \mathcal{F}_2) \in \mathcal{E}_{\Psi_0}^\xi [\![\Delta \Vdash K]\!]^k$.*

*Proof.* If $k = 0$ the proof is trivial. Consider $k = m + 1$. By assumption, we have $(\mathcal{C}_1 \mathcal{D}_1 \mathcal{F}_1; \mathcal{C}_2 \mathcal{D}_2 \mathcal{F}_2) \in \mathcal{E}_{\Psi_0}^\xi [\![\Delta \Vdash K]\!]^{m+1}$. By line (13) of the logical relation we get

$\star$ if $\mathcal{C}_1, \mathcal{D}_1, \mathcal{F}_1 \mapsto_{\Delta \Vdash K}^{*\Theta, \Upsilon_1} \mathcal{C}_1, \mathcal{D}_1', \mathcal{F}_1$, then
$\exists \mathcal{D}_2'. \mathcal{C}_2, \mathcal{D}_2, \mathcal{F}_2 \mapsto_{\Delta \Vdash K}^{*\Theta'; \Upsilon_1} \mathcal{C}_2, \mathcal{D}_2', \mathcal{F}_2$, and $\forall \mathcal{C}_1', \mathcal{F}_1', \mathcal{C}_2', \mathcal{F}_2'.$ if $\mathcal{C}_1, \mathcal{D}_1', \mathcal{F}_1 \mapsto_{\Delta \Vdash K}^{*\Theta_1; \Upsilon_1} \mathcal{C}_1', \mathcal{D}_1', \mathcal{F}_1'$ and $\mathcal{C}_2, \mathcal{D}_2', \mathcal{F}_2 \mapsto_{\Delta \Vdash K}^{*\Theta_2; \Upsilon_1} \mathcal{C}_2', \mathcal{D}_2', \mathcal{F}_2'$, then
$\forall x \in \Upsilon_1. (\mathcal{C}_1', \mathcal{D}_1', \mathcal{F}_1'; \mathcal{C}_2', \mathcal{D}_2', \mathcal{F}_2') \in \mathcal{V}_\Psi^\xi [\![\Delta \Vdash K]\!]_{:;x}^{m+1}$ and
$\forall x \in (\Theta_1 \cap \Theta_2). (\mathcal{C}_1', \mathcal{D}_1', \mathcal{F}_1'; \mathcal{C}_2', \mathcal{D}_2', \mathcal{F}_2') \in \mathcal{V}_\Psi^\xi [\![\Delta \Vdash K]\!]_{x;:}^{m+1}$.

Consider $\mathcal{D}_1''$, for which by the assumption we have $\mathcal{D}_1 \mapsto^* \mathcal{D}_1''$. Our goal is to prove $(\mathcal{C}_1 \mathcal{D}_1'' \mathcal{F}_1; \mathcal{C}_2 \mathcal{D}_2 \mathcal{F}_2) \in \mathcal{E}_{\Psi_0}^\xi [\![\Delta \Vdash K]\!]^{m+1}$. We need to show that

$\dagger$ if $\mathcal{C}_1, \mathcal{D}_1'', \mathcal{F}_1 \mapsto_{\Delta \Vdash K}^{*\Theta, \Upsilon_1} \mathcal{C}_1, \mathcal{D}_1', \mathcal{F}_1$, then
$\exists \mathcal{D}_2'. \mathcal{C}_2, \mathcal{D}_2, \mathcal{F}_2 \mapsto_{\Delta \Vdash K}^{*\Theta'; \Upsilon_1} \mathcal{C}_2, \mathcal{D}_2', \mathcal{F}_2$, and $\forall \mathcal{C}_1', \mathcal{F}_1', \mathcal{C}_2', \mathcal{F}_2'.$ if $\mathcal{C}_1, \mathcal{D}_1', \mathcal{F}_1 \mapsto_{\Delta \Vdash K}^{*\Theta_1; \Upsilon_1} \mathcal{C}_1', \mathcal{D}_1', \mathcal{F}_1'$ and $\mathcal{C}_2, \mathcal{D}_2', \mathcal{F}_2 \mapsto_{\Delta \Vdash K}^{*\Theta_2; \Upsilon_1} \mathcal{C}_2', \mathcal{D}_2', \mathcal{F}_2'$, then
$\forall x \in \Upsilon_1. (\mathcal{C}_1', \mathcal{D}_1', \mathcal{F}_1'; \mathcal{C}_2', \mathcal{D}_2', \mathcal{F}_2') \in \mathcal{V}_\Psi^\xi [\![\Delta \Vdash K]\!]_{:;x}^{m+1}$ and
$\forall x \in (\Theta_1 \cap \Theta_2). (\mathcal{C}_1', \mathcal{D}_1', \mathcal{F}_1'; \mathcal{C}_2', \mathcal{D}_2', \mathcal{F}_2') \in \mathcal{V}_\Psi^\xi [\![\Delta \Vdash K]\!]_{x;:}^{m+1}$.

With an if-Introduction on the goal, we assume $\mathcal{C}_1, \mathcal{D}_1'', \mathcal{F}_1 \mapsto_{\Delta \Vdash K}^{*\Theta, \Upsilon_1} \mathcal{C}_1', \mathcal{D}_1', \mathcal{F}_1'$. By assumption of $\mathcal{D}_1 \mapsto^* \mathcal{D}_1''$ we get $\mathcal{C}_1, \mathcal{D}_1, \mathcal{F}_1 \mapsto_{\Delta \Vdash K}^{*\Theta, \Upsilon_1} \mathcal{C}_1, \mathcal{D}_1', \mathcal{F}_1$. We use this to apply if- Elimination on the assumption, we get the assumption

$\star'$ $\exists \mathcal{D}_2'. \mathcal{C}_2, \mathcal{D}_2, \mathcal{F}_2 \mapsto_{\Delta \Vdash K}^{*\Theta'; \Upsilon_1} \mathcal{C}_2, \mathcal{D}_2', \mathcal{F}_2$, and $\forall \mathcal{C}_1', \mathcal{F}_1', \mathcal{C}_2', \mathcal{F}_2'.$ if $\mathcal{C}_1, \mathcal{D}_1', \mathcal{F}_1 \mapsto_{\Delta \Vdash K}^{*\Theta_1; \Upsilon_1} \mathcal{C}_1', \mathcal{D}_1', \mathcal{F}_1'$ and $\mathcal{C}_2, \mathcal{D}_2', \mathcal{F}_2 \mapsto_{\Delta \Vdash K}^{*\Theta_2; \Upsilon_1} \mathcal{C}_2', \mathcal{D}_2', \mathcal{F}_2'$, then
$\forall x \in \Upsilon_1. (\mathcal{C}_1', \mathcal{D}_1', \mathcal{F}_1'; \mathcal{C}_2', \mathcal{D}_2', \mathcal{F}_2') \in \mathcal{V}_\Psi^\xi [\![\Delta \Vdash K]\!]_{:;x}^{m+1}$ and
$\forall x \in (\Theta_1 \cap \Theta_2). (\mathcal{C}_1', \mathcal{D}_1', \mathcal{F}_1'; \mathcal{C}_2', \mathcal{D}_2', \mathcal{F}_2') \in \mathcal{V}_\Psi^\xi [\![\Delta \Vdash K]\!]_{x;:}^{m+1}$.

Which is exactly our goal and the proof is complete. $\square$

**Lemma C.13** (Forward closure on the second run with some specific conditions). *Consider $(\mathcal{B}_1; \mathcal{B}_2) \in \mathcal{E}_{\Psi_0}^\xi [\![\Delta \Vdash K]\!]^k$ and $\mathcal{B}_i = \mathcal{C}_i \mathcal{D}_i \mathcal{F}_i$ with $(\mathcal{C}_1, \mathcal{D}_1, \mathcal{F}_1; \mathcal{C}_2, \mathcal{D}_2, \mathcal{F}_2) \in$ Tree$(\Delta \Vdash K)$ and $\mathcal{D}_2 \mapsto^{*\Upsilon} \mathcal{D}_2''$ and $\mathcal{D}_2''$ sends along the same set $\Upsilon$ as $\mathcal{D}_1$ and is the closest set to $\mathcal{D}_2$ that sends along this set (as built in Lemma C.8). We can prove $(\mathcal{C}_1 \mathcal{D}_1 \mathcal{F}_1; \mathcal{C}_2 \mathcal{D}_2'' \mathcal{F}_2) \in \mathcal{E}_{\Psi_0}^\xi [\![\Delta \Vdash K]\!]^k$.*

*Proof.* If $k = 0$ the proof is trivial. Consider $k = m + 1$. By assumption, we have $(\mathcal{C}_1 \mathcal{D}_1 \mathcal{F}_1; \mathcal{C}_2 \mathcal{D}_2 \mathcal{F}_2) \in \mathcal{E}^\xi_{\Psi_0}[\![\Delta \Vdash K]\!]^{m+1}$. By line (13) of the logical relation we get

⋆ if $\mathcal{C}_1, \mathcal{D}_1, \mathcal{F}_1 \mapsto^{*\Theta;\Upsilon_1}_{\Delta \Vdash K} \mathcal{C}_1, \mathcal{D}_1, \mathcal{F}_1$, then
  $\exists \mathcal{D}'_2. \mathcal{C}_2, \mathcal{D}_2, \mathcal{F}_2 \mapsto^{*\Theta';\Upsilon_1}_{\Delta \Vdash K} \mathcal{C}_2, \mathcal{D}'_2, \mathcal{F}_2$, and $\forall \mathcal{C}'_1, \mathcal{F}'_1, \mathcal{C}'_2, \mathcal{F}'_2.$ if $\mathcal{C}_1, \mathcal{D}'_1, \mathcal{F}_1 \mapsto^{*\Theta_1;\Upsilon_1}_{\Delta \Vdash K} \mathcal{C}'_1, \mathcal{D}'_1, \mathcal{F}'_1$ and $\mathcal{C}_2, \mathcal{D}'_2, \mathcal{F}_2 \mapsto^{*\Theta_2;\Upsilon_1}_{\Delta \Vdash K} \mathcal{C}'_2, \mathcal{D}'_2, \mathcal{F}'_2$, then
  $\forall x \in \Upsilon_1. (\mathcal{C}'_1, \mathcal{D}'_1, \mathcal{F}'_1; \mathcal{C}'_2, \mathcal{D}'_2, \mathcal{F}'_2) \in \mathcal{V}^\xi_\Psi [\![\Delta \Vdash K]\!]^{m+1}_{:;x}$ and
  $\forall x \in (\Theta_1 \cap \Theta_2). (\mathcal{C}'_1, \mathcal{D}'_1, \mathcal{F}'_1; \mathcal{C}'_2, \mathcal{D}'_2, \mathcal{F}'_2) \in \mathcal{V}^\xi_\Psi [\![\Delta \Vdash K]\!]^{m+1}_{x;:}$.

Consider $\mathcal{D}''_2$, for which by the assumption we have $\mathcal{D}_2 \mapsto^* \mathcal{D}''_2$. Our goal is to prove $(\mathcal{C}_1 \mathcal{D}_1 \mathcal{F}_1; \mathcal{C}_2 \mathcal{D}''_2 \mathcal{F}_2) \in \mathcal{E}^\xi_{\Psi_0}[\![\Delta \Vdash K]\!]^{m+1}$. We need to show that

† if $\mathcal{C}_1, \mathcal{D}_1, \mathcal{F}_1 \mapsto^{*\Theta;\Upsilon_1}_{\Delta \Vdash K} \mathcal{C}_1, \mathcal{D}'_1, \mathcal{F}_1$, then
  $\exists \mathcal{D}'_2. \mathcal{C}_2, \mathcal{D}''_2, \mathcal{F}_2 \mapsto^{*\Theta';\Upsilon_1}_{\Delta \Vdash K} \mathcal{C}_2, \mathcal{D}'_2, \mathcal{F}_2$, and $\forall \mathcal{C}'_1, \mathcal{F}'_1, \mathcal{C}'_2, \mathcal{F}'_2.$ if $\mathcal{C}_1, \mathcal{D}'_1, \mathcal{F}_1 \mapsto^{*\Theta_1;\Upsilon_1}_{\Delta \Vdash K} \mathcal{C}'_1, \mathcal{D}'_1, \mathcal{F}'_1$ and $\mathcal{C}_2, \mathcal{D}'_2, \mathcal{F}_2 \mapsto^{*\Theta_2;\Upsilon_1}_{\Delta \Vdash K} \mathcal{C}'_2, \mathcal{D}'_2, \mathcal{F}'_2$, then
  $\forall x \in \Upsilon_1. (\mathcal{C}'_1, \mathcal{D}'_1, \mathcal{F}'_1; \mathcal{C}'_2, \mathcal{D}'_2, \mathcal{F}'_2) \in \mathcal{V}^\xi_\Psi [\![\Delta \Vdash K]\!]^{m+1}_{:;x}$ and
  $\forall x \in (\Theta_1 \cap \Theta_2). (\mathcal{C}'_1, \mathcal{D}'_1, \mathcal{F}'_1; \mathcal{C}'_2, \mathcal{D}'_2, \mathcal{F}'_2) \in \mathcal{V}^\xi_\Psi [\![\Delta \Vdash K]\!]^{m+1}_{x;:}$.

With an if-Introduction on the goal followed by an if- Elimination on the assumption, we get the assumption

⋆' $\exists \mathcal{D}'_2. \mathcal{C}_2, \mathcal{D}_2, \mathcal{F}_2 \mapsto^{*\Theta';\Upsilon_1}_{\Delta \Vdash K} \mathcal{C}_2, \mathcal{D}'_2, \mathcal{F}_2$, and $\forall \mathcal{C}'_1, \mathcal{F}'_1, \mathcal{C}'_2, \mathcal{F}'_2.$ if $\mathcal{C}_1, \mathcal{D}'_1, \mathcal{F}_1 \mapsto^{*\Theta_1;\Upsilon_1}_{\Delta \Vdash K} \mathcal{C}'_1, \mathcal{D}'_1, \mathcal{F}'_1$ and $\mathcal{C}_2, \mathcal{D}'_2, \mathcal{F}_2 \mapsto^{*\Theta_2;\Upsilon_1}_{\Delta \Vdash K} \mathcal{C}'_2, \mathcal{D}'_2, \mathcal{F}'_2$, then
  $\forall x \in \Upsilon_1. (\mathcal{C}'_1, \mathcal{D}'_1, \mathcal{F}'_1; \mathcal{C}'_2, \mathcal{D}'_2, \mathcal{F}'_2) \in \mathcal{V}^\xi_\Psi [\![\Delta \Vdash K]\!]^{m+1}_{:;x}$ and
  $\forall x \in (\Theta_1 \cap \Theta_2). (\mathcal{C}'_1, \mathcal{D}'_1, \mathcal{F}'_1; \mathcal{C}'_2, \mathcal{D}'_2, \mathcal{F}'_2) \in \mathcal{V}^\xi_\Psi [\![\Delta \Vdash K]\!]^{m+1}_{x;:}$.

and the goal

†' $\exists \mathcal{D}'_2. \mathcal{C}_2, \mathcal{D}''_2, \mathcal{F}_2 \mapsto^{*\Theta';\Upsilon_1}_{\Delta \Vdash K} \mathcal{C}_2, \mathcal{D}'_2, \mathcal{F}_2$, and $\forall \mathcal{C}'_1, \mathcal{F}'_1, \mathcal{C}'_2, \mathcal{F}'_2.$ if $\mathcal{C}_1, \mathcal{D}'_1, \mathcal{F}_1 \mapsto^{*\Theta_1;\Upsilon_1}_{\Delta \Vdash K} \mathcal{C}'_1, \mathcal{D}'_1, \mathcal{F}'_1$ and $\mathcal{C}_2, \mathcal{D}'_2, \mathcal{F}_2 \mapsto^{*\Theta_2;\Upsilon_1}_{\Delta \Vdash K} \mathcal{C}'_2, \mathcal{D}'_2, \mathcal{F}'_2$, then
  $\forall x \in \Upsilon_1. (\mathcal{C}'_1, \mathcal{D}'_1, \mathcal{F}'_1; \mathcal{C}'_2, \mathcal{D}'_2, \mathcal{F}'_2) \in \mathcal{V}^\xi_\Psi [\![\Delta \Vdash K]\!]^{m+1}_{:;x}$ and
  $\forall x \in (\Theta_1 \cap \Theta_2). (\mathcal{C}'_1, \mathcal{D}'_1, \mathcal{F}'_1; \mathcal{C}'_2, \mathcal{D}'_2, \mathcal{F}'_2) \in \mathcal{V}^\xi_\Psi [\![\Delta \Vdash K]\!]^{m+1}_{x;:}$.

We apply an ∃-Elimination on the assumption to get $\mathcal{D}'_2$ that satisfies the conditions. In particular, we know that $\Upsilon \subseteq \Upsilon_1$. We use the same $\mathcal{D}'_2$ to instantiate the existential quantifier in the goal. Since $\mathcal{D}''_2$ is the closest set sending along $\Upsilon_1$ to $\mathcal{D}_1$, we get $\mathcal{D}''_2 \mapsto^* \mathcal{D}'_2$, and thus $\mathcal{C}_2, \mathcal{D}''_2, \mathcal{F}_2 \mapsto^{*\Theta;\Upsilon_1}_{\Delta \Vdash K} \mathcal{C}_2, \mathcal{D}'_2, \mathcal{F}_2$, and the proof is complete. □

**Lemma C.14** (moving the existensial Q. over the universal Q.). *if we have*

$\forall m. \exists \mathcal{D}'_2. \mathcal{C}_2, \mathcal{D}_2, \mathcal{F}_2 \mapsto^{*\Theta;\Upsilon_1}_{\Delta \Vdash K} \mathcal{C}_2, \mathcal{D}'_2, \mathcal{F}_2$, and $\forall \mathcal{C}'_1, \mathcal{F}'_1, \mathcal{C}'_2, \mathcal{F}'_2.$ if $\mathcal{C}_1, \mathcal{D}'_1, \mathcal{F}_1 \mapsto^{*\Theta_1;\Upsilon_1}_{\Delta \Vdash K} \mathcal{C}'_1, \mathcal{D}'_1, \mathcal{F}'_1$ and $\mathcal{C}_2, \mathcal{D}'_2, \mathcal{F}_2 \mapsto^{*\Theta_2;\Upsilon_1}_{\Delta \Vdash K} \mathcal{C}'_2, \mathcal{D}'_2, \mathcal{F}'_2$, then
  $\forall x \in \Upsilon_1. (\mathcal{C}'_1, \mathcal{D}'_1, \mathcal{F}'_1; \mathcal{C}'_2, \mathcal{D}'_2, \mathcal{F}'_2) \in \mathcal{V}^\xi_\Psi [\![\Delta \Vdash K]\!]^{m+1}_{:;x}$ and
  $\forall x \in (\Theta_1 \cap \Theta_2). (\mathcal{C}'_1, \mathcal{D}'_1, \mathcal{F}'_1; \mathcal{C}'_2, \mathcal{D}'_2, \mathcal{F}'_2) \in \mathcal{V}^\xi_\Psi [\![\Delta \Vdash K]\!]^{m+1}_{x;:}$.

*then*

$\exists \mathcal{D}'_2. \mathcal{C}_2, \mathcal{D}_2, \mathcal{F}_2 \mapsto^{*\Theta;\Upsilon_1}_{\Delta \Vdash K} \mathcal{C}_2, \mathcal{D}'_2, \mathcal{F}_2$, and $\forall \mathcal{C}'_1, \mathcal{F}'_1, \mathcal{C}'_2, \mathcal{F}'_2.$ if $\mathcal{C}_1, \mathcal{D}'_1, \mathcal{F}_1 \mapsto^{*\Theta_1;\Upsilon_1}_{\Delta \Vdash K} \mathcal{C}'_1, \mathcal{D}'_1, \mathcal{F}'_1$ and $\mathcal{C}_2, \mathcal{D}'_2, \mathcal{F}_2 \mapsto^{*\Theta_2;\Upsilon_1}_{\Delta \Vdash K} \mathcal{C}'_2, \mathcal{D}'_2, \mathcal{F}'_2$, then
  $\forall x \in \Upsilon_1. \forall k. (\mathcal{C}'_1, \mathcal{D}'_1, \mathcal{F}'_1; \mathcal{C}'_2, \mathcal{D}'_2, \mathcal{F}'_2) \in \mathcal{V}^\xi_\Psi [\![\Delta \Vdash K]\!]^{k+1}_{:;x}$ and
  $\forall x \in (\Theta_1 \cap \Theta_2). \forall k. (\mathcal{C}'_1, \mathcal{D}'_1, \mathcal{F}'_1; \mathcal{C}'_2, \mathcal{D}'_2, \mathcal{F}'_2) \in \mathcal{V}^\xi_\Psi [\![\Delta \Vdash K]\!]^{k+1}_{x;:}$.

*Proof.* Put $m = 1$ and get a $\mathcal{D}'_2$ such that $\mathcal{D}_2 \mapsto^{*\Upsilon_1} \mathcal{D}'_2$. By Lemma C.8, we can find a $\mathcal{D}''_2$ that sends along $\Upsilon_1$ and for every $\mathcal{D}$ such that $\mathcal{D}_2 \mapsto^{*\Upsilon_1} \mathcal{D}$, we get $\mathcal{D}''_2 \mapsto^{*\Upsilon_1} \mathcal{D}$. Use this $\mathcal{D}''_2$ to instantiate the existential in the goal, and use ∀-Introduction and if-Introduction on the goal. In particular, instantiate $k$ with an arbitrary natural number. The goals are:

$$(\mathcal{C}'_1, \mathcal{D}'_1, \mathcal{F}'_1; \mathcal{C}'_2, \mathcal{D}''_2, \mathcal{F}'_2) \in \mathcal{V}^\xi_\Psi [\![\Delta \Vdash K]\!]^{k+1}_{:;x} \text{ and } (\mathcal{C}'_1, \mathcal{D}'_1, \mathcal{F}'_1; \mathcal{C}'_2, \mathcal{D}''_2, \mathcal{F}'_2) \in \mathcal{V}^\xi_\Psi [\![\Delta \Vdash K]\!]^{k+1}_{x;:}.$$

Instantiate the assumption with $m = k$ and apply ∃-Introduction to get a $\mathcal{D}'$ that satisfies the conditions. Instantiate $\mathcal{C}'_2, \mathcal{F}'_2$ and $x$ as those chosen for the goal. We have as assumptions:

$$(\mathcal{C}'_1, \mathcal{D}'_1, \mathcal{F}'_1; \mathcal{C}'_2, \mathcal{D}'_2, \mathcal{F}'_2) \in \mathcal{V}^\xi_\Psi [\![\Delta \Vdash K]\!]^{k+1}_{:;x} \text{ and } (\mathcal{C}'_1, \mathcal{D}'_1, \mathcal{F}'_1; \mathcal{C}'_2, \mathcal{D}'_2, \mathcal{F}'_2) \in \mathcal{V}^\xi_\Psi [\![\Delta \Vdash K]\!]^{k+1}_{x;:}.$$

and by the way we built $\mathcal{D}''_2$, we know that $\mathcal{D}''_2 \mapsto^* \mathcal{D}'_2$. We can apply the backward closure results of Lemma C.10 to get the goal from the assumptions.

□

**Lemma C.15.** *Consider* $\Delta, u_\alpha^c{:}T \Vdash \mathcal{D}_i :: K$ *and* $\Delta' \Vdash \mathcal{T}_i :: u_\alpha^c{:}T$, *and* $\cdot \Vdash \mathcal{C}_i :: \Delta$ *and* $\cdot \Vdash \mathcal{C}'_i :: \Delta'$ *and* $K \Vdash \mathcal{F}_i :: \_$. *If* $\dagger_1 \forall m. (\mathcal{C}_1\mathcal{C}'_1\mathcal{T}_1\mathcal{D}_1\mathcal{F}_1; \mathcal{C}_2\mathcal{C}'_2\mathcal{T}_2\mathcal{D}_2\mathcal{F}_2) \in \mathcal{E}^\xi_{\Psi_0}[\![\Delta, u_\alpha^c{:}T \Vdash K]\!]^m$ *and* $\dagger_2 \forall m. (\mathcal{C}'_1\mathcal{T}_1\mathcal{C}_1\mathcal{D}_1\mathcal{F}_1; \mathcal{C}'_2\mathcal{T}_2\mathcal{C}_2\mathcal{D}_2\mathcal{F}_2) \in \mathcal{E}^\xi_{\Psi_0}[\![\Delta' \Vdash u_\alpha^c{:}T]\!]^m$ *then* $\dagger \forall k. (\mathcal{C}_1\mathcal{C}'_1\mathcal{T}_1\mathcal{D}_1\mathcal{F}_1; \mathcal{C}_2\mathcal{C}'_2\mathcal{T}_2\mathcal{D}_2\mathcal{F}_2) \in \mathcal{E}^\xi_{\Psi_0}[\![\Delta', \Delta \Vdash K]\!]^k$.

*Proof.* It is a corollary of the following lemma in which we have multiple trees instead of just one. □

**Lemma C.16.** *Consider* $\Delta \Vdash \mathcal{D}_i :: K$, *and* $\cdot \Vdash \mathcal{C}_i :: \Delta$ *and* $K \Vdash \mathcal{F}_i :: \_$ *such that* $\mathcal{D}_i = \mathcal{A}_i^1, \cdots, \mathcal{A}_i^n$ *and for every* $\ell \leq n$, *we get* $\Delta'_\ell \Vdash \mathcal{A}_i^\ell :: K'_\ell$, *i.e.* $\mathcal{D}_i$ *can be disjointly separated into trees* $\{\mathcal{A}_i^\ell\}_{\ell \leq n}$ *that agree on the types. For every* $\Delta'_\ell \Vdash \mathcal{A}_i^\ell :: K'_\ell$, *we rewrite the closed configuration* $\mathcal{C}_i\mathcal{D}_i\mathcal{F}_i$ *as* $\mathcal{C}_i^{\ell_1}\mathcal{B}_i^\ell\mathcal{A}_i^\ell\mathcal{C}_i^{\ell_2}\mathcal{T}_i^\ell\mathcal{F}_i$ *such that* $\cdot \Vdash \mathcal{C}_i^{\ell_1}\mathcal{B}_i^\ell :: \Delta'_\ell$ *and* $K'_\ell \Vdash \mathcal{C}_i^{\ell_2}\mathcal{T}_i^\ell :: K$ *where* $\mathcal{C}_i = \mathcal{C}_i^{\ell_1}\mathcal{C}_i^{\ell_2}$ *and* $\mathcal{D}_i = \mathcal{B}_i^\ell\mathcal{A}_i^\ell\mathcal{T}_i^\ell$. *The intuition is that* $\mathcal{B}_i^\ell$ *is the collection of sub-trees in* $\mathcal{D}_i$ *required for providing the resources for* $\mathcal{A}_i^\ell$ *(they are at the bottom). And* $\mathcal{T}_i^\ell$ *are the rest of the sub-trees (they are at the top).*

*In this setup, if for all* $\ell$ *we have*

$$\dagger_1^\ell \, \forall m. (\mathcal{C}_1^{\ell_1}\mathcal{B}_1^\ell\mathcal{A}_1^\ell\mathcal{C}_1^{\ell_2}\mathcal{T}_1^\ell\mathcal{F}_1; \mathcal{C}_2^{\ell_1}\mathcal{B}_2^\ell\mathcal{A}_2^\ell\mathcal{C}_2^{\ell_2}\mathcal{T}_2^\ell\mathcal{F}_2) \in \mathcal{E}^\xi_{\Psi_0}[\![\Delta'_\ell \Vdash K'_\ell]\!]^m, \text{then}$$

$$\dagger \, \forall k. (\mathcal{C}_1\mathcal{D}_1\mathcal{F}_1; \mathcal{C}_2\mathcal{D}_2\mathcal{F}_2) \in \mathcal{E}^\xi_{\Psi_0}[\![\Delta \Vdash K]\!]^k.$$

*Proof.* Consider an arbitrary $k$ for which we want to prove $(\mathcal{C}_1\mathcal{D}_1\mathcal{F}_1; \mathcal{C}_2\mathcal{D}_2\mathcal{F}_2) \in \mathcal{E}^\xi_{\Psi_0}[\![\Delta \Vdash K]\!]^k$. We proceed the proof by induction on $k$. If $k = 0$ the proof trivially holds, otherwise consider $k = m + 1$. We need to show that the condition on line 13 of the logical relation holds. Assume

$$\mathcal{C}_1, \mathcal{D}_1, \mathcal{F}_1 \mapsto_{\Delta \Vdash K}^{*\Theta; \Upsilon_1} \mathcal{C}_1, \mathcal{D}'_1, \mathcal{F}_1$$

**in $j$ steps**. We want to show:

* $\exists \mathcal{D}'_2. \mathcal{C}_2, \mathcal{D}_2, \mathcal{F}_2 \mapsto_{\Delta \Vdash K}^{*\Theta'; \Upsilon_1} \mathcal{C}_2, \mathcal{D}'_2, \mathcal{F}_2$, and $\forall \mathcal{C}'_1, \mathcal{F}'_1, \mathcal{C}'_2, \mathcal{F}'_2$. if $\mathcal{C}_1, \mathcal{D}'_1, \mathcal{F}_1 \mapsto_{\Delta \Vdash K}^{*\Theta_1; \Upsilon_1} \mathcal{C}'_1, \mathcal{D}'_1, \mathcal{F}'_1$ and $\mathcal{C}_2, \mathcal{D}'_2, \mathcal{F}_2 \mapsto_{\Delta \Vdash K}^{*\Theta_2; \Upsilon_1} \mathcal{C}'_2, \mathcal{D}'_2, \mathcal{F}'_2$, then
$\forall x \in \Upsilon_1. (\mathcal{C}'_1, \mathcal{D}'_1, \mathcal{F}'_1; \mathcal{C}'_2, \mathcal{D}'_2, \mathcal{F}'_2) \in \mathcal{V}^\xi_\Psi[\![\Delta \Vdash K]\!]^{m+1}_{\cdot;x}$ and
$\forall x \in (\Theta_1 \cap \Theta_2). (\mathcal{C}'_1, \mathcal{D}'_1, \mathcal{F}'_1; \mathcal{C}'_2, \mathcal{D}'_2, \mathcal{F}'_2) \in \mathcal{V}^\xi_\Psi[\![\Delta \Vdash K]\!]^{m+1}_{x;\cdot}$.

We prove it by another nested induction on $j$:

**Case 1.** $j = 0$. We want to show:

* $\exists \mathcal{D}'_2. \mathcal{C}_2, \mathcal{D}_2, \mathcal{F}_2 \mapsto_{\Delta \Vdash K}^{*\Theta'; \Upsilon_1} \mathcal{C}_2, \mathcal{D}'_2, \mathcal{F}_2$, and $\forall \mathcal{C}'_1, \mathcal{F}'_1, \mathcal{C}'_2, \mathcal{F}'_2$. if $\mathcal{C}_1, \mathcal{D}_1, \mathcal{F}_1 \mapsto_{\Delta \Vdash K}^{*\Theta_1; \Upsilon_1} \mathcal{C}'_1, \mathcal{D}_1, \mathcal{F}'_1$ and $\mathcal{C}_2, \mathcal{D}'_2, \mathcal{F}_2 \mapsto_{\Delta \Vdash K}^{*\Theta_2; \Upsilon_1} \mathcal{C}'_2, \mathcal{D}'_2, \mathcal{F}'_2$, then
$\forall x \in \Upsilon_1. (\mathcal{C}'_1, \mathcal{D}_1, \mathcal{F}'_1; \mathcal{C}'_2, \mathcal{D}'_2, \mathcal{F}'_2) \in \mathcal{V}^\xi_\Psi[\![\Delta \Vdash K]\!]^{m+1}_{\cdot;x}$ and
$\forall x \in (\Theta_1 \cap \Theta_2). (\mathcal{C}'_1, \mathcal{D}_1, \mathcal{F}'_1; \mathcal{C}'_2, \mathcal{D}'_2, \mathcal{F}'_2) \in \mathcal{V}^\xi_\Psi[\![\Delta \Vdash K]\!]^{m+1}_{x;\cdot}$.

For each $\ell$, we know by assumption $\dagger_\ell$ that

* $\star_\ell \, \forall m. \exists \mathcal{A}_2^{\ell'}. \mathcal{C}_2^{\ell_1}\mathcal{B}_2^\ell, \mathcal{A}_2^\ell, \mathcal{C}_2^{\ell_2}\mathcal{T}_2^\ell\mathcal{F}_2 \mapsto_{\Delta'_\ell \Vdash K'_\ell}^{*\Theta^{\ell'}; \Upsilon_1^\ell} \mathcal{C}_2^{\ell_1}\mathcal{B}_2^\ell, \mathcal{A}_2^{\ell'}, \mathcal{C}_2^{\ell_2}\mathcal{T}_2^\ell\mathcal{F}_2$, and
$\forall \mathcal{C}'_1, \mathcal{F}'_1, \mathcal{C}'_2, \mathcal{F}'_2$. if $\mathcal{C}_1^{\ell_1}\mathcal{B}_1^\ell, \mathcal{A}_1^\ell, \mathcal{C}_1^{\ell_2}\mathcal{T}_1^\ell\mathcal{F}_1 \mapsto_{\Delta'_\ell \Vdash K'_\ell}^{*\Theta_1^\ell; \Upsilon_1^\ell} \mathcal{C}'_1, \mathcal{A}_1^\ell, \mathcal{F}'_1$ and $\mathcal{C}_2^{\ell_1}\mathcal{B}_2^\ell, \mathcal{A}_2^{\ell'}, \mathcal{C}_2^{\ell_2}\mathcal{T}_2^\ell\mathcal{F}_2 \mapsto_{\Delta'_\ell \Vdash K'_\ell}^{*\Theta_2^\ell; \Upsilon_1^\ell} \mathcal{C}'_2, \mathcal{A}_2^{\ell'}, \mathcal{F}'_2$, then
$\forall x \in \Upsilon_1^\ell. (\mathcal{C}'_1, \mathcal{A}_1^\ell, \mathcal{F}'_1; \mathcal{C}'_2, \mathcal{A}_2^{\ell'}, \mathcal{F}'_2) \in \mathcal{V}^\xi_\Psi[\![\Delta'_\ell \Vdash K'_\ell]\!]^{m+1}_{\cdot;x}$ and
$\forall x \in (\Theta_1^\ell \cap \Theta_2^\ell). (\mathcal{C}'_1, \mathcal{A}_1^\ell, \mathcal{F}'_1; \mathcal{C}'_2, \mathcal{A}_2^{\ell'}, \mathcal{F}'_2) \in \mathcal{V}^\xi_\Psi[\![\Delta'_\ell \Vdash K'_\ell]\!]^{m+1}_{x;\cdot}$.

By Lemma C.8, we get

* $\star'_\ell \, \exists \mathcal{A}_2^{\ell'}. \mathcal{C}_2^{\ell_1}\mathcal{B}_2^\ell, \mathcal{A}_2^\ell, \mathcal{C}_2^{\ell_2}\mathcal{T}_2^\ell\mathcal{F}_2 \mapsto_{\Delta'_\ell \Vdash K'_\ell}^{*\Theta^{\ell'}; \Upsilon_1^\ell} \mathcal{C}_2^{\ell_1}\mathcal{B}_2^\ell, \mathcal{A}_2^{\ell'}, \mathcal{C}_2^{\ell_2}\mathcal{T}_2^\ell\mathcal{F}_2$, and
$\forall \mathcal{C}'_1, \mathcal{F}'_1, \mathcal{C}'_2, \mathcal{F}'_2$. if $\mathcal{C}_1^{\ell_1}\mathcal{B}_1^\ell, \mathcal{A}_1^\ell, \mathcal{C}_1^{\ell_2}\mathcal{T}_1^\ell\mathcal{F}_1 \mapsto_{\Delta'_\ell \Vdash K'_\ell}^{*\Theta_1^\ell; \Upsilon_1^\ell} \mathcal{C}'_1, \mathcal{A}_1^\ell, \mathcal{F}'_1$ and $\mathcal{C}_2^{\ell_1}\mathcal{B}_2^\ell, \mathcal{A}_2^{\ell'}, \mathcal{C}_2^{\ell_2}\mathcal{T}_2^\ell\mathcal{F}_2 \mapsto_{\Delta'_\ell \Vdash K'_\ell}^{*\Theta_2^\ell; \Upsilon_1^\ell} \mathcal{C}'_2, \mathcal{A}_2^{\ell'}, \mathcal{F}'_2$, then
$\forall x \in \Upsilon_1^\ell. \forall m. (\mathcal{C}'_1, \mathcal{A}_1^\ell, \mathcal{F}'_2; \mathcal{C}'_2, \mathcal{A}_2^{\ell'}, \mathcal{F}'_2) \in \mathcal{V}^\xi_\Psi[\![\Delta'_\ell \Vdash K'_\ell]\!]^{m+1}_{\cdot;x}$ and
$\forall x \in (\Theta_1^\ell \cap \Theta_2^\ell). \forall m. (\mathcal{C}'_1, \mathcal{A}_1^\ell, \mathcal{F}'_1; \mathcal{C}'_2, \mathcal{A}_2^{\ell'}, \mathcal{F}'_2) \in \mathcal{V}^\xi_\Psi[\![\Delta'_\ell \Vdash K'_\ell]\!]^{m+1}_{x;\cdot}$.

We pick the $\mathcal{A}_2^{\ell'}$ which by lemma C.8 we know is the minimal sending configuration with respect to $\mathcal{A}_2^\ell$ and $\Upsilon_1$. We build $\mathcal{D}'_2$ as the collection of $\mathcal{A}_2^{\ell'}$ configurations, i.e. $\mathcal{D}'_2 = \{A_2^{\ell'}\}_{\ell \leq n}$. Instantiate the existential quantifier in the goal $(\star)$ with $\mathcal{D}'_2$ that we built for which we know $\mathcal{C}_2, \mathcal{D}_2, \mathcal{F}_2 \mapsto_{\Delta \Vdash K}^{*\Theta; \Upsilon_1} \mathcal{C}_2, \mathcal{D}'_2, \mathcal{F}_2$. We need to show

$\star'$ $\forall \mathcal{C}'_1, \mathcal{F}_1, \mathcal{C}'_2, \mathcal{F}'_2.$ if $\mathcal{C}_1, \mathcal{D}_1, \mathcal{F}_1 \mapsto^{*\Theta_1;\Upsilon_1}_{\Delta \Vdash K} \mathcal{C}'_1, \mathcal{D}_1, \mathcal{F}'_1$ and $\mathcal{C}_2, \mathcal{D}'_2, \mathcal{F}_2 \mapsto^{*\Theta_2;\Upsilon_1}_{\Delta \Vdash K} \mathcal{C}'_2, \mathcal{D}'_2, \mathcal{F}'_2$, then
$\forall x \in \Upsilon_1. (\mathcal{C}'_1, \mathcal{D}_1, \mathcal{F}'_1; \mathcal{C}'_2, \mathcal{D}'_2, \mathcal{F}'_2) \in \mathcal{V}^\xi_\Psi [\![\Delta \Vdash K]\!]^{m+1}_{:;x}$ and
$\forall x \in (\Theta_1 \cap \Theta_2). (\mathcal{C}'_1, \mathcal{D}_1, \mathcal{F}'_1; \mathcal{C}'_2, \mathcal{D}'_2, \mathcal{F}'_2) \in \mathcal{V}^\xi_\Psi [\![\Delta \Vdash K]\!]^{m+1}_{x;:}$.

We instantiate the universal quantifier in the goal with arbitrary $\mathcal{C}'_1, \mathcal{F}'_1, \mathcal{C}'_2,$ and $\mathcal{F}'_2$ such that $\mathcal{C}_1, \mathcal{D}_1, \mathcal{F}_1 \mapsto^{*\Theta_1;\Upsilon_1}_{\Delta \Vdash K} \mathcal{C}'_1, \mathcal{D}_1, \mathcal{F}'_1$ and $\mathcal{C}_2, \mathcal{D}'_2, \mathcal{F}_2 \mapsto^{*\Theta_2;\Upsilon_1}_{\Delta \Vdash K} \mathcal{C}'_2, \mathcal{D}'_2, \mathcal{F}'_2$. Now we can instantiate the universal quantifier in the assumptions. For each $\star'_\ell$, instantiate the quantifiers respectively with $\mathcal{C}^{\ell'_1}_1 \mathcal{B}^\ell_1$ (for the first lower substitution as), $\mathcal{C}^{\ell'_1}_1 \mathcal{T}^\ell_1 \mathcal{F}'_1$ (for the first upper substitution), $\mathcal{C}^{\ell'_1}_2 \mathcal{B}^{\ell'}_2$ (for the second lower substitution) and $\mathcal{C}^{\ell'_2}_2 \mathcal{T}^{\ell'}_2 \mathcal{F}'_2$ (for the second upper substitution). Here, $\mathcal{B}^\ell_1$ and $\mathcal{T}^\ell_1$ are the subtrees below, and at the top of $\mathcal{A}^\ell_1$ in $\mathcal{D}_1$, respectively. Similarly, $\mathcal{B}^{\ell'}_2$ and $\mathcal{T}^{\ell'}_2$ are the subtrees below, and at the top of $\mathcal{A}^{\ell'}_2$ in $\mathcal{D}'_2$, respectively. $\mathcal{C}^{\ell'_1}_i$ and $\mathcal{C}^{\ell'_2}_i$ are the sub-configurations of $\mathcal{C}'_i$ that provide the resources to bottom and top subtrees, respectively.

$\star''_\ell$ $\forall x \in \Upsilon^\ell_1. \forall m. (\mathcal{C}^{\ell'_1}_1 \mathcal{B}^\ell_1, \mathcal{A}^\ell_1, \mathcal{C}^{\ell'_2}_1 \mathcal{T}^\ell_1 \mathcal{F}'_1; \mathcal{C}^{\ell'_1}_2 \mathcal{B}^{\ell'}_2, \mathcal{A}^{\ell'}_2, \mathcal{C}^{\ell'_2}_2 \mathcal{T}^{\ell'}_2 \mathcal{F}'_2) \in \mathcal{V}^\xi_\Psi [\![\Delta'_\ell \Vdash K'_\ell]\!]^{m+1}_{:;x}$ and
$\forall x \in (\Theta^\ell_1 \cap \Theta^\ell_2). \forall m. (\mathcal{C}^{\ell'_1}_1 \mathcal{B}^\ell_1, \mathcal{A}^\ell_1, \mathcal{C}^{\ell'_2}_1 \mathcal{T}^\ell_1 \mathcal{F}'_1; \mathcal{C}^{\ell'_1}_2 \mathcal{B}^{\ell'}_2, \mathcal{A}^{\ell'}_2, \mathcal{C}^{\ell'_2}_2 \mathcal{T}^{\ell'}_2 \mathcal{F}'_2) \in \mathcal{V}^\xi_\Psi [\![\Delta'_\ell \Vdash K'_\ell]\!]^{m+1}_{x;:}$.

Consider a $y_\alpha \in \Upsilon_1$ or $y_\alpha \in \Theta_1 \cap \Theta_2$ for which we want to prove the goal. Without loss of generality, assume $y$ is a resource or the offering channel of $\mathcal{A}^\kappa_2$. In the first case we know that $y_\alpha \in \Upsilon^\kappa_1$, and in the later we have $y_\alpha \in \Theta^\kappa_1 \cap \Theta^\kappa_2$. Thus, if $y_\alpha \in \Upsilon_1$ we have as the assumption

$\star''_\kappa$ $\forall m. (\mathcal{C}^{\kappa'_1}_1 \mathcal{B}^\kappa_1, \mathcal{A}^\kappa_1, \mathcal{C}^{\kappa'_2}_1 \mathcal{T}^\kappa_1 \mathcal{F}'_1; \mathcal{C}^{\kappa'_1}_2 \mathcal{B}^{\kappa'}_2, \mathcal{A}^{\kappa'}_2, \mathcal{C}^{\kappa'_2}_2 \mathcal{T}^{\kappa'}_2 \mathcal{F}'_2) \in \mathcal{V}^\xi_\Psi [\![\Delta'_\kappa \Vdash K_\kappa]\!]^{m+1}_{:;y_\alpha}$

And if $y_\alpha \in \Theta_1 \cap \Theta_2$ we have as the assumption

$\star''_\kappa$ $\forall m. (\mathcal{C}^{\kappa'_1}_1 \mathcal{B}^\kappa_1, \mathcal{A}^\kappa_1, \mathcal{C}^{\kappa'_2}_1 \mathcal{T}^\kappa_1 \mathcal{F}'_1; \mathcal{C}^{\kappa'_1}_2 \mathcal{B}^{\kappa'}_2, \mathcal{A}^{\kappa'}_2, \mathcal{C}^{\kappa'_2}_2 \mathcal{T}^{\kappa'}_2 \mathcal{F}'_2) \in \mathcal{V}^\xi_\Psi [\![\Delta'_\kappa \Vdash K_\kappa]\!]^{m+1}_{y_\alpha;:}$.

Moreover, for every $\ell \neq \kappa$, we can apply forward closure on the second run and on the substitutions on $\dagger_\ell$ to get

$\dagger_{\ell''}$ $\forall m. (\mathcal{C}^{\ell'_1}_1 \mathcal{B}^\ell_1 \mathcal{A}^\ell_1 \mathcal{C}^{\ell'_2}_1 \mathcal{T}^\ell_1 \mathcal{F}'_1; \mathcal{C}^{\ell'_1}_2 \mathcal{B}^{\ell'}_2 \mathcal{A}^{\ell'}_2 \mathcal{C}^{\ell'_2}_2 \mathcal{T}^{\ell'}_2 \mathcal{F}'_2) \in \mathcal{E}^\xi_{\Psi_0} [\![\Delta'_\ell \Vdash K'_\ell]\!]^m$

Note that we can apply the forward closure on the second run since the conditions of Lemma C.13 are satisfied, i.e. $\mathcal{A}'^\ell_2$ sends the same messages as $\mathcal{A}^\ell_1$ and is the minimal sending configuration with respect to $\mathcal{A}^{\ell'}_2$ and $\Upsilon^\ell_1$.

We continue the proof by considering the row in the logical relation that makes the assumptions $\star''_\kappa$ correct. We only consider five interesting cases. The proof of other cases is similar.

**Row 4.** $\Delta'_\kappa = \Delta^3_\kappa, \Delta^4_\kappa$ and $K = K'_\kappa = y_\alpha{:}A \otimes B[c]$. We have $\mathcal{C}^{\kappa'_1}_1 = \mathcal{C}^{\kappa_3}_1 \mathcal{C}^{\kappa_4}_1$ and $\mathcal{B}^\kappa_1 = \mathcal{B}^{\kappa_3}_1 \mathcal{B}^{\kappa_4}_1$ such that $\mathcal{C}^{\kappa_3}_1 \mathcal{B}^{\kappa_3}_1 \in \mathsf{Tree}_\Psi(\cdot \Vdash \Delta^3_\kappa)$ and $\mathcal{C}^{\kappa_4}_1 \mathcal{B}^{\kappa_4}_1 \in \mathsf{Tree}_\Psi(\cdot \Vdash \Delta^4_\kappa)$. Moreover, $\mathcal{C}^{\kappa'_1}_2 = \mathcal{C}^{\kappa'_3}_2 \mathcal{C}^{\kappa'_4}_2$ and $\mathcal{B}^\kappa_2 = \mathcal{B}^{\kappa_3}_2 \mathcal{B}^{\kappa_4}_2$ such that $\mathcal{C}^{\kappa'_3}_2 \mathcal{B}^{\kappa_3}_2 \in \mathsf{Tree}_\Psi(\cdot \Vdash \Delta^3_\kappa)$ and $\mathcal{C}^{\kappa_4}_2 \mathcal{B}^{\kappa_4}_2 \in \mathsf{Tree}_\Psi(\cdot \Vdash \Delta^4_\kappa)$.

And $(\mathcal{C}^{\kappa_3}_1 \mathcal{C}^{\kappa_4}_1 \mathcal{B}^{\kappa_3}_1 \mathcal{B}^{\kappa_4}_1, \mathcal{A}^\kappa_1, \mathcal{C}^{\kappa_2}_1 \mathcal{T}^\kappa_1 \mathcal{F}'_1; \mathcal{C}^{\kappa'_3}_2 \mathcal{C}^{\kappa'_4}_2 \mathcal{B}^{\kappa_3}_2 \mathcal{B}^{\kappa_4}_2, \mathcal{A}^{\kappa'}_2, \mathcal{C}^{\kappa'_2}_2 \mathcal{T}^{\kappa'}_2 \mathcal{F}'_2) \in \mathsf{Tree}_\Psi(\Delta^3_\kappa, \Delta^4_\kappa \Vdash y_\alpha{:}A \otimes B[c])$.

Moreover, we have

$$\mathcal{A}^\kappa_1 = \mathcal{A}^{\kappa'}_1 \mathcal{H}_1 \mathbf{msg}(\mathbf{send}\, x^c_\beta\, y^c_\alpha; y^c_\alpha \leftarrow u^c_\delta) \text{ for } \mathcal{H}_1 \in \mathsf{Tree}_\Psi(\Delta^4_\kappa \Vdash x_\beta{:}A[c]) \text{ and}$$

$$\mathcal{A}^{\kappa'}_2 = \mathcal{A}^{\kappa''}_2 \mathcal{H}_2 \mathbf{msg}(\mathbf{send}\, x^c_\beta\, y^c_\alpha; y^c_\alpha \leftarrow u^c_\delta) \text{ for } \mathcal{H}_2 \in \mathsf{Tree}_\Psi(\Delta^4_\kappa \Vdash x_\beta{:}A[c]) \text{ and}$$

$\circ_1$ $\forall m. (\mathcal{C}^{\kappa_4}_1 \mathcal{B}^{\kappa_4}_1, \mathcal{H}_1, \mathcal{C}^{\kappa_3}_1 \mathcal{B}^{\kappa_3}_1 \mathcal{A}^{\kappa'}_1 \mathbf{msg}(\mathbf{send}\, x^c_\beta\, y^c_\alpha; y^c_\alpha \leftarrow u^c_\delta) \mathcal{C}^{\kappa_2}_1 \mathcal{T}^\kappa_1 \mathcal{F}'_1; \mathcal{C}^{\kappa'_4}_2 \mathcal{B}^{\kappa_4}_2, \mathcal{H}_2, \mathcal{C}^{\kappa'_3}_2 \mathcal{B}^{\kappa_3}_2 \mathcal{A}^{\kappa''}_2 \mathbf{msg}(\mathbf{send}\, x^c_\beta\, y^c_\alpha; y^c_\alpha \leftarrow u^c_\delta) \mathcal{C}^{\kappa'_2}_2 \mathcal{T}^{\kappa'}_2 \mathcal{F}'_2) \in \mathcal{E}^\xi_\Psi [\![\Delta^4_\kappa \Vdash x_\beta{:}A[c]]\!]^m$ and

$\circ_2$ $\forall m. (\mathcal{C}^{\kappa_3}_1 \mathcal{B}^{\kappa_3}_1, \mathcal{A}^{\kappa'}_1, \mathcal{C}^{\kappa_4}_1 \mathcal{B}^{\kappa_4}_1 \mathcal{H}_1 \mathbf{msg}(\mathbf{send}\, x^c_\beta\, y^c_\alpha; y^c_\alpha \leftarrow u^c_\delta) \mathcal{C}^{\kappa_2}_1 \mathcal{T}^\kappa_1 \mathcal{F}'_1; \mathcal{C}^{\kappa'_3}_2 \mathcal{B}^{\kappa_3}_2, \mathcal{A}^{\kappa''}_2, \mathcal{C}^{\kappa'_4}_2 \mathcal{B}^{\kappa_4}_2 \mathcal{H}_2 \mathbf{msg}(\mathbf{send}\, x^c_\beta\, y^c_\alpha; y^c_\alpha \leftarrow u^c_\delta) \mathcal{C}^{\kappa'_2}_2 \mathcal{T}^{\kappa'}_2 \mathcal{F}'_2) \in \mathcal{E}^\xi_\Psi [\![\Delta^3_\kappa \Vdash u_\delta{:}B[c]]\!]^m$.

By a closer look at the structure of the configurations $\mathcal{D}_i$, we observe that both $\mathcal{D}_i$ and its subtree $\mathcal{A}^\kappa_i$ provide the same channel. This means that the top configurations for $\mathcal{A}^\kappa_1$ and $\mathcal{A}^{\kappa'}_2$ are empty and can be ignored, i.e. $\mathcal{C}^{\kappa_2}_1 \mathcal{T}^\kappa_1 = \cdot$ and $\mathcal{C}^{\kappa'_2}_2 \mathcal{T}^{\kappa'}_2 = \cdot$.

By the assumptions above and the typing rules we have $\Delta = \Delta', \Delta''$ and

$$\mathcal{D}_1 = \mathcal{B}^{\kappa_3}_1 \mathcal{A}^{\kappa'}_1 \mathcal{B}^{\kappa_4}_1 \mathcal{H}_1 \mathbf{msg}(\mathbf{send}\, x^c_\beta\, y^c_\alpha; y^c_\alpha \leftarrow u^c_\delta) \text{ for } \mathcal{B}^{\kappa_4}_1 \mathcal{H}_1 \in \mathsf{Tree}_\Psi(\Delta'' \Vdash x_\beta{:}A[c]) \text{ and}$$

$$\mathcal{D}'_2 = \mathcal{B}^{\kappa_3}_2 \mathcal{A}^{\kappa''}_2 \mathcal{B}^{\kappa_4}_2 \mathcal{H}_2 \mathbf{msg}(\mathbf{send}\, x^c_\beta\, y^c_\alpha; y^c_\alpha \leftarrow u^c_\delta) \text{ for } \mathcal{B}^{\kappa_4}_2 \mathcal{H}_2 \in \mathsf{Tree}_\Psi(\Delta'' \Vdash x_\beta{:}A[c]).$$

It is important to observe that the configuration $\mathcal{B}^{\kappa_4}_i$ is a forest that provides the resources for the subtree $\mathcal{H}_1$ and is completely disjoint from $\mathcal{B}^{\kappa_3}_i$. In fact, we break down the original tree $\mathcal{A}^\kappa_1$ (similarly $\mathcal{A}^{\kappa'}_2$) in two subtrees $\mathcal{H}_1$ and $\mathcal{A}^{\kappa'}_1$ (similarly $\mathcal{H}_2$ and $\mathcal{A}^{\kappa''}_2$) and a message. The message is handed over to the upper substitution, while for the corresponding

subtrees we have the required relations as in $\circ_1$ and $\circ_2$ which hold for every $m$. Now we can apply the induction hypothesis. First, we apply the induction hypothesis on every $\dagger_{\ell''}$ corresponding to the trees in $\mathcal{B}_1^{\kappa_3}$ and $\mathcal{B}_2^{\kappa_3}$, and $\circ_2$ that states the relation for the trees $\mathcal{A}_1^{\kappa'}$ and $\mathcal{A}_2^{\kappa''}$:

$$(\mathcal{C}_1^{\kappa_3}, \mathcal{B}_1^{\kappa_3} \mathcal{A}_1^{\kappa'}, \mathcal{C}_1^{\kappa_4} \mathcal{B}_1^{\kappa_4} \mathcal{H}_1 \mathbf{msg}(\mathbf{send}\, x_\beta^c\, y_\alpha^c; y_\alpha^c \leftarrow u_\delta^c) \mathcal{F}_1; \mathcal{C}_2^{\kappa_3'}, \mathcal{B}_2^{\kappa_3} \mathcal{A}_2^{\kappa''}, \mathcal{C}_2^{\kappa_4'} \mathcal{B}_2^{\kappa_4} \mathcal{H}_2 \mathbf{msg}(\mathbf{send}\, x_\beta^c\, y_\alpha^c; y_\alpha^c \leftarrow u_\delta^c) \mathcal{F}_2') \in \mathcal{E}_\Psi^\xi[\![\Delta' \Vdash u_\delta{:}B[c]]\!]^m.$$

Next, we apply the induction hypothesis on every $\dagger_{\ell''}$ corresponding to the trees in $\mathcal{B}_1^{\kappa_4}$ and $\mathcal{B}_2^{\kappa_4}$, and $\circ_1$ that states the relation for the trees $\mathcal{H}_1$ and $\mathcal{H}_2$:

$$(\mathcal{C}_1^{\kappa_4}, \mathcal{B}_1^{\kappa_4} \mathcal{H}_1, \mathcal{C}_1^{\kappa_3} \mathcal{B}_1^{\kappa_3} \mathcal{A}_1^{\kappa'} \mathbf{msg}(\mathbf{send}\, x_\beta^c\, y_\alpha^c; y_\alpha^c \leftarrow u_\delta^c) \mathcal{F}_1; \mathcal{C}_2^{\kappa_4'}, \mathcal{B}_2^{\kappa_4} \mathcal{H}_2, \mathcal{C}_2^{\kappa_3'} \mathcal{B}_2^{\kappa_3} \mathcal{A}_2^{\kappa''} \mathbf{msg}(\mathbf{send}\, x_\beta^c\, y_\alpha^c; y_\alpha^c \leftarrow u_\delta^c) \mathcal{F}_2') \in \mathcal{E}_\Psi^\xi[\![\Delta'' \Vdash x_\beta{:}A[c]]\!]^m$$

And thus we can satisfy the conditions of the fourth row in the logical relation and prove the goal

$$(\mathcal{C}_1', \mathcal{D}_1, \mathcal{F}_1'; \mathcal{C}_2', \mathcal{D}_2', \mathcal{F}_2') \in \mathcal{V}_\Psi^\xi[\![\Delta \Vdash K]\!]^{m+1}_{\cdot; y_\alpha^c}$$

**Row 5.** $K = y_\alpha{:}A \multimap B[c]$, and $(\mathcal{C}_1^{\kappa_1'} \mathcal{B}_1^\kappa, \mathcal{A}_1^\kappa, \mathcal{C}_1^{\kappa_2'} \mathcal{T}_1^\kappa \mathcal{F}_1'; \mathcal{C}_2^{\kappa_1'} \mathcal{B}_2^{\kappa'}, \mathcal{A}_2^{\kappa'}, \mathcal{C}_2^{\kappa_2'} \mathcal{T}_2^{\kappa'} \mathcal{F}_2') \in \mathsf{Tree}_\Psi(\Delta'_\kappa \Vdash y_\alpha{:}A \multimap B[c])$.

By a closer look at the structure of the configurations $\mathcal{D}_i$, we observe that both $\mathcal{D}_i$ and its subtree $\mathcal{A}_i^\kappa$ provide the same channel. This means that the top configurations for $\mathcal{A}_1^{\kappa'}$ and $\mathcal{A}_2^{\kappa''}$ are empty and can be ignored, i.e. $\mathcal{T}_1^\kappa = \cdot$ and $\mathcal{T}_2^{\kappa'} = \cdot$, and for the same reason $\mathcal{C}_1^{\kappa_2'} = \cdot$ and $\mathcal{C}_2^{\kappa_2'} = \cdot$. And we have

$$\mathcal{F}_1' = \mathcal{H}_1 \mathbf{msg}(\mathbf{send}\, x_\beta^c\, y_\alpha^c; y_\alpha^c \leftarrow u_\delta^c) \mathcal{F}_1'' \text{ for } \mathcal{H}_1 \in \mathsf{Tree}_\Psi(\cdot \Vdash x_\beta{:}A[c]).$$

Assume that $\mathcal{F}_2' = \mathcal{H}_2 \mathbf{msg}(\mathbf{send}\, x_\beta^c\, y_\alpha^c; y_\alpha^c \leftarrow u_\delta^c) \mathcal{F}_2''$ for $\mathcal{H}_2 \in \mathsf{Tree}_\Psi(\cdot \Vdash x_\beta{:}A[c])$..
By Row 5 of the logical relation we get

$$\circ_1 \forall m.\, (\mathcal{C}_1^{\kappa''} \mathcal{B}_1^\kappa \mathcal{H}_1, \mathcal{A}_1^\kappa \mathbf{msg}(\mathbf{send}\, x_\beta^c\, y_\alpha^c; y_\alpha^c \leftarrow u_\delta^c), \mathcal{F}_1''; \mathcal{C}_2^{\kappa''} \mathcal{B}_2^{\kappa'} \mathcal{H}_2, \mathcal{A}_2^{\kappa'} \mathbf{msg}(\mathbf{send}\, x_\beta^c\, y_\alpha^c; y_\alpha^c \leftarrow u_\delta^c), \mathcal{F}_2'') \in \mathcal{E}_\Psi^\xi[\![\Delta'_\kappa, x_\beta{:}A[c] \Vdash u_\delta{:}B[c]]\!]^m$$

For every $\ell \neq \kappa$, we know that $\mathcal{T}_1^\ell = \mathcal{T}_1^{\ell'} \mathcal{A}_1^\kappa$ and $\mathcal{T}_2^\ell = \mathcal{T}_2^{\ell''} \mathcal{A}_2^{\kappa'}$. We can rewrite each $\dagger_{\ell''}$ for $\ell \neq \kappa$ as

$$\circ_\ell \forall m.\, (\mathcal{C}_1^{\ell'_1} \mathcal{B}_1^\ell \mathcal{A}_1^\ell \mathcal{C}_1^{\ell'_2} \mathcal{H}_1 \mathcal{T}_1^{\ell'} \mathcal{A}_1^\kappa \mathbf{msg}(\mathbf{send}\, x_\beta^c\, y_\alpha^c; y_\alpha^c \leftarrow u_\delta^c) \mathcal{F}_1''; \mathcal{C}_1^{\ell'_1} \mathcal{B}_2^{\ell'} \mathcal{A}_2^{\ell'} \mathcal{C}_2^{\ell'_2} \mathcal{H}_2 \mathcal{T}_2^{\ell''} \mathcal{A}_2^{\kappa'} \mathbf{msg}(\mathbf{send}\, x_\beta^c\, y_\alpha^c; y_\alpha^c \leftarrow u_\delta^c) \mathcal{F}_2'') \in \mathcal{E}_{\Psi_0}^\xi[\![\Delta'_\ell \Vdash K'_\ell]\!]^m$$

In the rewriting above, we used the fact that $\mathcal{F}_i' = \mathcal{H}_i \mathbf{msg}(\mathbf{send}\, x_\beta^c\, y_\alpha^c; y_\alpha^c \leftarrow u_\delta^c) \mathcal{F}_i''$. We use a simple rewrite of the configuration knowing that $\mathcal{H}_i$ does not use any resources to clarify how we handle the move of $\mathcal{H}_i$ to the lower substitution (e.g. along with $\mathcal{C}_2^{\ell'}$, it will be considered a lower substitution for the sub-tree $A_2^{\kappa'}$).

We apply the induction hypothesis on every $\circ_\ell$ for every $\ell \neq \kappa$ and $\circ_1$ (that states the relation for the trees $\mathcal{A}_1^\kappa \mathbf{msg}(\mathbf{send}\, x_\beta^c\, y_\alpha^c; y_\alpha^c \leftarrow u_\delta^c)$ and $\mathcal{A}_2^{\kappa'} \mathbf{msg}(\mathbf{send}\, x_\beta^c\, y_\alpha^c; y_\alpha^c \leftarrow u_\delta^c))$ to get:

$$(\mathcal{C}_1' \mathcal{H}_1, \mathcal{D}_1 \mathbf{msg}(\mathbf{send}\, x_\beta^c\, y_\alpha^c; y_\alpha^c \leftarrow u_\delta^c), \mathcal{F}_1''; \mathcal{C}_2' \mathcal{H}_2, \mathcal{D}_2' \mathbf{msg}(\mathbf{send}\, x_\beta^c\, y_\alpha^c; y_\alpha^c \leftarrow u_\delta^c), \mathcal{F}_2'') \in \mathcal{E}_\Psi^\xi[\![\Delta, x_\beta{:}A[c] \Vdash u_\delta{:}B[c]]\!]^m$$

And thus we can satisfy the conditions of the fifth row in the logical relation and prove the goal

$$(\mathcal{C}_1', \mathcal{D}_1, \mathcal{F}_1'; \mathcal{C}_2', \mathcal{D}_2', \mathcal{F}_2') \in \mathcal{V}_\Psi^\xi[\![\Delta \Vdash K]\!]^{m+1}_{\cdot; y_{\alpha\,\mathcal{F}}^c}$$

**Row 10.** $\Delta'_\kappa = \Delta_\kappa^3, \Delta_\kappa^4, y_\alpha{:}A \multimap B[c]$. And $\mathcal{C}_1^{\kappa_1'} = \mathcal{C}_1^{\kappa_3} \mathcal{C}_1^{\kappa_4}$ and $\mathcal{B}_1^\kappa = \mathcal{B}_1^{\kappa_3} \mathcal{B}_1^{\kappa_4}$ such that $\mathcal{C}_1^{\kappa_3} \mathcal{B}_1^{\kappa_3} \in \mathsf{Tree}_\Psi(\cdot \Vdash \Delta_\kappa^3, y_\alpha{:}A \multimap B[c])$ and $\mathcal{C}_1^{\kappa_4} \mathcal{B}_1^{\kappa_4} \in \mathsf{Tree}_\Psi(\cdot \Vdash \Delta_\kappa^4)$. Moreover, $\mathcal{C}_2^{\kappa_1'} = \mathcal{C}_2^{\kappa_3'} \mathcal{C}_2^{\kappa_4'}$ and $\mathcal{B}_2^\kappa = \mathcal{B}_2^{\kappa_3} \mathcal{B}_2^{\kappa_4}$ such that $\mathcal{C}_2^{\kappa_3'} \mathcal{B}_2^{\kappa_3} \in \mathsf{Tree}_\Psi(\cdot \Vdash \Delta_\kappa^3, y_\alpha{:}A \multimap B[c])$ and $\mathcal{C}_2^{\kappa_4'} \mathcal{B}_2^{\kappa_4} \in \mathsf{Tree}_\Psi(\cdot \Vdash \Delta_\kappa^4)$. And $(\mathcal{C}_1^{\kappa_3} \mathcal{C}_1^{\kappa_4} \mathcal{B}_1^{\kappa_3} \mathcal{B}_1^{\kappa_4}, \mathcal{A}_1^\kappa, \mathcal{C}_1^{\kappa_2'} \mathcal{T}_1^\kappa \mathcal{F}_1'; \mathcal{C}_2^{\kappa_3'} \mathcal{C}_2^{\kappa_4'} \mathcal{B}_2^{\kappa_3} \mathcal{B}_2^{\kappa_4}, \mathcal{A}_2^{\kappa'}, \mathcal{C}_2^{\kappa_2'} \mathcal{T}_2^{\kappa'} \mathcal{F}_2') \in \mathsf{Tree}_\Psi(\Delta_\kappa^3, \Delta_\kappa^4, y_\alpha{:}A \multimap B[c] \Vdash K)$.
Moreover, we have

$$\mathcal{A}_1^\kappa = \mathcal{H}_1 \mathbf{msg}(\mathbf{send}\, x_\beta^c\, y_\alpha^c; y_\alpha^c \leftarrow u_\delta^c) \mathcal{A}_1^{\kappa'} \text{ for } \mathcal{H}_1 \in \mathsf{Tree}_\Psi(\Delta_\kappa^4 \Vdash x_\beta{:}A[c]) \text{ and}$$

$$\mathcal{A}_2^{\kappa'} = \mathcal{H}_2 \mathbf{msg}(\mathbf{send}\, x_\beta^c\, y_\alpha^c; y_\alpha^c \leftarrow u_\delta^c) \mathcal{A}_2^{\kappa''} \text{ for } \mathcal{H}_2 \in \mathsf{Tree}_\Psi(\Delta_\kappa^4 \Vdash x_\beta{:}A[c]) \text{ and}$$

$\circ_1 \forall m.\, (\mathcal{C}_1^{\kappa_4} \mathcal{B}_1^{\kappa_4}, \mathcal{H}_1, \mathcal{C}_1^{\kappa_3} \mathcal{B}_1^{\kappa_3} \mathcal{A}_1^{\kappa'} \mathbf{msg}(\mathbf{send}\, x_\beta^c\, y_\alpha^c; y_\alpha^c \leftarrow u_\delta^c) \mathcal{C}_1^{\kappa_2'} \mathcal{T}_1^\kappa \mathcal{F}_1; \mathcal{C}_2^{\kappa_4'} \mathcal{B}_2^{\kappa_4}, \mathcal{H}_2, \mathcal{C}_2^{\kappa_3'} \mathcal{B}_2^{\kappa_3} \mathcal{A}_2^{\kappa''} \mathbf{msg}(\mathbf{send}\, x_\beta^c\, y_\alpha^c; y_\alpha^c \leftarrow u_\delta^c) \mathcal{C}_2^{\kappa_2'} \mathcal{T}_2^{\kappa'} \mathcal{F}_2) \in \mathcal{E}_\Psi^\xi[\![\Delta_\kappa^4 \Vdash x_\beta{:}A[c]]\!]^m$

and

$\circ_2 \forall m.\, (\mathcal{C}_1^{\kappa_3} \mathcal{B}_1^{\kappa_3} \mathcal{C}_1^{\kappa_4} \mathcal{B}_1^{\kappa_4} \mathcal{H}_1 \mathbf{msg}(\mathbf{send}\, x_\beta^c\, y_\alpha^c; y_\alpha^c \leftarrow u_\delta^c), \mathcal{A}_1^{\kappa'}, \mathcal{C}_1^{\kappa_2'} \mathcal{T}_1^\kappa \mathcal{F}_1; \mathcal{C}_2^{\kappa_3'} \mathcal{B}_2^{\kappa_3} \mathcal{C}_2^{\kappa_4'} \mathcal{B}_2^{\kappa_4} \mathcal{H}_2 \mathbf{msg}(\mathbf{send}\, x_\beta^c\, y_\alpha^c; y_\alpha^c \leftarrow u_\delta^c), \mathcal{A}_2^{\kappa''}, \mathcal{C}_2^{\kappa_2'} \mathcal{T}_2^{\kappa'} \mathcal{F}_2') \in \mathcal{E}_\Psi^\xi[\![\Delta_\kappa^3, u_\delta{:}B[c] \Vdash K]\!]^m.$

By the typing rules we have $\Delta = \Delta', \Delta'', y_\alpha{:}A \multimap B[c]$ and

$$\mathcal{D}_1 = \mathcal{B}_1^{\kappa_4} \mathcal{H}_1 \mathbf{msg}(\mathbf{send}\, x_\beta^c\, y_\alpha^c; y_\alpha^c \leftarrow u_\delta^c) \mathcal{B}_1^{\kappa_3} \mathcal{A}_1^{\kappa'} \mathcal{T}_1^\kappa \text{ for } \mathcal{B}_1^{\kappa_4} \mathcal{H}_1 \in \mathsf{Tree}_\Psi(\Delta'' \Vdash x_\beta{:}A[c]) \text{ and}$$

$$\mathcal{D}_2' = \mathcal{B}_2^{\kappa_4} \mathcal{H}_2 \mathbf{msg}(\mathbf{send}\, x_\beta^c\, y_\alpha^c; y_\alpha^c \leftarrow u_\delta^c) \mathcal{B}_2^{\kappa_3} \mathcal{A}_2^{\kappa''} \mathcal{T}_2^{\kappa'} \text{ for } \mathcal{B}_2^{\kappa_4} \mathcal{H}_2 \in \mathsf{Tree}_\Psi(\Delta'' \Vdash x_\beta{:}A[c]).$$

It is important to observe that the configuration $\mathcal{B}_i^{\kappa_4}$ is a forest that provides the resources for the subtree $\mathcal{H}_1$ and is completely disjoint from $\mathcal{B}_i^{\kappa_3}$. In fact, we break down the original tree $\mathcal{A}_1^\kappa$ (similarly $\mathcal{A}_2^{\kappa'}$) in two subtrees $\mathcal{H}_1$ and $\mathcal{A}_1^{\kappa'}$ (similarly $\mathcal{H}_2$ and $\mathcal{A}_2^{\kappa''}$) and a message. The message is handed over to the lower substitution, while for the corresponding subtrees we have the required relations as in $\circ_1$ and $\circ_2$ which hold for every $m$. Now we can apply the induction hypothesis. First, we apply the induction hypothesis on every $\dagger_{\ell''}$ corresponding to the trees in $\mathcal{B}_1^{\kappa_3}$ and $\mathcal{B}_2^{\kappa_3}$, and $\circ_2$ that states the relation for the trees $\mathcal{A}_1^{\kappa'}$ and $\mathcal{A}_2^{\kappa''}$:

$$(\mathcal{C}_1^{\kappa_3} \mathcal{C}_1^{\kappa_4} \mathcal{B}_1^{\kappa_4} \mathcal{H}_1 \mathbf{msg}(\mathbf{send}\, x_\beta^c\, y_\alpha^c; y_\alpha^c \leftarrow u_\delta^c), \mathcal{B}_1^{\kappa_3} \mathcal{A}_1^{\kappa'}, \mathcal{F}_1; \mathcal{C}_2^{\kappa_3'} \mathcal{C}_2^{\kappa_4'} \mathcal{B}_2^{\kappa_4} \mathcal{H}_2 \mathbf{msg}(\mathbf{send}\, x_\beta^c\, y_\alpha^c; y_\alpha^c \leftarrow u_\delta^c), \mathcal{B}_2^{\kappa_3} \mathcal{A}_2^{\kappa''}, \mathcal{F}_2') \in \mathcal{E}_\Psi^\xi [\![ \Delta', u_\delta{:}B[c] \Vdash K ]\!]^m.$$

Next, we apply the induction hypothesis on every $\dagger_{\ell''}$ corresponding to the trees in $\mathcal{B}_1^{\kappa_4}$ and $\mathcal{B}_2^{\kappa_4}$, and $\circ_1$ that states the relation for the trees $\mathcal{H}_1$ and $\mathcal{H}_2$:

$$(\mathcal{C}_1^{\kappa_4}, \mathcal{B}_1^{\kappa_4} \mathcal{H}_1, \mathcal{C}_1^{\kappa_3} \mathcal{B}_1^{\kappa_3} \mathcal{A}_1^{\kappa'} \mathbf{msg}(\mathbf{send}\, x_\beta^c\, y_\alpha^c; y_\alpha^c \leftarrow u_\delta^c) \mathcal{F}_1; \mathcal{C}_2^{\kappa_4'}, \mathcal{B}_2^{\kappa_4} \mathcal{H}_2, \mathcal{C}_2^{\kappa_3'} \mathcal{B}_2^{\kappa_3} \mathcal{A}_2^{\kappa''} \mathbf{msg}(\mathbf{send}\, x_\beta^c\, y_\alpha^c; y_\alpha^c \leftarrow u_\delta^c) \mathcal{F}_2') \in \mathcal{E}_\Psi^\xi [\![ \Delta'' \Vdash x_\beta{:}A[c] ]\!]^m$$

And thus we can satisfy the conditions of the fourth row in the logical relation and prove the goal

$$(\mathcal{C}_1', \mathcal{D}_1, \mathcal{F}_1'; \mathcal{C}_2', \mathcal{D}_2', \mathcal{F}_2') \in \mathcal{V}_\Psi^\xi [\![ \Delta \Vdash K ]\!]^{m+1}_{\cdot; y_\alpha^c}$$

**Row 11.** $K = y_\alpha{:}A[c]$, and $(\mathcal{C}_1^{\kappa_1'} \mathcal{B}_1^\kappa, \mathcal{A}_1^\kappa, \mathcal{C}_1^{\kappa_2'} \mathcal{T}_1^\kappa \mathcal{F}_1'; \mathcal{C}_2^{\kappa_1'} \mathcal{B}_2^{\kappa'}, \mathcal{A}_2^{\kappa'}, \mathcal{C}_2^{\kappa_2'} \mathcal{T}_2^{\kappa'} \mathcal{F}_2') \in \mathsf{Tree}_\Psi(\Delta_\kappa' \Vdash y_\alpha{:}A[c])$.

By a closer look at the structure of the configurations $\mathcal{D}_i$, we observe that both $\mathcal{D}_i$ and its subtree $\mathcal{A}_i^\kappa$ provide the same channel. This means that the top configurations for $\mathcal{A}_1^{\kappa'}$ and $\mathcal{A}_2^{\kappa''}$ are empty and can be ignored, i.e. $\mathcal{T}_1^\kappa = \cdot$ and $\mathcal{T}_2^{\kappa'} = \cdot$, and for the same reason $\mathcal{C}_1^{\kappa_2'} = \cdot$ and $\mathcal{C}_2^{\kappa_2'} = \cdot$.

We have

$$\mathcal{A}_1^\kappa = \mathcal{A}_1^{\kappa'} \mathbf{proc}(y_\alpha^c, y_\alpha^c \leftarrow u_\delta^c @ d_1) \text{ and } \mathcal{A}_2^{\kappa'} = \mathcal{A}_2^{\kappa''} \mathbf{proc}(y_\alpha^c, y_\alpha^c \leftarrow u_\delta^c @ d_2) \text{ and}$$

$$\circ_1 \forall m.\, ([u_\delta^c/y_\alpha^c] \mathcal{C}_1^{\kappa'} \mathcal{B}_1^\kappa, \mathcal{A}_1^{\kappa'}, \mathcal{F}_1'; [u_\delta^c/y_\alpha^c] \mathcal{C}_2^{\kappa'} \mathcal{B}_2^{\kappa'}, \mathcal{A}_2^{\kappa''}, \mathcal{F}_2') \in \mathcal{E}_\Psi^\xi [\![ \Delta_\kappa' \Vdash u_\delta{:}A[c] ]\!]^m$$

By the assumptions above and the typing rules we have

$$\mathcal{D}_1 = \mathcal{B}_1^\kappa \mathcal{A}_1^{\kappa'} \mathbf{proc}(y_\alpha^c, y_\alpha^c \leftarrow u_\delta^c @ d_1) \text{ and}$$

$$\mathcal{D}_2' = \mathcal{B}_2^{\kappa'} \mathcal{A}_2^{\kappa''} \mathbf{proc}(y_\alpha^c, y_\alpha^c \leftarrow u_\delta^c @ d_2).$$

For every $\ell \neq \kappa$, we know that $\mathcal{T}_1^\ell = \mathcal{T}_1^{\ell'} \mathcal{A}_1^{\kappa'} \mathbf{proc}(y_\alpha^c, y_\alpha^c \leftarrow u_\delta^c @ d_1)$ and $\mathcal{T}_2^{\ell'} = \mathcal{T}_2^{\ell''} \mathcal{A}_2^{\kappa''} \mathbf{proc}(y_\alpha^c, y_\alpha^c \leftarrow u_\delta^c @ d_2)$. We apply forward closure (Lemma C.11) on $\dagger_{\ell''}$ and the fact that $y_\alpha^c$ is a linear resource of $\mathcal{F}_i'$ to get

$$\circ_\ell \forall m.\, ([u_\delta^c/y_\alpha] \mathcal{C}_1^{\ell_1'} \mathcal{B}_1^\ell \mathcal{A}_1^\ell \mathcal{C}_1^{\ell_2'} \mathcal{T}_1^{\ell'} \mathcal{A}_1^{\kappa'} \mathcal{F}_1'; [u_\delta^c/y_\alpha] \mathcal{C}_2^{\ell_1'} \mathcal{B}_2^{\ell'} \mathcal{A}_2^{\ell'} \mathcal{C}_2^{\ell_2'} \mathcal{T}_2^{\ell''} \mathcal{A}_2^{\kappa''} \mathcal{F}_2') \in \mathcal{E}_{\Psi_0}^\xi [\![ \Delta_\ell' \Vdash K_\ell' ]\!]^m$$

We apply the induction hypothesis on every $\circ_\ell$ for every $\ell \neq \kappa$ and $\circ_1$ (that states the relation for the trees $\mathcal{A}_1^{\kappa'}$ and $\mathcal{A}_2^{\kappa''}$) to get:

$$([u_\delta^c/y_\alpha] \mathcal{C}_1', \mathcal{B}_1^\kappa \mathcal{A}_1^{\kappa'}, \mathcal{F}_1'; [u_\delta^c/y_\alpha] \mathcal{C}_2', \mathcal{B}_2^{\kappa'} \mathcal{A}_2^{\kappa''}, \mathcal{F}_2') \in \mathcal{E}_\Psi^\xi [\![ \Delta \Vdash u_\delta{:}A[c] ]\!]^m.$$

And thus we can satisfy the conditions of the 11th row in the logical relation and prove the goal

$$(\mathcal{C}_1', \mathcal{D}_1, \mathcal{F}_1'; \mathcal{C}_2', \mathcal{D}_2', \mathcal{F}_2') \in \mathcal{V}_\Psi^\xi [\![ \Delta \Vdash K ]\!]^{m+1}_{\cdot; y_\alpha^c}$$

**Row 12.** $\Delta = \Delta', y_\alpha{:}A[c]$, and $\Delta_\kappa' = \Delta_\kappa'', y_\alpha{:}A[c]$ and $(\mathcal{C}_1^{\kappa_1'} \mathcal{B}_1^\kappa, \mathcal{A}_1^\kappa, \mathcal{C}_1^{\kappa_2'} \mathcal{T}_1^\kappa \mathcal{F}_1'; \mathcal{C}_2^{\kappa_1'} \mathcal{B}_2^{\kappa'}, \mathcal{A}_2^{\kappa'}, \mathcal{C}_2^{\kappa_2'} \mathcal{T}_2^{\kappa'} \mathcal{F}_2') \in \mathsf{Tree}_\Psi(\Delta_\kappa'', y_\alpha{:}A[c] \Vdash K)$.

And we have

$$\mathcal{C}_1^{\kappa_1'} = \mathcal{C}_1^{\kappa_1''} \mathbf{proc}(y_\alpha^c, y_\alpha^c \leftarrow u_\delta^c @ d_1)$$

By a closer look at the structure of the configurations $\mathcal{D}_i$, we observe that both $\mathcal{D}_i$ and its subtree $\mathcal{A}_i^\kappa$ use the same channel as a resource. By the typing rules we have

$$\mathcal{C}_1' = \mathcal{C}_1^{\kappa_1''} \mathbf{proc}(y_\alpha^c, y_\alpha^c \leftarrow u_\delta^c @ d_1) \mathcal{C}_1^{\kappa_2'}.$$

Assume that $\mathcal{C}_2' = \mathcal{C}_2'' \mathbf{proc}(y_\alpha^c, y_\alpha^c \leftarrow u_\delta^c @ d_2)$, then by the typing rules we know that

$$\mathcal{C}_2^{\kappa_1'} = \mathcal{C}_2^{\kappa_1''} \mathbf{proc}(y_\alpha^c, y_\alpha^c \leftarrow u_\delta^c @ d_2).$$

By Row 12 of the logical relation we get

$$\circ_1 \, \forall m.\, ([u_\delta^c/y_\alpha^c]\mathcal{C}_1^{\kappa_1''}\mathcal{B}_1^\kappa, \mathcal{A}_1^\kappa, \mathcal{C}_1^{\kappa_2'}\mathcal{T}_1^\kappa \mathcal{F}_1; [u_\delta^c/y_\alpha^c]\mathcal{C}_2^{\kappa''}\mathcal{B}_2^{\kappa'}, \mathcal{A}_2^{\kappa'}, \mathcal{C}_2^{\kappa_2'}\mathcal{T}_2^{\kappa'}\mathcal{F}_2) \in \mathcal{E}_\Psi^\xi [\![\Delta_\kappa'', u_\delta{:}A[c] \Vdash K]\!]^m$$

For every $\ell \neq \kappa$, we know that $\mathcal{C}_i^{\ell_1'} = \mathcal{C}_i^{\ell_1''} \mathbf{proc}(y_\alpha^c, y_\alpha^c \leftarrow u_\delta^c @ d_i)$ or $\mathcal{C}_i^{\ell_2'} = \mathcal{C}_i^{\ell_2''} \mathbf{proc}(y_\alpha^c, y_\alpha^c \leftarrow u_\delta^c @ d_i)$. We apply forward closure (Lemma C.11) on $\dagger_{\ell''}$ and the fact that $y_\alpha^c$ is a linear resource of $\mathcal{A}_i^\kappa$ to get

$$\circ_\ell \, \forall m.\, ([u_\delta^c/y_\alpha]\mathcal{C}_1^{\ell_1''}\mathcal{B}_1^\ell \mathcal{A}_1^\ell \mathcal{C}_1^{\ell_2'}\mathcal{T}_1^\ell \mathcal{F}_1; [u_\delta^c/y_\alpha]\mathcal{C}_2^{\ell_1''}\mathcal{B}_2^{\ell'} \mathcal{A}_2^{\ell'} \mathcal{C}_2^{\ell_2'}\mathcal{T}_2^{\ell'}\mathcal{F}_2') \in \mathcal{E}_{\Psi_0}^\xi [\![\Delta_\ell' \Vdash K_\ell']\!]^m$$

in the cases with $\mathcal{C}_i^{\ell_1'} = \mathcal{C}_i^{\ell_1''} \mathbf{proc}(y_\alpha^c, y_\alpha^c \leftarrow u_\delta^c @ d_i)$ and

$$\circ_\ell' \, \forall m.\, ([u_\delta^c/y_\alpha]\mathcal{C}_1^{\ell_1'}\mathcal{B}_1^\ell \mathcal{A}_1^\ell \mathcal{C}_1^{\ell_2''}\mathcal{T}_1^\ell \mathcal{F}_1; [u_\delta^c/y_\alpha]\mathcal{C}_2^{\ell_1'}\mathcal{B}_2^{\ell'} \mathcal{A}_2^{\ell'} \mathcal{C}_2^{\ell_2''}\mathcal{T}_2^{\ell'}\mathcal{F}_2') \in \mathcal{E}_{\Psi_0}^\xi [\![\Delta_\ell' \Vdash K_\ell']\!]^m$$

in the cases with $\mathcal{C}_i^{\ell_2'} = \mathcal{C}_i^{\ell_2''} \mathbf{proc}(y_\alpha^c, y_\alpha^c \leftarrow u_\delta^c @ d_i)$.

We apply the induction hypothesis on every $\circ_\ell$ and $\circ_\ell'$ for every $\ell \neq \kappa$ and $\circ_1$ (that states the relation for the trees $\mathcal{A}_1^\kappa$ and $\mathcal{A}_2^{\kappa'}$) to get:

$$([u_\delta^c/y_\alpha]\mathcal{C}_1', \mathcal{D}_1, \mathcal{F}_1'; [u_\delta^c/y_\alpha]\mathcal{C}_2', \mathcal{D}_2', \mathcal{F}_2') \in \mathcal{E}_\Psi^\xi [\![\Delta, u_\delta{:}A[c] \Vdash K]\!]^m.$$

And thus we can satisfy the conditions of the 12th row in the logical relation and prove the goal

$$(\mathcal{C}_1', \mathcal{D}_1, \mathcal{F}_1'; \mathcal{C}_2', \mathcal{D}_2', \mathcal{F}_2') \in \mathcal{V}_\Psi^\xi [\![\Delta \Vdash K]\!]^{m+1}_{\cdot; y_\alpha^c{}_C}$$

**Case 2.** $j = j' + 1$, and the first step is an internal step but not a communication between the sub-trees.
Without loss of generality, lets assume that the communication is internal to $\mathcal{A}_1^\ell$, i.e. we have

$$\mathcal{C}_1^{\ell_1}, \mathcal{C}_1^{\ell_2}, \mathcal{B}_1^\ell \mathcal{A}_1^\ell \mathcal{T}_1^\ell, \mathcal{F}_1 \mapsto \mathcal{C}_1^{\ell_1}, \mathcal{C}_1^{\ell_2}, \mathcal{B}_1^\ell \mathcal{A}_1^{\ell'} \mathcal{T}_1^\ell, \mathcal{F}_1 \mapsto^{j'\Theta_1;\Upsilon_1}_{\Delta,\Delta' \Vdash K} \mathcal{C}_1^{\ell_1}, \mathcal{C}_1^{\ell_2}, \mathcal{D}_1', \mathcal{F}_1.$$

Observe that for all $\kappa \neq \ell$, the tree $\mathcal{A}_1^\kappa$ remains intact, but the step above occurs in either the configuration $\mathcal{B}_1^\kappa$ or $\mathcal{T}_1^\kappa$ resulting in $\mathcal{B}_1^{\kappa'}$ or $\mathcal{T}_1^{\kappa'}$.
By forward closure (Lemma C.12) on the assumption $\dagger_1^\ell$, we get

$$\dagger_1^{\ell'} \, \forall m.\, (\mathcal{C}_1^{\ell_1}\mathcal{B}_1^\ell \mathcal{A}_1^{\ell'} \mathcal{C}_1^{\ell_2}\mathcal{T}_1^\ell \mathcal{F}_1; \mathcal{C}_2^{\ell_1}\mathcal{B}_2^\ell \mathcal{A}_2^\ell \mathcal{C}_2^{\ell_2}\mathcal{T}_2^\ell \mathcal{F}_2) \in \mathcal{E}_{\Psi_0}^\xi [\![\Delta_\ell' \Vdash K_\ell']\!]^m,$$

Similarly, for alll $\kappa \neq \ell$, we use forward closure (Lemma C.11) on the assumption $\dagger_1^\kappa$ to get

$$\dagger_1^{\kappa'} \, \forall m.\, (\mathcal{C}_1^{\kappa_1}\mathcal{B}_1^{\kappa'} \mathcal{A}_1^\kappa \mathcal{C}_1^{\kappa_2}\mathcal{T}_1^{\kappa'}\mathcal{F}_1; \mathcal{C}_2^{\kappa_1}\mathcal{B}_2^\kappa \mathcal{A}_2^\kappa \mathcal{C}_2^{\kappa_2}\mathcal{T}_2^\kappa \mathcal{F}_2) \in \mathcal{E}_{\Psi_0}^\xi [\![\Delta_\kappa' \Vdash K_\kappa']\!]^m,$$

Now we can apply the induction hypothesis on the number of steps $j'$ to get:

* $\exists \mathcal{D}_2'. \mathcal{C}_2, \mathcal{D}_2, \mathcal{F}_2 \mapsto^{*\Theta';\Upsilon_1}_{\Delta \Vdash K} \mathcal{C}_2, \mathcal{D}_2', \mathcal{F}_2$, and $\forall \mathcal{C}_1', \mathcal{F}_1', \mathcal{C}_2', \mathcal{F}_2'.$ if $\mathcal{C}_1, \mathcal{D}_1', \mathcal{F}_1 \mapsto^{*\Theta_1;\Upsilon_1}_{\Delta \Vdash K} \mathcal{C}_1', \mathcal{D}_1', \mathcal{F}_1'$ and $\mathcal{C}_2, \mathcal{D}_2', \mathcal{F}_2 \mapsto^{*\Theta_2;\Upsilon_1}_{\Delta \Vdash K} \mathcal{C}_2', \mathcal{D}_2', \mathcal{F}_2'$, then
  $\forall x \in \Upsilon_1.\, (\mathcal{C}_1', \mathcal{D}_1', \mathcal{F}_1'; \mathcal{C}_2', \mathcal{D}_2', \mathcal{F}_2') \in \mathcal{V}_\Psi^\xi [\![\Delta \Vdash K]\!]^{m+1}_{\cdot; x}$ and
  $\forall x \in (\Theta_1 \cap \Theta_2).\, (\mathcal{C}_1', \mathcal{D}_1', \mathcal{F}_1'; \mathcal{C}_2', \mathcal{D}_2', \mathcal{F}_2') \in \mathcal{V}_\Psi^\xi [\![\Delta \Vdash K]\!]^{m+1}_{x; \cdot}$.

which completes the proof of this case.

**Case 3.** $j = j' + 1$, and the first step is a communication between two subtrees of $\mathcal{D}_1$.
Without loss of generality, assume that the communication is between $\Delta_\ell' \Vdash \mathcal{A}_1^\ell :: K_\ell'$ and $\Delta_\kappa' \Vdash \mathcal{A}_1^\kappa :: K_\kappa'$ along the channel $u_\alpha^c$ and $u_\alpha^c{:}T = K_\ell' \in \Delta_\kappa'$. (We write $\Delta_\kappa' = \Delta_\kappa'', u_\alpha^c{:}T$.) In this case we know that $\mathcal{B}_i^\kappa$ contains the tree $\mathcal{A}_i^\ell$, i.e. $\mathcal{B}_i^\kappa = \mathcal{B}_i^{\kappa'} \mathcal{A}_i^\ell$. Dually, $\mathcal{T}_i^\ell$ contains the tree $\mathcal{A}_i^\kappa$, i.e. $\mathcal{T}_i^\ell = \mathcal{A}_i^\kappa \mathcal{T}_i^{\ell'}$. In other words, $\mathcal{A}_i^\ell$ is a child of $\mathcal{A}_i^\kappa$, and there is no sub-tree connecting them. Thus for any $\iota \neq \kappa, \ell$, $\mathcal{A}_i^\ell$ and $\mathcal{A}_i^\kappa$ are both in $\mathcal{B}_i^\iota$ or they are both in $\mathcal{T}_i^\iota$.
We consider cases based on type of $u_\alpha^c$ and the message or forwarding process that is ready along $u_\alpha^c$.

29

**Subcase 1.** $u_\alpha^c{:}A$ and there is a forwarding process in $\mathcal{A}_1^\ell$ i.e. $\mathcal{A}_1^\ell = \mathcal{A}_1^{\ell'}\mathbf{proc}(u_\alpha^c, u_\alpha^c \leftarrow y_\beta^c@d_1)$, and $\mathcal{D}_1 = \mathcal{D}_1^1\mathcal{A}_1^{\ell'}\mathbf{proc}(u_\alpha^c, u_\alpha^c \leftarrow y_\beta^c@d_1)\mathcal{A}_1^\kappa\mathcal{D}_1^2$ and

$$\mathcal{C}_1, \mathcal{D}_1^1\mathcal{A}_1^{\ell'}\mathbf{proc}(u_\alpha^c, u_\alpha^c \leftarrow y_\beta^c@d_1)\mathcal{A}_1^\kappa\mathcal{D}_1^2, \mathcal{F}_1 \mapsto \mathcal{C}_1, \mathcal{D}_1^1\mathcal{A}_1^{\ell'}[y_\beta^c/u_\alpha^c]\mathcal{A}_1^\kappa\mathcal{D}_1^2, \mathcal{F}_1 \mapsto_{\Delta \Vdash K}^{j';\Theta;\Upsilon_1} \mathcal{C}_1, \mathcal{D}_1', \mathcal{F}_1.$$

In other words, $\mathcal{A}_1^\ell$ has a forwarding process that renames $u_\alpha^c$ to $y_\beta^c$ in $\mathcal{A}_1^\kappa$. We put $\mathcal{A}_1^{\kappa'} = [y_\beta^c/u_\alpha^c]\mathcal{A}_1^\kappa$. By linearity we know that the renaming does not affect any other subtree of $\mathcal{D}_1$.

By $\dagger_1^\ell$, and $\mathcal{C}_1^{\ell_1}\mathcal{B}_1^\ell, \mathcal{A}_1^\ell, \mathcal{C}_1^{\ell_2}\mathcal{T}_1^\ell\mathcal{F}_1 \mapsto_{\Delta_\ell' \Vdash K_\ell'}^{*\Theta^{\ell'};\Upsilon_1^\ell} \mathcal{C}_1^{\ell_1}\mathcal{B}_1^\ell, \mathcal{A}_1^\ell, \mathcal{C}_1^{\ell_2}\mathcal{T}_1^\ell\mathcal{F}_1$, and $\mathcal{A}_1^\ell = \mathcal{A}_1^{\ell'}\mathbf{proc}(u_\alpha^c, u_\alpha^c \leftarrow y_\beta^c@d_1)$ we know that

$\star_\ell \ \forall m. \exists \mathcal{A}_2^{\ell'}. \mathcal{C}_2^{\ell_1}\mathcal{B}_2^\ell, \mathcal{A}_2^\ell, \mathcal{C}_2^{\ell_2}\mathcal{T}_2^\ell\mathcal{F}_2 \mapsto_{\Delta_\ell' \Vdash K_\ell'}^{*\Theta^{\ell'};\Upsilon_1^\ell} \mathcal{C}_2^{\ell_1}\mathcal{B}_2^\ell, \mathcal{A}_2^{\ell'}, \mathcal{C}_2^{\ell_2}\mathcal{T}_2^\ell\mathcal{F}_2$, and
$\forall \mathcal{C}_1', \mathcal{F}_1', \mathcal{C}_2', \mathcal{F}_2'.$ if $\mathcal{C}_1^{\ell_1}\mathcal{B}_1^\ell, \mathcal{A}_1^\ell, \mathcal{C}_1^{\ell_2}\mathcal{T}_1^\ell\mathcal{F}_1 \mapsto_{\Delta_\ell' \Vdash K_\ell'}^{*\Theta_1^\ell;\Upsilon_1^\ell} \mathcal{C}_1', \mathcal{A}_1^\ell, \mathcal{F}_1'$ and $\mathcal{C}_2^{\ell_1}\mathcal{B}_2^\ell, \mathcal{A}_2^\ell, \mathcal{C}_2^{\ell_2}\mathcal{T}_2^\ell\mathcal{F}_2 \mapsto_{\Delta_\ell' \Vdash K_\ell'}^{*\Theta_2^\ell;\Upsilon_1^\ell} \mathcal{C}_2', \mathcal{A}_2^{\ell'}, \mathcal{F}_2'$, then
$\forall x \in \Upsilon_1^\ell. \ (\mathcal{C}_1', \mathcal{A}_1^{\ell'}\mathbf{proc}(u_\alpha^c, u_\alpha^c \leftarrow y_\beta^c@d_1), \mathcal{F}_1'; \mathcal{C}_2', \mathcal{A}_2^{\ell'}, \mathcal{F}_2') \in \mathcal{V}_\Psi^\xi [\![\Delta_\ell' \Vdash u_\alpha^c{:} \oplus \{\ell{:}A_\ell\}_{\ell \in L}]\!]_{:;x}^{m+1}$ and
$\forall x \in (\Theta_1^\ell \cap \Theta_2^\ell). \ (\mathcal{C}_1', \mathcal{A}_1^{\ell'}\mathbf{proc}(u_\alpha^c, u_\alpha^c \leftarrow y_\beta^c@d_1), \mathcal{F}_1'; \mathcal{C}_2', \mathcal{A}_2^{\ell'}, \mathcal{F}_2') \in \mathcal{V}_\Psi^\xi [\![\Delta_\ell' \Vdash u_\alpha^c{:} \oplus \{\ell{:}A_\ell\}_{\ell \in L}]\!]_{x;:}^{m+1}$.

By Lemma C.8, we get

$\star_\ell' \ \exists \mathcal{A}_2^{\ell'}. \mathcal{C}_2^{\ell_1}\mathcal{B}_2^\ell, \mathcal{A}_2^\ell, \mathcal{C}_2^{\ell_2}\mathcal{T}_2^\ell\mathcal{F}_2 \mapsto_{\Delta_\ell' \Vdash K_\ell'}^{*\Theta^{\ell'};\Upsilon_1^\ell} \mathcal{C}_2^{\ell_1}\mathcal{B}_2^\ell, \mathcal{A}_2^{\ell'}, \mathcal{C}_2^{\ell_2}\mathcal{T}_2^\ell\mathcal{F}_2$, and
$\forall \mathcal{C}_1', \mathcal{F}_1', \mathcal{C}_2', \mathcal{F}_2'.$ if $\mathcal{C}_1^{\ell_1}\mathcal{B}_1^\ell, \mathcal{A}_1^\ell, \mathcal{C}_1^{\ell_2}\mathcal{T}_1^\ell\mathcal{F}_1 \mapsto_{\Delta_\ell' \Vdash K_\ell'}^{*\Theta_1^\ell;\Upsilon_1^\ell} \mathcal{C}_1', \mathcal{A}_1^\ell, \mathcal{F}_1'$ and $\mathcal{C}_2^{\ell_1}\mathcal{B}_2^\ell, \mathcal{A}_2^\ell, \mathcal{C}_2^{\ell_2}\mathcal{T}_2^\ell\mathcal{F}_2 \mapsto_{\Delta_\ell' \Vdash K_\ell'}^{*\Theta_2^\ell;\Upsilon_1^\ell} \mathcal{C}_2', \mathcal{A}_2^{\ell'}, \mathcal{F}_2'$, then
$\forall x \in \Upsilon_1^\ell. \forall m. (\mathcal{C}_1', \mathcal{A}_1^{\ell'}\mathbf{proc}(u_\alpha^c, u_\alpha^c \leftarrow y_\beta^c@d_1), \mathcal{F}_1'; \mathcal{C}_2', \mathcal{A}_2^{\ell'}, \mathcal{F}_2') \in \mathcal{V}_\Psi^\xi [\![\Delta_\ell' \Vdash u_\alpha^c{:} \oplus \{\ell{:}A_\ell\}_{\ell \in L}]\!]_{:;x}^{m+1}$ and
$\forall x \in (\Theta_1^\ell \cap \Theta_2^\ell). \forall m. (\mathcal{C}_1', \mathcal{A}_1^{\ell'}\mathbf{proc}(u_\alpha^c, u_\alpha^c \leftarrow y_\beta^c@d_1), \mathcal{F}_1'; \mathcal{C}_2', \mathcal{A}_2^{\ell'}, \mathcal{F}_2') \in \mathcal{V}_\Psi^\xi [\![\Delta_\ell' \Vdash u_\alpha^c{:} \oplus \{\ell{:}A_\ell\}_{\ell \in L}]\!]_{x;:}^{m+1}$.

We pick that $\mathcal{A}_2^{\ell'}$ and apply existential elimination.

With a similar reasoning on $\dagger_1^\kappa$ and $\mathcal{C}_1^{\kappa_1}\mathcal{B}_1^\kappa, \mathcal{A}_1^\kappa, \mathcal{C}_1^{\kappa_2}\mathcal{T}_1^\kappa\mathcal{F}_1 \mapsto_{\Delta_\kappa' \Vdash K_\kappa'}^{*\Theta^{\kappa'};\Upsilon_1^\kappa} \mathcal{C}_1^{\kappa_1}\mathcal{B}_1^\kappa, \mathcal{A}_1^\kappa, \mathcal{C}_1^{\kappa_2}\mathcal{T}_1^\kappa\mathcal{F}_1$, we get

$\star_\kappa' \ \exists \mathcal{A}_2^{\kappa'}. \mathcal{C}_2^{\kappa_1}\mathcal{B}_2^\kappa, \mathcal{A}_2^\kappa, \mathcal{C}_2^{\kappa_2}\mathcal{T}_2^\kappa\mathcal{F}_2 \mapsto_{\Delta_\kappa' \Vdash K_\kappa'}^{*\Theta^{\kappa'};\Upsilon_1^\kappa} \mathcal{C}_2^{\kappa_1}\mathcal{B}_2^\kappa, \mathcal{A}_2^{\kappa'}, \mathcal{C}_2^{\kappa_2}\mathcal{T}_2^\kappa\mathcal{F}_2$, and
$\forall \mathcal{C}_1', \mathcal{F}_1', \mathcal{C}_2', \mathcal{F}_2'.$ if $\mathcal{C}_1^{\kappa_1}\mathcal{B}_1^\kappa, \mathcal{A}_1^\kappa, \mathcal{C}_1^{\kappa_2}\mathcal{T}_1^\kappa\mathcal{F}_1 \mapsto_{\Delta_\kappa' \Vdash K_\kappa'}^{*\Theta_1^\kappa;\Upsilon_1^\kappa} \mathcal{C}_1', \mathcal{A}_1^\kappa, \mathcal{F}_1'$ and $\mathcal{C}_2^{\kappa_1}\mathcal{B}_2^\kappa, \mathcal{A}_2^\kappa, \mathcal{C}_2^{\kappa_2}\mathcal{T}_2^\kappa\mathcal{F}_2 \mapsto_{\Delta_\kappa' \Vdash K_\kappa'}^{*\Theta_2;\Upsilon_1} \mathcal{C}_2', \mathcal{A}_2^{\kappa'}, \mathcal{F}_2'$, then
$\forall x \in \Upsilon_1^\kappa. \forall m. (\mathcal{C}_1', \mathcal{A}_1^\kappa, \mathcal{F}_1'; \mathcal{C}_2', \mathcal{A}_2^{\kappa'}, \mathcal{F}_2') \in \mathcal{V}_\Psi^\xi [\![\Delta_\kappa'', u_\alpha^c{:} \oplus \{\ell{:}A_\ell\}_{\ell \in L} \Vdash K_\kappa']\!]_{:;x}^{m+1}$ and
$\forall x \in (\Theta_1^\kappa \cap \Theta_2^\kappa). \forall m. (\mathcal{C}_1', \mathcal{A}_1^\kappa, \mathcal{F}_1'; \mathcal{C}_2', \mathcal{A}_2^{\kappa'}, \mathcal{F}_2') \in \mathcal{V}_\Psi^\xi [\![\Delta_\kappa'', u_\alpha^c{:} \oplus \{\ell{:}A_\ell\}_{\ell \in L} \Vdash K_\kappa']\!]_{x;:}^{m+1}$.

We pick that $\mathcal{A}_2^{\kappa'}$ and apply existential elimination. With these choices of $\mathcal{A}_2^{\ell'}$ and $\mathcal{A}_2^{\kappa'}$ fixed, we apply universal elimination on $\star_\ell'$ and $\star_\kappa'$.

We instantiate the quantifiers in $\star_\ell'$ respectively with, $\mathcal{C}_1^{\ell_1}\mathcal{B}_1^\ell$ (for $\mathcal{C}_1'$ or the lower substitution of the first run), $\mathcal{C}_1^{\ell_2}\mathcal{T}_1^\ell\mathcal{F}_1$ (for $\mathcal{F}_1'$ for the upper substitution of the first run), $\mathcal{C}_2^{\ell_1}\mathcal{B}_2^\ell$ (for $\mathcal{C}_2'$ or the lower substitution of the second run) and $\mathcal{C}_2^{\ell_2}\mathcal{A}_2^{\kappa'}\mathcal{T}_2^{\ell'}\mathcal{F}_2$ (for $\mathcal{F}_2'$ or the lower substitution of the second run).

We instantiate the quantifiers in $\star_\kappa'$ respectively with, $\mathcal{C}_1^{\kappa_1}\mathcal{B}_1^\kappa\mathcal{A}_1^{\ell'}\mathbf{proc}(u_\alpha^c, u_\alpha^c \leftarrow y_\beta^c@d_1)$ (for $\mathcal{C}_1'$ or the lower substitution of the first run), $\mathcal{C}_1^{\kappa_2}\mathcal{T}_1^\kappa\mathcal{F}_1$ (for $\mathcal{F}_1'$ for the upper substitution of the first run), $\mathcal{C}_2^{\kappa_1}\mathcal{B}_2^{\kappa'}\mathcal{A}_2^\ell$ (for $\mathcal{C}_2'$ or the lower substitution of the second run) and $\mathcal{C}_2^{\kappa_2}\mathcal{T}_2^\kappa\mathcal{F}_2$ (for $\mathcal{F}_2'$ or the lower substitution of the second run).

$\star_\ell'' \ \forall x \in \Upsilon_1^\ell. \forall m. (\mathcal{C}_1^{\ell_1}\mathcal{B}_1^\ell, \mathcal{A}_1^{\ell'}\mathbf{proc}(u_\alpha^c, u_\alpha^c \leftarrow y_\beta^c@d_1), \mathcal{C}_1^{\ell_2}\mathcal{A}_1^\kappa\mathcal{T}_1^{\ell'}\mathcal{F}_1; \mathcal{C}_2^{\ell_1}\mathcal{B}_2^\ell, \mathcal{A}_2^{\ell'}, \mathcal{C}_2^{\ell_2}\mathcal{A}_2^{\kappa'}\mathcal{T}_2^{\ell'}\mathcal{F}_2) \in \mathcal{V}_\Psi^\xi [\![\Delta_\ell' \Vdash u_\alpha^c{:} \oplus \{\ell{:}A_\ell\}_{\ell \in L}]\!]_{:;x}^{m+1}$ and
$\forall x \in (\Theta_1^\ell \cap \Theta_2^\ell). \forall m. (\mathcal{C}_1^{\ell_1}\mathcal{B}_1^\ell, \mathcal{A}_1^{\ell'}\mathbf{proc}(u_\alpha^c, u_\alpha^c \leftarrow y_\beta^c@d_1), \mathcal{C}_1^{\ell_2}\mathcal{A}_1^\kappa\mathcal{T}_1^{\ell'}\mathcal{F}_1; \mathcal{C}_2^{\ell_1}\mathcal{B}_2^\ell, \mathcal{A}_2^{\ell'}, \mathcal{C}_2^{\ell_2}\mathcal{A}_2^{\kappa'}\mathcal{T}_2^{\ell'}\mathcal{F}_2) \in \mathcal{V}_\Psi^\xi [\![\Delta_\ell' \Vdash u_\alpha^c{:} \oplus \{\ell{:}A_\ell\}_{\ell \in L}]\!]_{x;:}^{m+1}$.

$\star_\kappa'' \ \forall x \in \Upsilon_1^\kappa. \forall m. (\mathcal{C}_1^{\kappa_1}\mathcal{B}_1^\kappa\mathcal{A}_1^{\ell'}\mathbf{proc}(u_\alpha^c, u_\alpha^c \leftarrow y_\beta^c@d_1), \mathcal{A}_1^\kappa, \mathcal{C}_1^{\kappa_2}\mathcal{T}_1^\kappa\mathcal{F}_1; \mathcal{C}_2^{\kappa_1}\mathcal{B}_2^{\kappa'}\mathcal{A}_2^{\ell'}, \mathcal{A}_2^{\kappa'}, \mathcal{C}_2^{\kappa_2}\mathcal{T}_2^\kappa\mathcal{F}_2) \in \mathcal{V}_\Psi^\xi [\![\Delta_\kappa'', u_\alpha^c{:} \oplus \{\ell{:}A_\ell\}_{\ell \in L} \Vdash K_\kappa']\!]_{:;x}^{m+1}$ and
$\forall x \in (\Theta_1^\kappa \cap \Theta_2^\kappa). \forall m. (\mathcal{C}_1^{\kappa_1}\mathcal{B}_1^\kappa\mathcal{A}_1^{\ell'}\mathbf{proc}(u_\alpha^c, u_\alpha^c \leftarrow y_\beta^c@d_1), \mathcal{A}_1^\kappa, \mathcal{C}_1^{\kappa_2}\mathcal{T}_1^\kappa\mathcal{F}_1; \mathcal{C}_2^{\kappa_1}\mathcal{B}_2^{\kappa'}\mathcal{A}_2^{\ell'}, \mathcal{A}_2^{\kappa'}, \mathcal{C}_2^{\kappa_2}\mathcal{T}_2^\kappa\mathcal{F}_2) \in \mathcal{V}_\Psi^\xi [\![\Delta_\kappa'', u_\alpha^c{:} \oplus \{\ell{:}A_\ell\}_{\ell \in L} \Vdash K_\kappa']\!]_{x;:}^{m+1}$.

In particular, we know that $u_\alpha^c \in \Upsilon_1^\ell$ and $u_{\alpha\mathcal{C}}^c \in \Theta_1^\kappa \cap \Theta_2^\kappa$. We instantiate $x$ with $u_\alpha^c$ and look at the conditions of the logical relation (Rows 11 and 12) for satisfying

$\star_\ell'' \ \forall m. (\mathcal{C}_1^{\ell_1}\mathcal{B}_1^\ell, \mathcal{A}_1^{\ell'}\mathbf{proc}(u_\alpha^c, u_\alpha^c \leftarrow y_\beta^c@d_1), \mathcal{C}_1^{\ell_2}\mathcal{A}_1^\kappa\mathcal{T}_1^{\ell'}\mathcal{F}_1; \mathcal{C}_2^{\ell_1}\mathcal{B}_2^\ell, \mathcal{A}_2^{\ell'}, \mathcal{C}_2^{\ell_2}\mathcal{A}_2^{\kappa'}\mathcal{T}_2^{\ell'}\mathcal{F}_2) \in \mathcal{V}_\Psi^\xi [\![\Delta_\ell' \Vdash u_\alpha^c{:} \oplus \{\ell{:}A_\ell\}_{\ell \in L}]\!]_{:;u_\alpha^c}^{m+1}$

$\star_\kappa'' \ \forall m. (\mathcal{C}_1^{\kappa_1}\mathcal{B}_1^\kappa\mathcal{A}_1^{\ell'}\mathbf{proc}(u_\alpha^c, u_\alpha^c \leftarrow y_\beta^c@d_1), \mathcal{A}_1^\kappa, \mathcal{C}_1^{\kappa_2}\mathcal{T}_1^\kappa\mathcal{F}_1; \mathcal{C}_2^{\kappa_1}\mathcal{B}_2^{\kappa'}\mathcal{A}_2^{\ell'}, \mathcal{A}_2^{\kappa'}, \mathcal{C}_2^{\kappa_2}\mathcal{T}_2^\kappa\mathcal{F}_2) \in \mathcal{V}_\Psi^\xi [\![\Delta_\kappa'', u_\alpha^c{:} \oplus \{\ell{:}A_\ell\}_{\ell \in L} \Vdash K_\kappa']\!]_{u_{\alpha\mathcal{C}}^c;:}^{m+1}$.

By the conditions on row 11, we know $\mathcal{A}_2^{\ell'} = \mathcal{A}_2^{\ell''}\mathbf{proc}(u_\alpha^c, u_\alpha^c \leftarrow y_\beta^c@d_2)$ and

$\dagger_\ell' \ \forall m. ([y_\beta^c/u_\alpha^c]\mathcal{C}_1^{\ell_1}\mathcal{B}_1^\ell, \mathcal{A}_1^{\ell'}, \mathcal{C}_1^{\ell_2}\mathcal{A}_1^\kappa\mathcal{T}_1^{\ell'}\mathcal{F}_1; [y_\beta^c/u_\alpha^c]\mathcal{C}_2^{\ell_1}\mathcal{B}_2^\ell, \mathcal{A}_2^{\ell''}, \mathcal{C}_2^{\ell_2}\mathcal{A}_2^{\kappa'}\mathcal{T}_2^{\ell'}\mathcal{F}_2) \in \mathcal{E}_\Psi^\xi [\![\Delta_\ell' \Vdash y_\beta^c{:}A_k]\!]^m$

30

By the conditions on row 12, and the knowledge that $\mathcal{A}_2^{\ell'} = \mathcal{A}_2^{\ell''}\mathbf{proc}(u_\alpha^c, u_\alpha^c \leftarrow y_\beta^c @d_2)$, we get

$$\dagger_\kappa' \; \forall m.\, ([y_\beta^c/u_\alpha^c]\mathcal{C}_1^{\kappa_1}\mathcal{B}_1^{\kappa'}\mathcal{A}_1^{\ell'}, \mathcal{A}_1^\kappa, \mathcal{C}_1^{\kappa_2}\mathcal{T}_1^\kappa\mathcal{F}_1; [y_\beta^c/u_\alpha^c]\mathcal{C}_2^{\kappa_1}\mathcal{B}_2^{\kappa'}\mathcal{A}_2^{\ell''}, \mathcal{A}_2^{\kappa'}, \mathcal{C}_2^{\kappa_2}\mathcal{T}_2^\kappa\mathcal{F}_2) \in \mathcal{E}_\Psi^\xi[\![\Delta_\kappa'', y_\beta^c{:}A_k \Vdash K_\kappa']\!]^m.$$

We know that the above renaming only affects $\mathcal{A}_i^\kappa \neq \cdot$. We have $\mathcal{A}_1^{\kappa'} = [y_\beta^c/u_\alpha^c]\mathcal{A}_1^\kappa$, and we define $\mathcal{A}_2^{\kappa''} = [y_\beta^c/u_\alpha^c]\mathcal{A}_2^{\kappa'}$. From $\dagger_\ell'$ and $\dagger_\kappa'$ we get

$$\dagger_\ell'' \; \forall m.\,(\mathcal{C}_1^{\ell_1}\mathcal{B}_1^\ell, \mathcal{A}_1^{\ell'}, \mathcal{C}_1^{\ell_2}\mathcal{A}_1^{\kappa'}\mathcal{T}_1^{\ell'}\mathcal{F}_1; \mathcal{C}_2^{\ell_1}\mathcal{B}_2^\ell, \mathcal{A}_2^{\ell''}, \mathcal{C}_2^{\ell_2}\mathcal{A}_2^{\kappa''}\mathcal{T}_2^{\ell'}\mathcal{F}_2) \in \mathcal{E}_\Psi^\xi[\![\Delta_\ell' \Vdash y_\beta^c{:}A_k]\!]^m$$

$$\dagger_\kappa'' \; \forall m.\,(\mathcal{C}_1^{\kappa_1}\mathcal{B}_1^{\kappa'}\mathcal{A}_1^{\ell'}, \mathcal{A}_1^{\kappa'}, \mathcal{C}_1^{\kappa_2}\mathcal{T}_1^\kappa\mathcal{F}_1; \mathcal{C}_2^{\kappa_1}\mathcal{B}_2^{\kappa'}\mathcal{A}_2^{\ell''}, \mathcal{A}_2^{\kappa''}, \mathcal{C}_2^{\kappa_2}\mathcal{T}_2^\kappa\mathcal{F}_2) \in \mathcal{E}_\Psi^\xi[\![\Delta_\kappa'', y_\beta^c{:}A_k \Vdash K_\kappa']\!]^m.$$

By forward closure (Lemma C.11) on $\dagger_1^\iota$, for any other $\iota \neq \kappa, \ell$ we step $\mathcal{A}_i^\kappa$ in $\mathcal{B}_i^\iota$ or $\mathcal{T}_i^\iota$ to get

$$\dagger_1^{\iota'} \; \forall m.\,(\mathcal{C}_1^{\iota_1}\mathcal{B}_1^{\iota'}\mathcal{A}_1^\iota\mathcal{C}_1^{\iota_2}\mathcal{T}_1^{\iota'}\mathcal{F}_1; \mathcal{C}_2^{\iota_1}\mathcal{B}_2^{\iota'}\mathcal{A}_2^\iota\mathcal{C}_2^{\iota_2}\mathcal{T}_2^{\iota'}\mathcal{F}_2) \in \mathcal{E}_{\Psi_0}^\xi[\![\Delta_\iota' \Vdash K_\iota']\!]^m,$$

Define $\mathcal{D}_2''$ as $\{\mathcal{A}_2^{i''}\}_{i \leq n}$ where for $\iota \neq \kappa, \ell$, we put $\mathcal{A}_2^{\iota''} = \mathcal{A}_2^\iota$. From the above discussion, we know that $\mathcal{D}_2 \mapsto^*_{\Delta \Vdash K} \mathcal{D}_2''$. By induction on the number of steps $j'$, we get

$\exists \mathcal{D}_2'.\, \mathcal{C}_2, \mathcal{D}_2'', \mathcal{F}_2 \mapsto^{*\Theta'; \Upsilon_1}_{\Delta \Vdash K} \mathcal{C}_2, \mathcal{D}_2', \mathcal{F}_2$, and $\forall \mathcal{C}_1', \mathcal{F}_1', \mathcal{C}_2', \mathcal{F}_2'.\,$if $\mathcal{C}_1, \mathcal{D}_1', \mathcal{F}_1 \mapsto^{*\Theta_2; \Upsilon_1}_{\Delta \Vdash K} \mathcal{C}_1', \mathcal{D}_1', \mathcal{F}_1'$ and $\mathcal{C}_2, \mathcal{D}_2', \mathcal{F}_2 \mapsto^{*\Theta_2; \Upsilon_1}_{\Delta \Vdash K} \mathcal{C}_2', \mathcal{D}_2', \mathcal{F}_2'$, then
$\forall x \in \Upsilon_1.\,(\mathcal{C}_1', \mathcal{D}_1', \mathcal{F}_1'; \mathcal{C}_2', \mathcal{D}_2', \mathcal{F}_2') \in \mathcal{V}_\Psi^\xi[\![\Delta \Vdash K]\!]_{:,x}^{m+1}$ and
$\forall x \in (\Theta_1 \cap \Theta_2).\,(\mathcal{C}_1', \mathcal{D}_1', \mathcal{F}_1'; \mathcal{C}_2', \mathcal{D}_2', \mathcal{F}_2') \in \mathcal{V}_\Psi^\xi[\![\Delta \Vdash K]\!]_{x;:}^{m+1}.$

By the knowledge that $\mathcal{C}_2, \mathcal{D}_2, \mathcal{F}_2 \mapsto^{*\Theta''; \Upsilon_1}_{\Delta \Vdash K} \mathcal{C}_2, \mathcal{D}_2'', \mathcal{F}_2$, we know that for the same $\mathcal{D}_2'$, we have

$\exists \mathcal{D}_2'.\, \mathcal{C}_2, \mathcal{D}_2, \mathcal{F}_2 \mapsto^{*\Theta'; \Upsilon_1}_{\Delta \Vdash K} \mathcal{C}_2, \mathcal{D}_2', \mathcal{F}_2$, and $\forall \mathcal{C}_1', \mathcal{F}_1', \mathcal{C}_2', \mathcal{F}_2'.\,$if $\mathcal{C}_1, \mathcal{D}_1', \mathcal{F}_1 \mapsto^{*\Theta_2; \Upsilon_1}_{\Delta \Vdash K} \mathcal{C}_1', \mathcal{D}_1', \mathcal{F}_1'$ and $\mathcal{C}_2, \mathcal{D}_2', \mathcal{F}_2 \mapsto^{*\Theta_2; \Upsilon_1}_{\Delta \Vdash K} \mathcal{C}_2', \mathcal{D}_2', \mathcal{F}_2'$, then
$\forall x \in \Upsilon_1.\,(\mathcal{C}_1', \mathcal{D}_1', \mathcal{F}_1'; \mathcal{C}_2', \mathcal{D}_2', \mathcal{F}_2') \in \mathcal{V}_\Psi^\xi[\![\Delta \Vdash K]\!]_{:,x}^{m+1}$ and
$\forall x \in (\Theta_1 \cap \Theta_2).\,(\mathcal{C}_1', \mathcal{D}_1', \mathcal{F}_1'; \mathcal{C}_2', \mathcal{D}_2', \mathcal{F}_2') \in \mathcal{V}_\Psi^\xi[\![\Delta \Vdash K]\!]_{x;:}^{m+1}.$

and the proof of this subcase is complete.

**Subcase 2.** $u_\alpha^c{:}\oplus\{\ell{:}A_\ell\}_{\ell \in L}$ and there is a message in $\mathcal{A}_1^\ell$, i.e. $\mathcal{A}_1^\ell = \mathcal{A}_1^{\ell'}\mathbf{msg}(u_\alpha^c.k; u_\alpha^c \leftarrow y_\beta^c)$, and $\mathcal{D}_1 = \mathcal{D}_1^1 \mathcal{A}_1^{\ell'}\mathbf{msg}(u_\alpha^c.k; u_\alpha^c \leftarrow y_\beta^c)\mathcal{A}_1^\kappa \mathcal{D}_1^2$ and

$$\mathcal{C}_1, \mathcal{D}_1^1\mathcal{A}_1^{\ell'}\mathbf{msg}(u_\alpha^c.k; u_\alpha^c \leftarrow y_\beta^c)\mathcal{A}_1^\kappa \mathcal{D}_1^2, \mathcal{F}_1 \mapsto \mathcal{C}_1, \mathcal{D}_1^1\mathcal{A}_1^{\ell'}\mathcal{A}_1^{\kappa'}\mathcal{D}_1^2, \mathcal{F}_1 \mapsto^{j'\Theta; \Upsilon_1}_{\Delta \Vdash K} \mathcal{C}_1, \mathcal{D}_1', \mathcal{F}_1.$$

In other words, $\mathcal{A}_1^\ell$ has a message that $\mathcal{A}_1^\kappa$ receives and steps to $\mathcal{A}_1^{\kappa'}$.

By $\dagger_1^\ell$, and $\mathcal{C}_1^{\ell_1}\mathcal{B}_1^\ell, \mathcal{A}_1^\ell, \mathcal{C}_1^{\ell_2}\mathcal{T}_1^\ell\mathcal{F}_1 \mapsto^{*\Theta^{\ell'}; \Upsilon_1^\ell}_{\Delta_\ell' \Vdash K_\ell'} \mathcal{C}_1^{\ell_1}\mathcal{B}_1^\ell, \mathcal{A}_1^\ell, \mathcal{C}_1^{\ell_2}\mathcal{T}_1^\ell\mathcal{F}_1$, and $\mathcal{A}_1^\ell = \mathcal{A}_1^{\ell'}\mathbf{msg}(u_\alpha^c.k; u_\alpha^c \leftarrow y_\beta^c)$ we know that

$\star_\ell \; \forall m. \exists \mathcal{A}_2^{\ell'}.\, \mathcal{C}_2^{\ell_1}\mathcal{B}_2^\ell, \mathcal{A}_2^\ell, \mathcal{C}_2^{\ell_2}\mathcal{T}_2^\ell\mathcal{F}_2 \mapsto^{*\Theta^{\ell'}; \Upsilon_1^\ell}_{\Delta_\ell' \Vdash K_\ell'} \mathcal{C}_2^{\ell_1}\mathcal{B}_2^\ell, \mathcal{A}_2^{\ell'}, \mathcal{C}_2^{\ell_2}\mathcal{T}_2^\ell\mathcal{F}_2$, and
$\forall \mathcal{C}_1', \mathcal{F}_1', \mathcal{C}_2', \mathcal{F}_2'.$ if $\mathcal{C}_1^{\ell_1}\mathcal{B}_1^\ell, \mathcal{A}_1^\ell, \mathcal{C}_1^{\ell_2}\mathcal{T}_1^\ell\mathcal{F}_1 \mapsto^{*\Theta_1^\ell; \Upsilon_1^\ell}_{\Delta_\ell' \Vdash K_\ell'} \mathcal{C}_1', \mathcal{A}_1^\ell, \mathcal{F}_1'$ and $\mathcal{C}_2^{\ell_1}\mathcal{B}_2^\ell, \mathcal{A}_2^\ell, \mathcal{C}_2^{\ell_2}\mathcal{T}_2^\ell\mathcal{F}_2 \mapsto^{*\Theta_2^\ell; \Upsilon_1^\ell}_{\Delta_\ell' \Vdash K_\ell'} \mathcal{C}_2', \mathcal{A}_2^{\ell'}, \mathcal{F}_2'$, then
$\forall x \in \Upsilon_1^\ell.\,(\mathcal{C}_1', \mathcal{A}_1^{\ell'}\mathbf{msg}(u_\alpha^c.k; u_\alpha^c \leftarrow y_\beta^c), \mathcal{F}_1'; \mathcal{C}_2', \mathcal{A}_2^{\ell'}, \mathcal{F}_2') \in \mathcal{V}_\Psi^\xi[\![\Delta_\ell' \Vdash u_\alpha^c{:} \oplus\{\ell{:}A_\ell\}_{\ell \in L}]\!]_{:,x}^{m+1}$ and
$\forall x \in (\Theta_1^\ell \cap \Theta_2^\ell).\,(\mathcal{C}_1', \mathcal{A}_1^{\ell'}\mathbf{msg}(u_\alpha^c.k; u_\alpha^c \leftarrow y_\beta^c), \mathcal{F}_1'; \mathcal{C}_2', \mathcal{A}_2^{\ell'}, \mathcal{F}_2') \in \mathcal{V}_\Psi^\xi[\![\Delta_\ell' \Vdash u_\alpha^c{:} \oplus\{\ell{:}A_\ell\}_{\ell \in L}]\!]_{x;:}^{m+1}.$

By Lemma C.8, we get

$\star_\ell' \; \exists \mathcal{A}_2^{\ell'}.\, \mathcal{C}_2^{\ell_1}\mathcal{B}_2^\ell, \mathcal{A}_2^\ell, \mathcal{C}_2^{\ell_2}\mathcal{T}_2^\ell\mathcal{F}_2 \mapsto^{*\Theta^{\ell'}; \Upsilon_1^\ell}_{\Delta_\ell' \Vdash K_\ell'} \mathcal{C}_2^{\ell_1}\mathcal{B}_2^\ell, \mathcal{A}_2^{\ell'}, \mathcal{C}_2^{\ell_2}\mathcal{T}_2^\ell\mathcal{F}_2$, and
$\forall \mathcal{C}_1', \mathcal{F}_1', \mathcal{C}_2', \mathcal{F}_2'.$ if $\mathcal{C}_1^{\ell_1}\mathcal{B}_1^\ell, \mathcal{A}_1^\ell, \mathcal{C}_1^{\ell_2}\mathcal{T}_1^\ell\mathcal{F}_1 \mapsto^{*\Theta_1^\ell; \Upsilon_1^\ell}_{\Delta_\ell' \Vdash K_\ell'} \mathcal{C}_1', \mathcal{A}_1^\ell, \mathcal{F}_1'$ and $\mathcal{C}_2^{\ell_1}\mathcal{B}_2^\ell, \mathcal{A}_2^\ell, \mathcal{C}_2^{\ell_2}\mathcal{T}_2^\ell\mathcal{F}_2 \mapsto^{*\Theta_2^\ell; \Upsilon_1^\ell}_{\Delta_\ell' \Vdash K_\ell'} \mathcal{C}_2', \mathcal{A}_2^{\ell'}, \mathcal{F}_2'$, then
$\forall x \in \Upsilon_1^\ell. \forall m.\,(\mathcal{C}_1', \mathcal{A}_1^{\ell'}\mathbf{msg}(u_\alpha^c.k; u_\alpha^c \leftarrow y_\beta^c), \mathcal{F}_1'; \mathcal{C}_2', \mathcal{A}_2^{\ell'}, \mathcal{F}_2') \in \mathcal{V}_\Psi^\xi[\![\Delta_\ell' \Vdash u_\alpha^c{:} \oplus\{\ell{:}A_\ell\}_{\ell \in L}]\!]_{:,x}^{m+1}$ and
$\forall x \in (\Theta_1^\ell \cap \Theta_2^\ell). \forall m.\,(\mathcal{C}_1', \mathcal{A}_1^{\ell'}\mathbf{msg}(u_\alpha^c.k; u_\alpha^c \leftarrow y_\beta^c), \mathcal{F}_1'; \mathcal{C}_2', \mathcal{A}_2^{\ell'}, \mathcal{F}_2') \in \mathcal{V}_\Psi^\xi[\![\Delta_\ell' \Vdash u_\alpha^c{:} \oplus\{\ell{:}A_\ell\}_{\ell \in L}]\!]_{x;:}^{m+1}.$

We pick that $\mathcal{A}_2^{\ell'}$ and apply existential elimination.

With a similar reasoning on $\dagger_1^\kappa$ and $\mathcal{C}_1^{\kappa_1}\mathcal{B}_1^\kappa, \mathcal{A}_1^\kappa, \mathcal{C}_1^{\kappa_2}\mathcal{T}_1^\kappa\mathcal{F}_1 \mapsto^{*\Theta^{\kappa'}; \Upsilon_1^\kappa}_{\Delta_\kappa' \Vdash K_\kappa'} \mathcal{C}_1^{\kappa_1}\mathcal{B}_1^\kappa, \mathcal{A}_1^\kappa, \mathcal{C}_1^{\kappa_2}\mathcal{T}_1^\kappa\mathcal{F}_1$, we get

$\star_\kappa' \; \exists \mathcal{A}_2^{\kappa'}.\, \mathcal{C}_2^{\kappa_1}\mathcal{B}_2^\kappa, \mathcal{A}_2^\kappa, \mathcal{C}_2^{\kappa_2}\mathcal{T}_2^\kappa\mathcal{F}_2 \mapsto^{*\Theta^{\kappa'}; \Upsilon_1^\kappa}_{\Delta_\kappa' \Vdash K_\kappa'} \mathcal{C}_2^{\kappa_1}\mathcal{B}_2^\kappa, \mathcal{A}_2^{\kappa'}, \mathcal{C}_2^{\kappa_2}\mathcal{T}_2^\kappa\mathcal{F}_2$, and
$\forall \mathcal{C}_1', \mathcal{F}_1', \mathcal{C}_2', \mathcal{F}_2'.$ if $\mathcal{C}_1^{\kappa_1}\mathcal{B}_1^\kappa, \mathcal{A}_1^\kappa, \mathcal{C}_1^{\kappa_2}\mathcal{T}_1^\kappa\mathcal{F}_1 \mapsto^{*\Theta_1^\kappa; \Upsilon_1^\kappa}_{\Delta_\kappa' \Vdash K_\kappa'} \mathcal{C}_1', \mathcal{A}_1^\kappa, \mathcal{F}_1'$ and $\mathcal{C}_2^{\kappa_1}\mathcal{B}_2^\kappa, \mathcal{A}_2^\kappa, \mathcal{C}_2^{\kappa_2}\mathcal{T}_2^\kappa\mathcal{F}_2 \mapsto^{*\Theta_2; \Upsilon_1}_{\Delta_\kappa' \Vdash K_\kappa'} \mathcal{C}_2', \mathcal{A}_2^{\kappa'}, \mathcal{F}_2'$, then
$\forall x \in \Upsilon_1^\kappa. \forall m.\,(\mathcal{C}_1', \mathcal{A}_1^\kappa, \mathcal{F}_1'; \mathcal{C}_2', \mathcal{A}_2^{\kappa'}, \mathcal{F}_2') \in \mathcal{V}_\Psi^\xi[\![\Delta_\kappa'', u_\alpha^c{:}\oplus\{\ell{:}A_\ell\}_{\ell \in L} \Vdash K_\kappa']\!]_{:,x}^{m+1}$ and
$\forall x \in (\Theta_1^\kappa \cap \Theta_2^\kappa). \forall m.\,(\mathcal{C}_1', \mathcal{A}_1^\kappa, \mathcal{F}_1'; \mathcal{C}_2', \mathcal{A}_2^{\kappa'}, \mathcal{F}_2') \in \mathcal{V}_\Psi^\xi[\![\Delta_\kappa'', u_\alpha^c{:}\oplus\{\ell{:}A_\ell\}_{\ell \in L} \Vdash K_\kappa']\!]_{x;:}^{m+1}.$

31

We pick that $\mathcal{A}_2^{\kappa'}$ and apply existential elimination. With these choices of $\mathcal{A}_2^{\ell'}$ and $\mathcal{A}_2^{\kappa'}$ fixed, we apply universal elimination on $\star'_\ell$ and $\star'_\kappa$.

We instantiate the quantifiers in $\star'_\ell$ respectively with, $\mathcal{C}_1^{\ell_1}\mathcal{B}_1^\ell$ (for $\mathcal{C}'_1$ or the lower substitution of the first run), $\mathcal{C}_1^{\ell_2}\mathcal{T}_1^{\ell'}\mathcal{F}_1$ (for $\mathcal{F}'_1$ for the upper substitution of the first run), $\mathcal{C}_2^{\ell_1}\mathcal{B}_2^\ell$ (for $\mathcal{C}'_2$ or the lower substitution of the second run) and $\mathcal{C}_2^{\ell_2}\mathcal{A}_2^{\kappa'}\mathcal{T}_2^{\ell'}\mathcal{F}_2$ (for $\mathcal{F}'_2$ or the lower substitution of the second run).

We instantiate the quantifiers in $\star'_\kappa$ respectively with, $\mathcal{C}_1^{\kappa_1}\mathcal{B}_1^{\kappa'}\mathcal{A}_1^{\ell'}\mathbf{msg}(u_\alpha^c.k; u_\alpha^c \leftarrow y_\beta^c)$ (for $\mathcal{C}'_1$ or the lower substitution of the first run), $\mathcal{C}_1^{\kappa_2}\mathcal{T}_1^\kappa\mathcal{F}_1$ (for $\mathcal{F}'_1$ for the upper substitution of the first run), $\mathcal{C}_2^{\kappa_1}\mathcal{B}_2^{\kappa'}\mathcal{A}_2^\ell$ (for $\mathcal{C}'_2$ or the lower substitution of the second run) and $\mathcal{C}_2^{\kappa_2}\mathcal{T}_2^\kappa\mathcal{F}_2$ (for $\mathcal{F}'_2$ or the lower substitution of the second run).

$\star''_\ell \quad \forall x \in \Upsilon_1^\ell. \forall m. (\mathcal{C}_1^{\ell_1}\mathcal{B}_1^\ell, \mathcal{A}_1^{\ell'}\mathbf{msg}(u_\alpha^c.k; u_\alpha^c \leftarrow y_\beta^c), \mathcal{C}_1^{\ell_2}\mathcal{A}_1^\kappa\mathcal{T}_1^{\ell'}\mathcal{F}_1; \mathcal{C}_2^{\ell_1}\mathcal{B}_2^\ell, \mathcal{A}_2^{\ell'}, \mathcal{C}_2^{\ell_2}\mathcal{A}_2^{\kappa'}\mathcal{T}_2^{\ell'}\mathcal{F}_2) \in \mathcal{V}_\Psi^\xi[\![\Delta'_\ell \Vdash u_\alpha^c: \oplus \{\ell{:}A_\ell\}_{\ell \in L}]\!]_{\cdot; x}^{m+1}$ and
$\forall x \in (\Theta_1^\ell \cap \Theta_2^\ell). \forall m. (\mathcal{C}_1^{\ell_1}\mathcal{B}_1^\ell, \mathcal{A}_1^{\ell'}\mathbf{msg}(u_\alpha^c.k; u_\alpha^c \leftarrow y_\beta^c), \mathcal{C}_1^{\ell_2}\mathcal{A}_1^\kappa\mathcal{T}_1^{\ell'}\mathcal{F}_1; \mathcal{C}_2^{\ell_1}\mathcal{B}_2^\ell, \mathcal{A}_2^{\ell'}, \mathcal{C}_2^{\ell_2}\mathcal{A}_2^{\kappa'}\mathcal{T}_2^{\ell'}\mathcal{F}_2) \in \mathcal{V}_\Psi^\xi[\![\Delta'_\ell \Vdash u_\alpha^c: \oplus \{\ell{:}A_\ell\}_{\ell \in L}]\!]_{x;\cdot}^{m+1}$.

$\star''_\kappa \quad \forall x \in \Upsilon_1^\kappa. \forall m. (\mathcal{C}_1^{\kappa_1}\mathcal{B}_1^\kappa\mathcal{A}_1^{\ell'}\mathbf{msg}(u_\alpha^c.k; u_\alpha^c \leftarrow y_\beta^c), \mathcal{A}_1^\kappa, \mathcal{C}_1^{\kappa_2}\mathcal{T}_1^\kappa\mathcal{F}_1; \mathcal{C}_2^{\kappa_1}\mathcal{B}_2^{\kappa'}\mathcal{A}_2^{\ell'}, \mathcal{A}_2^{\kappa'}, \mathcal{C}_2^{\kappa_2}\mathcal{T}_2^\kappa\mathcal{F}_2) \in \mathcal{V}_\Psi^\xi[\![\Delta''_\kappa, u_\alpha^c: \oplus \{\ell{:}A_\ell\}_{\ell \in L} \Vdash K'_\kappa]\!]_{\cdot; x}^{m+1}$ and
$\forall x \in (\Theta_1^\kappa \cap \Theta_2^\kappa). \forall m. (\mathcal{C}_1^{\kappa_1}\mathcal{B}_1^\kappa\mathcal{A}_1^{\ell'}\mathbf{msg}(u_\alpha^c.k; u_\alpha^c \leftarrow y_\beta^c), \mathcal{A}_1^\kappa, \mathcal{C}_1^{\kappa_2}\mathcal{T}_1^\kappa\mathcal{F}_1; \mathcal{C}_2^{\kappa_1}\mathcal{B}_2^{\kappa'}\mathcal{A}_2^{\ell'}, \mathcal{A}_2^{\kappa'}, \mathcal{C}_2^{\kappa_2}\mathcal{T}_2^\kappa\mathcal{F}_2) \in \mathcal{V}_\Psi^\xi[\![\Delta''_\kappa, u_\alpha^c: \oplus \{\ell{:}A_\ell\}_{\ell \in L} \Vdash K'_\kappa]\!]_{x;\cdot}^{m+1}$.

In particular, we know that $u_\alpha^c \in \Upsilon_1^\ell$ and $u_\alpha^c \in \Theta_1^\kappa \cap \Theta_2^\kappa$. We instantiate $x$ with $u_\alpha^c$ and look at the conditions of the logical relation (Rows 2 and 7) for satisfying

$\star''_\ell \quad \forall m.(\mathcal{C}_1^{\ell_1}\mathcal{B}_1^\ell, \mathcal{A}_1^{\ell'}\mathbf{msg}(u_\alpha^c.k; u_\alpha^c \leftarrow y_\beta^c), \mathcal{C}_1^{\ell_2}\mathcal{A}_1^\kappa\mathcal{T}_1^{\ell'}\mathcal{F}_1; \mathcal{C}_2^{\ell_1}\mathcal{B}_2^\ell, \mathcal{A}_2^{\ell'}, \mathcal{C}_2^{\ell_2}\mathcal{A}_2^{\kappa'}\mathcal{T}_2^{\ell'}\mathcal{F}_2) \in \mathcal{V}_\Psi^\xi[\![\Delta'_\ell \Vdash u_\alpha^c: \oplus \{\ell{:}A_\ell\}_{\ell \in L}]\!]_{\cdot; u_\alpha^c}^{m+1}$

$\star''_\kappa \quad \forall m. (\mathcal{C}_1^{\kappa_1}\mathcal{B}_1^\kappa\mathcal{A}_1^{\ell'}\mathbf{msg}(u_\alpha^c.k; u_\alpha^c \leftarrow y_\beta^c), \mathcal{A}_1^\kappa, \mathcal{C}_1^{\kappa_2}\mathcal{T}_1^\kappa\mathcal{F}_1; \mathcal{C}_2^{\kappa_1}\mathcal{B}_2^{\kappa'}\mathcal{A}_2^{\ell'}, \mathcal{A}_2^{\kappa'}, \mathcal{C}_2^{\kappa_2}\mathcal{T}_2^\kappa\mathcal{F}_2) \in \mathcal{V}_\Psi^\xi[\![\Delta''_\kappa, u_\alpha^c: \oplus \{\ell{:}A_\ell\}_{\ell \in L} \Vdash K'_\kappa]\!]_{u_\alpha^c;\cdot}^{m+1}$.

By the conditions on row 2, we know $\mathcal{A}_2^{\ell'} = \mathcal{A}_2^{\ell''}\mathbf{msg}(u_\alpha^c.k; u_\alpha^c \leftarrow y_\beta^c)$ and

$\dagger'_\ell \quad \forall m.(\mathcal{C}_1^{\ell_1}\mathcal{B}_1^\ell, \mathcal{A}_1^{\ell'}, \mathcal{C}_1^{\ell_2}\mathbf{msg}(u_\alpha^c.k; u_\alpha^c \leftarrow y_\beta^c)\mathcal{A}_1^\kappa\mathcal{T}_1^{\ell'}\mathcal{F}_1; \mathcal{C}_2^{\ell_1}\mathcal{B}_2^\ell, \mathcal{A}_2^{\ell''}, \mathcal{C}_2^{\ell_2}\mathbf{msg}(u_\alpha^c.k; u_\alpha^c \leftarrow y_\beta^c)\mathcal{A}_2^{\kappa'}\mathcal{T}_2^{\ell'}\mathcal{F}_2) \in \mathcal{E}_\Psi^\xi[\![\Delta'_\ell \Vdash y_\beta^c{:}A_k]\!]^m$

By the conditions on row 7, and the knowledge that $\mathcal{A}_2^{\ell'} = \mathcal{A}_2^{\ell''}\mathbf{msg}(u_\alpha^c.k; u_\alpha^c \leftarrow y_\beta^c)$, we get

$\dagger'_\kappa \quad \forall m. (\mathcal{C}_1^{\kappa_1}\mathcal{B}_1^{\kappa'}\mathcal{A}_1^{\ell'}, \mathbf{msg}(u_\alpha^c.k; u_\alpha^c \leftarrow y_\beta^c)\mathcal{A}_1^\kappa, \mathcal{C}_1^{\kappa_2}\mathcal{T}_1^\kappa\mathcal{F}_1; \mathcal{C}_2^{\kappa_1}\mathcal{B}_2^{\kappa'}\mathcal{A}_2^{\ell''}, \mathbf{msg}(u_\alpha^c.k; u_\alpha^c \leftarrow y_\beta^c)\mathcal{A}_2^{\kappa'}, \mathcal{C}_2^{\kappa_2}\mathcal{T}_2^\kappa\mathcal{F}_2) \in \mathcal{E}_\Psi^\xi[\![\Delta''_\kappa, y_\beta^c{:}A_k \Vdash K'_\kappa]\!]^m$.

By applying forward closure (Lemma C.11 on $\dagger'_\ell$ and forward closure (Lemma C.12) on $\dagger'_\kappa$, we get

$\dagger''_\ell \quad \forall m.(\mathcal{C}_1^{\ell_1}\mathcal{B}_1^\ell, \mathcal{A}_1^{\ell'}, \mathcal{C}_1^{\ell_2}\mathcal{A}_1^{\kappa'}\mathcal{T}_1^{\ell'}\mathcal{F}_1; \mathcal{C}_2^{\ell_1}\mathcal{B}_2^\ell, \mathcal{A}_2^{\ell''}, \mathcal{C}_2^{\ell_2}\mathbf{msg}(u_\alpha^c.k; u_\alpha^c \leftarrow y_\beta^c)\mathcal{A}_2^{\kappa'}\mathcal{T}_2^{\ell'}\mathcal{F}_2) \in \mathcal{E}_\Psi^\xi[\![\Delta'_\ell \Vdash y_\beta^c{:}A_k]\!]^m$

$\dagger''_\kappa \quad \forall m. (\mathcal{C}_1^{\kappa_1}\mathcal{B}_1^{\kappa'}\mathcal{A}_1^{\ell'}, \mathcal{A}_1^{\kappa'}, \mathcal{C}_1^{\kappa_2}\mathcal{T}_1^\kappa\mathcal{F}_1; \mathcal{C}_2^{\kappa_1}\mathcal{B}_2^{\kappa'}\mathcal{A}_2^{\ell''}, \mathbf{msg}(u_\alpha^c.k; u_\alpha^c \leftarrow y_\beta^c)\mathcal{A}_2^{\kappa'}, \mathcal{C}_2^{\kappa_2}\mathcal{T}_2^\kappa\mathcal{F}_2) \in \mathcal{E}_\Psi^\xi[\![\Delta''_\kappa, y_\beta^c{:}A_k \Vdash K'_\kappa]\!]^m$.

By forward closure, for any other $\iota \neq \kappa, \ell$ we get

$\dagger''_1 \quad \forall m. (\mathcal{C}_1^{\iota_1}\mathcal{B}_1^{\iota'}\mathcal{A}_1^\iota\mathcal{C}_1^{\iota_2}\mathcal{T}_1^{\iota'}\mathcal{F}_1; \mathcal{C}_2^{\iota_1}\mathcal{B}_2^\iota\mathcal{A}_2^\iota\mathcal{C}_2^{\iota_2}\mathcal{T}_2^\iota\mathcal{F}_2) \in \mathcal{E}_{\Psi_0}^\xi[\![\Delta'_\iota \Vdash K'_\iota]\!]^m$ ,

Define $\mathcal{D}''_2$ as $\{\mathcal{A}_2^{i''}\}_{i \leq n}$ where for $\iota \neq \kappa, \ell$, we put $\mathcal{A}_2^{\iota''} = \mathcal{A}_2^\iota$ and $\mathcal{A}_2^{\kappa''} = \mathcal{A}_2^{\kappa'}\mathbf{msg}(u_\alpha^c.k; u_\alpha^c \leftarrow y_\beta^c)$. From the above discussion, we know that $\mathcal{D}_2 \mapsto^*_{\Delta \Vdash K} \mathcal{D}''_2$. By induction on the number of steps $j'$, we get

$\exists \mathcal{D}'_2. \mathcal{C}_2, \mathcal{D}''_2, \mathcal{F}_2 \mapsto^{*\Theta';\Upsilon_1}_{\Delta \Vdash K} \mathcal{C}_2, \mathcal{D}'_2, \mathcal{F}_2$, and $\forall \mathcal{C}'_1, \mathcal{F}'_1, \mathcal{C}'_2, \mathcal{F}'_2$.if $\mathcal{C}_1, \mathcal{D}'_1, \mathcal{F}_1 \mapsto^{*\Theta_2;\Upsilon_1}_{\Delta \Vdash K} \mathcal{C}'_1, \mathcal{D}'_1, \mathcal{F}'_1$ and $\mathcal{C}_2, \mathcal{D}'_2, \mathcal{F}_2 \mapsto^{*\Theta_2;\Upsilon_1}_{\Delta \Vdash K} \mathcal{C}'_2, \mathcal{D}'_2, \mathcal{F}'_2$, then
$\forall x \in \Upsilon_1. (\mathcal{C}'_1, \mathcal{D}'_1, \mathcal{F}'_1; \mathcal{C}'_2, \mathcal{D}'_2, \mathcal{F}'_2) \in \mathcal{V}_\Psi^\xi[\![\Delta \Vdash K]\!]_{\cdot; x}^{m+1}$ and
$\forall x \in (\Theta_1 \cap \Theta_2). (\mathcal{C}'_1, \mathcal{D}'_1, \mathcal{F}'_1; \mathcal{C}'_2, \mathcal{D}'_2, \mathcal{F}'_2) \in \mathcal{V}_\Psi^\xi[\![\Delta \Vdash K]\!]_{x;\cdot}^{m+1}$.

By the knowledge that $\mathcal{C}_2, \mathcal{D}_2, \mathcal{F}_2 \mapsto^{*\Theta'';\Upsilon_1}_{\Delta \Vdash K} \mathcal{C}_2, \mathcal{D}''_2, \mathcal{F}_2$, we know that for the same $\mathcal{D}'_2$, we have

$\exists \mathcal{D}'_2. \mathcal{C}_2, \mathcal{D}_2, \mathcal{F}_2 \mapsto^{*\Theta';\Upsilon_1}_{\Delta \Vdash K} \mathcal{C}_2, \mathcal{D}'_2, \mathcal{F}_2$, and $\forall \mathcal{C}'_1, \mathcal{F}'_1, \mathcal{C}'_2, \mathcal{F}'_2$.if $\mathcal{C}_1, \mathcal{D}'_1, \mathcal{F}_1 \mapsto^{*\Theta_2;\Upsilon_1}_{\Delta \Vdash K} \mathcal{C}'_1, \mathcal{D}'_1, \mathcal{F}'_1$ and $\mathcal{C}_2, \mathcal{D}'_2, \mathcal{F}_2 \mapsto^{*\Theta_2;\Upsilon_1}_{\Delta \Vdash K} \mathcal{C}'_2, \mathcal{D}'_2, \mathcal{F}'_2$, then
$\forall x \in \Upsilon_1. (\mathcal{C}'_1, \mathcal{D}'_1, \mathcal{F}'_1; \mathcal{C}'_2, \mathcal{D}'_2, \mathcal{F}'_2) \in \mathcal{V}_\Psi^\xi[\![\Delta \Vdash K]\!]_{\cdot; x}^{m+1}$ and
$\forall x \in (\Theta_1 \cap \Theta_2). (\mathcal{C}'_1, \mathcal{D}'_1, \mathcal{F}'_1; \mathcal{C}'_2, \mathcal{D}'_2, \mathcal{F}'_2) \in \mathcal{V}_\Psi^\xi[\![\Delta \Vdash K]\!]_{x;\cdot}^{m+1}$.

and the proof of this subcase is complete.

**Subcase 3.** $u_\alpha^c$:&$\{\ell:A_\ell\}_{\ell \in L}$ and there is a message in $\mathcal{A}_1^\kappa$. The proof is similar to the previous case.

**Subcase 4.** $u_\alpha^c$:$A \otimes B$ and there is a message and a tree in $\mathcal{A}_1^\ell$, i.e. $\mathcal{A}_1^\ell = \mathcal{A}_1^{\ell'}\mathcal{H}_1\mathbf{msg}(\mathbf{send}x_\delta^c u_\alpha^c; u_\alpha^c \leftarrow y_\beta^c)$, and $\mathcal{D}_1 = \mathcal{D}_1^1\mathcal{A}_1^{\ell'}\mathcal{H}_1\mathbf{msg}(\mathbf{send}x_\delta^c u_\alpha^c; u_\alpha^c \leftarrow y_\beta^c)\mathcal{A}_1^\kappa\mathcal{D}_1^2$ and

$$\mathcal{C}_1, \mathcal{D}_1^1 \mathcal{A}_1^{\ell'}\mathcal{H}_1\mathbf{msg}(\mathbf{send}x_\delta^c u_\alpha^c; u_\alpha^c \leftarrow y_\beta^c)\mathcal{A}_1^\kappa\mathcal{D}_1^2, \mathcal{F}_1 \mapsto \mathcal{C}_1, \mathcal{D}_1^1\mathcal{A}_1^{\ell'}\mathcal{H}_1\mathcal{A}_1^{\kappa'}\mathcal{D}_1^2, \mathcal{F}_1 \mapsto_{\Delta \Vdash K}^{j'\Theta;\Upsilon_1} \mathcal{C}_1, \mathcal{D}_1', \mathcal{F}_1.$$

By $\dagger_1^\ell$ and $\mathcal{C}_1^{\ell_1}\mathcal{B}_1^\ell, \mathcal{A}_1^\ell, \mathcal{C}_1^{\ell_2}\mathcal{T}_1^\ell\mathcal{F}_1 \mapsto_{\Delta'_\ell \Vdash K'_\ell}^{*\Theta^{\ell'};\Upsilon_1^\ell} \mathcal{C}_1^{\ell_1}\mathcal{B}_1^\ell, \mathcal{A}_1^\ell, \mathcal{C}_1^{\ell_2}\mathcal{T}_1^\ell\mathcal{F}_1$, we know that

$\star_\ell \; \forall m. \exists \mathcal{A}_2^{\ell'}. \mathcal{C}_2^{\ell_1}\mathcal{B}_2^\ell, \mathcal{A}_2^\ell, \mathcal{C}_2^{\ell_2}\mathcal{T}_2^\ell\mathcal{F}_2 \mapsto_{\Delta'_\ell \Vdash K'_\ell}^{*\Theta^{\ell'};\Upsilon_1^\ell} \mathcal{C}_2^{\ell_1}\mathcal{B}_2^\ell, \mathcal{A}_2^{\ell'}, \mathcal{C}_2^{\ell_2}\mathcal{T}_2^\ell\mathcal{F}_2$, and

$\forall \mathcal{C}_1', \mathcal{F}_1', \mathcal{C}_2', \mathcal{F}_2'. \text{if } \mathcal{C}_1^{\ell_1}\mathcal{B}_1^\ell, \mathcal{A}_1^{\ell'}\mathcal{H}_1\mathbf{msg}(\mathbf{send}x_\delta^c u_\alpha^c; u_\alpha^c \leftarrow y_\beta^c), \mathcal{C}_1^{\ell_2}\mathcal{T}_1^\ell\mathcal{F}_1 \mapsto_{\Delta\Vdash K}^{*\Theta_1^\ell;\Upsilon_1^\ell} \mathcal{C}_1', \mathcal{A}_1^{\ell'}\mathcal{H}_1\mathbf{msg}(\mathbf{send}x_\delta^c u_\alpha^c; u_\alpha^c \leftarrow y_\beta^c), \mathcal{F}_1'$ and
$\quad\quad\quad \mathcal{C}_2^{\ell_1}\mathcal{B}_2^\ell, \mathcal{A}_2^{\ell'}, \mathcal{C}_2^{\ell_2}\mathcal{T}_2^\ell\mathcal{F}_2 \mapsto_{\Delta\Vdash K}^{*\Theta_2^\ell;\Upsilon_1^\ell} \mathcal{C}_2', \mathcal{A}_2^{\ell'}, \mathcal{F}_2'$, then
$\quad \forall x \in \Upsilon_1^\ell. (\mathcal{C}_1', \mathcal{A}_1^{\ell'}\mathcal{H}_1\mathbf{msg}(\mathbf{send}x_\delta^c u_\alpha^c; u_\alpha^c \leftarrow y_\beta^c), \mathcal{F}_1'; \mathcal{C}_2', \mathcal{A}_2^{\ell'}, \mathcal{F}_2') \in \mathcal{V}_\Psi^\xi[\![\Delta'_\ell \Vdash K'_\ell]\!]_{\cdot;x}^{m+1}$ and
$\quad \forall x \in (\Theta_1^\ell \cap \Theta_2^\ell). (\mathcal{C}_1', \mathcal{A}_1^{\ell'}\mathcal{H}_1\mathbf{msg}(\mathbf{send}x_\delta^c u_\alpha^c; u_\alpha^c \leftarrow y_\beta^c), \mathcal{F}_1'; \mathcal{C}_2', \mathcal{A}_2^{\ell'}, \mathcal{F}_2') \in \mathcal{V}_\Psi^\xi[\![\Delta'_\ell \Vdash K'_\ell]\!]_{x;\cdot}^{m+1}$.

By Lemma C.8, we get

$\star_\ell' \; \exists \mathcal{A}_2^{\ell'}. \mathcal{C}_2^{\ell_1}\mathcal{B}_2^\ell, \mathcal{A}_2^\ell, \mathcal{C}_2^{\ell_2}\mathcal{T}_2^\ell\mathcal{F}_2 \mapsto_{\Delta'_\ell \Vdash K'_\ell}^{*\Theta';\Upsilon_1^\ell} \mathcal{C}_2^{\ell_1}\mathcal{B}_2^\ell, \mathcal{A}_2^{\ell'}, \mathcal{C}_2^{\ell_2}\mathcal{T}_2^\ell\mathcal{F}_2$, and

$\forall \mathcal{C}_1', \mathcal{F}_1', \mathcal{C}_2', \mathcal{F}_2'. \text{if } \mathcal{C}_1^{\ell_1}\mathcal{B}_1^\ell, \mathcal{A}_1^{\ell'}\mathcal{H}_1\mathbf{msg}(\mathbf{send}x_\delta^c u_\alpha^c; u_\alpha^c \leftarrow y_\beta^c), \mathcal{C}_1^{\ell_2}\mathcal{T}_1^\ell\mathcal{F}_1 \mapsto_{\Delta\Vdash K}^{*\Theta_1^\ell;\Upsilon_1^\ell} \mathcal{C}_1', \mathcal{A}_1^{\ell'}\mathcal{H}_1\mathbf{msg}(\mathbf{send}x_\delta^c u_\alpha^c; u_\alpha^c \leftarrow y_\beta^c), \mathcal{F}_1'$ and
$\quad\quad\quad \mathcal{C}_2^{\ell_1}\mathcal{B}_2^\ell, \mathcal{A}_2^{\ell'}, \mathcal{C}_2^{\ell_2}\mathcal{T}_2^\ell\mathcal{F}_2 \mapsto_{\Delta\Vdash K}^{*\Theta_2^\ell;\Upsilon_1^\ell} \mathcal{C}_2', \mathcal{A}_2^{\ell'}, \mathcal{F}_2'$, then
$\quad \forall x \in \Upsilon_1^\ell. \forall m. (\mathcal{C}_1', \mathcal{A}_1^{\ell'}\mathcal{H}_1\mathbf{msg}(\mathbf{send}x_\delta^c u_\alpha^c; u_\alpha^c \leftarrow y_\beta^c), \mathcal{F}_1'; \mathcal{C}_2', \mathcal{A}_2^{\ell'}, \mathcal{F}_2') \in \mathcal{V}_\Psi^\xi[\![\Delta'_\ell \Vdash K'_\ell]\!]_{\cdot;x}^{m+1}$ and
$\quad \forall x \in (\Theta_1^\ell \cap \Theta_2^\ell). \forall m. (\mathcal{C}_1', \mathcal{A}_1^{\ell'}\mathcal{H}_1\mathbf{msg}(\mathbf{send}x_\delta^c u_\alpha^c; u_\alpha^c \leftarrow y_\beta^c), \mathcal{F}_1'; \mathcal{C}_2', \mathcal{A}_2^{\ell'}, \mathcal{F}_2') \in \mathcal{V}_\Psi^\xi[\![\Delta'_\ell \Vdash K'_\ell]\!]_{x;\cdot}^{m+1}$.

We pick that $\mathcal{A}_2^{\ell'}$ and apply existential elimination.

With a similar reasoning on $\dagger_1^\kappa$ and $\mathcal{C}_1^{\kappa_1}\mathcal{B}_1^\kappa, \mathcal{A}_1^\kappa, \mathcal{C}_1^{\kappa_2}\mathcal{T}_1^\kappa\mathcal{F}_1 \mapsto_{\Delta'_\kappa \Vdash K'_\kappa}^{*\Theta^{\kappa'};\Upsilon_1^\kappa} \mathcal{C}_1^{\kappa_1}\mathcal{B}_1^\kappa, \mathcal{A}_1^\kappa, \mathcal{C}_1^{\kappa_2}\mathcal{T}_1^\kappa\mathcal{F}_1$, we get

$\star_\kappa' \; \exists \mathcal{A}_2^{\kappa'}. \mathcal{C}_2^{\kappa_1}\mathcal{B}_2^\kappa, \mathcal{A}_2^\kappa, \mathcal{C}_2^{\kappa_2}\mathcal{T}_2^\kappa\mathcal{F}_2 \mapsto_{\Delta'_\kappa \Vdash K'_\kappa}^{*\Theta^{\kappa'};\Upsilon_1^\kappa} \mathcal{C}_2^{\kappa_1}\mathcal{B}_2^\kappa, \mathcal{A}_2^{\kappa'}, \mathcal{C}_2^{\kappa_2}\mathcal{T}_2^\kappa\mathcal{F}_2$, and

$\forall \mathcal{C}_1', \mathcal{F}_1', \mathcal{C}_2', \mathcal{F}_2'. \text{if } \mathcal{C}_1^{\kappa_1}\mathcal{B}_1^\kappa, \mathcal{A}_1^\kappa, \mathcal{C}_1^{\kappa_2}\mathcal{T}_1^\kappa\mathcal{F}_1 \mapsto_{\Delta'_\kappa \Vdash K'_\kappa}^{*\Theta_1^\kappa;\Upsilon_1^\kappa} \mathcal{C}_1', \mathcal{A}_1^\kappa, \mathcal{F}_1'$ and $\mathcal{C}_2^{\kappa_1}\mathcal{B}_2^\kappa, \mathcal{A}_2^{\kappa'}, \mathcal{C}_2^{\kappa_2}\mathcal{T}_2^\kappa\mathcal{F}_2 \mapsto_{\Delta'_\kappa \Vdash K'_\kappa}^{*\Theta_2;\Upsilon_1} \mathcal{C}_2', \mathcal{A}_2^{\kappa'}, \mathcal{F}_2'$, then
$\quad \forall x \in \Upsilon_1^\kappa. \forall m. (\mathcal{C}_1', \mathcal{A}_1^\kappa, \mathcal{F}_1'; \mathcal{C}_2', \mathcal{A}_2^{\kappa'}, \mathcal{F}_2') \in \mathcal{V}_\Psi^\xi[\![\Delta''_\kappa, u_\alpha^c:A \otimes B \Vdash K'_\kappa]\!]_{\cdot;x}^{m+1}$ and
$\quad \forall x \in (\Theta_1^\kappa \cap \Theta_2^\kappa). \forall m. (\mathcal{C}_1', \mathcal{A}_1^\kappa, \mathcal{F}_1'; \mathcal{C}_2', \mathcal{A}_2^{\kappa'}, \mathcal{F}_2') \in \mathcal{V}_\Psi^\xi[\![\Delta''_\kappa, u_\alpha^c:A \otimes B \Vdash K'_\kappa]\!]_{x;\cdot}^{m+1}$.

We pick that $\mathcal{A}_2^{\kappa'}$ and apply existential elimination.

With these choices of $\mathcal{A}_2^{\ell'}$ and $\mathcal{A}_2^{\kappa'}$ fixed, we apply universal elimination on $\star_\ell'$ and $\star_\kappa'$. We instantiate the universal quantifiers in $\star_\ell'$ respectively with $\mathcal{C}_1^{\ell_1}\mathcal{B}_1^\ell$ (for $\mathcal{C}_1'$ or the lower substitution of the first run) and $\mathcal{C}_1^{\ell_2}\mathcal{T}_1^\ell\mathcal{F}_1$ (for $\mathcal{F}_1'$ or the upper substitution of the first run) and $\mathcal{C}_2^{\ell_1}\mathcal{B}_2^\ell$ (for $\mathcal{C}_2'$ or the lower substitution of the second run) and $\mathcal{C}_2^{\ell_2}\mathcal{A}_2^{\kappa'}\mathcal{T}_2^{\ell'}\mathcal{F}_2$ (for $\mathcal{F}_2'$ or the upper substitution of the second run).

We instantiate the universal quantifiers in $\star_\kappa'$ respectively with $\mathcal{C}_1^{\kappa_1}\mathcal{B}_1^{\kappa'}\mathcal{A}_1^{\ell'}\mathcal{H}_1\mathbf{msg}(\mathbf{send}x_\delta^c u_\alpha^c; u_\alpha^c \leftarrow y_\beta^c)$ (for $\mathcal{C}_1'$ or the lower substitution of the first run) and $\mathcal{C}_1^{\kappa_2}\mathcal{T}_1^\kappa\mathcal{F}_1$ (for $\mathcal{F}_1'$ or the upper substitution of the first run) and $\mathcal{C}_2^{\kappa_1}\mathcal{B}_2^{\kappa'}\mathcal{A}_2^\ell$ (for $\mathcal{C}_2'$ or the lower substitution of the second run) and $\mathcal{C}_2^{\kappa_2}\mathcal{T}_2^\kappa\mathcal{F}_2$ (for $\mathcal{F}_1'$ or the upper substitution of the second run).

$\star_\ell'' \; \forall x \in \Upsilon_1^\ell. \forall m. (\mathcal{C}_1^{\ell_1}\mathcal{B}_1^\ell, \mathcal{A}_1^{\ell'}\mathcal{H}_1\mathbf{msg}(\mathbf{send}x_\delta^c u_\alpha^c; u_\alpha^c \leftarrow y_\beta^c), \mathcal{C}_1^{\ell_2}\mathcal{A}_1^\kappa\mathcal{T}_1^{\ell'}\mathcal{F}_1; \mathcal{C}_2^{\ell_1}\mathcal{B}_2^\ell, \mathcal{A}_2^{\ell'}, \mathcal{C}_2^{\ell_2}\mathcal{A}_2^{\kappa'}\mathcal{T}_2^{\ell'}\mathcal{F}_2) \in \mathcal{V}_\Psi^\xi[\![\Delta'_\ell \Vdash u_\alpha^c:A \otimes B]\!]_{\cdot;x}^{m+1}$ and
$\quad \forall x \in (\Theta_1^\ell \cap \Theta_2^\ell). \forall m. (\mathcal{C}_1^{\ell_1}\mathcal{B}_1^\ell, \mathcal{A}_1^{\ell'}\mathcal{H}_1\mathbf{msg}(\mathbf{send}x_\delta^c u_\alpha^c; u_\alpha^c \leftarrow y_\beta^c), \mathcal{C}_1^{\ell_2}\mathcal{A}_1^\kappa\mathcal{T}_1^{\ell'}\mathcal{F}_1; \mathcal{C}_2^{\ell_1}\mathcal{B}_2^\ell, \mathcal{A}_2^{\ell'}, \mathcal{C}_2^{\ell_2}\mathcal{A}_2^{\kappa'}\mathcal{T}_2^{\ell'}\mathcal{F}_2) \in \mathcal{V}_\Psi^\xi[\![\Delta'_\ell \Vdash u_\alpha^c:A \otimes B]\!]_{x;\cdot}^{m+1}$.

$\star_\kappa'' \; \forall x \in \Upsilon_1^\kappa. \forall m. (\mathcal{C}_1^{\kappa_1}\mathcal{B}_1^\kappa \mathcal{A}_1^{\ell'}\mathcal{H}_1\mathbf{msg}(\mathbf{send}x_\delta^c u_\alpha^c; u_\alpha^c \leftarrow y_\beta^c), \mathcal{A}_1^\kappa, \mathcal{C}_1^{\kappa_2}\mathcal{T}_1^\kappa\mathcal{F}_1; \mathcal{C}_2^{\kappa_1}\mathcal{B}_2^{\kappa'}\mathcal{A}_2^\ell, \mathcal{A}_2^{\kappa'}, \mathcal{C}_2^{\kappa_2}\mathcal{T}_2^\kappa\mathcal{F}_2) \in \mathcal{V}_\Psi^\xi[\![\Delta''_\kappa, u_\alpha^c:A \otimes B \Vdash K'_\kappa]\!]_{\cdot;x}^{m+1}$ and
$\quad \forall x \in (\Theta_1^\kappa \cap \Theta_2^\kappa). \forall m. (\mathcal{C}_1^{\kappa_1}\mathcal{B}_1^\kappa \mathcal{A}_1^{\ell'}\mathcal{H}_1\mathbf{msg}(\mathbf{send}x_\delta^c u_\alpha^c; u_\alpha^c \leftarrow y_\beta^c), \mathcal{A}_1^\kappa, \mathcal{C}_1^{\kappa_2}\mathcal{T}_1^\kappa\mathcal{F}_1; \mathcal{C}_2^{\kappa_1}\mathcal{B}_2^{\kappa'}\mathcal{A}_2^\ell, \mathcal{A}_2^{\kappa'}, \mathcal{C}_2^{\kappa_2}\mathcal{T}_2^\kappa\mathcal{F}_2) \in \mathcal{V}_\Psi^\xi[\![\Delta''_\kappa, u_\alpha^c:A \otimes B \Vdash K'_\kappa]\!]_{x;\cdot}^{m+1}$.

In particular, we have $u_\alpha^c \in \Upsilon_1^\ell$ and $u_\alpha^c \in \Theta_1^\kappa \cap \Theta_2^\kappa$. We instantiate $x$ with $u_\alpha^c$ and look at the conditions of the logical relation (Rows 4 and 9) for satisfying

$\star_\ell'' \; \forall m. (\mathcal{C}_1^{\ell_1}\mathcal{B}_1^\ell, \mathcal{A}_1^{\ell'}\mathcal{H}_1\mathbf{msg}(\mathbf{send}x_\delta^c u_\alpha^c; u_\alpha^c \leftarrow y_\beta^c), \mathcal{C}_1^{\ell_2}\mathcal{A}_1^\kappa\mathcal{T}_1^{\ell'}\mathcal{F}_1; \mathcal{C}_2^{\ell_1}\mathcal{B}_2^\ell, \mathcal{A}_2^{\ell'}, \mathcal{C}_2^{\ell_2}\mathcal{A}_2^{\kappa'}\mathcal{T}_2^{\ell'}\mathcal{F}_2) \in \mathcal{V}_\Psi^\xi[\![\Delta'_\ell \Vdash u_\alpha^c:A \otimes B]\!]_{\cdot;u_\alpha^c}^{m+1}$

$\star_\kappa'' \; \forall m. (\mathcal{C}_1^{\kappa_1}\mathcal{B}_1^\kappa \mathcal{A}_1^{\ell'}\mathcal{H}_1\mathbf{msg}(\mathbf{send}x_\delta^c u_\alpha^c; u_\alpha^c \leftarrow y_\beta^c), \mathcal{A}_1^\kappa, \mathcal{C}_1^{\kappa_2}\mathcal{T}_1^\kappa\mathcal{F}_1; \mathcal{C}_2^{\kappa_1}\mathcal{B}_2^{\kappa'}\mathcal{A}_2^{\ell'}, \mathcal{A}_2^{\kappa'}, \mathcal{C}_2^{\kappa_2}\mathcal{T}_2^\kappa\mathcal{F}_2) \in \mathcal{V}_\Psi^\xi[\![\Delta''_\kappa, u_\alpha^c:A \otimes B \Vdash K'_\kappa]\!]_{u_\alpha^c;\cdot}^{m+1}$.

By the conditions on row 4, we know $\mathcal{A}_2^{\ell'} = \mathcal{A}_2^{\ell''}\mathcal{H}_2\mathbf{msg}(\mathbf{send}\, x_\delta^c\, u_\alpha^c; u_\alpha^c \leftarrow y_\beta^c)$. Moreover $\mathcal{B}_i^\ell = \mathcal{B}_i^{\mathcal{H}}\mathcal{B}_i^{\ell'}$ and $\mathcal{C}_i^{\ell_1} = \mathcal{C}_i^{\mathcal{H}}\mathcal{C}_i^{\ell'_1}$ and $\Delta'_\ell = \Delta'_{\ell'}, \Delta'_{\mathcal{H}}$ such that

$\dagger'_\ell\ \forall m.(\mathcal{C}_1^{\ell'_1}\mathcal{B}_1^{\ell'}, \mathcal{A}_1^{\ell'}, \mathcal{C}_1^{\mathcal{H}}\mathcal{B}_1^{\mathcal{H}}\mathcal{C}_1^{\ell_2}\mathcal{H}_1\mathbf{msg}(\mathbf{send}\, x_\delta^c\, u_\alpha^c; u_\alpha^c \leftarrow y_\beta^c)\mathcal{A}_1^\kappa \mathcal{T}_1^{\ell'}\mathcal{F}_1; \mathcal{C}_2^{\ell'_1}\mathcal{B}_2^{\ell'}, \mathcal{A}_2^{\ell''}, \mathcal{C}_2^{\mathcal{H}}\mathcal{B}_2^{\mathcal{H}}\mathcal{C}_2^{\ell_2}\mathcal{H}_2\mathbf{msg}(\mathbf{send}\, x_\delta^c\, u_\alpha^c; u_\alpha^c \leftarrow y_\beta^c)\mathcal{A}_2^{\kappa'}\mathcal{T}_2^{\ell'}\mathcal{F}_2) \in \mathcal{E}_\Psi^\xi[\![\Delta'_{\ell'} \Vdash y_\beta^c{:}B]\!]^m$

$\dagger'_{\mathcal{H}}\ \forall m.(\mathcal{C}_1^{\mathcal{H}}\mathcal{B}_1^{\mathcal{H}}, \mathcal{H}_1, \mathcal{C}_1^{\ell'_1}\mathcal{B}_1^{\ell'}\mathcal{C}_1^{\ell_2}\mathcal{A}_1^{\ell'}\mathbf{msg}(\mathbf{send}\, x_\delta^c\, u_\alpha^c; u_\alpha^c \leftarrow y_\beta^c)\mathcal{A}_1^\kappa\mathcal{T}_1^{\ell'}\mathcal{F}_1; \mathcal{C}_2^{\mathcal{H}}\mathcal{B}_2^{\mathcal{H}}, \mathcal{H}_2, \mathcal{C}_2^{\ell'_1}\mathcal{B}_2^{\ell'}\mathcal{C}_2^{\ell_2}\mathcal{A}_2^{\ell''}\mathbf{msg}(\mathbf{send}\, x_\delta^c\, u_\alpha^c; u_\alpha^c \leftarrow y_\beta^c)\mathcal{A}_2^{\kappa'}\mathcal{T}_2^{\ell'}\mathcal{F}_2) \in \mathcal{E}_\Psi^\xi[\![\Delta'_{\mathcal{H}} \Vdash x_\alpha^c{:}A]\!]^m$

By the conditions on row 9, and the knowledge that $\mathcal{A}_2^{\ell'} = \mathcal{A}_2^{\ell''}\mathcal{H}_2\mathbf{msg}(\mathbf{send}\, x_\delta^c\, u_\alpha^c; u_\alpha^c \leftarrow y_\beta^c)$, we get

$\dagger'_\kappa\ \forall m.(\mathcal{C}_1^{\kappa_1}\mathcal{B}_1^{\kappa'}\mathcal{A}_1^{\ell'}\mathcal{H}_1, \mathbf{msg}(\mathbf{send}\, x_\delta^c\, u_\alpha^c; u_\alpha^c \leftarrow y_\beta^c)\mathcal{A}_1^\kappa, \mathcal{C}_1^{\kappa_2}\mathcal{T}_1^\kappa\mathcal{F}_1; \mathcal{C}_2^{\kappa_1}\mathcal{B}_2^{\kappa'}\mathcal{A}_2^{\ell''}\mathcal{H}_2, \mathbf{msg}(\mathbf{send}\, x_\delta^c\, u_\alpha^c; u_\alpha^c \leftarrow y_\beta^c)\mathcal{A}_2^{\kappa'}, \mathcal{C}_2^{\kappa_2}\mathcal{T}_2^\kappa\mathcal{F}_2) \in \mathcal{E}_\Psi^\xi[\![\Delta''_\kappa, x_\delta^c{:}A, y_\beta^c{:}B \Vdash K'_\kappa]\!]^m$

By applying forward closure (Lemma C.11) on $\dagger'_\ell$ and $\dagger'_{\mathcal{H}}$ we get

$\dagger''_\ell\ \forall m.(\mathcal{C}_1^{\ell'_1}\mathcal{B}_1^{\ell'}, \mathcal{A}_1^{\ell'}, \mathcal{C}_1^{\mathcal{H}}\mathcal{B}_1^{\mathcal{H}}\mathcal{C}_1^{\ell_2}\mathcal{H}_1\mathcal{A}_1^{\kappa'}\mathcal{T}_1^{\ell'}\mathcal{F}_1; \mathcal{C}_2^{\ell'_1}\mathcal{B}_2^{\ell'}, \mathcal{A}_2^{\ell''}, \mathcal{C}_2^{\mathcal{H}}\mathcal{B}_2^{\mathcal{H}}\mathcal{C}_2^{\ell_2}\mathcal{H}_2\mathbf{msg}(\mathbf{send}\, x_\delta^c\, u_\alpha^c; u_\alpha^c \leftarrow y_\beta^c)\mathcal{A}_2^{\kappa'}\mathcal{T}_2^{\ell'}\mathcal{F}_2) \in \mathcal{E}_\Psi^\xi[\![\Delta'_{\ell'} \Vdash y_\beta^c{:}B]\!]^m$

$\dagger''_{\mathcal{H}}\ \forall m.(\mathcal{C}_1^{\mathcal{H}}\mathcal{B}_1^{\mathcal{H}}, \mathcal{H}_1, \mathcal{C}_1^{\ell'_1}\mathcal{B}_1^{\ell'}\mathcal{C}_1^{\ell_2}\mathcal{A}_1^{\ell'}\mathcal{A}_1^{\kappa'}\mathcal{T}_1^{\ell'}\mathcal{F}_1; \mathcal{C}_2^{\mathcal{H}}\mathcal{B}_2^{\mathcal{H}}, \mathcal{H}_2, \mathcal{C}_2^{\ell'_1}\mathcal{B}_2^{\ell'}\mathcal{C}_2^{\ell_2}\mathcal{A}_2^{\ell''}\mathbf{msg}(\mathbf{send}\, x_\delta^c\, u_\alpha^c; u_\alpha^c \leftarrow y_\beta^c)\mathcal{A}_2^{\kappa'}\mathcal{T}_2^{\ell'}\mathcal{F}_2) \in \mathcal{E}_\Psi^\xi[\![\Delta'_{\mathcal{H}} \Vdash x_\alpha^c{:}A]\!]^m$

By applying forward closure (Lemma C.12) on $\dagger'_\kappa$, we get

$\dagger''_\kappa\ \forall m.(\mathcal{C}_1^{\kappa_1}\mathcal{B}_1^{\kappa'}\mathcal{A}_1^{\ell'}\mathcal{H}_1, \mathcal{A}_1^{\kappa'}, \mathcal{C}_1^{\kappa_2}\mathcal{T}_1^\kappa\mathcal{F}_1; \mathcal{C}_2^{\kappa_1}\mathcal{B}_2^{\kappa'}\mathcal{A}_2^{\ell''}\mathcal{H}_2, \mathbf{msg}(\mathbf{send}\, x_\delta^c\, u_\alpha^c; u_\alpha^c \leftarrow y_\beta^c)\mathcal{A}_2^{\kappa'}, \mathcal{C}_2^{\kappa_2}\mathcal{T}_2^\kappa\mathcal{F}_2) \in \mathcal{E}_\Psi^\xi[\![\Delta''_\kappa, x_\delta^c{:}A, y_\beta^c{:}B \Vdash K'_\kappa]\!]^m$.

By forward closure, for any other $\iota \neq \kappa, \ell$ we get

$\dagger_1^{\iota'}\ \forall m.(\mathcal{C}_1^{\iota_1}\mathcal{B}_1^{\iota'}\mathcal{A}_1^\iota\mathcal{C}_1^{\iota_2}\mathcal{T}_1^{\iota'}\mathcal{F}_1; \mathcal{C}_2^{\iota_1}\mathcal{B}_2^\iota\mathcal{A}_2^\iota\mathcal{C}_2^{\iota_2}\mathcal{T}_2^\iota\mathcal{F}_2) \in \mathcal{E}_{\Psi_0}^\xi[\![\Delta'_\iota \Vdash K'_\iota]\!]^m$,

Define $\mathcal{D}''_2$ as $\{\mathcal{A}_2^{i''}\}_{i \leq n+1}$ where for $\iota \neq \kappa, \ell \leq n$, we put $\mathcal{A}_2^{\iota''} = \mathcal{A}_2^\iota$, and $\mathcal{A}_2^{\kappa''} = \mathcal{A}_2^{\kappa'}\mathbf{msg}(\mathbf{send}\, x_\delta^c\, u_\alpha^c; u_\alpha^c \leftarrow y_\beta^c)$ and $\mathcal{A}_2^{n+1''} = \mathcal{A}_2^{\mathcal{H}}$. In other words, we break down the tree $\mathcal{A}_2^\ell$ to tree parts: $\mathcal{A}_2^{\ell''}$, $\mathcal{H}_2$ and the message $\mathbf{msg}(\mathbf{send}\, x_\delta^c\, u_\alpha^c; u_\alpha^c \leftarrow y_\beta^c)$. The two trees $\mathcal{A}_2^{\ell''}$ and $\mathcal{H}$ stand alone on their own, while the message $\mathbf{msg}(\mathbf{send}\, x_\delta^c\, u_\alpha^c; u_\alpha^c \leftarrow y_\beta^c)$ becomes the property of $\mathcal{A}_2^{k'}$. From the above discussion, we know that $\mathcal{D}_2 \mapsto^* \mathcal{D}''_2$. By induction on the number of steps $j'$, we get

$\exists \mathcal{D}'_2.\, \mathcal{C}_2, \mathcal{D}''_2, \mathcal{F}_2 \mapsto_{\Delta \Vdash K}^{*\Theta';\Upsilon_1} \mathcal{C}_2, \mathcal{D}'_2, \mathcal{F}_2$, and $\forall \mathcal{C}'_1, \mathcal{F}'_1, \mathcal{C}'_2, \mathcal{F}'_2.$if $\mathcal{C}_1, \mathcal{D}'_1, \mathcal{F}_1 \mapsto_{\Delta \Vdash K}^{*\Theta_2;\Upsilon_1} \mathcal{C}'_1, \mathcal{D}'_1, \mathcal{F}'_1$ and $\mathcal{C}_2, \mathcal{D}'_2, \mathcal{F}_2 \mapsto_{\Delta \Vdash K}^{*\Theta_2;\Upsilon_1} \mathcal{C}'_2, \mathcal{D}'_2, \mathcal{F}'_2$, then
$\forall x \in \Upsilon_1.\, (\mathcal{C}'_1, \mathcal{D}'_1, \mathcal{F}'_1; \mathcal{C}'_2, \mathcal{D}'_2, \mathcal{F}'_2) \in \mathcal{V}_\Psi^\xi[\![\Delta \Vdash K]\!]_{:;x}^{m+1}$ and
$\forall x \in (\Theta_1 \cap \Theta_2).\, (\mathcal{C}'_1, \mathcal{D}'_1, \mathcal{F}'_1; \mathcal{C}'_2, \mathcal{D}'_2, \mathcal{F}'_2) \in \mathcal{V}_\Psi^\xi[\![\Delta \Vdash K]\!]_{x;:}^{m+1}$.

Now by the knowledge that $\mathcal{C}_2, \mathcal{D}_2, \mathcal{F}_2 \mapsto_{\Delta \Vdash K}^{*\Theta;\Upsilon_1} \mathcal{C}_2, \mathcal{D}''_2, \mathcal{F}_2$ the proof of this subcase is complete.

**Subcase 5.** $u_\alpha^c{:}A \multimap B$ and there is a message and a tree in $\mathcal{A}_1^\kappa$. The proof is similar to the previous case.
**Subcase 6.** $u_\alpha^c{:}\mathbf{1}$ and there is a message in $\mathcal{A}_1^\ell$. The proof is similar to the previous cases.

$\square$

### C.3 Equivalence

**Definition C.4.** *The relation* $(\Delta_1 \Vdash \mathcal{D}_1 :: x_\alpha{:}A_1[c_1]) \equiv_\xi^{\Psi_0} (\Delta_2 \Vdash \mathcal{D}_2 :: y_\beta{:}A_2[c_2])$ *is defined as* $\Psi_0; \Delta_1 \Vdash \mathcal{D}_1 :: x_\alpha{:}A_1[c_1]$ *and* $\Psi_0; \Delta_2 \Vdash \mathcal{D}_2 :: y_\beta{:}A_2[c_2]$ *and* $\Delta_1 \Downarrow \xi = \Delta_2 \Downarrow \xi = \Delta$ *and* $x_\alpha{:}A_1[c_1] \Downarrow \xi = y_\beta{:}A_2[c_2] \Downarrow \xi = K$ *and* $\mathcal{C}_1, \mathcal{C}_2, \mathcal{F}_1, \mathcal{F}_2$, *with* $\Psi_0; \cdot \Vdash \mathcal{C}_1 :: \Delta_1$ *and* $\Psi_0; \cdot \Vdash \mathcal{C}_2 :: \Delta_2$ *and* $\Psi_0; x_\alpha{:}A_1[c_1] \Vdash \mathcal{F}_1 :: \cdot$ *and* $\Psi_0; y_\beta{:}A_2[c_2] \Vdash \mathcal{F}_2 :: \cdot$, *and for all $m$ we have*

- $(\mathcal{C}_1\mathcal{D}_1\mathcal{F}_1, \mathcal{C}_2\mathcal{D}_2\mathcal{F}_2) \in \mathcal{E}_{\Psi_0}^\xi[\![\Delta \Vdash K]\!]^m$ *and*
- $(\mathcal{C}_2\mathcal{D}_2\mathcal{F}_2, \mathcal{C}_1\mathcal{D}_1\mathcal{F}_1)) \in \mathcal{E}_{\Psi_0}^\xi[\![\Delta \Vdash K]\!]^m$.

Consider $\Delta = \{d_j{:}T_j[c_j]\}_{j \leq n}$, and the configuration $\mathcal{D}_i$ being a forest consisting of $n$ trees $\mathcal{D}_i = \mathcal{B}_i^1 \cdots \mathcal{B}_i^n$ each offering along a channel in $\Delta$. In other words, $\cdot \Vdash \mathcal{B}_i^j :: d_j{:}T_j$. We generalize $\equiv_\xi^{\Psi_0}$ for forests $(\cdot \Vdash \mathcal{D}_1 :: \Delta) \equiv_\xi^{\Psi_0} (\cdot \Vdash \mathcal{D}_2 :: \Delta)$ as $\forall \ell.(\cdot \Vdash \mathcal{B}_1^\ell :: d_\ell{:}T_\ell) \equiv_\xi^{\Psi_0} (\cdot \Vdash \mathcal{B}_2^\ell :: d_\ell{:}T_\ell[c_\ell])$.

### C.3.1 Noninterference

**Definition C.5** (Quasi Running Secrecy). *In the configuration tree, the quasi running secrecy of a message or process is determined based on its running secrecy, its process term, and the running secrecy of its parent.*

- If the node is a process with a process term other than **recv** or **case**, then its quasi running secrecy is equal to its running secrecy.
- If the process term is of the form $\mathbf{case}\, y_\alpha^c(\ell \Rightarrow P_\ell)_{\ell \in L} @ d_1$ or $x \leftarrow \mathbf{recv}\, y_\alpha^c; P_x @ d_1$, then its quasi running secrecy is $d_1 \sqcup c$.
- If the node is a message of a negative type along channel $y_\alpha^c$, its quasi running secrecy is $c$.
- If the node is a message of a positive type along channel $y_\alpha^c$ and it has a parent with quasi running secrecy $d_1$, its quasi running secrecy is $d_1 \sqcup c$, otherwise its quasi-running secrecy is $c$.

*The quasi running secrecy can be determined by traversing the tree top to bottom.*

**Definition C.6** (Relevant nodes). *Consider configuration $\Delta \Vdash \mathcal{D} :: K$ and observer level $\xi$. A channel is relevant in $\mathcal{D}$ if 1) it is has a maximal secrecy level less than or equal to $\xi$, and 2) it is either an observable channel or it shares a process or message with quasi running secrecy less than or equal to $\xi$ with a relevant channel. (A channel shares a process with another channel if they are siblings or one is the parent of another.)*

*A relevant process or message has quasi running secrecy less than or equal to $\xi$ and at least one relevant channel. $\mathcal{D}\Downarrow\xi$ are the relevant processes and messages in $\mathcal{D}$. We write $\mathcal{D}_1\Downarrow\xi =_\xi \mathcal{D}_2\Downarrow\xi$ if they are identical up to renaming of channels with higher or incomparable secrecy than the observer.*

We can build the set of all relevant nodes in a configuration by traversing the tree bottom-up as explained in Section 4. Def. C.6 provides us with a local guide to identify whether a process is relevant or not.

**Lemma C.17.** *Consider $(\mathcal{C}_1, \mathcal{D}_1, \mathcal{F}_1; \mathcal{C}_2, \mathcal{D}_2, \mathcal{F}_2) \in \mathsf{Tree}_\Psi(\Delta \Vdash K)$. If $\mathcal{C}_1, \mathcal{D}_1, \mathcal{F}_1 \mapsto_{\Delta \Vdash K} \mathcal{C}_1', \mathcal{D}_1', \mathcal{F}_1'$ such that $\mathcal{D}_1\Downarrow\xi = \mathcal{D}_2\Downarrow\xi$, then for some $\mathcal{C}_2', \mathcal{D}_2', \mathcal{F}_2'$, we have $\mathcal{C}_2, \mathcal{D}_2, \mathcal{F}_2 \mapsto^{0,1}_{\Delta \Vdash K} \mathcal{C}_2, \mathcal{D}_2', \mathcal{F}_2$ such that $\mathcal{D}_1'\Downarrow\xi = \mathcal{D}_2'\Downarrow\xi$.*

*Proof.* The proof is by cases on the possible $\mapsto_{\Delta \Vdash K}$ steps. If the step takes place in $\mathcal{C}_1$ or $\mathcal{F}_1$, the proof trivially holds. We only consider the cases in which $\mathcal{D}$ steps. In each case we prove that either the step does not change relevancy of any process in $\mathcal{D}_1$ or we can step $\mathcal{D}_2$ such that the same change of relevancy occurs in it too.

**Case 1.** $\mathcal{D}_1 = \mathcal{D}_1'\mathbf{proc}(y_\alpha[c], y_\alpha^c.k; P @ d_1)\mathcal{D}_1''$ and

$$\mathcal{D}_1'\mathbf{proc}(y_\alpha[c], y_\alpha^c.k; P @ d_1)\mathcal{D}_1'' \mapsto_{\Delta \Vdash K} \mathcal{D}_1'\mathbf{proc}(y_{\alpha+1}[c], [y_{\alpha+1}^c/y_\alpha^c]P @ d_1)\mathbf{msg}(y_\alpha^c.k; y_\alpha^c \leftarrow y_{\alpha+1}^c)\mathcal{D}_1''$$

We consider subcases based on relevancy of process offering along $y_\alpha[c]$:

**Subcase 1.** $\mathbf{proc}(y_\alpha[c], y_\alpha^c.k; P @ d_1)$ **is not relevant.** By inversion on the typing rules $d_1 \sqsubseteq c$. By definition either $d_1 \not\sqsubseteq \xi$ or none of the channels connected to $P$ including its offering channel $y_\alpha^c$ are relevant. In both cases neither $\mathbf{proc}(y_{\alpha+1}[c], [y_{\alpha+1}^c/y_\alpha^c]P @ d_1)$, nor $\mathbf{msg}(y_\alpha^c.k; y_\alpha^c \leftarrow y_{\alpha+1}^c)$ are relevant in the post step. Note that from $d_1 \sqsubseteq c$ and $d_1 \not\sqsubseteq \xi$, we get $c \not\sqsubseteq \xi$. Channel $y_\alpha^c$ is not relevant in the pre-step, and both $y_\alpha^c$ and $y_{\alpha+1}^c$ are not relevant in pre-step and post-step configurations. Every not relevant resource of $\mathbf{proc}(y_\alpha^c, y_\alpha^c.k; P @ d_1)$ will remain irrelevant in the post-step too.

In this subcase, our goal is to show

$$\mathcal{D}_1'\mathbf{proc}(y_{\alpha+1}[c], [y_{\alpha+1}^c/y_\alpha^c]P @ d_1)\mathbf{msg}(y_\alpha^c.k; y_\alpha^c \leftarrow y_{\alpha+1}^c)\mathcal{D}_1''\Downarrow\xi =_\xi \mathcal{D}_1'\mathcal{D}_1''\Downarrow\xi =_\xi \mathcal{D}_1\Downarrow\xi =_\xi \mathcal{D}_2\Downarrow\xi.$$

To prove this we need two observations:

- Neither $\mathbf{proc}(y_{\alpha+1}[c], [y_{\alpha+1}^c/y_\alpha^c]P @ d_1)$ nor $\mathbf{msg}(y_\alpha^c.k; y_\alpha^c \leftarrow y_{\alpha+1}^c)$ are relevant and they will be dismissed by the projection. (As explained above.)
- Replacing $\mathbf{proc}(y_\alpha[c], y_\alpha^c.k; P @ d_1)$ with these two nodes, does not affect relevancy of the rest of processes in $\mathcal{D}_1'\mathcal{D}_1''$. Relevancy of processes in $\mathcal{D}_1''$ remains intact since $y_\alpha^c$ and $y_{\alpha+1}^c$ are irrelevant.
  The relevancy of processes in $\mathcal{D}_1'$ remains intact too as we replace their irrelevant root with another irrelevant process. However, we need to be careful about the changes in the quasi running secrecy of a process and their effect on its (grand)children. The quasi-running secrecy of the process offering along $y_{\alpha+1}^c$ may be higher or incomparable to $d_1$ based on the code of $P$ (if it starts with a **recv** or **case**). This is of significance only if $d_1 \sqsubseteq \xi$, and in the pre-step the process has a chain of positive messages along a channel with secrecy level lower than or equal to the observer level (e.g. $x : \_[d]$ where $d \sqsubseteq \xi$) as (grand)children. But by the assumption of the subcase, $x : \_[d]$ cannot be a relevant channel in the pre-step. Thus the chain of messages is not relevant in neither the pre-step and nor the post-step.

**Subcase 2.** $\mathbf{proc}(y_\alpha[c], y_\alpha^c.k; P @ d_1)$ **is relevant.** By assumption ($\mathcal{D}_1\Downarrow\xi =_\xi \mathcal{D}_2\Downarrow\xi$):

$$\mathcal{D}_2 = \mathcal{D}_2'\mathbf{proc}(z_\delta[c], z_\delta^c.k; P_{|\xi} @ d_1)\mathcal{D}_2'',$$

35

such that $z_\delta = y_\alpha$ if $c \sqsubseteq \xi$ and $P_{|\xi}$ is equal to $P$ modulo renaming of some channels with higher or incomparable secrecy to the observer. We have

$$\mathcal{C}_2, \mathcal{D}_2, \mathcal{F}_2 \mapsto_{\Delta \Vdash K} \mathcal{C}_2, \mathcal{D}_2' \mathbf{proc}(z_{\delta+1}[c], [z_{\delta+1}^c/z_\delta^c]P_{|\xi}@d_1)\mathbf{msg}(z_\delta^c.k; z_\delta^c \leftarrow z_{\delta+1}^c)\mathcal{D}_2'', \mathcal{F}_2$$

It is enough to show

$$\mathcal{D}_1' \mathbf{proc}(y_{\alpha+1}[c], [y_{\alpha+1}^c/y_\alpha^c]P@d_1)\mathbf{msg}(y_\alpha^c.k; y_{\alpha+1}^c \leftarrow y_\alpha^c)\mathcal{D}_1''\Downarrow\xi =_\xi \mathcal{D}_2' \mathbf{proc}(z_{\delta+1}[c], [z_{\delta+1}^c/z_\delta^c]P_{|\xi}@d_1)\mathbf{msg}(z_\delta^c.k; z_\delta^c \leftarrow z_{\delta+1}^c)\mathcal{D}_2''\Downarrow\xi.$$

If $c \not\sqsubseteq \xi$, then $\mathbf{msg}(y_\alpha^c.k; y_\alpha^c \leftarrow y_{\alpha+1}^c)$ and $\mathbf{msg}(z_\delta^c.k; z_\delta^c \leftarrow z_{\delta+1}^c)$ are not relevant in both runs, and will be dismissed by the projections. Moreover, in this case neither $y_\alpha^c$, nor $z_\delta^c$ are relevant in the pre-step and post-step configurations. Thus the relevancy of processes in $\mathcal{D}_1''$ and $\mathcal{D}_2''$ will remain intact.

If $c \sqsubseteq \xi$, then $z_\delta^c = y_\alpha^c$ are relevant in the pre-step in both runs. In the post-step, $z_{\delta+1}^c = y_{\alpha+1}^c$ are relevant in both runs. Relevancy of messages $\mathbf{msg}(y_\alpha^c.k; y_\alpha^c \leftarrow y_{\alpha+1}^c)$ and $z_\delta^c = y_\alpha^c$ in the post-steps are determined by the quasi running secrecies ($d$ and $d'$) of their parents ($X$ and $X'$) in $\mathcal{D}_1''$ and $\mathcal{D}_2''$. If $d \sqsubseteq \xi$, then the parent ($X$) is relevant in the first run and by assumption is equal to a relevant $X'$ in the second run. Thus messages $\mathbf{msg}(y_\alpha^c.k; y_\alpha^c \leftarrow y_{\alpha+1}^c)$ are relevant in both runs and $z_\delta^c = y_\alpha^c$ are relevant in the post-step too. The same holds when $d' \sqsubseteq \xi$.

Otherwise, in both runs the quasi running secrecy of the parent is higher than or incomparable to the observer (the parents are both irrelevant). Thus messages $\mathbf{msg}(y_\alpha^c.k; y_\alpha^c \leftarrow y_{\alpha+1}^c)$ are not relevant in the post step of both runs, and will be dismissed by the projections. The channels $z_\delta^c = y_\alpha^c$ will be irrelevant in the post-step too. However, this does not affect the processes in $\mathcal{D}_1''$ and $\mathcal{D}_2''$ as the parents of messages ($X$ and $X'$) are already irrelevant in the pre-step.

It remains to show that projections of $\mathcal{D}_1'$ and $\mathcal{D}_2'$ are equal in the post-step too. The resources with secrecies higher than or incomparable to the observer offered along $\mathcal{D}_1'$ and $\mathcal{D}_2'$ in the pre-step will remain higher than or incomparable to the observer and thus irrelevant in the post-step too. For a relevant resources ($w$) offered along $\mathcal{D}_1$ and $\mathcal{D}_2$, we need to consider the change in quasi running secrecy. The quasi running secrecy of the processes offering along $y_{\alpha+1}^c$ and $z_{\delta+1}^c$ may increase based on their code (if the code of $P$ and $P_{|\xi}$ starts with a $\mathbf{recv}$ or $\mathbf{case}$) and become irrelevant. This means that a relevant sub-tree $\mathcal{T}_i$ of $\mathcal{D}_i'$ offering along $w$ in the pre-step will become irrelevant in the post-step. But by the assumption of the theorem, we know that $\mathcal{T}_1 = \mathcal{T}_2$. In the post step, we project out the equal subtrees $\mathcal{T}_1$ and $\mathcal{T}_2$ from the relevant parts of $\mathcal{D}_1'$ and $\mathcal{D}_2'$; the projections will remain equal in the post-step.

**Case 2.** $\mathcal{D}_1 = \mathcal{D}_1' \mathbf{proc}(x_\beta[d], y_\alpha^c.k; P@d_1)\mathcal{D}_1''$ and

$$\mathcal{D}_1' \mathbf{proc}(x_\beta[d], y_\alpha^c.k; P@d_1)\mathcal{D}_1'' \mapsto_{\Delta \Vdash K} \mathcal{D}_1' \mathbf{msg}(y_\alpha^c.k; y_{\alpha+1}^c \leftarrow y_\alpha^c)\mathbf{proc}(x_\beta[d], [y_{\alpha+1}^c/y_\alpha^c]P@d_1)\mathcal{D}_1''$$

We consider subcases based on relevancy of the process offering along $x_d^\beta$:

**Subcase 1.** $\mathbf{proc}(x_\beta[d], y_\alpha^c.k; P@d_1)$ **is irrelevant.** By inversion on the typing rules, $d_1 \sqsubseteq c \sqsubseteq d$. By definition either $d_1 \not\sqsubseteq \xi$ or none of the channels connected to $P$ including $y_\alpha^c$ and $x_\beta^d$ are relevant. In both cases, neither $\mathbf{msg}(y_\alpha^c.k; y_{\alpha+1}^c \leftarrow y_\alpha^c)$ nor $\mathbf{proc}(x_\beta[d], [y_{\alpha+1}^c/y_\alpha^c]P@d_1)$ are relevant. Channel $x_\beta^d$ is irrelevant in the pre-step and post-step configurations. Channel $y_\alpha^c$ is irrelevant in the pre-step, and both $y_\alpha^c$ and $y_{\alpha+1}^c$ are irrelevant in pre-step and post-step configurations. Every other irrelevant resource of $\mathbf{proc}(x_\beta[d], y_\alpha^c.k; P@d_1)$ will remain irrelevant in the post-step too.

In this subcase, our goal is to show

$$\mathcal{D}_1' \mathbf{msg}(y_\alpha^c.k; y_{\alpha+1}^c \leftarrow y_\alpha^c)\mathbf{proc}(x_\beta[d], [y_{\alpha+1}^c/y_\alpha^c]P@d_1)\mathcal{D}_1''\Downarrow\xi =_\xi \mathcal{D}_1'\mathcal{D}_1''\Downarrow\xi =_\xi \mathcal{D}_1\Downarrow\xi =_\xi \mathcal{D}_2\Downarrow\xi$$

With a same argument as in **Case 1. Subcase 1.**, we can prove that the relevancy of processes in $\mathcal{D}_1'$ and $\mathcal{D}_1''$ remain intact.

**Subcase 2.** $\mathbf{proc}(x_\beta[d], y_\alpha^c.k; P@d_1)$ **is relevant.** By assumption that $\mathcal{D}_1\Downarrow\xi =_\xi \mathcal{D}_2\Downarrow\xi$, and definition of $=_\xi$:

$$\mathcal{D}_2 = \mathcal{D}_2' \mathbf{proc}(u_\gamma[d], z_\delta^c.k; P_{|\xi}@d_1)\mathcal{D}_2'',$$

such that $z_\delta = y_\alpha$ if $c \sqsubseteq \xi$ and $u_\gamma = x_\beta$ if $d \sqsubseteq \xi$ and $P_{|\xi}$ is equal to $P$ modulo renaming of some channels with higher than or incomparable to the observer level. We have

$$\mathcal{C}_2, \mathcal{D}_2, \mathcal{F}_2 \mapsto_{\Delta \Vdash K} \mathcal{C}_2, \mathcal{D}_2' \mathbf{msg}(z_\delta^c.k; z_{\delta+1}^c \leftarrow z_\delta^c)\mathbf{proc}(u_\gamma[d], [z_{\delta+1}^c/z_\delta^c]P_{|\xi}@d_1)\mathcal{D}_2'', \mathcal{F}_2$$

If $d \sqsubseteq \xi$, then $u_\gamma = x_\beta$ are relevant in the pre-steps of both runs and remain relevant in the post-steps. Even if the quasi running secrecy increases based on the code of $P$ and $P_{|\xi}$, it will be lower than or equal to the observer level by the tree invariant. Thus the relevancy of processes in $\mathcal{D}_i''$ remain intact. Moreover, every resource of the processes in $\mathcal{D}_i'$ is relevant in the pre-steps and post-steps.

36

If $d \not\sqsubseteq \xi$, then $u_\gamma$ and $x_\beta$ remain irrelevant in the pre-steps and post-steps of both runs: the relevancy of processes in $\mathcal{D}_i''$ remain intact.

It remains to show that the projections of $\mathcal{D}_1'$ and $\mathcal{D}_2'$ in post-steps are still equal. Here we only consider the trees offered along $y_\alpha^c$ and $z_\delta^c$ in both runs. The argument for the rest of $\mathcal{D}_i'$ is similar to **Case 1. Subcase 2.**

The adaptive running secrecy of the negative message $\mathbf{msg}(z_\delta^c.k; z_{\delta+1}^c \leftarrow z_\delta^c)$ is $c$ in the post-steps. If $c \sqsubseteq \xi$ then the same message exists in both runs, and the tree offered along $y_\alpha^c = z_\delta^c$ is relevant in the pre-steps and post-steps.

If $c \not\sqsubseteq \xi$ then the message is irrelevant in both runs. $y_\alpha^c$ and $z_\delta^c$ are both irrelevant in the pre-step and remain irrelevant in the post-step too. By typing rules $\mathbf{msg}(z_\delta^c.k; z_{\delta+1}^c \leftarrow z_\delta^c)$ is not a parent of any positive messages.

**Case 3.** $\mathcal{D}_1 = \mathcal{D}_1' \mathcal{C}_1'' \mathbf{proc}(y_\alpha^c, \mathbf{send}\, x_\beta^d\, y_\alpha^c @ d_1) \mathcal{D}_1''$ and

$$\mathcal{D}_1' \mathcal{C}_1'' \mathbf{proc}(y_\alpha^c, \mathbf{send}\, x_\beta^c\, y_\alpha^c; P @ d_1) \mathcal{D}_1'' \mapsto_{\Delta \Vdash K} \mathcal{D}_1' \mathcal{C}_1'' \mathbf{proc}(y_{\alpha+1}^c, [y_{\alpha+1}^c / y_\alpha^c] P @ d_1) \mathbf{msg}(\mathbf{send}\, x_\beta^c\, y_\alpha^c; y_\alpha^c \leftarrow y_{\alpha+1}^c) \mathcal{D}_1''$$

such that $\Delta = \Delta' \Delta''$ and $\Delta' \Vdash \mathcal{D}_1' :: \Lambda$ and $\Delta'' \Vdash \mathcal{C}_1'' :: (x_\beta:A[c])$ and $\Lambda, x_\beta:A[c] \Vdash \mathbf{proc}(y_\alpha[c], \mathbf{send}\, x_\beta^c\, y_\alpha^c; P @ d_1) :: (y_\alpha:A \otimes B[c])$. (In the case that $\mathcal{C}_1''$ is empty, we have $\Delta = \Delta', x_\beta:A[c]$.)

We consider subcases based on relevancy of process offering along $y_\alpha^c$:

**Subcase 1.** $\mathbf{proc}(y_c^\alpha, \mathbf{send}\, x_\beta^d\, y_\alpha^c; P @ d_1)$ **is not relevant.** By inversion on the typing rules $d_1 \sqsubseteq c$. By definition either $d_1 \not\sqsubseteq \xi$ or none of the channels connected to $P$ including $y_\alpha^c$ and $x_\beta^c$ are relevant. In both cases, neither $\mathbf{proc}(y_{\alpha+1}[c], [y_{\alpha+1}^c / y_\alpha^c] P @ d_1)$ nor $\mathbf{msg}(\mathbf{send}\, x_\beta^c\, y_\alpha^c; y_\alpha^c \leftarrow y_{\alpha+1}^c)$ are relevant. Channel $y_\alpha^c$ is irrelevant in the pre-step, and both $y_\alpha^c$ and $y_{\alpha+1}^c$ are irrelevant in pre-step and post-step configurations. In this subcase, our goal is to show

$$\mathcal{D}_1' \mathcal{C}_1'' \mathbf{proc}(y_{\alpha+1}[c], [y_{\alpha+1}^c / y_\alpha^c] P @ d_1) \mathbf{msg}(\mathbf{send}\, x_\beta^c\, y_\alpha^c; y_\alpha^c \leftarrow y_{\alpha+1}^c) \mathcal{D}_1'' \Downarrow \xi =_\xi \mathcal{D}_1' \mathcal{C}_1'' \mathcal{D}_1'' \Downarrow \xi =_\xi \mathcal{D}_1 \Downarrow \xi =_\xi \mathcal{D}_2 \Downarrow \xi.$$

We first prove that the relevancy status of $\mathcal{C}_1''$ remain intact too. Note that all channels in $\mathcal{C}_1''$, except $x_\beta^c$ have the same connections in the pre-step and post-step. So it is enough to consider the changes made to the tree rooted at $x_\beta^c$.

We show that the relevancy of $x_\beta^c$ remains intact after the step: If $c \not\sqsubseteq \xi$, then the quasi running secrecy of the message is higher than or incomparable to the observer and thus relevancy of $x_\beta^c$ remains intact in the post-step. If $c \sqsubseteq \xi$, then $d_1 \sqsubseteq \xi$, and thus both $x_c^\beta$ and $y_\alpha^c$ are irrelevant in the pre-step. In the post-step, the parent of the message has to be irrelevant, otherwise $y_\alpha^c$ would be relevant in the pre-step. Since the parent of the message is irrelevant, we know that $x_\beta^c$ remains irrelevant in the post-step. Thus relevancy status of $\mathcal{C}_1''$ remains intact. Moreover, the message is irrelevant. The tree rooted at $x_\beta^c$ offers to an irrelevant node before and after the step. Thus, the relevancy of tree rooted at $x_\beta^c$ is not affected. With a same argument as in **Case 1. Subcase 1.** and the one given for $\mathcal{C}_1''$, we can prove that the relevancy status of processes in $\mathcal{D}_1'$ and $\mathcal{D}_1''$ remain intact.

**Subcase 2.** $\mathbf{proc}(y_\alpha^c, \mathbf{send}\, x_\beta^c\, y_\alpha^c; P @ d_1)$ **is relevant.** By assumption that $\mathcal{D}_1 \Downarrow \xi =_\xi \mathcal{D}_2 \Downarrow \xi$, and definition of $=_\xi$:

$$\mathcal{D}_2 = \mathcal{D}_2' \mathcal{C}_2'' \mathbf{proc}(z_\delta[c], \mathbf{send}\, u_\gamma^c\, z_\delta^c; P_{|\xi} @ d_1) \mathcal{D}_2''$$

such that $z_\delta = y_\alpha$ and $u_\gamma = x_\beta$ if $c \sqsubseteq \xi$, and $P_{|\xi}$ is equal to $P$ modulo renaming of some channels with higher than or incomparable to the observer level.

We have

$$\mathcal{C}_2, \mathcal{D}_2, \mathcal{F}_2 \mapsto_{\Delta \Vdash K} \mathcal{C}_2, \mathcal{D}_2' \mathcal{C}_2'' \mathbf{proc}(z_{\delta+1}^c, [z_{\delta+1}^c / z_\delta^c] P_{|\xi} @ d_1) \mathbf{msg}(\mathbf{send}\, u_\gamma^d\, z_\delta^c; z_\delta^c \leftarrow z_{\delta+1}^c) \mathcal{D}_2'', \mathcal{F}_2$$

If $c \not\sqsubseteq \xi$, then $\mathbf{msg}(\mathbf{send}\, x_\beta^c\, y_\alpha^c; y_\alpha^c \leftarrow y_{\alpha+1}^c)$ and $\mathbf{msg}(\mathbf{send}\, u_\gamma^c\, z_\delta^c; z_\delta^c \leftarrow z_{\delta+1}^c)$ are not relevant in both runs, and will be dismissed by the projections. Moreover, neither $y_\alpha^c$, nor $z_\delta^c$ are relevant in the pre-step and post-step configurations. Thus the relevancy of processes in $\mathcal{D}_1''$ and $\mathcal{D}_2''$ will remain intact. Moreover, in this case $x_\beta^c$ and $u_\gamma^c$ are irrelevant in both pre-steps and post-steps. Which means that relevancy status of $\mathcal{C}_1''$ and $\mathcal{C}_2''$ remains intact.

If $c \sqsubseteq \xi$, then $z_\delta^c = y_\alpha^c$ are relevant in the pre-step in both runs. We also know that $x_\beta = u_\gamma$ are relevant in the pre-step and $\mathcal{C}_1'' = \mathcal{C}_2''$ are relevant in the pre-step. In the post-step, $z_{\delta+1}^c = y_{\alpha+1}^c$ are relevant in both runs. Relevancy of messages $\mathbf{msg}(\mathbf{send}\, x_\beta^c\, y_\alpha^c; y_\alpha^c \leftarrow y_{\alpha+1}^c)$ in the post-steps are determined by the quasi running secrecy ($d'$ and $d''$) of their parents ($X$ and $X'$) in $\mathcal{D}_1''$ and $\mathcal{D}_2''$. If $d' \sqsubseteq \xi$, then the parent ($X$) is relevant in the first run and by assumption is equal to a relevant $X'$ in the second run. Thus messages $\mathbf{msg}(\mathbf{send}\, x_\beta^c\, y_\alpha^c; y_\alpha^c \leftarrow y_{\alpha+1}^c)$ are relevant in both runs, $z_\delta^c = y_\alpha^c$ are relevant in the post-step, and trees $\mathcal{C}_1'' = \mathcal{C}_2''$ and their offering channels $x_\beta^c = u_\gamma^c$ are relevant in the post-steps too. The same holds when $d' \sqsubseteq \xi$.

Otherwise, in both runs the quasi running secrecy of the parent is higher than or incomparable to the observer level (the parents are both irrelevant). Thus messages $\mathbf{msg}(\mathbf{send}\, x_\beta^c\, y_\alpha^c; y_\alpha^c \leftarrow y_{\alpha+1}^c)$ are not relevant in the post step of both runs, and will be dismissed by the projections. The channels $z_\delta^c = y_\alpha^c$ will be irrelevant in the post-step too. However, this does not affect the processes in $\mathcal{D}_1''$ and $\mathcal{D}_2''$ as the parents of messages ($X$ and $X'$) are already irrelevant in the pre-step. The channels $x_\beta^c = u_\gamma^c$ both become irrelevant in the post-steps. However, we still have $\mathcal{C}_1'' \Downarrow \xi = \mathcal{C}_2'' \Downarrow \xi$ as they have the same type and their parents have the same quasi running secrecy.

We show that projections of $\mathcal{D}'_1$ and $\mathcal{D}'_2$ are equal in the post-step too. The resources with secrecy level higher than or incomparable to the observer level offered along $\mathcal{D}'_1$ and $\mathcal{D}'_2$ in the pre-step will remain higher than or incomparable to and thus irrelevant in the post-step too. For a relevant resources ($w$) offered along $\mathcal{D}_1$ and $\mathcal{D}_2$, we need to consider the change in quasi running secrecy as in **Case 1. Subcase 2.** Moreover, we need to consider the scenario that a relevant resource ($w^{c'}$) in the pre-step loses its relevancy in the post-step because the channel offered along $x_\beta^c$ is transferred to the message. This case only happens if $c', c \sqsubseteq \xi$ and thus the trees $\mathcal{T}_1$ and $\mathcal{T}_2$ offered along $w^{c'}$ is present in both runs and $\mathcal{T}_1 = \mathcal{T}_2$. We know that $w^{c'}$ is irrelevant in the post-step of both runs, and the quasi-running secrecy of the processes using the resource $w^{c'}$ in both runs are the same.

The relevant and irrelevant processes in $\mathcal{D}''_i$ remain intact.

**Case 4.** $\mathcal{D}_1 = \mathcal{D}'_1 \mathcal{C}''_1 \mathbf{proc}(w_\eta[c'], \mathbf{send}\, x_\beta^c\, y_\alpha^c @ d_1) \mathcal{D}''_1$ and

$$\mathcal{D}'_1 \mathcal{C}''_1 \mathbf{proc}(w_\eta[c'], \mathbf{send}\, x_\beta^c\, y_\alpha^c; P@d_1) \mathcal{D}''_1 \mapsto_{\Delta \Vdash K} \mathcal{D}'_1 \mathcal{C}''_1 \mathbf{msg}(\mathbf{send}\, x_\beta^c\, y_\alpha^c; y_{\alpha+1}^c \leftarrow y_\alpha^c) \mathbf{proc}(w_\eta[c'], [y_{\alpha+1}^c/y_\alpha^c]P@d_1) \mathcal{D}''_1$$

such that $\Delta = \Delta_1 \Delta_2$ and $\Delta_1 \Vdash \mathcal{D}'_1 :: \Lambda, y_\alpha:(A \multimap B)[c]$ and $\Delta_2 \Vdash \mathcal{C}''_1 :: (x_\beta:A[c])$ and

$$\Lambda, y_\alpha:(A \multimap B)[c], x_\beta:A[c] \Vdash \mathbf{proc}(w_\eta[c'], \mathbf{send}\, x_\beta^c\, y_\alpha^c; P@d_1) :: (w_\eta:C[c']).$$

In the case where $\mathcal{C}''_1$ is empty we have $\Delta_2 = x_\beta:A[c]$. We proceed by considering subcases based on relevancy of the process offering along $w_\eta[c']$:

**Subcase 1.** $\mathbf{proc}(w_\eta[c'], \mathbf{send}\, x_\beta^c\, y_\alpha^c; P@d_1)$ **is not relevant.** By inversion on the typing rules $d_1 \sqsubseteq c \sqsubseteq c'$. By definition either $d_1 \not\sqsubseteq \xi$ or none of the channels connected to $P$ including $y_\alpha^c$, and $x_\beta^c$ are relevant. In both cases, neither $\mathbf{msg}(\mathbf{send}\, x_\beta^c\, y_\alpha^c; y_{\alpha+1}^c \leftarrow y_\alpha^c)$ nor $\mathbf{proc}(w_\eta[c'], [y_{\alpha+1}^c/y_\alpha^c]P@d_1)$ are relevant. Channel $w_\eta^{c'}$ is irrelevant in the pre-step and post-step configurations.

Channel $y_\alpha^c$ is irrelevant in the pre-step, and both $y_\alpha^c$ and $y_{\alpha+1}^c$ are irrelevant in pre-step and post-step configurations. Every other irrelevant resource of the process in the pre-step will remain irrelevant in the post-step too. See **Case 3. Subcase 1.** for the discussion on the relevancy of $\mathcal{C}''_1$.

$$\mathcal{D}'_1 \mathcal{C}''_1 \mathbf{msg}(\mathbf{send}\, x_\beta^c\, y_\alpha^c; y_{\alpha+1}^c \leftarrow y_\alpha^c) \mathbf{proc}(w_\eta[c'], [y_{\alpha+1}^c/y_\alpha^c]P@d_1) \mathcal{D}''_1 \Downarrow \xi =_\xi \mathcal{D}'_1 \mathcal{D}''_1 \Downarrow \xi =_\xi \mathcal{D}_1 \Downarrow \xi =_\xi \mathcal{D}_2 \Downarrow \xi$$

**Subcase 2.** $\mathbf{proc}(w_\eta^{c'}, \mathbf{send}\, x_\beta^c\, y_\alpha^c; P@d_1)$ **is relevant.** By assumption that $\mathcal{D}_1 \Downarrow \xi =_\xi \mathcal{D}_2 \Downarrow \xi$, and definition of $=_\xi$:

$$\mathcal{D}_2 = \mathcal{D}'_2 \mathcal{C}''_2 \mathbf{proc}(v_\omega[c'], \mathbf{send}\, u_\gamma^c\, z_\delta^c; P_{|\xi} @ d_1) \mathcal{D}''_2$$

such that $z_\delta = y_\alpha$ and $u_\gamma = x_\beta$ if $c \sqsubseteq \xi$, and $P_{|\xi}$ is equal to $P$ modulo renaming of some channels with secrecy level higher than or incomparable secrecy level to the observer.

We have

$$\mathcal{C}_2, \mathcal{D}_2, \mathcal{F}_2 \mapsto_{\Delta \Vdash K} \mathcal{C}_2, \mathcal{D}'_2 \mathcal{C}''_2 \mathbf{msg}(\mathbf{send}\, u_\gamma^c\, z_\delta^c; z_{\delta+1}^c \leftarrow z_\delta^c) \mathbf{proc}(v_\omega[c'], [z_{\delta+1}^c/z_\delta^c]P_{|\xi} @ d_1) \mathcal{D}''_2, \mathcal{F}_2$$

With the same argument as in **Case 2. Subcase 2.** we can show that relevancy of $\mathcal{D}''_i$ remains intact.

For $\mathcal{C}''_i$, we argue that if $c \sqsubset \xi$, then $\mathcal{C}''_1 = \mathcal{C}''_2$ is relevant in the pre-step and remains relevant in the post-step too. If $c \not\sqsubseteq \xi$, then the relevancy of $\mathcal{C}''_i$ remain intact from pre-step to post-step. See **Case 3. Subcase 2.** for a more detailed discussion on transferring a tree via message.

The discussion on relevancy of $\mathcal{D}'_i$ is similar to the previous cases.

**Case 5.** $\mathcal{D}_1 = \mathcal{D}'_1 \mathbf{msg}(y_\alpha^c.k; y_\alpha^c \leftarrow v^c) \mathbf{proc}(x_\beta[d], \mathbf{case}\, y_\alpha^c(\ell \Rightarrow P_\ell)_{\ell \in L} @ d_1) \mathcal{D}''_1$ and

$$\mathcal{D}'_1 \mathbf{msg}(y_\alpha^c.k; y_\alpha^c \leftarrow v^c) \mathbf{proc}(x_\beta[d], \mathbf{case}\, y_\alpha^c(\ell \Rightarrow P_\ell)_{\ell \in L} @ d_1) \mathcal{D}''_1 \mapsto_{\Delta \Vdash K} \mathcal{D}'_1 \mathbf{proc}(w_\beta[d], [v^c/y_\alpha^c] P_k @ (d_1 \sqcup c)) \mathcal{D}''_1$$

We consider sub-cases based on relevancy of process offering along $w_\beta^d$. Observe that $y_\alpha^c$ is relevant if an only if $v^c$ is relevant, since they share a message of secrecy $c$.

**Subcase 1.** $\mathbf{proc}(x_\beta[d], \mathbf{case}\, y_\alpha^c(\ell \Rightarrow P_\ell)_{\ell \in L} @ d_1)$ **is not relevant.** By definition either $d_1 \sqcup c \not\sqsubseteq \xi$ or none of the channels connected to $P$ including its offering channel $y_\alpha^c$ are relevant. In both cases the messages $\mathbf{msg}(y_\alpha^c.k; y_\alpha^c \leftarrow v^c)$ and the continuation process $\mathbf{proc}(w_\beta[d], [v^c/y_\alpha^c] P_k @ (d_1 \sqcup c))$ are not relevant either. It is straightforward to see that

$$\mathcal{D}'_1 \mathbf{proc}(w_\beta[d], [v^c/y_\alpha^c] P_k @ (d_1 \sqcup c)) \mathcal{D}''_1 \Downarrow \xi =_\xi \mathcal{D}'_1 \mathcal{D}''_1 \Downarrow \xi =_\xi \mathcal{D}_1 \Downarrow \xi =_\xi \mathcal{D}_2 \Downarrow \xi.$$

**Subcase 2. $\mathbf{proc}(x_\beta[d], \mathbf{case}\, y_\alpha^c(\ell \Rightarrow P_\ell)_{\ell \in L}@d_1)$ is relevant.** By definition of relevancy, we get that $c \sqcup d_1 \sqsubseteq \xi$ and thus $y_\alpha^c$ is relevant. This means that $\mathbf{msg}(y_\alpha^c.k; y_\alpha^c \leftarrow v^c)$ is relevant too. By assumption that $\mathcal{D}_1 \Downarrow \xi =_\xi \mathcal{D}_2 \Downarrow \xi$, and definition of $=_\xi$:
$$\mathcal{D}_2 = \mathcal{D}_2' \mathbf{msg}(y_\alpha^c.k; y_\alpha^c \leftarrow v^c) \mathbf{proc}(u_\gamma[d], \mathbf{case}\, y_\alpha^c(\ell \Rightarrow P_{\ell|\xi})_{\ell \in L}@d_1) \mathcal{D}_2'',$$
such that $P_{\ell|\xi}$ is equal to $P_\ell$ modulo renaming of some channels with secrecy level higher than or incomparable to the observer.

We have
$$\mathcal{C}_2, \mathcal{D}_2, \mathcal{F}_2 \mapsto_{\Delta \Vdash K} \mathcal{C}_2, \mathcal{D}_2' \mathbf{proc}(u_\gamma[d], [v^c/y_\alpha^c]P_{k|\xi}@(d_1 \sqcup c)) \mathcal{D}_2'', \mathcal{F}_2$$

This completes the proof of the subcase as we know that the relevancy of channels in $\mathcal{D}_i'$ and $\mathcal{D}_i''$ remain intact.

**Case 6.** $\mathcal{D}_1 = \mathcal{D}_1' \mathbf{proc}(y_\alpha[c], \mathbf{case}\, y_\alpha^c(\ell \Rightarrow P_\ell)_{\ell \in L}@d_1) \mathbf{msg}(y_\alpha^c.k; x_1^c \leftarrow y_\alpha^c) \mathcal{D}_1''$ and
$$\mathcal{D}_1' \mathbf{proc}(y_\alpha[c], \mathbf{case}\, y_\alpha^c(\ell \Rightarrow P_\ell)_{\ell \in L}@d_1) \mathbf{msg}(y_\alpha^c.k; x_1^c \leftarrow y_\alpha^c) \mathcal{D}_1'' \mapsto_{\Delta \Vdash K} \mathcal{D}_1' \mathbf{proc}(x_1[c], [x_1^c/y_\alpha^c]P_k@c) \mathcal{D}_1''$$

We consider sub-cases based on relevancy of process offering along $y_\alpha^c$. Observe that $y_\alpha^c$ is relevant if an only if $x_1^c$ is relevant, since they share a message of secrecy $c$.

**Subcase 1. $\mathbf{proc}(y_\alpha[c], \mathbf{case}\, y_\alpha^c(\ell \Rightarrow P_\ell)_{\ell \in L}@d_1)$ is not relevant.** By definition either $d_1 \sqcup c = c \not\sqsubseteq \xi$ or none of the channels connected to $P$ including its offering channel $y_\alpha^c$ are relevant. In both cases, means that $x_1^c$ is not relevant and $\mathbf{msg}(y_\alpha^c.k; x_1^c \leftarrow y_\alpha^c)$ is not relevant either. Moreover the continuation process $\mathbf{proc}(x_1[c], [x_1^c/y_\alpha^c]P_k@c)$ won't be relevant. And
$$\mathcal{D}_1' \mathbf{proc}(x_1[c], [x_1^c/y_\alpha^c]P_k@c) \mathcal{D}_1'' \Downarrow \xi =_\xi \mathcal{D}_1' \mathcal{D}_1'' \Downarrow \xi =_\xi \mathcal{D}_1 \Downarrow \xi =_\xi \mathcal{D}_2 \Downarrow \xi$$

**Subcase 2. $\mathbf{proc}(y_\alpha[c], \mathbf{case}\, y_\alpha^c(\ell \Rightarrow P_\ell)_{\ell \in L}@d_1)$ is relevant.** By definition of relevancy, we get that $c \sqcup d_1 = c \sqsubseteq \xi$ and thus $y_\alpha^c$ is relevant. This means that $\mathbf{msg}(y_\alpha^c.k; x_1^c \leftarrow y_\alpha^c)$ is relevant too. By assumption that $\mathcal{D}_1 \Downarrow \xi =_\xi \mathcal{D}_2 \Downarrow \xi$, and definition of $=_\xi$:
$$\mathcal{D}_2 = \mathcal{D}_2' \mathbf{proc}(y_\alpha[c], \mathbf{case}\, y_\alpha^c(\ell \Rightarrow P_{\ell|\xi})_{\ell \in L}@d_1) \mathbf{msg}(y_\alpha^c.k; x_1^c \leftarrow y_\alpha^c) \mathcal{D}_2'',$$
such that $P_{\ell|\xi}$ is equal to $P_\ell$ modulo renaming of some channels with secrecy level higher than or incomparable to the observer.

We have
$$\mathcal{C}_2, \mathcal{D}_2, \mathcal{F}_2 \mapsto_{\Delta \Vdash K} \mathcal{C}_2, \mathcal{D}_2' \mathbf{proc}(x_1[c], [x_1^c/y_\alpha^c]P_{k|\xi}@c) \mathcal{D}_2'', \mathcal{F}_2$$

This completes the proof of the subcase as we know that the relevancy of channels in $\mathcal{D}_i'$ and $\mathcal{D}_i''$ remain intact.

**Case 7.** $\mathcal{D}_1 = \mathcal{D}_1' \mathbf{msg}(\mathbf{send}\, x_\eta^c\, y_\alpha^c; y_\alpha^c \leftarrow x_1^c) \mathbf{proc}(v_\eta[c'], w^c \leftarrow \mathbf{recv}\, y_\alpha^c; P@d_1) \mathcal{D}_1''$ and
$$\mathcal{D}_1' \mathbf{msg}(\mathbf{send}\, x_\eta^c\, y_\alpha^c; y_\alpha^c \leftarrow x_1^c) \mathbf{proc}(v_\eta[c'], w^c \leftarrow \mathbf{recv}\, y_\alpha^c; P@d_1) \mathcal{D}_1'' \mapsto_{\Delta \Vdash K} \mathcal{D}_1' \mathbf{proc}(v_\eta[c'], [x_\eta/w][x_1^c/y_\alpha^c]P@c \sqcup d_1) \mathcal{D}_1''$$

We consider sub-cases based on relevancy of process offering along $v_\eta^{c'}$.

**Subcase 1. $\mathbf{proc}(v_\eta[c'], w^c \leftarrow \mathbf{recv}\, y_\alpha^c; P@d_1)$ is not relevant.** By definition either $d_1 \sqcup c \not\sqsubseteq \xi$ or none of the channels connected to $P$ including $y_\alpha^c$ are relevant.

In both cases by the definition of quasi running secrecy we know that neither $\mathbf{msg}(\mathbf{send}\, x_\eta^c\, y_\alpha^c; x_1^c \leftarrow y_\alpha^c)$ nor the continuation process $\mathbf{proc}(v_\eta[c'], [x_\eta/w][x_1^c/y_\alpha^c]P@c \sqcup d_1)$ are relevant. It is then straightforward to see that
$$\mathcal{D}_1' \mathbf{proc}(v_\eta[c'], [x_\eta/w][x_1^c/y_\alpha^c]P@c \sqcup d_1) \mathcal{D}_1'' \Downarrow \xi =_\xi \mathcal{D}_1' \mathcal{D}_1'' \Downarrow \xi =_\xi \mathcal{D}_1 \Downarrow \xi =_\xi \mathcal{D}_2 \Downarrow \xi.$$

**Subcase 2. $\mathbf{proc}(v_\eta[c'], w^c \leftarrow \mathbf{recv}\, y_\alpha^c; P@d_1)$ is relevant.** By definition of relevancy, we get that $c \sqcup d_1 \sqsubseteq \xi$. This implies that $y_\alpha^c$ is relevant in the pre-step. From relevancy of $y_\alpha^c$ and the quasi running secrecy lower than or equal to the observer of the positive message $\mathbf{msg}(\mathbf{send}\, x_\eta^c\, y_\alpha^c; y_\alpha^c \leftarrow x_1^c)$ we get that the message is relevant too. By assumption:
$$\mathcal{D}_2 = \mathcal{D}_2' \mathbf{msg}(\mathbf{send}\, x_\eta^c\, y_\alpha^c; y_\alpha^c \leftarrow x_1^c) \mathbf{proc}(u_\gamma[c'], w^c \leftarrow \mathbf{recv}\, y_\alpha^c; P_{|\xi}@d_1) \mathcal{D}_2'',$$
such that $P_{\ell|\xi}$ is equal to $P_\ell$ modulo renaming of some channels with secrecy level higher than or incomparable to the observer.

We have
$$\mathcal{C}_2, \mathcal{D}_2, \mathcal{F}_2 \mapsto_{\Delta \Vdash K} \mathcal{C}_2, \mathcal{D}_2' \mathbf{proc}(u_\gamma[c'], [x_\eta/w][x_1^c/y_\alpha^c]P_{|\xi}@d_1 \sqcup c) \mathcal{D}_2'', \mathcal{F}_2$$

We need to consider that the quasi running secrecy of the process may increases in the post step based on the code of $P$ and $P_{|\xi}$. The argument for this case is similar to the previous cases of the proof. See **Case 1.Subcase 2.**. One interesting situation is when the relevancy of chain of positive and relevant messages in the pre-step of $\mathcal{D}_i'$ changes in the post-step. By relevancy in the pre-step we know that these chains exist in both runs, so the same chain of messages will become irrelevant in the post-step of both runs.

**Case 8.** $\mathcal{D}_1 = \mathcal{D}'_1 \mathbf{proc}(y_\alpha[c], w^c \leftarrow \mathbf{recv}\, y_\alpha^c; P@d_1) \mathbf{msg}(\mathbf{send}\, x_\eta^c y_\alpha^c; x_1^c \leftarrow y_\alpha^c) \mathcal{D}''_1$ and

$$\mathcal{D}'_1 \mathbf{proc}(y_\alpha[c], w^c \leftarrow \mathbf{recv}\, y_\alpha^c; P@d_1) \mathbf{msg}(\mathbf{send}\, x_\eta^c y_\alpha^c; x_1^c \leftarrow y_\alpha^c) \mathcal{D}''_1 \mapsto_{\Delta \Vdash K} \mathcal{D}'_1 \mathbf{proc}(x_1[c], [x_\eta/w][x_1^c/y_\alpha^c] P@d_1) \mathcal{D}''_1$$

We consider sub-cases based on relevancy of process offering along $y_\alpha^c$.

**Subcase 1.** $\mathbf{proc}(y_\alpha[c], w^c \leftarrow \mathbf{recv}\, y_\alpha^c; P@d_1)$ **is not relevant.** By definition either $d_1 \sqcup c = c \not\sqsubseteq \xi$ or none of the channels connected to $P$ including $y_\alpha^c$ are relevant. In both cases the negative message $\mathbf{msg}(\mathbf{send}\, x_\eta^c y_\alpha^c; x_1^c \leftarrow y_\alpha^c)$ and the continuation process $\mathbf{proc}(x_1[c], [x_\eta/w][x_1^c/y_\alpha^c] P@d_1)$ are not relevant either. It is then straightforward to see that

$$\mathcal{D}'_1 \mathbf{proc}(x_1[c], [x_\eta/w][x_1^c/y_\alpha^c] P@d_1) \mathcal{D}''_1 \Downarrow \xi =_\xi \mathcal{D}'_1 \mathcal{D}''_1 \Downarrow \xi =_\xi \mathcal{D}_1 \Downarrow \xi =_\xi \mathcal{D}_2 \Downarrow \xi.$$

**Subcase 2.** $\mathbf{proc}(y_\alpha[c], w^c \leftarrow \mathbf{recv}\, y_\alpha^c; P@d_1)$ **is relevant.** By definition of relevancy, we get that $c \sqcup d_1 = c \sqsubseteq \xi$ and $y_\alpha^c$ is relevant. This means that $\mathbf{msg}(\mathbf{send}\, x_\eta^c y_\alpha^c; x_1^c \leftarrow y_\alpha^c)$ and the channel $x_\eta^c$ are relevant. By assumption that $\mathcal{D}_1 \Downarrow \xi =_\xi \mathcal{D}_2 \Downarrow \xi$, and definition of $=_\xi$:

$$\mathcal{D}_2 = \mathcal{D}'_2 \mathbf{proc}(y_\alpha[c], w \leftarrow \mathbf{recv}\, y_\alpha^c; P_{|\xi}@d_1) \mathbf{msg}(\mathbf{send}\, x_\eta^c y_\alpha^c; x_1^c \leftarrow y_\alpha^c) \mathcal{D}''_2,$$

such that $P_{|\xi}$ is equal to $P$ modulo renaming of some channels with secrecy level higher than or incomparable to the observer.

We have

$$\mathcal{C}_2, \mathcal{D}_2, \mathcal{F}_2 \mapsto_{\Delta \Vdash K} \mathcal{C}_2, \mathcal{D}'_2 \mathbf{proc}(x_1[c], [x_\eta^c/w][x_1^c/y_\alpha^c] P_{|\xi}@d_1) \mathcal{D}''_2, \mathcal{F}_2$$

The proof is similar to previous cases.

**Case 9.** $\mathcal{D}_1 = \mathcal{D}'_1 \mathbf{proc}(y_\alpha[c], y_\alpha^c \leftarrow x_\beta^c @ d_1) \mathcal{D}''_1$ and

$$\mathcal{D}'_1 \mathbf{proc}(y_\alpha[c], y_\alpha^c \leftarrow x_\beta^c @ d_1) \mathcal{D}''_1 \mapsto_{\Delta \Vdash K} \mathcal{D}'_1 [x_\beta^c / y_\alpha^c] \mathcal{D}''_1$$

We consider sub-cases based on relevancy of process offering along $y_\alpha^c$. Observe that $y_\alpha^c$ is relevant if and only if $x_\beta^c$ is relevant.

**Subcase 1.** $\mathbf{proc}(y_\alpha[c], y_\alpha^c \leftarrow x_\beta^c @ d_1)$ **is not relevant.** By definition either $c \not\sqsubseteq \xi$ or none of the channels $y_\alpha^c$ and $x_\beta^c$ are relevant. In both cases, it means that $y_\alpha^c$ and $x_\beta^c$ are not relevant. As a result, we can safely make the substitution $[x_\beta^c/y_\alpha^c]$ in $\mathcal{D}''_1$. If $c \not\sqsubseteq \xi$, it is only the matter of renaming channels with secrecy level higher than or incomparable to the observer level. And if $c \sqsubseteq \xi$, then $y_\alpha^c$ does not occur in any relevant process and we can rename it to $x_\beta^c$.

Moreover, deleting $\mathbf{proc}(y_\alpha[c], y_\alpha^c \leftarrow x_\beta^c @ d_1)$ from the configuration does not decrease the quasi-running secrecy of any of its (grand)children since by the tree invariant $d_1 \sqsubseteq c$. In particular, a chain of positive messages offered along $x_\beta^c$, has its minimum quasi running secrecy of $d_1 \sqsubseteq c = c$. Removing $\mathbf{proc}(y_\alpha[c], y_\alpha^c \leftarrow x_\beta^c @ d_1)$ may increase the quasi running secrecy of such chain of messages. This case is only of significance if $c \sqsubseteq \xi$. By the assumption of subcase, we know that $x_\beta^c$ is irrelevant, and thus the chain of messages have to be irrelevant in the pre-state.

$$\mathcal{D}'_1 \mathbf{proc}(y_\alpha[c], y_\alpha^c \leftarrow x_\beta^c @ d_1) \mathcal{D}''_1 \Downarrow \xi =_\xi \mathcal{D}'_1 [x_\beta^c/y_\alpha^c] \mathcal{D}''_1 \Downarrow \xi =_\xi \mathcal{D}_1 \Downarrow \xi =_\xi \mathcal{D}_2 \Downarrow \xi$$

**Subcase 2.** $\mathbf{proc}(y_\alpha[c], y_\alpha^c \leftarrow x_\beta^c @ d_1)$ **is relevant.** By definition of relevancy, we get that $c \sqsubseteq \xi$ and both $y_\alpha^c$ and $x_\beta^c$ are relevant. By assumption that $\mathcal{D}_1 \Downarrow \xi =_\xi \mathcal{D}_2 \Downarrow \xi$, and definition of $=_\xi$:

$$\mathcal{D}_2 = \mathcal{D}'_2 \mathbf{proc}(y_\alpha[c], y_\alpha^c \leftarrow x_\beta^c @ d_1) \mathcal{D}''_2.$$

We have

$$\mathcal{C}_2, \mathcal{D}_2, \mathcal{F}_2 \mapsto_{\Delta \Vdash K} \mathcal{C}_2, \mathcal{D}'_2 [x_\beta^c/y_\beta^c] \mathcal{D}''_2, \mathcal{F}_2$$

With the same reasoning as in **Subcase 1.**, the quasi running secrecy of a chain of messages offered along $x_\beta^c$ is $c$ and does not decrease after deleting the process. However, it may increase based on the quasi-running secrecy of the parents of $\mathbf{proc}(y_\alpha[c], y_\alpha^c \leftarrow x_\beta^c @ d_1)$ in $\mathcal{D}''_i$. If the parent in one run has a quasi running secrecy lower than or equal to the observer level, then it has to be relevant and thus there is a counterpart in the other run with a quasi running secrecy lower than or equal to the observer level. Thus in both runs the relevancy of the chain of messages does not change in the post-step. If both of the parents have running secrecy higher than or incomparable to the observer level, then the same chain of messages become irrelevant in both runs.

**Case 10.** $\mathcal{D}_1 = \mathcal{D}'_1 \mathbf{proc}(y_\alpha[c], (x^d \leftarrow X \leftarrow \Lambda)@d_2; Q@d_1) \mathcal{D}''_1$ and $\mathbf{def}(\Psi'; \Lambda' \vdash X = P@\psi' :: x : B'[\psi])$ and

$$\mathcal{D}'_1 \mathbf{proc}(y_\alpha[c], (x^d \leftarrow X \leftarrow \Lambda)@d_2; Q@d_1) \mathcal{D}''_1 \mapsto_{\Delta \Vdash K} \mathcal{D}'_1 \mathbf{proc}(a_0[d], [a_0^d, \Lambda/x^d, \Lambda'] \hat{\gamma}\, P@d_2) \mathbf{proc}(y_\alpha[c], [a_0^d/x^d] Q@d_1) \mathcal{D}''_1$$

We consider sub-cases based on relevancy of process offering along $y_\alpha^c$.

**Subcase 1.** $\mathbf{proc}(y_\alpha[c], (x^d \leftarrow X \leftarrow \Lambda)@d_2; Q@d_1)$ **is not relevant.** By definition either $d_1 \not\sqsubseteq \xi$ or none of the channels of this process including $y_\alpha^c$ are relevant.

In both cases, it means that both $\mathbf{proc}(a_0[d], [a_0^d, \Lambda/x^d, \Lambda']\hat{\gamma}P@d_2)$ and $\mathbf{proc}(y_\alpha[c], [a_0^d/x^d]Q@d_1)$ are not relevant either. Note that $d_1 \sqsubseteq d_2$ and thus $d_2 \not\sqsubseteq \xi$ if $d_1 \not\sqsubseteq \xi$.

$$\mathcal{D}_1'\mathbf{proc}(a_0[d], [a_0^d, \Lambda/x^d, \Lambda']\hat{\gamma}P@d_2)\mathbf{proc}(y_\alpha[c], [a_0^d/x^d]Q@d_1)\mathcal{D}_1''\Downarrow\xi =_\xi \mathcal{D}_1'\mathcal{D}_1''\Downarrow\xi =_\xi \mathcal{D}_1\Downarrow\xi =_\xi \mathcal{D}_2\Downarrow\xi$$

**Subcase 2.** $\mathbf{proc}(y_\alpha[c], (x^d \leftarrow X \leftarrow \Lambda)@d_2; Q@d_1)$ **is relevant.** By definition of relevancy, we get that $d_1 \sqsubseteq \xi$ and all channels of this process with secrecy levels lower than or equal to the observer level are relevant. By assumption that $\mathcal{D}_1\Downarrow\xi =_\xi \mathcal{D}_2\Downarrow\xi$, and definition of $=_\xi$:

$$\mathcal{D}_2 = \mathcal{D}_2'\mathbf{proc}(z_\delta[c], (x^d \leftarrow X \leftarrow \Lambda_{|\xi})@d_2; Q_{|\xi}@d_1)\mathcal{D}_2''.$$

We have

$$\mathcal{C}_2, \mathcal{D}_2, \mathcal{F}_2 \mapsto_{\Delta \Vdash K} \mathcal{C}_2, \mathcal{D}_2'\mathbf{proc}(a_0[d], [a_0^d, \Lambda_{|\xi}/x^d, \Lambda']\hat{\gamma}P_{|\xi}@d_2)\mathbf{proc}(z_\delta[c], [a_0^d/x^d]Q_{|\xi}@d_1)\mathcal{D}_2'', \mathcal{F}_2$$

Remark: we can assume that the fresh channel being spawned will be $a$ in both runs. Moreover, the (unique) substitution $\hat{\gamma}$ is the same when stepping both runs.

Note that if $a_0^d$ has secrecy level lower than or equal to the observer level, then taking this step won't change relevancy of any channels. Otherwise some resource of the process may become irrelevant after this step, e.g. $a_0^d$ may block their relevancy path if $d \not\sqsubseteq \xi$ or in the case where $d_2 \not\sqsubseteq \xi$. But this happens to the processes in the both runs and can only change relevant processes/messages in the pre-step to irrelevant processes/messages in the post-step. (similar to the cases 3 and 4 for $\otimes$ and $\multimap$)

$\square$

**Lemma C.18** (Build a Meganode). *Consider $\mathcal{C}_1\mathcal{D}_1\mathcal{F}_1$, and $\mathcal{C}_2\mathcal{D}_2\mathcal{F}_2$ such that $\Psi; \cdot \Vdash \mathcal{C}_1 :: \Delta_1$ and $\Psi; \cdot \Vdash \mathcal{C}_2 :: \Delta_2$, and $\Psi; \Delta_1 \Vdash \mathcal{D}_1 :: x_\alpha{:}A[c]$, and $\Psi; \Delta_2 \Vdash \mathcal{D}_2 :: y_\beta{:}B[d]$, and $\Psi; x_\alpha{:}A[c] \Vdash \mathcal{F}_1 :: \cdot$, and $\Psi; y_\beta{:}B[d] \Vdash \mathcal{F}_2 :: \cdot$. If $\Delta_1\Downarrow\xi = \Delta_2\Downarrow\xi = \Delta$ and $x_\alpha{:}A[c]\Downarrow\xi = y_\beta{:}B[d]\Downarrow\xi = K$, then we can rewrite $\mathcal{C}_i\mathcal{D}_i\mathcal{F}_i$ as $\mathcal{C}_i'\mathcal{D}_i'\mathcal{F}_i'$ such that*

$$(\mathcal{C}_1', \mathcal{D}_1', \mathcal{F}_1'; \mathcal{C}_2', \mathcal{D}_2', \mathcal{F}_2') \in \mathsf{Tree}_\Psi(\Delta \Vdash K).$$

*Moreover, if $\mathcal{D}_1\Downarrow\xi = \mathcal{D}_2\Downarrow\xi$, then $\mathcal{D}_1'\Downarrow\xi = \mathcal{D}_2'\Downarrow\xi$.*

*Proof.* We break down the proof into the following cases based on the structure of $K$ and $\Delta$.

**Case 1.** $x_\alpha{:}A[c]\Downarrow\xi = y_\beta{:}B[d]\Downarrow\xi = K \neq \cdot$. This means that $x_\alpha{:}A[c] = y_\beta{:}B[d]$ and $c \sqsubseteq \xi$. By the tree invariant, we know that $\Delta_1 = \Delta_2 = \Delta$. So without any rewrite configurations $\mathcal{C}_i\mathcal{D}_i\mathcal{F}_i$ satisfy the properties that we are looking for.

**Case 2.** $x_\alpha{:}A[c]\Downarrow\xi = y_\beta{:}B[d]\Downarrow\xi = \cdot$, and $\Delta_1 = \Lambda_1, \Delta$, and $\Delta_2 = \Lambda_2, \Delta$. By typing of configurations we have $\mathcal{C}_i = \mathcal{T}_i\mathcal{T}_i'$ such that $\Psi; \cdot \Vdash \mathcal{T}_1 :: \Delta$, and $\Psi; \cdot \Vdash \mathcal{T}_1' :: \Lambda_1$, and $\Psi; \cdot \Vdash \mathcal{T}_2 :: \Delta$, and $\Psi; \cdot \Vdash \mathcal{T}_2' :: \Lambda_2$.

From the definition of projections, we have $c \not\sqsubseteq \xi$ and $d \not\sqsubseteq \xi$ and for every $u_\gamma[c'] \in \Lambda_1, \Lambda_2$, we know that $c' \not\sqsubseteq \xi$. We build $\mathcal{D}_i' = \mathcal{T}_i'\mathcal{D}_i\mathcal{F}_i$, and $\mathcal{C}_i' = \mathcal{T}_i$, and $\mathcal{F}_i' = \cdot$. By the typing rules, we know that $\Psi; \cdot \Vdash \mathcal{F}_i' :: \cdot$ and $\Psi; \cdot \Vdash \mathcal{T}_i :: \Delta$ and $\Psi; \Delta \Vdash \mathcal{D}_i' :: \cdot$ as we need to establish

$$(\mathcal{C}_1', \mathcal{D}_1', \mathcal{F}_1'; \mathcal{C}_2', \mathcal{D}_2', \mathcal{F}_2') \in \mathsf{Tree}_\Psi(\Delta \Vdash K).$$

Moreover, by the definition of projections, since $c \not\sqsubseteq \xi$ and $d \not\sqsubseteq \xi$, we know that none of the processes in $\mathcal{F}_i$ will be relevant in $\mathcal{D}_i'$. Also, adding $\mathcal{F}_1$ as a parent of the process/message offering along $x_\alpha^c$, does not switch relevancy of any process in $\mathcal{D}_1$, since we already know that $c \not\sqsubseteq \xi$. In particular if we have a message offering along $x_\alpha^c$, its quasi running secrecy is higher than or incomparable to the observer level before adding $\mathcal{F}_1$ as its parent and will stay higher than or incomparable to the observer level after too. The same reasoning goes with $\mathcal{F}_2$ and $y_\beta^c$.

Similarly, since $c' \not\sqsubseteq \xi$ for every $u_\gamma[c'] \in \Lambda_1, \Lambda_2$, and the fact that $\mathcal{T}_i'$ does not use any resources in $\Delta$, we know that all trees in $\mathcal{T}_i'$ are irrelevant. As a result, we have $\mathcal{D}_1'\Downarrow\xi = \mathcal{D}_1\Downarrow\xi = \mathcal{D}_2\Downarrow\xi = \mathcal{D}_2'\Downarrow\xi$.

$\square$

**Theorem C.19** (Fundamental Theorem). *For all security levels $\xi$, and configurations $\Psi; \Delta_1 \Vdash \mathcal{D}_1 :: u_\alpha{:}T_1[c_1]$ and $\Psi; \Delta_2 \Vdash \mathcal{D}_2 :: v_\beta{:}T_2[c_2]$ with $\mathcal{D}_1\Downarrow\xi =_\xi \mathcal{D}_2\Downarrow\xi$, $\Delta_1 \Downarrow \xi = \Delta_2 \Downarrow \xi$, and $u_\alpha{:}T_1[c_1] \Downarrow \xi = v_\beta{:}T_2[c_2] \Downarrow \xi$ we have*

$$(\Delta_1 \Vdash \mathcal{D}_1 :: u_\alpha{:}T_1[c_1]) \equiv_\xi^\Psi (\Delta_2 \Vdash \mathcal{D}_2 :: v_\beta{:}T_2[c_2]).$$

*Proof.* Put $\Delta = \Delta_1 \Downarrow \xi = \Delta_2 \Downarrow \xi$, and $K = u_\alpha{:}T_1[c_1]\Downarrow \xi = v_\beta{:}T_2[c_2]\Downarrow \xi$. The goal is to prove that for all $m$, and $\mathcal{C}_1, \mathcal{C}_2$, and $\mathcal{F}_1, \mathcal{F}_2$ such that $\Psi; \cdot \Vdash \mathcal{C}_1 :: \Delta_1$ and $\Psi; \cdot \Vdash \mathcal{C}_2 :: \Delta_2$ and $\Psi; u_\alpha{:}T_1[c_1] \Vdash \mathcal{F}_1 :: \cdot$ and $\Psi; v_\beta{:}T_2[c_2] \Vdash \mathcal{F}_2 :: \cdot$

$$(\mathcal{C}_1 \mathcal{D}_1 \mathcal{F}_1; \mathcal{C}_2 \mathcal{D}_2 \mathcal{F}_2) \in \mathcal{E}_\Psi^\xi [\![ \Delta \Vdash K ]\!]^m.$$

The proof is by induction on the index $m$.

By Lemma C.18 we can rewrite $\mathcal{C}_1 \mathcal{D}_1 \mathcal{F}_1$ as $\mathcal{C}_1^1 \mathcal{D}_1^1 \mathcal{F}_1^1$ and $\mathcal{C}_2 \mathcal{D}_2 \mathcal{F}_2$ as $\mathcal{C}_2^1 \mathcal{D}_2^1 \mathcal{F}_2^1$ such that $(\mathcal{C}_1^1, \mathcal{D}_1^1, \mathcal{F}_1^1; \mathcal{C}_2^1, \mathcal{D}_2^1, \mathcal{F}_2^1) \in \mathsf{Tree}_\Psi(\Delta \Vdash K)$. Moreover, $\mathcal{D}_1^1 \Downarrow \xi = \mathcal{D}_2^1 \Downarrow \xi$. If $m = 0$, then it is enough to complete the proof. Otherwise consider $m = m' + 1$.

Consider an arbitrary $\mathcal{C}_1, \mathcal{D}_1', \mathcal{F}_1$ such that $\mathcal{C}_1^1, \mathcal{D}_1^1, \mathcal{F}_1^1 \mapsto^{*\Theta;\Upsilon}_{\Delta \Vdash K} \mathcal{C}_1^1, \mathcal{D}_1', \mathcal{F}_1^1$. By Lemma C.17, $\mathcal{C}_2^1, \mathcal{D}_2^1, \mathcal{F}_2^1 \mapsto^{*\Theta';\Upsilon}_{\Delta \Vdash K} \mathcal{C}_2^1, \mathcal{D}_2', \mathcal{F}_2^1$ with $\Delta \Vdash \mathcal{D}_2' :: K$, and $\mathcal{D}_1' \Downarrow \xi = \mathcal{D}_2' \Downarrow \xi$. Consider an arbitrary $\mathcal{C}_1', \mathcal{F}_1', \mathcal{C}_2', \mathcal{F}_2'$ such that $\mathcal{C}_1^1, \mathcal{D}_1', \mathcal{F}_1^1 \mapsto^{*\Theta_1;\Upsilon}_{\Delta \Vdash K} \mathcal{C}_1', \mathcal{D}_1', \mathcal{F}_1'$ and $\mathcal{C}_2^1, \mathcal{D}_2', \mathcal{F}_2^1 \mapsto^{*\Theta_2;\Upsilon}_{\Delta \Vdash K} \mathcal{C}_2', \mathcal{D}_2', \mathcal{F}_2'$.

We show that

$$\dagger_1 \; \forall u' \in \Upsilon. \, (\mathcal{C}_1', \mathcal{D}_1', \mathcal{F}_1'; \mathcal{C}_2', \mathcal{D}_2', \mathcal{F}_2') \in \mathcal{V}_\Psi^\xi [\![ \Delta \Vdash K ]\!]^m_{\cdot;u'}$$

and

$$\dagger_2 \; \forall u' \in (\Theta_1 \cap \Theta_2). \, (\mathcal{C}_1', \mathcal{D}_1', \mathcal{F}_1'; \mathcal{C}_2', \mathcal{D}_2', \mathcal{F}_2') \in \mathcal{V}_\Psi^\xi [\![ \Delta \Vdash K ]\!]^m_{u';\cdot}.$$

We prove both $\dagger_1$, and $\dagger_2$. If there is no such $u'$ in the specified set, then the statement trivially holds. On the other hand, there might be a $u'$ occurring in both $\Upsilon$, and $\Theta_1 \cap \Theta_2$. In this case, we need to prove both $\dagger_1$ and $\dagger_2$ separately. Similarly, there might be a $u'_\mathcal{C}$ and $u'_\mathcal{F}$ in $\Theta_1 \cap \Theta_2$. In this case, we need to prove $\dagger_2$ based on the rows that provide the condition according to the subscript of $u'$.

**Proof of $\dagger_1$.** Consider an arbitrary channel $u' : T[c'] \in \Upsilon$. We consider cases based on the type of this channel, and whether there is message or forwarding ready to be sent along it; it cannot be both.

**Case 1.** $u'{:}T[c'] = K = u_\alpha{:}T_1[c_1] = u_\alpha{:}1[c_1]$, and $\mathcal{C}_1', \mathcal{D}_1', \mathcal{F}_1' = \mathcal{C}_1', \mathcal{D}_1'' \mathbf{msg}(\mathbf{close}\, u_\alpha^{c_1}), \mathcal{F}_1'$. By the typing rules we know that $\mathcal{D}_1'' = \cdot$, and $\mathcal{C}_1' = \cdot$, and $\Delta = \cdot$. Note that by the definition of relevancy, all processes in $\mathcal{D}_1'$ and $\mathcal{D}_2'$ are relevant. As a result, $\mathcal{D}_2' = \mathbf{msg}(\mathbf{close}\, u_\alpha^{c_1})$, and by the typing rules $\mathcal{C}_2' = \cdot$. By line (1) of the definition, we get

$$\dagger_1 (\mathcal{C}_1', \mathcal{D}_1', \mathcal{F}_1'; \mathcal{C}_2', \mathcal{D}_2', \mathcal{F}_2') \in \mathcal{V}_\Psi^\xi [\![ \cdot \Vdash u_\alpha{:}1[c_1] ]\!]^m_{\cdot;u'},$$

**Case 2.** $u'{:}T[c'] = K = u_\alpha{:}T_1[c_1] = u_\alpha{:}\oplus\{\ell{:}A_\ell\}_{\ell \in I}[c_1]$, and $\mathcal{D}_1' = \mathcal{D}_1'' \mathbf{msg}(u_\alpha^{c_1}.k; u_\alpha^{c_1} \leftarrow w_\gamma^{c_1})$. By the assumption of the theorem, $\mathcal{D}_1'$ and $\mathcal{D}_2'$ are both relevant, and we have $\mathcal{D}_2' = \mathcal{D}_2'' \mathbf{msg}(u_\alpha^{c_1}.k; u_\alpha^{c_1} \leftarrow w_\gamma^{c_1})$.

Removing $\mathbf{msg}(u_\alpha^{c_1}.k; u_\alpha^{c_1} \leftarrow w_\gamma^{c_1})$ from $\mathcal{D}_1'$ and $\mathcal{D}_2'$ does not change relevancy of the remaining configuration when $c_1 \sqsubseteq \xi$:

$$\mathcal{D}_1'' \Downarrow \xi = \mathcal{D}_2'' \Downarrow \xi.$$

We can apply the induction hypothesis on the index $m' < m$ to get

$$\star (\mathcal{C}_1' \mathcal{D}_1'' \mathbf{msg}(u_\alpha^{c_1}.k; u_\alpha^{c_1} \leftarrow w_\gamma^{c_1}) \mathcal{F}_1', \mathcal{C}_2' \mathcal{D}_2'' \mathbf{msg}(u_\alpha^{c_1}.k; u_\alpha^{c_1} \leftarrow w_\gamma^{c_1}) \mathcal{F}_2') \in \mathcal{E}_\Psi^\xi [\![ \Delta \Vdash w_\gamma{:}A_k[c_1] ]\!]^{m'}.$$

By line (2) in the definition of $\mathcal{V}_\Psi^\xi$:

$$\dagger_1 (\mathcal{C}_1', \mathcal{D}_1'' \mathbf{msg}(u_\alpha^{c_1}.k; u_\alpha^{c_1} \leftarrow w_\gamma^{c_1}), \mathcal{F}_1'; \mathcal{C}_2', \mathcal{D}_2'' \mathbf{msg}(u_\alpha^{c_1}.k; u_\alpha^{c_1} \leftarrow w_\gamma^{c_1}), \mathcal{F}_2') \in \mathcal{V}_\Psi^\xi [\![ \Delta \Vdash u_\alpha{:}\oplus\{\ell : A_\ell\}_{\ell \in I}[c_1] ]\!]^{m'+1}_{\cdot;u'}.$$

**Case 3.** $u'{:}T[c'] = K = u_\alpha{:}T_1[c_1] = u_\alpha{:}(A \otimes B)[c_1]$, and

$$\mathcal{D}_1' = \mathcal{D}_1'' \mathcal{T}_1 \mathbf{msg}(\mathbf{send}\, x_\beta^{c_1}\, u_\alpha^{c_1}; u_\alpha^{c_1} \leftarrow w_\gamma^{c_1})$$

where $\Delta = \Lambda_1, \Lambda_2$ and $\Psi; \Lambda_2 \Vdash \mathcal{T}_1 :: (x_\beta{:}A[c_1])$.

By assumption of the theorem, we know $\mathcal{D}_1'$ and $\mathcal{D}_2'$ are both relevant and

$$\mathcal{D}_2' = \mathcal{D}_2'' \mathcal{T}_2 \mathbf{msg}(\mathbf{send}\, x_\beta^{c_1}\, u_\alpha^{c_1}; u_\alpha^{c_1} \leftarrow w_\gamma^{c_1}),$$

such that $\Psi; \Lambda_2 \Vdash \mathcal{T}_2 :: (x_\beta{:}A[c_1])$. Moreover, by relevancy of $\mathcal{D}_i'$ (and relevancy of $x_\beta^{c_1}$) we get $\mathcal{D}_1'' \Downarrow \xi = \mathcal{D}_2'' \Downarrow \xi$ and $\mathcal{T}_1 \Downarrow \xi = \mathcal{T}_2 \Downarrow \xi$.

By configuration typing we can break down $\mathcal{C}_i'$ into $\mathcal{H}_i$ and $\mathcal{H}_i'$ such that $\Psi; \cdot \Vdash \mathcal{H}_i :: \Lambda_1$, and $\Psi; \cdot \Vdash \mathcal{H}_i' :: \Lambda_2$. In the case where $x_\beta{:}A[c_1] \in \Delta$, we have $\mathcal{T}_i = \cdot$ and $\Lambda_2 = x_\beta{:}A[c_1]$.

We can apply the induction hypothesis on the index $m' < m$ to get

$\star' (\mathcal{H}'_1 \mathcal{T}_1 \mathcal{H}_1 \mathcal{D}''_1 \mathbf{msg}(\mathbf{send} x_\beta^{c_1}\, u_\alpha^{c_1}; u_\alpha^{c_1} \leftarrow w_\gamma^{c_1}) \mathcal{F}'_1, \mathcal{H}'_2 \mathcal{T}_2 \mathcal{H}_2 \mathcal{D}''_1 \mathbf{msg}(\mathbf{send} x_\beta^{c_1}\, u_\alpha^{c_1}; u_\alpha^{c_1} \leftarrow w_\gamma^{c_1}) \mathcal{F}'_2) \in \mathcal{E}^\xi_\Psi [\![ \Lambda_2 \Vdash x_\beta{:}A[c_1] ]\!]^{m'}.$

and

$\star' (\mathcal{H}_1 \mathcal{D}''_1 \mathcal{H}'_1 \mathcal{T}_1 \mathbf{msg}(\mathbf{send} x_\beta^{c_1}\, u_\alpha^{c_1}; u_\alpha^{c_1} \leftarrow w_\gamma^{c_1}) \mathcal{F}'_1, \mathcal{H}_2 \mathcal{D}''_2 \mathcal{H}'_2 \mathcal{T}_2 \mathbf{msg}(\mathbf{send} x_\beta^{c_1}\, u_\alpha^{c_1}; u_\alpha^{c_1} \leftarrow w_\gamma^{c_1}) \mathcal{F}'_2) \in \mathcal{E}^\xi_\Psi [\![ \Lambda_1 \Vdash w_\gamma{:}B[c_1] ]\!]^{m'}.$

By line (4) in the definition of $\mathcal{V}^\xi_\Psi$:

$\dagger_1 (\mathcal{C}'_1, \mathcal{D}''_1 \mathcal{T}_1 \mathbf{msg}(\mathbf{send} x_\beta^{c_1}\, u_\alpha^{c_1}; u_\alpha^{c_1} \leftarrow w_\gamma^{c_1}), \mathcal{F}'_1; \mathcal{C}'_2, \mathcal{D}''_2 \mathcal{T}_2 \mathbf{msg}(\mathbf{send} x_\beta^{c_1}\, u_\alpha^{c_1}; u_\alpha^{c_1} \leftarrow w_\gamma^{c_1}), \mathcal{F}'_2) \in \mathcal{V}^\xi_\Psi [\![ \Delta \Vdash u_\alpha{:}A \otimes B[c_1] ]\!]^{m'+1}_{\cdot; u'}.$

**Case 4.** $u'{:}T'[c'] = K = u_\alpha{:}A[c_1]$, and $\mathcal{D}'_1 = \mathcal{D}''_1 \mathbf{proc}(u_\alpha[c_1], u_\alpha^{c_1} \leftarrow w_\gamma^{c_1} @ d_1)$.
By assumption of the theorem, $\mathcal{D}'_2 = \mathcal{D}''_2 \mathbf{proc}(u_\alpha[c_1], u_\alpha^{c_1} \leftarrow w_\gamma^{c_1} @ d_1)$.
By assumption of the theorem, $\mathcal{D}''_1 \Downarrow \xi = \mathcal{D}''_2 \Downarrow \xi$, since $w_\gamma^{c_1}$ is relevant in $\mathcal{D}'_i$.
We can apply the induction hypothesis on $m' < m$ to get

$\star (\mathcal{C}'_1 \mathcal{D}''_1 [w_\gamma^{c_1}/u_\alpha^{c_1}] \mathcal{F}'_1, \mathcal{C}'_2 \mathcal{D}''_2 [w_\gamma^{c_1}/u_\alpha^{c_1}] \mathcal{F}'_2) \in \mathcal{E}^\xi_\Psi [\![ \Delta \Vdash w_\gamma{:}A[c_1] ]\!]^{m'}.$

By line (11) in the definition of $\mathcal{V}^\xi_\Psi$:

$\dagger_1 (\mathcal{C}'_1, \mathcal{D}''_1 \mathbf{proc}(u_\alpha^{c_1}, u_\alpha^{c_1} \leftarrow w_\gamma^{c_1} @ d_1), \mathcal{F}'_1; \mathcal{C}'_2, \mathcal{D}''_2 \mathbf{proc}(u_\alpha^{c_1}, u_\alpha^{c_1} \leftarrow w_\gamma^{c_1} @ d_1), \mathcal{F}'_2) \in \mathcal{V}^\xi_\Psi [\![ \Delta \Vdash u_\alpha{:}A[c_1] ]\!]^{m'+1}_{\cdot; u'}.$

**Case 5.** $\Delta = \Delta', x_\gamma{:}\&\{\ell : A_\ell\}_{\ell \in L}[c]$, and $u'{:}T'[c'] = x_\gamma{:}\&\{\ell : A_\ell\}_{\ell \in L}[c]$, and $\mathcal{C}'_1, \mathcal{D}'_1, \mathcal{F}'_1 = \mathcal{C}'_1, \mathbf{msg}(x_\gamma^c.k; w_\delta^c \leftarrow x_\gamma^c) \mathcal{D}''_1, \mathcal{F}'_1$.
By the assumption of the theorem, $\mathcal{D}'_2 = \mathbf{msg}(x_\gamma^c.k; w_\delta^c \leftarrow x_\gamma^c) \mathcal{D}''_2$. Moreover, $\mathcal{D}''_1 \Downarrow \xi = \mathcal{D}'_2 \Downarrow \xi$: $c \sqsubseteq \xi$ and thus $x_{\gamma+1}$ remains relevant in $\mathcal{D}_i$ and no relevancy changes in the configurations after removing the negative message.
We can apply the induction hypothesis on the smaller index $m' < m$ to get

$\star (\mathcal{C}'_1 \mathbf{msg}(x_\gamma^c.k; w_\delta^c \leftarrow x_\gamma^c) \mathcal{D}''_1 \mathcal{F}'_1, \mathcal{C}'_2 \mathbf{msg}(x_\gamma^c.k; w_\delta^c \leftarrow x_\gamma^c) \mathcal{D}''_2 \mathcal{F}'_2) \in \mathcal{E}^\xi_\Psi [\![ \Delta', x_{\gamma+1}{:}A_k[c] \Vdash K ]\!]^{m'}.$

By line (8) in the definition of $\mathcal{V}^\xi_\Psi$:

$\dagger_1 (\mathcal{C}'_1, \mathbf{msg}(x_\gamma^c.k; w_\delta^c \leftarrow x_\gamma^c) \mathcal{D}''_1, \mathcal{F}'_1; \mathcal{C}'_2, \mathbf{msg}(x_\gamma^c.k; w_\delta^c \leftarrow x_\gamma^c) \mathcal{D}''_2, \mathcal{F}'_2) \in \mathcal{V}^\xi_\Psi [\![ \Delta', x_\gamma{:}\&\{\ell{:}A_\ell\}_{\ell \in L}[c] \Vdash K ]\!]^{m'+1}_{\cdot; u'}.$

**Case 6.** $\Delta = \Delta', \Delta'', x_\gamma{:}(A \multimap B)[c]$, and $u'{:}T'[c'] = x_\gamma{:}(A \multimap B)[c]$, and

$\mathcal{C}'_1, \mathcal{D}'_1, \mathcal{F}'_1 = \mathcal{C}'_1, \mathcal{T}_1 \mathbf{msg}(\mathbf{send} y_\delta^c\, x_\gamma^c; w_\eta^c \leftarrow x_\gamma^c) \mathcal{D}''_1, \mathcal{F}'_1,$

such that $\Psi; \Delta'' \Vdash \mathcal{T}_1 :: y_\delta{:}A[c]$ and $\Psi; \Delta', w_\eta{:}B[c] \Vdash \mathcal{D}''_1 :: K$.
The message $\mathbf{msg}(\mathbf{send} y_\delta^c\, x_\gamma^c; w_\eta^c \leftarrow x_\gamma^c)$ is relevant in $\mathcal{D}'_1$. By assumption of the theorem, $\mathcal{D}'_2 = \mathcal{T}_2 \mathbf{msg}(\mathbf{send} y_\delta^c\, x_\gamma^c) \mathcal{D}''_2$, such that $\Psi; \Delta'' \Vdash \mathcal{T}_2 :: y_\delta{:}A[c]$ and $\Psi; \Delta', w_\eta{:}B[c] \Vdash \mathcal{D}''_2 :: K$.
By typing configuration, we have $\mathcal{C}'_i = \mathcal{A}_i \mathcal{A}'_i$ such that $\cdot \Vdash \mathcal{A}_i :: \Delta'$ and $\cdot \Vdash \mathcal{A}'_i :: \Delta'', x_\gamma{:}(A \multimap B)[c]$.
By assumption we know that $\mathcal{D}'_1 \Downarrow \xi = \mathcal{D}'_2 \Downarrow \xi$. By definition of relevancy, we know that $\mathbf{msg}(\mathbf{send} y_\delta^c\, x_\gamma^c; w_\eta^c \leftarrow x_\gamma^c)$ and the tree $\mathcal{T}_1$ connected to it are relevant in $\mathcal{D}_1$ and thus $\mathcal{T}_1$ is equal to $\mathcal{T}_2$ in $\mathcal{D}_2$. Removing these from both configurations does not change relevancy of the rest of the configuration since $w_\eta^c$ will remain relevant, and the relevancy of the message's parent does not change: $\mathcal{D}''_1 \Downarrow \xi = \mathcal{D}''_2 \Downarrow \xi$.
We can apply the induction hypothesis on the smaller index $m' < m$ to get

$\star (\mathcal{A}'_1 \mathcal{T}_1 \mathcal{A}_1 \mathbf{msg}(\mathbf{send} y_\delta^c\, x_\gamma^c; w_\eta^c \leftarrow x_\gamma^c) \mathcal{D}''_1 \mathcal{F}'_1, \mathcal{A}'_2 \mathcal{T}_2 \mathcal{A}_2 \mathbf{msg}(\mathbf{send} y_\delta^c\, x_\gamma^c; w_\eta^c \leftarrow x_\gamma^c) \mathcal{D}''_2 \mathcal{F}'_2) \in \mathcal{E}^\xi_\Psi [\![ \Delta'' \Vdash y_\delta{:}A[c] ]\!]^{m'}.$

and

$\star' (\mathcal{A}_1 \mathcal{A}'_1 \mathcal{T}_1 \mathbf{msg}(\mathbf{send} y_\delta^c\, x_\gamma^c; w_\eta^c \leftarrow x_\gamma^c) \mathcal{D}''_1 \mathcal{F}'_1, \mathcal{A}_2 \mathcal{A}'_2 \mathcal{T}_2 \mathbf{msg}(\mathbf{send} y_\delta^c\, x_\gamma^c; w_\eta^c \leftarrow x_\gamma^c) \mathcal{D}''_2 \mathcal{F}'_2) \in \mathcal{E}^\xi_\Psi [\![ \Delta', w_\eta{:}B[c] \Vdash K ]\!]^{m'}.$

By line (10) in the definition of $\mathcal{V}^\xi_\Psi$:

$\dagger_1 (\mathcal{C}'_1, \mathcal{T}_1 \mathbf{msg}(\mathbf{send} y_\delta^c\, x_\gamma^c; w_\eta^c \leftarrow x_\gamma^c) \mathcal{D}''_1, \mathcal{F}'_1; \mathcal{C}'_2, \mathcal{T}_2 \mathbf{msg}(\mathbf{send} y_\delta^c\, x_\gamma^c; w_\eta^c \leftarrow x_\gamma^c) \mathcal{D}''_2, \mathcal{F}'_2) \in \mathcal{V}^\xi_\Psi [\![ \Delta'_1, x_\gamma{:}A \multimap B[c] \Vdash K ]\!]^{m'+1}_{\cdot; u'}.$

43

**Proof of** $\dagger_2$. Consider an arbitrary channel $u':T[c'] \in \Theta_1 \cap \Theta_2$. We consider cases based on the type of this channel, and whether there is message or forwarding ready to be sent along it; it cannot be both.

**Case 1.** $u':T[c'] = K = u_\alpha:T_1[c_1] = u_\alpha:\&\{\ell:A_\ell\}_{\ell \in I}[c_1]$, and $\mathcal{C}'_1, \mathcal{D}'_1, \mathcal{F}'_1 = \mathcal{C}'_1, \mathcal{D}'_1, \mathbf{msg}(u_\alpha^{c_1}.k; w_\gamma^{c_1} \leftarrow u_\alpha^{c_1})\mathcal{F}''_1$.
We consider two subcases:

**Subcase 1.** $\mathcal{F}'_2 \neq \mathbf{msg}(u_\alpha^{c_1}.k; w_\gamma^{c_1} \leftarrow u_\alpha^{c_1})\mathcal{F}''_2$. By line (3) in the definition of $\mathcal{V}_\Psi^\xi$:

$$\dagger_2(\mathcal{C}'_1, \mathcal{D}'_1, \mathbf{msg}(u_\alpha^{c_1}.k; w_\gamma^{c_1} \leftarrow u_\alpha^{c_1})\mathcal{F}''_1; \mathcal{C}'_2, \mathcal{D}'_2, \mathcal{F}'_2) \in \mathcal{V}_\Psi^\xi[\![\Delta \Vdash u_\alpha:\&\{\ell:A_\ell\}_{\ell \in I}[c_1]]\!]_{u';\cdot}^m.$$

**Subcase 2.** $\mathcal{F}'_2 = \mathbf{msg}(u_\alpha^{c_1}.k; w_\gamma^{c_1} \leftarrow u_\alpha^{c_1})\mathcal{F}''_2$.
We can apply the induction hypothesis on the index $m' < m$, but first we need to show that the invariant of the induction holds. From $\mathcal{D}'_1 \Downarrow \xi = \mathcal{D}'_2 \Downarrow \xi$ and $c_1 \sqsubseteq \xi$, we get

$$\mathcal{D}'_1 \mathbf{msg}(u_\alpha^{c_1}.k; w_\gamma^{c_1} \leftarrow u_\alpha^{c_1})\Downarrow \xi = \mathcal{D}'_2 \mathbf{msg}(u_\alpha^{c_1}.k; w_\gamma^{c_1} \leftarrow u_\alpha^{c_1})\Downarrow \xi.$$

By induction hypothesis

$$\star (\mathcal{C}'_1 \mathcal{D}'_1 \mathbf{msg}(u_\alpha^{c_1}.k; w_\gamma^{c_1} \leftarrow u_\alpha^{c_1})\mathcal{F}''_1, \mathcal{C}'_2 \mathcal{D}'_2 \mathbf{msg}(u_\alpha^{c_1}.k; w_\gamma^{c_1} \leftarrow u_\alpha^{c_1})\mathcal{F}''_2) \in \mathcal{E}_\Psi^\xi[\![\Delta \Vdash w_\gamma:A_k[c_1]]\!]^{m'}.$$

By line (3) in the definition of $\mathcal{V}_\Psi^\xi$:

$$\dagger_2(\mathcal{C}'_1, \mathcal{D}'_1 \mathbf{msg}(u_\alpha^{c_1}.k; w_\gamma^{c_1} \leftarrow u_\alpha^{c_1}), \mathcal{F}''_1; \mathcal{C}'_2, \mathcal{D}'_2 \mathbf{msg}(u_\alpha^{c_1}.k; w_\gamma^{c_1} \leftarrow u_\alpha^{c_1}), \mathcal{F}''_2) \in \mathcal{V}_\Psi^\xi[\![\Delta \Vdash u_\alpha:\&\{\ell:A_\ell\}_{\ell \in I}[c_1]]\!]_{u';\cdot}^{m'+1}.$$

**Case 2.** $u':T[c'] = K = u_\alpha:T_1[c_1] = u_\alpha:A \multimap B[c_1]$ and $\mathcal{C}'_1, \mathcal{D}'_1, \mathcal{F}'_1 = \mathcal{C}'_1, \mathcal{D}'_1, \mathcal{T}_1 \mathbf{msg}(\mathbf{send} x_\beta^{c_1} u_\alpha^{c_1}; w_\gamma^{c_1} \leftarrow u_\alpha^{c_1})\mathcal{F}''_1$. There are two subcases to consider:

**Subcase 1.** $\mathcal{F}'_2 \neq \mathcal{T}_2 \mathbf{msg}(\mathbf{send} x_\beta^{c_1} u_\alpha^{c_1}; w_\gamma^{c_1} \leftarrow u_\alpha^{c_1})\mathcal{F}''_2$. By line (5) in the definition of $\mathcal{V}_\Psi^\xi$:

$$\dagger_2(\mathcal{C}'_1, \mathcal{D}'_1, \mathcal{T}_1 \mathbf{msg}(\mathbf{send} x_\beta^{c_1} u_\alpha^{c_1}; w_\gamma^{c_1} \leftarrow u_\alpha^{c_1})\mathcal{F}''_1; \mathcal{C}'_2, \mathcal{D}'_2, \mathcal{F}'_2) \in \mathcal{V}_\Psi^\xi[\![\Delta \Vdash u_\alpha:A \multimap B[c_1]]\!]_{u';\cdot}^m.$$

**Subcase 2.** $\mathcal{F}'_2 = \mathcal{T}_2 \mathbf{msg}(\mathbf{send} x_\beta^{c_1} u_\alpha^{c_1}; w_\gamma^{c_1} \leftarrow u_\alpha^{c_1})\mathcal{F}''_2,$.
By assumption and $c_1 \sqsubseteq \xi$, we know that $\mathcal{D}'_1$ and $\mathcal{D}'_2$ are relevant. Adding a negative message along $c_1 \sqsubseteq \xi$ as the root does not change their relevancy:

$$\mathcal{D}'_1 \mathbf{msg}(\mathbf{send} x_\beta^{c_1} u_\alpha^{c_1}; w_\gamma^{c_1} \leftarrow u_\alpha^{c_1})\Downarrow \xi = \mathcal{D}'_2 \mathbf{msg}(\mathbf{send} x_\beta^{c_1} u_\alpha^{c_1}; w_\gamma^{c_1} \leftarrow u_\alpha^{c_1})\Downarrow \xi$$

We can apply the induction hypothesis on $m' < m$:

$$\star \ (\mathcal{C}'_1 \mathcal{T}_1 \mathcal{D}'_1 \mathbf{msg}(\mathbf{send} x_\beta^{c_1} u_\alpha^{c_1}; w_\gamma^{c_1} \leftarrow u_\alpha^{c_1})\mathcal{F}''_2, \mathcal{C}'_2 \mathcal{T}_2 \mathcal{D}'_2 \mathbf{msg}(\mathbf{send} x_\beta^{c_1} u_\gamma^{c_1}; w_\gamma^{c_1} \leftarrow u_\alpha^{c_1})\mathcal{F}''_2) \in \mathcal{E}_\Psi^\xi[\![\Delta, x_\beta:A[c_1] \Vdash w_\gamma:B[c_1]]\!]^{m'}.$$

By line (5) in the definition of $\mathcal{V}_\Psi^\xi$:

$$\dagger_2(\mathcal{C}'_1, \mathcal{D}'_1, \mathcal{T}_1 \mathbf{msg}(\mathbf{send} x_\beta^{c_1} u_\alpha^{c_1}; w_\gamma^{c_1} \leftarrow u_\alpha^{c_1})\mathcal{F}''_1; \mathcal{C}'_2, \mathcal{D}'_2, \mathcal{T}_2 \mathbf{msg}(\mathbf{send} x_\beta^{c_1} u_\alpha^{c_1}; w_\gamma^{c_1} \leftarrow u_\alpha^{c_1})\mathcal{F}''_2) \in \mathcal{V}_\Psi^\xi[\![\Delta \Vdash u_\alpha:A \multimap B[c_1]]\!]_{u';\cdot}^{m'+1}.$$

**Case 3.** $\Delta = \Delta', x_\gamma:1[c]$, and $u':T'[c'] = x_\gamma:1[c]$, and $\mathcal{C}'_1, \mathcal{D}'_1, \mathcal{F}'_1 = \mathcal{C}''_1 \mathbf{msg}(\mathbf{close}\, x_\gamma^c), \mathcal{D}'_1, \mathcal{F}'_1$.
There are two subcases to consider:

**Subcase 1.** $\mathcal{C}'_2 \neq \mathcal{C}''_2 \mathbf{msg}(\mathbf{close}\, x_\gamma^c)$.
By line (6) in the definition of $\mathcal{V}_\Psi^\xi$,

$$\dagger_2 \left(\mathcal{C}''_1 \mathbf{msg}(\mathbf{close}\, x_\gamma^c), \mathcal{D}'_1, \mathcal{F}'_1, \mathcal{C}'_2, \mathcal{D}'_2, \mathcal{F}'_2\right) \in \mathcal{V}_\Psi^\xi[\![\Delta', x_\gamma:1[c] \Vdash K]\!]_{u';\cdot}^m.$$

**Subcase 2.** $\mathcal{C}'_2 = \mathcal{C}''_2 \mathbf{msg}(\mathbf{close}\, x_\gamma^c)$.
We first briefly explain why the invariant of induction holds after we bring the closing message inside $\mathcal{D}'_i$ and remove the channel $x_\gamma^c$ from $\Delta$, i.e.

$$\mathbf{msg}(\mathbf{close}\, x_\gamma^c)\mathcal{D}'_1 \Downarrow \xi = \mathbf{msg}(\mathbf{close}\, x_\gamma^c)\mathcal{D}'_2 \Downarrow \xi.$$

If the parent of $\mathbf{msg}(\mathbf{close}\, x_\gamma^c)$ in $\mathcal{D}'_1$ is relevant in $\mathcal{D}'_1$ before bringing the message inside then by the assumption of the theorem, it is the same as the parent of $\mathbf{msg}(\mathbf{close}\, x_\gamma^c)$ in $\mathcal{D}'_2$. If after adding the message $\mathcal{D}'_1$ the parent still remains

44

relevant, it means that it has at least one other relevant channel other than $x_\gamma^c$ in $\mathcal{D}_1'$ which also exists in $\mathcal{D}_2'$ and will be relevant after adding the message to $\mathcal{D}_2'$. If after adding the message, $\mathcal{D}_1'$ becomes irrelevant, it means that it does not have a relevant path to any other channel in $\Delta'$ and $K$. By the assumption of theorem the parent of the message in $\mathcal{D}_2'$ does not have such path either. The same argument holds for any other node that becomes irrelevant because of adding the closing message to $\mathcal{D}_i'$. As a result the same processes becomes irrelevant in both $\mathcal{D}_1'$ and $\mathcal{D}_2'$ after adding the message to them and the proof of this case is complete. The same argument holds for the case in which the parent of $\mathbf{msg}(\mathbf{close}\, x_\gamma^c)$ in $\mathcal{D}_2'$ before adding the message is relevant. Otherwise the parent of $\mathbf{msg}(\mathbf{close}\, x_\gamma^c)$ is irrelevant in $\mathcal{D}_1'$ and $\mathcal{D}_2'$ before adding the message and remains irrelevant after that too. The proof in this case is straightforward.
Now that the invariant holds, we can apply the induction hypothesis on the smaller index $m' < m$:

$$\star\ (\mathcal{C}_1''\mathbf{msg}(\mathbf{close}\, x_\gamma^c)\mathcal{D}_1'\mathcal{F}_1', \mathcal{C}_2''\mathbf{msg}(\mathbf{close}\, x_\gamma^c)\mathcal{D}_2'\mathcal{F}_2') \in \mathcal{E}_\Psi^\xi[\![\Delta' \Vdash K]\!]^{m'}.$$

By line (6) in the definition of $\mathcal{V}_\Psi^\xi$,

$$\dagger_2\ (\mathcal{C}_1''\mathbf{msg}(\mathbf{close}\, x_\gamma^c), \mathcal{D}_1', \mathcal{F}_1', \mathcal{C}_2''\mathbf{msg}(\mathbf{close}\, x_\gamma^c), \mathcal{D}_2', \mathcal{F}_2') \in \mathcal{V}_\Psi^\xi[\![\Delta', x_\gamma{:}\mathbf{1}[c] \Vdash K]\!]_{u';\cdot}^{m'+1}.$$

**Case 4.** $\Delta = \Delta', x_\gamma{:}\oplus\{\ell : A_\ell\}_{\ell \in L}[c]$, and $u'{:}T'[c'] = x_\gamma{:}\oplus\{\ell : A_\ell\}_{\ell \in L}[c]$, and $\mathcal{C}_1', \mathcal{D}_1', \mathcal{F}_1' = \mathcal{C}_1''\mathbf{msg}(x_\gamma^c.k; x_\gamma^c \leftarrow w_\delta^c), \mathcal{D}_1', \mathcal{F}_1'$. By configuration typing, we have $\mathcal{C}_1' = \mathcal{T}_i\mathcal{T}_i'$ such that $\Psi; \cdot \Vdash \mathcal{T}_i :: \Delta'$ and $\Psi; \cdot \Vdash \mathcal{T}_i' :: x_\gamma{:}\oplus\{\ell : A_\ell\}_{\ell \in L}[c]$. We know that $\mathcal{T}_1' = \mathcal{T}_1''\mathbf{msg}(x_\gamma^c.k; x_\gamma^c \leftarrow w_\delta^c)$. We consider two subcases:
**Subcase 1.** $\mathcal{T}_2' \neq \mathcal{T}_2''\mathbf{msg}(x_\gamma^c.k; x_\gamma^c \leftarrow w_\delta^c)$.
By line (7) in the definition of $\mathcal{V}_\Psi^\xi$:

$$\dagger_2(\mathcal{T}_1\mathcal{T}_1''\mathbf{msg}(x_\gamma^c.k; x_\gamma^c \leftarrow w_\delta^c), \mathcal{D}_1', \mathcal{F}_1'; \mathcal{T}_2\mathcal{T}_2', \mathcal{D}_2', \mathcal{F}_2'') \in \mathcal{V}_\Psi^\xi[\![\Delta', x_\gamma{:}\oplus\{\ell : A_\ell\}_{\ell \in L}[c] \Vdash K]\!]^m.$$

**Subcase 2.** $\mathcal{T}_2' = \mathcal{T}_2''\mathbf{msg}(x_\gamma^c.k; x_\gamma^c \leftarrow w_\delta^c)$.
We have $\mathbf{msg}(x_\gamma^c.k; x_\gamma^c \leftarrow w_\delta^c)\mathcal{D}_1'\Downarrow\xi = \mathbf{msg}(x_\gamma^c.k; x_\gamma^c \leftarrow w_\delta^c)\mathcal{D}_2'\Downarrow\xi$, since $w_\delta^c$ is relevant and the positive messages $\mathbf{msg}(x_\gamma^c.k; x_\gamma^c \leftarrow w_\delta^c)$ in both configurations are relevant only if their parents are. Thus adding the message does not change relevancy of any other process.
We can apply the induction hypothesis on the smaller index $m' < m$ to get

$$\star\ (\mathcal{T}_1\mathcal{T}_1''\mathbf{msg}(x_\gamma^c.k; x_\gamma^c \leftarrow w_\delta^c)\mathcal{D}_1'\mathcal{F}_1', \mathcal{T}_2\mathcal{T}_2''\mathbf{msg}(x_\gamma^c.k; x_\gamma^c \leftarrow w_\delta^c)\mathcal{D}_2'\mathcal{F}_2') \in \mathcal{E}_\Psi^\xi[\![\Delta', w_\delta : A_k[c] \Vdash K]\!]^{m'}.$$

By line (7) in the definition of $\mathcal{V}_\Psi^\xi$:

$$\dagger_2(\mathcal{T}_1\mathcal{T}_1''\mathbf{msg}(x_\gamma^c.k; x_\gamma^c \leftarrow w_\delta^c), \mathcal{D}_1', \mathcal{F}_1'; \mathcal{T}_2\mathcal{T}_2''\mathbf{msg}(x_\gamma^c.k; x_\gamma^c \leftarrow w_\delta^c), \mathcal{D}_2', \mathcal{F}_2'') \in \mathcal{V}_\Psi^\xi[\![\Delta', x_\gamma{:}\oplus\{\ell : A_\ell\}_{\ell \in L}[c] \Vdash K]\!]_{u';\cdot}^{m'+1}.$$

**Case 5.** $\Delta = \Delta', x_\gamma{:}(A \otimes B)[c]$, and $u'{:}T'[c'] = x_\gamma{:}(A \otimes B)[c]$, and

$$\mathcal{C}_1', \mathcal{D}_1', \mathcal{F}_1' = \mathcal{C}_1''\mathcal{T}_1\mathbf{msg}(\mathbf{send}\, y_\delta^c\, x_\gamma^c; x_\gamma^c \leftarrow w_\eta^c), \mathcal{D}_1', \mathcal{F}_1'.$$

We consider two subcases:
**Subcase 1.** $\mathcal{C}_2' \neq \mathcal{C}_2''\mathcal{T}_2\mathbf{msg}(\mathbf{send}\, y_\delta^c\, x_\gamma^c; x_\gamma^c \leftarrow w_\eta^c)$.
By line (9) in the definition of $\mathcal{V}_\Psi^\xi$:

$$\dagger_2\ (\mathcal{C}_1''\mathcal{T}_1\mathbf{msg}(\mathbf{send}\, y_\delta^c\, x_\gamma^c; x_\gamma^c \leftarrow w_\eta^c), \mathcal{D}_1', \mathcal{F}_1'; \mathcal{C}_2', \mathcal{D}_2', \mathcal{F}_2') \in \mathcal{V}_\Psi^\xi[\![\Delta', x_\gamma{:}A \otimes B[c] \Vdash K]\!]^m.$$

**Subcase 2.** $\mathcal{C}_2' = \mathcal{C}_2''\mathcal{T}_2\mathbf{msg}(\mathbf{send}\, y_\delta^c\, x_\gamma^c; x_\gamma^c \leftarrow w_\eta^c)$.
Moreover, we have
$$\mathbf{msg}(\mathbf{send}\, y_\delta^c\, x_\gamma^c; x_\gamma^c \leftarrow w_\eta^c)\mathcal{D}_1'\Downarrow\xi =_\xi \mathbf{msg}(\mathbf{send}\, y_\delta^c\, x_\gamma^c; x_\gamma^c \leftarrow w_\eta^c)\mathcal{D}_2'\Downarrow\xi :$$

The quasi-running secrecy of $\mathbf{msg}(\mathbf{send}\, y_\delta^c\, x_\gamma^c; x_\gamma^c \leftarrow w_\eta^c)$ is lower than or equal to the observer level if the quasi-running secrecy of its parent is lower than or equal to the observer level. So the relevancy of the parent of the message and thus the rest of configurations do not change by adding the message to the configuration.
We can apply the induction hypothesis on the smaller index $m' < m$ to get

$$\star\ (\mathcal{C}_1''\mathcal{T}_1\mathbf{msg}(\mathbf{send}\, y_\delta^c\, x_\gamma^c; x_\gamma^c \leftarrow w_\eta^c)\mathcal{D}_1'\mathcal{F}_1', \mathcal{C}_2''\mathcal{T}_2\mathbf{msg}(\mathbf{send}\, y_\delta^c\, x_\gamma^c; x_\gamma^c \leftarrow w_\eta^c)\mathcal{D}_2'\mathcal{F}_2') \in \mathcal{E}_\Psi^\xi[\![\Delta', y_\delta{:}A[c], w_\eta{:}B[c] \Vdash K]\!]^{m'}.$$

By line (9) in the definition of $\mathcal{V}_\Psi^\xi$:

$$\dagger_2 \; (\mathcal{C}_1''\mathcal{T}_1\mathbf{msg}(\mathbf{send}\,y_\delta^c\,x_\gamma^c; x_\gamma^c \leftarrow w_\eta^c), \mathcal{D}_1', \mathcal{F}_1'; \mathcal{C}_2''\mathcal{T}_2\mathbf{msg}(\mathbf{send}\,y_\delta^c\,x_\gamma^c; x_\gamma^c \leftarrow w_\eta^c), \mathcal{D}_2', \mathcal{F}_2') \in \mathcal{V}_\Psi^\xi[\![\Delta', x_\gamma{:}A \otimes B[c] \Vdash K]\!]_{u';\cdot}^{m'+1}.$$

**Case 6.** $\Delta = \Delta', x_\delta{:}A[c]$. $\mathcal{C}_1' = \mathcal{C}_1''\mathbf{proc}(x_\delta[c], x_\delta^c \leftarrow w_\gamma^c @ d_1)$.
We consider two subcases:

**Subcase 1.** $\mathcal{C}_2' \neq \mathcal{C}_2''\mathbf{proc}(x_\delta^c, x_\delta^c \leftarrow w_\gamma^c @ d_1)$. By line (12) in the definition of $\mathcal{V}_\Psi^\xi$:

$$\dagger_2(\mathcal{C}_1''\mathbf{proc}(x_\delta[c], x_\delta^c \leftarrow w_\gamma^c @ d_1), \mathcal{D}_1', \mathcal{F}_1'; \mathcal{C}_2, \mathcal{D}_2', \mathcal{F}_2') \in \mathcal{V}_\Psi^\xi[\![\Delta, x_\delta{:}A[c] \Vdash [w_\gamma^c/u_\delta^c]K]\!]^m.$$

**Subcase 2.** $\mathcal{C}_2' = \mathcal{C}_2''\mathbf{proc}(x_\delta[c], x_\delta^c \leftarrow w_\gamma^c @ d_1)$.
By assumption of the theorem, we have $[w_\gamma^c/x_\delta^c]\mathcal{D}_1' \Downarrow \xi = [w_\gamma^c/x_\delta^c]\mathcal{D}_2' \Downarrow \xi$, since we just rename a relevant channel in both configurations. We can apply the induction hypothesis on the smaller $m' < m$ to get

$$\star \; (\mathcal{C}_1''[w_\gamma^c/u_\delta^c]\mathcal{D}_1'\mathcal{F}_1', \mathcal{C}_2''[w_\gamma^c/x_\delta^c]\mathcal{D}_2''\mathcal{F}_2') \in \mathcal{E}_\Psi^\xi[\![\Delta', x_\delta{:}A[c] \Vdash [w_\gamma^c/u_\delta^c]K]\!]^{m'}.$$

By line (12) in the definition of $\mathcal{V}_\Psi^\xi$:

$$\dagger_2(\mathcal{C}_1''\mathbf{proc}(x_\delta[c], x_\delta^c \leftarrow w_\gamma^c @ d_1), \mathcal{D}_1', \mathcal{F}_1'; \mathcal{C}_2''\mathbf{proc}(x_\delta[c], x_\delta^c \leftarrow w_\gamma^c @ d_1), \mathcal{D}_2', \mathcal{F}_2') \in \mathcal{V}_\Psi^\xi[\![\Delta, x_\delta{:}A[c] \Vdash K]\!]_{u';\cdot}^{m'+1}.$$

$\square$

### C.3.2 Logical Equivalence

**Lemma C.20** (Reflexivity)**.** *For all security levels $\xi$ and configurations $\Psi_0; \Delta \Vdash \mathcal{D} :: x_\alpha{:}T[c]$, we have $(\Delta \Vdash \mathcal{D} :: x_\alpha{:}T[c]) \equiv_\xi^{\Psi_0} (\Delta \Vdash \mathcal{D} :: x_\alpha{:}T[c])$.*

*Proof.* Corollary of the fundamental theorem (Theorem C.19). $\square$

**Lemma C.21** (Symmetry)**.** *For all security levels $\xi$ and configurations $\Psi_0; \Delta_1 \Vdash \mathcal{D}_1 :: x_\alpha{:}T_1[c_1]$, and $\Psi_0; \Delta_2 \Vdash \mathcal{D}_2 :: y_\beta{:}T_2[c_2]$ we have if $(\Delta_1 \Vdash \mathcal{D}_1 :: x_\alpha{:}T_1[c_1]) \equiv_\xi^{\Psi_0} (\Delta_2 \Vdash \mathcal{D}_2 :: y_\beta{:}T_2[c_2])$, then $(\Delta_2 \Vdash \mathcal{D}_2 :: y_\beta{:}T_2[c_2]) \equiv_\xi^{\Psi_0} (\Delta_1 \Vdash \mathcal{D}_1 :: x_\alpha{:}T_1[c_1])$.*

*Proof.* The proof is straightforward by Definition C.4 $\square$

**Lemma C.22** (Transitivity)**.** *For all security levels $\xi$ and configurations $\Psi_0; \Delta_1 \Vdash \mathcal{D}_1 :: x_\alpha{:}T_1[c_1]$, and $\Psi_0; \Delta_2 \Vdash \mathcal{D}_2 :: y_\beta{:}T_2[c_2]$, and $\Psi_0; \Delta_3 \Vdash \mathcal{D}_3 :: z_\eta{:}T_3[c_3]$ we have if $(\Delta_1 \Vdash \mathcal{D}_1 :: x_\alpha{:}T_1[c_1]) \equiv_\xi^{\Psi_0} (\Delta_2 \Vdash \mathcal{D}_2 :: y_\beta{:}T_2[c_2])$, and $(\Delta_2 \Vdash \mathcal{D}_2 :: y_\beta{:}T_2[c_2]) \equiv_\xi^{\Psi_0} (\Delta_3 \Vdash \mathcal{D}_3 :: z_\eta{:}T_3[c_3])$, then $(\Delta_1 \Vdash \mathcal{D}_1 :: x_\alpha{:}T_1[c_1]) \equiv_\xi^{\Psi_0} (\Delta_3 \Vdash \mathcal{D}_3 :: z_\eta{:}T_3[c_3])$.*

*Proof.* We want to prove that $\Psi_0; \Delta_1 \Vdash \mathcal{D}_1 :: x_\alpha{:}3_1[c_1]$ and $\Psi_0; \Delta_3 \Vdash \mathcal{D}_3 :: z_\eta{:}T_3[c_3]$ and $\Delta_1 \Downarrow \xi = \Delta_3 \Downarrow \xi = \Delta$ and $x_\alpha{:}T_1[c_1] \Downarrow \xi = z_\eta{:}T_3[c_3] \Downarrow \xi = K$ and for all $\mathcal{C}_1, \mathcal{C}_3, \mathcal{F}_1, \mathcal{F}_3$, with $\Psi_0; \cdot \Vdash \mathcal{C}_1 :: \Delta_1$ and $\Psi_0; \cdot \Vdash \mathcal{C}_3 :: \Delta_3$ and $\Psi_0; x_\alpha{:}T_1[c_1] \Vdash \mathcal{F}_1 :: \cdot$ and $\Psi_0; z_\eta{:}T_3[c_3] \Vdash \mathcal{F}_3 :: \cdot$, and for all $m$ we have $(\mathcal{C}_1\mathcal{D}_1\mathcal{F}_1, \mathcal{C}_3\mathcal{D}_3\mathcal{F}_3) \in \mathcal{E}_{\Psi_0}^\xi[\![\Delta \Vdash K]\!]^m$ and $(\mathcal{C}_3\mathcal{D}_3\mathcal{F}_3, \mathcal{C}_1\mathcal{D}_1\mathcal{F}_1) \in \mathcal{E}_{\Psi_0}^\xi[\![\Delta \Vdash K]\!]^m$.

By assumptions we know that $\Delta_1 \Downarrow \xi = \Delta_2 \Downarrow \xi$ and $\Delta_2 \Downarrow \xi = \Delta_3 \Downarrow \xi$, which gives us $\Delta = \Delta_1 \Downarrow \xi = \Delta_2 \Downarrow \xi = \Delta_3 \Downarrow \xi$. Similarly, by assumptions we know that $x_\alpha{:}T_1[c_1] \Downarrow \xi = y_\beta{:}T_2[c_2] \Downarrow \xi$ and $y_\beta{:}T_2[c_2] \Downarrow \xi = z_\eta{:}T_3[c_3] \Downarrow \xi$, which gives us $K = x_\alpha{:}T_1[c_1] \Downarrow \xi = y_\beta{:}T_2[c_2] \Downarrow \xi = z_\eta{:}T_3[c_3] \Downarrow \xi$. Rewrite $\mathcal{C}_1\mathcal{D}_1\mathcal{F}_1$ as $\mathcal{C}_1'\mathcal{D}_1'\mathcal{F}_1'$ such that $\cdot \Vdash \mathcal{C}_1' :: \Delta$ and $\Delta \Vdash \mathcal{D}_1' :: K$ and $K \Vdash \mathcal{F}_1' :: \_$. Similarly, we can rewrite $\mathcal{C}_3\mathcal{D}_3\mathcal{F}_3$ as $\mathcal{C}_3'\mathcal{D}_3'\mathcal{F}_3'$ such that $\cdot \Vdash \mathcal{C}_3' :: \Delta$ and $\Delta \Vdash \mathcal{D}_3' :: K$ and $K \Vdash \mathcal{F}_3' :: \_$.

Consider the assumption $(\Delta_1 \Vdash \mathcal{D}_1 :: x_\alpha{:}T_1[c_1]) \equiv_\xi^{\Psi_0} (\Delta_2 \Vdash \mathcal{D}_2 :: y_\beta{:}T_2[c_2])$. We know that $\Delta_2 = \Delta, \Delta_2'$. Consider an arbitrary forest $\mathcal{C}_2''$ offering along $\Delta_2'$, then we get $\cdot \Vdash \mathcal{C}_3'\mathcal{C}_2'' :: \Delta_2$. If $c_2 \leq \xi$, then $y_\beta{:}T_2[c_2] = K$, in that case we know that $y_\beta{:}T_2[c_2] \Vdash \mathcal{F}_3' :: \_$ and we define $\mathcal{F}_2'' = \cdot$. Otherwise, we have $K = \cdot$ and $K \Vdash \mathcal{F}_3' :: \_$, and we consider an arbitrary tree $\mathcal{F}_2''$ that uses the resource $y_\beta{:}T_2[c_2]$. Now we instantiate the substitutions of $\mathcal{D}_1$ with $\mathcal{C}_1$ and $\mathcal{F}_1$ and the substitutions of $\mathcal{D}_2$ with $\mathcal{C}_3'\mathcal{C}_2''$ and $\mathcal{F}_3'\mathcal{F}_2''$. We get

$$\forall m. \, (\mathcal{C}_1'\mathcal{D}_1'\mathcal{F}_1', \mathcal{C}_3'\mathcal{C}_2''\mathcal{D}_2\mathcal{F}_2''\mathcal{F}_3') \in \mathcal{E}_{\Psi_0}^\xi[\![\Delta \Vdash K]\!]^m,$$

With a similar trick for instantiation of the assumption $(\Delta_2 \Vdash \mathcal{D}_2 :: y_\beta{:}T_2[c_2]) \equiv^{\Psi_0}_\xi (\Delta_3 \Vdash \mathcal{D}_3 :: z_\eta{:}T_3[c_3])$, we get

$$\forall m. (\mathcal{C}'_3\mathcal{C}''_2\mathcal{D}_2\mathcal{F}''_2\mathcal{F}'_3, \mathcal{C}'_3\mathcal{D}'_3\mathcal{F}'_3) \in \mathcal{E}^\xi_{\Psi_0}[\![\Delta \Vdash K]\!]^m,$$

and

$$\forall m. (\mathcal{C}'_3\mathcal{D}'_3\mathcal{F}'_3, \mathcal{C}'_3\mathcal{C}''_2\mathcal{D}_2\mathcal{F}''_2\mathcal{F}'_3) \in \mathcal{E}^\xi_{\Psi_0}[\![\Delta \Vdash K]\!]^m,$$

By the reflexivity lemma, we get

$$\forall m. (\mathcal{C}'_3\mathcal{D}'_3\mathcal{F}'_3\cdot, \mathcal{C}'_3\mathcal{C}''_2\mathcal{D}_2\mathcal{F}''_2\mathcal{F}'_3\cdot) \in \mathcal{E}^\xi_{\Psi_0}[\![K \Vdash \cdot]\!]^m,$$

and for every tree $\mathcal{A}^i_3$ in the forest $\mathcal{C}'_3$, such that $\cdot \Vdash \mathcal{A}^i_3 :: d_i{:}A_i$ where $\Delta = \{d_i{:}A_i\}_{i \in I}$ we have

$$(\cdot\mathcal{A}^i_3\mathcal{C}^{i'}_3\mathcal{D}'_3\mathcal{F}'_3\cdot, \cdot\mathcal{A}^i_3\mathcal{C}^{i'}_2\mathcal{D}'_2\mathcal{F}'_2\cdot) \in \mathcal{E}^\xi_{\Psi_0}[\![\cdot \Vdash d_i{:}A_i]\!]^m.$$

With these assumptions we can apply Lemma C.23 to get

$$\forall m. (\mathcal{C}_1\mathcal{D}_1\mathcal{F}_1, \mathcal{C}_3\mathcal{D}_3\mathcal{F}_3) \in \mathcal{E}^\xi_{\Psi_0}[\![\Delta \Vdash K]\!]^m.$$

With a similar reasoning we get

$$\forall m. (\mathcal{C}_3\mathcal{D}_3\mathcal{F}_3, \mathcal{C}_1\mathcal{D}_1\mathcal{F}_1) \in \mathcal{E}^\xi_{\Psi_0}[\![\Delta \Vdash K]\!]^m.$$

□

**Lemma C.23** (Generalized transitivity). *If*

$$\dagger_1 \forall k. (\mathcal{C}_1\mathcal{D}_1\mathcal{F}_1; \mathcal{C}_2\mathcal{D}_2\mathcal{F}_2) \in \mathcal{E}^\xi_{\Psi_0}[\![\Delta \Vdash K]\!]^k$$

*and*

$$\dagger_2 \forall k. (\mathcal{C}_2\mathcal{D}_2\mathcal{F}_2; \mathcal{C}_3\mathcal{D}_3\mathcal{F}_3) \in \mathcal{E}^\xi_{\Psi_0}[\![\Delta \Vdash K]\!]^k$$

*and*

$$\dagger_3 \forall k. (\mathcal{C}_3\mathcal{D}_3\mathcal{F}_3; \mathcal{C}_2\mathcal{D}_2\mathcal{F}_2) \in \mathcal{E}^\xi_{\Psi_0}[\![\Delta \Vdash K]\!]^k$$

*such that considering $\Delta = \{d_\ell{:}T_\ell\}_{\ell \in I}$, for every tree $\cdot \vdash \mathcal{A}^\ell_i :: d_\ell{:}T_\ell$ in the forest $\mathcal{C}_i$ (for $i = 2, 3$) we have*

$$\dagger^\ell_4 \forall k. (\cdot\mathcal{A}^\ell_3\mathcal{C}^\ell_3\mathcal{D}_3\mathcal{F}_3; \cdot\mathcal{A}^\ell_2\mathcal{C}^\ell_2\mathcal{D}_2\mathcal{F}_2) \in \mathcal{E}^\xi_{\Psi_0}[\![\cdot \Vdash d_\ell{:}T_\ell]\!]^k$$

*where $\mathcal{C}^\ell_i$ is defined as all trees in the forest $\mathcal{C}_i$ except the tree $\mathcal{A}^\ell_i$, and*

$$\dagger_5 \forall k. (\mathcal{C}_3\mathcal{D}_3\mathcal{F}_3\cdot; \mathcal{C}_2\mathcal{D}_2\mathcal{F}_2\cdot) \in \mathcal{E}^\xi_{\Psi_0}[\![K \Vdash \cdot]\!]^k$$

*then we get*

$$\forall k. (\mathcal{C}_1\mathcal{D}_1\mathcal{F}_1; \mathcal{C}_3\mathcal{D}_3\mathcal{F}_3) \in \mathcal{E}^\xi_{\Psi_0}[\![\Delta \Vdash K]\!]^k.$$

*Proof.* We need to prove for all $k$

$$\star (\mathcal{C}_1\mathcal{D}_1\mathcal{F}_1; \mathcal{C}_3\mathcal{D}_3\mathcal{F}_3) \in \mathcal{E}^\xi_{\Psi_0}[\![\Delta \Vdash K]\!]^k.$$

by induction on $k$. If $k = 0$, the proof trivially holds. We consider $k = m + 1$.

We can rewrite $\star$ as

$\star'$ if $\mathcal{C}_1, \mathcal{D}_1, \mathcal{F}_1 \mapsto^{*\Theta, \Upsilon_1}_{\Delta \Vdash K} \mathcal{C}_1, \mathcal{D}'_1, \mathcal{F}_1$, then
$\exists \mathcal{D}'_3. \mathcal{C}_3, \mathcal{D}_3, \mathcal{F}_3 \mapsto^{*\Theta'; \Upsilon_1}_{\Delta \Vdash K} \mathcal{C}_3, \mathcal{D}'_3, \mathcal{F}_3$, and $\forall \mathcal{C}'_1, \mathcal{F}'_1, \mathcal{C}'_3, \mathcal{F}'_3.$ if $\mathcal{C}_1, \mathcal{D}'_1, \mathcal{F}_1 \mapsto^{*\Theta_1, \Upsilon_1}_{\Delta \Vdash K} \mathcal{C}'_1, \mathcal{D}'_1, \mathcal{F}'_1$ and $\mathcal{C}_3, \mathcal{D}'_3, \mathcal{F}_3 \mapsto^{*\Theta_3; \Upsilon_1}_{\Delta \Vdash K} \mathcal{C}'_3, \mathcal{D}'_3, \mathcal{F}'_3$, then
$\forall x \in \Upsilon_1. (\mathcal{C}'_1, \mathcal{D}'_1, \mathcal{F}'_1; \mathcal{C}'_3, \mathcal{D}'_3, \mathcal{F}'_3) \in \mathcal{V}^\xi_\Psi[\![\Delta \Vdash K]\!]^{m+1}_{\cdot;x}$ and
$\forall x \in (\Theta_1 \cap \Theta_3). (\mathcal{C}'_1, \mathcal{D}'_1, \mathcal{F}'_1; \mathcal{C}'_3, \mathcal{D}'_3, \mathcal{F}'_3) \in \mathcal{V}^\xi_\Psi[\![\Delta \Vdash K]\!]^{m+1}_{x;\cdot}$.

Assume $\circ_1 \mathcal{C}_1, \mathcal{D}_1, \mathcal{F}_1 \mapsto^{*\Theta, \Upsilon_1}_{\Delta \Vdash K} \mathcal{C}_1, \mathcal{D}'_1, \mathcal{F}_1$. We want to show

$\star''$ $\exists \mathcal{D}'_3. \mathcal{C}_3, \mathcal{D}_3, \mathcal{F}_3 \mapsto^{*\Theta'; \Upsilon_1}_{\Delta \Vdash K} \mathcal{C}_3, \mathcal{D}'_3, \mathcal{F}_3$, and $\forall \mathcal{C}'_1, \mathcal{F}'_1, \mathcal{C}'_3, \mathcal{F}'_3.$ if $\mathcal{C}_1, \mathcal{D}'_1, \mathcal{F}_1 \mapsto^{*\Theta_1, \Upsilon_1}_{\Delta \Vdash K} \mathcal{C}'_1, \mathcal{D}'_1, \mathcal{F}'_1$ and $\mathcal{C}_3, \mathcal{D}'_3, \mathcal{F}_3 \mapsto^{*\Theta_3; \Upsilon_1}_{\Delta \Vdash K} \mathcal{C}'_3, \mathcal{D}'_3, \mathcal{F}'_3$, then
$\forall x \in \Upsilon_1. (\mathcal{C}'_1, \mathcal{D}'_1, \mathcal{F}'_1; \mathcal{C}'_3, \mathcal{D}'_3, \mathcal{F}'_3) \in \mathcal{V}^\xi_\Psi[\![\Delta \Vdash K]\!]^{m+1}_{\cdot;x}$ and
$\forall x \in (\Theta_1 \cap \Theta_3). (\mathcal{C}'_1, \mathcal{D}'_1, \mathcal{F}'_1; \mathcal{C}'_3, \mathcal{D}'_3, \mathcal{F}'_3) \in \mathcal{V}^\xi_\Psi[\![\Delta \Vdash K]\!]^{m+1}_{x;\cdot}$.

47

By †$_1$, and ○$_1$, we get

†$'_1$ $\forall k. \exists \mathcal{D}'_2. \mathcal{C}_2, \mathcal{D}_2, \mathcal{F}_2 \mapsto_{\Delta \Vdash K}^{*\Theta';\Upsilon_1} \mathcal{C}_2, \mathcal{D}'_2, \mathcal{F}_2$, and $\forall \mathcal{C}'_1, \mathcal{F}'_1, \mathcal{C}'_2, \mathcal{F}'_2$.if $\mathcal{C}_1, \mathcal{D}'_1, \mathcal{F}_1 \mapsto_{\Delta \Vdash K}^{*\Theta_1,\Upsilon_1} \mathcal{C}'_1, \mathcal{D}'_1, \mathcal{F}'_1$ and $\mathcal{C}_2, \mathcal{D}'_2, \mathcal{F}_2 \mapsto_{\Delta \Vdash K}^{*\Theta_2;\Upsilon_1} \mathcal{C}'_2, \mathcal{D}'_2, \mathcal{F}'_2$, then
$\forall x \in \Upsilon_1. (\mathcal{C}'_1, \mathcal{D}'_1, \mathcal{F}'_1; \mathcal{C}'_2, \mathcal{D}'_2, \mathcal{F}'_2) \in \mathcal{V}^\xi_\Psi [\![\Delta \Vdash K]\!]^k_{\cdot;x}$ and
$\forall x \in (\Theta_1 \cap \Theta_2). (\mathcal{C}'_1, \mathcal{D}'_1, \mathcal{F}'_1; \mathcal{C}'_2, \mathcal{D}'_2, \mathcal{F}'_2) \in \mathcal{V}^\xi_\Psi [\![\Delta \Vdash K]\!]^k_{x;\cdot}$.

By lemma C.14 we know that there exists a $\mathcal{D}'_2$ that is the closest sending configuration to $\mathcal{D}_2$ along $\Upsilon_1$ such that ○$_2$ $\mathcal{C}_2, \mathcal{D}_2, \mathcal{F}_2 \mapsto_{\Delta \Vdash K}^{*\Theta';\Upsilon_1} \mathcal{C}_2, \mathcal{D}'_2, \mathcal{F}_2$ and

†$''_1$ $\forall \mathcal{C}'_1, \mathcal{F}'_1, \mathcal{C}'_2, \mathcal{F}'_2$.if $\mathcal{C}_1, \mathcal{D}'_1, \mathcal{F}_1 \mapsto_{\Delta \Vdash K}^{*\Theta_1,\Upsilon_1} \mathcal{C}'_1, \mathcal{D}'_1, \mathcal{F}'_1$ and $\mathcal{C}_2, \mathcal{D}'_2, \mathcal{F}_2 \mapsto_{\Delta \Vdash K}^{*\Theta_2;\Upsilon_1} \mathcal{C}'_2, \mathcal{D}'_2, \mathcal{F}'_2$, then
$\forall x \in \Upsilon_1. \forall k. (\mathcal{C}'_1, \mathcal{D}'_1, \mathcal{F}'_1; \mathcal{C}'_2, \mathcal{D}'_2, \mathcal{F}'_2) \in \mathcal{V}^\xi_\Psi [\![\Delta \Vdash K]\!]^k_{\cdot;x}$ and
$\forall x \in (\Theta_1 \cap \Theta_2). \forall k. (\mathcal{C}'_1, \mathcal{D}'_1, \mathcal{F}'_1; \mathcal{C}'_2, \mathcal{D}'_2, \mathcal{F}'_2) \in \mathcal{V}^\xi_\Psi [\![\Delta \Vdash K]\!]^k_{x;\cdot}$.

By †$_2$ and ○$_2$, we get

†$'_2$ $\forall k. \exists \mathcal{D}'_3. \mathcal{C}_3, \mathcal{D}_3, \mathcal{F}_3 \mapsto_{\Delta \Vdash K}^{*\Theta';\Upsilon_1} \mathcal{C}_3, \mathcal{D}'_3, \mathcal{F}_3$, and $\forall \mathcal{C}'_2, \mathcal{F}'_2, \mathcal{C}'_3, \mathcal{F}'_3$.if $\mathcal{C}_2, \mathcal{D}'_2, \mathcal{F}_2 \mapsto_{\Delta \Vdash K}^{*\Theta_2,\Upsilon_1} \mathcal{C}'_2, \mathcal{D}'_2, \mathcal{F}'_2$ and $\mathcal{C}_3, \mathcal{D}'_3, \mathcal{F}_3 \mapsto_{\Delta \Vdash K}^{*\Theta_3;\Upsilon_1} \mathcal{C}'_3, \mathcal{D}'_3, \mathcal{F}'_3$, then
$\forall x \in \Upsilon_1. (\mathcal{C}'_2, \mathcal{D}'_2, \mathcal{F}'_2; \mathcal{C}'_3, \mathcal{D}'_3, \mathcal{F}'_3) \in \mathcal{V}^\xi_\Psi [\![\Delta \Vdash K]\!]^k_{\cdot;x}$ and
$\forall x \in (\Theta_2 \cap \Theta_3). (\mathcal{C}'_2, \mathcal{D}'_2, \mathcal{F}'_2; \mathcal{C}'_3, \mathcal{D}'_3, \mathcal{F}'_3) \in \mathcal{V}^\xi_\Psi [\![\Delta \Vdash K]\!]^k_{x;\cdot}$.

By lemma C.14 we know that there exists a $\mathcal{D}'_3$ that is the closest sending configuration to $\mathcal{D}_3$ along $\Upsilon_1$ such that ○$_3$ $\mathcal{C}_3, \mathcal{D}_3, \mathcal{F}_3 \mapsto_{\Delta \Vdash K}^{*\Theta';\Upsilon_1} \mathcal{C}_3, \mathcal{D}'_3, \mathcal{F}_3$ and

†$''_2$ $\forall \mathcal{C}'_2, \mathcal{F}'_2, \mathcal{C}'_3, \mathcal{F}'_3$.if $\mathcal{C}_2, \mathcal{D}'_2, \mathcal{F}_2 \mapsto_{\Delta \Vdash K}^{*\Theta_2,\Upsilon_1} \mathcal{C}'_2, \mathcal{D}'_2, \mathcal{F}'_2$ and $\mathcal{C}_3, \mathcal{D}'_3, \mathcal{F}_3 \mapsto_{\Delta \Vdash K}^{*\Theta_3;\Upsilon_1} \mathcal{C}'_3, \mathcal{D}'_3, \mathcal{F}'_3$, then
$\forall x \in \Upsilon_1. \forall k. (\mathcal{C}'_2, \mathcal{D}'_2, \mathcal{F}'_2; \mathcal{C}'_3, \mathcal{D}'_3, \mathcal{F}'_3) \in \mathcal{V}^\xi_\Psi [\![\Delta \Vdash K]\!]^k_{\cdot;x}$ and
$\forall x \in (\Theta_2 \cap \Theta_3). \forall k. (\mathcal{C}'_2, \mathcal{D}'_2, \mathcal{F}'_2; \mathcal{C}'_3, \mathcal{D}'_3, \mathcal{F}'_3) \in \mathcal{V}^\xi_\Psi [\![\Delta \Vdash K]\!]^k_{x;\cdot}$.

By †$_3$ and ○$_3$, we get

†$'_3$ $\forall k. \exists \mathcal{D}'_2. \mathcal{C}_2, \mathcal{D}_2, \mathcal{F}_2 \mapsto_{\Delta \Vdash K}^{*\Theta';\Upsilon_1} \mathcal{C}_2, \mathcal{D}'_2, \mathcal{F}_2$, and $\forall \mathcal{C}'_3, \mathcal{F}'_3, \mathcal{C}'_2, \mathcal{F}'_2$.if $\mathcal{C}_3, \mathcal{D}'_3, \mathcal{F}_3 \mapsto_{\Delta \Vdash K}^{*\Theta_3,\Upsilon_1} \mathcal{C}'_3, \mathcal{D}'_3, \mathcal{F}'_3$ and $\mathcal{C}_2, \mathcal{D}'_2, \mathcal{F}_2 \mapsto_{\Delta \Vdash K}^{*\Theta_2;\Upsilon_1} \mathcal{C}'_2, \mathcal{D}'_2, \mathcal{F}'_2$, then
$\forall x \in \Upsilon_1. (\mathcal{C}'_3, \mathcal{D}'_3, \mathcal{F}'_3; \mathcal{C}'_2, \mathcal{D}'_2, \mathcal{F}'_2) \in \mathcal{V}^\xi_\Psi [\![\Delta \Vdash K]\!]^k_{\cdot;x}$ and
$\forall x \in (\Theta_3 \cap \Theta_2). (\mathcal{C}'_3, \mathcal{D}'_3, \mathcal{F}'_3; \mathcal{C}'_2, \mathcal{D}'_2, \mathcal{F}'_2) \in \mathcal{V}^\xi_\Psi [\![\Delta \Vdash K]\!]^k_{x;\cdot}$.

By lemma C.14 we know that for $\mathcal{D}'_2$ that is the closest sending configuration to $\mathcal{D}_2$ along $\Upsilon_1$ (By the uniqueness of such configuration, it is the same configuration we got in ○$_2$) we have

†$''_3$ $\forall \mathcal{C}'_3, \mathcal{F}'_3, \mathcal{C}'_2, \mathcal{F}'_2$.if $\mathcal{C}_3, \mathcal{D}'_3, \mathcal{F}_3 \mapsto_{\Delta \Vdash K}^{*\Theta_3,\Upsilon_1} \mathcal{C}'_3, \mathcal{D}'_3, \mathcal{F}'_3$ and $\mathcal{C}_2, \mathcal{D}'_2, \mathcal{F}_2 \mapsto_{\Delta \Vdash K}^{*\Theta_2;\Upsilon_1} \mathcal{C}'_2, \mathcal{D}'_2, \mathcal{F}'_2$, then
$\forall x \in \Upsilon_1. \forall k. (\mathcal{C}'_3, \mathcal{D}'_3, \mathcal{F}'_3; \mathcal{C}'_2, \mathcal{D}'_2, \mathcal{F}'_2) \in \mathcal{V}^\xi_\Psi [\![\Delta \Vdash K]\!]^k_{\cdot;x}$ and
$\forall x \in (\Theta_3 \cap \Theta_2). \forall k. (\mathcal{C}'_3, \mathcal{D}'_3, \mathcal{F}'_3; \mathcal{C}'_2, \mathcal{D}'_2, \mathcal{F}'_2) \in \mathcal{V}^\xi_\Psi [\![\Delta \Vdash K]\!]^k_{x;\cdot}$.

We can instantiate the goal $\star''$ with $\mathcal{D}'_3$ and ○$_3$ to get the new goal

$\star_1$ $\forall \mathcal{C}'_1, \mathcal{F}'_1, \mathcal{C}'_3, \mathcal{F}'_3$.if $\mathcal{C}_1, \mathcal{D}'_1, \mathcal{F}_1 \mapsto_{\Delta \Vdash K}^{*\Theta_1,\Upsilon_1} \mathcal{C}'_1, \mathcal{D}'_1, \mathcal{F}'_1$ and $\mathcal{C}_3, \mathcal{D}'_3, \mathcal{F}_3 \mapsto_{\Delta \Vdash K}^{*\Theta_3;\Upsilon_1} \mathcal{C}'_3, \mathcal{D}'_3, \mathcal{F}'_3$, then
$\forall x \in \Upsilon_1. (\mathcal{C}'_1, \mathcal{D}'_1, \mathcal{F}'_1; \mathcal{C}'_3, \mathcal{D}'_3, \mathcal{F}'_3) \in \mathcal{V}^\xi_\Psi [\![\Delta \Vdash K]\!]^{m+1}_{\cdot;x}$ and
$\forall x \in (\Theta_1 \cap \Theta_3). (\mathcal{C}'_1, \mathcal{D}'_1, \mathcal{F}'_1; \mathcal{C}'_3, \mathcal{D}'_3, \mathcal{F}'_3) \in \mathcal{V}^\xi_\Psi [\![\Delta \Vdash K]\!]^{m+1}_{x;\cdot}$.

Next, we instantiate the $\forall$ quantifiers in the goal with arbitrary $\mathcal{C}'_1, \mathcal{F}'_1, \mathcal{C}'_3, \mathcal{F}'_3$ (also implicitly $\Theta_1, \Theta_3$) such that ○$_4$ $\mathcal{C}_1, \mathcal{D}'_1, \mathcal{F}_1 \mapsto_{\Delta \Vdash K}^{*\Theta_1,\Upsilon_1} \mathcal{C}'_1, \mathcal{D}'_1, \mathcal{F}'_1$ and ○$_5$ $\mathcal{C}_3, \mathcal{D}'_3, \mathcal{F}_3 \mapsto_{\Delta \Vdash K}^{*\Theta_3;\Upsilon_1} \mathcal{C}'_3, \mathcal{D}'_3, \mathcal{F}'_3$. Our new goal is to show

$\star_2$ $\forall x \in \Upsilon_1. (\mathcal{C}'_1, \mathcal{D}'_1, \mathcal{F}'_1; \mathcal{C}'_3, \mathcal{D}'_3, \mathcal{F}'_3) \in \mathcal{V}^\xi_\Psi [\![\Delta \Vdash K]\!]^{m+1}_{\cdot;x}$ and
$\forall x \in (\Theta_1 \cap \Theta_3). (\mathcal{C}'_1, \mathcal{D}'_1, \mathcal{F}'_1; \mathcal{C}'_3, \mathcal{D}'_3, \mathcal{F}'_3) \in \mathcal{V}^\xi_\Psi [\![\Delta \Vdash K]\!]^{m+1}_{x;\cdot}$.

Consider $\Delta = \{d_\ell : T_\ell\}_{\ell \in I}$, and $\mathcal{C}_3 = \{\mathcal{A}^\ell_3\}_{\ell \in I}$. we break down $\Theta_3$ into disjoint sets of channels $\Theta^\ell_3$ for the channels with messages incoming from $\mathcal{A}^\ell_3$ and $\Theta^\mathcal{F}_3$ for the channels with messages incoming from $\mathcal{F}_3$. By the definition of internal computations, the fact that these trees are all disjoint, and ○$_5$, we get ○$^\ell_5 \cdot, \mathcal{A}^\ell_3, \mathcal{C}^\ell_3 \mathcal{D}_3 \mathcal{F}_3 \mapsto_{\cdot \Vdash T^\ell}^{*\Theta';\Theta^\ell_3} \cdot, \mathcal{A}^{\ell'}_3, \mathcal{C}^\ell_3 \mathcal{D}_3 \mathcal{F}_3$, where $\mathcal{C}'_3 = \{\mathcal{A}^{\ell'}_3\}_{\ell \in I}$. Also ○$^\mathcal{F}_5$ $\mathcal{C}_3 \mathcal{D}_3, \mathcal{F}_3, \cdot \mapsto_{K \Vdash \cdot}^{*\Theta'';\Theta^\mathcal{F}_3} \mathcal{C}_3 \mathcal{D}_3, \mathcal{F}'_3, \cdot$.

We apply each ○$^\ell_5$ on †$^\ell_4$ to get

48

$\dagger_4^{\ell'}$ $\forall k. \exists \mathcal{A}_2^{\ell'}. \cdot, \mathcal{A}_2^{\ell}, \mathcal{C}_2^{\ell} \mathcal{D}_2 \mathcal{F}_2 \mapsto_{\cdot \Vdash T^{\ell}}^{*\Theta; \Theta_3^{\ell}} \cdot, \mathcal{A}_2^{\ell'}, \mathcal{C}_2^{\ell} \mathcal{D}_2 \mathcal{F}_2$, and $\forall \mathcal{F}_4^{\ell}, \mathcal{F}_5^{\ell}$.if $\cdot, \mathcal{A}_3^{\ell'}, \mathcal{C}_3^{\ell} \mathcal{D}_3 \mathcal{F}_3 \mapsto_{\Delta \Vdash K}^{*\Theta_4, \Theta_3^{\ell}} \cdot, \mathcal{A}_3^{\ell'}, \mathcal{F}_4^{\ell}$ and $\cdot, \mathcal{A}_2^{\ell'}, \mathcal{C}_2^{\ell} \mathcal{D}_2 \mathcal{F}_2 \mapsto_{\Delta \Vdash K}^{*\Theta_5, \Theta_3^{\ell}} \cdot, \mathcal{A}_2^{\ell'}, \mathcal{F}_5^{\ell}$, then
$\forall x \in \Theta_3^{\ell}. (\cdot, \mathcal{A}_3^{\ell'}, \mathcal{F}_4^{\ell}; \cdot, \mathcal{A}_2^{\ell'}, \mathcal{F}_5^{\ell}) \in \mathcal{V}_{\Psi}^{\xi} [\![ \cdot \Vdash T^{\ell} ]\!]_{\cdot;x}^{k}$ and
$\forall x \in (\Theta_4 \cap \Theta_5). \cdot, \mathcal{A}_3^{\ell'}, \mathcal{F}_4^{\ell}; \cdot, \mathcal{A}_2^{\ell'}, \mathcal{F}_5^{\ell}) \in \mathcal{V}_{\Psi}^{\xi} [\![ \cdot \Vdash T^{\ell} ]\!]_{x;\cdot}^{k}$.

By lemma C.14 we know that there exists a $\mathcal{A}_2^{\ell'}$ that is the closest sending configuration to $\mathcal{A}_2^{\ell}$ along $\Theta_3^{\ell}$ such that $\circ_6 \cdot, \mathcal{A}_2^{\ell}, \mathcal{C}_2^{\ell} \mathcal{D}_2 \mathcal{F}_2 \mapsto_{\cdot \Vdash T^{\ell}}^{*\Theta; \Theta_3^{\ell}} \cdot, \mathcal{A}_2^{\ell'}, \mathcal{C}_2^{\ell} \mathcal{D}_2 \mathcal{F}_2$ and

$\dagger_4^{\ell''}$ $\forall \mathcal{F}_4^{\ell}, \mathcal{F}_5^{\ell}$.if $\cdot, \mathcal{A}_3^{\ell'}, \mathcal{C}_3^{\ell} \mathcal{D}_3 \mathcal{F}_3 \mapsto_{\Delta \Vdash K}^{*\Theta_4, \Theta_3^{\ell}} \cdot, \mathcal{A}_3^{\ell'}, \mathcal{F}_4^{\ell}$ and $\cdot, \mathcal{A}_2^{\ell'}, \mathcal{C}_2^{\ell} \mathcal{D}_2 \mathcal{F}_2 \mapsto_{\Delta \Vdash K}^{*\Theta_5, \Theta_3^{\ell}} \cdot, \mathcal{A}_2^{\ell'}, \mathcal{F}_5^{\ell}$, then
$\forall x \in \Theta_3^{\ell}. \forall k. (\cdot, \mathcal{A}_3^{\ell'}, \mathcal{F}_4^{\ell}; \cdot, \mathcal{A}_2^{\ell'}, \mathcal{F}_5^{\ell}) \in \mathcal{V}_{\Psi}^{\xi} [\![ \cdot \Vdash T^{\ell} ]\!]_{\cdot;x}^{k}$ and
$\forall x \in (\Theta_4 \cap \Theta_5). \forall k. \cdot, \mathcal{A}_3^{\ell'}, \mathcal{F}_4^{\ell}; \cdot, \mathcal{A}_2^{\ell'}, \mathcal{F}_5^{\ell}) \in \mathcal{V}_{\Psi}^{\xi} [\![ \cdot \Vdash T^{\ell} ]\!]_{x;\cdot}^{k}$.

We build $\mathcal{C}_2'$ as a forest consisting of trees $\mathcal{A}_2^{\ell'}$, i.e. $\mathcal{C}_2' = \{\mathcal{A}_2^{\ell'}\}_{\ell \in I}$.

Similarly from $\circ_5^{\mathcal{F}}$, and $\dagger_5$, lemma C.14 we know that there exists a $\mathcal{F}_2'$ that is the closest sending configuration to $\mathcal{F}_2$ along $\Theta_3^{\mathcal{F}}$ such that $\circ_7 \mathcal{C}_2 \mathcal{D}_2, \mathcal{F}_2, \cdot \mapsto_{K \Vdash \cdot}^{*\Theta_6; \Theta_3^{\mathcal{F}}} \mathcal{C}_2 \mathcal{D}_2, \mathcal{F}_2', \cdot$ and

$\dagger_5''$ $\forall \mathcal{C}_6, \mathcal{C}_7$.if $\mathcal{C}_3 \mathcal{D}_3, \mathcal{F}_3', \cdot \mapsto_{K \Vdash \cdot}^{*\Theta_6, \Theta_3^{\mathcal{F}}} \mathcal{C}_6, \mathcal{F}_3', \cdot \mapsto^{*\Theta_6, \Theta_3^{\mathcal{F}}}$ and $\mathcal{C}_2 \mathcal{D}_2, \mathcal{F}_2', \cdot \mapsto_{K \Vdash \cdot}^{*\Theta_7, \Theta_3^{\mathcal{F}}} \mathcal{C}_7, \mathcal{F}_2', \cdot$, then
$\forall x \in \Theta_3^{\mathcal{F}}. \forall k. (\mathcal{C}_6, \mathcal{F}_3', \cdot; \mathcal{C}_7, \mathcal{F}_2', \cdot \in \mathcal{V}_{\Psi}^{\xi} [\![ K \Vdash \cdot ]\!]_{\cdot;x}^{k}$ and
$\forall x \in (\Theta_6 \cap \Theta_7). \forall k. \mathcal{C}_6, \mathcal{F}_3', \cdot; \mathcal{C}_7, \mathcal{F}_2', \cdot) \in \mathcal{V}_{\Psi}^{\xi} [\![ K \Vdash \cdot ]\!]_{x;\cdot}^{k}$.

With the way that we built such $\mathcal{C}_2'$ and $\mathcal{F}_2'$, it is straightforward to see $\circ_8 \mathcal{C}_2, \mathcal{D}_2', \mathcal{F}_2 \mapsto_{\Delta \Vdash K}^{*\Theta_3; \Upsilon_1} \mathcal{C}_2', \mathcal{D}_2', \mathcal{F}_2'$.

We apply $\circ_4$ and $\circ_5$ on $\dagger_1''$ to get

$\dagger_1'''$ $\forall x \in \Upsilon_1. \forall k. (\mathcal{C}_1', \mathcal{D}_1', \mathcal{F}_1'; \mathcal{C}_2', \mathcal{D}_2', \mathcal{F}_2') \in \mathcal{V}_{\Psi}^{\xi} [\![ \Delta \Vdash K ]\!]_{\cdot;x}^{k}$ and
$\forall x \in (\Theta_1 \cap \Theta_3). \forall k. (\mathcal{C}_1', \mathcal{D}_1', \mathcal{F}_1'; \mathcal{C}_2', \mathcal{D}_2', \mathcal{F}_2') \in \mathcal{V}_{\Psi}^{\xi} [\![ \Delta \Vdash K ]\!]_{x;\cdot}^{k}$.

We apply $\circ_5$ and $\circ_6$ on $\dagger_2''$ to get

$\dagger_2'''$ $\forall x \in \Upsilon_1. \forall k. (\mathcal{C}_2', \mathcal{D}_2', \mathcal{F}_2'; \mathcal{C}_3', \mathcal{D}_3', \mathcal{F}_3') \in \mathcal{V}_{\Psi}^{\xi} [\![ \Delta \Vdash K ]\!]_{\cdot;x}^{k}$ and
$\forall x \in (\Theta_3 \cap \Theta_3). \forall k. (\mathcal{C}_2', \mathcal{D}_2', \mathcal{F}_2'; \mathcal{C}_3', \mathcal{D}_3', \mathcal{F}_3') \in \mathcal{V}_{\Psi}^{\xi} [\![ \Delta \Vdash K ]\!]_{x;\cdot}^{k}$.

We apply $\circ_4$ and $\circ_5$ on $\dagger_3''$ to get

$\dagger_3'''$ $\forall x \in \Upsilon_1. \forall k. (\mathcal{C}_3', \mathcal{D}_3', \mathcal{F}_3'; \mathcal{C}_2', \mathcal{D}_2', \mathcal{F}_2') \in \mathcal{V}_{\Psi}^{\xi} [\![ \Delta \Vdash K ]\!]_{\cdot;x}^{k}$ and
$\forall x \in (\Theta_3 \cap \Theta_3). \forall k. (\mathcal{C}_3', \mathcal{D}_3', \mathcal{F}_3'; \mathcal{C}_2', \mathcal{D}_2', \mathcal{F}_2') \in \mathcal{V}_{\Psi}^{\xi} [\![ \Delta \Vdash K ]\!]_{x;\cdot}^{k}$.

Finally, we define $\mathcal{F}_4^{\ell} = \mathcal{C}_3^{\ell'} \mathcal{D}_3' \mathcal{F}_3'$, and $\mathcal{F}_5^{\ell} = \mathcal{C}_2^{\ell'} \mathcal{D}_2' \mathcal{F}_2'$ and $\mathcal{C}^6 = \mathcal{C}_3' \mathcal{D}_3'$ and $\mathcal{C}^7 = \mathcal{C}_2' \mathcal{D}_2'$. Observe that in this case $\Theta_4 = \Theta_5$ and $\Theta_6 = \Theta_7$ and we get:

$\dagger_4^{\ell'''}$ $\forall x \in \Theta_3^{\ell}. \forall k. (\cdot, \mathcal{A}_3^{\ell'}, \mathcal{C}_3^{\ell'} \mathcal{D}_3' \mathcal{F}_3'; \cdot, \mathcal{A}_2^{\ell'}, \mathcal{C}_2^{\ell'} \mathcal{D}_2' \mathcal{F}_2') \in \mathcal{V}_{\Psi}^{\xi} [\![ \cdot \Vdash T^{\ell} ]\!]_{\cdot;x}^{k}$ and
$\forall x \in (\Theta_4). \forall k. \cdot, \mathcal{A}_3^{\ell'}, \mathcal{C}_3^{\ell'} \mathcal{D}_3' \mathcal{F}_3'; \cdot, \mathcal{A}_2^{\ell'}, \mathcal{C}_2^{\ell'} \mathcal{D}_2' \mathcal{F}_2') \in \mathcal{V}_{\Psi}^{\xi} [\![ \cdot \Vdash T^{\ell} ]\!]_{x;\cdot}^{k}$.

and

$\dagger_5'''$ $\forall x \in \Theta_3^{\mathcal{F}}. \forall k. (\mathcal{C}_3' \mathcal{D}_3', \mathcal{F}_3', \cdot; \mathcal{C}_2' \mathcal{D}_2', \mathcal{F}_2', \cdot \in \mathcal{V}_{\Psi}^{\xi} [\![ K \Vdash \cdot ]\!]_{\cdot;x}^{k}$ and
$\forall x \in (\Theta_6 \cap \Theta_7). \forall k. \mathcal{C}_3' \mathcal{D}_3', \mathcal{F}_3', \cdot; \mathcal{C}_2' \mathcal{D}_2', \mathcal{F}_2', \cdot) \in \mathcal{V}_{\Psi}^{\xi} [\![ K \Vdash \cdot ]\!]_{x;\cdot}^{k}$.

Now that we gathered all the needed assumptions, let's get back to our goal $\star_2$:

$\star_2$ $\forall x \in \Upsilon_1. (\mathcal{C}_1', \mathcal{D}_1', \mathcal{F}_1'; \mathcal{C}_3', \mathcal{D}_3', \mathcal{F}_3') \in \mathcal{V}_{\Psi}^{\xi} [\![ \Delta \Vdash K ]\!]_{\cdot;x}^{m+1}$ and
$\forall x \in (\Theta_1 \cap \Theta_3). (\mathcal{C}_1', \mathcal{D}_1', \mathcal{F}_1'; \mathcal{C}_3', \mathcal{D}_3', \mathcal{F}_3') \in \mathcal{V}_{\Psi}^{\xi} [\![ \Delta \Vdash K ]\!]_{x;\cdot}^{m+1}$.

We have two cases to consider

**Case 1.** $y_{\alpha} \in \Upsilon_1$. In this case our goal is to prove

$\star_2'$ $(\mathcal{C}_1', \mathcal{D}_1', \mathcal{F}_1'; \mathcal{C}_3', \mathcal{D}_3', \mathcal{F}_3') \in \mathcal{V}_{\Psi}^{\xi} [\![ \Delta \Vdash K ]\!]_{\cdot;y_{\alpha}}^{m+1}$

49

We consider subcases based on the type of $y^\alpha{:}T$ and the message or forwarding process along $y^\alpha$ that is in $\mathcal{D}_1$. The idea is to rewrite all three assumptions accordingly based on the message/forwarding process and the type of $y^\alpha{:}T$. Here, we only provide three interesting cases. The proof of other cases is similar.

**Subcase 1.** $K = y_c^\alpha{:}A\otimes B$ and a message along $y^\alpha$ is in $\mathcal{D}_1'$, i.e. $\mathcal{D}_1' = \mathcal{D}_1''\mathcal{T}_1\mathbf{msg}(\mathbf{send}x_\beta^c\ y_\alpha^c; y_\alpha^c \leftarrow u_\delta^c)$ for $\mathcal{T}_1 \in \mathsf{Tree}_\Psi(\Delta'' \Vdash x_\beta{:}A[c])$, and $\Delta = \Delta', \Delta''$, and $\mathcal{C}_1' = \mathcal{C}_1^1\mathcal{C}_1^2$ for $\cdot \Vdash \mathcal{C}_1^1 :: \Delta'$ and $\cdot \Vdash \mathcal{C}_1^2 :: \Delta''$.

By $\dagger_1'''$, we know that $\mathcal{C}_2' = \mathcal{C}_2^1\mathcal{C}_2^2$ for $\cdot \Vdash \mathcal{C}_2^1 :: \Delta'$ and $\cdot \Vdash \mathcal{C}_2^2 :: \Delta''$ and

$\ddagger_1$ $\mathcal{D}_2' = \mathcal{D}_2''\mathcal{T}_2\mathbf{msg}(\mathbf{send}, x_\beta^c\ y_\alpha^c; y_\alpha^c \leftarrow u_\delta^c)$ for $\mathcal{T}_2 \in \mathsf{Tree}_\Psi(\Delta'' \Vdash x_\beta{:}A[c])$ and

$\ddagger_2$ $\forall k.\,(\mathcal{C}_1^2\mathcal{T}_1\mathcal{C}_1^1\mathcal{D}_1''\mathbf{msg}(\mathbf{send}x_\beta^c\ y_\alpha^c; y_\alpha^c \leftarrow u_\delta^c)\mathcal{F}_1', \mathcal{C}_2^2\mathcal{T}_2\mathcal{C}_2^1\mathcal{D}_2''\mathbf{msg}(\mathbf{send}x_\beta^c\ y_\alpha^c; y_\alpha^c \leftarrow u_\delta^c)\mathcal{F}_2') \in \mathcal{E}_\Psi^\xi[\![\Delta'' \Vdash x_\beta{:}A[c]]\!]^k$ and

$\ddagger_3$ $\forall k.\,(\mathcal{C}_1^1\mathcal{D}_1''\mathcal{C}_1^2\mathcal{T}_1\mathbf{msg}(\mathbf{send}x_\beta^c\ y_\alpha^c; y_\alpha^c \leftarrow u_\delta^c)\mathcal{F}_1', \mathcal{C}_2^1\mathcal{D}_2''\mathcal{C}_2^2\mathcal{T}_2\mathbf{msg}(\mathbf{send}x_\beta^c\ y_\alpha^c; y_\alpha^c \leftarrow u_\delta^c)\mathcal{F}_2') \in \mathcal{E}_\Psi^\xi[\![\Delta' \Vdash u_\delta{:}B[c]]\!]^k$

hold by Row 4 in the logical relation.

By $\ddagger_1$ and $\dagger_2'''$, we know that $\mathcal{C}_3' = \mathcal{C}_3^1\mathcal{C}_3^2$ for $\cdot \Vdash \mathcal{C}_3^1 :: \Delta'$ and $\cdot \Vdash \mathcal{C}_3^2 :: \Delta''$ and

$\ddagger_4$ $\mathcal{D}_3' = \mathcal{D}_3''\mathcal{T}_3\mathbf{msg}(\mathbf{send}, x_\beta^c\ y_\alpha^c; y_\alpha^c \leftarrow u_\delta^c)$ for $\mathcal{T}_3 \in \mathsf{Tree}_\Psi(\Delta'' \Vdash x_\beta{:}A[c])$ and

$\ddagger_5$ $\forall k.\,(\mathcal{C}_2^2\mathcal{T}_2\mathcal{C}_2^1\mathcal{D}_2''\mathbf{msg}(\mathbf{send}x_\beta^c\ y_\alpha^c; y_\alpha^c \leftarrow u_\delta^c)\mathcal{F}_2', \mathcal{C}_3^2\mathcal{T}_3\mathcal{C}_3^1\mathcal{D}_3''\mathbf{msg}(\mathbf{send}x_\beta^c\ y_\alpha^c; y_\alpha^c \leftarrow u_\delta^c)\mathcal{F}_3') \in \mathcal{E}_\Psi^\xi[\![\Delta'' \Vdash x_\beta{:}A[c]]\!]^k$ and

$\ddagger_6$ $\forall k.\,(\mathcal{C}_2^1\mathcal{D}_2''\mathcal{C}_2^2\mathcal{T}_2\mathbf{msg}(\mathbf{send}x_\beta^c\ y_\alpha^c; y_\alpha^c \leftarrow u_\delta^c)\mathcal{F}_2', \mathcal{C}_3^1\mathcal{D}_3''\mathcal{C}_3^2\mathcal{T}_3\mathbf{msg}(\mathbf{send}x_\beta^c\ y_\alpha^c; y_\alpha^c \leftarrow u_\delta^c)\mathcal{F}_3') \in \mathcal{E}_\Psi^\xi[\![\Delta' \Vdash u_\delta{:}B[c]]\!]^k$

hold by Row 4 in the logical relation.

By $\ddagger_4$ and $\dagger_3'''$, we know that

$\ddagger_8$ $\forall k.\,(\mathcal{C}_3^2\mathcal{T}_3\mathcal{C}_3^1\mathcal{D}_3''\mathbf{msg}(\mathbf{send}x_\beta^c\ y_\alpha^c; y_\alpha^c \leftarrow u_\delta^c)\mathcal{F}_3', \mathcal{C}_2^2\mathcal{T}_2\mathcal{C}_2^1\mathcal{D}_2''\mathbf{msg}(\mathbf{send}x_\beta^c\ y_\alpha^c; y_\alpha^c \leftarrow u_\delta^c)\mathcal{F}_2') \in \mathcal{E}_\Psi^\xi[\![\Delta'' \Vdash x_\beta{:}A[c]]\!]^k$ and

$\ddagger_9$ $\forall k.\,(\mathcal{C}_3^1\mathcal{D}_3''\mathcal{C}_3^2\mathcal{T}_3\mathbf{msg}(\mathbf{send}x_\beta^c\ y_\alpha^c; y_\alpha^c \leftarrow u_\delta^c)\mathcal{F}_3', \mathcal{C}_2^1\mathcal{D}_2''\mathcal{C}_2^2\mathcal{T}_2\mathbf{msg}(\mathbf{send}x_\beta^c\ y_\alpha^c; y_\alpha^c \leftarrow u_\delta^c)\mathcal{F}_2') \in \mathcal{E}_\Psi^\xi[\![\Delta' \Vdash u_\delta{:}B[c]]\!]^k$

hold by Row 4 in the logical relation.

By $\dagger_5'''$ and $\ddagger_1$ and $\ddagger_4$ we get

$\ddagger_{10}$ $\forall k.\,(\mathcal{C}_3^2\mathcal{C}_3^1\mathcal{D}_3''\mathcal{T}_3\mathbf{msg}(\mathbf{send}x_\beta^c\ y_\alpha^c; y_\alpha^c \leftarrow u_\delta^c)\mathcal{F}_3'\cdot, \mathcal{C}_2^2\mathcal{C}_2^1\mathcal{D}_2''\mathcal{T}_2\mathbf{msg}(\mathbf{send}x_\beta^c\ y_\alpha^c; y_\alpha^c \leftarrow u_\delta^c)\mathcal{F}_2'\cdot) \in \mathcal{E}_\Psi^\xi[\![x_\beta{:}A[c], u_\delta{:}B[c] \Vdash \cdot]\!]^k$

by Row 9 in the logical relation.

By $\dagger_4^\ell$ and the forward closure lemmas for every $\ell \in I$ we get

$\ddagger_{11}^\ell$ $\forall k.\,(\cdot\mathcal{A}_3^{\ell'}\mathcal{C}_3^{\ell'}\mathcal{D}_3'\mathcal{F}_3'; \cdot\mathcal{A}_2^{\ell'}\mathcal{C}_2^{\ell'}\mathcal{D}_2'\mathcal{F}_2') \in \mathcal{E}_{\Psi_0}^\xi[\![\cdot \Vdash d_\ell{:}T_\ell]\!]^k$.

Note that we can apply forward closure on the second run since $\mathcal{A}_2^{\ell'}$ is the minimal configuration that sends along the same channels as $\mathcal{A}_3^{\ell'}$.

By Lemma C.16, from $\ddagger_{10}$, $\ddagger_8$, and $\ddagger_{11}^\ell$ for every tree in $\mathcal{C}_3^2$ we get

$\ddagger_{12}$ $\forall k.\,(\mathcal{C}_3^1\mathcal{D}_3''\mathcal{C}_3^2\mathcal{T}_3\mathbf{msg}(\mathbf{send}x_\beta^c\ y_\alpha^c; y_\alpha^c \leftarrow u_\delta^c)\mathcal{F}_3'\cdot, \mathcal{C}_2^1\mathcal{D}_2''\mathcal{C}_2^2\mathcal{T}_2\mathbf{msg}(\mathbf{send}x_\beta^c\ y_\alpha^c; y_\alpha^c \leftarrow u_\delta^c)\mathcal{F}_2'\cdot) \in \mathcal{E}_\Psi^\xi[\![u_\delta{:}B[c] \Vdash \cdot]\!]^k$

We can apply the induction hypothesis on $m < m+1$, having the required assumptions $\ddagger_3$, $\ddagger_6$, $\ddagger_9$, $\ddagger_{11}^\ell$ relating corresponding trees in $\mathcal{C}_3^1$ and $\mathcal{C}_2^1$, and $\ddagger_{12}$ to get

$\star\star_1$ $(\mathcal{C}_1^1\mathcal{D}_1''\mathcal{C}_1^2\mathcal{T}_1\mathbf{msg}(\mathbf{send}x_\beta^c\ y_\alpha^c; y_\alpha^c \leftarrow u_\delta^c)\mathcal{F}_1', \mathcal{C}_3^1\mathcal{D}_3''\mathcal{C}_3^2\mathcal{T}_3\mathbf{msg}(\mathbf{send}x_\beta^c\ y_\alpha^c; y_\alpha^c \leftarrow u_\delta^c)\mathcal{F}_3') \in \mathcal{E}_\Psi^\xi[\![\Delta' \Vdash u_\delta{:}B[c]]\!]^m$

$\star\star_1$ is one relation required by row 4 of the logical relation to satisfy $\star_2'$. To get the other relation, we use a similar strategy. By Lemma C.16, from $\ddagger_{10}$, $\ddagger_9$, and $\ddagger_{11}^\ell$ for every tree in $\mathcal{C}_3^1$ we get

$\ddagger_{13}$ $\forall k.\,(\mathcal{C}_3^2\mathcal{T}_3\mathcal{C}_3^1\mathcal{D}_3''\mathbf{msg}(\mathbf{send}x_\beta^c\ y_\alpha^c; y_\alpha^c \leftarrow u_\delta^c)\mathcal{F}_3'\cdot, \mathcal{C}_2^2\mathcal{T}_2\mathcal{C}_2^1\mathcal{D}_2''\mathbf{msg}(\mathbf{send}x_\beta^c\ y_\alpha^c; y_\alpha^c \leftarrow u_\delta^c)\mathcal{F}_2'\cdot) \in \mathcal{E}_\Psi^\xi[\![x_\beta{:}A[c] \Vdash \cdot]\!]^k$

We can apply the induction hypothesis on $m < m+1$, having the required assumptions $\ddagger_2$, $\ddagger_5$, $\ddagger_8$, $\ddagger_{11}^\ell$ relating corresponding trees in $\mathcal{C}_3^2$ and $\mathcal{C}_2^2$, and $\ddagger_13$ to get

$\star\star_2$ $(\mathcal{C}_1^2\mathcal{T}_1\mathcal{C}_1^1\mathcal{D}_1''\mathbf{msg}(\mathbf{send}x_\beta^c\ y_\alpha^c; y_\alpha^c \leftarrow u_\delta^c)\mathcal{F}_1', \mathcal{C}_3^2\mathcal{T}_3\mathcal{C}_3^1\mathcal{D}_3''\mathbf{msg}(\mathbf{send}x_\beta^c\ y_\alpha^c; y_\alpha^c \leftarrow u_\delta^c)\mathcal{F}_3') \in \mathcal{E}_\Psi^\xi[\![\Delta'' \Vdash x_\beta{:}A[c]]\!]^m$.

This completes the proof of the subcase since we now satisfy all conditions of Row 4.

**Subcase 2.** $K = y_c^\alpha{:}A$ and a forwarding process along $y^\alpha$ is in $\mathcal{D}_1'$ i.e. $\mathcal{D}_1' = \mathcal{D}_1''\mathbf{proc}(y_\alpha^c, y_\alpha^c \leftarrow x_\beta^c @ d_1)$.
By $\dagger_1'''$, we know that

$$\ddagger_1 \; \mathcal{D}_2' = \mathcal{D}_2''\mathbf{proc}(y_\alpha^c, y_\alpha^c \leftarrow x_\beta^c @ d_2) \text{ and}$$

$$\ddagger_2 \; \forall k. \, ([x_\beta^c/y_\alpha^c]\mathcal{C}_1'\mathcal{D}_1''\mathcal{F}_1', [x_\beta^c/y_\alpha^c]\mathcal{C}_2'\mathcal{D}_2''\mathcal{F}_2') \in \mathcal{E}_\Psi^\xi [\![\Delta \Vdash x_\beta{:}A[c]]\!]^k$$

hold by Row 11 in the logical relation.
By $\ddagger_1$ and $\dagger_2'''$, we know that

$$\ddagger_3 \; \mathcal{D}_3' = \mathcal{D}_3''\mathbf{proc}(y_\alpha^c, y_\alpha^c \leftarrow x_\beta^c @ d_3) \text{ and}$$

$$\ddagger_4 \; \forall k. \, ([x_\beta^c/y_\alpha^c]\mathcal{C}_2'\mathcal{D}_2''\mathcal{F}_2', [x_\beta^c/y_\alpha^c]\mathcal{C}_3'\mathcal{D}_3''\mathcal{F}_3') \in \mathcal{E}_{\Psi_0}^\xi [\![\Delta \Vdash x_\beta{:}A[c]]\!]^k$$

hold by Row 11 in the logical relation.
By $\ddagger_3$ and $\dagger_3'''$, we know that

$$\ddagger_5 \; \forall k. \, ([x_\beta^c/y_\alpha^c]\mathcal{C}_3'\mathcal{D}_3''\mathcal{F}_3', \mathcal{C}_2'\mathcal{D}_2''\mathcal{F}_2') \in \mathcal{E}_{\Psi_0}^\xi [\![\Delta \Vdash x_\beta{:}A[c]]\!]^k$$

holds by Row 11 in the logical relation.
By $\dagger_5'''$ and $\ddagger_1$ and $\ddagger_3$ we get

$$\ddagger_6 \; \forall k. \, ([x_\beta^c/y_\alpha^c]\mathcal{C}_3'\mathcal{D}_3''\mathcal{F}_3'\cdot, [x_\beta^c/y_\alpha^c]\mathcal{C}_2'\mathcal{D}_2''\mathcal{F}_2'\cdot) \in \mathcal{E}_{\Psi_0}^\xi [\![x_\beta{:}A[c] \Vdash \cdot]\!]^k$$

by Row 12 in the logical relation.
By $\dagger_4^\ell$ and the forward closure lemmas for every $\ell \in I$ we get

$$\ddagger_7^\ell \; \forall k. \, (\cdot \mathcal{A}_3^{\ell'} \mathcal{C}_3^{\ell'} \mathcal{D}_3' \mathcal{F}_3'; \cdot \mathcal{A}_2^{\ell'} \mathcal{C}_2^{\ell'} \mathcal{D}_2' \mathcal{F}_2') \in \mathcal{E}_{\Psi_0}^\xi [\![\cdot \Vdash d_\ell{:}T_\ell]\!]^k.$$

Note that we can apply forward closure on the second run since $\mathcal{A}_2^{\ell'}$ is the minimal configuration that sends along the same channels as $\mathcal{A}_3^{\ell'}$.
We apply forward closure lemmas on the substitutions to get

$$\ddagger_8^\ell \; \forall k. \, (\cdot \mathcal{A}_3^{\ell'}[x_\beta^c/y_\alpha^c]\mathcal{C}_3^{\ell'} \mathcal{D}_3'' \mathcal{F}_3'; \cdot \mathcal{A}_2^{\ell'}[x_\beta^c/y_\alpha^c]\mathcal{C}_2^{\ell'} \mathcal{D}_2'' \mathcal{F}_2') \in \mathcal{E}_{\Psi_0}^\xi [\![\cdot \Vdash d_\ell{:}T_\ell]\!]^k.$$

Since the renaming $[x_\beta^c/y_\alpha^c]$ does not affect the trees $\mathcal{A}_3^{\ell'}$ and $\mathcal{A}_2^{\ell'}$, we can rewrite $\ddagger_8^\ell$ as

$$\ddagger_9^\ell \; \forall k. \, ([x_\beta^c/y_\alpha^c]\cdot \mathcal{A}_3^{\ell'} \mathcal{C}_3^{\ell'} \mathcal{D}_3'' \mathcal{F}_3'; [x_\beta^c/y_\alpha^c]\cdot \mathcal{A}_2^{\ell'} \mathcal{C}_2^{\ell'} \mathcal{D}_2'' \mathcal{F}_2') \in \mathcal{E}_{\Psi_0}^\xi [\![\cdot \Vdash d_\ell{:}T_\ell]\!]^k.$$

We can apply the induction hypothesis on $m < m+1$, having the required assumptions $\ddagger_2, \ddagger_4, \ddagger_5, \ddagger_6$, and $\ddagger_9^\ell$ to get

$$\star\star \;\; ([x_\beta^c/y_\alpha^c]\mathcal{C}_1'\mathcal{D}_1''\mathcal{F}_1', [x_\beta^c/y_\alpha^c]\mathcal{C}_3'\mathcal{D}_3''\mathcal{F}_3') \in \mathcal{E}_\Psi^\xi [\![\Delta \Vdash x_\beta{:}A[c]]\!]^m$$

$\star\star$ is the relation required by row 4 of the logical relation to satisfy $\star_2'$, and this completes the proof of this subcase.

**Subcase 3.** $y_c^\alpha{:}A \multimap B \in \Delta$ and a message along $y^\alpha$ is in $\mathcal{D}_1'$, i.e. $\mathcal{D}_1' = \mathcal{T}_1\mathbf{msg}(\mathbf{send}\, x_\beta^c \, y_\alpha^c; y_\alpha^c \leftarrow u_\delta^c)\mathcal{D}_1''$ for $\mathcal{T}_1 \in \mathsf{Tree}_\Psi(\Delta'' \Vdash x_\beta{:}A[c])$, and $\Delta = \Delta', \Delta'', y_c^\alpha{:}A \multimap B[c]$, and $\mathcal{C}_1' = \mathcal{C}_1^1 \mathcal{C}_1^2$ for $\cdot \Vdash \mathcal{C}_1^1 :: \Delta', y_c^\alpha{:}A \multimap B$ and $\cdot \Vdash \mathcal{C}_1^2 :: \Delta''$.
By $\dagger_1'''$, we know that $\mathcal{C}_2' = \mathcal{C}_2^1 \mathcal{C}_2^2$ for $\cdot \Vdash \mathcal{C}_2^1 :: \Delta', y_c^\alpha{:}A \multimap B$ and $\cdot \Vdash \mathcal{C}_2^2 :: \Delta''$ and

$$\ddagger_1 \; \mathcal{D}_2' = \mathcal{T}_2\mathbf{msg}(\mathbf{send}, x_\beta^c \, y_\alpha^c; y_\alpha^c \leftarrow u_\delta^c)\mathcal{D}_2'' \text{ for } \mathcal{T}_2 \in \mathsf{Tree}_\Psi(\Delta'' \Vdash x_\beta{:}A[c]) \text{ and}$$

$$\ddagger_2 \; \forall k. \, (\mathcal{C}_1^2\mathcal{T}_1\mathcal{C}_1^1\mathcal{D}_1''\mathbf{msg}(\mathbf{send}\, x_\beta^c \, y_\alpha^c; y_\alpha^c \leftarrow u_\delta^c)\mathcal{F}_1', \mathcal{C}_2^2\mathcal{T}_2\mathcal{C}_2^1\mathcal{D}_2''\mathbf{msg}(\mathbf{send}\, x_\beta^c \, y_\alpha^c; y_\alpha^c \leftarrow u_\delta^c)\mathcal{F}_2') \in \mathcal{E}_\Psi^\xi [\![\Delta'' \Vdash x_\beta{:}A[c]]\!]^k \text{ and}$$

$$\ddagger_3 \; \forall k. \, (\mathcal{C}_1^1\mathcal{C}_1^2\mathcal{T}_1\mathbf{msg}(\mathbf{send}, x_\beta^c \, y_\alpha^c; y_\alpha^c \leftarrow u_\delta^c)\mathcal{D}_1''\mathcal{F}_1', \mathcal{C}_2^1\mathcal{C}_2^2\mathcal{T}_2\mathbf{msg}(\mathbf{send}\, x_\beta^c \, y_\alpha^c; y_\alpha^c \leftarrow u_\delta^c)\mathcal{D}_2''\mathcal{F}_2') \in \mathcal{E}_\Psi^\xi [\![\Delta', u_\delta{:}B[c] \Vdash K]\!]^k$$

hold by Row 10 in the logical relation.
By $\ddagger_1$ and $\dagger_2'''$, we know that $\mathcal{C}_3' = \mathcal{C}_3^1 \mathcal{C}_3^2$ for $\cdot \Vdash \mathcal{C}_3^1 :: \Delta', y_c^\alpha{:}A \multimap B$ and $\cdot \Vdash \mathcal{C}_3^2 :: \Delta''$ and

$$\ddagger_4 \; \mathcal{D}_3' = \mathcal{T}_3\mathbf{msg}(\mathbf{send}, x_\beta^c \, y_\alpha^c; y_\alpha^c \leftarrow u_\delta^c)\mathcal{D}_3'' \text{ for } \mathcal{T}_3 \in \mathsf{Tree}_\Psi(\Delta'' \Vdash x_\beta{:}A[c]) \text{ and}$$

$$\ddagger_5 \; \forall k. \, (\mathcal{C}_2^2\mathcal{T}_2\mathcal{C}_2^1\mathcal{D}_2''\mathbf{msg}(\mathbf{send}\, x_\beta^c \, y_\alpha^c; y_\alpha^c \leftarrow u_\delta^c)\mathcal{F}_2', \mathcal{C}_3^2\mathcal{T}_3\mathcal{C}_3^1\mathcal{D}_3''\mathbf{msg}(\mathbf{send}\, x_\beta^c \, y_\alpha^c; y_\alpha^c \leftarrow u_\delta^c)\mathcal{F}_3') \in \mathcal{E}_\Psi^\xi [\![\Delta'' \Vdash x_\beta{:}A[c]]\!]^k \text{ and}$$

$$\ddagger_6 \; \forall k. \, (\mathcal{C}_2^1\mathcal{C}_2^2\mathcal{T}_2\mathbf{msg}(\mathbf{send}\, x_\beta^c \, y_\alpha^c; y_\alpha^c \leftarrow u_\delta^c)\mathcal{D}_2''\mathcal{F}_2', \mathcal{C}_3^1\mathcal{C}_3^2\mathcal{T}_3\mathbf{msg}(\mathbf{send}\, x_\beta^c \, y_\alpha^c; y_\alpha^c \leftarrow u_\delta^c)\mathcal{D}_3''\mathcal{F}_3') \in \mathcal{E}_\Psi^\xi [\![\Delta', u_\delta{:}B[c] \Vdash K]\!]^k$$

hold by Row 10 in the logical relation.

By $\ddagger_4$ and $\dagger_3'''$, we know that

$\ddagger_8\ \forall k.\ (\mathcal{C}_3^2\mathcal{T}_3\mathcal{C}_3^1\mathcal{D}_3''\mathbf{msg}(\mathbf{send}x_\beta^c\ y_\alpha^c; y_\alpha^c \leftarrow u_\delta^c)\mathcal{F}_3', \mathcal{C}_2^2\mathcal{T}_2\mathcal{C}_2^1\mathcal{D}_2''\mathbf{msg}(\mathbf{send}x_\beta^c\ y_\alpha^c; y_\alpha^c \leftarrow u_\delta^c)\mathcal{F}_2') \in \mathcal{E}_\Psi^\xi[\![\Delta'' \Vdash x_\beta{:}A[c]]\!]^k$ and

$\ddagger_9\ \forall k.\ (\mathcal{C}_3^1\mathcal{C}_3^2\mathcal{T}_3\mathbf{msg}(\mathbf{send}x_\beta^c\ y_\alpha^c; y_\alpha^c \leftarrow u_\delta^c)\mathcal{D}_3''\mathcal{F}_3', \mathcal{C}_2^1\mathcal{C}_2^2\mathcal{T}_2\mathbf{msg}(\mathbf{send}x_\beta^c\ y_\alpha^c; y_\alpha^c \leftarrow u_\delta^c)\mathcal{D}_2''\mathcal{F}_2') \in \mathcal{E}_\Psi^\xi[\![\Delta', u_\delta{:}B[c] \Vdash K]\!]^k$

hold by Row 10 in the logical relation.
Without loss of generality assume that $\mathcal{A}_i^{\kappa'}$ is the subtree that provides $y_c^\alpha{:}A \multimap B[c]$. We have $\mathcal{C}_i^{\kappa'} = \mathcal{C}_i^2, \mathcal{C}_i^{\kappa''}$, i.e. $\mathcal{A}_i^{\kappa'}$ does not occur in $\mathcal{C}_i^2$ and $\mathcal{C}_i^1 = \mathcal{C}_i^{\kappa''}, \mathcal{A}_i^{\kappa'}$. By $\dagger_5''', \ddagger_1$ and $\ddagger_4$ we get

$\ddagger_{10}\ \forall k.\ (\mathcal{C}_3^2\mathcal{T}_3\mathcal{A}_3^{\kappa''}\mathbf{msg}(\mathbf{send}x_\beta^c\ y_\alpha^c; y_\alpha^c \leftarrow u_\delta^c)\mathcal{C}_3^{\kappa''}\mathcal{D}_3'\mathcal{F}_3'; \mathcal{C}_2^2\mathcal{T}_2\mathcal{A}_2^{\kappa''}\mathbf{msg}(\mathbf{send}x_\beta^c\ y_\alpha^c; y_\alpha^c \leftarrow u_\delta^c)\mathcal{C}_2^{\kappa''}\mathcal{D}_2'\mathcal{F}_2') \in \mathcal{E}_{\Psi_0}^\xi[\![x_\beta{:}A[c] \Vdash u_\delta{:}B[c]]\!]^k$

by Row 5 in the logical relation.
By $\dagger_4^\ell$ and the forward closure lemmas for every $\ell \neq \kappa \in I$ we get

$\ddagger_{11}^\ell\ \forall k.\ (\cdot\mathcal{A}_3^{\ell'}\mathcal{C}_3^{\ell'}\mathcal{D}_3'\mathcal{F}_3'; \cdot\mathcal{A}_2^{\ell'}\mathcal{C}_2^{\ell'}\mathcal{D}_2'\mathcal{F}_2') \in \mathcal{E}_{\Psi_0}^\xi[\![\cdot \Vdash d_\ell{:}T_\ell]\!]^k.$

Note that we can apply forward closure on the second run since $\mathcal{A}_2^{\ell'}$ is the minimal configuration that sends along the same channels as $\mathcal{A}_3^{\ell'}$.
By Lemma C.16, from $\ddagger_{10}, \ddagger_8$, and $\ddagger_{11}^\ell$ for every tree in $\mathcal{C}_3^2$ we get

$\ddagger_{12}\ \forall k.\ (\cdot\mathcal{C}_3^2\mathcal{T}_3\mathcal{A}_3^{\kappa''}\mathbf{msg}(\mathbf{send}x_\beta^c\ y_\alpha^c; y_\alpha^c \leftarrow u_\delta^c)\mathcal{C}_3^{\kappa''}\mathcal{D}_3'\mathcal{F}_3'; \cdot\mathcal{C}_2^2\mathcal{T}_2\mathcal{A}_2^{\kappa''}\mathbf{msg}(\mathbf{send}x_\beta^c\ y_\alpha^c; y_\alpha^c \leftarrow u_\delta^c)\mathcal{C}_2^{\kappa''}\mathcal{D}_2'\mathcal{F}_2') \in \mathcal{E}_{\Psi_0}^\xi[\![\cdot \Vdash u_\delta{:}B[c]]\!]^k$

We can apply the induction hypothesis on $m < m{+}1$, having the required assumptions $\ddagger_3, \ddagger_6, \ddagger_9, \ddagger_{11}^\ell$ relating corresponding trees in $\mathcal{C}_3^{\kappa''}$ and $\mathcal{C}_2^{\kappa''}$, and $\ddagger_{12}$ and $\mathcal{C}_i^1 = \mathcal{C}_i^{\kappa''}, \mathcal{A}_i^{\kappa'}$ to get

$\star\star_1\ (\mathcal{C}_1^1\mathcal{C}_1^2\mathcal{T}_1\mathbf{msg}(\mathbf{send}x_\beta^c\ y_\alpha^c; y_\alpha^c \leftarrow u_\delta^c)\mathcal{D}_1''\mathcal{F}_1', \mathcal{C}_3^1\mathcal{C}_3^2\mathcal{T}_3\mathbf{msg}(\mathbf{send}x_\beta^c\ y_\alpha^c; y_\alpha^c \leftarrow u_\delta^c)\mathcal{D}_3''\mathcal{F}_3') \in \mathcal{E}_\Psi^\xi[\![\Delta', u_\delta{:}B[c] \Vdash K]\!]^m$

$\star\star_1$ is one relation required by row 4 of the logical relation to satisfy $\star_2'$. To get the other relation, we use a similar strategy. By Lemma C.16, from $\ddagger_{10}, \ddagger_9$, and $\ddagger_{11}^\ell$ for every tree in $\mathcal{C}_3^1$ we get

$\ddagger_{13}\ \forall k.\ (\mathcal{C}_3^2\mathcal{T}_3\mathcal{C}_3^1\mathcal{D}_3''\mathbf{msg}(\mathbf{send}x_\beta^c\ y_\alpha^c; y_\alpha^c \leftarrow u_\delta^c)\mathcal{F}_3'\cdot, \mathcal{C}_2^2\mathcal{T}_2\mathcal{C}_2^1\mathcal{D}_2''\mathbf{msg}(\mathbf{send}x_\beta^c\ y_\alpha^c; y_\alpha^c \leftarrow u_\delta^c)\mathcal{F}_2'\cdot) \in \mathcal{E}_\Psi^\xi[\![x_\beta{:}A[c] \Vdash \cdot]\!]^k$

We can apply the induction hypothesis on $m < m{+}1$, having the required assumptions $\ddagger_2, \ddagger_5, \ddagger_8, \ddagger_{11}^\ell$ relating corresponding trees in $\mathcal{C}_3^2$ and $\mathcal{C}_2^2$, and $\ddagger_1 3$ to get

$\star\star_2\ (\mathcal{C}_1^2\mathcal{T}_1\mathcal{C}_1^1\mathcal{D}_1''\mathbf{msg}(\mathbf{send}x_\beta^c\ y_\alpha^c; y_\alpha^c \leftarrow u_\delta^c)\mathcal{F}_1', \mathcal{C}_3^2\mathcal{T}_3\mathcal{C}_3^1\mathcal{D}_3''\mathbf{msg}(\mathbf{send}x_\beta^c\ y_\alpha^c; y_\alpha^c \leftarrow u_\delta^c)\mathcal{F}_3') \in \mathcal{E}_\Psi^\xi[\![\Delta'' \Vdash x_\beta{:}A[c]]\!]^m.$

This completes the proof of the subcase since we now satisfy all conditions of Row 10.
**Case 2.** $y_\alpha \in \Theta_1 \cap \Theta_3$. In this case, our goal is to prove

$\star_2''\ (\mathcal{C}_1', \mathcal{D}_1', \mathcal{F}_1'; \mathcal{C}_3', \mathcal{D}_3', \mathcal{F}_3') \in \mathcal{V}_\Psi^\xi[\![\Delta \Vdash K]\!]_{y_\alpha;\cdot}^{m+1}.$

We consider subcases based on the type of $y^\alpha$ and the message or forwarding process along $y^\alpha$ that is in $\mathcal{C}_1$ or $\mathcal{F}_1$. We provide the proof of two interesting cases. The proof of other cases is similar.
**Subcase 1.** $\Delta = \Delta', y_c^\alpha{:}A \otimes B$ and a message along $y^\alpha$ is in $\mathcal{C}_1'$, i.e. $\mathcal{C}_1' = \mathcal{C}_1''\mathcal{T}_1\mathbf{msg}(\mathbf{send}x_\beta^c\ y_\alpha^c; y_\alpha^c \leftarrow u_\delta^c)$ for $\mathcal{T}_1 \in \mathsf{Tree}_\Psi(\cdot \Vdash x_\beta{:}A[c])$. If $\mathcal{C}_3' \neq \mathcal{C}_3''\mathcal{T}_3\mathbf{msg}(\mathbf{send}x_\beta^c\ y_\alpha^c; y_\alpha^c \leftarrow u_\delta^c)$ then $\star_2''$ holds trivially by row 9 of the logical relation.
Assume $\ddagger_1\ \mathcal{C}_3' = \mathcal{C}_3''\mathcal{T}_3\mathbf{msg}(\mathbf{send}x_\beta^c\ y_\alpha^c; y_\alpha^c \leftarrow u_\delta^c)$. Without loss of generality, assume $\mathbf{msg}(\mathbf{send}x_\beta^c\ y_\alpha^c; y_\alpha^c \leftarrow u_\delta^c)$ is in the tree $\mathcal{A}_3^{\kappa'}$, i.e. $\mathcal{A}_3^{\kappa'} = \mathcal{A}_3^{\kappa''}\mathcal{T}_3\mathbf{msg}(\mathbf{send}x_\beta^c\ y_\alpha^c; y_\alpha^c \leftarrow u_\delta^c)$.
By $\dagger_4^{\kappa'''}$ we get $\ddagger_1\ \mathcal{A}_2^{\kappa'} = \mathcal{A}_2^{\kappa''}\mathcal{T}_2\mathbf{msg}(\mathbf{send}x_\beta^c\ y_\alpha^c; y_\alpha^c \leftarrow u_\delta^c)$, and

$\ddagger_2^\kappa\ \forall k.\ (\cdot\mathcal{A}_3^{\kappa''}\mathcal{T}_3\mathbf{msg}(\mathbf{send}x_\beta^c\ y_\alpha^c; y_\alpha^c \leftarrow u_\delta^c)\mathcal{C}_3^{\kappa'}\mathcal{D}_3'\mathcal{F}_3'; \cdot\mathcal{A}_2^{\kappa''}\mathcal{T}_2\mathbf{msg}(\mathbf{send}x_\beta^c\ y_\alpha^c; y_\alpha^c \leftarrow u_\delta^c)\mathcal{C}_2^{\kappa'}\mathcal{D}_2'\mathcal{F}_2') \in \mathcal{E}_{\Psi_0}^\xi[\![\cdot \Vdash u_\delta{:}B[c]]\!]^k$, and

$\ddagger_3^\kappa\ \forall k.\ (\cdot\mathcal{T}_3\mathcal{A}_3^{\kappa''}\mathbf{msg}(\mathbf{send}x_\beta^c\ y_\alpha^c; y_\alpha^c \leftarrow u_\delta^c)\mathcal{C}_3^{\kappa'}\mathcal{D}_3'\mathcal{F}_3'; \cdot\mathcal{T}_2\mathcal{A}_2^{\kappa''}\mathbf{msg}(\mathbf{send}x_\beta^c\ y_\alpha^c; y_\alpha^c \leftarrow u_\delta^c)\mathcal{C}_2^{\kappa'}\mathcal{D}_2'\mathcal{F}_2') \in \mathcal{E}_{\Psi_0}^\xi[\![\cdot \Vdash x_\beta{:}A[c]]\!]^k.$

by row 4 of the logical relation.
By $\ddagger_1$, we get $\ddagger_1'\ \mathcal{C}_2' = \mathcal{C}_2''\mathcal{T}_2\mathbf{msg}(\mathbf{send}x_\beta^c\ y_\alpha^c; y_\alpha^c \leftarrow u_\delta^c)$. By $\dagger_1'''$ and $\ddagger_1'$, we know that

$\ddagger_4\ \forall k.\ (\mathcal{C}_1''\mathcal{T}_1\mathbf{msg}(\mathbf{send}x_\beta^c\ y_\alpha^c; y_\alpha^c \leftarrow u_\delta^c)\mathcal{D}_1'\mathcal{F}_1', \mathcal{C}_2''\mathcal{T}_2\mathbf{msg}(\mathbf{send}x_\beta^c\ y_\alpha^c; y_\alpha^c \leftarrow u_\delta^c)\mathcal{D}_2'\mathcal{F}_2') \in \mathcal{E}_\Psi^\xi[\![\Delta', x_\beta{:}A[c], u_\delta{:}B[c] \Vdash K]\!]^k$

52

holds by Row 9 in the logical relation.
By †$_2'''$, we know that

$$\ddagger_5 \forall k. (\mathcal{C}_2''\mathcal{T}_2\mathbf{msg}(\mathbf{send}x_\beta^c\, y_\alpha^c; y_\alpha^c \leftarrow u_\delta^c)\mathcal{D}_2'\mathcal{F}_2', \mathcal{C}_3''\mathcal{T}_3\mathbf{msg}(\mathbf{send}x_\beta^c\, y_\alpha^c; y_\alpha^c \leftarrow u_\delta^c)\mathcal{D}_3'\mathcal{F}_3') \in \mathcal{E}_\Psi^\xi[\![\Delta', x_\beta{:}A[c], u_\delta{:}A[c] \Vdash K]\!]^k$$

holds by Row 9 in the logical relation.
By †$_3'''$, we know that

$$\ddagger_6 \forall k. (\mathcal{C}_3''\mathcal{T}_3\mathbf{msg}(\mathbf{send}x_\beta^c\, y_\alpha^c; y_\alpha^c \leftarrow u_\delta^c)\mathcal{D}_3'\mathcal{F}_3', \mathcal{C}_2''\mathcal{T}_2\mathbf{msg}(\mathbf{send}x_\beta^c\, y_\alpha^c; y_\alpha^c \leftarrow u_\delta^c)\mathcal{D}_2'\mathcal{F}_2') \in \mathcal{E}_\Psi^\xi[\![\Delta', x_\beta{:}A[c], u_\delta{:}A[c] \Vdash K]\!]^k$$

holds by Row 9 in the logical relation.
By †$_4^\ell$ and the forward closure lemmas for every $\ell \neq \kappa \in I$ we get

$$\ddagger_7^\ell \forall k. (\cdot\mathcal{A}_3^{\ell'}\mathcal{C}_3^{\ell'}\mathcal{D}_3'\mathcal{F}_3'; \cdot\mathcal{A}_2^{\ell'}\mathcal{C}_2^{\ell'}\mathcal{D}_2'\mathcal{F}_2') \in \mathcal{E}_{\Psi_0}^\xi[\![\cdot \Vdash d_\ell{:}T_\ell]\!]^k.$$

Note that we can apply forward closure on the second run since $\mathcal{A}_2^{\ell'}$ is the minimal configuration that sends along the same channels as $\mathcal{A}_3^{\ell'}$.
by forward closure on †$_5$ we get

$$\ddagger_8 \forall k. (\mathcal{C}_3'\mathcal{D}_3'\mathcal{F}_3'\cdot, \mathcal{C}_2'\mathcal{D}_2'\mathcal{F}_2'\cdot) \in \mathcal{E}_\Psi^\xi[\![K \Vdash \cdot]\!]^k.$$

We can apply the induction hypothesis on $m < m+1$, having the required assumptions $\ddagger_4, \ddagger_5, \ddagger_6, \ddagger_2^\kappa$ and $\ddagger_2^\kappa$ and $\ddagger_7^\ell$, and $\ddagger_8$ to get

$$\star\star\ (\mathcal{C}_1''\mathcal{T}_1\mathbf{msg}(\mathbf{send}x_\beta^c\, y_\alpha^c; y_\alpha^c \leftarrow u_\delta^c)\mathcal{D}_1'\mathcal{F}_1', \mathcal{C}_3''\mathcal{T}_3\mathbf{msg}(\mathbf{send}x_\beta^c\, y_\alpha^c; y_\alpha^c \leftarrow u_\delta^c)\mathcal{D}_3'\mathcal{F}_3') \in \mathcal{E}_\Psi^\xi[\![\Delta', x_\beta{:}A[c], u_\delta{:}B[c] \Vdash K]\!]^m$$

$\star\star$ is the relation required by row 9 of the logical relation to satisfy $\star_2'$, and this completes the proof of this subcase.

**Subcase 2.** $\Delta = \Delta', y_c^\alpha{:}A$ and a forwarding process along $y^\alpha$ is in $\mathcal{C}_1'$, i.e. $\mathcal{C}_1' = \mathcal{C}_1''\mathbf{proc}(y_\alpha^c, y_\alpha^c \leftarrow x_\beta^c@d_1)$. If $\mathcal{C}_3' \neq \mathcal{C}_3''\mathbf{proc}(y_\alpha^c, y_\alpha^c \leftarrow x_\beta^c@d_2)$ then $\star_2''$ holds trivially by row 12 of the logical relation. Assume $\ddagger_1\ \mathcal{C}_3' = \mathcal{C}_3''\mathbf{proc}(y_\alpha^c, y_\alpha^c \leftarrow x_\beta^c@d_2)$. Without loss of generality, assume $\mathbf{proc}(y_\alpha^c, y_\alpha^c \leftarrow x_\beta^c@d_1)$ is in the tree $\mathcal{A}_3^{\kappa'}$, i.e. $\mathcal{A}_3^{\kappa'} = \mathcal{A}_3^{\kappa''}\mathbf{proc}(y_\alpha^c, y_\alpha^c \leftarrow x_\beta^c@d_1)$.
By †$_4^{\kappa'''}$ we get $\ddagger_1\ \mathcal{A}_2^{\kappa'} = \mathcal{A}_2^{\kappa''}\mathbf{proc}(y_\alpha^c, y_\alpha^c \leftarrow x_\beta^c@d_1)$, and

$$\ddagger_2^\kappa \forall k. ([x_\beta^c/y_\alpha^c]\cdot\mathcal{A}_3^{\kappa''}\mathcal{C}_3^{\kappa'}\mathcal{D}_3'\mathcal{F}_3'; [x_\beta^c/y_\alpha^c]\cdot\mathcal{A}_2^{\kappa''}\mathcal{C}_2^{\kappa'}\mathcal{D}_2'\mathcal{F}_2') \in \mathcal{E}_{\Psi_0}^\xi[\![\cdot \Vdash x_\beta{:}A[c]]\!]^k$$

by row 11 of the logical relation.
By $\ddagger_1$, we get $\ddagger_1'\ \mathcal{C}_2' = \mathcal{C}_2''\mathbf{proc}(y_\alpha^c, y_\alpha^c \leftarrow x_\beta^c@d_1)$.
By †$_1'''$ and $\ddagger_1'$, we know that

$$\ddagger_3 \forall k. ([x_\beta^c/y_\alpha^c]\mathcal{C}_1''\mathcal{D}_1'\mathcal{F}_1', [x_\beta^c/y_\alpha^c]\mathcal{C}_2''\mathcal{D}_2'\mathcal{F}_2') \in \mathcal{E}_\Psi^\xi[\![\Delta', x_\beta{:}A[c] \Vdash [x_\beta^c/y_\alpha^c]K]\!]^k$$

hold by Row 12 in the logical relation.
By †$_2'''$,

$$\ddagger_4 \forall k. ([x_\beta^c/y_\alpha^c]\mathcal{C}_2''\mathcal{D}_2'\mathcal{F}_2', [x_\beta^c/y_\alpha^c]\mathcal{C}_3''\mathcal{D}_3'\mathcal{F}_3') \in \mathcal{E}_\Psi^\xi[\![\Delta', x_\beta{:}A[c] \Vdash [x_\beta^c/y_\alpha^c]K]\!]^k$$

hold by Row 12 in the logical relation.
By †$_3'''$, we know that

$$\ddagger_5 \forall k. ([x_\beta^c/y_\alpha^c]\mathcal{C}_3''\mathcal{D}_3'\mathcal{F}_3', [x_\beta^c/y_\alpha^c]\mathcal{C}_2''\mathcal{D}_2'\mathcal{F}_2') \in \mathcal{E}_\Psi^\xi[\![\Delta', x_\beta{:}A[c] \Vdash [x_\beta^c/y_\alpha^c]K]\!]^k$$

holds by Row 12 in the logical relation.
By †$_4^\ell$ and the forward closure lemmas for every $\ell \neq \kappa \in I$ we get

$$\ddagger_6^\ell \forall k. (\cdot\mathcal{A}_3^{\ell'}\mathcal{C}_3^{\ell'}\mathcal{D}_3'\mathcal{F}_3'; \cdot\mathcal{A}_2^{\ell'}\mathcal{C}_2^{\ell'}\mathcal{D}_2'\mathcal{F}_2') \in \mathcal{E}_{\Psi_0}^\xi[\![\cdot \Vdash d_\ell{:}T_\ell]\!]^k.$$

Note that we can apply forward closure on the second run since $\mathcal{A}_2^{\ell'}$ is the minimal configuration that sends along the same channels as $\mathcal{A}_3^{\ell'}$.
We apply forward closure lemmas on the substitutions to get

$$\ddagger_7^\ell \forall k. (\cdot\mathcal{A}_3^{\ell'}[x_\beta^c/y_\alpha^c]\mathcal{C}_3^{\ell''}\mathcal{D}_3'\mathcal{F}_3'; \cdot\mathcal{A}_2^{\ell'}[x_\beta^c/y_\alpha^c]\mathcal{C}_2^{\ell''}\mathcal{D}_2'\mathcal{F}_2') \in \mathcal{E}_{\Psi_0}^\xi[\![\cdot \Vdash d_\ell{:}T_\ell]\!]^k,$$

53

where $\mathcal{C}_2^{\ell''}$ is the configuration we get by stepping the forwarding process in $\mathcal{A}_2^{\kappa'}$. In other words, $\mathcal{C}_2^{\ell''} = \{\mathcal{A}^{\iota''}\}_{\iota \neq \ell \in I}$ where for $\iota \neq \kappa$ we define $\mathcal{A}^{\iota''} = \mathcal{A}^{\iota'}$. $\mathcal{C}_3^{\ell''}$ is defined in the same way.

Since the renaming $[x_\beta^c/y_\alpha^c]$ does not affect the trees $\mathcal{A}_3^{\ell'}$ and $\mathcal{A}_2^{\ell'}$, we can rewrite $\ddagger_7^\ell$ as

$$\ddagger_8^\ell \ \forall k. \ ([x_\beta^c/y_\alpha^c]\cdot\mathcal{A}_3^{\ell'}\mathcal{C}_3^{\ell''}\mathcal{D}_3'\mathcal{F}_3'; [x_\beta^c/y_\alpha^c]\cdot\mathcal{A}_2^{\ell'}\mathcal{C}_2^{\ell''}\mathcal{D}_2'\mathcal{F}_2') \in \mathcal{E}_{\Psi_0}^\xi[\![\cdot \Vdash d_\ell{:}T_\ell]\!]^k.$$

So far, we proved the first four assumptions required for applying the induction hypothesis. To build the last one, we consider cases on $K$, if $K = y_\alpha{:}A[c]$, then we know that $\mathcal{D}_3' = \mathcal{D}_2' = \cdot$. In this case, by $\dagger_5'''$ we get

$$\ddagger_9 \ \forall k. \ ([x_\beta^c/y_\alpha^c]\mathcal{C}_3''\mathcal{D}_3'\mathcal{F}_3'\cdot, [x_\beta^c/y_\alpha^c]\mathcal{C}_2''\mathcal{D}_2'\mathcal{F}_2'\cdot) \in \mathcal{E}_\Psi^\xi[\![x_\beta{:}A[c] \Vdash \cdot]\!]^k$$

by Row 12 in the logical relation.

We can apply the induction hypothesis on $m < m+1$, having the required assumptions $\ddagger_3$, $\ddagger_4$, $\ddagger_5$, $\ddagger_2^\kappa$ and $\ddagger_8^\ell$, and $\ddagger_9$ to get

$$\star\star_1 \ ([x_\beta^c/y_\alpha^c]\mathcal{C}_1''\mathcal{D}_1'\mathcal{F}_1', [x_\beta^c/y_\alpha^c]\mathcal{C}_3''\mathcal{D}_3'\mathcal{F}_3') \in \mathcal{E}_\Psi^\xi[\![\Delta', x_\beta{:}A[c] \Vdash x_\beta{:}A[c]]\!]^m$$

$\star\star_1$ is the relation required by row 12 of the logical relation to satisfy $\star_2'$, and this completes the proof of this subcase. If $K \neq y_\alpha{:}A[c]$, then we know that $y_\alpha$ does not occur in $\mathcal{F}_2$ and $\mathcal{F}_3$. In this case, by forward closure on $\dagger_5$ we get

$$\ddagger_{10} \ \forall k. \ (\mathcal{C}_3'\mathcal{D}_3'\mathcal{F}_3'\cdot, \mathcal{C}_2'\mathcal{D}_2'\mathcal{F}_2'\cdot) \in \mathcal{E}_\Psi^\xi[\![K \Vdash \cdot]\!]^k.$$

By forward closure on the substitutions and the fact that $y_\alpha$ does not occur in $\mathcal{F}_2$ and $\mathcal{F}_3$, we get

$$\ddagger_{11} \ \forall k. \ ([x_\beta^c/y_\alpha^c]\mathcal{C}_3''\mathcal{D}_3'\mathcal{F}_3'\cdot, [x_\beta^c/y_\alpha^c]\mathcal{C}_2''\mathcal{D}_2'\mathcal{F}_2'\cdot) \in \mathcal{E}_\Psi^\xi[\![K \Vdash \cdot]\!]^k$$

We can apply the induction hypothesis on $m < m+1$, having the required assumptions $\ddagger_3$, $\ddagger_4$, $\ddagger_5$, $\ddagger_2^\kappa$ and $\ddagger_8^\ell$, and $\ddagger_{11}$ to get

$$\star\star_2 \ ([x_\beta^c/y_\alpha^c]\mathcal{C}_1''\mathcal{D}_1'\mathcal{F}_1', [x_\beta^c/y_\alpha^c]\mathcal{C}_3''\mathcal{D}_3'\mathcal{F}_3') \in \mathcal{E}_\Psi^\xi[\![\Delta', x_\beta{:}A[c] \Vdash K]\!]^m$$

$\star\star_2$ is the relation required by row 12 of the logical relation to satisfy $\star_2'$, and this completes the proof of this subcase. $\square$

## C.4 Adequacy

**Definition C.7.** *We write $\mathcal{C}_1, \mathcal{D}_1, \mathcal{F}_1 \sim^{m+1} \mathcal{C}_2, \mathcal{D}_2, \mathcal{F}_2$ iff for some $\Delta, K$, we have $(\mathcal{C}_1, \mathcal{D}_1, \mathcal{F}_1; \mathcal{C}_2, \mathcal{D}_2, \mathcal{F}_2) \in \mathsf{Tree}_{\Psi_0}(\Delta \Vdash K)$, and*

- *if $(\mathcal{C}_1, \mathcal{D}_1, \mathcal{F}_1) \xrightarrow[\Delta \Vdash K]{x.q_1} (\mathcal{C}_1', \mathcal{D}_1', \mathcal{F}_1')$ then $(\mathcal{C}_2, \mathcal{D}_2, \mathcal{F}_2) \xrightarrow[\Delta \Vdash K]{x.q_1} (\mathcal{C}_2', \mathcal{D}_2', \mathcal{F}_2')$ and $\mathcal{C}_1', \mathcal{D}_1', \mathcal{F}_1' \sim^m \mathcal{C}_2', \mathcal{D}_2', \mathcal{F}_2'$, and*
- *if $(\mathcal{C}_1, \mathcal{D}_1, \mathcal{F}_1) \xrightarrow[\Delta \Vdash K]{x.\overline{q_1}} (\mathcal{C}_1', \mathcal{D}_1', \mathcal{F}_1')$ then $(\mathcal{C}_2, \mathcal{D}_2, \mathcal{F}_2) \xrightarrow[\Delta \Vdash K]{x.\overline{q_1}} (\mathcal{C}_2', \mathcal{D}_2', \mathcal{F}_2')$ and $\mathcal{C}_1', \mathcal{D}_1', \mathcal{F}_1' \sim^m \mathcal{C}_2', \mathcal{D}_2', \mathcal{F}_2'$, and*
- *if $(\mathcal{C}_2, \mathcal{D}_2, \mathcal{F}_2) \xrightarrow[\Delta \Vdash K]{x.q_1} (\mathcal{C}_2', \mathcal{D}_2', \mathcal{F}_2')$ then $(\mathcal{C}_1, \mathcal{D}_1, \mathcal{F}_1) \xrightarrow[\Delta \Vdash K]{x.q_1} (\mathcal{C}_1', \mathcal{D}_1', \mathcal{F}_1')$ and $\mathcal{C}_1', \mathcal{D}_1', \mathcal{F}_1' \sim^m \mathcal{C}_2', \mathcal{D}_2', \mathcal{F}_2'$, and*
- *if $(\mathcal{C}_2, \mathcal{D}_2, \mathcal{F}_2) \xrightarrow[\Delta \Vdash K]{x.\overline{q_1}} (\mathcal{C}_2', \mathcal{D}_2', \mathcal{F}_2')$ then $(\mathcal{C}_1, \mathcal{D}_1, \mathcal{F}_1) \xrightarrow[\Delta \Vdash K]{x.\overline{q_1}} (\mathcal{C}_1', \mathcal{D}_1', \mathcal{F}_1')$ and $\mathcal{C}_1', \mathcal{D}_1', \mathcal{F}_1' \sim^m \mathcal{C}_2', \mathcal{D}_2', \mathcal{F}_2'$.*

*Moreover, $\mathcal{C}_1, \mathcal{D}_1, \mathcal{F}_1 \sim^0 \mathcal{C}_2, \mathcal{D}_2, \mathcal{F}_2$ iff for some $\Delta, K$, we have $(\mathcal{C}_1, \mathcal{D}_1, \mathcal{F}_1; \mathcal{C}_2, \mathcal{D}_2, \mathcal{F}_2) \in \mathsf{Tree}_{\Psi_0}(\Delta \Vdash K)$.*

**Theorem C.24.** *labelthm:ctx For $(\mathcal{C}, \mathcal{D}_1, \mathcal{F}; \mathcal{C}, \mathcal{D}_2, \mathcal{F}) \in \mathsf{Tree}_{\Psi_0}(\Delta \Vdash K)$, if $\dagger_2 (\Delta \Vdash \mathcal{D}_1 :: K) \equiv_\xi^{\Psi_0} (\Delta \Vdash \mathcal{D}_2 :: K)$ then $\forall m. \mathcal{C}_1, \mathcal{D}_1, \mathcal{F}_1 \sim^m \mathcal{C}_2, \mathcal{D}_2, \mathcal{F}_2$.*

*Proof.* It is a straightforward corollary of Theorem C.25 and Reflexivity C.20. $\square$

**Theorem C.25** (Adequacy). *For $(\mathcal{C}_1, \mathcal{D}_1, \mathcal{F}_1; \mathcal{C}_2, \mathcal{D}_2, \mathcal{F}_2) \in \mathsf{Tree}_{\Psi_0}(\Delta \Vdash K)$, if $\dagger_1 (\cdot \Vdash \mathcal{C}_1 :: \Delta) \equiv_\xi^{\Psi_0} (\cdot \Vdash \mathcal{C}_2 :: \Delta)$ and if $\dagger_2 (\Delta \Vdash \mathcal{D}_1 :: K) \equiv_\xi^{\Psi_0} (\Delta \Vdash \mathcal{D}_2 :: K)$ and $\dagger_3 (K \Vdash \mathcal{F}_1 :: \cdot) \equiv_\xi^{\Psi_0} (K \Vdash \mathcal{F}_2 :: \cdot)$, then $\forall m. \mathcal{C}_1, \mathcal{D}_1, \mathcal{F}_1 \sim^m \mathcal{C}_2, \mathcal{D}_2, \mathcal{F}_2$.*

*Proof.* It is a corollary of the following lemma (Generalized adequacy). Consider $\Delta = \{d_j{:}T_j\}_{j \leq n}$, the configuration $\mathcal{C}_i$ is a forest consisting of $n$ trees $\mathcal{C}_i = \mathcal{B}_i^1 \cdots \mathcal{B}_i^n$ each offering along a channel in $\Delta$. In other words, $\cdot \Vdash \mathcal{B}_i^j :: d_j{:}T_j$. We put $\mathcal{C}_i^j$ to be the rest of the forest other than $\mathcal{B}_i^j$. In other words, $\mathcal{C}_i^j = \{\mathcal{B}_i^\ell\}_{\ell \neq j}$. By $\dagger_1$, we have

$$\dagger_1^j \ \forall m_1. (\cdot \mathcal{B}_i^j \mathcal{C}_1^j \mathcal{D}_1 \mathcal{F}_1; \cdot \mathcal{B}_2^j \mathcal{C}_2^j \mathcal{D}_2 \mathcal{F}_2) \in \mathcal{E}_{\Psi_0}^\xi [\![\cdot \Vdash d_j{:}T_j]\!]^{m_1},$$

Moreover, by assumptions $\dagger_2$ and $\dagger_3$ we get:

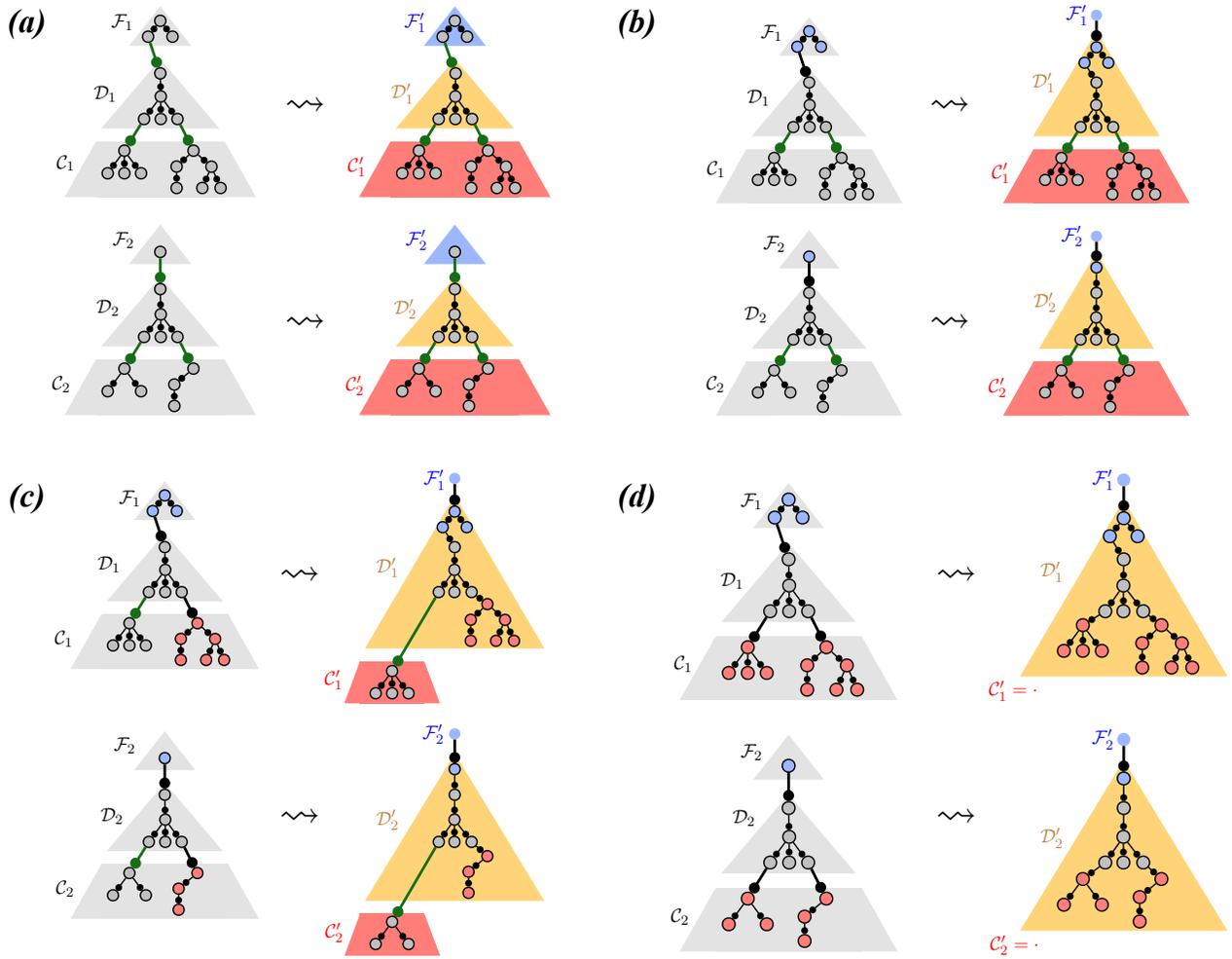

Figure 3: Schematic illustration of Lemma C.18. Observable channels are depicted in green.

$$\dagger'_2 \ \forall m_2.(\mathcal{C}_1\mathcal{D}_1\mathcal{F}_1; \mathcal{C}_2\mathcal{D}_2\mathcal{F}_2) \in \mathcal{E}^\xi_{\Psi_0}[\![\Delta \Vdash K]\!]^{m_2},$$

$$\dagger'_3 \ \forall m_3.(\mathcal{C}_1\mathcal{D}_1\mathcal{F}_1\cdot; \mathcal{C}_2\mathcal{D}_2\mathcal{F}_2\cdot) \in \mathcal{E}^\xi_{\Psi_0}[\![K \Vdash \cdot]\!]^{m_3}.$$

We can prove $\forall m. \mathcal{C}_1, \mathcal{D}_1, \mathcal{F}_1 \sim^m \mathcal{C}_2, \mathcal{D}_2, \mathcal{F}_2$ by applying Lemma C.26 on $\dagger^j_1$ and $\dagger'_2$ and $\dagger'_3$. □

**Lemma C.26** (Generalized adequacy). *Put $\mathcal{C}_i = \mathcal{B}^1_i \cdots \mathcal{B}^n_i$, and $\Delta = \{d_j:T_j\}_{j\leq n}$, and $\cdot \Vdash \mathcal{B}^j_i :: d_j:T_j$. We put $\mathcal{C}^j_i$ to be the rest of the forest other than $\mathcal{B}^j_i$. In other words, $\mathcal{C}^j_i = \{\mathcal{B}^\ell_i\}_{\ell\neq j}$.*

$$\text{if } \dagger^j_1 \ \forall m_1.(\cdot\mathcal{B}^j_1\mathcal{C}^j_1\mathcal{D}_1\mathcal{F}_1; \cdot\mathcal{B}^j_2\mathcal{C}^j_2\mathcal{D}_2\mathcal{F}_2) \in \mathcal{E}^\xi_{\Psi_0}[\![\cdot \Vdash d_j:T_j]\!]^{m_1}, \text{ and}$$

$$\dagger_2 \ \forall m_2.(\mathcal{C}_1\mathcal{D}_1\mathcal{F}_1; \mathcal{C}_2\mathcal{D}_2\mathcal{F}_2) \in \mathcal{E}^\xi_{\Psi_0}[\![\Delta \Vdash K]\!]^{m_2}, \text{ and}$$

$$\dagger_3 \ \forall m_3.(\mathcal{C}_1\mathcal{D}_1\mathcal{F}_1\cdot; \mathcal{C}_2\mathcal{D}_2\mathcal{F}_2\cdot) \in \mathcal{E}^\xi_{\Psi_0}[\![K \Vdash \cdot]\!]^{m_3},$$

*then $\forall m. \mathcal{C}_1, \mathcal{D}_1, \mathcal{F}_1 \sim^m \mathcal{C}_2, \mathcal{D}_2, \mathcal{F}_2$.*

*Proof.* We want to prove $\forall m. \mathcal{C}_1, \mathcal{D}_1, \mathcal{F}_1 \sim^m \mathcal{C}_2, \mathcal{D}_2, \mathcal{F}_2$. Consider one arbitrary $m$, we prove $\mathcal{C}_1, \mathcal{D}_1, \mathcal{F}_1 \sim^m \mathcal{C}_2, \mathcal{D}_2, \mathcal{F}_2$ by induction on $m$. If $m = 0$ the proof is straightforward. Consider $m = k + 1$. There are four cases to look at.

**Case 1.** $(\mathcal{C}_1, \mathcal{D}_1, \mathcal{F}_1) \xrightarrow[\Delta \Vdash K]{y_c^\alpha \cdot q_1} (\mathcal{C}_1', \mathcal{D}_1', \mathcal{F}_1')$, and $y_c^\alpha \in K$.

By definition $\star$ $(\mathcal{C}_1, \mathcal{D}_1, \mathcal{F}_1) \mapsto_{\Delta \Vdash K}^{*\Theta; \Upsilon_1, y_c^\alpha} (\mathcal{C}_1'', \mathcal{D}_1'', \mathcal{F}_1'')$, and $\dagger$ $(\mathcal{C}_1'', \mathcal{D}_1'', \mathcal{F}_1'') \xrightarrow[\Delta \Vdash K]{y_c^\alpha \cdot q_1} (\mathcal{C}_1', \mathcal{D}_1', \mathcal{F}_1')$. We apply $\star$ on the condition of line 13 in the logical relation in $\dagger_2$ for any $m_2 \geq 1$ to eliminate the if. By $\forall E$ and $\exists E$ and Lemma C.14, we know that there exists a minimal $\mathcal{D}_2''$ such that

$\circ_1$ $\mathcal{C}_2, \mathcal{D}_2, \mathcal{F}_2 \mapsto_{\Delta \Vdash K}^{*\Theta'; \Upsilon_1} \mathcal{C}_2, \mathcal{D}_2', \mathcal{F}_2$ and

$\dagger_2'$ $\forall \mathcal{C}_1', \mathcal{F}_1', \mathcal{C}_2', \mathcal{F}_2'$.if $\mathcal{C}_1, \mathcal{D}_1'', \mathcal{F}_1 \mapsto_{\Delta \Vdash K}^{*\Theta_1, \Upsilon_1} \mathcal{C}_1', \mathcal{D}_1'', \mathcal{F}_1'$ and $\mathcal{C}_2, \mathcal{D}_2'', \mathcal{F}_2 \mapsto_{\Delta \Vdash K}^{*\Theta_2; \Upsilon_1} \mathcal{C}_2', \mathcal{D}_2'', \mathcal{F}_2'$, then
$\quad \forall x \in \Upsilon_1. \forall m_2. (\mathcal{C}_1', \mathcal{D}_1'', \mathcal{F}_1'; \mathcal{C}_2', \mathcal{D}_2'', \mathcal{F}_2') \in \mathcal{V}_\Psi^\xi [\![\Delta \Vdash K]\!]_{\cdot; x}^{m_2}$ and
$\quad \forall x \in (\Theta_2 \cap \Theta_3). \forall m_2. (\mathcal{C}_1', \mathcal{D}_1'', \mathcal{F}_1'; \mathcal{C}_2', \mathcal{D}_2'', \mathcal{F}_2') \in \mathcal{V}_\Psi^\xi [\![\Delta \Vdash K]\!]_{x; \cdot}^{m_2}$.

By $\star$, we know that $(\mathcal{C}_1 \mathcal{D}_1, \mathcal{F}_1, \cdot) \mapsto_{K \Vdash \cdot}^{*\Theta_i, y_c^\alpha; \Upsilon_2} (\mathcal{C}_1 \mathcal{D}_1, \mathcal{F}_1'', \cdot)$. We apply it on the condition of line 13 in the logical relation in $\dagger_3$. By $\forall E$ and $\exists E$ and Lemma C.14, we know that there exists a minimal $\mathcal{F}_2''$ such that $\mathcal{C}_2 \mathcal{D}_2, \mathcal{F}_2, \cdot \mapsto_{K \Vdash \cdot}^{*\Theta'; \Upsilon_1} \mathcal{C}_2 \mathcal{D}_2, \mathcal{F}_2'', \cdot$ and

$\dagger_3'$ $\forall \mathcal{C}_1^1, \mathcal{C}_2^1$.if $\mathcal{C}_1 \mathcal{D}_1, \mathcal{F}_1'', \cdot \mapsto_{K \Vdash \cdot}^{*\Theta_1, \Upsilon_1} \mathcal{C}_1^1, \mathcal{F}_1'', \cdot$ and $\mathcal{C}_2 \mathcal{D}_2, \mathcal{F}_2'', \cdot \mapsto_{K \Vdash \cdot}^{*\Theta_2; \Upsilon_1} \mathcal{C}_2^1, \mathcal{F}_2'', \cdot$, then
$\quad \forall x \in \Upsilon_1. \forall m_3. (\mathcal{C}_1^1, \mathcal{F}_1'', \cdot; \mathcal{C}_2^1, \mathcal{F}_2'', \cdot) \in \mathcal{V}_\Psi^\xi [\![K \Vdash \cdot]\!]_{\cdot; x}^{m_3}$ and
$\quad \forall x \in (\Theta_2 \cap \Theta_3). \forall m_3. (\mathcal{C}_1^1, \mathcal{F}_1'', \cdot; \mathcal{C}_2^1, \mathcal{F}_2'', \cdot) \in \mathcal{V}_\Psi^\xi [\![K \Vdash \cdot]\!]_{x; \cdot}^{m_3}$.

Pick this $\mathcal{F}_2''$.

By $\star$, we know that for all trees $\mathcal{A}_1^\ell$ in $\mathcal{C}_1$ we have $(\cdot, \mathcal{A}_1^\ell, \mathcal{C}_1^\ell \mathcal{D}_1 \mathcal{F}_1) \mapsto_{\cdot \Vdash d_\ell : T_\ell}^{*\Theta_i, y_c^\alpha; \Upsilon_2} (\cdot, \mathcal{A}_1^{\ell''}, \mathcal{C}_1^\ell \mathcal{D}_1 \mathcal{F}_1)$. We apply it on the condition of line 13 in the logical relation in $\dagger_1^\ell$. By $\forall E$ and $\exists E$ and Lemma C.14, we know that there exists a minimal $\mathcal{A}_2^{\ell''}$ such that $\cdot, \mathcal{A}_2^\ell, \mathcal{C}_2^\ell \mathcal{D}_2 \mathcal{F}_2 \mapsto_{K \Vdash \cdot}^{*\Theta'; \Upsilon_1} \cdot, \mathcal{A}_2^{\ell''}, \mathcal{C}_2^\ell \mathcal{D}_2 \mathcal{F}_2$ and

$\dagger_1^{\ell'}$ $\forall \mathcal{F}_1^1, \mathcal{F}_2^1$.if $\cdot, \mathcal{A}_1^{\ell''}, \mathcal{C}_1^\ell \mathcal{D}_1 \mathcal{F}_1 \mapsto_{\cdot \Vdash d_\ell : T_\ell}^{*\Theta_1, \Upsilon_1} \cdot, \mathcal{A}_1^{\ell''}, \mathcal{F}_1^1$ and $\cdot, \mathcal{A}_2^{\ell''}, \mathcal{C}_2^\ell \mathcal{D}_2 \mathcal{F}_2 \mapsto_{\cdot \Vdash d_\ell : T_\ell}^{*\Theta_1, \Upsilon_1} \cdot, \mathcal{A}_2^{\ell''}, \mathcal{F}_2^1$, then
$\quad \forall x \in \Upsilon_1. \forall m_1. (\cdot, \mathcal{A}_1^{\ell''}, \mathcal{F}_1^1; \cdot, \mathcal{A}_2^{\ell''}, \mathcal{F}_2^1) \in \mathcal{V}_\Psi^\xi [\![\cdot \Vdash d_\ell : T_\ell]\!]_{\cdot; x}^{m_1}$ and
$\quad \forall x \in (\Theta_2 \cap \Theta_3). \forall m_1. (\cdot, \mathcal{A}_1^{\ell''}, \mathcal{F}_1^1; \cdot, \mathcal{A}_2^{\ell''}, \mathcal{F}_2^1 \in \mathcal{V}_\Psi^\xi [\![\cdot \Vdash d_\ell : T_\ell]\!]_{x; \cdot}^{m_1}$.

Build $\mathcal{C}_2''$ as a forest consisting of the trees $\mathcal{A}_2^{\ell''}$, i.e. $\mathcal{C}_2'' = \{\mathcal{A}_2^{\ell''}\}_{\ell \in I}$ and put $\mathcal{C}_2^{\ell''}$ to be all trees in $\mathcal{C}_2''$ other than $\mathcal{A}_2^{\ell''}$. We instantiate the universal quantifiers in $\dagger_2'$ with $\mathcal{C}_1'', \mathcal{F}_1'', \mathcal{C}_2''$, and $\mathcal{F}_2''$ to get

$\dagger_1''$ $\forall x \in \Upsilon_1. \forall m_2. (\mathcal{C}_1'', \mathcal{D}_1'', \mathcal{F}_1''; \mathcal{C}_2'', \mathcal{D}_2'', \mathcal{F}_2'') \in \mathcal{V}_\Psi^\xi [\![\Delta \Vdash K]\!]_{\cdot; x}^{m_2}$ and
$\quad \forall x \in (\Theta_2 \cap \Theta_3). \forall m_2. (\mathcal{C}_1'', \mathcal{D}_1'', \mathcal{F}_1''; \mathcal{C}_2'', \mathcal{D}_2'', \mathcal{F}_2'') \in \mathcal{V}_\Psi^\xi [\![\Delta \Vdash K]\!]_{x; \cdot}^{m_2}$.

Similarly, from $\dagger_3'$, we get

$\dagger_3''$ $\forall x \in \Upsilon_1. \forall m_3. (\mathcal{C}_1'' \mathcal{D}_1'', \mathcal{F}_1'', \cdot; \mathcal{C}_2'' \mathcal{D}_2'', \mathcal{F}_2'', \cdot) \in \mathcal{V}_\Psi^\xi [\![K \Vdash \cdot]\!]_{\cdot; x}^{m_3}$ and
$\quad \forall x \in (\Theta_2 \cap \Theta_3). \forall m_3. (\mathcal{C}_1'' \mathcal{D}_1'', \mathcal{F}_1'', \cdot; \mathcal{C}_2'' \mathcal{D}_2'', \mathcal{F}_2'', \cdot) \in \mathcal{V}_\Psi^\xi [\![K \Vdash \cdot]\!]_{x; \cdot}^{m_3}$.

In particular, we have

$\star_2$ $\forall m_2. (\mathcal{C}_1'', \mathcal{D}_1'', \mathcal{F}_1''; \mathcal{C}_2'', \mathcal{D}_2'', \mathcal{F}_2'') \in \mathcal{V}_\Psi^\xi [\![\Delta \Vdash K]\!]_{\cdot; y_c^\alpha}^{m_2}$, and

$\star_3$ $\forall m_3. (\mathcal{C}_1'' \mathcal{D}_1'', \mathcal{F}_1'', \cdot; \mathcal{C}_2'' \mathcal{D}_2'', \mathcal{F}_2'', \cdot) \in \mathcal{V}_\Psi^\xi [\![K \Vdash \cdot]\!]_{y_c^\alpha; \cdot}^{m_3}$.

Moreover by the forward closure lemma on $\dagger_1^i$, we get

$\star_1^\ell$ $\forall m_1. (\cdot \mathcal{A}_1^{\ell''} \mathcal{C}_1^{\ell''} \mathcal{D}_1'' \mathcal{F}_1''; \cdot \mathcal{A}_2^{\ell''} \mathcal{C}_2^{\ell''} \mathcal{D}_2'' \mathcal{F}_2'') \in \mathcal{E}_\Psi^\xi [\![\cdot \Vdash d_j : T_j]\!]^{m_1}$.

Now we consider cases based on the row in the definition of logical relation (Figure 2) with which we can prove $\star_2$. It is a unique row determined based on the type of $y_c^\alpha$ and whether the process sending over $y_c^\alpha$ is a message or a forwarding process.

**Row (1).** The conditions validating this row are:

$\circ_1$ $(\cdot, \mathcal{D}_1'', \mathcal{F}_1''; \cdot, \mathcal{D}_2'', \mathcal{F}_2'') \in \mathsf{Tree}_\Psi(\cdot \Vdash y_\alpha^c : 1)$

$\circ_2$ $\mathcal{D}_1'' = \mathbf{msg}(\mathbf{close}\, y_\alpha^c)$

$\circ_3$ $\mathcal{D}_2'' = \mathbf{msg}(\mathbf{close}\, y_\alpha^c)$

Thus $(\mathcal{C}_i'', \mathcal{D}_i'', \mathcal{F}_i'') = (\cdot, \mathbf{msg}(\mathbf{close}\, y_\alpha^c), \mathcal{F}_i'')$. By Figure 4, we know that $(\cdot, \mathbf{msg}(\mathbf{close}\, y_\alpha^c), \mathcal{F}_1'') \xrightarrow[\Delta \Vdash K]{y_\alpha^c \cdot \mathbf{msg}(\mathbf{close}\, y_\alpha^c)} \cdot, \cdot, \cdot$, and thus $\mathcal{C}_1', \mathcal{D}_1', \mathcal{F}_1' = \cdot, \cdot, \cdot$. By Figure 4, we can step $(\mathcal{C}_2'', \mathcal{D}_2'', \mathcal{F}_2'') \xrightarrow[\Delta \Vdash K]{y_\alpha^c \cdot \mathbf{msg}(\mathbf{close}\, y_\alpha^c)} \cdot, \cdot, \cdot$. This completes the proof since $\cdot, \cdot, \cdot \sim^k \cdot, \cdot, \cdot$.

**Row (2).** The conditions validating this row are:

$$\circ_1 \; (\mathcal{C}_1'', \mathcal{D}_1'', \mathcal{F}_1''; \mathcal{C}_2'', \mathcal{D}_2'', \mathcal{F}_2'') \in \mathsf{Tree}_\Psi(\Delta \Vdash y_\alpha : \oplus\{\ell{:}A_\ell\}_{\ell \in I})$$

$$\circ_2 \; \mathcal{D}_1'' = \mathcal{D}_1^1 \mathbf{msg}(y_\alpha^c.k; y_\alpha^c \leftarrow u_\delta^c)$$

$$\circ_3 \; \mathcal{D}_2'' = \mathcal{D}_2^1 \mathbf{msg}(y_\alpha^c.k; y_\alpha^c \leftarrow u_\delta^c)$$

$$\circ_4 \; \forall k. \, (\mathcal{C}_1'' \mathcal{D}_1^1 \mathbf{msg}(y_\alpha^c.k; y_\alpha^c \leftarrow u_\delta^c) \mathcal{F}_1'', \mathcal{C}_2'' \mathcal{D}_2^1 \mathbf{msg}(y_\alpha^c.k; y_\alpha^c \leftarrow u_\delta^c) \mathcal{F}_2'') \in \mathcal{E}_\Psi^\xi [\![ \Delta \Vdash u_\delta{:}A_k[c] ]\!]^k$$

Thus $(\mathcal{C}_i'', \mathcal{D}_i'', \mathcal{F}_i'') = (\mathcal{C}_i'', \mathcal{D}_i^1 \mathbf{msg}(y_\alpha^c.k; y_\alpha^c \leftarrow u_\delta^c), \mathcal{F}_i'')$. By Figure 4, we know that

$$(\mathcal{C}_1'', \mathcal{D}_1^1 \mathbf{msg}(y_\alpha^c.k; y_\alpha^c \leftarrow u_\delta^c), \mathcal{F}_1'') \xrightarrow[\Delta \Vdash y_\alpha : \oplus\{\ell{:}A_\ell\}_{\ell \in I}[c]]{y_\alpha^c.\mathbf{msg}(y_\alpha^c.k; y_\alpha^c \leftarrow u_\delta^c)} (\mathcal{C}_1'', \mathcal{D}_1^1, \mathbf{msg}(y_\alpha^c.k; y_\alpha^c \leftarrow u_\delta^c) \mathcal{F}_1''),$$

and thus $\mathcal{C}_1', \mathcal{D}_1', \mathcal{F}_1' = (\mathcal{C}_1'', \mathcal{D}_1^1, \mathbf{msg}(y_\alpha^c.k; y_\alpha^c \leftarrow u_\delta^c) \mathcal{F}_1'')$. By Figure 4, we can step

$$(\mathcal{C}_2'', \mathcal{D}_2'', \mathcal{F}_2'') \xrightarrow[\Delta \Vdash y_\alpha : \oplus\{\ell{:}A_\ell\}_{\ell \in I}[c]]{y_\alpha^c.\mathbf{msg}(y_\alpha^c.k; y_\alpha^c \leftarrow u_\delta^c)} (\mathcal{C}_2'', \mathcal{D}_2^1, \mathbf{msg}(y_\alpha^c.k; y_\alpha^c \leftarrow u_\delta^c) \mathcal{F}_2'').$$

By the typing rules, we know that

$$\circ_5 \; (\mathcal{C}_1'', \mathcal{D}_1^1, \mathbf{msg}(y_\alpha^c.k; y_\alpha^c \leftarrow u_\delta^c) \mathcal{F}_1''), \mathcal{C}_2'', \mathcal{D}_2^1, \mathbf{msg}(y_\alpha^c.k; y_\alpha^c \leftarrow u_\delta^c) \mathcal{F}_2'') \in \mathsf{Tree}(\Delta \Vdash u_\delta{:}A_k[c]).$$

And by row 7 and $\star_3$, we get

$$\circ_6 \; \forall m_3. \, (\mathcal{C}_1'' \mathcal{D}_1^1, \mathbf{msg}(y_\alpha^c.k; y_\alpha^c \leftarrow u_\delta^c) \mathcal{F}_1'', \cdot; \mathcal{C}_2'' \mathcal{D}_2^1, \mathbf{msg}(y_\alpha^c.k; y_\alpha^c \leftarrow u_\delta^c) \mathcal{F}_2'', \cdot) \in \mathcal{V}_\Psi^\xi [\![ K \Vdash \cdot ]\!]^{m_3}_{y_\alpha^c; \cdot}.$$

We can apply the induction hypothesis, on $\circ_4$, $\circ_5$, $\circ_6$, and $\star_1^\ell$ to get

$$\mathcal{C}_1'', \mathcal{D}_1^1, \mathbf{msg}(y_\alpha^c.k; y_\alpha^c \leftarrow u_\delta^c) \mathcal{F}_1'' \sim^k \mathcal{C}_2'', \mathcal{D}_2^1, \mathbf{msg}(y_\alpha^c.k; y_\alpha^c \leftarrow u_\delta^c) \mathcal{F}_2''.$$

which completes the proof.

**Row (4).** The conditions validating this row are:

$$\circ_1 \; \mathcal{C}_i^1 \in \mathsf{Tree}_\Psi(\cdot \Vdash \Delta') \text{ and } (\mathcal{C}_1^1 \mathcal{C}_1^2, \mathcal{D}_1'', \mathcal{F}_1''; \mathcal{C}_2^1 \mathcal{C}_2^2, \mathcal{D}_2'', \mathcal{F}_2'') \in \mathsf{Tree}_\Psi(\Delta', \Delta'' \Vdash y_\alpha{:}A \otimes B[c])$$

$$\circ_2 \; \mathcal{D}_1'' = \mathcal{D}_1^1 \mathcal{T}_1 \mathbf{msg}(\mathbf{send} x_\beta^c \, y_\alpha^c; y_\alpha^c \leftarrow u_\delta^c) \text{ for } \mathcal{T}_1 \in \mathsf{Tree}_\Psi(\Delta'' \Vdash x_\beta{:}A[c])$$

$$\circ_3 \; \mathcal{D}_2'' = \mathcal{D}_2^1 \mathcal{T}_2 \mathbf{msg}(\mathbf{send}, x_\beta^c \, y_\alpha^c; y_\alpha^c \leftarrow u_\delta^c) \text{ for } \mathcal{T}_2 \in \mathsf{Tree}_\Psi(\Delta'' \Vdash x_\beta{:}A[c])$$

$$\circ_4 \; \forall k. \, (\mathcal{C}_1^2 \mathcal{T}_1 \mathcal{C}_1^1 \mathcal{D}_1^1 \mathbf{msg}(\mathbf{send} x_\beta^c \, y_\alpha^c; y_\alpha^c \leftarrow u_\delta^c) \mathcal{F}_1'', \mathcal{C}_2^2 \mathcal{T}_2 \mathcal{C}_2^1 \mathcal{D}_2^1 \mathbf{msg}(\mathbf{send} x_\beta^c \, y_\alpha^c; y_\alpha^c \leftarrow u_\delta^c) \mathcal{F}_2'') \in \mathcal{E}_\Psi^\xi [\![ \Delta'' \Vdash x_\beta{:}A[c] ]\!]^k$$

$$\circ_5 \; \forall k. \, (\mathcal{C}_1^1 \mathcal{D}_1^1 \mathcal{C}_1^2 \mathcal{T}_1 \mathbf{msg}(\mathbf{send} x_\beta^c \, y_\alpha^c; y_\alpha^c \leftarrow u_\delta^c) \mathcal{F}_1'', \mathcal{C}_2^1 \mathcal{D}_2^1 \mathcal{C}_2^2 \mathcal{T}_2 \mathbf{msg}(\mathbf{send} x_\beta^c \, y_\alpha^c; y_\alpha^c \leftarrow u_\delta^c) \mathcal{F}_2'') \in \mathcal{E}_\Psi^\xi [\![ \Delta' \Vdash u_\delta{:}B[c] ]\!]^k$$

Thus $(\mathcal{C}_i'', \mathcal{D}_i'', \mathcal{F}_i'') = (\mathcal{C}_i^1 \mathcal{C}_i^2, \mathcal{D}_i^1 \mathcal{T}_i \mathbf{msg}(\mathbf{send} x_\beta^c \, y_\alpha^c; y_\alpha^c \leftarrow u_\delta^c), \mathcal{F}_i'')$. By Figure 4, we know that

$$(\mathcal{C}_1^1 \mathcal{C}_1^2, \mathcal{D}_1^1 \mathcal{T}_2 \mathbf{msg}(\mathbf{send} x_\beta^c \, y_\alpha^c; y_\alpha^c \leftarrow u_\delta^c), \mathcal{F}_1'') \xrightarrow[\Delta', \Delta'' \Vdash y_\alpha : (A \otimes B)[c]]{y_\alpha^c.\mathbf{msg}(\mathbf{send} x_\beta^c \, y_\alpha^c; y_\alpha^c \leftarrow u_\delta^c)} (\mathcal{C}_1^2, \mathcal{T}_1, \mathcal{C}_1^1 \mathcal{D}_1^1 \mathbf{msg}(\mathbf{send} x_\beta^c \, y_\alpha^c; y_\alpha^c \leftarrow u_\delta^c) \mathcal{F}_1''),$$

and

$$(\mathcal{C}_1^1 \mathcal{C}_1^2, \mathcal{D}_1^1 \mathcal{T}_1 \mathbf{msg}(\mathbf{send} x_\beta^c \, y_\alpha^c; y_\alpha^c \leftarrow u_\delta^c), \mathcal{F}_1'') \xrightarrow[\Delta', \Delta'' \Vdash y_\alpha : (A \otimes B)[c]]{y_\alpha^c.\mathbf{msg}(\mathbf{send} x_\beta^c \, y_\alpha^c; y_\alpha^c \leftarrow u_\delta^c)} (\mathcal{C}_1^1, \mathcal{D}_1^1, \mathcal{C}_1^2 \mathcal{T}_1 \mathbf{msg}(\mathbf{send} x_\beta^c \, y_\alpha^c; y_\alpha^c \leftarrow u_\delta^c) \mathcal{F}_1'').$$

Thus either

$$(a) \; (\mathcal{C}_1', \mathcal{D}_1', \mathcal{F}_1') = (\mathcal{C}_1^2, \mathcal{T}_1, \mathcal{C}_1^1 \mathcal{D}_1^1 \mathbf{msg}(\mathbf{send} x_\beta^c \, y_\alpha^c; y_\alpha^c \leftarrow u_\delta^c) \mathcal{F}_1'') \text{ or}$$

$$(b) \; (\mathcal{C}_1', \mathcal{D}_1', \mathcal{F}_1') = (\mathcal{C}_1^1, \mathcal{D}_1^1, \mathcal{C}_1^2 \mathcal{T}_1 \mathbf{msg}(\mathbf{send} x_\beta^c \, y_\alpha^c; y_\alpha^c \leftarrow u_\delta^c) \mathcal{F}_1'').$$

We proceed by cases:

(a) $(\mathcal{C}'_1, \mathcal{D}'_1, \mathcal{F}'_1) = (\mathcal{C}_1^2, \mathcal{T}_1, \mathcal{C}_1^1 \mathcal{D}_1^1 \mathbf{msg}(\mathbf{send} x_\beta^c\, y_\alpha^c; y_\alpha^c \leftarrow u_\delta^c) \mathcal{F}''_1)$.
By Figure 4, we can step
$$(\mathcal{C}''_2, \mathcal{D}''_2, \mathcal{F}''_2) \xrightarrow[\Delta', \Delta'' \Vdash y_\alpha:(A \otimes B)[c]]{y_\alpha^c \cdot \mathbf{msg}(\mathbf{send} x_\beta^c\, y_\alpha^c; y_\alpha^c \leftarrow u_\delta^c)} (\mathcal{C}_2^2, \mathcal{T}_2, \mathcal{C}_2^1 \mathcal{D}_2^1 \mathbf{msg}(\mathbf{send} x_\beta^c\, y_\alpha^c; y_\alpha^c \leftarrow u_\delta^c) \mathcal{F}''_2).$$

By the typing rules, we know that

$\circ_6$ $(\mathcal{C}_1^2 \mathcal{T}_1 \mathcal{C}_1^1 \mathcal{D}_1^1 \mathbf{msg}(\mathbf{send} x_\beta^c\, y_\alpha^c; y_\alpha^c \leftarrow u_\delta^c) \mathcal{F}''_1, \mathcal{C}_2^2 \mathcal{T}_2 \mathcal{C}_2^1 \mathcal{D}_2^1 \mathbf{msg}(\mathbf{send} x_\beta^c\, y_\alpha^c; y_\alpha^c \leftarrow u_\delta^c) \mathcal{F}''_2) \in \mathsf{Tree}(\Delta'' \Vdash x_\beta:A[c])$.

By $\star_3$, we get

$\circ_7$ $\forall m_3.\, (\mathcal{C}''_1 \mathcal{D}_1^1 \mathcal{T}_1 \mathbf{msg}(\mathbf{send} x_\beta^c\, y_\alpha^c; y_\alpha^c \leftarrow u_\delta^c) \mathcal{F}''_1 \cdot; \mathcal{C}''_2 \mathcal{D}_2^2 \mathcal{T}_2 \mathbf{msg}(\mathbf{send} x_\beta^c\, y_\alpha^c; y_\alpha^c \leftarrow u_\delta^c) \mathcal{F}''_2 \cdot) \in \mathcal{E}_\Psi^\xi [\![ x_\beta:A[c], u_\delta:B[c] \Vdash \cdot ]\!]^{m_3}$.

By Lemma C.16, from $\circ_7$ and $\circ_5$ and $\star_1^\ell$ for the trees in $\mathcal{C}_1^1$ and $\mathcal{C}_2^1$, we get

$\circ'_7$ $\forall m_3.\, (\mathcal{C}_1^2 \mathcal{T}_1 \mathcal{C}_1^1 \mathcal{D}_1^1 \mathbf{msg}(\mathbf{send} x_\beta^c\, y_\alpha^c; y_\alpha^c \leftarrow u_\delta^c) \mathcal{F}''_1 \cdot; \mathcal{C}''_2 \mathcal{T}_2 \mathcal{C}_1^1 \mathcal{D}_2^2 \mathbf{msg}(\mathbf{send} x_\beta^c\, y_\alpha^c; y_\alpha^c \leftarrow u_\delta^c) \mathcal{F}''_2 \cdot) \in \mathcal{E}_\Psi^\xi [\![ x_\beta:A[c] \Vdash \cdot ]\!]^{m_3}$.

We can apply the induction hypothesis, on $\circ_4$ and $\circ_6$ and $\circ'_7$ and $\star_1^\ell$ to get

$$\mathcal{C}_1^2, \mathcal{T}_1, \mathcal{C}_1^1 \mathcal{D}_1^1 \mathbf{msg}(\mathbf{send} x_\beta^c\, y_\alpha^c; y_\alpha^c \leftarrow u_\delta^c) \mathcal{F}''_1 \sim^k \mathcal{C}_2^2, \mathcal{T}_2, \mathcal{C}_2^1 \mathcal{D}_2^1 \mathbf{msg}(\mathbf{send} x_\beta^c\, y_\alpha^c; y_\alpha^c \leftarrow u_\delta^c) \mathcal{F}''_2.$$

which completes the proof of this case.

(b) $(\mathcal{C}'_1, \mathcal{D}'_1, \mathcal{F}'_1) = (\mathcal{C}_1^1, \mathcal{D}_1^1, \mathcal{C}_1^2 \mathcal{T}_1 \mathbf{msg}(\mathbf{send} x_\beta^c\, y_\alpha^c; y_\alpha^c \leftarrow u_\delta^c) \mathcal{F}''_1)$.
By Figure 4, we can step
$$(\mathcal{C}''_2, \mathcal{D}''_2, \mathcal{F}''_2) \xrightarrow[\Delta', \Delta'' \Vdash y_\alpha:(A \otimes B)[c]]{y_\alpha^c \cdot \mathbf{msg}(\mathbf{send} x_\beta^c\, y_\alpha^c; y_\alpha^c \leftarrow u_\delta^c)} (\mathcal{C}_2^1, \mathcal{D}_2^1, \mathcal{C}_2^2 \mathcal{T}_2 \mathbf{msg}(\mathbf{send} x_\beta^c\, y_\alpha^c; y_\alpha^c \leftarrow u_\delta^c) \mathcal{F}''_2).$$

By the typing rules, we know that

$\circ_8$ $(\mathcal{C}_1^1, \mathcal{D}_1^1, \mathcal{C}_1^2 \mathcal{T}_1 \mathbf{msg}(\mathbf{send} x_\beta^c\, y_\alpha^c; y_\alpha^c \leftarrow u_\delta^c) \mathcal{F}''_1), (\mathcal{C}_2^1, \mathcal{D}_2^1, \mathcal{C}_2^2 \mathcal{T}_2 \mathbf{msg}(\mathbf{send} x_\beta^c\, y_\alpha^c; y_\alpha^c \leftarrow u_\delta^c) \mathcal{F}''_2) \in \mathsf{Tree}(\Delta' \Vdash u_\delta:B[c])$.

By $\star_3$, we get

$\circ_9$ $\forall m_3.\, (\mathcal{C}''_1 \mathcal{D}_1^1 \mathcal{T}_1 \mathbf{msg}(\mathbf{send} x_\beta^c\, y_\alpha^c; y_\alpha^c \leftarrow u_\delta^c) \mathcal{F}''_1 \cdot; \mathcal{C}''_2 \mathcal{D}_2^2 \mathcal{T}_2 \mathbf{msg}(\mathbf{send} x_\beta^c\, y_\alpha^c; y_\alpha^c \leftarrow u_\delta^c) \mathcal{F}''_2 \cdot) \in \mathcal{E}_\Psi^\xi [\![ x_\beta:A[c], u_\delta:B[c] \Vdash \cdot ]\!]^{m_3}$.

By Lemma C.16, from $\circ_9$ and $\circ_4$ and $\star_1^\ell$ for the trees in $\mathcal{C}_1^2$ and $\mathcal{C}_2^2$, we get

$\circ'_9$ $\forall m_3.\, (\mathcal{C}_1^1 \mathcal{D}_1^1 \mathcal{C}_1^2 \mathcal{T}_1 \mathbf{msg}(\mathbf{send} x_\beta^c\, y_\alpha^c; y_\alpha^c \leftarrow u_\delta^c) \mathcal{F}''_1 \cdot; \mathcal{C}_2^1 \mathcal{D}_2^1 \mathcal{C}_2^2 \mathcal{T}_2 \mathbf{msg}(\mathbf{send} x_\beta^c\, y_\alpha^c; y_\alpha^c \leftarrow u_\delta^c) \mathcal{F}''_2 \cdot) \in \mathcal{E}_\Psi^\xi [\![ u_\delta:B[c] \Vdash \cdot ]\!]^{m_3}$.

We can apply the induction hypothesis, on $\circ_5$ and $\circ 8$ and $\circ'_9$ and $\star_1^\ell$ to get

$$\mathcal{C}_1^1, \mathcal{D}_1^1, \mathcal{C}_1^2 \mathcal{T}_1 \mathbf{msg}(\mathbf{send} x_\beta^c\, y_\alpha^c; y_\alpha^c \leftarrow u_\delta^c) \mathcal{F}''_1) \sim^k \mathcal{C}_2^1, \mathcal{D}_2^1, \mathcal{C}_2^2 \mathcal{T}_2 \mathbf{msg}(\mathbf{send} x_\beta^c\, y_\alpha^c; y_\alpha^c \leftarrow u_\delta^c) \mathcal{F}''_2.$$

which completes the proof of this case.

**Row (11).** The conditions validating this row are:

$\circ_1$ $(\mathcal{C}''_1, \mathcal{D}''_1, \mathcal{F}''_1; \mathcal{C}''_2, \mathcal{D}''_2, \mathcal{F}''_2) \in \mathsf{Tree}_\Psi(\Delta \Vdash y_\alpha:A[c])$

$\circ_2$ $\mathcal{D}''_1 = \mathcal{D}_1^1 \mathbf{proc}(y_\alpha^c, y_\alpha^c \leftarrow x_\beta^c @ d_1)$

$\circ_3$ $\mathcal{D}''_2 = \mathcal{D}_2^1 \mathbf{proc}(y_\alpha^c, y_\alpha^c \leftarrow x_\beta^c @ d_2)$

$\circ_4$ $\forall m_2.([x_\beta^c/y_\alpha^c]\mathcal{C}''_1 \mathcal{D}_1^1 \mathcal{F}''_1, [x_\beta^c/y_\alpha^c]\mathcal{C}''_2 \mathcal{D}_2^1 \mathcal{F}''_2) \in \mathcal{E}_\Psi^\xi [\![ \Delta \Vdash x_\beta:A[c] ]\!]^{m_2}$

Thus $(\mathcal{C}''_i, \mathcal{D}''_i, \mathcal{F}''_i) = (\mathcal{C}''_i, \mathcal{D}_i^1 \mathbf{proc}(y_\alpha^c, y_\alpha^c \leftarrow x_\beta^c @ d_i), \mathcal{F}''_i)$. By Figure 4, we know that

$$(\mathcal{C}''_1, \mathcal{D}_1^1 \mathbf{proc}(y_\alpha^c, y_\alpha^c \leftarrow x_\beta^c @ d_1), \mathcal{F}''_1) \xrightarrow[\Delta \Vdash y_\alpha:A[c]]{y_\alpha^c \cdot \{x_\beta^c/y_\alpha^c\}} ([x_\beta^c/y_\alpha^c]\mathcal{C}''_1, \mathcal{D}_1^1, \mathcal{F}''_1),$$

58

and thus $\mathcal{C}'_1, \mathcal{D}'_1, \mathcal{F}'_1 = ([x^c_\beta/y^c_\alpha]\mathcal{C}''_1, \mathcal{D}^1_1, \mathcal{F}''_1)$. By Figure 4, we can step

$$(\mathcal{C}''_2, \mathcal{D}''_2, \mathcal{F}''_2) \xrightarrow[\Delta \Vdash y_\alpha:A[c]]{y^c_\alpha \cdot \{x^c_\beta/y^c_\alpha\}} ([x^c_\beta/y^c_\alpha]\mathcal{C}''_2, \mathcal{D}^1_2, \mathcal{F}''_2).$$

By the typing rules, we know that

$$\circ_5 \ ([x^c_\beta/y^c_\alpha]\mathcal{C}''_1, \mathcal{D}^1_1, \mathcal{F}''_1), [x^c_\beta/y^c_\alpha]\mathcal{C}''_2, \mathcal{D}^1_2, \mathcal{F}''_2) \in \mathsf{Tree}(\Delta \Vdash x_\beta:A[c]).$$

By $\star_3$, we get

$$\circ_6 \ \forall m_3. \ ([x^c_\beta/y^c_\alpha]\mathcal{C}''_1\mathcal{D}^1_1\mathcal{F}''_1 \cdot ; [x^c_\beta/y^c_\alpha]\mathcal{C}''_2\mathcal{D}^1_2\mathcal{F}''_2 \cdot) \in \mathcal{E}^\xi_\Psi [\![ x_\beta{:}A[c] \Vdash \cdot ]\!]^{m_3}.$$

By forward closure lemma on $\star^\ell_1$ and the fact that $y^c_\alpha$ does not occur in $\mathcal{C}''_i$, we get

$$\star_1^{\ell'} \ \forall m_1. \ ([x^c_\beta/y^c_\alpha]\cdot \mathcal{A}^{\ell''}_1 \mathcal{C}^{\ell''}_1 \mathcal{D}''_1 \mathcal{F}''_1; [x^c_\beta/y^c_\alpha]\cdot \mathcal{A}^{\ell''}_2 \mathcal{C}^{\ell''}_2 \mathcal{D}''_2 \mathcal{F}''_2) \in \mathcal{E}^\xi_\Psi [\![ \cdot \Vdash d_j{:}T_j ]\!]^{m_1}.$$

We can apply the induction hypothesis, on $\circ_5$ and $\circ_4$ and $\circ_6$ and $\star_1^{\ell'}$ to get

$$[x^c_\beta/y^c_\alpha]\mathcal{C}''_1, \mathcal{D}^1_1, \mathcal{F}''_1 \sim^k [x^c_\beta/y^c_\alpha]\mathcal{C}''_2, \mathcal{D}^1_2, \mathcal{F}''_2.$$

which completes the proof.

**Case 2.** $(\mathcal{C}_1, \mathcal{D}_1, \mathcal{F}_1) \xrightarrow[\Delta \Vdash K]{y^\alpha_c \cdot q_1} (\mathcal{C}'_1, \mathcal{D}'_1, \mathcal{F}'_1)$, and $y^\alpha_c \in \Delta$. Without loss of generality assume that the tree in $\mathcal{C}_i$ that offers along $y^\alpha_c$ is $\mathcal{A}^\kappa_i$, i.e. $\cdot \Vdash \mathcal{A}^\kappa_i :: y^\alpha_c{:}T$.

By definition $\star \ (\mathcal{C}_1, \mathcal{D}_1, \mathcal{F}_1) \mapsto^{*\Theta;\Upsilon, y^\alpha_c}_{\Delta \Vdash K} (\mathcal{C}''_1, \mathcal{D}''_1, \mathcal{F}''_1)$, and $\dagger \ (\mathcal{C}''_1, \mathcal{D}''_1, \mathcal{F}''_1) \xrightarrow[\Delta \Vdash K]{y^\alpha_c \cdot q_1} (\mathcal{C}'_1, \mathcal{D}'_1, \mathcal{F}'_1)$.
We apply $\star$ on the condition of line 13 in the logical relation in $\dagger_2$ for any $m_2 \geq 1$ to eliminate the if. By $\forall E$ and $\exists E$ and Lemma C.14, we know that there exists a minimal $\mathcal{D}''_2$ such that
$\circ_1 \ \mathcal{C}_2, \mathcal{D}_2, \mathcal{F}_2 \mapsto^{*\Theta';\Upsilon_1}_{\Delta \Vdash K} \mathcal{C}_2, \mathcal{D}'_2, \mathcal{F}_2$ and

$$\dagger'_2 \ \forall \mathcal{C}'_1, \mathcal{F}'_1, \mathcal{C}'_2, \mathcal{F}'_2. \text{if } \mathcal{C}_1, \mathcal{D}''_1, \mathcal{F}_1 \mapsto^{*\Theta_1,\Upsilon_1}_{\Delta \Vdash K} \mathcal{C}'_1, \mathcal{D}''_1, \mathcal{F}'_1 \text{ and } \mathcal{C}_2, \mathcal{D}''_2, \mathcal{F}_2 \mapsto^{*\Theta_2;\Upsilon_1}_{\Delta \Vdash K} \mathcal{C}'_2, \mathcal{D}''_2, \mathcal{F}'_2, \text{ then }$$
$$\forall x \in \Upsilon_1. \forall m_2. \ (\mathcal{C}'_1, \mathcal{D}''_1, \mathcal{F}'_1; \mathcal{C}'_2, \mathcal{D}''_2, \mathcal{F}'_2) \in \mathcal{V}^\xi_\Psi [\![ \Delta \Vdash K ]\!]^{m_2}_{\cdot; x} \text{ and }$$
$$\forall x \in (\Theta_2 \cap \Theta_3). \forall m_2. \ (\mathcal{C}'_1, \mathcal{D}''_1, \mathcal{F}'_1; \mathcal{C}'_2, \mathcal{D}''_2, \mathcal{F}'_2) \in \mathcal{V}^\xi_\Psi [\![ \Delta \Vdash K ]\!]^{m_2}_{x; \cdot}.$$

By $\star$, we know that $(\mathcal{C}_1 \mathcal{D}_1, \mathcal{F}_1, \cdot) \mapsto^{*\Theta_i, y^\alpha_c;\Upsilon_2}_{K \Vdash \cdot} (\mathcal{C}_1 \mathcal{D}_1, \mathcal{F}''_1, \cdot)$. We apply it on the condition of line 13 in the logical relation in $\dagger_3$.
By $\forall E$ and $\exists E$ and Lemma C.14, we know that there exists a minimal $\mathcal{F}''_2$ such that $\mathcal{C}_2 \mathcal{D}_2, \mathcal{F}_2, \cdot \mapsto^{*\Theta';\Upsilon_1}_{K \Vdash \cdot} \mathcal{C}_2 \mathcal{D}_2, \mathcal{F}''_2, \cdot$ and

$$\dagger'_3 \ \forall \mathcal{C}^1_1, \mathcal{C}^1_2. \text{if } \mathcal{C}_1 \mathcal{D}_1, \mathcal{F}''_1, \cdot \mapsto^{*\Theta_1,\Upsilon_1}_{K \Vdash \cdot} \mathcal{C}^1_1, \mathcal{F}''_1, \cdot \text{ and } \mathcal{C}_2 \mathcal{D}_2, \mathcal{F}''_2, \cdot \mapsto^{*\Theta_2;\Upsilon_1}_{K \Vdash \cdot} \mathcal{C}^1_2, \mathcal{F}''_2, \cdot, \text{ then }$$
$$\forall x \in \Upsilon_1. \forall m_3. \ (\mathcal{C}^1_1, \mathcal{F}''_1, \cdot; \mathcal{C}^1_2, \mathcal{F}''_2, \cdot) \in \mathcal{V}^\xi_\Psi [\![ K \Vdash \cdot ]\!]^{m_3}_{\cdot; x} \text{ and }$$
$$\forall x \in (\Theta_2 \cap \Theta_3). \forall m_3. \ (\mathcal{C}^1_1, \mathcal{F}''_1, \cdot; \mathcal{C}^1_2, \mathcal{F}''_2, \cdot) \in \mathcal{V}^\xi_\Psi [\![ K \Vdash \cdot ]\!]^{m_3}_{x; \cdot}.$$

Pick this $\mathcal{F}''_2$.
By $\star$, we know that for all trees $\mathcal{A}^\ell_1$ in $\mathcal{C}_1$ we have $(\cdot, \mathcal{A}^\ell_1, \mathcal{C}^\ell_1 \mathcal{D}_1 \mathcal{F}_1) \mapsto^{*\Theta_i, y^\alpha_c;\Upsilon_2}_{\cdot \Vdash d_\ell:T_\ell} (\cdot, \mathcal{A}^{\ell''}_1, \mathcal{C}^\ell_1 \mathcal{D}_1 \mathcal{F}_1)$. We apply it on the condition of line 13 in the logical relation in $\dagger^\ell_1$. By $\forall E$ and $\exists E$ and Lemma C.14, we know that there exists a minimal $\mathcal{A}^{\ell''}_2$ such that $\cdot, \mathcal{A}^\ell_2, \mathcal{C}^\ell_2 \mathcal{D}_2 \mathcal{F}_2 \mapsto^{*\Theta';\Upsilon_1}_{K \Vdash \cdot} \cdot, \mathcal{A}^{\ell''}_2, \mathcal{C}^\ell_2 \mathcal{D}_2 \mathcal{F}_2$ and

$$\dagger^{\ell'}_1 \ \forall \mathcal{F}^1_1, \mathcal{F}^1_2. \text{if } \cdot, \mathcal{A}^{\ell''}_1, \mathcal{C}^\ell_1 \mathcal{D}_1 \mathcal{F}_1 \mapsto^{*\Theta_1,\Upsilon_1}_{\cdot \Vdash d_\ell:T_\ell} \cdot, \mathcal{A}^{\ell''}_1, \mathcal{F}^1_1 \text{ and } \cdot, \mathcal{A}^{\ell''}_2, \mathcal{C}^\ell_2 \mathcal{D}_2 \mathcal{F}_2 \mapsto^{*\Theta_1,\Upsilon_1}_{\cdot \Vdash d_\ell:T_\ell} \cdot, \mathcal{A}^{\ell''}_2, \mathcal{F}^1_2, \text{ then }$$
$$\forall x \in \Upsilon_1. \forall m_1. \ (\cdot, \mathcal{A}^{\ell''}_1, \mathcal{F}^1_1; \cdot, \mathcal{A}^{\ell''}_2, \mathcal{F}^1_2) \in \mathcal{V}^\xi_\Psi [\![ \cdot \Vdash d_\ell:T_\ell ]\!]^{m_1}_{\cdot; x} \text{ and }$$
$$\forall x \in (\Theta_2 \cap \Theta_3). \forall m_1. \ (\cdot, \mathcal{A}^{\ell''}_1, \mathcal{F}^1_1; \cdot, \mathcal{A}^{\ell''}_2, \mathcal{F}^1_2 \in \mathcal{V}^\xi_\Psi [\![ \cdot \Vdash d_\ell:T_\ell ]\!]^{m_1}_{x; \cdot}.$$

Build $\mathcal{C}''_2$ as a forest consisting of the trees $\mathcal{A}^{\ell''}_2$, i.e. $\mathcal{C}''_2 = \{\mathcal{A}^{\ell''}_2\}_{\ell \in I}$ and put $\mathcal{C}^{\ell''}_2$ to be all trees in $\mathcal{C}''_2$ other than $\mathcal{A}^{\ell''}_2$. We instantiate the universal quantifiers in $\dagger'_2$ with $\mathcal{C}''_1, \mathcal{F}''_1, \mathcal{C}''_2$, and $\mathcal{F}''_2$ to get

$$\dagger''_1 \ \forall x \in \Upsilon_1. \forall m_2. \ (\mathcal{C}''_1, \mathcal{D}''_1, \mathcal{F}''_1; \mathcal{C}''_2, \mathcal{D}''_2, \mathcal{F}''_2) \in \mathcal{V}^\xi_\Psi [\![ \Delta \Vdash K ]\!]^{m_2}_{\cdot; x} \text{ and }$$
$$\forall x \in (\Theta_2 \cap \Theta_3). \forall m_2. \ (\mathcal{C}''_1, \mathcal{D}''_1, \mathcal{F}''_1; \mathcal{C}''_2, \mathcal{D}''_2, \mathcal{F}''_2) \in \mathcal{V}^\xi_\Psi [\![ \Delta \Vdash K ]\!]^{m_2}_{x; \cdot}.$$

Similarly, from $\dagger^{\ell'}_1$, we get

$$\dagger_{\ell''} \ \forall x \in \Upsilon_1. \forall m_1. \ (\cdot, \mathcal{A}^{\ell''}_1, \mathcal{C}^{\ell''}_1 \mathcal{D}''_1 \mathcal{F}''_1; \cdot, \mathcal{A}^{\ell''}_2, \mathcal{C}^{\ell''}_2 \mathcal{D}''_2 \mathcal{F}''_2) \in \mathcal{V}^\xi_\Psi [\![ K \Vdash \cdot ]\!]^{m_1}_{\cdot; x} \text{ and }$$
$$\forall x \in (\Theta_2 \cap \Theta_3). \forall m_1. \ (\cdot, \mathcal{A}^{\ell''}_1, \mathcal{C}^{\ell''}_1 \mathcal{D}''_1 \mathcal{F}''_1; \cdot, \mathcal{A}^{\ell''}_2, \mathcal{C}^{\ell''}_2 \mathcal{D}''_2 \mathcal{F}''_2) \in \mathcal{V}^\xi_\Psi [\![ K \Vdash \cdot ]\!]^{m_1}_{x; \cdot}.$$

In particular, we have

$$\star_2 \,\forall m_2.\, (\mathcal{C}_1'', \mathcal{D}_1'', \mathcal{F}_1''; \mathcal{C}_2'', \mathcal{D}_2'', \mathcal{F}_2'') \in \mathcal{V}_\Psi^\xi [\![ \Delta \Vdash K ]\!]_{\cdot; y_c^\alpha}^{m_2}, \text{ and}$$

$$\star_1^\ell \,\forall m_3.\, (\cdot, \mathcal{A}_1^{\ell''}, \mathcal{C}_1^{\ell''} \mathcal{D}_1'' \mathcal{F}_1''; \cdot, \mathcal{A}_2^{\ell''}, \mathcal{C}_2^{\ell''} \mathcal{D}_2'' \mathcal{F}_2'') \in \mathcal{V}_\Psi^\xi [\![ \cdot \Vdash y_c^\alpha {:} T ]\!]_{y_c^\alpha; \cdot}^{m_3}.$$

Moreover by the forward closure lemma on $\dagger_3'$, we get

$$\star_3 \,\forall m_3.\, (\mathcal{C}_1'' \mathcal{D}_1'' \mathcal{F}_1'' \cdot; \mathcal{C}_2'' \mathcal{D}_2'' \mathcal{F}_2'' \cdot) \in \mathcal{E}_\Psi^\xi [\![ K \Vdash \cdot ]\!]^{m_3}.$$

And by the forward closure lemma on $\dagger_1^i$, we get

$$\star_4^\ell \,\forall m_1.\, (\cdot \mathcal{A}_1^{\ell''} \mathcal{C}_1^{\ell''} \mathcal{D}_1'' \mathcal{F}_1''; \cdot \mathcal{A}_2^{\ell''} \mathcal{C}_2^{\ell''} \mathcal{D}_2'' \mathcal{F}_2'') \in \mathcal{E}_\Psi^\xi [\![ \cdot \Vdash d_j{:}T_j ]\!]^{m_1}.$$

Now we consider cases based on the row in the definition of logical relation (Figure 2) with which we can prove $\star_2$. It is a unique row determined based on the type of $y_c^\alpha$ and whether the process sending over $y_c^\alpha$ is a message or a forwarding process.

**Row (8).** The conditions validating this row are:

$$\circ_1 \,(\mathcal{C}_1'', \mathcal{D}_1'', \mathcal{F}_1''; \mathcal{C}_2'', \mathcal{D}_2'', \mathcal{F}_2'') \in \mathsf{Tree}_\Psi(\Delta, y_\alpha{:}\&\{\ell{:}A_\ell\}_{\ell \in I}[c] \Vdash K)$$

$$\circ_2 \,\mathcal{D}_1'' = \mathbf{msg}(y_\alpha^c.k; u_\delta^c \leftarrow y_\alpha^c) \mathcal{D}_1^1$$

$$\circ_3 \,\mathcal{D}_2'' = \mathbf{msg}(y_\alpha^c.k; u_\delta^c \leftarrow y_\alpha^c) \mathcal{D}_2^1$$

$$\circ_4 \,\forall m_2.\, (\mathcal{C}_1'' \mathbf{msg}(y_\alpha^c.k; u_\delta^c \leftarrow y_\alpha^c) \mathcal{D}_1^1 \mathcal{F}_1'', \mathcal{C}_2'' \mathbf{msg}(y_\alpha^c.k; u_\delta^c \leftarrow y_\alpha^c) \mathcal{D}_2^1 \mathcal{F}_2'') \in \mathcal{E}_\Psi^\xi [\![ \Delta, u_\delta{:}A_k[c] \Vdash K ]\!]^{m_2}$$

Thus $(\mathcal{C}_i'', \mathcal{D}_i'', \mathcal{F}_i'') = (\mathcal{C}_i'', \mathbf{msg}(y_\alpha^c.k; u_\delta^c \leftarrow y_\alpha^c) \mathcal{D}_i^1, \mathcal{F}_i'')$. By Figure 4, we know that

$$(\mathcal{C}_1'', \mathbf{msg}(y_\alpha^c.k; u_\delta^c \leftarrow y_\alpha^c) \mathcal{D}_i^1, \mathcal{F}_1'') \xrightarrow[\Delta, y_\alpha{:}\&\{\ell{:}A_\ell\}_{\ell \in I}[c] \Vdash K]{y_\alpha^c.\mathbf{msg}(y_\alpha^c.k; u_\delta^c \leftarrow y_\alpha^c)} (\mathcal{C}_1'' \mathbf{msg}(y_\alpha^c.k; u_\delta^c \leftarrow y_\alpha^c), \mathcal{D}_1^1, \mathcal{F}_1''),$$

and thus $\mathcal{C}_1', \mathcal{D}_1', \mathcal{F}_1' = (\mathcal{C}_1'' \mathbf{msg}(y_\alpha^c.k; u_\delta^c \leftarrow y_\alpha^c), \mathcal{D}_1^1, \mathcal{F}_1'')$. By Figure 4, we can step

$$(\mathcal{C}_2'', \mathcal{D}_2'', \mathcal{F}_2'') \xrightarrow[\Delta, y_\alpha{:}\&\{\ell{:}A_\ell\}_{\ell \in I}[c] \Vdash K]{y_\alpha^c.\mathbf{msg}(y_\alpha^c.k; u_\delta^c \leftarrow y_\alpha^c)} (\mathcal{C}_2'' \mathbf{msg}(y_\alpha^c.k; u_\delta^c \leftarrow y_\alpha^c), \mathcal{D}_2^1, \mathcal{F}_2'').$$

By the typing rules, we know that

$$\circ_5 \,(\mathcal{C}_1'' \mathbf{msg}(y_\alpha^c.k; u_\delta^c \leftarrow y_\alpha^c), \mathcal{D}_1^1, \mathcal{F}_1''), \mathcal{C}_2'' \mathbf{msg}(y_\alpha^c.k; u_\delta^c \leftarrow y_\alpha^c), \mathcal{D}_2^1, \mathcal{F}_2'') \in \mathsf{Tree}(\Delta, u_\delta{:}A_k[c] \Vdash K).$$

And by row 8 and $\star_1^\kappa$ we get

$$\circ_6 \,\forall m_1.\, (\cdot \mathcal{A}_1^{\kappa''} \mathbf{msg}(y_\alpha^c.k; y_\alpha^c \leftarrow u_\delta^c) \mathcal{C}_1^{\kappa''} \mathcal{D}_1^1 \mathcal{F}_1'' \cdot; \mathcal{A}_1^{\kappa''} \mathbf{msg}(y_\alpha^c.k; y_\alpha^c \leftarrow u_\delta^c) \mathcal{C}_2^{\kappa''} \mathcal{D}_2^1 \mathcal{F}_2'' \cdot) \in \mathcal{E}_\Psi^\xi [\![ \cdot \Vdash y_\alpha{:}A_k[c] ]\!]^{m_1}.$$

We can apply the induction hypothesis, on $\circ_4$, $\circ_5$, $\circ_6$, $\star_3$, and $\star_4^\ell$ (for every $\ell \neq \kappa$) to get

$$\mathcal{C}_1'' \mathbf{msg}(y_\alpha^c.k; u_\delta^c \leftarrow y_\alpha^c), \mathcal{D}_1^1, \mathcal{F}_1'' \sim^k \mathcal{C}_2'' \mathbf{msg}(y_\alpha^c.k; u_\delta^c \leftarrow y_\alpha^c), \mathcal{D}_2^1, \mathcal{F}_2''.$$

which completes the proof.

**Row (10).** The conditions validating this row are:

$$\circ_1 \,\mathcal{C}_i^2 \in \mathsf{Tree}(\cdot \Vdash \Delta') \text{ and } (\mathcal{C}_1^1 \mathcal{C}_1^2, \mathcal{D}_1'', \mathcal{F}_1''; \mathcal{C}_2^1 \mathcal{C}_2^2, \mathcal{D}_2'', \mathcal{F}_2'') \in \mathsf{Tree}_\Psi(\Delta', \Delta'', y_\alpha{:}A \multimap B[c] \Vdash K)$$

$$\circ_2 \,\mathcal{D}_1'' = \mathcal{T}_1 \mathbf{msg}(\mathbf{send}\, x_\beta^c \, y_\alpha^c; ; u_\delta^c \leftarrow y_\alpha^c)\, \mathcal{D}_1^1 \text{ for } \mathcal{T}_1 \in \mathsf{Tree}(\Delta' \Vdash x_\beta^c)$$

$$\circ_3 \,\mathcal{D}_2'' = \mathcal{T}_2 \mathbf{msg}(\mathbf{send}\, x_\beta^c \, y_\alpha^c; ; u_\delta^c \leftarrow y_\alpha^c)\, \mathcal{D}_2^1 \text{ for } \mathcal{T}_2 \in \mathsf{Tree}(\Delta' \Vdash x_\beta^c)$$

$$\circ_4 \,\forall m_2.\, (\mathcal{C}_1^2 \mathcal{T}_1 \mathcal{C}_1^1 \mathbf{msg}(\mathbf{send}\, x_\beta^c \, y_\alpha^c; u_\delta^c \leftarrow y_\alpha^c) \mathcal{D}_1^1 \mathcal{F}_1'', \mathcal{C}_2^2 \mathcal{T}_2 \mathcal{C}_2^1 \mathbf{msg}(\mathbf{send}\, x_\beta^c \, y_\alpha^c; u_\delta^c \leftarrow y_\alpha^c) \mathcal{D}_2^1 \mathcal{F}_2'') \in \mathcal{E}_\Psi^\xi [\![ \Delta' \Vdash x_\beta{:}A[c] ]\!]^{m_2}$$

$\circ_5 \forall m_2. (\mathcal{C}_1^1 \mathcal{C}_1^2 \mathcal{T}_1 \mathbf{msg}(\mathbf{send}x_\beta^c\, y_\alpha^c; u_\delta^c \leftarrow y_\alpha^c)\mathcal{D}_1^1 \mathcal{F}_1'', \mathcal{C}_2^1 \mathcal{C}_2^2 \mathcal{T}_2 \mathbf{msg}(\mathbf{send}x_\beta^c\, y_\alpha^c; u_\delta^c \leftarrow y_\alpha^c)\mathcal{D}_2^1 \mathcal{F}_2'') \in \mathcal{E}_\Psi^\xi [\![\Delta'', u_\delta{:}B[c] \Vdash K]\!]^{m_2}$.

Thus $(\mathcal{C}_i'', \mathcal{D}_i'', \mathcal{F}_i'') = (\mathcal{C}_i^1 \mathcal{C}_i^2, \mathcal{T}_i \mathbf{msg}(\mathbf{send}x_\beta^c\, y_\alpha^c; u_\delta^c \leftarrow y_\alpha^c)\mathcal{D}_i^1, \mathcal{F}_i'')$. By Figure 4, we know that

$$(\mathcal{C}_1^1 \mathcal{C}_1^2, \mathcal{T}_1 \mathbf{msg}(\mathbf{send}x_\beta^c\, y_\alpha^c; u_\delta^c \leftarrow y_\alpha^c)\mathcal{D}_1^1, \mathcal{F}_1'') \xrightarrow[\Delta', \Delta'', y_\alpha:(A \multimap B)[c] \Vdash K]{y_\alpha^c.\mathbf{msg}(\mathbf{send}x_\beta^c\, y_\alpha^c; u_\delta^c \leftarrow y_\alpha^c)} (\mathcal{C}_1^2, \mathcal{T}_1, \mathcal{C}_1^1 \mathbf{msg}(\mathbf{send}x_\beta^c\, y_\alpha^c; u_\delta^c \leftarrow y_\alpha^c)\mathcal{D}_1^1 \mathcal{F}_1''),$$

and

$$(\mathcal{C}_1^1 \mathcal{C}_1^2, \mathcal{T}_1 \mathbf{msg}(\mathbf{send}x_\beta^c\, y_\alpha^c; u_\delta^c \leftarrow y_\alpha^c)\mathcal{D}_1^1, \mathcal{F}_1'') \xrightarrow[\Delta', \Delta'', y_\alpha:(A \multimap B)[c] \Vdash K]{y_\alpha^c.\mathbf{msg}(\mathbf{send}x_\beta^c\, y_\alpha^c; u_\delta^c \leftarrow y_\alpha^c)} (\mathcal{C}_1^1 \mathcal{C}_1^2 \mathcal{T}_1 \mathbf{msg}(\mathbf{send}x_\beta^c\, y_\alpha^c; u_\delta^c \leftarrow y_\alpha^c), \mathcal{D}_1^1, \mathcal{F}_1'').$$

Thus either (a) $(\mathcal{C}_1', \mathcal{D}_1', \mathcal{F}_1') = (\mathcal{C}_1^2, \mathcal{T}_1, \mathcal{C}_1^1 \mathbf{msg}(\mathbf{send}x_\beta^c\, y_\alpha^c; u_\delta^c \leftarrow y_\alpha^c)\mathcal{D}_1^1 \mathcal{F}_1'')$ or
(b) $(\mathcal{C}_1', \mathcal{D}_1', \mathcal{F}_1') = (\mathcal{C}_1^1 \mathcal{C}_1^2 \mathcal{T}_1 \mathbf{msg}(\mathbf{send}x_\beta^c\, y_\alpha^c; u_\delta^c \leftarrow y_\alpha^c), \mathcal{D}_1^1, \mathcal{F}_1'')$.

We proceed by cases:

(a) $(\mathcal{C}_1', \mathcal{D}_1', \mathcal{F}_1') = (\mathcal{C}_1^2, \mathcal{T}_1, \mathcal{C}_1^1 \mathbf{msg}(\mathbf{send}x_\beta^c\, y_\alpha^c; u_\delta^c \leftarrow y_\alpha^c)\mathcal{D}_1^1 \mathcal{F}_1'')$.

By Figure 4, we can step

$$(\mathcal{C}_2'', \mathcal{D}_2'', \mathcal{F}_2'') \xrightarrow[\Delta', \Delta'', y_\alpha:(A \multimap B)[c] \Vdash K]{y_\alpha^c.\mathbf{msg}(\mathbf{send}x_\beta^c\, y_\alpha^c; u_\delta^c \leftarrow y_\alpha^c)} (\mathcal{C}_2^2, \mathcal{T}_2, \mathcal{C}_2^1 \mathbf{msg}(\mathbf{send}x_\beta^c\, y_\alpha^c; u_\delta^c \leftarrow y_\alpha^c)\mathcal{D}_2^1 \mathcal{F}_2'').$$

By the typing rules, we know that

$\circ_6 (\mathcal{C}_1^2 \mathcal{T}_1 \mathcal{C}_1^1 \mathbf{msg}(\mathbf{send}x_\beta^c\, y_\alpha^c; u_\delta^c \leftarrow y_\alpha^c)\mathcal{D}_1^1 \mathcal{F}_1'', \mathcal{C}_2^2 \mathcal{T}_2 \mathcal{C}_2^1 \mathbf{msg}(\mathbf{send}x_\beta^c\, y_\alpha^c; u_\delta^c \leftarrow y_\alpha^c)\mathcal{D}_2^1 \mathcal{F}_2'') \in \mathsf{Tree}(\Delta' \Vdash x_\beta{:}A[c])$.

We know that $\mathcal{A}^{\kappa''}$ is a tree in the configuration $\mathcal{C}_i^1$, i.e. $\mathcal{C}_i^1 = \mathcal{A}_i^{\kappa''}, \mathcal{C}_i^{1'}$ and $\mathcal{C}_i^{\kappa''} = \mathcal{C}_i^{1'} \mathcal{C}_i^2$. By $\star_1^\kappa$, we get

$\circ_7 \forall m_1. (\mathcal{C}_1^2 \mathcal{T}_1 \mathcal{A}_1^{\kappa''} \mathbf{msg}(\mathbf{send}x_\beta^c\, y_\alpha^c; y_\alpha^c \leftarrow u_\delta^c)\mathcal{C}_1^{1'} \mathcal{D}_1^1 \mathcal{F}_1''; \mathcal{C}_2^2 \mathcal{T}_2 \mathcal{A}_2^{\kappa''} \mathbf{msg}(\mathbf{send}x_\beta^c\, y_\alpha^c; y_\alpha^c \leftarrow u_\delta^c)\mathcal{C}_2^{1'} \mathcal{D}_2^2 \mathcal{F}_2'') \in \mathcal{E}_\Psi^\xi [\![x_\beta{:}A[c] \Vdash u_\delta{:}B[c]]\!]^{m_1}$

by row 5 of the logical relation.

By Lemma C.16, from $\circ_7$ and $\circ_5$, $\star_3$, and $\star_4^\ell$ for the trees in $\mathcal{C}_i^{1'}$, we get

$\circ_7' \forall m_3. (\mathcal{C}_1^2 \mathcal{T}_1 \mathcal{C}_1^1 \mathcal{D}_1^1 \mathbf{msg}(\mathbf{send}x_\beta^c\, y_\alpha^c; y_\alpha^c \leftarrow u_\delta^c)\mathcal{F}_1'' \cdot; \mathcal{C}_2^2 \mathcal{T}_2 \mathcal{C}_2^1 \mathcal{D}_2^2 \mathbf{msg}(\mathbf{send}x_\beta^c\, y_\alpha^c; y_\alpha^c \leftarrow u_\delta^c)\mathcal{F}_2'' \cdot) \in \mathcal{E}_\Psi^\xi [\![x_\beta{:}A[c] \Vdash \cdot]\!]^{m_3}$.

We can apply the induction hypothesis, on $\circ_4$ and $\circ_6$ and $\circ_7'$ and $\star_4^\ell$ for the trees in $\mathcal{C}_i^2$ to get

$$\mathcal{C}_1^2, \mathcal{T}_1, \mathcal{C}_1^1 \mathbf{msg}(\mathbf{send}x_\beta^c\, y_\alpha^c; u_\delta^c \leftarrow y_\alpha^c)\mathcal{D}_1^1 \mathcal{F}_1'' \sim^k \mathcal{C}_2^2, \mathcal{T}_2, \mathcal{C}_2^1 \mathbf{msg}(\mathbf{send}x_\beta^c\, y_\alpha^c; u_\delta^c \leftarrow y_\alpha^c)\mathcal{D}_2^1 \mathcal{F}_2''.$$

which completes the proof of this case.

(b) $(\mathcal{C}_1', \mathcal{D}_1', \mathcal{F}_1') = (\mathcal{C}_1^1 \mathcal{C}_1^2 \mathcal{T}_1 \mathbf{msg}(\mathbf{send}x_\beta^c\, y_\alpha^c; u_\delta^c \leftarrow y_\alpha^c), \mathcal{D}_1^1, \mathcal{F}_1'')$.

By Figure 4, we can step

$$(\mathcal{C}_2'', \mathcal{D}_2'', \mathcal{F}_2'') \xrightarrow[\Delta', \Delta'', y_\alpha:(A \multimap B)[c] \Vdash K]{y_\alpha^c.\mathbf{msg}(\mathbf{send}x_\beta^c\, y_\alpha^c; u_\delta^c \leftarrow y_\alpha^c)} (\mathcal{C}_2^1 \mathcal{C}_2^2 \mathcal{T}_2 \mathbf{msg}(\mathbf{send}x_\beta^c\, y_\alpha^c; u_\delta^c \leftarrow y_\alpha^c), \mathcal{D}_2^1, \mathcal{F}_2'').$$

By the typing rules, we know that

$\circ_8 (\mathcal{C}_1^1 \mathcal{C}_1^2 \mathcal{T}_1 \mathbf{msg}(\mathbf{send}x_\beta^c\, y_\alpha^c; u_\delta^c \leftarrow y_\alpha^c), \mathcal{D}_1^1, \mathcal{F}_1''), (\mathcal{C}_2^1 \mathcal{C}_2^2 \mathcal{T}_2 \mathbf{msg}(\mathbf{send}x_\beta^c\, y_\alpha^c; u_\delta^c \leftarrow y_\alpha^c), \mathcal{D}_2^1, \mathcal{F}_2'') \in \mathsf{Tree}(\Delta'', u_\delta{:}B[c] \Vdash K)$.

We know that $\mathcal{A}^{\kappa''}$ is a tree in the configuration $\mathcal{C}_i^1$, i.e. $\mathcal{C}_i^1 = \mathcal{A}_i^{\kappa''}, \mathcal{C}_i^{1'}$ and $\mathcal{C}_i^{\kappa''} = \mathcal{C}_i^{1'} \mathcal{C}_i^2$. By $\star_1^\kappa$, we get

$\circ_9 \forall m_1. (\mathcal{C}_1^2 \mathcal{T}_1 \mathcal{A}_1^{\kappa''} \mathbf{msg}(\mathbf{send}x_\beta^c\, y_\alpha^c; y_\alpha^c \leftarrow u_\delta^c)\mathcal{C}_1^{1'} \mathcal{D}_1^1 \mathcal{F}_1''; \mathcal{C}_2^2 \mathcal{T}_2 \mathcal{A}_2^{\kappa''} \mathbf{msg}(\mathbf{send}x_\beta^c\, y_\alpha^c; y_\alpha^c \leftarrow u_\delta^c)\mathcal{C}_2^{1'} \mathcal{D}_2^2 \mathcal{F}_2'') \in \mathcal{E}_\Psi^\xi [\![x_\beta{:}A[c] \Vdash u_\delta{:}B[c]]\!]^{m_1}$

by row 5 of the logical relation.

By Lemma C.16, from $\circ_9$ and $\circ_4$, and $\star_4^\ell$ for the trees in $\mathcal{C}_i^2$, we get

$\circ_9' \forall m_3. (\cdot \mathcal{C}_1^2 \mathcal{T}_1 \mathcal{A}_1^{\kappa''} \mathbf{msg}(\mathbf{send}x_\beta^c\, y_\alpha^c; y_\alpha^c \leftarrow u_\delta^c)\mathcal{C}_1^{1'} \mathcal{D}_1^1 \mathcal{F}_1''; \mathcal{C}_2^2 \mathcal{T}_2 \mathcal{A}_2^{\kappa''} \mathbf{msg}(\mathbf{send}x_\beta^c\, y_\alpha^c; y_\alpha^c \leftarrow u_\delta^c)\mathcal{C}_2^{1'} \mathcal{D}_2^2 \mathcal{F}_2'') \in \mathcal{E}_\Psi^\xi [\![\cdot \Vdash u_\delta{:}B[c]]\!]^{m_3}$.

We can apply the induction hypothesis, on $\circ_5$ and $\circ_8$ and $\circ_9'$ and $\star_3$ and $\star_4^\ell$ for the trees in $\mathcal{C}_i^{1'}$ to get $\circ_5'$ to get

$$\mathcal{C}_1^1 \mathcal{C}_1^2 \mathcal{T}_1 \mathbf{msg}(\mathbf{send}x_\beta^c\, y_\alpha^c; u_\delta^c \leftarrow y_\alpha^c), \mathcal{D}_1^1, \mathcal{F}_1'' \sim^k \mathcal{C}_2^1 \mathcal{C}_2^2 \mathcal{T}_2 \mathbf{msg}(\mathbf{send}x_\beta^c\, y_\alpha^c; u_\delta^c \leftarrow y_\alpha^c), \mathcal{D}_2^1, \mathcal{F}_2''.$$

which completes the proof of this case.

**Case 3.** $(\mathcal{C}_1, \mathcal{D}_1, \mathcal{F}_1) \xrightarrow[\Delta \Vdash K]{y_c^c \cdot \bar{q}} (\mathcal{C}_1', \mathcal{D}_1', \mathcal{F}_1')$, and $y_c^\alpha \in K$.

By definition $\star (\mathcal{C}_1, \mathcal{D}_1, \mathcal{F}_1) \mapsto_{\Delta \Vdash K}^{*\Theta, y_c^\alpha; \Upsilon} (\mathcal{C}_1'', \mathcal{D}_1'', \mathcal{F}_1'')$, and $\dagger (\mathcal{C}_1'', \mathcal{D}_1'', \mathcal{F}_1'') \xrightarrow[\Delta \Vdash K]{y_c^\alpha \cdot \bar{q}} (\mathcal{C}_1', \mathcal{D}_1', \mathcal{F}_1')$.

We apply $\star$ on the condition of line 13 in the logical relation in $\dagger_2$ for any $m_2 \geq 1$ to eliminate the if. By $\forall E$ and $\exists E$ and Lemma C.14, we know that there exists a minimal $\mathcal{D}_2''$ such that

$\circ_1 \mathcal{C}_2, \mathcal{D}_2, \mathcal{F}_2 \mapsto_{\Delta \Vdash K}^{*\Theta'; \Upsilon_1} \mathcal{C}_2, \mathcal{D}_2', \mathcal{F}_2$ and

$\dagger_2' \, \forall \mathcal{C}_1', \mathcal{F}_1', \mathcal{C}_2', \mathcal{F}_2'. \text{if } \mathcal{C}_1, \mathcal{D}_1'', \mathcal{F}_1 \mapsto_{\Delta \Vdash K}^{*\Theta_1, \Upsilon_1} \mathcal{C}_1', \mathcal{D}_1'', \mathcal{F}_1' \text{ and } \mathcal{C}_2, \mathcal{D}_2'', \mathcal{F}_2 \mapsto_{\Delta \Vdash K}^{*\Theta_2; \Upsilon_1} \mathcal{C}_2', \mathcal{D}_2'', \mathcal{F}_2', \text{ then}$
$\forall x \in \Upsilon_1. \forall m_2. (\mathcal{C}_1', \mathcal{D}_1'', \mathcal{F}_1'; \mathcal{C}_2', \mathcal{D}_2'', \mathcal{F}_2') \in \mathcal{V}_\Psi^\xi [\![ \Delta \Vdash K ]\!]_{\cdot; x}^{m_2}$ and
$\forall x \in (\Theta_2 \cap \Theta_3). \forall m_2. (\mathcal{C}_1', \mathcal{D}_1'', \mathcal{F}_1'; \mathcal{C}_2', \mathcal{D}_2'', \mathcal{F}_2') \in \mathcal{V}_\Psi^\xi [\![ \Delta \Vdash K ]\!]_{x; \cdot}^{m_2}.$

By $\star$, we know that $(\mathcal{C}_1 \mathcal{D}_1, \mathcal{F}_1, \cdot) \mapsto_{K \Vdash \cdot}^{*\Theta_i, y_c^\alpha; \Upsilon_2} (\mathcal{C}_1 \mathcal{D}_1, \mathcal{F}_1'', \cdot)$. We apply it on the condition of line 13 in the logical relation in $\dagger_3$. By $\forall E$ and $\exists E$ and Lemma C.14, we know that there exists a minimal $\mathcal{F}_2''$ such that $\mathcal{C}_2 \mathcal{D}_2, \mathcal{F}_2, \cdot \mapsto_{K \Vdash \cdot}^{*\Theta'; \Upsilon_1} \mathcal{C}_2 \mathcal{D}_2, \mathcal{F}_2'', \cdot$ and

$\dagger_3' \, \forall \mathcal{C}_1^1, \mathcal{C}_2^1. \text{if } \mathcal{C}_1 \mathcal{D}_1, \mathcal{F}_1'', \cdot \mapsto_{K \Vdash \cdot}^{*\Theta_1, \Upsilon_1} \mathcal{C}_1^1, \mathcal{F}_1'', \cdot \text{ and } \mathcal{C}_2 \mathcal{D}_2, \mathcal{F}_2'', \cdot \mapsto_{K \Vdash \cdot}^{*\Theta_2; \Upsilon_1} \mathcal{C}_2^1, \mathcal{F}_2'', \cdot, \text{ then}$
$\forall x \in \Upsilon_1. \forall m_3. (\mathcal{C}_1^1, \mathcal{F}_1'', \cdot; \mathcal{C}_2^1, \mathcal{F}_2'', \cdot) \in \mathcal{V}_\Psi^\xi [\![ K \Vdash \cdot ]\!]_{\cdot; x}^{m_3}$ and
$\forall x \in (\Theta_2 \cap \Theta_3). \forall m_3. (\mathcal{C}_1^1, \mathcal{F}_1'', \cdot; \mathcal{C}_2^1, \mathcal{F}_2'', \cdot) \in \mathcal{V}_\Psi^\xi [\![ K \Vdash \cdot ]\!]_{x; \cdot}^{m_3}.$

Pick this $\mathcal{F}_2''$.

By $\star$, we know that for all trees $\mathcal{A}_1^\ell$ in $\mathcal{C}_1$ we have $(\cdot, \mathcal{A}_1^\ell, \mathcal{C}_1^\ell \mathcal{D}_1 \mathcal{F}_1) \mapsto_{\cdot \Vdash d_\ell : T_\ell}^{*\Theta_i, y_c^\alpha; \Upsilon_2} (\cdot, \mathcal{A}_1^{\ell''}, \mathcal{C}_1^\ell \mathcal{D}_1 \mathcal{F}_1)$. We apply it on the condition of line 13 in the logical relation in $\dagger_1^\ell$. By $\forall E$ and $\exists E$ and Lemma C.14, we know that there exists a minimal $\mathcal{A}_2^{\ell''}$ such that $\cdot, \mathcal{A}_2^\ell, \mathcal{C}_2^\ell \mathcal{D}_2 \mathcal{F}_2 \mapsto_{K \Vdash \cdot}^{*\Theta'; \Upsilon_1} \cdot, \mathcal{A}_2^{\ell''}, \mathcal{C}_2^\ell \mathcal{D}_2 \mathcal{F}_2$ and

$\dagger_1^{\ell'} \, \forall \mathcal{F}_1^1, \mathcal{F}_2^1. \text{if } \cdot, \mathcal{A}_1^{\ell''}, \mathcal{C}_1^\ell \mathcal{D}_1 \mathcal{F}_1 \mapsto_{\cdot \Vdash d_\ell : T_\ell}^{*\Theta_1, \Upsilon_1} \cdot, \mathcal{A}_1^{\ell''}, \mathcal{F}_1^1 \text{ and } \cdot, \mathcal{A}_2^{\ell''}, \mathcal{C}_2^\ell \mathcal{D}_2 \mathcal{F}_2 \mapsto_{\cdot \Vdash d_\ell : T_\ell}^{*\Theta_1, \Upsilon_1} \cdot, \mathcal{A}_2^{\ell''}, \mathcal{F}_2^1, \text{ then}$
$\forall x \in \Upsilon_1. \forall m_1. (\cdot, \mathcal{A}_1^{\ell''}, \mathcal{F}_1^1; \cdot, \mathcal{A}_2^{\ell''}, \mathcal{F}_2^1) \in \mathcal{V}_\Psi^\xi [\![ \cdot \Vdash d_\ell : T_\ell ]\!]_{\cdot; x}^{m_1}$ and
$\forall x \in (\Theta_2 \cap \Theta_3). \forall m_1. (\cdot, \mathcal{A}_1^{\ell''}, \mathcal{F}_1^1; \cdot, \mathcal{A}_2^{\ell''}, \mathcal{F}_2^1 \in \mathcal{V}_\Psi^\xi [\![ \cdot \Vdash d_\ell : T_\ell ]\!]_{x; \cdot}^{m_1}.$

Build $\mathcal{C}_2''$ as a forest consisting of the trees $\mathcal{A}_2^{\ell''}$, i.e. $\mathcal{C}_2'' = \{\mathcal{A}_2^{\ell''}\}_{\ell \in I}$ and put $\mathcal{C}_2^{\ell''}$ to be all trees in $\mathcal{C}_2''$ other than $\mathcal{A}_2^{\ell''}$. We instantiate the universal quantifiers in $\dagger_2'$ with $\mathcal{C}_1'', \mathcal{F}_1'', \mathcal{C}_2''$, and $\mathcal{F}_2''$ to get

$\dagger_1'' \, \forall x \in \Upsilon_1. \forall m_2. (\mathcal{C}_1'', \mathcal{D}_1'', \mathcal{F}_1''; \mathcal{C}_2'', \mathcal{D}_2'', \mathcal{F}_2'') \in \mathcal{V}_\Psi^\xi [\![ \Delta \Vdash K ]\!]_{\cdot; x}^{m_2}$ and
$\forall x \in (\Theta_2 \cap \Theta_3). \forall m_2. (\mathcal{C}_1'', \mathcal{D}_1'', \mathcal{F}_1''; \mathcal{C}_2'', \mathcal{D}_2'', \mathcal{F}_2'') \in \mathcal{V}_\Psi^\xi [\![ \Delta \Vdash K ]\!]_{x; \cdot}^{m_2}.$

Similarly, from $\dagger_3'$, we get

$\dagger_3'' \, \forall x \in \Upsilon_1. \forall m_3. (\mathcal{C}_1'' \mathcal{D}_1'', \mathcal{F}_1'', \cdot; \mathcal{C}_2'' \mathcal{D}_2'', \mathcal{F}_2'', \cdot) \in \mathcal{V}_\Psi^\xi [\![ K \Vdash \cdot ]\!]_{\cdot; x}^{m_3}$ and
$\forall x \in (\Theta_2 \cap \Theta_3). \forall m_3. (\mathcal{C}_1'' \mathcal{D}_1'', \mathcal{F}_1'', \cdot; \mathcal{C}_2'' \mathcal{D}_2'', \mathcal{F}_2'', \cdot) \in \mathcal{V}_\Psi^\xi [\![ K \Vdash \cdot ]\!]_{x; \cdot}^{m_3}.$

In particular, we have

$\star_2 \forall m_2. (\mathcal{C}_1'', \mathcal{D}_1'', \mathcal{F}_1''; \mathcal{C}_2'', \mathcal{D}_2'', \mathcal{F}_2'') \in \mathcal{V}_\Psi^\xi [\![ \Delta \Vdash K ]\!]_{y_c^\alpha; \cdot}^{m_2},$ and

$\star_3 \forall m_3. (\mathcal{C}_1'' \mathcal{D}_1'', \mathcal{F}_1'', \cdot; \mathcal{C}_1'' \mathcal{D}_1'', \mathcal{F}_2'', \cdot) \in \mathcal{V}_\Psi^\xi [\![ K \Vdash \cdot ]\!]_{\cdot; y_c^\alpha}^{m_3}.$

Moreover by the forward closure lemma on $\dagger_3'$, we get

$\star_3 \forall m_3. (\mathcal{C}_1'' \mathcal{D}_1'' \mathcal{F}_1''; \mathcal{C}_2'' \mathcal{D}_2'' \mathcal{F}_2'' \cdot) \in \mathcal{E}_\Psi^\xi [\![ K \Vdash \cdot ]\!]^{m_3}.$

And by the forward closure lemma on $\dagger_1^i$, we get

$\star_4^\ell \forall m_1. (\cdot \mathcal{A}_1^{\ell''} \mathcal{C}_1^{\ell''} \mathcal{D}_1'' \mathcal{F}_1''; \cdot \mathcal{A}_2^{\ell''} \mathcal{C}_2^{\ell''} \mathcal{D}_2'' \mathcal{F}_2'') \in \mathcal{E}_\Psi^\xi [\![ \cdot \Vdash d_j : T_j ]\!]^{m_1}.$

Now we consider cases based on the row in the definition of logical relation (Figure 2) with which we can prove $\star_3$. It is a unique row determined based on the type of $y_c^\alpha$ and whether the process sending over $y_c^\alpha$ is a message or a forwarding process.

**Row (8).** The conditions validating this row are:

$\circ_1 (\mathcal{C}_1'' \mathcal{D}_1'', \mathcal{F}_1'', \cdot; \mathcal{C}_2'' \mathcal{D}_2'', \mathcal{F}_2'', \cdot) \in \mathsf{Tree}_\Psi(y_\alpha : \& \{\ell : A_\ell\}_{\ell \in I}[c] \Vdash \cdot)$

$\circ_2 \mathcal{F}_1'' = \mathbf{msg}(y_\alpha^c.k; u_\delta^c \leftarrow y_\alpha^c)\mathcal{F}_1^1$

$$\circ_3 \mathcal{F}_2'' = \mathbf{msg}(y_\alpha^c.k; u_\delta^c \leftarrow y_\alpha^c)\mathcal{F}_2^1$$

$$\circ_4 \forall m_3 \, (\mathcal{C}_1''\mathcal{D}_1''\mathbf{msg}(y_\alpha^c.k; u_\delta^c \leftarrow y_\alpha^c)\mathcal{F}_1^1\cdot, \mathcal{C}_2''\mathcal{D}_2''\mathbf{msg}(y_\alpha^c.k; u_\delta^c \leftarrow y_\alpha^c)\mathcal{F}_2^1\cdot) \in \mathcal{E}_\Psi^\xi [\![u_\delta{:}A_k[c] \Vdash \cdot]\!]^{m_3}$$

Thus $(\mathcal{C}_i'', \mathcal{D}_i'', \mathcal{F}_i'') = (\mathcal{C}_i'', \mathcal{D}_i'', \mathbf{msg}(y_\alpha^c.k; u_\delta^c \leftarrow y_\alpha^c)\mathcal{F}_i^1)$. By Figure 4, we know that

$$(\mathcal{C}_1'', \mathcal{D}_1'', \mathbf{msg}(y_\alpha^c.k; u_\delta^c \leftarrow y_\alpha^c)\mathcal{F}_1^1) \xrightarrow[\Delta \Vdash y_\alpha: \&\{\ell: A_\ell\}_{\ell \in I}[c]]{y_\alpha^c \cdot \overline{\mathbf{msg}(y_\alpha^c.k; u_\delta^c \leftarrow y_\alpha^c)}} (\mathcal{C}_1'', \mathcal{D}_1''\mathbf{msg}(y_\alpha^c.k; u_\delta^c \leftarrow y_\alpha^c), \mathcal{F}_1^1),$$

and thus $\mathcal{C}_1', \mathcal{D}_1', \mathcal{F}_1' = (\mathcal{C}_1'', \mathcal{D}_1''\mathbf{msg}(y_\alpha^c.k; u_\delta^c \leftarrow y_\alpha^c), \mathcal{F}_1^1)$.
By Figure 4, we can step

$$(\mathcal{C}_2'', \mathcal{D}_2'', \mathcal{F}_2'') \xrightarrow[\Delta \Vdash y_\alpha: \&\{\ell: A_\ell\}_{\ell \in I}[c]]{y_\alpha^c \cdot \overline{\mathbf{msg}(y_\alpha^c.k; y_\alpha^c \leftarrow u_\delta^c)}} (\mathcal{C}_2'', \mathcal{D}_2''\mathbf{msg}(y_\alpha^c.k; u_\delta^c \leftarrow y_\alpha^c), \mathcal{F}_2^1).$$

By the typing rules, we know that

$$\circ_5 \, (\mathcal{C}_1'', \mathcal{D}_1''\mathbf{msg}(y_\alpha^c.k; u_\delta^c \leftarrow y_\alpha^c), \mathcal{F}_1^1), \mathcal{C}_2'', \mathcal{D}_2''\mathbf{msg}(y_\alpha^c.k; u_\delta^c \leftarrow y_\alpha^c), \mathcal{F}_2^1) \in \mathsf{Tree}(\Delta \Vdash u_\delta{:}A_k[c]).$$

And by row 3 and $\star_2$ we get

$$\circ_6 \forall m_2. (\mathcal{C}_1''\mathcal{D}_1''\mathbf{msg}(y_\alpha^c.k; y_\alpha^c \leftarrow u_\delta^c)\mathcal{F}_1^1; \mathcal{C}_2''\mathcal{D}_2''\mathbf{msg}(y_\alpha^c.k; y_\alpha^c \leftarrow u_\delta^c)\mathcal{F}_2^1) \in \mathcal{E}_\Psi^\xi [\![\Delta \Vdash y_\alpha{:}A_k[c]]\!]^{m_2}.$$

We can apply the induction hypothesis on $\circ_4$, $\circ_5$, $\circ_6$, and $\star_4^\ell$ to get

$$\mathcal{C}_1'', \mathcal{D}_1''\mathbf{msg}(y_\alpha^c.k; u_\delta^c \leftarrow y_\alpha^c), \mathcal{F}_1^1 \sim^k \mathcal{C}_2'', \mathcal{D}_2''\mathbf{msg}(y_\alpha^c.k; u_\delta^c \leftarrow y_\alpha^c), \mathcal{F}_2^1.$$

which completes the proof.

**Row (10).** The conditions validating this row are:

$$\circ_1 \, (\mathcal{C}_1''\mathcal{D}_1'', \mathcal{F}_1'', \cdot; \mathcal{C}_2''\mathcal{D}_2'', \mathcal{F}_2'', \cdot) \in \mathsf{Tree}_\Psi(y_\alpha{:}(A \multimap B)[c] \Vdash \cdot)$$
$$\circ_2 \, \mathcal{F}_1'' = \mathcal{T}_1\mathbf{msg}(\mathbf{send}\, x_\beta^c \, y_\alpha^c; u_\delta^c \leftarrow y_\alpha^c)\mathcal{F}_1^1 \text{ for } \mathcal{T}_1 \in \mathsf{Tree}(\cdot \Vdash x_\beta{:}A[c])$$
$$\circ_3 \, \mathcal{F}_2'' = \mathcal{T}_2\mathbf{msg}(\mathbf{send}\, x_\beta^c \, y_\alpha^c; u_\delta^c \leftarrow y_\alpha^c)\mathcal{F}_2^1 \text{ for } \mathcal{T}_2 \in \mathsf{Tree}(\cdot \Vdash x_\beta{:}A[c])$$

$$\circ_4 \, \forall m_3. \, (\cdot \mathcal{T}_1\mathcal{C}_1''\mathcal{D}_1''\mathbf{msg}(\mathbf{send}\, x_\beta^c \, y_\alpha^c; u_\delta^c \leftarrow y_\alpha^c)\mathcal{F}_1^1, \cdot \mathcal{T}_2\mathcal{C}_2''\mathcal{D}_2''\mathbf{msg}(\mathbf{send}\, x_\beta^c \, y_\alpha^c; u_\delta^c \leftarrow y_\alpha^c)\mathcal{F}_2^1) \in \mathcal{E}_\Psi^\xi [\![\cdot \Vdash x_\beta{:}A[c]]\!]^{m_3}$$
$$\circ_5 \, \forall m_3. \, (\mathcal{C}_1''\mathcal{T}_1\mathcal{D}_1''\mathbf{msg}(\mathbf{send}\, x_\beta^c \, y_\alpha^c; u_\delta^c \leftarrow y_\alpha^c)\mathcal{F}_1^1\cdot, \mathcal{C}_2''\mathcal{T}_2\mathcal{D}_2''\mathbf{msg}(\mathbf{send}\, x_\beta^c \, y_\alpha^c; u_\delta^c \leftarrow y_\alpha^c)\mathcal{F}_2^1\cdot) \in \mathcal{E}_\Psi^\xi [\![u_\delta{:}B[c] \Vdash \cdot]\!]^{m_3}$$

We get

$$(\mathcal{C}_i'', \mathcal{D}_i'', \mathcal{F}_i'') = (\mathcal{C}_i'', \mathcal{D}_i'', \mathcal{T}_2\mathbf{msg}(\mathbf{send}\, x_\beta^c \, y_\alpha^c; u_\delta^c \leftarrow y_\alpha^c)\mathcal{F}_i^1).$$

By Figure 4, we know that

$$(\mathcal{C}_1'', \mathcal{D}_1'', \mathcal{T}_1\mathbf{msg}(\mathbf{send}\, x_\beta^c \, y_\alpha^c; u_\delta^c \leftarrow y_\alpha^c)\mathcal{F}_1^1) \xrightarrow[\Delta \Vdash y_\alpha{:}(A \multimap B)[c]]{y_\alpha^c \cdot \overline{\mathbf{msg}(\mathbf{send}\, x_\beta^c \, y_\alpha^c; u_\delta^c \leftarrow y_\alpha^c)}} (\mathcal{C}_1''\mathcal{T}_1, \mathcal{D}_1''\mathbf{msg}(\mathbf{send}\, x_\beta^c \, y_\alpha^c; u_\delta^c \leftarrow y_\alpha^c), \mathcal{F}_1^1),$$

and thus $\mathcal{C}_1', \mathcal{D}_1', \mathcal{F}_1' = (\mathcal{C}_1''\mathcal{T}_1, \mathcal{D}_1''\mathbf{msg}(\mathbf{send}\, x_\beta^c \, y_\alpha^c; u_\delta^c \leftarrow y_\alpha^c), \mathcal{F}_1^1)$.
By Figure 4, we can step

$$(\mathcal{C}_2'', \mathcal{D}_2'', \mathcal{F}_2'') \xrightarrow[\Delta \Vdash y_\alpha{:}(A \multimap B)[c]]{y_\alpha^c \cdot \overline{\mathbf{msg}(\mathbf{send}\, x_\beta^c \, y_\alpha^c; u_\delta^c \leftarrow y_\alpha^c)}} (\mathcal{C}_2''\mathcal{T}_2, \mathcal{D}_2''\mathbf{msg}(\mathbf{send}\, x_\beta^c \, y_\alpha^c; u_\delta^c \leftarrow y_\alpha^c), \mathcal{F}_2^1).$$

By the typing rules, we know that

$$\circ_6 \, (\mathcal{C}_1''\mathcal{T}_1, \mathcal{D}_1''\mathbf{msg}(\mathbf{send}\, x_\beta^c \, y_\alpha^c; u_\delta^c \leftarrow y_\alpha^c), \mathcal{F}_1^1, \mathcal{C}_2''\mathcal{T}_2, \mathcal{D}_2''\mathbf{msg}(\mathbf{send}\, x_\beta^c \, y_\alpha^c; u_\delta^c \leftarrow y_\alpha^c), \mathcal{F}_2^1) \in \mathsf{Tree}(\Delta, x_\beta{:}A[c] \Vdash u_\delta{:}B[c]).$$

And by row 5 and $\star_2$ we get

$$\circ_7 \forall m_2. (\mathcal{C}_1''\mathcal{T}_1\mathcal{D}_1''\mathbf{msg}(\mathbf{send}\, x_\beta^c \, y_\alpha^c; u_\delta^c \leftarrow y_\alpha^c)\mathcal{F}_1^1, \mathcal{C}_2''\mathcal{T}_2\mathcal{D}_2''\mathbf{msg}(\mathbf{send}\, x_\beta^c \, y_\alpha^c; u_\delta^c \leftarrow y_\alpha^c)\mathcal{F}_2^1) \in \mathcal{E}_\Psi^\xi [\![\Delta, x_\beta{:}A[c] \Vdash u_\delta{:}B[c]]\!]^{m_2}$$

We can apply the induction hypothesis, on $\circ_4$ and $\circ_5$ and $\circ_6$ and $\circ_7$ and $\star_4^\ell$ to get

$$\mathcal{C}_1''\mathcal{T}_1, \mathcal{D}_1''\mathbf{msg}(\mathbf{send}\, x_\beta^c \, y_\alpha^c; u_\delta^c \leftarrow y_\alpha^c), \mathcal{F}_1^1 \sim^k \mathcal{C}_2''\mathcal{T}_2, \mathcal{D}_2''\mathbf{msg}(\mathbf{send}\, x_\beta^c \, y_\alpha^c; u_\delta^c \leftarrow y_\alpha^c), \mathcal{F}_2^1.$$

which completes the proof.

**Case 4.** $(\mathcal{C}_1, \mathcal{D}_1, \mathcal{F}_1) \xrightarrow{y_c^c \cdot \overline{q}}_{\Delta \Vdash K} (\mathcal{C}_1', \mathcal{D}_1', \mathcal{F}_1')$, and $y_c^\alpha \in \Delta$, i.e. $\Delta = \Delta', y_c^\alpha{:}T$.

By definition $\star$ $(\mathcal{C}_1, \mathcal{D}_1, \mathcal{F}_1) \mapsto^{*\Theta, y_c^\alpha; \Upsilon}_{\Delta \Vdash K} (\mathcal{C}_1'', \mathcal{D}_1'', \mathcal{F}_1'')$, and $\dagger$ $(\mathcal{C}_1'', \mathcal{D}_1'', \mathcal{F}_1'') \xrightarrow{y_c^\alpha \cdot \overline{q}}_{\Delta \Vdash K} (\mathcal{C}_1', \mathcal{D}_1', \mathcal{F}_1')$.

We apply $\star$ on the condition of line 13 in the logical relation in $\dagger_2$ for any $m_2 \geq 1$ to eliminate the if. By $\forall E$ and $\exists E$ and Lemma C.14, we know that there exists a minimal $\mathcal{D}_2''$ such that
$\circ_1$ $\mathcal{C}_2, \mathcal{D}_2, \mathcal{F}_2 \mapsto^{*\Theta'; \Upsilon_1}_{\Delta \Vdash K} \mathcal{C}_2, \mathcal{D}_2', \mathcal{F}_2$ and

$$\dagger_2' \; \forall \mathcal{C}_1', \mathcal{F}_1', \mathcal{C}_2', \mathcal{F}_2'. \text{if } \mathcal{C}_1, \mathcal{D}_1'', \mathcal{F}_1 \mapsto^{*\Theta_1, \Upsilon_1}_{\Delta \Vdash K} \mathcal{C}_1', \mathcal{D}_1'', \mathcal{F}_1' \text{ and } \mathcal{C}_2, \mathcal{D}_2'', \mathcal{F}_2 \mapsto^{*\Theta_2; \Upsilon_1}_{\Delta \Vdash K} \mathcal{C}_2', \mathcal{D}_2'', \mathcal{F}_2', \text{ then}$$
$$\forall x \in \Upsilon_1. \forall m_2. (\mathcal{C}_1', \mathcal{D}_1'', \mathcal{F}_1'; \mathcal{C}_2', \mathcal{D}_2'', \mathcal{F}_2') \in \mathcal{V}_\Psi^\xi [\![\Delta \Vdash K]\!]_{;:x}^{m_2} \text{ and}$$
$$\forall x \in (\Theta_2 \cap \Theta_3). \forall m_2. (\mathcal{C}_1', \mathcal{D}_1'', \mathcal{F}_1'; \mathcal{C}_2', \mathcal{D}_2'', \mathcal{F}_2') \in \mathcal{V}_\Psi^\xi [\![\Delta \Vdash K]\!]_{x;:}^{m_2}.$$

By $\star$, we know that $(\mathcal{C}_1 \mathcal{D}_1, \mathcal{F}_1, \cdot) \mapsto^{*\Theta_i, y_c^\alpha; \Upsilon_2}_{K \Vdash \cdot} (\mathcal{C}_1 \mathcal{D}_1, \mathcal{F}_1'', \cdot)$. We apply it on the condition of line 13 in the logical relation in $\dagger_3$. By $\forall E$ and $\exists E$ and Lemma C.14, we know that there exists a minimal $\mathcal{F}_2''$ such that $\mathcal{C}_2 \mathcal{D}_2, \mathcal{F}_2, \cdot \mapsto^{*\Theta'; \Upsilon_1}_{K \Vdash \cdot} \mathcal{C}_2 \mathcal{D}_2, \mathcal{F}_2'', \cdot$ and

$$\dagger_3' \; \forall \mathcal{C}_1^1, \mathcal{C}_2^1. \text{if } \mathcal{C}_1 \mathcal{D}_1, \mathcal{F}_1'', \cdot \mapsto^{*\Theta_1, \Upsilon_1}_{K \Vdash \cdot} \mathcal{C}_1^1, \mathcal{F}_1'', \cdot \text{ and } \mathcal{C}_2 \mathcal{D}_2, \mathcal{F}_2'', \cdot \mapsto^{*\Theta_2; \Upsilon_1}_{K \Vdash \cdot} \mathcal{C}_2^1, \mathcal{F}_2'', \cdot, \text{ then}$$
$$\forall x \in \Upsilon_1. \forall m_3. (\mathcal{C}_1^1, \mathcal{F}_1'', \cdot; \mathcal{C}_2^1, \mathcal{F}_2'', \cdot) \in \mathcal{V}_\Psi^\xi [\![K \Vdash \cdot]\!]_{;:x}^{m_3} \text{ and}$$
$$\forall x \in (\Theta_2 \cap \Theta_3). \forall m_3. (\mathcal{C}_1^1, \mathcal{F}_1'', \cdot; \mathcal{C}_2^1, \mathcal{F}_2'', \cdot) \in \mathcal{V}_\Psi^\xi [\![K \Vdash \cdot]\!]_{x;:}^{m_3}.$$

Pick this $\mathcal{F}_2''$.

By $\star$, we know that for all trees $\mathcal{A}_1^\ell$ in $\mathcal{C}_1$ we have $(\cdot, \mathcal{A}_1^\ell, \mathcal{C}_1^\ell \mathcal{D}_1 \mathcal{F}_1) \mapsto^{*\Theta_i, y_c^\alpha; \Upsilon_2}_{\cdot \Vdash d_\ell : T_\ell} (\cdot, \mathcal{A}_1^{\ell''}, \mathcal{C}_1^\ell \mathcal{D}_1 \mathcal{F}_1)$. We apply it on the condition of line 13 in the logical relation in $\dagger_1^\ell$. By $\forall E$ and $\exists E$ and Lemma C.14, we know that there exists a minimal $\mathcal{A}_2^{\ell''}$ such that $\cdot, \mathcal{A}_2^\ell, \mathcal{C}_2^\ell \mathcal{D}_2 \mathcal{F}_2 \mapsto^{*\Theta'; \Upsilon_1}_{K \Vdash \cdot} \cdot, \mathcal{A}_2^{\ell''}, \mathcal{C}_2^\ell \mathcal{D}_2 \mathcal{F}_2$ and

$$\dagger_1^{\ell'} \; \forall \mathcal{F}_1^1, \mathcal{F}_2^1. \text{if } \cdot, \mathcal{A}_1^{\ell''}, \mathcal{C}_1^\ell \mathcal{D}_1 \mathcal{F}_1 \mapsto^{*\Theta_1, \Upsilon_1}_{\cdot \Vdash d_\ell : T_\ell} \cdot, \mathcal{A}_1^{\ell''}, \mathcal{F}_1^1 \text{ and } \cdot, \mathcal{A}_2^{\ell''}, \mathcal{C}_2^\ell \mathcal{D}_2 \mathcal{F}_2 \mapsto^{*\Theta_1, \Upsilon_1}_{\cdot \Vdash d_\ell : T_\ell} \cdot, \mathcal{A}_2^{\ell''}, \mathcal{F}_2^1, \text{ then}$$
$$\forall x \in \Upsilon_1. \forall m_1. (\cdot, \mathcal{A}_1^{\ell''}, \mathcal{F}_1^1; \cdot, \mathcal{A}_2^{\ell''}, \mathcal{F}_2^1) \in \mathcal{V}_\Psi^\xi [\![\cdot \Vdash d_\ell : T_\ell]\!]_{;:x}^{m_1} \text{ and}$$
$$\forall x \in (\Theta_2 \cap \Theta_3). \forall m_1. (\cdot, \mathcal{A}_1^{\ell''}, \mathcal{F}_1^1; \cdot, \mathcal{A}_2^{\ell''}, \mathcal{F}_2^1 \in \mathcal{V}_\Psi^\xi [\![\cdot \Vdash d_\ell : T_\ell]\!]_{x;:}^{m_1}.$$

Build $\mathcal{C}_2''$ as a forest consisting of the trees $\mathcal{A}_2^{\ell''}$, i.e. $\mathcal{C}''_2 = \{\mathcal{A}_2^{\ell''}\}_{\ell \in I}$ and put $\mathcal{C}_2^{\ell''}$ to be all trees in $\mathcal{C}_2''$ other than $\mathcal{A}_2^{\ell''}$. We instantiate the universal quantifiers in $\dagger_2'$ with $\mathcal{C}_1'', \mathcal{F}_1'', \mathcal{C}_2''$, and $\mathcal{F}_2''$ to get

$$\dagger_1'' \; \forall x \in \Upsilon_1. \forall m_2. (\mathcal{C}_1'', \mathcal{D}_1'', \mathcal{F}_1''; \mathcal{C}_2'', \mathcal{D}_2'', \mathcal{F}_2'') \in \mathcal{V}_\Psi^\xi [\![\Delta \Vdash K]\!]_{;:x}^{m_2} \text{ and}$$
$$\forall x \in (\Theta_2 \cap \Theta_3). \forall m_2. (\mathcal{C}_1'', \mathcal{D}_1'', \mathcal{F}_1''; \mathcal{C}_2'', \mathcal{D}_2'', \mathcal{F}_2'') \in \mathcal{V}_\Psi^\xi [\![\Delta \Vdash K]\!]_{x;:}^{m_2}.$$

Similarly, from $\dagger_1^{\ell'}$, we get

$$\dagger_1^{\ell''} \; \forall x \in \Upsilon_1. \forall m_1. (\cdot, \mathcal{A}_1^{\ell''}, \mathcal{C}_1^{\ell''} \mathcal{D}_1'' \mathcal{F}_1''; \cdot, \mathcal{A}_2^{\ell''}, \mathcal{C}_2^{\ell''} \mathcal{D}_2'' \mathcal{F}_2'') \in \mathcal{V}_\Psi^\xi [\![K \Vdash \cdot]\!]_{;:x}^{m_1} \text{ and}$$
$$\forall x \in (\Theta_2 \cap \Theta_3). \forall m_1. (\cdot, \mathcal{A}_1^{\ell''}, \mathcal{C}_1^{\ell''} \mathcal{D}_1'' \mathcal{F}_1''; \cdot, \mathcal{A}_2^{\ell''}, \mathcal{C}_2^{\ell''} \mathcal{D}_2'' \mathcal{F}_2'') \in \mathcal{V}_\Psi^\xi [\![K \Vdash \cdot]\!]_{x;:}^{m_1}.$$

And from $\dagger_3'$ we get

$$\dagger_3'' \; \forall x \in \Upsilon_1. \forall m_3. (\mathcal{C}_1'' \mathcal{D}_1'', \mathcal{F}_1'', \cdot; \mathcal{C}_2'' \mathcal{D}_2'', \mathcal{F}_2'', \cdot) \in \mathcal{V}_\Psi^\xi [\![K \Vdash \cdot]\!]_{;:x}^{m_3} \text{ and}$$
$$\forall x \in (\Theta_2 \cap \Theta_3). \forall m_3. (\mathcal{C}_1'' \mathcal{D}_1'', \mathcal{F}_1'', \cdot; \mathcal{C}_2'' \mathcal{D}_2'', \mathcal{F}_2'', \cdot) \in \mathcal{V}_\Psi^\xi [\![K \Vdash \cdot]\!]_{x;:}^{m_3}.$$

Without loss of generality assume the tree in $\mathcal{C}_i''$ that offers along $y_\alpha^c$ is $\mathcal{A}_i^{\kappa''}$.
In particular, we have

$$\star_2 \; \forall m_2. (\mathcal{C}_1'', \mathcal{D}_1'', \mathcal{F}_1''; \mathcal{C}_2'', \mathcal{D}_2'', \mathcal{F}_2'') \in \mathcal{V}_\Psi^\xi [\![\Delta \Vdash K]\!]_{y_c^\alpha;:}^{m_2}, \text{ and}$$

$$\star_1^\kappa \; \forall m_1. (\cdot, \mathcal{A}_1^{\kappa''}, \mathcal{C}_1^{\kappa''} \mathcal{D}_1'' \mathcal{F}_1''; \cdot, \mathcal{A}_2^{\kappa''}, \mathcal{C}_2^{\kappa''} \mathcal{D}_2'' \mathcal{F}_2'') \in \mathcal{V}_\Psi^\xi [\![K \Vdash \cdot]\!]_{;:y_c^\alpha}^{m_1}.$$

Moreover by the forward closure lemma on $\dagger_3'$, we get

$$\star_3 \; \forall m_3. (\mathcal{C}_1'' \mathcal{D}_1'' \mathcal{F}_1'' \cdot; \mathcal{C}_2'' \mathcal{D}_2'' \mathcal{F}_2'' \cdot) \in \mathcal{E}_\Psi^\xi [\![K \Vdash \cdot]\!]^{m_3}.$$

And by the forward closure lemma on $\dagger_1^i$, we get

$$\star_4^\ell \; \forall m_1. (\cdot \mathcal{A}_1^{\ell''} \mathcal{C}_1^{\ell''} \mathcal{D}_1'' \mathcal{F}_1''; \cdot \mathcal{A}_2^{\ell''} \mathcal{C}_2^{\ell''} \mathcal{D}_2'' \mathcal{F}_2'') \in \mathcal{E}_\Psi^\xi [\![\cdot \Vdash d_j : T_j]\!]^{m_1}.$$

Now we consider cases based on the row in the definition of logical relation (Figure 2) with which we can prove $\star_1^\kappa$. It is a unique row determined based on the type of $y_c^\alpha$ and whether the process sending over $y_c^\alpha$ is a message or a forwarding process.

64

**Row (1).** The conditions validating this row are:

$$\circ_1 \ (\cdot, \mathcal{A}_1^{\kappa''}, \mathcal{C}_1^{\kappa''}\mathcal{D}_1''\mathcal{F}_1''; \cdot, \mathcal{A}_2^{\kappa''}, \mathcal{C}_2^{\kappa''}\mathcal{D}_2''\mathcal{F}_2'') \in \mathsf{Tree}_\Psi(\cdot \Vdash y_\alpha{:}\mathbf{1}[c])$$

$$\circ_2 \ \mathcal{A}_1^{\kappa''} = \mathbf{msg}(\mathbf{close}\ y_\alpha^c), \ \text{i.e.}\ \mathcal{C}_1'' = \mathcal{C}_1^{\kappa''}\mathbf{msg}(\mathbf{close}\ y_\alpha^c)$$

$$\circ_3 \ \mathcal{A}_2^{\kappa''} = \mathbf{msg}(\mathbf{close}\ y_\alpha^c), \ \text{i.e.}\ \mathcal{C}_2'' = \mathcal{C}_2^{\kappa''}\mathbf{msg}(\mathbf{close}\ y_\alpha^c).$$

We get

$$(\mathcal{C}_i'', \mathcal{D}_i'', \mathcal{F}_i'') = (\mathcal{C}_i^{\kappa''}\mathbf{msg}(\mathbf{close}\ y_\alpha^c), \mathcal{D}_i'', \mathcal{F}_i'').$$

By Figure 4, we know that

$$(\mathcal{C}_1^{\kappa''}\mathbf{msg}(\mathbf{close}\ y_\alpha^c), \mathcal{D}_1'', \mathcal{F}_1'') \xrightarrow[\Delta', y_\alpha{:}\mathbf{1}[c]\Vdash K]{y_\alpha^c \cdot \overline{\mathbf{msg}(\mathbf{close}\ y_\alpha^c)}} (\mathcal{C}_1^{\kappa''}, \mathbf{msg}(\mathbf{close}\ y_\alpha^c)\mathcal{D}_1'', \mathcal{F}_1''),$$

and thus $\mathcal{C}_1', \mathcal{D}_1', \mathcal{F}_1' = (\mathcal{C}_1^{\kappa''}, \mathbf{msg}(\mathbf{close}\ y_\alpha^c)\mathcal{D}_1'', \mathcal{F}_1'')$.
By Figure 4, we can step

$$(\mathcal{C}_2'', \mathcal{D}_2'', \mathcal{F}_2'') \xrightarrow[\Delta', y_\alpha{:}\mathbf{1}[c]\Vdash K]{y_\alpha^c \cdot \overline{\mathbf{msg}(\mathbf{close}\ y_\alpha^c)}} (\mathcal{C}_2^{\kappa''}, \mathbf{msg}(\mathbf{close}\ y_\alpha^c)\mathcal{D}_2'', \mathcal{F}_2'').$$

By the typing rules, we know that

$$\circ_4 \ (\mathcal{C}_1^{\kappa''}, \mathbf{msg}(\mathbf{close}\ y_\alpha^c)\mathcal{D}_1'', \mathcal{F}_1''), \mathcal{C}_2^{\kappa''}, \mathbf{msg}(\mathbf{close}\ y_\alpha^c)\mathcal{D}_2'', \mathcal{F}_2'') \in \mathsf{Tree}(\Delta' \Vdash K).$$

By row 6 of the logical relation and $\star_2$ we know that

$$\circ_5 \ (\mathcal{C}_1^{\kappa''}\mathbf{msg}(\mathbf{close}\ y_\alpha^c)\mathcal{D}_1''\mathcal{F}_1'', \mathcal{C}_2^{\kappa''}\mathbf{msg}(\mathbf{close}\ y_\alpha^c)\mathcal{D}_2''\mathcal{F}_2'') \in \mathcal{E}_\Psi^\xi[\![\Delta' \Vdash K]\!]^k$$

We can apply the induction hypothesis, on $\circ_4$ and $\circ_5$ and $\star_3$ and $\star_4^\ell$ for $\ell \neq \kappa$ (every tree in $\mathcal{C}_i^{\kappa''}$) to get

$$\mathcal{C}_1^{\kappa''}, \mathbf{msg}(\mathbf{close}\ y_\alpha^c)\mathcal{D}_1'', \mathcal{F}_1'' \sim^k \mathcal{C}_2^{\kappa''}, \mathbf{msg}(\mathbf{close}\ y_\alpha^c)\mathcal{D}_2'', \mathcal{F}_2''.$$

which completes the proof.

**Row (2).** The conditions validating this row are:

$$\circ_1 \ (\cdot, \mathcal{A}_1^{\kappa''}, \mathcal{C}_1^{\kappa''}\mathcal{D}_1''\mathcal{F}_1''; \cdot, \mathcal{A}_2^{\kappa''}, \mathcal{C}_2^{\kappa''}\mathcal{D}_2''\mathcal{F}_2'') \in \mathsf{Tree}_\Psi(\cdot \Vdash y_\alpha{:} \oplus \{\ell{:}A_\ell\}_{\ell \in I}[c])$$

$$\circ_2 \ \mathcal{A}_1^{\kappa''} = \mathcal{A}_1^{\kappa_1}\mathbf{msg}(y_\alpha^c.k; y_\alpha^c \leftarrow u_\delta^c), \ \text{i.e.}\ \mathcal{C}_1'' = \mathcal{C}_1^{\kappa''}\mathcal{A}_1^{\kappa_1}\mathbf{msg}(y_\alpha^c.k; y_\alpha^c \leftarrow u_\delta^c)$$

$$\circ_3 \ \mathcal{A}_2^{\kappa''} = \mathcal{A}_2^{\kappa_1}\mathbf{msg}(y_\alpha^c.k; y_\alpha^c \leftarrow u_\delta^c), \ \text{i.e.}\ \mathcal{C}_2'' = \mathcal{C}_2^{\kappa''}\mathcal{A}_2^{\kappa_1}\mathbf{msg}(y_\alpha^c.k; y_\alpha^c \leftarrow u_\delta^c)$$

$$\circ_4 \ \forall m_1. \ (\cdot \mathcal{A}_1^{\kappa_1}\mathbf{msg}(y_\alpha^c.k; y_\alpha^c \leftarrow u_\delta^c)\mathcal{C}_1^{\kappa''}\mathcal{D}_1''\mathcal{F}_1''; \cdot \mathcal{A}_2^{\kappa_1}\mathbf{msg}(y_\alpha^c.k; y_\alpha^c \leftarrow u_\delta^c)\mathcal{C}_2^{\kappa''}\mathcal{D}_2''\mathcal{F}_2'') \in \mathcal{E}_\Psi^\xi(\cdot \Vdash u_\delta{:}A_k[c])^{m_1}$$

We get

$$(\mathcal{C}_i'', \mathcal{D}_i'', \mathcal{F}_i'') = (\mathcal{C}_i^{\kappa''}\mathcal{A}_i^{\kappa_1}\mathbf{msg}(y_\alpha^c.k; y_\alpha^c \leftarrow u_\delta^c), \mathcal{D}_i'', \mathcal{F}_i'').$$

By Figure 4, we know that

$$(\mathcal{C}_1^{\kappa''}\mathcal{A}_1^{\kappa_1}\mathbf{msg}(y_\alpha^c.k; y_\alpha^c \leftarrow u_\delta^c), \mathcal{D}_1'', \mathcal{F}_1'') \xrightarrow[\Delta', y_\alpha{:}\oplus\{\ell{:}A_\ell\}_{\ell \in I}[c]\Vdash K]{y_\alpha^c \cdot \overline{\mathbf{msg}(y_\alpha^c.k; y_\alpha^c \leftarrow u_\delta^c)}} (\mathcal{C}_1^{\kappa''}\mathcal{A}_1^{\kappa_1}, \mathbf{msg}(y_\alpha^c.k; y_\alpha^c \leftarrow u_\delta^c)\mathcal{D}_1'', \mathcal{F}_1''),$$

and thus $\mathcal{C}_1', \mathcal{D}_1', \mathcal{F}_1' = (\mathcal{C}_1^{\kappa''}\mathcal{A}_1^{\kappa_1}, \mathbf{msg}(y_\alpha^c.k; y_\alpha^c \leftarrow u_\delta^c)\mathcal{D}_1'', \mathcal{F}_1'')$.
By Figure 4, we can step

$$(\mathcal{C}_2'', \mathcal{D}_2'', \mathcal{F}_2'') \xrightarrow[\Delta', y_\alpha{:}\oplus\{\ell{:}A_\ell\}_{\ell \in I}[c]\Vdash K]{y_\alpha^c \cdot \overline{\mathbf{msg}(y_\alpha^c.k; y_\alpha^c \leftarrow u_\delta^c)}} (\mathcal{C}_2^{\kappa''}\mathcal{A}_2^{\kappa_1}, \mathbf{msg}(y_\alpha^c.k; y_\alpha^c \leftarrow u_\delta^c)\mathcal{D}_2'', \mathcal{F}_2'').$$

By the typing rules, we know that

$$\circ_5 \ (\mathcal{C}_1^{\kappa''}\mathcal{A}_1^{\kappa_1}, \mathbf{msg}(y_\alpha^c.k; y_\alpha^c \leftarrow u_\delta^c)\mathcal{D}_1'', \mathcal{F}_1''), \mathcal{C}_2^{\kappa''}\mathcal{A}_2^{\kappa_1}, \mathbf{msg}(y_\alpha^c.k; y_\alpha^c \leftarrow u_\delta^c)\mathcal{D}_2'', \mathcal{F}_2'') \in \mathsf{Tree}(\Delta', u_\delta{:}A_k[c] \Vdash K).$$

By row 7 of the logical relation and $\star_2$ we get

$$\circ_6 \ \forall m_2. \ (\mathcal{C}_1^{\kappa''}\mathcal{A}_1^{\kappa_1}\mathbf{msg}(y_\alpha^c.k; y_\alpha^c \leftarrow u_\delta^c)\mathcal{D}_1''\mathcal{F}_1'', \mathcal{C}_2^{\kappa''}\mathcal{A}_2^{\kappa_1}\mathbf{msg}(y_\alpha^c.k; y_\alpha^c \leftarrow u_\delta^c)\mathcal{D}_2''\mathcal{F}_2'') \in \mathcal{E}_\Psi^\xi[\![\Delta', u_\delta{:}A_k[c] \Vdash K]\!]^{m_2}$$

We can apply the induction hypothesis, on $\circ_4$ and $\circ_5$ and $\circ_6$ and $\star_3$ and $star_4^\ell$ for $\ell \neq \kappa$ (every tree in $\mathcal{C}_i^{\kappa''}$) to get

$$\mathcal{C}_1^{\kappa''}\mathcal{A}_1^{\kappa_1}, \mathbf{msg}(y_\alpha^c.k; y_\alpha^c \leftarrow u_\delta^c)\mathcal{D}_1'', \mathcal{F}_1'' \sim^k \mathcal{C}_2^{\kappa''}\mathcal{A}_2^{\kappa_1}, \mathbf{msg}(y_\alpha^c.k; y_\alpha^c \leftarrow u_\delta^c)\mathcal{D}_2'', \mathcal{F}_2''.$$

which completes the proof.

**Row (4).** The conditions validating this row are:

$$\circ_1 (\cdot, \mathcal{A}_1^{\kappa''}, \mathcal{C}_1^{\kappa''}\mathcal{D}_1''\mathcal{F}_1''; \cdot, \mathcal{A}_2^{\kappa''}, \mathcal{C}_2^{\kappa''}\mathcal{D}_2''\mathcal{F}_2'') \in \mathsf{Tree}_\Psi(\cdot \Vdash y_\alpha{:}A \otimes B[c])$$

$\circ_2\, \mathcal{A}_1^{\kappa''} = \mathcal{A}_1^{\kappa_1}\mathcal{T}_1\mathbf{msg}(\mathbf{send}\,x_\beta^c\,y_\alpha^c; y_\alpha^c \leftarrow u_\delta^c)$ for $\mathcal{T}_1 \in \mathsf{Tree}(\cdot \Vdash x_\beta{:}A[c])$, i.e. $\mathcal{C}_1'' = \mathcal{C}_1^{\kappa''}\mathcal{A}_1^{\kappa_1}\mathcal{T}_1\mathbf{msg}(\mathbf{send}\,x_\beta^c\,y_\alpha^c; y_\alpha^c \leftarrow u_\delta^c)$

$\circ_3\, \mathcal{A}_2^{\kappa''} = \mathcal{A}_2^{\kappa_1}\mathcal{T}_1\mathbf{msg}(\mathbf{send}\,x_\beta^c\,y_\alpha^c; y_\alpha^c \leftarrow u_\delta^c)$ for $\mathcal{T}_2 \in \mathsf{Tree}(\cdot \Vdash x_\beta{:}A[c])$, i.e. $\mathcal{C}_2'' = \mathcal{C}_2^{\kappa''}\mathcal{A}_2^{\kappa_1}\mathcal{T}_2\mathbf{msg}(\mathbf{send}\,x_\beta^c\,y_\alpha^c; y_\alpha^c \leftarrow u_\delta^c)$

$\circ_4\, \forall m_1.\, (\cdot\mathcal{T}_1\mathcal{A}_1^{\kappa_1}\mathbf{msg}(\mathbf{send}\,x_\beta^c\,y_\alpha^c; y_\alpha^c \leftarrow u_\delta^c)\mathcal{C}_1^{\kappa''}\mathcal{D}_1''\mathcal{F}_1''; \cdot\mathcal{T}_2\mathcal{A}_2^{\kappa_1}\mathbf{msg}(\mathbf{send}\,x_\beta^c\,y_\alpha^c; y_\alpha^c \leftarrow u_\delta^c)\mathcal{C}_2^{\kappa''}\mathcal{D}_2''\mathcal{F}_2'') \in \mathcal{E}_\Psi^\xi(\cdot \Vdash x_\beta{:}A[c])^{m_1}$

$\circ_5\, \forall m_1.\, (\cdot\mathcal{A}_1^{\kappa_1}\mathcal{T}_1\mathbf{msg}(\mathbf{send}\,x_\beta^c\,y_\alpha^c; y_\alpha^c \leftarrow u_\delta^c)\mathcal{C}_1^{\kappa''}\mathcal{D}_1''\mathcal{F}_1''; \cdot\mathcal{A}_2^{\kappa_1}\mathcal{T}_2\mathbf{msg}(\mathbf{send}\,x_\beta^c\,y_\alpha^c; y_\alpha^c \leftarrow u_\delta^c)\mathcal{C}_2^{\kappa''}\mathcal{D}_2''\mathcal{F}_2'') \in \mathcal{E}_\Psi^\xi(\cdot \Vdash u_\delta{:}B[c])^{m_1}$

We get

$$(\mathcal{C}_i'', \mathcal{D}_i'', \mathcal{F}_i'') = (\mathcal{C}_i^1\mathbf{msg}(\mathbf{send}\,x_\beta^c\,y_\alpha^c; y_\alpha^c \leftarrow u_\delta^c), \mathcal{D}_i'', \mathcal{F}_i''),$$

where we define $\mathcal{C}_i^1$ as $\mathcal{C}_i^{\kappa''}\mathcal{A}_i^{\kappa_1}\mathcal{T}_i$. By Figure 4, we know that

$$(\mathcal{C}_1^1\mathbf{msg}(\mathbf{send}\,x_\beta^c\,y_\alpha^c; y_\alpha^c \leftarrow u_\delta^c), \mathcal{D}_1'', \mathcal{F}_1'') \xrightarrow[\Delta', y_\alpha:(A\otimes B)[c] \Vdash K]{y_\alpha^c.\overline{\mathbf{msg}(\mathbf{send}\,x_\beta^c\,y_\alpha^c; y_\alpha^c \leftarrow u_\delta^c)}} (\mathcal{C}_1^1, \mathbf{msg}(\mathbf{send}\,x_\beta^c\,y_\alpha^c; y_\alpha^c \leftarrow u_\delta^c)\mathcal{D}_1'', \mathcal{F}_1''),$$

and thus $\mathcal{C}_1', \mathcal{D}_1', \mathcal{F}_1' = (\mathcal{C}_1^1, \mathbf{msg}(\mathbf{send}\,x_\beta^c\,y_\alpha^c; y_\alpha^c \leftarrow u_\delta^c)\mathcal{D}_1'', \mathcal{F}_1'')$.
By Figure 4, we can step

$$(\mathcal{C}_2'', \mathcal{D}_2'', \mathcal{F}_2'') \xrightarrow[\Delta', y_\alpha:(A\otimes B)[c] \Vdash K]{y_\alpha^c.\overline{\mathbf{msg}(\mathbf{send}\,x_\beta^c\,y_\alpha^c; y_\alpha^c \leftarrow u_\delta^c)}} (\mathcal{C}_2^1, \mathbf{msg}(\mathbf{send}\,x_\beta^c\,y_\alpha^c; y_\alpha^c \leftarrow u_\delta^c)\mathcal{D}_2'', \mathcal{F}_2'').$$

By the typing rules, we know that

$\circ_6\, (\mathcal{C}_1^1, \mathbf{msg}(\mathbf{send}\,x_\beta^c\,y_\alpha^c; y_\alpha^c \leftarrow u_\delta^c)\mathcal{D}_1'', \mathcal{F}_1''), \mathcal{C}_2^1, \mathbf{msg}(\mathbf{send}\,x_\beta^c\,y_\alpha^c; y_\alpha^c \leftarrow u_\delta^c)\mathcal{D}_2'', \mathcal{F}_2'') \in \mathsf{Tree}(\Delta', x_\beta{:}A[c], u_\delta{:}B[c] \Vdash K).$

By row 9 of the logical relation and $\star_2$

$\circ_7\, \forall m_2.\, (\mathcal{C}_1^1\mathbf{msg}(\mathbf{send}\,x_\beta^c\,y_\alpha^c; y_\alpha^c \leftarrow u_\delta^c)\mathcal{D}_1''\mathcal{F}_1'', \mathcal{C}_2^1\mathbf{msg}(\mathbf{send}\,x_\beta^c\,y_\alpha^c; y_\alpha^c \leftarrow u_\delta^c)\mathcal{D}_2''\mathcal{F}_2'') \in \mathcal{E}_\Psi^\xi[\![\Delta, x_\beta{:}A[c], u_\delta{:}B[c] \Vdash K]\!]^{m_2}$

We can apply the induction hypothesis, on $\circ_4$ and $\circ_5$ and $\circ_6$ and $\circ_7$ and $\star_3$ and $\star_4^\ell$ for $\ell \neq \kappa$ (every tree in $\mathcal{C}_i^{\kappa''}$) to get

$$\mathcal{C}_1^1, \mathbf{msg}(\mathbf{send}\,x_\beta^c\,y_\alpha^c; y_\alpha^c \leftarrow u_\delta^c)\mathcal{D}_1'', \mathcal{F}_1'' \sim^k \mathcal{C}_2^1, \mathbf{msg}(\mathbf{send}\,x_\beta^c\,y_\alpha^c; y_\alpha^c \leftarrow u_\delta^c)\mathcal{D}_2'', \mathcal{F}_2''.$$

which completes the proof.

**Row (11).** The conditions validating this row are:

$$\circ_1 (\cdot, \mathcal{A}_1^{\kappa''}, \mathcal{C}_1^{\kappa''}\mathcal{D}_1''\mathcal{F}_1''; \cdot, \mathcal{A}_2^{\kappa''}, \mathcal{C}_2^{\kappa''}\mathcal{D}_2''\mathcal{F}_2'') \in \mathsf{Tree}_\Psi(\cdot \Vdash y_\alpha{:}A[c])$$

$\circ_2\, \mathcal{A}_1^{\kappa''} = \mathcal{A}_1^{\kappa_1}\mathbf{proc}(y_\alpha^c, y_\alpha^c \leftarrow x_\beta^c@d_1)$, i.e. $\mathcal{C}_1'' = \mathcal{C}_1^{\kappa''}\mathcal{A}_1^{\kappa_1}\mathbf{proc}(y_\alpha^c, y_\alpha^c \leftarrow x_\beta^c@d_1)$

$\circ_3\, \mathcal{A}_2^{\kappa''} = \mathcal{A}_2^{\kappa_1}\mathbf{proc}(y_\alpha^c, y_\alpha^c \leftarrow x_\beta^c@d_2)$, i.e. $\mathcal{C}_2'' = \mathcal{C}_2^{\kappa''}\mathcal{A}_2^{\kappa_1}\mathbf{proc}(y_\alpha^c, y_\alpha^c \leftarrow x_\beta^c@d_2)$

$\circ_4\, \forall m_1.\, ([x_\beta^c/y_\alpha^c]\cdot\mathcal{A}_1^{\kappa_1}\mathcal{C}_1^{\kappa''}\mathcal{D}_1''\mathcal{F}_1''; [x_\beta^c/y_\alpha^c]\cdot\mathcal{A}_2^{\kappa_1}\mathcal{C}_2^{\kappa''}\mathcal{D}_2''\mathcal{F}_2'') \in \mathcal{E}_\Psi^\xi(\cdot \Vdash x_\beta{:}A[c])^{m_1}$

We get

$$(\mathcal{C}_i'', \mathcal{D}_i'', \mathcal{F}_i'') = ([x_\beta^c/y_\alpha^c]\mathcal{C}_i^{\kappa''}\mathcal{A}_i^{\kappa_1}, \mathcal{D}_i'', \mathcal{F}_i'').$$

By Figure 4, we know that

$$(\mathcal{C}_1^{\kappa''}\mathcal{A}_1^{\kappa_1}\mathbf{proc}(y_\alpha^c, y_\alpha^c \leftarrow x_\beta^c@d_1), \mathcal{D}_1'', \mathcal{F}_1'') \xrightarrow[\Delta', y_\alpha:A[c] \Vdash K]{y_\alpha^c.\overline{\{x_\beta^c/y_\alpha^c\}}} ([x_\beta^c/y_\alpha^c]\mathcal{C}_1^{\kappa''}\mathcal{A}_1^{\kappa_1}, \mathcal{D}_1'', \mathcal{F}_1''),$$

66

and thus $\mathcal{C}_1', \mathcal{D}_1', \mathcal{F}_1' = ([x_\beta^c/y_\alpha^c]\mathcal{C}_1^{\kappa''}\mathcal{A}_1^{\kappa_1}, \mathcal{D}_1'', \mathcal{F}_1'')$.
By Figure 4, we can step
$$(\mathcal{C}_2'', \mathcal{D}_2'', \mathcal{F}_2'') \xrightarrow[\Delta', y_\alpha:A[c] \Vdash K]{y_\alpha^c \cdot \overline{\{x_\beta^c/y_\alpha^c\}}} ([x_\beta^c/y_\alpha^c]\mathcal{C}_2^{\kappa''}\mathcal{A}_2^{\kappa_1}, \mathcal{D}_2'', \mathcal{F}_2'').$$

By the typing rules, we know that
$$\circ_5 \ ([x_\beta^c/y_\alpha^c]\mathcal{C}_1^{\kappa''}\mathcal{A}_1^{\kappa_1}, \mathcal{D}_1'', \mathcal{F}_1''), [x_\beta^c/y_\alpha^c]\mathcal{C}_2^{\kappa''}\mathcal{A}_2^{\kappa_1}, \mathcal{D}_2'', \mathcal{F}_2'') \in \mathsf{Tree}(\Delta', x_\beta:A[c] \Vdash [x_\beta^c/y_\alpha^c]K).$$

By row 12 and $\star_2$ we get
$$\circ_6 \ \forall m_2. \ ([x_\beta^c/y_\alpha^c]\mathcal{C}_1^{\kappa''}\mathcal{A}_1^{\kappa_1}\mathcal{D}_1''\mathcal{F}_1'', [x_\beta^c/y_\alpha^c]\mathcal{C}_2^{\kappa''}\mathcal{A}_2^{\kappa_1}\mathcal{D}_2''\mathcal{F}_2'') \in \mathcal{E}_\Psi^\xi [\![ \Delta', x_\beta:A[c] \Vdash [x_\beta^c/y_\alpha^c]K ]\!]^{m_2}.$$

By forward closure lemma on $\star_1^\ell$ for every $\ell \neq \kappa$ and the fact that $y_\alpha^c$ does not occur in $\mathcal{C}_i^{\kappa''}$, we get
$$\star_1^{\ell'} \ \forall m_1. \ ([x_\beta^c/y_\alpha^c] \cdot A_1^{\ell''} \mathcal{C}_1^{\ell''} \mathcal{D}_1'' \mathcal{F}_1''; [x_\beta^c/y_\alpha^c] \cdot A_2^{\ell''} \mathcal{C}_2^{\ell''} \mathcal{D}_2'' \mathcal{F}_2'') \in \mathcal{E}_\Psi^\xi [\![ \cdot \Vdash d_j:T_j ]\!]^{m_1}.$$

Here we consider two cases. If $\mathcal{D} \neq \cdot$, we know that $x_\beta^c$ does not occur in $K$ and thus it does not occur in $\mathcal{F}_i''$. By forward closure lemma on $\star_3$ we get
$$\star_3' \ \forall m_3. \ ([x_\beta^c/y_\alpha^c]\mathcal{C}_1^{\kappa''}\mathcal{A}_1^{\kappa_1}\mathcal{D}_1''\mathcal{F}_1''\cdot; [x_\beta^c/y_\alpha^c]\mathcal{C}_2^{\kappa''}\mathcal{A}_2^{\kappa_1}\mathcal{D}_2''\mathcal{F}_2''\cdot) \in \mathcal{E}_\Psi^\xi [\![ K \Vdash \cdot ]\!]^{m_3}.$$

We can apply the induction hypothesis, on $\circ_4$ and $\circ_5$ and $\circ_6$ and $\star_1^{\ell'}$ for every $\ell \neq \kappa$ and $\star_3'$ to get
$$[x_\beta^c/y_\alpha^c]\mathcal{C}_1^{\kappa''}\mathcal{A}_1^{\kappa_1}, \mathcal{D}_1'', \mathcal{F}_1'' \sim^k [x_\beta^c/y_\alpha^c]\mathcal{C}_1^{\kappa''}\mathcal{A}_1^{\kappa_1}, \mathcal{D}_2'', \mathcal{F}_2''.$$

Otherwise, if $\mathcal{D} = \cdot$ we know that $K = y_\alpha:A[c]$. By $\dagger_3''$
$$\circ_8 \ \forall m_3. \ (\mathcal{C}_1^{\kappa''}\mathcal{A}_1^{\kappa_1}\mathbf{proc}(y_\alpha^c, y_\alpha^c \leftarrow x_\beta^c@d_1)\mathcal{D}_1'', \mathcal{F}_1'', \cdot; \mathcal{C}_2^{\kappa''}\mathcal{A}_2^{\kappa_1}\mathbf{proc}(y_\alpha^c, y_\alpha^c \leftarrow x_\beta^c@d_2)\mathcal{D}_2'', \mathcal{F}_2'', \cdot) \in \mathcal{V}_\Psi^\xi [\![ y_\alpha:A[c] \Vdash \cdot ]\!]^{m_3}_{y_\alpha^c;\cdot}.$$

By row 12 and $\circ_8$ we get
$$\circ_9 \ \forall m_3. \ ([x_\beta^c/y_\alpha^c]\mathcal{C}_1^{\kappa''}\mathcal{A}_1^{\kappa_1}\mathcal{D}_1''\mathcal{F}_1''\cdot; [x_\beta^c/y_\alpha^c]\mathcal{C}_2^{\kappa''}\mathcal{A}_2^{\kappa_1}\mathcal{D}_2''\mathcal{F}_2''\cdot) \in \mathcal{E}_\Psi^\xi [\![ x_\beta:A[c] \Vdash \cdot ]\!]^{m_3}.$$

We can apply the induction hypothesis, on $\circ_4$ and $\circ_5$ and $\circ_6$ and $\star_1^{\ell'}$ for every $\ell \neq \kappa$ and $\circ_9$ to get
$$[x_\beta^c/y_\alpha^c]\mathcal{C}_1^{\kappa''}\mathcal{A}_1^{\kappa_1}, \mathcal{D}_1'', \mathcal{F}_1'' \sim^k [x_\beta^c/y_\alpha^c]\mathcal{C}_1^{\kappa''}\mathcal{A}_1^{\kappa_1}, \mathcal{D}_2'', \mathcal{F}_2''.$$

which completes the proof.

□

| Property | Condition |
|---|---|
| (1) $\mathcal{C}, \mathcal{D}, \mathcal{F} \xRightarrow[\Delta \Vdash K]{x.q} \mathcal{C}', \mathcal{D}', \mathcal{F}'$ | $\mathcal{C}, \mathcal{D}, \mathcal{F} \in \mathsf{Tree}(\Delta \Vdash K)$ and |
| | $\mathcal{C}, \mathcal{D}, \mathcal{F} \mapsto^{*\Theta;\Upsilon,x}_{\Delta \Vdash K} \mathcal{C}'', \mathcal{D}'', \mathcal{F}''$ and $\mathcal{C}'', \mathcal{D}'', \mathcal{F}'' \xRightarrow[\Delta \Vdash K]{x.q} \mathcal{C}', \mathcal{D}', \mathcal{F}'$ |
| (2) $\mathcal{C}, \mathcal{D}, \mathcal{F} \xRightarrow[\Delta \Vdash K]{x.\overline{q}} \mathcal{C}', \mathcal{D}', \mathcal{F}'$ | $\mathcal{C}, \mathcal{D}, \mathcal{F} \in \mathsf{Tree}(\Delta \Vdash K)$ and |
| | $\mathcal{C}, \mathcal{D}, \mathcal{F} \mapsto^{*\Theta,x;\Upsilon}_{\Delta \Vdash K} \mathcal{C}'', \mathcal{D}'', \mathcal{F}''$ and $\mathcal{C}'', \mathcal{D}'', \mathcal{F}'' \xRightarrow[\Delta \Vdash K]{x.\overline{q}} \mathcal{C}', \mathcal{D}', \mathcal{F}'$ |

*definition of* $\xRightarrow[\Delta \Vdash K]{x.q}$ :

(1) $\cdot, \mathbf{msg}(\mathbf{close}\ y_\alpha^c), \mathcal{F} \qquad \xrightarrow[\cdot \Vdash y_\alpha:1[c]]{y_\alpha^c.\mathbf{msg}(\mathbf{close}\ y_\alpha^c)} \qquad \cdot, \cdot, \cdot$

(2) $\mathcal{C}'\mathbf{msg}(\mathbf{close}\ x_\beta^{c'}), \mathcal{D}, \mathcal{F} \qquad \xrightarrow[\Delta', x_\beta:1[c'] \Vdash K]{x_\beta^{c'}.\overline{\mathbf{msg}(\mathbf{close}\ x_\beta^{c'})}} \qquad \mathcal{C}', \mathbf{msg}(\mathbf{close}\ x_\beta^{c'})\mathcal{D}, \mathcal{F}$

(3) $\mathcal{C}, \mathcal{D}'\mathbf{msg}(y_\alpha^c.k; y_\alpha^c \leftarrow u_\delta^c), \mathcal{F} \qquad \xrightarrow[\Delta \Vdash y_\alpha:\oplus\{\ell:A_\ell\}_{\ell \in L}[c]]{y_\alpha^c.\mathbf{msg}(y_\alpha^c.k; y_\alpha^c \leftarrow u_\delta^c)} \qquad \mathcal{C}, \mathcal{D}', \mathbf{msg}(y_\alpha^c.k; y_\alpha^c \leftarrow u_\delta^c)\mathcal{F}$

(4) $\mathcal{C}_1\mathbf{msg}(x_\beta^{c'}.k; x_\beta^{c'} \leftarrow u_\delta^{c'}), \mathcal{D}_1', \mathcal{F}_1 \qquad \xrightarrow[\Delta, x_\beta:\oplus\{\ell:A_\ell\}_{\ell \in L}[c'] \Vdash K]{x_\beta^{c'}.\overline{\mathbf{msg}(x_\beta^{c'}.k; x_\beta^{c'} \leftarrow u_\delta^{c'})}} \qquad \mathcal{C}, \mathbf{msg}(x_\beta^{c'}.k; x_\beta^{c'} \leftarrow u_\delta^{c'})\mathcal{D}', \mathcal{F}$

(5) $\mathcal{C}, \mathcal{D}', \mathbf{msg}(y_\alpha^c.k; u_\delta^c \leftarrow y_\alpha^c)\mathcal{F} \qquad \xrightarrow[\Delta \Vdash y_\alpha:\&\{\ell:A_\ell\}_{\ell \in L}[c]]{y_\alpha^c.\overline{\mathbf{msg}(y_\alpha^c.k; u_\delta^c \leftarrow y_\alpha^c)}} \qquad \mathcal{C}, \mathcal{D}'\mathbf{msg}(y_\alpha^c.k; u_\delta^c \leftarrow y_\alpha^c), \mathcal{F}$

(6) $\mathcal{C}, \mathbf{msg}(x_\beta^{c'}.k; u_\delta^{c'} \leftarrow x_\beta^{c'})\mathcal{D}', \mathcal{F} \qquad \xrightarrow[\Delta, x_\beta:\&\{\ell:A_\ell\}_{\ell \in L}[c'] \Vdash K]{x_\beta^{c'}.\mathbf{msg}(x_\beta^{c'}.k; u_\delta^{c'} \leftarrow x_\beta^{c'})} \qquad \mathcal{C}\mathbf{msg}(x_\beta^{c'}.k; u_\delta^{c'} \leftarrow x_\beta^{c'}), \mathcal{D}', \mathcal{F}$

(7) $\mathcal{C}'\mathcal{C}'', \mathcal{D}'\mathcal{T}\mathbf{msg}(\mathbf{send}\ z_\gamma^c\ y_\alpha^c; y_\alpha^c \leftarrow u_\delta^c), \mathcal{F} \qquad \xrightarrow[\Delta', \Delta'' \Vdash y_\alpha:(A \otimes B)[c]]{y_\alpha^c.\mathbf{msg}(\mathbf{send}\ z_\gamma^c\ y_\alpha^c; y_\alpha^c \leftarrow u_\delta^c)} \qquad \mathcal{C}', \mathcal{D}', \mathcal{C}\mathcal{T}\mathbf{msg}(\mathbf{send}\ z_\gamma^c\ y_\alpha^c; y_\alpha^c \leftarrow u_\delta^c)\mathcal{F}$
where $\mathcal{T} \in \mathsf{Tree}_\Psi(\Delta'' \Vdash z_\gamma:A[c])$ and
$\mathcal{C}' \in \mathsf{Tree}_\Psi(\cdot \Vdash \Delta')$ and $\mathcal{C}'' \in \mathsf{Tree}_\Psi(\cdot \Vdash \Delta'')$

(7') $\mathcal{C}'\mathcal{C}'', \mathcal{D}'\mathcal{T}\mathbf{msg}(\mathbf{send}\ z_\gamma^c\ y_\alpha^c; y_\alpha^c \leftarrow u_\delta^c), \mathcal{F} \qquad \xrightarrow[\Delta', \Delta'' \Vdash y_\alpha:(A \otimes B)[c]]{y_\alpha^c.\mathbf{msg}(\mathbf{send}\ z_\gamma^c\ y_\alpha^c; y_\alpha^c \leftarrow u_\delta^c)} \qquad \mathcal{C}'', \mathcal{T}, \mathcal{C}'\mathcal{D}'\mathbf{msg}(\mathbf{send}\ z_\gamma^c\ y_\alpha^c; y_\alpha^c \leftarrow u_\delta^c)\mathcal{F}$
where $\mathcal{T} \in \mathsf{Tree}_\Psi(\Delta'' \Vdash z_\gamma:A[c])$ and
$\mathcal{C}' \in \mathsf{Tree}_\Psi(\cdot \Vdash \Delta')$ and $\mathcal{C}'' \in \mathsf{Tree}_\Psi(\cdot \Vdash \Delta'')$

(8) $\mathcal{C}'\mathbf{msg}(\mathbf{send}\ z_\gamma^{c'}\ x_\beta^{c'}; x_\beta^{c'} \leftarrow u_\delta^{c'}), \mathcal{D}, \mathcal{F} \qquad \xrightarrow[\Delta, x_\beta:(A \otimes B)[c'] \Vdash K]{x_\beta^{c'}.\overline{\mathbf{msg}(\mathbf{send}\ z_\gamma^{c'}\ x_\beta^{c'}; x_\beta^{c'} \leftarrow u_\delta^{c'})}} \qquad \mathcal{C}', \mathbf{msg}(\mathbf{send}\ z_\gamma^{c'}\ x_\beta^{c'}; x_\beta^{c'} \leftarrow u_\delta^{c'})\mathcal{D}, \mathcal{F}$

(9) $\mathcal{C}, \mathcal{D}, \mathcal{T}\mathbf{msg}(\mathbf{send}\ z_\gamma^c\ y_\alpha^c; u_\delta^c \leftarrow y_\alpha^c)\mathcal{F} \qquad \xrightarrow[\Delta \Vdash y_\alpha:(A \multimap B)[c]]{y_\alpha^c.\overline{\mathbf{msg}(\mathbf{send}\ z_\gamma^c\ y_\alpha^c; u_\delta^c \leftarrow y_\alpha^c)}} \qquad \mathcal{C}\mathcal{T}, \mathcal{D}\mathbf{msg}(\mathbf{send}\ z_\gamma^c\ y_\alpha^c; u_\delta^c \leftarrow y_\alpha^c), \mathcal{F}$

(10) $\mathcal{C}'\mathcal{C}'', \mathcal{T}\mathbf{msg}(\mathbf{send}\ z_\gamma^{c'}\ x_\beta^{c'}; u_\delta^{c'} \leftarrow x_\beta^{c'})\mathcal{D}, \mathcal{F} \qquad \xrightarrow[\Delta', \Delta'', x_\beta:(A \multimap B)[c'] \Vdash K]{x_\beta^{c'}.\mathbf{msg}(\mathbf{send}\ z_\gamma^{c'}\ x_\beta^{c'}; u_\delta^{c'} \leftarrow x_\beta^{c'})} \qquad \mathcal{C}'\mathcal{C}''\mathcal{T}\mathbf{msg}(\mathbf{send}\ z_\gamma^{c'}\ x_\beta^{c'}; u_\delta^{c'} \leftarrow x_\beta^{c'}), \mathcal{D}, \mathcal{F}$
where $\mathcal{T} \in \mathsf{Tree}_\Psi(\Delta'' \Vdash z_\gamma:A[c'])$ and
$\mathcal{C}' \in \mathsf{Tree}_\Psi(\cdot \Vdash \Delta')$ and $\mathcal{C}'' \in \mathsf{Tree}_\Psi(\cdot \Vdash \Delta'')$

(10') $\mathcal{C}'\mathcal{C}'', \mathcal{T}\mathbf{msg}(\mathbf{send}\ z_\gamma^{c'}\ x_\beta^{c'}; u_\delta^{c'} \leftarrow x_\beta^{c'})\mathcal{D}, \mathcal{F} \qquad \xrightarrow[\Delta', \Delta'', x_\beta:(A \multimap B)[c'] \Vdash K]{x_\beta^{c'}.\mathbf{msg}(\mathbf{send}\ z_\gamma^{c'}\ x_\beta^{c'}; u_\delta^{c'} \leftarrow x_\beta^{c'})} \qquad \mathcal{C}'', \mathcal{T}, \mathcal{C}'\mathbf{msg}(\mathbf{send}\ z_\gamma^{c'}\ x_\beta^{c'}; u_\delta^{c'} \leftarrow x_\beta^{c'})\mathcal{D}'\mathcal{F}$
where $\mathcal{T} \in \mathsf{Tree}_\Psi(\Delta'' \Vdash z_\gamma:A[c'])$ and
$\mathcal{C}' \in \mathsf{Tree}_\Psi(\cdot \Vdash \Delta')$ and $\mathcal{C}'' \in \mathsf{Tree}_\Psi(\cdot \Vdash \Delta'')$

(11) $\mathcal{C}, \mathcal{D}'\mathbf{proc}(y_\alpha^c, y_\alpha^c \leftarrow x_\beta^c @d_1), \mathcal{F} \qquad \xrightarrow[\Delta \Vdash y_\alpha:A[c]]{y_\alpha^c.\{x_\beta^c/y_\alpha^c\}} \qquad [x_\beta^c/y_\alpha^c]\mathcal{C}, \mathcal{D}, \mathcal{F}$

(12) $\mathcal{C}'\mathbf{proc}(x_\beta^{c'}, x_\beta^{c'} \leftarrow z_\gamma^{c'} @d_1), \mathcal{D}_1, \mathcal{F}_1 \qquad \xrightarrow[\Delta, x_\beta:A[c'] \Vdash K]{x_\beta^{c'}.\overline{\{z_\gamma^{c'}/x_\beta^{c'}\}}} \qquad [z_\gamma^{c'}/x_\beta^{c'}]\mathcal{C}', \mathcal{D}, \mathcal{F}$

Figure 4: Definition of $\xRightarrow[\Delta \Vdash K]{x.q}$. An overline indicates that a message is sent from $\mathcal{C}$ or $\mathcal{F}$ to $\mathcal{D}$, otherwise the message is sent from $\mathcal{D}$ to $\mathcal{C}$ or $\mathcal{F}$.



\end{document}

%% file: sections/abstract.tex
Program \emph{equivalence} is the fulcrum for reasoning about and proving properties of programs.
For noninterference, for example, program equivalence up to the secrecy level of an observer is shown.
A powerful enabler for such proofs are \emph{logical relations}.
Logical relations only recently were adopted for session types---but exclusively for terminating languages.
This paper scales logical relations to \emph{recursive} session types.
It develops a logical relation for progress-sensitive noninterference for linear session types, tackling the challenges non-termination and concurrency pose.
The contributions include secrecy-polymorphic processes and the logical relation with metatheory.
A distinguishing feature is the choice of \emph{``step index''} of the logical relation, allowing for a natural proof of transitivity and soundness.

%% file: sections/1.intro.tex
\section{Introduction}

\emph{Logical relations} are a powerful and well-established proof method to show logical \emph{equivalence} of programs.
Asserting equivalence of two programs is at the heart of many program verification problems.
One such verification problem is \emph{noninterference}, guaranteeing that an adversary cannot infer any secrets from made observations.
Noninterference amounts to a program equivalence up to the secrecy level of the adversary.

While logical relations have been applied to a wide application domain---from functional languages to imperative languages with higher-order store---they have only recently been explored in a session-typed setting.
In this setting computation is carried out by \emph{concurrently} interacting \emph{processes} that exchange messages along channels.
Computation is inherently \emph{effectful} because message exchange enact state transitions in the involved processes.
Message exchange can also lead to a change in channel topology when channels are sent along channels.
Given their ability to prescribe change enacted by message exchange, session types \cite{HondaCONCUR1993, HondaESOP1998,CairesCONCUR2010,WadlerICFP2012} are prominent in message-passing concurrency.
Session types moreover enjoy a strong theoretical foundation by a Curry-Howard correspondence between \emph{linear logic} and the session-typed $\pi$-calculus~\cite{CairesCONCUR2010,WadlerICFP2012,KokkePOPL2019}.

Surprisingly, logical relations for session types have primarily focused on unary logical relations for proving termination~\cite{PerezESOP2012,PerezARTICLE2014,DeYoungFSCD2020}, except for a binary relation for
parametricity~\cite{CairesESOP2013} and noninterference \cite{DerakhshanLICS2021}, and have been \emph{exclusively} formulated for the \emph{terminating} linear session type fragment.

This paper develops a logical relation for program equivalence for \emph{recursive} session types.
Support of possibly non-terminating programs is necessary for many realistic applications, especially for emerging new platforms and computing models, such as cloud computing.
In these domains, noninterference is an important correctness concern.
To demonstrate practicality of our logical relation, we instantiate it for \emph{progress-sensitive noninterference} \cite{HeidinSabelfeldMartkoberdorf2011}.
A progress-sensitive formulation treats divergences as an observable outcome and thus provides a stronger guarantee in a concurrent setting where the \emph{termination channel} \cite{SabelfeldIEE2003} can be scaled to many parallel computations, each leaking ``just'' one bit~\cite{StefanICFP2012,AlpernasOOPSLA2018}.

To assert progress-sensitive noninterference by typing checking, we develop a flow-sensitive information flow control (IFC) type system for intuitionistic linear session types, yielding the language $\lang$.
A distinguishing feature of $\lang$ is its support of \emph{secrecy polymorphism}, allowing processes to be instantiated at varying secrecy levels.

Our meta theoretic results include, besides fundamental theorem and noninterference, the proof that our logical relation is an equivalence relation as well as soundness.
The former requires proof of transitivity, a notoriously challenging property to prove, particularly in the presence of recursion.
Recursive types mandate the use of a measure like step indexing or later modalities~\cite{AppelMcAllesterTOPLAS2001,AhmedESOP2006,DreyerLICS2009} to keep the logical relation well-founded.
A naive adoption of step indexing is too crude for program equivalence in a  message-passing setting.
Instead, we devise an index that uses the number of messages observed as an upper bound.
This choice paves the way for proving transitivity.
To ensure soundness of the logical relation, we prove that any two logically equivalent programs are indistinguishable by forming a bisimulation bounded by the number of observations between them.


Additional challenges that we had to address when developing the logical relation where posed by the fact that \textit{(i)} concurrency rules out lock-step execution of programs, complicating synchronization of the two programs to ensure relatedness, and that \textit{(ii)} high-secrecy code can include non-termination, complicating a progress-sensitive statement, which demands that both runs will either be observably productive or not.

In summary, this paper makes the following contributions:

\begin{itemize}

\item  flow-sensitive IFC type system for an intuitionistic linear session type language with \emph{recursive types} and \emph{secrecy-polymorphic} process definitions, yielding the language $\lang$ with proof of type safety;

\item \emph{recursive session logical relation} for noninterference bounded by the number of observable messages, amounting to the first logical relation for session types with recursive types;

\item statement and proof of logical equivalence, proving \emph{transitivity}, symmetry, and \emph{progress-sensitive noninterference} for well-typed $\lang$ programs;

\item statement and proof of adequacy theorem, asserting a \emph{weak stratified bisimulation} for logically equivalent programs closed with the same context.

\end{itemize}

%% file: sections/2.background.tex
\section{Background: IFC \& session types}\label{sec:2.background}

This section familiarizes with intuitionistic linear session types \cite{ToninhoESOP2013,ToninhoPhD2015,CairesARTICLE2016} and prior work on information flow control (IFC) in that setting \cite{DerakhshanLICS2021}.

\subsection{Motivating example}\label{sec:2.background-motivation}

We illustrate intuitionistic session-typed programming and IFC on a banking example.
Let's assume that our bank has two customers, Alice and Bob, with accounts each.
In a secure system, we expect that Alice's account can only be queried by Alice or the bank, but neither by Bob or any walk-in customer.
The same must hold for Bob's account.
We can express this policy by defining corresponding secrecy levels and a lattice on them:

\begin{center}
\begin{small}
$\mb{guest} \sqsubseteq \mb{alice} \sqsubseteq \mb{bank} \qquad \mb{guest} \sqsubseteq \mb{bob} \sqsubseteq \mb{bank}$
\end{small}
\end{center}

\Cref{fig:example} shows two run-time instances of our example implemented in an intuitionisticaly linear session-typed language.
\emph{Session types} prescribe the \emph{protocols} of message exchange between concurrently executing \emph{processes}.
A run-time program instance thus amounts to a configuration of processes connected by \emph{channels}.
The process configuration on the left consists of the {\small$\m{Bank}$} process, the customer processes {\small$\m{Alice}$} and {\small$\m{Bob}$}, along with their bank account processes {\small$\m{aAcc}$} and {\small$\m{bAcc}$}, resp., and corresponding processes for authentication.
The latter comprises, per customer, an authorization process ({\small$\m{\_Auth}$}), which provides access to a customer's bank account ({\small$\m{\_Acc}$}), given successful authentication by a verifier process ({\small$\m{\_Ver}$}) using the customer's PIN ({\small$\m{\_Pin}$}).
The process configuration on the right shows the configuration after Alice's authorization process has sent Alice's PIN to Alice's verifier ({\small$\mathbf{send}\,u\,x$}) for authentication.
The ability to pass channels along channels gives rise to \emph{effectful} computation, changing the connectivity structure and communication pathways of a configuration.

\begin{figure}
\centering
\includegraphics[scale=0.25]{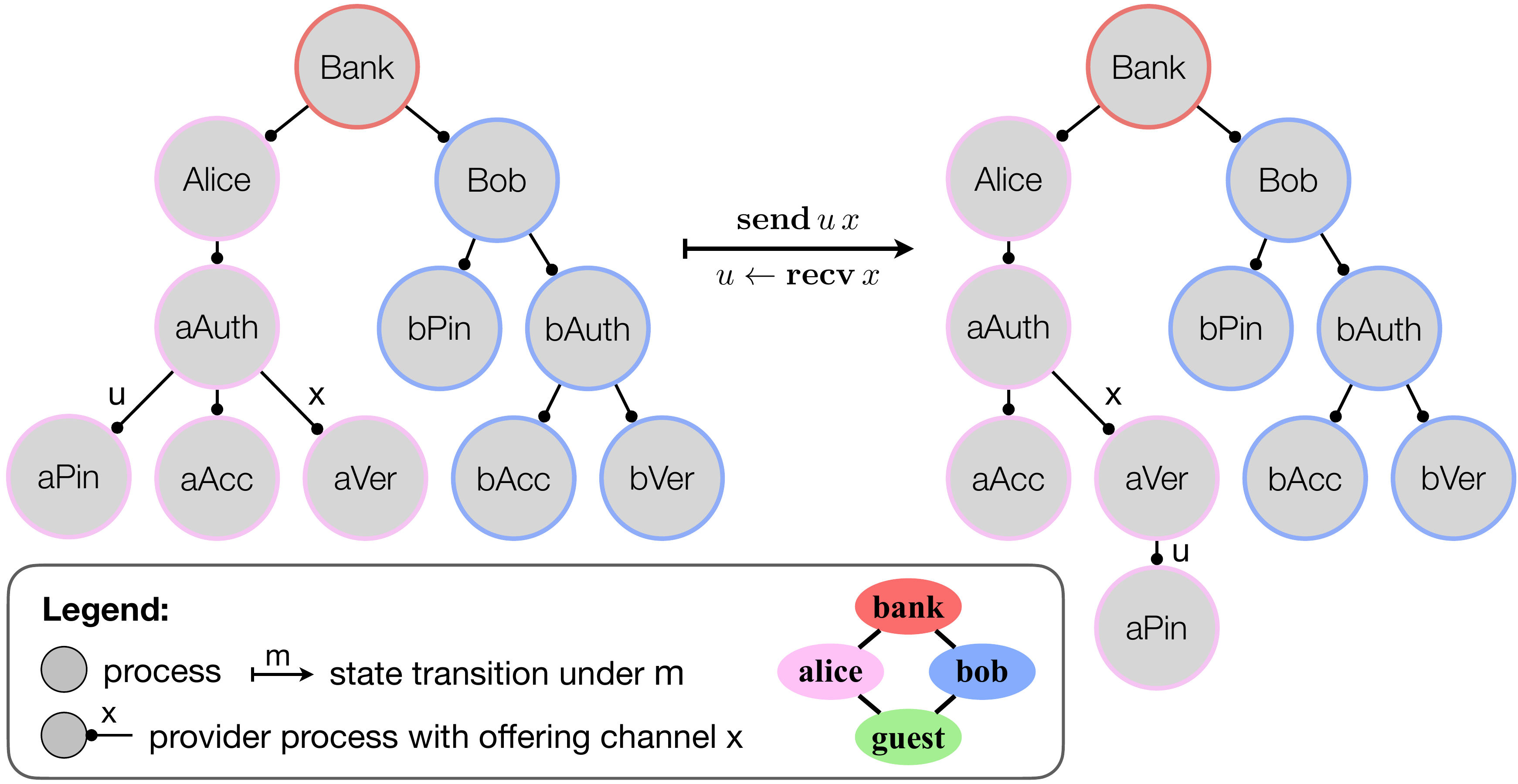}
\caption{State transitions among process configurations.}
\label{fig:example}
\end{figure}

To provide a flavor of session-typed programming, we give the session types that define the protocol governing the banking application.
We use the language $\sill$ \cite{ToninhoESOP2013,ToninhoPhD2015}, which is based on a Curry-Howard correspondence between intuitionistic linear logic and the session-typed $\pi$-calculus~\cite{CairesCONCUR2010,CairesARTICLE2016} and whose types and terms we summarize in \Cref{tab:session_types}:

\begin{center}
\begin{small}
\begin{minipage}[t]{\textwidth}
\begin{tabbing}
$\m{customer}$ \=$=$ \= $\m{pin} \multimap \m{auth}_{\m{out}} \multimap  1$ \\
$\m{pin}$ \> $=$ \>$\oplus \{\mathit{tok_1}{:}\m{pin}, \dots, \mathit{tok_n}{:} \m{pin}\}$\\ 
$\m{auth}_{\m{out}}$\> $=$ \> $\m{pin} \multimap \oplus \{$\= $\mathit{succ}{:} \m{account} \otimes \m{auth}_{\m{in}}, \, \mathit{fail}{:} \m{pin} \otimes \m{auth}_{\m{out}}\}$ \\
$\m{auth}_{\m{in}}$\> $=$ \> $\m{account} \multimap \m{pin} \otimes \m{auth}_{\m{out}}$\\
$\m{ver}$ \> $=$ \> $\m{pin} \multimap \oplus \{\mathit{succ}{:}\m{pin} \otimes \m{ver}, \mathit{fail}{:} \m{pin} \otimes \m{ver}\}$ \\
$\m{account}$ \> $=$ \> $\oplus\{\mathit{high}{:} \m{account},\,\mathit{med}{:}\m{account},\, \mathit{low}{:}\m{account}\}$ 
\end{tabbing}
\end{minipage}
\end{small}
\end{center}

\noindent Most of the above session type definitions are recursive, {\small$\m{auth}_{\m{out}}$} and {\small$\m{auth}_{\m{in}}$} are even mutually recursive.
$\sill$ adopts an \emph{equi-recursive}~\cite{CraryPLDI1999} interpretation, which avoids explicit (un)fold messages and relates types up to their unfolding.
Recursive session types must also be \emph{contractive}~\cite{GayHoleARTICLE2005}, demanding a message exchange before recurring.

We briefly comment on the overview of intuitionistic linear session types given in \Cref{tab:session_types}.
A crucial---and probably unusual---characteristics of session-typed processes is that a process \emph{changes} its type along with the messages it exchanges.
As a result, a process' type always reflects the current protocol state.
\Cref{tab:session_types} lists state transitions inflicted by a message exchange in the first and second column and corresponding process terms in the third and fourth column.
The fifth column provides the operational meaning of a type and the last column its polarity.
Positive types have a sending semantics, negative types a receiving semantics.

To guarantee type preservation (\aka session fidelity) in the presence of perpetual protocol change, session types use linearity.
Linearity ensures that a channel connects exactly two processes and thus imposes a \emph{tree} structure on the process configuration.
An intuitionistic viewpoint moreover interprets the process configuration as a \emph{rooted tree}, allowing one channel endpoint to become the \emph{provider} (child node) and the other to become its unique \emph{client} (parent node).
Channels are then typed with the session type of the providing process, requiring provider and client to behave dually: when one sends, the other receives.
In linear session types based on classical logic~\cite{WadlerICFP2012}, the two endpoints of a channel are typed separately using linear negation to enforce duality.

\input{tables/session_types}

The above session types dictate the following protocol:
When a customer opens an account with the bank, the bank sends them a PIN and authorization process ({\small$\m{customer}$}).
A customer can now repeatedly query their account's balance ({\small$\mi{high}$}, {\small$\mi{med}$}, {\small$\mi{low}$}), preceding authentication.
To authenticate themselves, a customer sends their PIN to the authorization process.
The authorization process then goes ahead and sends the PIN to the verifier.
If the authentication by the verifier is successful, the verifier sends the message {\small$\mi{succ}$} to the authorization process, otherwise the message {\small$\mi{fail}$}.
In either case, the verifier returns the PIN to the authorization process.
The latter relays the outcome of the authentication to the customer, yielding the bank account in case of {\small$\mi{succ}$} and the PIN in case of {\small$\mi{fail}$}.
Once done, a customer returns the account to the authorization process and receives their PIN, allowing them to start over.

Next we provide the process implementations {\small$\m{aVerifier}$} and {\small$\m{aPin}$} of types {\small$\m{ver}$} and {\small$\m{pin}$}, resp., for Alice:

\begin{flushleft}
\begin{small}
\begin{minipage}[t]{\textwidth}
\begin{tabbing}
$\cdot \vdash \m{aVerifier} :: x{:}\m{ver}$ \\
$x \leftarrow \m{aVerifier}\leftarrow \cdot = ($ \\
\quad \= $u \leftarrow \mb{recv}\, x$ \commentc{// $u{:}\m{pin} \vdash x{:}\oplus\{\mathit{succ}:\m{pin} \otimes \m{ver}, \mathit{fail}{:}\m{pin} \otimes \m{ver}\}$} \\
\> $\mb{case} \, u \, ($\= $\mathit{tok}_j \Rightarrow x.\mathit{succ}; \mb{send}\, u\, x;   x \leftarrow \m{aVerifier} \leftarrow \cdot$\\
\> \> $\mid \mathit{tok}_{i \neq j} \Rightarrow x.\mathit{fail}; \mb{send}\, u\, x;   x \leftarrow \m{aVerifier} \leftarrow \cdot))$ \\[4pt]

$\cdot \vdash \m{aPin} :: u{:}\m{pin}$ \\
$u \leftarrow \m{aPin}\leftarrow \cdot = (u.\mathit{tok_j}; u \leftarrow \m{aPin} \leftarrow \cdot)$
\end{tabbing}
\end{minipage}
\end{small}
\end{flushleft}

\noindent A process implementation consists of the process signature (first two lines) and body (after {\small$=$}).
The first line indicates the typing of channel variables used by the process (left of {\small$\vdash$}) and the type of the providing channel variable (right of {\small$\vdash$}).
The former are going to be child nodes of the process.
The second line binds the channel variables.
In $\sill$, {\small$\leftarrow$} generally denotes variable bindings.
Both {\small$\m{aVerifier}$} and {\small$\m{aPin}$} are leaves and thus do not have any variables occurring on the left of {\small$\vdash$}.
We leave it to the reader to convince themselves, consulting \Cref{tab:session_types}, that the code in the body of the two processes execute the protocol defined by their session types.

While session types ensure that messages are exchanged at run-time according to the protocol specified by the session type, they leave unconstrained the propagation of secret information.
For example, a villain verifier could send a customer's PIN to an attacker, amounting to a \emph{direct} leak of information.
More subtle forms of information leaks are \emph{indirect}, exploiting branching to convey secret information.
This strategy is taken by process {\small$\m{SneakyVerifier}$} defined below.
It signals to an attacker whether the customer authenticated themselves successfully or not by either sending the label {\small$\mathit{s}$} or {\small$\mathit{f}$}, resp.

\begin{flushleft}
\begin{small}
\begin{minipage}[t]{\textwidth}
\begin{tabbing}
$\m{attacker} = \& \{\mathit{s}{:}\m{attacker}, \mathit{f}{:} \m{attacker}\}$ \\[3pt]
$y{:}\m{attacker} \vdash \m{SneakyVerifier} :: x{:}\m{ver}$ \\
$x \leftarrow \m{SneakyVerifier}\leftarrow y = ( z \leftarrow \mb{recv}\, x$ \\
\quad \= $\mb{case} \, z \, ($\= $\mathit{tok}_j \Rightarrow x.\mathit{succ};y.\mathit{s}; \mb{send}\, z\, x;   x \leftarrow \m{SneakyVerifier} \leftarrow y$\\
\> \> $\mid \mathit{tok}_{i \neq j} \Rightarrow$ \=$x.\mathit{fail}; y.\mathit{f}; \mb{send}\, z\, x;$\\
\> \> \> $x \leftarrow \m{SneakyVerifier} \leftarrow y))$
\end{tabbing}
\end{minipage}
\end{small}
\end{flushleft}

While the above example leaks by virtue of sending an actual message, the \emph{absence} of a message due to non-termination can also be exploited for leaking confidential information.
Especially in a concurrent setting, this \emph{termination covert channel} \cite{SabelfeldIEE2003} can be scaled to many parallel computations, each leaking ``just'' one bit~\cite{StefanICFP2012,AlpernasOOPSLA2018}.
The interested reader can find examples illustrating these scenarios in Appendix~A.1.

\subsection{Possible worlds for IFC}\label{sec:2.background.treeinv}

IFC type systems for binary session types have only recently been contributed by \citet{DerakhshanLICS2021}.
\citeauthor{DerakhshanLICS2021} extend the \emph{terminating} subset of $\sill$ with a flow-sensitive IFC type system, yielding the language $\oldlang$.
Our work is based on $\oldlang$, so we briefly summarize its main features.

$\oldlang$ exploits the rooted tree structure resulting from intuitionism.
Observing that a process configuration naturally aligns with the security lattice (see \Cref{fig:example}), $\oldlang$ enforces this alignment by corresponding invariants preserved by typing.
To this end, the typing judgment is enriched with possible worlds from hybrid logic to denote a process' running secrecy {\small$\runsec{@c}$}, maximal secrecy {\small$\maxsec{[d]}$}, and the security lattice {\small$\Psi_0$} of the application:

\begin{center}
\begin{small}
$\Psi_0; \Delta \vdash P \runsec{@c}:: (x{:}A\maxsec{[d]})$
\end{small}
\end{center}

\noindent The idea is that a process' \emph{maximal secrecy} indicates the maximal level of secret information the process may ever obtain whereas its \emph{running secrecy} denotes the highest level of secret information the process has obtained so far.
We read the typing judgment as: \textit{``Process {\small$P$}, with maximal secrecy {\small$\maxsec{[d]}$} and running secrecy {\small$\runsec{@c}$}, provides a session of type {\small$A$} along channel variable {\small$x$}, given the typing of sessions offered along channel variables in {\small$\Delta$} and given the secrecy levels in the security lattice {\small$\langle \Psi_0, \sqsubseteq, \sqcup \rangle$}''}.
{\small$\Delta$} is a \emph{linear}
context that consists of a finite set of assumptions of the form {\small$y_i{:}B_i\maxsec{[d_i']}$},
indicating for each channel variable {\small$y_i$} its maximal secrecy {\small$\maxsec{[d_i']}$} and the
offered session type {\small$B_i$}.
The typing judgment uses channel variables, admitting {\small$\alpha$}-variance as usual.
Channel variables contrast with channels that are created at run-time upon process spawning.
When the distinction is clear from the context we refer to either as ``channel'' for brevity.

To enforce that a process configuration is in alignment with the security lattice, the following invariants are maintained by typing, referred to as the \emph{tree invariant}:
{\small \textit{(i)} $\forall y{:}B\maxsec{[d']} \in \Delta\,  (\Psi_0 \Vdash \maxsec{d'} \sqsubseteq \maxsec{d})$} and
{\small \textit{(ii)} $\Psi_0 \Vdash \runsec{c} \sqsubseteq \maxsec{d}$}
ensuring that the maximal secrecy of a child node is capped by the maximal secrecy of its parent and that the running secrecy of a process is less than or equal to its maximal secrecy, resp.






The tree invariant is crucial for preventing both direct and indirect flows but not sufficient to exclude them completely.
To ensure secure information flow, the type system needs to additionally ensure (1) that no channel of high secrecy is sent along one with lower secrecy and (2) that a process does not send any information along a low-secrecy channel after obtaining some high secrecy information.
Condition (1) prevents direct flows and is enforced by requiring that the maximal secrecy level of a carrier channel and the channels sent over it match.
Condition (2) is met by ensuring that a process' running secrecy is a sound approximation of the level of secret information the process has obtained so far.
To this end, the type system increases the running secrecy upon each receive to \textit{at least} the maximal secrecy of the sending process and correspondingly guards sends by making sure that the running secrecy of the sending process is \textit{at most} the maximal secrecy of the receiving process.

Our work extends $\oldlang$ with secrecy-polymorphic processes and recursive types, yielding the language $\lang$.
$\lang$ carries over $\oldlang$'s tree invariant and conditions (1) and (2).
An impatient reader may want to glance at \Cref{fig:type-system}, ignoring the rules for process definitions, to see how the tree invariant and conditions (1) and (2) are maintained.

%% file: tables/session_types.tex
\begin{table*}
\caption{Intuitionistic linear session types: types and terms (\textbf{b}efore and \textbf{a}fter exchange), operational meaning, and \textbf{p}olarity.}
\label{tab:session_types}
\begin{small}
\begin{tabular}{@{}lllllr@{}}
\toprule
\multicolumn{2}{@{}l}{\textbf{Session type (b/a)}} &
\multicolumn{2}{l}{\textbf{Process term (b/a)}} &
\textbf{Description} &
\textbf{p} \\[3pt]
$x {:} \oplus\{\ell{:}A\}_{\ell \in L}$ & $x {:} A_k$ & $x.k; P$ &
$P$ & provider sends label $k$ along $x$, continues with $P$ & $+$ \\
 & & $\mathbf{case}\, x (\ell {\Rightarrow} Q_{\ell})_{\ell \in L}$ & $Q_k$ &
\multicolumn{2}{l}{client receives label $k$ along $x$, continues with $Q_k$} \\[2pt]
$x {:} \&\{\ell{:}A\}_{\ell \in L}$ & $x {:} A_k$ & $\mathbf{case}\, x (\ell {\Rightarrow} P_{\ell})_{\ell \in L}$ &
$P_k$ & provider receives label $k$ along $x$, continues with $P_k$ & $-$  \\
 & & $x.k; Q$ & $Q$ & \multicolumn{2}{l}{client sends label $k$ along $x$, continues with $Q$}  \\[2pt]
$x {:} A \otimes B$ & $x {:} B$ & $\mathbf{send}\,y\,x; P$ &
$P$ & provider sends channel $y {:} A$ along $x$, continues with $P$  & $+$ \\
 & & $z {\leftarrow} \mathbf{recv}\,x;Q_z$ & ${Q_y}$ &
\multicolumn{2}{l}{client receives channel $y {:} A$ along $x$, continues with $Q_y$}  \\[2pt]
$x {:} A \multimap B$ & $x {:} B$ & $z {\leftarrow} \mathbf{recv}\,x;P_z$ & ${P_y}$ &
provider receives channel $y {:} A$ along $x$, continues with $P_y$   & $-$\\
 & & $\mathbf{send}\,y\,x; Q$ & $Q$ &
\multicolumn{2}{l}{client sends channel $y {:} A$ along $x$, continues with $Q$}  \\[2pt]
$x {:} 1$ & - & $\mathbf{close}\, x$ &
- & provider sends ``$\m{end}$'' along $x$ and terminates  & $+$ \\
 & & $\mathbf{wait}\,x;Q$ & $Q$ & \multicolumn{2}{l}{client receives ``$\m{end}$'' along $x$, continues with Q} \\
$x : A$ & $x : A$ & $\pseq{\pcut{z}{Q}}{P_z}$ & $P_z$ & \multicolumn{2}{l}{provider spawns new process running $Q$ along channel $z$, continues with $P_z$}  \\[2pt]
$x : A$ & - & $\pfwd{x}{y}$ & - & \multicolumn{2}{l}{provider forwards itself to process providing along channel $y {:} A$ and terminates} \\[2pt]
 $x {:} Y$ & $x{:}A$ & - &
- & \multicolumn{2}{l}{recursive and contractive session type definition $Y=A$ ($Y$ occurs in $A$)} \\
\bottomrule
\end{tabular}
\end{small}
\end{table*}

%% file: sections/3.recursion.tex
\section{IFC type system with recursion}\label{sec:3.recursion}

This section introduces $\lang$'s type system, configuration typing, and dynamics.
$\lang$ extends $\oldlang$ with recursive types and secrecy-polymorphic processes.
\Cref{sec:5.metatheory} shows that well-typed $\lang$ programs enjoy progress-sensitive noninterference, using the logical relation developed in \Cref{sec:4.logical-relation}.

\subsection{Polymorphic IFC type system}\label{sec:3.rec.type}

The support of recursive session types necessitates process definitions.
Process definitions in $\lang$ are \emph{polymorphic} in their secrecy levels, allowing a process to be instantiated at varying secrecy levels.
To ensure that instantiations respect the tree invariant and constraints imposed on the typing (see \Cref{sec:2.background.treeinv}), polymorphic process definitions come equipped with constraints on the secrecy variables.
The body of a process is then verified assuming that these constraints are met.
It is the obligation of the caller---including the process itself in case of a recursive call---to assert that the constraints are met for the provided instantiations.

\paragraph{Signature checking.}
 
Process definitions are collected in a global signature {\small$\Sig$}.
The signature is populated prior to type checking with all process definitions of the form 
{\small\[\Psi; x_1{:}A_1\maxsec{[\psi_1]}, \cdots, x_n{:}A_n\maxsec{[\psi_n]}\vdash X=P \runsec{@\psi'}:: x{:}B\maxsec{[\psi]}\]}
\noindent (rules {\small$\Sig_1$} and {\small$\Sig_3$} in \Cref{fig:type-system}).
The judgment indicates the process' name {\small$X$}, body {\small$P$}, maximal secrecy {\small$\maxsec{\psi}$}, and running secrecy {\small$\runsec{@\psi'}$} as well as the typing of its used and providing channels.
{\small$\Psi$} is the \emph{extended lattice} containing the security lattice {\small$\Psi_0$} and the constraints on secrecy variables required for verifying process term {\small$P$}. The secrecy labels {\small$\maxsec{\psi}, \runsec{\psi'}$ and $\maxsec{\psi_i}$} are all \emph{secrecy variables}, we use labels {\small$c,d, \dots$} to refer to both concrete secrecy levels and variables. 

\input{figs/type-system}

A signature {\small$\Sigma$} is well-typed over the security lattice {\small$\Psi_0$}, i.e., {\small${\Vdash_{\Sigma; \Psi_0}\Sigma \; \mathbf{sig}}$}, if the body {\small$P$} of every process in the signature is well-typed.
The programmer must define the extended lattice {\small$\Psi$} for each process definition such that its constraints over the concrete secrecy levels match precisely the constraints in the security lattice ({\small$\mathtt{concrete}(\Psi)=\Psi_0$}), i.e., for concrete secrecy levels {\small $c,d \in \Psi_0$}, we have {\small $\Psi_0 \Vdash c \sqsubseteq d\, \m{iff}\, \Psi \Vdash c \sqsubseteq d$}.
{\small$\Psi$} moreover must satisfy the tree invariant.

The signature {\small$\Sig$} additionally collects recursive session type definitions (rule $\Sig_2$).
As detailed in \Cref{sec:2.background-motivation}, recursive session types are \emph{equi-recursive}~\cite{CraryPLDI1999}, and \emph{contractive} ({\small $\Vdash_{\Sigma; \Psi_0} A \; \mb{wfmd}$}).
Equi-recursiveness avoids explicit (un)fold messages and relates types up to their unfolding.
Contractiveness demands an exchange before recurring.


\paragraph{Process typing.}

We slightly modify the process term typing judgment introduced in \Cref{sec:2.background.treeinv} to depend on the signature {\small$\Sigma$}, and the extended lattice {\small$\Psi$}:

\begin{center}
\begin{small}
$\Psi; \Delta \vdash_\Sigma P \runsec{@c}:: (x{:}A\maxsec{[d]})$
\end{small}
\end{center}

\noindent The corresponding typing rules are shown in \Cref{fig:type-system}.
The rules maintain the tree invariant and meet the conditions on sends and receives of \Cref{sec:2.background.treeinv}.
This can be achieved by assuming the tree invariant for the conclusion of each rule and asserting it for the premises.
Moreover, the rules that receive along a channel of maximal secrecy {\small$c$}, increase the running secrecy of the receiving process from {\small$d_1$} to {\small$d_1 \sqcup c$} and the rules that send along a channel of secrecy {\small$c$} are guarded with the condition that the running secrecy of the sending process is less than or equal to {\small$c$}.

\Cref{fig:type-system} does not have a special rule for recursive calls.
In line with prior work~\cite{ToninhoESOP2013,ToninhoPhD2015,CairesARTICLE2016,BalzerICFP2017}, we type check recursive calls using rules {\small$\msc{Spawn}$} and {\small$\msc{Fwd}$}.
Since recursive calls are tail calls, every recursive call {\small$y^{d}{:}B\leftarrow X \leftarrow \Delta_1$} can be translated into spawning {\small$X$} offering a fresh channel {\small$x^{d}{:}B$} and then terminating the continuation by identifying {\small$x$} with \small{$y$}, i.e.
{\small$(x^{d}{:}B\leftarrow X \leftarrow \Delta_1); y^d\leftarrow x^d$}.

Rule {\small$\msc{Spawn}$} relies on an order-preserving substitution mapping {\small$\gamma$} for secrecy variables in {\small$\Psi'$}.
The substitution {\small$\gamma$} identifies the specific instantiation provided by the caller to assert the constraints assumed by the callee in {\small$\Psi'$}. 
The premise {\small$d_1 \sqsubseteq d_2$} ensures that the spawner's level of secret information cannot surpass the one of the spawnee, ruling out indirect flows.
The premise {\small$d' \sqsubseteq d$} preserves the tree invariant.

\paragraph{Sneaky verifier - revisited.}

We briefly illustrate polymorphic processes on the $\m{SneakyVerifier}$ discussed in \Cref{sec:2.background-motivation} and show that it is not a well-typed process in $\lang$.
The interested reader can find further examples in Appendix~A.

\begin{flushleft}
\begin{small}
\begin{minipage}[t]{\textwidth}
\begin{tabbing}
$\m{attacker} = \& \{\mathit{s}{:}\m{attacker}, \mathit{f}{:} \m{attacker}\}$ \\[3pt]
$\Psi; y{:}\m{attacker}[\maxsec{\psi_1}]  \vdash \m{SneakyVerifier}\runsec{@\psi'} :: x{:}\m{ver}\maxsec{[\psi]}$ \\
$x \leftarrow \m{SneakyVerifier}\leftarrow y = ( z \leftarrow \mb{recv}\, x;$ \\
\;\;\= $\mb{case} \, z \, ($\= $\mathit{tok}_j \Rightarrow x.\mathit{succ};y.\mathit{s}; \mb{send}\, z\, x;   x \leftarrow \m{SneakyVerifier}@d_2 {\leftarrow} y$\\
\> \> $\mid \mathit{tok}_{i \neq j} \Rightarrow$ \=$x.\mathit{fail}; y.\mathit{f}; \mb{send}\, z\, x;$\\
\> \> \> $x \leftarrow \m{SneakyVerifier}@d_3 \leftarrow y))\runsec{@\mb{\psi'}}$
\end{tabbing}
\end{minipage}
\end{small}
\end{flushleft}

Assume that the extended lattice {\small$\Psi$} contains the constraints {\small$\psi_1=\mb{guest}$} and {\small$\psi=\mb{alice}$}. 
By {\small$\Sig_3$}, {\small$\Psi$} must satisfy {\small$\Psi \Vdash \mb{guest}=\psi_1 \sqsubset \mb{alice}=\psi$} and {\small$\Psi \Vdash \psi' \sqsubseteq \psi=\mb{alice}$}.
Moreover, by {\small$\& L$}, we must have {\small$\Psi \Vdash \psi' \sqsubseteq \psi_1$}, to execute {\small$y.\mathit{s}$}.
By {\small$\multimap R$}, we know that the running secrecy of the process after executing {\small$z \leftarrow \mb{recv}\,x^\psi$} increases to {\small$\psi$}.
And by {\small$\msc{Spawn}$}, we need to know that the running secrecy {\small$d_2$} specified for the callee is at least as high as the caller's, i.e. {\small$\Psi \Vdash \psi \sqsubseteq d_2$}.
It is straightforward to observe that there is no possible substitution {\small$\gamma$} to unify {\small$d_2$} with {\small$\psi'$} as required by the spawn rule, i.e. {\small$\hat{\gamma}(\Psi) \not \Vdash d_2 \sqsubseteq \psi_1 = \mb{guest} \sqsubset  \mb{alice}= \psi \sqsubseteq d_2$}.

\subsection{Asynchronous dynamics}

We define an asynchronous dynamics for $\lang$ that we adopt from prior work~\cite{BalzerICFP2017,DasCSF2021,DerakhshanLICS2021} and extend to account for polymorphic process definitions.
The dynamics is given in \Cref{fig:dynamics} and uses multiset rewriting rules \cite{CervesatoARTICLE2009}.
Multiset rewriting rules express the dynamics as state transitions between configurations and are \emph{local} in that they only mention the parts of a configuration they rewrite.

\input{figs/dynamics}

The rules in \Cref{fig:dynamics} achieve non-blocking sends by spawning off outputted messages as special message processes of the form {\small$\mb{msg}(\_; y_\alpha \leftarrow y_{\alpha+1})$} or {\small$\mb{msg}(\_; y_{\alpha+1} \leftarrow y_\alpha)$} where {\small$\_$} stands for the message being sent.
For example, in case of {\small$\otimes_{\m{snd}}$}, the message is {\small$\mathbf{send}\,x_\beta\,y_\alpha$}, indicating that the channel {\small$x_\beta$} is sent over channel {\small$y_\alpha$}.
Each spawning of such a message results in the allocation of a new generation {\small$y_{\alpha+1}$} of the carrier channel {\small$y_\alpha$}.
To guarantee proper sequencing between the message and the sender's continuation, a forward is included in the message that links the two generations {\small$y_{\alpha+1}$} and {\small$y_\alpha$} appropriately.
The direction of the forward is determined by the polarity of the session type.

Rule {\small$\msc{Spawn}$} relies on the substitution mapping {\small$\gamma$} given by the typing derivation and its lifting {\small$\hat{\gamma}$} to the process term level.
It looks up the definition of process {\small$X$} in the signature and instantiates the secrecy variables occurring in the process body using {\small$\gamma$}.
The condition {\small$\hat{\gamma}(\Psi')=\Psi_0$} ensures that all secrecy variables are being instantiated with a concrete secrecy level.
The rule also allocates a fresh channel {\small$x_0$} at generation 0, which is substituted for {\small$x$}.

\subsection{Configuration typing.}

Similar to~\cite{DerakhshanLICS2021}, we allow configurations of processes to be \emph{open}.
An open configuration can be \emph{plugged} with other process configurations to form a closed program.
Open configurations are typed with the judgment {\small$\Psi_0;\Delta\Vdash \mc{C}:: \Delta'$}, indicating that the configuration {\small$\mc{C}$} provides sessions along the channels in {\small$\Delta'$}, using sessions provided along channels in {\small$\Delta$} and given the concrete security lattice {\small$\Psi_0$}.
{\small$\Delta$} and {\small$\Delta'$} are both linear contexts, consisting of actual run-time channels of the form {\small$y_i{:}B_i\maxsec{[d_i']}$}.
For brevity, we do not indicate a channel's generation.
The typing rules are shown below.

{\small
\[
\inferrule*[right=$\mathbf{emp}_1$]
{\strut}
{\Psi_0; x{:}A[d] \Vdash \cdot :: (x{:}A[d])}
\qquad
\inferrule*[right=$\mathbf{emp}_2$]
{\strut}
{\Psi_0; \cdot \Vdash \cdot :: (\cdot)}
\]
\[
\inferrule*[right=$\mathbf{proc}$]
{\Psi_0 \Vdash d_1 \sqsubseteq d \\
\forall y{:}B[d'] \in \Delta'_0, \Delta \, (\Psi_0 \Vdash d' \sqsubseteq d) \\\\
\Psi_0; \Delta_0 \Vdash \mathcal{C}:: \Delta \\
\Psi_0; \Delta'_0, \Delta \vdash P@d_1:: (x{:}A[d])}
{\Psi_0; \Delta_0, \Delta'_0 \Vdash \mathcal{C}, \mathbf{proc}(x[d], P@d_1):: (x{:}A[d])}
\]
\[
\inferrule*[right=$\mathbf{msg}$]
{\forall y{:}B[d'] \in \Delta'_0, \Delta \, (\Psi_0 \Vdash d' \sqsubseteq d)\\
\Psi_0; \Delta_0 \Vdash \mathcal{C}:: \Delta \\
\Psi_0; \Delta'_0, \Delta \vdash P@d:: (x{:}A[d])}
{\Psi_0; \Delta_0, \Delta'_0 \Vdash \mathcal{C}, \mathbf{msg}(P):: (x{:}A[d])}
\]
\[
\inferrule*[right=$\mathbf{comp}$]
{\Psi_0; \Delta_0 \Vdash \mathcal{C}:: \Delta \\
\Psi_0; \Delta'_0 \Vdash \mc{C}_1:: x{:}A[d]}
{\Psi_0; \Delta_0, \Delta'_0 \Vdash \mathcal{C}, \mc{C}_1 :: \Delta, x{:} A[d]}
\]}

The dynamics steps an open configuration {\small$\Psi_0;\Delta\Vdash \mc{C}:: \Delta'$} to {\small$\Psi;\Delta\Vdash \mc{C}':: \Delta'$} by applying the rewriting {\small$\mc{C}_1 \ostepp \mc{C}_2$} to {\small$\mc{C}$}. 
Transitions happen entirely within {\small$\mc{C}$}, which is conveyed by the fact that the used and provided channels {\small$\Delta$} and {\small$\Delta'$} remain unchanged.
This is also reflected by the subscript {\small$\Delta \Vdash \Delta'$} of {\small$\mapsto_{\Delta \Vdash \Delta'}$}.

%% file: figs/type-system.tex
\begin{figure*}
\begin{center}
\begin{small}
\begin{mathpar}
\inferrule*[right=$\oplus L$]
{\Psi \Vdash d_2=c \sqcup d_1 \\
\Psi;  \Delta, x:A_{k}[c] \vdash  Q_k@d_2::  y:C[c'] \quad
\forall k \in L}
{\Psi; \Delta, x:\oplus\{\ell:A_{\ell}\}_{\ell \in L}[c] \vdash (\mathbf{case}\, x(\ell \Rightarrow Q_\ell)_{\ell \in L})@d_1::  y:C[c']}

\inferrule*[right=$\oplus R$]
{\Psi; \Delta \vdash P@d_1::  y:A_{k}[c] \\
k \in L}
{\Psi;\Delta \vdash (y.k; P)@d_1 ::  y:\oplus\{\ell:A_{\ell}\}_{\ell \in L}[c]}

\inferrule*[right=$\& L$]
{\Psi \Vdash d_1 \sqsubseteq c \\
\Psi;  \Delta, x:A_{k}[c] \vdash  P@d_1::  y:C[c'] \\
k \in L}
{\Psi; \Delta, x:\&\{\ell:A_{\ell}\}_{\ell \in I}[c] \vdash (x.k; P)@d_1::  y:C[c']}

\inferrule*[right=$\& R$]
{\Psi; \Delta \vdash Q_k@c::  y:A_{k}[c] \\
\forall k \in L}
{\Psi;\Delta \vdash (\mathbf{case}\, y(\ell \Rightarrow Q_\ell)_{\ell \in I})@d_1::  y:\&\{\ell:A_{\ell}\}_{\ell \in L}[c]}

\inferrule*[right=$\otimes L$]
{\Psi \Vdash d_2= c\sqcup d_1 \\
\Psi;\Delta, z:A[c], x:B[c] \vdash  P@d_2:: y: C[c']}
{\Psi;\Delta, x:A \otimes B[c] \vdash (z\leftarrow \mathbf{recv}\, x; P)@d_1::  y:C[c']}

\inferrule*[right=$\otimes R$]
{\Psi;\Delta \vdash  P@d_1:: y:B[c]}
{\Psi;\Delta, z:A[c]\vdash (\mathbf{send}\,z\,y; P)@d_1::  y:A \otimes B[c]}

\inferrule*[right=$\chanin L$]
{\Psi \Vdash d_1 \sqsubseteq d \\
\Psi; \Delta, x:B[d] \vdash  P@d_1:: y:C[c']}
{\Psi;\Delta, z:A[d], x:A \multimap B[d] \vdash (\mathbf{send}\, z\,x; P)@d_1::  y:C[c']}

\inferrule*[right=$\chanin R$]
{\Psi;\Delta,  z:A[c] \vdash  P@c:: y: B[c]}
{\Psi; \Delta \vdash (z \leftarrow \mathbf{recv}\,y; P)@d_1::  y:A \multimap B[c]}

\inferrule*[right=$1 L$]
{\Psi; \Delta \vdash_{\Sigma} Q@d_2 ::  y:T[d] \\
\Psi \Vdash d_2=c \sqcup d_1}
{\Psi;\Delta, x:1[c] \vdash_{\Sigma} (\mathbf{wait}\,x;Q)@d_1 ::  y:T[d]}

\inferrule*[right=$1 R$]
{\strut}
{\Psi;\cdot \vdash_{\Sigma} (\mathbf{close}\,y)@d_1 ::  y:1[c]}

\inferrule*[right=\msc{Fwd}]
{\strut}
{\Psi; x:T[d] \vdash_{\Sigma} (y \leftarrow x)@{c} ::  y:T[d]}

\inferrule*[right=$\Sig_1$]
{\strut}
{\Vdash_{\Sigma; \Psi_0} (\cdot)\; \mathbf{sig}}

\inferrule*[right=$\Sig_2$]
{\Vdash_{\Sigma; \Psi_0} A \; \mb{wfmd} \\
\Vdash_{\Sigma; \Psi_0} \Sigma'\; \mathbf{sig}}
{\Vdash_{\Sigma; \Psi_0} Y = A, \Sigma' \; \mathbf{sig}}

\inferrule*[right=$\Sig_3$]
{\mathtt{concrete}(\Psi) = \Psi_0 \\
\Psi \Vdash \psi' \sqsubseteq \psi, \forall i. \psi_i \sqsubseteq \psi \\
\Psi; x_1:A_1[\psi_1], \dots, x_n:A_n[\psi_n] \vdash_{\Sigma} P @ \psi':: x:B[\psi] \\
\Vdash_{\Sigma; \Psi_0} \Sigma'\; \mathbf{sig}}
{\Vdash_{\Sigma; \Psi_0} \Psi; x_1:A_1[\psi_1], \dots, x_n:A_n[\psi_n]\vdash X=P @ \psi':: x:B[\psi], \Sigma'\; \mathbf{sig}}

\inferrule*[right=\msc{Spawn}]
{\hat{\gamma}(\Delta'_1; B[\psi];\psi')=\Delta_1; B[d'];d_2 \\
\Psi \Vdash \hat{\gamma}(\Psi') \\
\Psi'; \Delta'_1 \vdash X = P @{\psi'} ::x: B[\psi] \in \Sigma \\\\
\Psi \Vdash d_1 \sqsubseteq d_2, d' \sqsubseteq d \\
\Psi; x:B[d'], \Delta_2 \vdash_{\Sigma} Q@{d_1} ::  y:T[d]}
{\Psi;\Delta_1, \Delta_2 \vdash_{\Sigma}  (x^{d'} \leftarrow X@d_2 \leftarrow \Delta_1; Q)@{d_1} ::  y:T[d]}
\end{mathpar}
\end{small}
\end{center}
    \caption{Process term typing rules of $\lang$.}
    \label{fig:type-system}
\end{figure*}

%% file: figs/dynamics.tex
\begin{figure*}
\begin{center}
\begin{small}
\begin{tabbing}
$\msc{Fwd} \qquad$ \= $\mb{proc}(y_\alpha[c], (y_\alpha \leftarrow x_\beta)@d_1) \ostep [x_\beta/y_\alpha]$  \` $(y_\alpha \not \in \Delta')$ \\

$\msc{Spawn}$ \> $\mb{proc}(y_\alpha[c],  (x^d \leftarrow X \leftarrow \D)@d_2 ; Q@d_1)  \ostep$ \` $(\Psi'; \D' \vdash X = P @{\psi'} ::x: B'[\psi] \in \Sigma)$ \\
\> $\mb{proc}(x_0[d], (\subst{x_0, \D}{x, \m{dom}(\D')} \hat{\gamma}(P))@d_2) \; \mb{proc}(y_\alpha[c],  (\subst{x_0}{x}{Q})@d_1)$ \` $(\hat{\gamma}(\Psi')=\Psi_0\, \text{and}\, x_0\, \mi{fresh})$ \\

$1_{\m{snd}}$ \> $\mb{proc}(y_\alpha[c],(\mathbf{close}\, y_\alpha)@d_1) \ostep \mb{msg}(\mathbf{close}\, y_\alpha)$ \\ 

$1_{\m{rcv}}$ \> $\mb{msg}(\mathbf{close}\, y_\alpha)\; \mb{proc}(x_\beta[{c'}], (\mathbf{wait}\,y_\alpha;Q)@d_1) \ostep \mb{proc}(x_\beta[{c'}], Q@(d_1 \sqcup c))$ \\ 

$\oplus_{\m{snd}}$ \>$\mb{proc}(y_\alpha[c], y_\alpha.k; P@d_1) \ostep \mb{proc}(y_{\alpha+1}[c], ([y_{\alpha+1}/y_{\alpha}]P)@d_1)\; \mb{msg}( y_\alpha.k;y_\alpha \leftarrow y_{\alpha+1})$ \\

$\oplus_{\m{rcv}}$ \> $\mb{msg}(y_\alpha[{c}].k; y_\alpha \leftarrow v_{\delta})) \; \mb{proc}(u_\gamma[{c'}],\mb{case}\, y_\alpha (( \ell \Rightarrow P_\ell)_{\ell \in L})@d_1) \ostep \mb{proc}(u_{\gamma}[{c'}], ([v_{\delta}/y_{\alpha}] P_k)@(d_1\sqcup c))$ \\

$\&_{\m{snd}}$ \> $\mb{proc}(y_\alpha[c], (x_\beta.k; P)@d_1) \ostep \mb{msg}( x_\beta.k; x_{\beta+1} \leftarrow x_\beta)\; \mb{proc}(y_{\alpha}[c], ([x_{\beta+1}/x_{\beta}]P)@d_1)$ \\ 

$\&_{\m{rcv}}$ \> $\mb{proc}(y_\alpha[c],(\mb{case}\, y_\alpha ( \ell \Rightarrow P_\ell)_{\ell \in L})@d_1) \; \mb{msg}(y_\alpha.k; v_\delta \leftarrow y_\alpha) \ostep \mb{proc}(v_{\delta}[c],([v_{\delta}/y_{\alpha}] P_k)@c)$ \\

$\otimes_{\m{snd}}$ \> $\mb{proc}(y_\alpha[c],(\mathbf{send}\,x_\beta\,y_\alpha; P)@d_1) \ostep \mb{proc}(y_{\alpha+1}[c], ([y_{\alpha+1}/y_{\alpha}]P)@d_1)\; \mb{msg}( \mathbf{send}\,x_\beta\,y_\alpha; y_\alpha \leftarrow y_{\alpha+1})$ \\

$\otimes_{\m{rcv}}$ \> $\mb{msg}(\mathbf{send}\,x_\beta\,y_\alpha;y_{\alpha}\leftarrow v_\delta) \; \mb{proc}(u_\gamma[{c'}],(w_\eta\leftarrow \mathbf{recv}\,y_\alpha; P)@d_1) \ostep \mb{proc}(u_{\gamma}[{c'}], ([x_\beta/w_\eta][v_{\delta}/y_{\alpha}] P)@(d_1\sqcup c))$ \\

$\multimap_{\m{snd}}$ \> $\mb{proc}(y_\alpha[c],(\mathbf{send}\,x_\beta\,u_\gamma; P)@d_1) \ostep \mb{msg}( \mathbf{send}\,x_\beta\,u_\gamma; u_{\gamma+1}\leftarrow u_\gamma)\; \mb{proc}(y_{\alpha}[c], ([u_{\gamma+1}/u_{\gamma}]P)@d_1)$ \\

$\multimap_{\m{rcv}}$ \> $\mb{proc}(y_\alpha[c],(w_\eta\leftarrow \mathbf{recv}\,y_\alpha; P)@d_1) \; \mb{msg}(\mathbf{send}\,x_\beta\,y_\alpha;v_{\delta}\,\leftarrow y_\alpha) \ostep \mb{proc}(v_\delta[c], ([x_\beta/w_\eta][v_{\delta}/y_{\alpha}] P)@ c)$
\end{tabbing}
\caption{Asynchronous dynamics of $\lang$.}
\label{fig:dynamics}
\end{small}
\end{center}
\end{figure*}

%% file: sections/4.logical_relation.tex
\section{Noninterference recursive logical relation}\label{sec:4.logical-relation}

This section develops the logical relation used to prove progress-sensitive noninterference for $\lang$.
Transitivity and adequacy of the relation are shown in \Cref{sec:5.metatheory}.

\subsection{Challenges}\label{sec:4.logical-relation:challenges}

\emph{Noninterference} essentially amounts to a \emph{program equivalence up to} the secrecy level {\small$\xi$} of the observer, requiring that two runs of a program can only differ in outcomes of secrecy level higher than or incomparable to {\small$\xi$}.
In a functional setting, the fundamental theorem of the \emph{logical relation} for equivalence is stated for two open expressions, showing that they evaluate to related values, if given related substitutions.
In a session type setting, on the other hand, the fundamental theorem is stated by \emph{plugging} an open configuration of processes (``open programs'') with closing configurations, showing that the programs send related messages, if receiving related messages from the configurations as well~\cite{DerakhshanLICS2021}.

The \emph{session logical relation} for noninterference developed by \citet{DerakhshanLICS2021}, while serving as a baseline for our development, does not scale to our setting.
It is phrased for the terminating language $\oldlang$ and its statement intimately relies on this assertion.
To support \emph{recursive types}, and thus \emph{non-terminating} programs, we had to address the following challenges:

\begin{itemize}

\item \emph{Well-foundedness:} Recursive types mandate use of a measure like step-indexing or later modalities~\cite{AppelMcAllesterTOPLAS2001,AhmedESOP2006,DreyerLICS2009} to keep the logical relation well-founded.
A naive adoption, typically tied to the number of unfoldings or computation steps, however, is too crude for a progress-sensitive statement of noninterference for message passing and challenge transitivity.

\item \emph{Nondeterminism:} While linear session-typed programs enjoy confluence, processes run concurrently.
As a result, two runs of a program cannot be aligned in lockstep, and the logical relation has to accommodate messages that are not being sent simultaneously while ensuring their relatedness.

\item \emph{Non-termination:} Noninterference allows two program runs to execute different code, as long as the same observable outcomes are produced.
In the presence of recursive types, any such high-secrecy code includes divergence for either of the two runs.
This discrepancy complicates a progress-sensitive statement, which demands that both runs will either diverge or not.

\end{itemize}

\subsection{Logical relation}\label{sec:4.logical-relation:logical-relation}

In line with classical logical relation developments for functional and imperative languages, we phrase our logical relation in terms of a value {\small$\mc{V} \llbracket \tau \rrbracket$} and term {\small$\mc{E} \llbracket \tau \rrbracket$} interpretation at a type {\small$\tau$}.
These notions have to be transposed to the session-typed, message-passing setting.
Prior work~\cite{PerezESOP2012,CairesESOP2013,PerezARTICLE2014,DeYoungFSCD2020,DerakhshanLICS2021} has done that for terminating session type languages.
We develop a logical relation for \emph{recursive} session types.
Using \Cref{fig:slrni-schema} as a guide, we first discuss the basic schema of the logical relation and then address the challenges identified above.

\begin{figure*}
    \centering
    \includegraphics[width=0.75\textwidth]{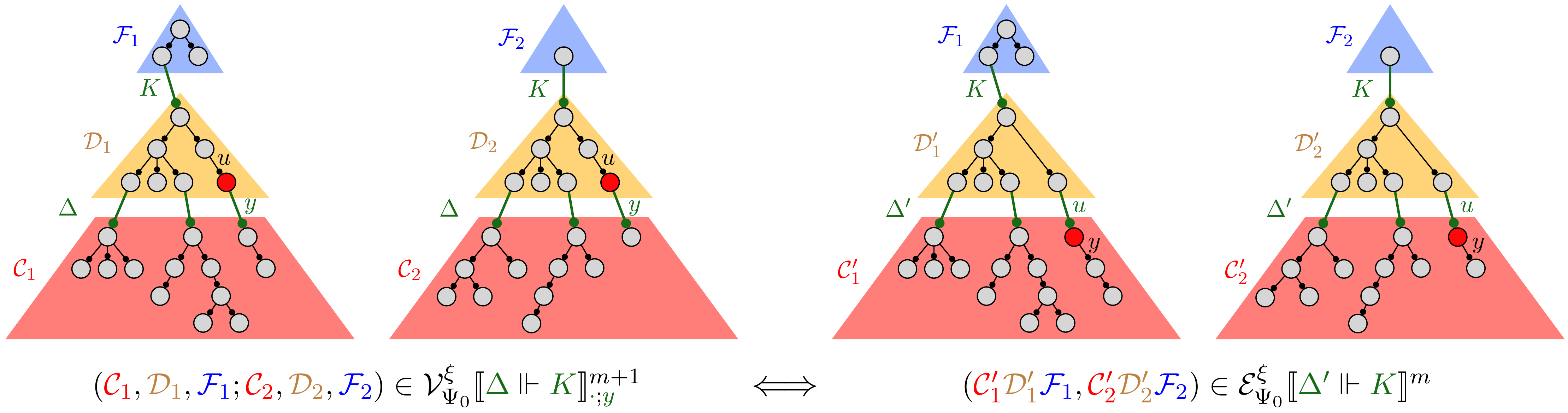}
    \caption{Schema of recursive session logical relation.}
    \label{fig:slrni-schema}
\end{figure*}

\subsubsection{Schema}

\Cref{fig:slrni-schema} distills the schema underlying the value interpretation of our logical relation, serving as a blueprint for the actual value interpretation defined in \Cref{fig:rel_value_right} and \Cref{fig:rel_value_left}.
On the left, it depicts two open programs {\small$\mnode_1$} and {\small$\mnode_2$}---the two runs of the program under consideration---with closing configurations (``substitutions'') {\small$\substc_1$}, {\small$\substf_1$} and {\small$\substc_2$}, {\small$\substf_2$}, \resp being related by the value interpretation.
This is formally stated as

\begin{center}
\begin{small}
$(\lr{\mc{C}_1} {\mc{D}_1} {\mc{F}_1};\lr{\mc{C}_2} {\mc{D}_2} {\mc{F}_2}) \in \mc{V}^\xi_{\Psi_0}\llbracket {\color{mygreen}\Delta} \Vdash {\color{mygreen}K}\rrbracket _{{\color{mygreen}\cdot}; {\color{mygreen}y}}^{m+1}$
\end{small}
\end{center}

\noindent where {\small${\color{mygreen}\Delta}$} amounts to the typing of channels connecting {\small$\mnode_1$} and {\small$\mnode_2$} with {\small$\substc_1$} and {\small$\substc_2$}, \resp and {\small${\color{mygreen}K}$} to the typing of the channel connecting {\small$\mnode_1$} and {\small$\mnode_2$} with {\small$\substf_1$} and {\small$\substf_2$}, resp.
We refer to {\small${\color{mygreen}\Delta} \Vdash {\color{mygreen}K}$} as the \emph{interface} of {\small$\mnode_1$} and {\small$\mnode_2$}, comprising the channels along which the two program runs interact with their closing configurations---including a potential attacker.
Superscript {\small$\xi$} and subscript {\small$\Psi_0$} of {\small$\mc{V}$} denote the secrecy level of the observer and the globally fixed security lattice, resp.
We defer an explanation of superscript {\small$m+1$} and subscript {\small${\color{mygreen}\cdot}; {\color{mygreen}y}$} of the index {\small$\llbracket{\color{mygreen}\Delta} \Vdash {\color{mygreen}K}\rrbracket$} until we address well-foundedness and nondeterminism, resp.

What are \emph{values} in a functional or imperative setting are open configurations in a session type setting that are \emph{ready to send} a message along their external interface.
For example, program runs {\small$\mnode_1$} and {\small$\mnode_2$} have already spawned off two messages, depicted in red in \Cref{fig:slrni-schema}, ready to be transmitted over interface channel {\small${\color{mygreen}y}$}.
The logical relation now mandates that the messages are related and pushes them into the substitutions {\small$\substc_1$} and {\small$\substc_2$}, yielding {\small$\mnode_1'$}, {\small$\mnode_2'$}, {\small$\substc_1'$}, and {\small$\substc_2'$} shown on the right in \Cref{fig:slrni-schema}.
{\small$\mnode_1'$}, {\small$\mnode_2'$}, {\small$\substc_1'$}, and {\small$\substc_2'$} can now each step internally, for example to consume the incoming messages, requiring them to be in the term interpretation

\begin{center}
\begin{small}
$(\lre{\mc{C}_1'} {\mc{D}_1'} {\mc{F}_1},\lre{\mc{C}_2'} {\mc{D}_2'} {\mc{F}_2}) \in \mc{E}^\xi_{ \Psi_0}\llbracket {\color{mygreen}\Delta'} \Vdash {\color{mygreen}K}\rrbracket ^{m}$
\end{small}
\end{center}

\noindent We defer a detailed discussion of the term interpretation defined in \Cref{fig:rel_term} for now but note the change of the interface from {\small${\color{mygreen}\Delta} \Vdash {\color{mygreen}K}$} to {\small${\color{mygreen}\Delta'} \Vdash {\color{mygreen}K}$} due to pushing the message into the substitution, replacing {\small${\color{mygreen}y}$} with the new interface channel {\small${\color{mygreen}u}$}.

The value interpretation of the logical relation accounts for every message sent from program runs {\small$\mnode_1$} and {\small$\mnode_2$} to substitutions {\small$\substc_1$}, {\small$\substf_1$} and {\small$\substc_2$}, {\small$\substf_2$}, \resp and vice versa.
This amounts to two cases per connective, one for a message exchanged along a channel in {\small${\color{mygreen}K}$} and one for a message exchanged along a channel in {\small${\color{mygreen}\Delta}$}.
We refer to the former as communications to the \emph{right} and the latter as communications to the \emph{left}.
Right communications are defined in \Cref{fig:rel_value_right}, left communications in \Cref{fig:rel_value_left}.
The predicate {\small$(\lr{\mc{C}_1}{\mc{D}_1}{\mc{F}_1}; \lr{\mc{C}_2}{\mc{D}_2}{\mc{F}_2})\in \m{Tree}_{\Psi_0}({\color{mygreen}\Delta} \Vdash {\color{mygreen}K})$} indicates that the runs {\small$\mnode_1$} and {\small$\mnode_2$} are well-typed and have matching well-typed substitutions {\small$\substc_1$}, {\small$\substf_1$} and {\small$\substc_2$}, {\small$\substf_2$}, resp.

The definitions showcase \emph{extensionality} of the logical relation: \emph{positive} connectives \emph{assert} the sending of related messages when communicating on the \emph{right} and \emph{assume} receipt of related messages when communicating on the \emph{left}.
Conversely, \emph{negative} connectives \emph{assume} receipt of related messages when communicating on the \emph{right} and \emph{assert} the sending of related messages when communicating on the \emph{left}.
We note that \emph{relatedness} of messages means for ground types ({\small$1$}, {\small$\intchoice$}, and {\small$\extchoice$}) that the same messages are being exchanged.
For higher-order types ({\small$\chanout$} and {\small$\chanin$}), however, whose messages transport entire subtrees, relatedness is more subtle.
In particular when the subtrees {\small$\mtree_1$} and {\small$\mtree_2$} originate from the program runs {\small$\mnode_1$} and {\small$\mnode_2$} (case {\small$\chanout$} in \Cref{fig:rel_value_right} and {\small$\chanin$} in \Cref{fig:rel_value_left}), \resp {\small$\mtree_1$} and {\small$\mtree_2$} must be shown to be related by the term interpretation of the logical relation.
\Cref{fig:rel_value_right} and \Cref{fig:rel_value_left} also include cases for forward, even though forwards do not result in a message exchange.
These cases make sure to update the names of channels in the interface due to a forward, guaranteeing that they are the same in both program runs.
Another subtle point to note is the use of {\small$\alpha$}-variance to match up continuation channels of messages and root channels of sent subtrees, ensuring that subsequent messages produced along those channels continue to be related.
For {\small$\chanout$} in \Cref{fig:rel_value_right}, for example, the value interpretation insists that the continuation channel {\small$u^c_\delta$} and root channel {\small$x^c_\beta$} are the same in both messages.
The superscript $c$ of a channel denotes its maximal secrecy.  While not required, we include it for convenience in our formal development.

\input{figs/rel_value_right}

\input{figs/rel_value_left}

We remark that the logical relation is defined for program runs {\small$\Psi_0; {\color{mygreen}\Delta} \Vdash {\color{brown}\mc{D}_1} :: {\color{mygreen}K}$} and {\small$\Psi_0; {\color{mygreen}\Delta} \Vdash {\color{brown}\mc{D}_2} :: {\color{mygreen}K}$} with corresponding substitutions {\small$\Psi_0; \cdot \Vdash {\color{red}\mc{C}_1} :: {\color{mygreen}\Delta}$},  {\small$\Psi_0; {\color{mygreen}K} \Vdash {\color{blue}\mc{F}_1} :: \cdot$}, {\small$\Psi_0; \cdot \Vdash {\color{red}\mc{C}_2} :: {\color{mygreen}\Delta}$,} and  {\small$\Psi_0; {\color{mygreen}K} \Vdash {\color{blue}\mc{F}_2} :: \cdot$} such that the channels in the interface {\small${\color{mygreen}\Delta} \Vdash {\color{mygreen}K}$} are \emph{{\color{mygreen}observable}}, requiring their maximal secrecy to be less than or equal to the secrecy level of the observer {\small$\xi$}.
This choice renders the definition of the logical relation more lightweight, avoiding casing on the maximal secrecy level of carrier channels in each case of the value interpretation.
To accommodate arbitrary runs {\small$\Psi_0; \Delta_3 \Vdash \mc{D}_3:: K_3$} and {\small$\Psi_0; \Delta_4 \Vdash \mc{D}_4:: K_4$} of a program with corresponding substitutions {\small$\Psi_0; \cdot \Vdash \mc{C}_3 :: \Delta_3$},  {\small$\Psi_0; K_3 \Vdash \mc{F}_3 :: \cdot$}, {\small$\Psi_0; \cdot \Vdash \mc{C}_4 :: \Delta_4$}, and {\small$\Psi_0; K_4 \Vdash \mc{F}_4 :: \cdot$} we define the projection {\small$\Downarrow \xi$} on a context, which filters out the channels with a maximal secrecy higher than or incomparable to {\small$\xi$}.
By fixing {\small$\Delta_3 \Downarrow \xi = \Delta_4 \Downarrow \xi = {\color{mygreen}\Delta'}$} and {\small$K_3\Downarrow \xi = K_4 \Downarrow \xi = {\color{mygreen}K'}$} we get {\small$\Psi_0; {\color{mygreen}\Delta'} \Vdash {\color{brown}\mc{D}_3'} :: {\color{mygreen}K'}$} and {\small$\Psi_0; {\color{mygreen}\Delta'} \Vdash {\color{brown}\mc{D}_4'} :: {\color{mygreen}K'}$} with corresponding substitutions {\small$\Psi_0; \cdot \Vdash {\color{red}\mc{C}_3'} :: {\color{mygreen}\Delta'}$},  {\small$\Psi_0; {\color{mygreen}K'} \Vdash {\color{blue}\mc{F}_3'} :: \cdot$}, {\small$\Psi_0; \cdot \Vdash {\color{red}\mc{C}_4'} :: {\color{mygreen}\Delta'}$}, and {\small$\Psi_0; {\color{mygreen}K'} \Vdash {\color{blue}\mc{F}_4'} :: \cdot$} where {\small${\color{brown}\mc{D}_3'}$} and {\small${\color{brown}\mc{D}_4'}$} internalize any subtrees in {\small$\mc{C}_3$}, {\small$\mc{F}_3$}, {\small$\mc{C}_4$}, and {\small$\mc{F}_4$} rooted at filtered out non-observable channels.
This projection step happens in the term interpretation shown in \Cref{fig:rel_term}, which we explore in more detail below.
The interested reader can consult Lemma~I.12 and Fig.~3 in the appendix to understand how {\small${\color{brown}\mc{D}_3'}$} and {\small${\color{brown}\mc{D}_4'}$} are built.
We remark that any interesting program will include branching on high-secrecy channels and thus will at least internalize the top substitution {\small$\substf_i$}, since the offering channel of the program {\small$\mnode_i$}will be of maximal secrecy {\small$\not\sqsubseteq \xi$}.

\subsubsection{Well-foundedness}

Because of recursive types, the logical relation can no longer be defined inductively by the multiset of the size of the types in its interface {\small${\color{mygreen}\Delta} \Vdash {\color{mygreen}K}$}.
To restore well-foundedness a measure like step-indexing~\cite{AppelMcAllesterTOPLAS2001,AhmedESOP2006} can be employed.
Step-indexing is typically tied to the number of type unfoldings or computation steps by which one of the two programs is bound.
A more natural and symmetric measure for our setting is the number of \emph{observations} that can be made along the interface {\small${\color{mygreen}\Delta} \Vdash {\color{mygreen}K}$} of the logical relation.
We thus bound the value and term interpretation of our logical relation by the number of observations and attach them as superscripts {\small$m$}, for {\small$m \geq 0$}, to the relation's interface {\small${\color{mygreen}\Delta} \Vdash {\color{mygreen}K}$}.
Since the two interpretations are mutually recursive, we define a second measure {\small$\mc{V} < \mc{E}$}, yielding a lexicographic order.
Given this lexicographic order, well-foundedness of our logical relation is immediate: \textit{(i)} the recursive invocations of {\small$\mc{E}$} in the value interpretation {\small$\mc{V}$} (see \Cref{fig:rel_value_right} and \Cref{fig:rel_value_left}) decrease in the first measure, the number of observations, and \textit{(ii)}  the recursive invocations of {\small$\mc{V}$} in the term interpretation {\small$\mc{E}$} (see \Cref{fig:rel_term}) remain the same for the first measure but decrease in the second measure.

\subsubsection{Nondeterminism}

It is now time to look at the definition of the term interpretation of our logical relation given in \Cref{fig:rel_term}.
The value interpretation calls the term interpretation on the updated configurations {\small$\lre{\mc{C}_1} {\mc{D}_1} {\mc{F}_1},\lre{\mc{C}_2} {\mc{D}_2} {\mc{F}_2}$}, after having asserted or assumed existence of related messages and having pushed these messages across the interface to the recipient.
The term interpretation now allows each subconfiguration {\small$\substc_1$}, {\small$\mnode_1$}, {\small$\substf_1$}, {\small$\substc_2$}, {\small$\mnode_2$}, and {\small$\substf_2$} to step internally until two configurations are reached where at least one message is ready to be transmitted across the same interface channel in both of the two configurations.
These configurations are then required to be in the value interpretation.

\Cref{fig:rel_term} makes this reasoning precise using the transition

\begin{center}
\begin{small}
$\lr{\mc{C}_i}{\mc{D}_i}{\mc{F}_i}\mapsto^{*_{\Theta; \Upsilon}}_{{\color{mygreen}\Delta} \Vdash {\color{mygreen}K}} \lr{\mc{C}'_i}{\mc{D}'_i}{\mc{F}'_i}$
\end{small}
\end{center}

\noindent which allows stepping a configuration by iterated application of the rewriting rules defined in \Cref{fig:dynamics} to each subconfiguration.
These steps are internal to each subconfiguration.
The star expresses that zero to multiple internal steps can be taken.
The superscripts {\small$\Theta; \Upsilon$} denote set of channels occurring in the interface {\small${\color{mygreen}\Delta} \Vdash {\color{mygreen}K}$} that have a message ready to be transported along the channel.
The set {\small$\Theta$} collects the \emph{incoming} channels, \ie channels that have messages ready  to be sent from {\small$\substc_i$} and {\small$\substf_i$} to {\small$\mnode_i$}, and the set {\small$\Upsilon$} collects the \emph{outgoing} channels, \ie channels that have messages ready to be sent from {\small$\mnode_i$} to {\small$\substc_i$} and {\small$\substf_i$}.

\input{figs/rel_term}

The term interpretation uses the transition {\small$\lr{\mc{C}_i}{\mc{D}_i}{\mc{F}_i}\mapsto^{*_{\Theta; \Upsilon}}_{{\color{mygreen}\Delta} \Vdash {\color{mygreen}K}} \lr{\mc{C}'_i}{\mc{D}'_i}{\mc{F}'_i}$} in a nested way to account for the fact that messages may not be ready for transmission along the same channel simultaneously in both runs due to nondeterminism.
Assuming that {\small$\mnode_1$} steps a finite number of times {\small$\lr{\mc{C}_1}{\mc{D}_1}{\mc{F}_1}\;\mapsto^{{*}_{\Theta; \Upsilon_1}}_{\Delta \Vdash K} \; \lr{\mc{C}_1}{\mc{D}'_1}{\mc{F}_1}$}, generating the outgoing channels {\small$\Upsilon_1$}, {\small$\mnode_2$} must be stepped {\small$\lr{\mc{C}_2}{\mc{D}_2}{\mc{F}_2}\mapsto^{{*}_{\Theta'; \Upsilon_1}}_{\Delta \Vdash K} \lr{\mc{C}_2}{\mc{D}'_2}{\mc{F}_2}$} to produce the same set of outgoing channels {\small$\Upsilon_1$}.
To provide a witness for this existential in the proof of the fundamental theorem we use \Cref{lem:indinvariant1} that intimately relies on nondeterminism, allowing us to choose a schedule that yields the appropriate {\small$\mnode_2'$}.
The logical relation assumes that {\small$\substc_i$}, {\small$\substf_i$} can step arbitrarily {\small$\lr{\mc{C}_1}{\mc{D}'_1}{\mc{F}_1}\mapsto^{{*}_{\Theta_1;\Upsilon_1}}_{\Delta \Vdash K} \lr{\mc{C}'_1}{\mc{D}'_1}{\mc{F}'_1}$} and {\small$\lr{\mc{C}_2}{\mc{D}'_2}{\mc{F}_2}\mapsto^{{*}_{\Theta_2;\Upsilon_1}}_{\Delta \Vdash K} \lr{\mc{C}'_2}{\mc{D}'_2}{\mc{F}'_2}$}, generating the incoming channels {\small$\Theta_1$} and {\small$\Theta_2$}, resp.
Having generated the set of incoming channels {\small$\Theta_1$} and {\small$\Theta_2$} and outgoing channels {\small$\Upsilon_1$}, the term interpretation then calls the value interpretation on the resulting configurations {\small$\lre{\mc{C}_1'} {\mc{D}_1'} {\mc{F}_1'},\lre{\mc{C}_2'} {\mc{D}_2'} {\mc{F}_2'}$} \emph{for every} channel that has a message ready for transmission in both runs.
Non-productive runs are trivially accommodated by the term interpretation, by allowing the set of incoming and outgoing channels to be empty.

The value interpretation uses the subscripts {\small${\color{mygreen}x}; {\color{mygreen}\cdot}$} and {\small${\color{mygreen}\cdot}; {\color{mygreen}x}$} to indicate the incoming and outgoing channels, \resp in focus.
The introduction of a \emph{focus channel} is key to resolving nondeterminism present in the logical relation itself, justifying the use of an {\small$\m{iff}$} in \Cref{fig:rel_value_right} and \Cref{fig:rel_value_left}.
Some focus channels use the subscripts {\small$\mc{C}$} or {\small$\mc{F}$}, denoting the sender, information necessary to distinguish messages in case {\small$\mnode_i$} is empty.

We may wonder why the logical relation seems oblivious of recursive type unfoldings.
Since the relation is indexed by the number of external observations, it is invariant in type unfoldings.
For example, ${\small ({\mc{B}_1}, {\mc{B}_2}) \in \mc{E}^\xi_{ \Psi_0}\llbracket \Delta \Vdash x {:} \m{pin} [c] \rrbracket^{m+1}}$ holds for the same set of configuration pairs as ${\small ({\mc{B}_1}, {\mc{B}_2}) \in \mc{E}^\xi_{ \Psi_0}\llbracket \Delta \Vdash x {:} \oplus \{\mathit{tok_i}{:}\m{pin}\}_{i \leq n} [c] \rrbracket^{m+1}}$, and the value interpretation with focus channel $x$ is only triggered when a message along $x$ is ready (\ie token or forward).
Any recursive calls of the term interpretation upon consumption of the message in the value interpretation will then happen at index $m$.

\subsubsection{Non-termination}\label{sec:4.logical-relation:nontermination}

Non-termination complicates the proof of noninterference, if a \emph{progress-sensitive} statement is wished for.
The crux is to show that either both runs are observably productive or not.
This requirement is expressed by the term interpretation of the logical relation (see \Cref{fig:rel_term}), which demands stepping {\small$\mnode_2$} to produce the same set of outgoing channels as {\small$\mnode_1$}.
The challenge is aggravated by the fact that two runs of a program can differ after branching on a high-secrecy channel, including divergence for just one of them.
Given this discrepancy between executed high-secrecy code, it seems almost impossible to show that the same observable outcomes are produced nonetheless.

Let's remind ourselves that noninterference amounts to a program equivalence, equating the program to itself up to the secrecy level of the observer.
This means that the two program runs will initially be the same and only diverge upon branching on high-secrecy messages.
If we can identify, in the proof of the fundamental theorem, which parts of the two program runs align and which ones differ, we should be able to use this information to our advantage.
In addition, we can exploit the guarantee made by the type system that ``low-secrecy code'' cannot sequentially depend on ``high-secrecy code'' because flows from high to low are ruled out by typing.

To make this reasoning more tangible, we introduce the notion of a \emph{relevant node}.
A relevant node is a process or message in {\small$\mnode$} that can directly or indirectly influence what messages are made ready for transmission across the interface {\small${\color{mygreen}\Delta} \Vdash {\color{mygreen}K}$}.
Relevant nodes thus affect the \emph{observable outcome}.
\Cref{fig:relevant-nodes} illustrates the set of relevant nodes for a {\small$\mnode$} and the evolution of this set when {\small$\mnode$} makes an internal transition to {\small$\mnode'$}.
We provide the main intuitions here, the interested reader is deferred to Def.~C.6 in the appendix.


\begin{figure*}
\centering
\includegraphics[width=0.75\textwidth]{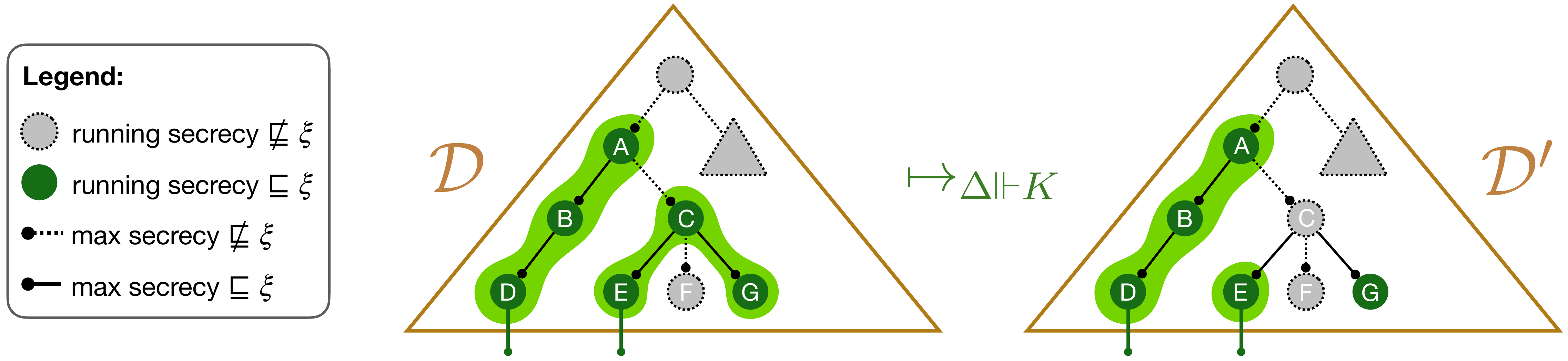}
\caption{Relevant nodes for program run $\mnode$ and successor $\mnode'$.}
\label{fig:relevant-nodes}
\end{figure*}

The relevant nodes in \Cref{fig:relevant-nodes} are highlighted in bright green and amount to the nodes {\small$\m{a}$}, {\small$\m{b}$}, {\small$\m{c}$}, {\small$\m{d}$}, {\small$\m{e}$}, and {\small$\m{g}$} in {\small$\mnode$} and to the nodes {\small$\m{a}$}, {\small$\m{b}$}, {\small$\m{d}$}, and {\small$\m{e}$} in {\small$\mnode'$}.
To identify the observable nodes, we must consider the \emph{running secrecy} of a node and the \emph{maximal secrecy} of its offering channel.
Nodes with a running secrecy {\small$\sqsubseteq \xi$} are displayed  as green solid circles and nodes with a running secrecy {\small$\not \sqsubseteq \xi$} as gray dotted circles.
Channels with a maximal secrecy {\small$\sqsubseteq \xi$} are displayed as black solid lines and channels with a maximal secrecy {\small$\not \sqsubseteq \xi$} as dotted black lines.
The observable nodes can be found by a bottom-up traversal, starting with the observable channels in the interface {\small${\color{mygreen}\Delta} \Vdash {\color{mygreen}K}$}.
Every node in {\small$\mnode_i$} that is connected to an observable channel is relevant, \ie nodes {\small$\m{d}$} and {\small$\m{e}$} in {\small$\mnode$} and nodes {\small$\m{d}$} and {\small$\m{e}$} in {\small$\mnode'$}.
Climbing up the ancestor line starting at those nodes, we include an ancestor node if its offering channel has maximal secrecy {\small$\sqsubseteq \xi$}.
Moreover, we include the closest ancestor node whose offering channel has maximal secrecy {\small$\not \sqsubseteq \xi$} and whose running secrecy is {\small$\sqsubseteq \xi$}.
We refer to the maximal ancestor node in this set as the \emph{pivot} node.
{\small$\mnode$} has the pivot nodes {\small$\m{a}$} and {\small$\m{c}$} and {\small$\mnode'$} has the pivot node {\small$\m{a}$}.
We then expand the set of relevant nodes with all of the pivot node's direct children that have an offering channel with maximal secrecy {\small$\sqsubseteq \xi$}.
In addition, we include the subtrees rooted at those children.
For pivot node {\small$\m{a}$} and {\small$\m{c}$} in {\small$\mnode$} these are the nodes rooted at {\small$\m{b}$}, {\small$\m{e}$}, and {\small$\m{g}$}.
For pivot node {\small$\m{a}$} in {\small$\mnode'$} these are the nodes rooted at {\small$\m{b}$}.

As illustrated by \Cref{fig:relevant-nodes}, the set of relevant nodes can become smaller when the program transitions and the running secrecy of a pivot node contained in the set increases as a result.
The transition in \Cref{fig:relevant-nodes} increases the running secrecy of pivot node {\small$\m{c}$}, which hence no longer is a pivot  node in {\small$\mnode'$}.
This transition can come about when {\small$\m{c}$} receives a message either from {\small$\m{a}$} or {\small$\m{f}$}.

The description above is phrased for a synchronous semantics for simplicity.
Our technical development accounts for an asynchronous semantics.
The notion of a relevant node is vital for stating noninterference in a progress-sensitive way.
In the proof of the fundamental theorem, we strengthen the inductive hypothesis to demand that the relevant nodes in {\small$\mnode_1$} and {\small$\mnode_2$} must execute exactly the same code, meaning that for every relevant {\small$\mathbf{proc}(x[d], P_1@d_1)$} or {\small$\mathbf{msg}(P_1)$} in {\small$\mnode_1$} and its counterpart {\small$\mathbf{proc}(x[d], P_2@d_1)$} or {\small$\mathbf{msg}(P_2)$} in {\small$\mnode_2$} it must hold that {\small$P_1 = P_2$}.
With this assumption we can easily discharge the proof obligation that {\small$\mnode_2$} produces the same set of outgoing channels as {\small$\mnode_1$} and thus the same observable outcome.

%% file: figs/rel_value_right.tex
\begin{figure*}
\begin{small}
\[\begin{array}{lcl}
({\lr{\cdot} {\mc{D}_1} {\mc{F}_1};{\lr{\cdot} {\mc{D}_1} {\mc{F}_1}}})\in \mc{V}^\xi_{ \Psi_0}\llbracket \cdot \Vdash y_\alpha{:}1[c]\rrbracket_{\cdot;y^\alpha}^{m+1} &
\m{iff} &
(\lr{\cdot} {\mc{D}_1} {\mc{F}_1};\lr{\cdot} {\mc{D}_2} {\mc{F}_2}) \in \m{Tree}_{\Psi_0}(\cdot \Vdash y_\alpha)\, \m{and} \,
\mc{D}_1=\mathbf{msg}(\mathbf{close}\,y_\alpha^c ) \, \m{and}\, \mc{D}_2=\mathbf{msg}( \mathbf{close}\,y_\alpha^c)\\[4pt]
 
(\lr{\mc{C}_1} {\mc{D}_1} {\mc{F}_1};\lr{\mc{C}_2} {\mc{D}_2} {\mc{F}_2})\in &
\m{iff} &
(\lr{\mc{C}_1} {\mc{D}_1} {\mc{F}_1};\lr{\mc{C}_2} {\mc{D}_2} {\mc{F}_2}) \in \m{Tree}_{\Psi_0}(\Delta \Vdash y_\alpha:\oplus\{\ell{:}A_\ell\}_{\ell \in I}[c])\, \m{and} \\
\mathcal{V}^\xi_{ \Psi_0}\llbracket (\Delta \Vdash y_\alpha{:}\oplus\{\ell{:}A_\ell\}_{\ell \in I}[c])\rrbracket_{\cdot;y^\alpha}^{m+1} & & \mathcal{D}_1=\mathcal{D}'_1\mathbf{msg}( y_\alpha^c.k;y_\alpha^c \leftarrow u^c_\delta ) \, \m{and}\, \mathcal{D}_2=\mathcal{D}'_2 \mathbf{msg}(y_\alpha^c.k; y_\alpha^c \leftarrow u^c_\delta )\, \m{and} \\
&& (\lre{\mc{C}_1} {\mc{D}'_1} {\mathbf{msg}(y_\alpha^c.k;y_\alpha^c \leftarrow u^c_\delta  )\mc{F}_1}, \lre{\mc{C}_2} {\mc{D}'_2} {\mathbf{msg}(y_\alpha^c.k;y_\alpha^c \leftarrow u^c_\delta  )\mc{F}_2}) \in  \mathcal{E}^\xi_{ \Psi_0}\llbracket \Delta \Vdash u_{\delta}{:}A_k[c]\rrbracket^{m}   \\[4pt]
   
(\lr{\mc{C}_1} {\mc{D}_1} {\mc{F}_1};\lr{\mc{C}_2} {\mc{D}_2} {\mc{F}_2}) \in &
\m{iff} &
(\lr{\mc{C}_1} {\mc{D}_1} {\mc{F}_1};\lr{\mc{C}_2} {\mc{D}_2} {\mc{F}_2}) \in \m{Tree}_{\Psi_0}(\Delta \Vdash y_\alpha{:}\&\{\ell{:}A_\ell\}_{\ell \in I}[c])\,\m{and}\\
\mathcal{V}^\xi_{ \Psi_0}\llbracket \Delta \Vdash y_\alpha{:}\&\{\ell{:}A_\ell\}_{\ell \in I}[c]\rrbracket_{y^\alpha_{\mc{F}}; \cdot}^{m+1} &&
\mathcal{F}_1=\mathbf{msg}( y_\alpha^c.k; u_\delta^c \leftarrow y_\alpha^c ) \mathcal{F}'_1 \, \m{and}\,\m{if} \, \mathcal{F}_2=\mathbf{msg}( y_\alpha^c.k;u_\delta^c \leftarrow y_\alpha^c)\mathcal{F}'_2 \, \m{then} \\
&& (\lre{\mc{C}_1} {\mc{D}_1 \mathbf{msg}( y_\alpha^c.k;u_\delta^c \leftarrow y_\alpha^c)} {\mc{F}'_1},\lre{\mc{C}_2} {\mc{D}_2\mathbf{msg}( y_\alpha^c.k;u_\delta^c \leftarrow y_\alpha^c)} {\mc{F}'_2})\in  \mathcal{E}^\xi_{\Psi_0}\llbracket \Delta \Vdash u_{\delta}{:}A_k[c]\rrbracket^{m}    \\[4pt]
   
(\lr{\mc{C}'_1\mc{C}''_1} {\mc{D}_1} {\mc{F}_1};\lr{\mc{C}'_2\mc{C}''_2} {\mc{D}_2} {\mc{F}_2})  \in&
\m{iff} &
\mc{C}'_i \in \m{Tree}_{\Psi_0}(\cdot \Vdash \Delta')\,\m{and}\, (\lr{\mc{C}'_1\mc{C}''_1} {\mc{D}_1} {\mc{F}_1};\lr{\mc{C}'_2\mc{C}''_2} {\mc{D}_2} {\mc{F}_2}) \in \m{Tree}_{\Psi_0}(\Delta', \Delta'' \Vdash y_\alpha{:}A\otimes B[c])\\
\mathcal{V}^\xi_{ \Psi_0}\llbracket \Delta',\Delta'' \Vdash y_\alpha{:}A\otimes B[c]\rrbracket_{\cdot; y^\alpha}^{m+1} & &
\m{and}\, \mathcal{D}_1=\mathcal{D}'_1\mathcal{T}_1\mathbf{msg}( \mb{send}\,x_\beta^c\,y_\alpha^c; y_\alpha^c \leftarrow u^c_\delta )\, \m{for}\, \mc{T}_1 \in \m{Tree}_{\Psi_0}(\Delta'' \Vdash x_\beta{:}A[c])\, \m{and}\\
&& \mathcal{D}_2=\mathcal{D}'_2\mc{T}_2 \mathbf{msg}( \mb{send}\,x_\beta^c\,y_\alpha^c; y_\alpha^c \leftarrow u^c_\delta )\,\m{for}\, \mc{T}_2 \in \m{Tree}_{\Psi_0}(\Delta'' \Vdash x_\beta{:}A[c])\, \m{and}\\

&& (\lre{\mc{C}''_1} {\mc{T}_1} {\mc{C}'_1 \mc{D}'_1 \mathbf{msg}( \mb{send}\,x_\beta^c\,y_\alpha^c; y_\alpha^c \leftarrow u^c_\delta )\mc{F}_1},\\
&& \lre{\mc{C}''_2} {\mc{T}_2} {\mc{C}'_2\mc{D}'_2 \mathbf{msg}( \mb{send}\,x_\beta^c\,y_\alpha^c; y_\alpha^c \leftarrow u^c_\delta )\mc{F}_2})  \in  \mathcal{E}^\xi_{ \Psi_0}\llbracket \Delta'' \Vdash x_{\beta}{:}A[c]\rrbracket^{m}  \,\m{and}\\
&& (\lre{\mc{C}'_1} {\mc{D}'_1} {\mc{C}''_1\mc{T}_1 \mathbf{msg}( \mb{send}\,x_\beta^c\, y_\alpha^c; y_\alpha^c \leftarrow u^c_\delta )\mc{F}_1},\\
&& \lre{\mc{C}'_2} {\mc{D}'_2} {\mc{C}''_2\mc{T}_2 \mathbf{msg}( \mb{send}\,x_\beta^c\,y_\alpha^c; y_\alpha^c \leftarrow u^c_\delta )\mc{F}_2})  \in  \mathcal{E}^\xi_{ \Psi_0}\llbracket \Delta' \Vdash u_{\delta}{:}B[c]\rrbracket^{m}  \\[4pt]
  
(\lr{\mc{C}_1} {\mc{D}_1} {\mc{F}_1};\lr{\mc{C}_2} {\mc{D}_2} {\mc{F}_2})  \in &
\m{iff} &
(\lr{\mc{C}_1} {\mc{D}_1} {\mc{F}_1};\lr{\mc{C}_2} {\mc{D}_2} {\mc{F}_2}) \in \m{Tree}_{\Psi_0}(\Delta \Vdash y_\alpha{:}A\multimap B[c])\, \m{and}\\
\mathcal{V}^\xi_{ \Psi_0}\llbracket \Delta \Vdash y_\alpha{:}A\multimap B[c]\rrbracket_{y^\alpha_{\mc{F}}; \cdot}^{m+1} &&
\mathcal{F}_1=\mathcal{T}_1\mathbf{msg}( \mb{send}\,x_\beta^c\,y_\alpha^c;u_\delta^c \leftarrow y_\alpha^c)\mathcal{F}'_1\,\m{for}\, \mc{T}_1 \in \m{Tree}_{\Psi_0}(\cdot \Vdash x_\beta{:}A[c])\, \m{and}\,\m{if} \\
&&\mathcal{F}_2=\mathcal{T}_2\mathbf{msg}( \mb{send}\,x_\beta^c\,y_\alpha^c;u_\delta^c \leftarrow y_\alpha^c)\mathcal{F}'_2\,\m{for}\, \mc{T}_2 \in \m{Tree}_{\Psi_0}(\cdot \Vdash x_\beta{:}A[c])\,\m{then}\,\\
&& (\lre{\mc{C}_1\mc{T}_1} {\mc{D}_1\mathbf{msg}( \mb{send}\,x_\beta^c\,y_\alpha^c;u_\delta^c \leftarrow y_\alpha^c)} {\mc{F}'_1},\\
&& \lre{\mc{C}_2\mc{T}_2} {\mc{D}_2\mathbf{msg}(\mb{send}\,x_\beta^c\,y_\alpha^c;u_\delta^c \leftarrow y_\alpha^c)} {\mc{F}'_2})  \in  \mathcal{E}^\xi_{ \Psi_0}\llbracket \Delta, x_\beta{:}A[c]\Vdash u_{\delta}{:}B[c]\rrbracket^{m}  \\[4pt]
   
({\lr{\mc{C}_1} {\mc{D}_1} {\mc{F}_1};{\lr{\mc{C}_2} {\mc{D}_2} {\mc{F}_2}}}) \in &
\m{iff}&
(\lr{\mc{C}_1} {\mc{D}_1} {\mc{F}_1};\lr{\mc{C}_2} {\mc{D}_2} {\mc{F}_2}) \in \m{Tree}_{\Psi_0}(\Delta \Vdash y_\alpha{:}A[c])\,\m{and}\\
\mc{V}^\xi_{ \Psi_0}\llbracket  \Delta \Vdash y_\alpha{:}A[c] \rrbracket_{\cdot;y^\alpha}^{m+1} &&
\mc{D}_1=\mc{D}'_1\mathbf{proc}(y_\alpha^c, y_\alpha^c\leftarrow x_\beta^c @d_1) \, \m{and}\,  \mc{D}_2=\mc{D}'_2\mathbf{proc}(y_\alpha^c, y_\alpha^c\leftarrow x_\beta^c @d_2)\,\m{and}\\
&& ([x_{\beta}^c/y_\alpha^c]{\lre{\mc{C}_1} {\mc{D}'_1} {\mc{F}_1},[x_{\beta}^c/y_\alpha^c]{\lre{\mc{C}_2} {\mc{D}'_2} {\mc{F}_2}}}) \in \mathcal{E}^\xi_{ \Psi_0}\llbracket \Delta\Vdash x_{\beta}{:}A[c] \rrbracket^{m}
   \end{array}\]%
\end{small}
\caption{Value interpretation of logical relation for \emph{right} communications.}
\label{fig:rel_value_right}
\end{figure*}

%% file: figs/rel_value_left.tex
\begin{figure*}
\begin{small}
\begin{mathpar}
\begin{array}{lcl}
(\lr{\mc{C}_1} {\mc{F}_1}{\mc{D}_1}; \lr{\mc{C}_2} {\mc{F}_2}{\mc{D}_2})\in &
\m{iff} &
(\lr{\mc{C}_1} {\mc{F}_1}{\mc{D}_1}; \lr{\mc{C}_2} {\mc{F}_2}{\mc{D}_2}) \in \m{Tree}_{\Psi_0}(\Delta, y_\alpha{:}1[c] \Vdash K)\, \m{and} \\
\mc{V}^\xi_{ \Psi_0}\llbracket \Delta, y_\alpha{:}1[c] \Vdash K\rrbracket_{y^\alpha_{\mc{C}}; \cdot}^{m+1} &&
\mc{C}_1=\mathcal{C}'_1\mathbf{msg}( \mathbf{close}\,y_\alpha^c) \, \m{and}\, \m{if}\, \mc{C}_2=\mathcal{C}'_2\mathbf{msg}( \mathbf{close}\,y_\alpha^c)\, \m{then} \\
&&(\lre{\mc{C}'_1} {\mathbf{msg}( \mathbf{close}\,y_\alpha^c)\mc{D}_1} {\mc{F}_1},\lre{\mc{C}'_2} {\mathbf{msg}( \mathbf{close}\,y_\alpha^c )\mc{D}_2} {\mc{F}_2})\in \mathcal{E}^\xi_{\Psi_0}\llbracket \Delta \Vdash K\rrbracket^{m}    \\[4pt]

(\lr{\mc{C}_1} {\mc{D}_1} {\mc{F}_1};\lr{\mc{C}_2} {\mc{D}_2} {\mc{F}_2}) \in &
\m{iff} &
(\lr{\mc{C}_1} {\mc{D}_1} {\mc{F}_1};\lr{\mc{C}_2} {\mc{D}_2} {\mc{F}_2}) \in \m{Tree}_{\Psi_0}(\Delta, y_\alpha{:}\oplus\{\ell{:}A_\ell\}_{\ell \in I}[c] \Vdash K)\, \m{and}\\
\mathcal{V}^\xi_{ \Psi_0}\llbracket \Delta, y_\alpha:\oplus\{\ell{:}A_\ell\}_{\ell \in I}[c] \Vdash K\rrbracket_{y^\alpha_{\mc{C}}; \cdot}^{m+1}
&&\mathcal{C}_1=\mathcal{C}'_1\mathbf{msg}( y_\alpha^c.k; y_\alpha^c \leftarrow u^c_\delta) \, \m{and}\, \m{if}\, \mathcal{C}_2=\mathcal{C}'_2 \mathbf{msg}(y_\alpha^c.k;y_\alpha^c \leftarrow u^c_\delta )\, \m{then}\\
&& (\lre{\mc{C}'_1} {\mathbf{msg}(y_\alpha^c.k; y_\alpha^c \leftarrow u^c_\delta)\mc{D}_1} {\mc{F}_1},\lre{\mc{C}'_2} {\mathbf{msg}(y_\alpha^c.k; y_\alpha^c \leftarrow u^c_\delta)\mc{D}_2} {\mc{F}_2}) \in  \mathcal{E}^\xi_{ \Psi_0}\llbracket \Delta, u_{\delta}{:}A_k[c] \Vdash K \rrbracket^{m}    \\[4pt]
   
(\lr{\mc{C}_1} {\mc{D}_1} {\mc{F}_1};\lr{\mc{C}_2} {\mc{D}_2} {\mc{F}_2}) \in &
\m{iff} &
(\lr{\mc{C}_1} {\mc{D}_1} {\mc{F}_1};\lr{\mc{C}_2} {\mc{D}_2} {\mc{F}_2}) \in \m{Tree}_{\Psi_0}(\Delta, y_\alpha{:}\&\{\ell{:}A_\ell\}_{\ell \in I}[c] \Vdash K)\, \m{and}\\
\mathcal{V}^\xi_{ \Psi_0}\llbracket \Delta, y_\alpha{:}\&\{\ell{:}A_\ell\}_{\ell \in I}[c] \Vdash K \rrbracket_{\cdot; y^\alpha}^{m+1} &&\mathcal{D}_1=\mathbf{msg}( y_\alpha^c.k;u_\delta^c \leftarrow y_\alpha^c) \mathcal{D}'_1 \, \m{and}\ \mathcal{D}_2=\mathbf{msg}( y_\alpha^c.k;u_\delta^c \leftarrow y_\alpha^c)\mathcal{D}'_2\,\m{and} \\
&&(\lre{\mc{C}_1 \mathbf{msg}( y_\alpha^c.k;u_\delta^c \leftarrow y_\alpha^c)} {\mc{D}'_1} {\mc{F}_1},\lre{\mc{C}_2\mathbf{msg}( y_\alpha^c.k;u_\delta^c \leftarrow y_\alpha^c)} {\mc{D}'_2} {\mc{F}_2}) \in  \mathcal{E}^\xi_{ \Psi_0}\llbracket \Delta, u_{\delta}{:}A_k[c] \Vdash K \rrbracket^{m}    \\[4pt]
   
(\lr{\mc{C}_1} {\mc{D}_1} {\mc{F}_1};\lr{\mc{C}_2} {\mc{D}_2} {\mc{F}_2})  \in &
\m{iff} &
(\lr{\mc{C}_1} {\mc{D}_1} {\mc{F}_1};\lr{\mc{C}_2} {\mc{D}_2} {\mc{F}_2}) \in \m{Tree}_{\Psi_0}(\Delta, y_\alpha{:}A\otimes B[c] \Vdash K)\,\m{and} \\
\mathcal{V}^\xi_{ \Psi_0}\llbracket \Delta, y_\alpha{:}A\otimes B[c] \Vdash K \rrbracket_{y^\alpha_{\mc{C}}; \cdot}^{m+1}&&\mathcal{C}_1=\mathcal{C}'_1\mathbf{msg}( \mb{send}x_\beta^c\,y_\alpha^c ; y_\alpha^c \leftarrow u^c_\delta)\, \m{and}\,\m{if}\, \mathcal{C}_2=\mathcal{C}'_2 \mathbf{msg}( \mb{send}x_\beta^c\,y_\alpha^c ; y_\alpha^c \leftarrow u^c_\delta) \, \m{then}  \\
&&(\lre{\mc{C}'_1}{\mathbf{msg}( \mb{send}x_\beta^c,y_\alpha^c;y_\alpha^c \leftarrow u^c_\delta )\mc{D}_1} {\mc{F}_1},\\
&&\lre{\mc{C}'_2} {\mathbf{msg}( \mb{send}x_\beta^c,y_\alpha^c ; y_\alpha^c \leftarrow u^c_\delta )\mc{D}_2} {\mc{F}_2})  \in  \mathcal{E}^\xi_{ \Psi_0}\llbracket \Delta,x_\beta{:}A[c], u_{\delta}{:}B[c] \Vdash K \rrbracket^{m}  \\[4pt]
  
(\lr{\mc{C}'_1\mc{C}''_1} {\mc{D}_1} {\mc{F}_1};\lr{\mc{C}'_2\mc{C}''_2} {\mc{D}_2} {\mc{F}_2})  \in &
\m{iff} &
\mc{C}''_i \in \m{Tree}_{\Psi_0}(\cdot \Vdash \Delta')\,\m{and} \, (\lr{\mc{C}'_1\mc{C}''_1} {\mc{D}_1} {\mc{F}_1};\lr{\mc{C}'_2\mc{C}''_2} {\mc{D}_2} {\mc{F}_2}) \in \m{Tree}_{\Psi_0}(\Delta', \Delta'', y_\alpha{:}A\multimap B[c] \Vdash K)\\
\mathcal{V}^\xi_{\Psi_0}\llbracket \Delta', \Delta'', y_\alpha{:}A\multimap B[c] \Vdash K \rrbracket_{\cdot;y^\alpha}^{m+1}
&& \m{and}\,\mathcal{D}_1=\mathcal{T}_1\mathbf{msg}( \mb{send}x_\beta^c\,y_\alpha^c; u_\delta^c \leftarrow y_\alpha^c)\,\mathcal{D}''_1\, \m{and} \,\m{for}\, \mc{T}_1 \in \m{Tree}_{\Psi_0}(\Delta' \Vdash x_\beta{:}A[c])\\
&& \mathcal{D}_2=\mathcal{T}_2\mathbf{msg}( \mb{send}x_\beta^c\,y_\alpha^c;u_\delta^c \leftarrow y_\alpha^c)\,\mc{D}''_2\, \m{and}  \,\m{for}\, \mc{T}_2 \in \m{Tree}_{\Psi_0}(\Delta' \Vdash x_\beta{:}A[c])\,\m{and} \\
&& (\lre{\mc{C}''_1}{\mathcal{T}_1}{\mc{C}'_1\mathbf{msg}( \mb{send}x_\beta^c\,y_\alpha^c; u_\delta^c \leftarrow y_\alpha^c)\mc{D}''_1\mc{F}_1},\\
&& \lre{\mc{C}''_2}{\mathcal{T}_2}{\mc{C}'_2\mathbf{msg}( \mb{send}x_\beta^c\,y_\alpha^c; u_\delta^c \leftarrow y_\alpha^c) \mc{D}''_2 \mc{F}_2}) \in  \mathcal{E}^\xi_{ \Psi_0}\llbracket \Delta' \Vdash x_\beta{:}A[c] \rrbracket^{m} \,\m{and}\\
&& (\lre{\mc{C}'_1\mc{C}''_1\mathcal{T}_1\mathbf{msg}( \mb{send}x_\beta^c\,y_\alpha^c; u_\delta^c \leftarrow y_\alpha^c)} {\mc{D}''_1} {\mc{F}_1},\\
&& \lre{\mc{C}'_2\mc{C}''_2\mathcal{T}_2\mathbf{msg}( \mb{send}x_\beta^c\,y_\alpha^c; u_\delta^c \leftarrow y_\alpha^c)} {\mc{D}''_2} {\mc{F}_2}) \in  \mathcal{E}^\xi_{ \Psi_0}\llbracket \Delta'', u_{\delta}{:}B[c] \Vdash K \rrbracket^{m}  \\[4pt]

({\lr{\mc{C}_1} {\mc{D}_1} {\mc{F}_1};{\lr{\mc{C}_2} {\mc{D}_2} {\mc{F}_2}}}) \in &
\m{iff} &
(\lr{\mc{C}_1} {\mc{D}_1} {\mc{F}_1};\lr{\mc{C}_2} {\mc{D}_2} {\mc{F}_2}) \in \m{Tree}_{\Psi_0}(\Delta', y_\alpha{:}A[c] \Vdash K)\,\m{and}\\
\mc{V}^\xi_{ \Psi_0}\llbracket \Delta, y_\alpha{:}A[c] \Vdash K \rrbracket_{y^\alpha_{\mc{C}}; \cdot}^{m+1}
&& \mc{C}_1=\mc{C}'_1\mathbf{proc}(y_\alpha^c, y_\alpha^c\leftarrow x_\beta^c @d_1) \, \m{and}\, \m{if}\, \mc{C}_2=\mc{C}'_2\mathbf{proc}(y_\alpha^c, y_\alpha^c\leftarrow x_\beta^c @d_2) \, \m{then}\\ 
&& ([x_{\beta}^c/y_\alpha^c]{\lre{\mc{C}'_1} {\mc{D}_1} {\mc{F}_1},[x_{\beta}^c/y_\alpha^c]{\lre{\mc{C}'_2} {\mc{D}_2} {\mc{F}_2}}}) \in \mathcal{E}^\xi_{ \Psi_0}\llbracket \Delta, x_{\beta}{:}A[c] \Vdash [x_{\beta}^c/y_\alpha^c] K \rrbracket^{m}
\end{array}
\end{mathpar}
\end{small}
\caption{Value interpretation of logical relation for \emph{left} communications.}
\label{fig:rel_value_left}
\end{figure*}

%% file: figs/rel_term.tex
\begin{figure*}
\begin{small}
\[\begin{array}{lcl}
({\mc{B}_1}, {\mc{B}_2}) \in \mc{E}^\xi_{ \Psi_0}\llbracket \Delta \Vdash K \rrbracket^{m+1} &
\m{iff} &
\mc{B}_1=\mc{C}_1\mc{D}_1\mc{F}_1 \,\m{and}\,\mc{B}_2=\mc{C}_2\mc{D}_2\mc{F}_2\,\m{and} \,
(\lr{\mc{C}_1}{\mc{D}_1}{\mc{F}_1}; \lr{\mc{C}_2}{\mc{D}_2}{\mc{F}_2})\in \m{Tree}_{\Psi_0}(\Delta \Vdash K)\,\m{and}\\
&& \m{if}\; \lr{\mc{C}_1}{\mc{D}_1}{\mc{F}_1}\;\mapsto^{{*}_{\Theta; \Upsilon_1}}_{\Delta \Vdash K} \; \lr{\mc{C}_1}{\mc{D}'_1}{\mc{F}_1}\, \m{then}\,\exists \mc{D}'_2 \m{\,such \, that\,}\,\lr{\mc{C}_2}{\mc{D}_2}{\mc{F}_2}\mapsto^{{*}_{\Theta'; \Upsilon_1}}_{\Delta \Vdash K} \lr{\mc{C}_2}{\mc{D}'_2}{\mc{F}_2}, \; \m{and}\\[2pt]
&& \forall \mc{C}'_1,\mc{F}'_1,\mc{C}'_2, \mc{F}'_2\,\m{if}\, \lr{\mc{C}_1}{\mc{D}'_1}{\mc{F}_1}\mapsto^{{*}_{\Theta_1;\Upsilon_1}}_{\Delta \Vdash K} \lr{\mc{C}'_1}{\mc{D}'_1}{\mc{F}'_1}\, \m{and}\,\lr{\mc{C}_2}{\mc{D}'_2}{\mc{F}_2}\mapsto^{{*}_{\Theta_2;\Upsilon_1}}_{\Delta \Vdash K} \lr{\mc{C}'_2}{\mc{D}'_2}{\mc{F}'_2},\,\m{then}\\ 
&& \forall\, x \in \Upsilon_1.\, (\lr{\mc{C}'_1}{\mc{D}'_1}{\mc{F}'_1}; \lr{\mc{C}'_2}{\mc{D}'_2}{\mc{F}'_2})\in \mc{V}^\xi_{\Psi_0}\llbracket  \Delta \Vdash K \rrbracket_{\cdot;x}^{m+1}\,\m{and}\,\\[2pt]
&& \forall\, x \in (\Theta_1 \cap \Theta_2).\, (\lr{\mc{C}'_1}{\mc{D}'_1}{\mc{F}'_1}; \lr{\mc{C}'_2}{\mc{D}'_2}{\mc{F}'_2})\in \mc{V}^\xi_{\Psi_0}\llbracket  \Delta \Vdash K \rrbracket_{x;\cdot}^{m+1}\\[4pt]
  
({\mc{B}_1}, {\mc{B}_2}) \in \mc{E}^\xi_{ \Psi_0}\llbracket \Delta \Vdash K \rrbracket^{0} &
\m{iff}&
\mc{B}_1=\mc{C}_1\mc{D}_1\mc{F}_1 \,\m{and}\,\mc{B}_2=\mc{C}_2\mc{D}_2\mc{F}_2\,\m{and}\,(\lr{\mc{C}_1}{\mc{D}_1}{\mc{F}_1}; \lr{\mc{C}_2}{\mc{D}_2}{\mc{F}_2})\in \m{Tree}_{\Psi_0}(\Delta \Vdash K)
    \end{array}\]
\end{small}
\caption{Term interpretation of logical relation.}
\label{fig:rel_term}
\end{figure*}

%% file: sections/5.metatheory.tex
\section{Adequacy and equivalence}\label{sec:5.metatheory}

This section establishes the main properties of our logical relation.\footnote{Standard safety properties can be found in the appendix (Thms.~C.3, C.4).}
We first show that it is an equivalence relation, necessitating proofs of reflexivity, symmetry, and transitivity.
Reflexivity, in particular, coincides with the noninterference property of $\lang$.
Then we establish adequacy, asserting that two logically equivalent programs are bisimilar when closed with the same context.



\subsection{Logical equivalence}
We define the logical equivalence, {\small$\equiv^{\Psi_0}_{\xi}$}, of two open configurations as being related by the term interpretation for every possible pair of substitutions and any number of observations, {\small$m$}. 

\begin{definition}\label{def:noninterference}
The relation {\small\[(\Delta_1 \Vdash \mc{D}_1:: x_\alpha {:}A_1[{c_1}]) \equiv^{\Psi_0}_{\xi} (\Delta_2 \Vdash \mc{D}_2:: y_\beta {:}A_2[{c_2}])\]}is defined as 
{\small $\Psi_0; \Delta_1 \Vdash \mc{D}_1:: x_\alpha {:}A_1[{c_1}]$ and $\Psi_0; \Delta_2 \Vdash \mc{D}_2:: y_\beta {:}A_2[{c_2}]$} and
{\small $\Delta_1 {\Downarrow} \xi= \Delta_2 {\Downarrow} \xi=\Delta$ and  $x_\alpha {:}A_1[{c_1}]{\Downarrow} \xi= y_\beta {:}A_2[{c_2}] {\Downarrow} \xi= K$} and
          for all {\small$\mc{C}_1, \mc{C}_2, \mc{F}_1, \mc{F}_2,$}
          with {\small$\ctype{\Psi_0}{\cdot}{\mc{C}_i}{\Delta_i}$}  and {\small$\ctype{\Psi_0}{x_\alpha {:}A_1[{c_1}]}{\mc{F}_1}{\cdot}$} and 
            {\small$\ctype{\Psi_0}{y_\beta {:}A_2[{c_2}]}{\mc{F}_2}{\cdot}$}, and  for all {\small$m$}  we have{\small
\[\begin{array}{l}
(\mc{C}_1\mc{D}_1\mc{F}_1, \mc{C}_2\mc{D}_2\mc{F}_2) \in \mc{E}_{\Psi_0}^\xi\llbracket \Delta \Vdash K \rrbracket^{m},\,
\m{and} \\[1pt]
 (\mc{C}_2\mc{D}_2\mc{F}_2, \mc{C}_1\mc{D}_1\mc{F}_1) \in \mc{E}_{\Psi_0}^\xi\llbracket \Delta \Vdash K \rrbracket^{m}.
\end{array}
\]}
\end{definition}
Symmetry is immediate from the above definition. We focus on the proofs of reflexivity and transitivity and our techniques to overcome the challenges imposed by non-termination and concurrency.

\subsubsection{Reflexivity}


Reflexivity is at the heart of noninterference and relies on the notion of relevant nodes introduced in \Cref{sec:4.logical-relation:nontermination}.
A crucial lemma for establishing noninterference is that two configurations with equal relevant nodes can maintain this equality while stepping.
If one configuration takes a step, the other one can catch up and sustain the equality of their relevant nodes.
\begin{lemma}\label{lem:indinvariant1}
Consider {\small$(\lr{\mc{C}_1}{\mc{D}_1}{\mc{F}_1}; \lr{\mc{C}_2}{\mc{D}_2}{\mc{F}_2}) \in \m{Tree}_{\Psi_0}(\Delta \Vdash K)$}. 
If {\small$\lr{\mc{C}_1} {\mc{D}_1}{ \mc{F}_1} \mapsto_{\Delta \Vdash K}\lr{\mc{C}_1} {\mc{D}'_1}{ \mc{F}_1}$} 
such that 
{\small$\cproj{\mc{D}_1}{\xi}=\cproj{\mc{D}_2}{\xi}$}, then for some {\small$\lr{\mc{C}'_2}{\mc{D}'_2} {\mc{F}'_2}$} 
, we have {\small$\lr{\mc{C}_2} {\mc{D}_2}{ \mc{F}_2} \mapsto^{0,1}_{\Delta \Vdash K}\lr{\mc{C}_2} {\mc{D}'_2}{ \mc{F}_2}$} such  that 
{\small$\cproj{\mc{D}'_1}{\xi}=\cproj{\mc{D}'_2}{\xi}$}.
\end{lemma}
\begin{proof}
For the complete proof see Lem.~C.17 in the appendix.
\end{proof}
Using \Cref{lem:indinvariant1}, we can establish our fundamental theorem, requiring two open programs to be equal up to their relevant nodes.
\begin{theorem}[Fundamental theorem]\label{thm:ft}
For all security levels $\xi$, and configurations ${\small\Psi_0; \Delta_1 \Vdash \mathcal{D}_1:: u_\alpha {:}T_1[c_1]}$ and ${\small\Psi_0; \Delta_2 \Vdash \mathcal{D}_2:: v_\beta {:}T_2[c_2]}$  with $\cproj{\mathcal{D}_1}{\xi} \meq \cproj{\mathcal{D}_2}{\xi}$, $\Delta_1 {\Downarrow} \xi = \Delta_2 {\Downarrow} \xi$, and $u_\alpha {:}T_1[c_1]{\Downarrow} \xi = v_\beta {:}T_2[c_2] {\Downarrow} \xi$  we have {\small$(\Delta_1 \Vdash \mathcal{D}_1:: u_\alpha {:}T_1[c_1])  \equiv^{\Psi_0}_{\xi} (\Delta_2 \Vdash \mathcal{D}_2:: v_\beta {:}T_2[c_2]).$}
\end{theorem}
\begin{proof}
For the complete proof see Thm.~C.19 in the appendix.
\end{proof}

The reflexivity of the relation $\equiv^{\Psi_0}_{\xi}$, otherwise known as noninterference, is a straightforward corollary of the fundamental theorem. 
\begin{theorem}[Noninterference]
For all security levels $\xi$ and configurations $\Psi_0; \Delta \Vdash \mathcal{D}:: x_\alpha {:}T[c]$, we have 
{\small
\[(\Delta \Vdash \mathcal{D}:: x_\alpha {:}T[c])  \equiv^{\Psi_0}_{\xi} (\Delta \Vdash \mathcal{D}:: x_\alpha {:}T[c]).\]
}
\end{theorem}

\subsubsection{Transitivity}\label{subsec:transitivity}

Proofs of transitivity can be notoriously challenging.
We focus here on the key ideas underlying our proof and explain the important role the choice of our step index assumed.

\begin{lemma}[Transitivity]
For all security levels {\small$\xi$} and configurations {\small$\Psi_0; \Delta_1 \Vdash \mathcal{D}_1:: x_\alpha {:}T_1[c_1]$}, and {\small$\Psi_0; \Delta_2 \Vdash \mathcal{D}_2:: y_\beta {:}T_2[c_2]$}, and {\small$\Psi_0; \Delta_3 \Vdash \mathcal{D}_3:: z_\eta {:}T_3[c_3]$}
{\small\[\begin{array}{ll}
\m{if}\quad\dagger_1\,(\Delta_1 \Vdash \mathcal{D}_1:: x_\alpha {:}T_1[c_1])  \equiv^{\Psi_0}_{\xi} (\Delta_2 \Vdash \mathcal{D}_2:: y_\beta {:}T_2[c_2])\;\m{and}\\
\;\quad\;\,\,\dagger_2\,\,(\Delta_2 \Vdash \mathcal{D}_2:: y_\beta {:}T_2[c_2])  \equiv^{\Psi_0}_{\xi} (\Delta_3 \Vdash \mathcal{D}_3:: z_\eta {:}T_3[c_3]),\\
\m{then}\star\,(\Delta_1 \Vdash \mathcal{D}_1:: x_\alpha {:}T_1[c_1]) \equiv^{\Psi_0}_{\xi} (\Delta_3 \Vdash \mathcal{D}_3:: z_\eta {:}T_3[c_3]).
\end{array}\]}
\end{lemma}
\begin{proof}
See Lem.~C.22 in the appendix for the complete proof.
\end{proof}
 To explain the key ideas of our proof more smoothly, we consider a trivial lattice with only one element. A trivial lattice in the transitivity lemma implies $\Delta_i=\Delta$ and $x_\alpha {:}T_1[c_1]=y_\beta {:}T_2[c_2]=z_\eta {:}T_3[c_3]$.
 
 Our first technique in the proof of transitivity is similar to~\citet{AhmedESOP2006}, i.e., to provide the same substitutions for the second and third programs and use reflexivity. Consider an arbitrary pair of substitutions for $\mc{D}_1$ and $\mc{D}_3$ in $\star$:  $\mc{C}_1, \mc{F}_1$ for $\mc{D}_1$ and $\mc{C}_3,\mc{F}_3$ for $\mc{D}_3$. We instantiate the substitutions in $\dagger_1$ with $\mc{C}_1, \mc{F}_1$ for $\mc{D}_1$ and $\mc{C}_3,\mc{F}_3$ for $\mc{D}_2$, and the substitutions in $\dagger_2$ with $\mc{C}_3, \mc{F}_3$ for $\mc{D}_2$ and $\mc{C}_3,\mc{F}_3$ for $\mc{D}_3$. Reflexivity ensures $\circ_1\,\mc{C}_3\equiv^{\Psi_0}_{\xi}\mc{C}_3$ and $\circ_2\,\mc{F}_3\equiv^{\Psi_0}_{\xi}\mc{F}_3$.\footnote{$\mc{C}_3\equiv^{\Psi_0}_{\xi}\mc{C}_3$ stands for reflexivity of each tree in $\mc{C}_3$.} In particular, when we put together the substitutions $\mc{C}_3,\mc{F}_3$ with programs $\mc{D}_2$ and $\mc{D}_3$, they send the same messages to $\mc{D}_2$ and $\mc{D}_3$ as long as $\mc{D}_2$ and $\mc{D}_3$ respond with the same messages. Therefore, we are in a rely-guarantee situation; for example, when $T_3$ is a negative protocol, $\mc{D}_2$ and $\mc{D}_3$ rely on $\mc{F}_3$ to send the same messages, and $\mc{F}_3$ guarantees equal messages by $\circ_2$, and when the protocol $T_3$ is positive, $\mc{F}_3$ relies on $\mc{D}_2$ and $\mc{D}_3$ to send the same messages, and $\mc{D}_2$ and $\mc{D}_3$ guarantee equal messages by $\dagger_2$. The rely-guarantee situation would imply that all messages sent or received along $\Delta \Vdash z_\eta {:}T_3[c_3]$ are identical in $\lre{\mc{C}_3}{\mc{D}_2}{\mc{F}_3}$ and $\lre{\mc{C}_3}{\mc{D}_3}{\mc{F}_3}$. The proof of transitivity follows from this observation and $\dagger_1$.

However, this technique on its own is insufficient in our effectful language with higher-order channels. In our setting, substitutions are dynamic configurations of processes and can grow by receiving a tree from the program along a higher-order channel. For example, in the proof of transitivity sketched above, consider a case where $\mc{D}_2$ and $\mc{D}_3$ are connected to $\mc{F}_3$ along $z_\eta$ with protocol $A \otimes B$, and  $\mc{D}_2$ sends a tree $\Psi_0;\cdot \Vdash  \mc{T}_2::x_\delta{:}A$ to $\mc{F}_3$ along $z_\eta$ and  $\mc{D}_3$ sends a tree $\Psi_0;\cdot \Vdash  \mc{T}_3::x_\delta{:}A$ to $\mc{F}_3$ along $z_\eta$. At this point, the substitutions evolve to $\mc{T}_3\mb{msg}(\mb{send}x_\delta\,z_\eta;\_)\mc{F}_3$  and  $\mc{T}_2\mb{msg}(\mb{send}x_\delta\,z_\eta;\_)\mc{F}_3$, respectively, and cease to be identical (by the value interpretation for $\otimes$ in \Cref{fig:rel_value_right}). To rescue our proof of transitivity, we use the fact that $\mc{T}_3$ and $\mc{T}_2$, although not necessarily identical, are related as assured by the value interpretation of $\otimes$ in \Cref{fig:rel_value_right}. \Cref{lem:compositionality} shows that the compositions of two pairs of related trees are related.

\begin{lemma}[Compositionality]\label{lem:compositionality}
For {\small $\Psi_0;\Delta , u_\alpha^c{:}T \Vdash \mc{D}_i:: K$} and {\small $\Psi_0;\Delta'  \Vdash \mc{T}_i::  u_\alpha^c{:}T$} and {\small$\cdot \Vdash \mc{C}_i :: \Delta$ and $\Psi_0;\cdot \Vdash \mc{C}'_i :: \Delta'$} and {\small $\Psi_0;K \Vdash \mc{F}_i:: \cdot$}, if {\small $\forall m.\, (\lre{\mc{C}_1\mc{C}'_1\mc{T}_1}{\mc{D}_1}{\mc{F}_1};\lre{\mc{C}_2\mc{C}'_2\mc{T}_2}{\mc{D}_2}{\mc{F}_2}) \in \mc{E}^\xi_{\Psi_0}\llbracket \Delta, u_\alpha^c{:}T \Vdash K \rrbracket^m$} and {\small $\forall m.\, (\lre{\mc{C}'_1}{\mc{T}_1}{\mc{C}_1\mc{D}_1\mc{F}_1};\lre{\mc{C}'_2}{\mc{T}_2}{\mc{C}_2\mc{D}_2\mc{F}_2}) \in \mc{E}^\xi_{\Psi_0}\llbracket \Delta'\Vdash  u_\alpha^c{:}T \rrbracket^m$} then {\small$\forall m.\, (\lre{\mc{C}_1\mc{C}'_1}{\mc{T}_1 \mc{D}_1}{\mc{F}_1};\lre{\mc{C}_2\mc{C}'_2}{\mc{T}_2\mc{D}_2}{\mc{F}_2}) \in \mc{E}^\xi_{\Psi_0}\llbracket \Delta', \Delta\Vdash K \rrbracket^m$}.
\end{lemma}
\begin{proof}
For the complete proof see Lem.~C.15 in the appendix.
\end{proof}
The lemma allows us to compose two pairs of trees only if they are related \emph{for all} $m$. In the situation described above, we need to ensure that the value interpretation establishing the relation between $\mc{T}_3$ and $\mc{T}_2$ holds for every $m$. In other words, we need to prove that the quantifier, ${\bf\color{violet}\forall\,m.}$, originally given over the term interpretation carries over to the value interpretation too.
\begin{lemma}\label{lem:moveq}If ${\small{ \bf\color{violet}\forall m.}({\lr{\mc{C}_1}{\mc{D}_1}{\mc{F}_1}}, \lr{\mc{C}_2}{\mc{D}_2}{\mc{F}_2}) \in \mc{E}^\xi_{ \Psi_0}\llbracket \Delta \Vdash K \rrbracket^{m+1},}$then:
{\small\[\begin{array}{l}
\m{if}\; \lr{\mc{C}_1}{\mc{D}_1}{\mc{F}_1}\;\mapsto^{{*}_{\Theta; \Upsilon_1}}_{\Delta \Vdash K} \; \lr{\mc{C}_1}{\mc{D}'_1}{\mc{F}_1}\, \m{then}\\[2pt]\exists \mc{D}'_2 \m{\,such \, that\,}\,\lr{\mc{C}_2}{\mc{D}_2}{\mc{F}_2}\mapsto^{{*}_{\Theta'; \Upsilon_1}}_{\Delta \Vdash K} \lr{\mc{C}_2}{\mc{D}'_2}{\mc{F}_2}, \; \m{and}\\[2pt]
 \forall \mc{C}'_1,\mc{F}'_1,\mc{C}'_2, \mc{F}'_2\,\m{if}\, \lr{\mc{C}_1}{\mc{D}'_1}{\mc{F}_1}\mapsto^{{*}_{\Theta_1;\Upsilon_1}}_{\Delta \Vdash K} \lr{\mc{C}'_1}{\mc{D}'_1}{\mc{F}'_1}\,\\[2pt] \qquad \qquad \quad\;\;\m{and}\,\lr{\mc{C}_2}{\mc{D}'_2}{\mc{F}_2}\mapsto^{{*}_{\Theta_2;\Upsilon_1}}_{\Delta \Vdash K} \lr{\mc{C}'_2}{\mc{D}'_2}{\mc{F}'_2},\,\m{then}\\[2pt] 
 \forall\, x \in \Upsilon_1.\,{\bf\color{violet}\forall m.} (\lr{\mc{C}'_1}{\mc{D}'_1}{\mc{F}'_1}; \lr{\mc{C}'_2}{\mc{D}'_2}{\mc{F}'_2})\in \mc{V}^\xi_{\Psi_0}\llbracket  \Delta \Vdash K \rrbracket_{\cdot;x}^{m+1}\,\m{and}\,\\[2pt]
\forall\, x \in (\Theta_1 \cap \Theta_2).\,{\bf\color{violet}\forall m.} (\lr{\mc{C}'_1}{\mc{D}'_1}{\mc{F}'_1}; \lr{\mc{C}'_2}{\mc{D}'_2}{\mc{F}'_2})\in \mc{V}^\xi_{\Psi_0}\llbracket  \Delta \Vdash K \rrbracket_{x;\cdot}^{m+1}
\end{array}
\]}
\end{lemma}
 \begin{proof}
  The proof relies on two results: (1) For each {\small$\Upsilon_1$}, if {\small$\mc{D}_2 \mapsto^{*_{\Upsilon_1}}_{\Delta \Vdash K}\mc{D}$} for some $\mc{D}$, we can build a unique \emph{minimal} configuration {\small$\mc{D}''_2$} with {\small$\mc{D}_2 \mapsto^{*_{\Upsilon_1}}_{\Delta \Vdash K}\mc{D}''_2$}. It is minimal in the sense that for every $\mc{D}'_2$, if {\small$\mc{D}_2 \mapsto^{*_{\Upsilon_1}}_{\Delta \Vdash K}\mc{D}'_2$}, then {\small$\mc{D}''_2 \mapsto^{*_{\Upsilon_1}}_{\Delta \Vdash K}\mc{D}'_2$}. (2) A standard backward closure property for the value interpretation of the logical relation. An interested reader may consult Lem.~C.8,~C.10,~C.14 in the appendix.
 \end{proof}
Our choice of step index is an enabler for the proof of \Cref{lem:moveq}. Choosing the number of observations as the index ensures that the set $\Upsilon_1$ built in $\lr{\mc{C}_1}{\mc{D}_1}{\mc{F}_1}\;\mapsto^{{*}_{\Theta; \Upsilon_1}}_{\Delta \Vdash K} \; \lr{\mc{C}_1}{\mc{D}'_1}{\mc{F}_1}$ is independent from the index $m$. This allows us to build the unique minimal $\mc{D}''_2$ for each $\Upsilon_1$ without the need to instantiate {\bf\color{violet}$\forall\,m.$} in the assumption. If we were to use the number of computation or unfolding steps as a step index, the set $\Upsilon_1$ would depend on the number of steps of the first program, and proving \Cref{lem:moveq} would have been more challenging, if not impossible. 
\subsection{Adequacy}
We show that two logically equivalent programs are indistinguishable by any well-typed context.

The relation $\sim^m$ is a weak stratified bisimulation based on a labeled transition system (LTS).
It is bounded by the number of observations that can be made and thus ``stratified'' to make the definition inductive.
For {\small${(\lr{\mathcal{C}_1}{\mathcal{D}_1} {\mathcal{F}_1}}, \lr{\mathcal{C}_2}{\mathcal{D}_2} {\mathcal{F}_2}) \in \m{Tree}_{\Psi_0}(\Delta \Vdash K)$}, we read ${\lr{\mathcal{C}_1}{\mathcal{D}_1} {\mathcal{F}_1}} \sim^{m+1} \lr{\mathcal{C}_2}{\mathcal{D}_2} {\mathcal{F}_2}$ as: if $\lr{\mathcal{C}_1}{\mathcal{D}_1} {\mathcal{F}_1}$, after none or many internal steps, produces an outgoing message $q$ or an incoming message $\overline{q}$ along an observable channel $x \in \Delta$ or $x \in K$, then  $\lr{\mathcal{C}_2}{\mathcal{D}_2} {\mathcal{F}_2}$ will produce the same outgoing message $q$ or incoming message $\overline{q}$ along $x$ with none or many internal steps. Moreover, the post-step configurations are related by $\sim^{m}$. The same holds when $\lr{\mathcal{C}_2}{\mathcal{D}_2} {\mathcal{F}_2}$ produces a message.


\begin{theorem}[Adequacy]\label{thm:adeq}
If {\small $(\Delta \Vdash \mathcal{D}_1::K)  \equiv^{\Psi_0}_{\xi} (\Delta \Vdash \mathcal{D}_2:: K)$} and {\small${(\lr{\mathcal{C}}{\mathcal{D}_1} {\mathcal{F}}}, \lr{\mathcal{C}}{\mathcal{D}_2} {\mathcal{F}}) \in \m{Tree}_{\Psi_0}(\Delta \Vdash K)$},
 then \[\forall m. {\lr{\mathcal{C}}{\mathcal{D}_1} {\mathcal{F}}} \sim^m \lr{\mathcal{C}}{\mathcal{D}_2} {\mathcal{F}}.\]
\end{theorem}
\begin{proof}
The techniques we used in the proof are similar to those used in the proof of transitivity and depend on the rely-guarantee property that we explained in \Cref{subsec:transitivity}. See Thm.~C.24 in the appendix. 
\end{proof}



%% file: sections/6.related_work.tex
\section{Related work}

The seminal works on IFC type systems for sequential imperative languages~\cite{Volpano96} and multi-threaded imperative languages~\cite{SmithVolpanoPOPL1998} have spurred a multitude of work on IFC type systems for imperative programming (c.f.,~\cite{SabelfeldIEE2003}).  We focus here on the more closely related systems that we categorize as follows:

\textbf{IFC type systems for binary session types:} The work most closely related to ours is the one by \citet{DerakhshanLICS2021}, which develops an IFC type system for intuitionistic linear logic session types.
The authors prove noninterference for their language, using a logical relation, the first logical relation for noninterference for session types.
In contrast to our work, their work does not support recursive types and exploits that programs are terminating.
Moreover, we phrase a stronger adequacy statement and prove our logical relation to be an equivalence relation.

\textbf{IFC type systems for multiparty session types:}
IFC type systems have also been explored for multiparty session types \cite{CapecchiCONCUR2010,CapecchiARTICLE2014,CastellaniARTICLE2016,Ciancaglini2016}.
These works explore declassification \cite{CapecchiCONCUR2010,CapecchiARTICLE2014} and flexible run-time monitoring techniques \cite{CastellaniARTICLE2016,Ciancaglini2016}.
Our work not only differs in use of session type paradigm (\ie binary vs. multiparty) but also in use of a logical relation for showing noninterference.

\textbf{IFC type systems for process calculi:}
Our work is more distantly related with IFC type systems developed for process calculi \cite{HondaESOP2000,HondaYoshidaPOPL2002,CrafaARTICLE2002,CrafaTGC2005,CrafaFMSE2006,HondaYoshidaTOPLAS2007,Crafa2007,HENNESSYRIELY2002,HENNESSY20053,KOBAYASHI2005,ZDANCEWIC2003,POTTIER2002} by virtue of its underlying concurrent, message-passing computation model.
These works prevent information leakage by associating a security label with channels/types/actions \cite{HondaYoshidaPOPL2002,HondaYoshidaTOPLAS2007}, read/write policies with channels \cite{HENNESSYRIELY2002,HENNESSY20053}, or capabilities with expressions \cite{CrafaARTICLE2002}.
\citet{HondaYoshidaPOPL2002,HondaYoshidaTOPLAS2007} moreover also use a substructural type system and prove a sound embedding of Dependency Core Calculus (DCC)~\cite{AbadiPOPL1999} into their calculus.
Our work sets itself apart from these works in its use of session types and meta theoretic developments based on a novel recursive logical relation.


\textbf{Logical relations for session types and noninterference in general:}
The application of logical relations to session types has focused predominantly on unary logical relations for proving termination~\cite{PerezESOP2012, PerezARTICLE2014,DeYoungFSCD2020} with the exception of a binary logical relation for parametricity \cite{CairesESOP2013} and noninterference \cite{DerakhshanLICS2021} discussed above.
Our noninterference definition is moreover
inspired by \citet{BowmanIAhmedCFP2015}, which also define noninterference in terms of a logical
relation, albeit in a sequential purely functional setting.
\citet{OddershedePOPL2021} recently contributed a logical relation for a higher-order functional language with higher-order store, recursive types, existential types, impredicative type polymorphism, and label polymorphism.
They use a semantic approach to proving noninterference,  allowing integration of syntactically well-typed and ill-typed components, if the latter are shown to be semantically sound.
Besides the difference in language, the authors consider termination-insensitive noninterference only.

\textbf{Relationship to logical relations for stateful languages:}
Logical relations have been scaled to accommodate state using Kripke logical relations \cite{PittsStarkHOOTS1998}.
Kripke logical relations are indexed by a \emph{possible world W}, providing a semantic model of the heap.
Invariants can then be imposed that must be preserved by any future worlds W'.
Kripke logical relations have been successfully combined with step indexing \cite{AppelMcAllesterTOPLAS2001,AhmedESOP2006} to address circularity arising in higher-order stores \cite{AhmedPOPL2009,DreyerPOPL2010,DreyerJFP2012} and to express state transition systems.
Our \emph{session logical relation}  does \emph{not} use a Kripke logical relation because it relies on session typing of the configurations that it relates.
Session typing in turn governs state transitions and avoids circular configuration topologies by virtue of linearity.




